\documentclass[BCOR=15mm,DIV=12,headings=optiontoheadandtoc]{scrbook}

\usepackage{scrpage2}
\usepackage{amsthm, amsmath, amssymb, mathrsfs}
\usepackage{wasysym}

\usepackage{graphicx}
\usepackage{enumitem}
\usepackage{pgf,tikz}
\usepackage[all]{xy}

\numberwithin{equation}{chapter}
\usepackage[bookmarksnumbered=true]{hyperref}
\usepackage{cite}
\usepackage{caption}
\setkomafont{publishers}{\large}

\setheadsepline{.2mm}
\usepackage{color}

\renewcommand{\it}{\emph}

\newcommand{\boxd}[1]{\boxed{\phantom{\Biggl(}#1\phantom{\Biggl)}}}
\renewcommand{\sec}[1]{section \ref{sec:#1}}
\newcommand{\fig}[1]{figure \ref{fig:#1}}
\newcommand{\dref}[1]{definition\,\ref{def:#1}}

\newcommand{\rem}[1]{remark \ref{rem:#1}}

\usetikzlibrary{arrows, positioning, patterns, shapes, external}
\definecolor{jred}{rgb}{0.8,0,0}
\definecolor{jgreen}{rgb}{0,0.7,0}
\definecolor{jblue}{rgb}{0,0,0.8}

\tikzstyle{c} =	[coordinate]
\tikzstyle{v} = 	[circle, draw=black, line width=.2pt, fill=black, inner sep=0pt, minimum size=1.5mm]
\tikzstyle{vb} =	[circle, draw=black, line width=.2pt, fill=jred, inner sep=0pt, minimum size=1.5mm]
\tikzstyle{vh} =	[circle, draw=black, line width=.2pt, fill=jblue, inner sep=0pt, minimum size=1.5mm]
\tikzstyle{vs} =	[circle, draw=black, line width=.2pt, fill=jgreen, inner sep=0pt, minimum size=1.5mm]
\tikzstyle{e} =	[draw=jred,line width=1.6pt]
\tikzstyle{eb} =	[draw=jgreen,line width=1.6pt]
\tikzstyle{eh} =	[dashed]
\tikzstyle{es} =	[draw=jblue,line width=1.6pt]
\tikzstyle{f} = 	[line width=0.01pt,dotted,fill=blue, fill opacity=.1]
\tikzstyle{cs} = 	[draw=none,fill=jred, fill opacity=.7]
\tikzstyle{bb} =	[dashed,line width=1.4pt]
\tikzstyle{bh} =	[dotted,line width=1pt]
\tikzstyle{h} = 	[ellipse, inner sep=0.1pt, draw=black]
\tikzstyle{hv} = 	[ellipse, inner sep=0.1pt, draw=red]
\tikzstyle{he} = 	[ellipse, inner sep=0.1pt, draw=blue]
\tikzstyle{hb} = 	[ellipse, inner sep=0.1pt, draw=jgreen]

\tikzstyle{venn1} = [fill=gray!20!white, draw=black, pattern=north east lines]
\tikzstyle{venn2} = [fill=gray!20!white, draw=black, pattern=north west lines]

\newcommand{\trace}{{trace}}

\newcommand{\gn}{G_{\textsc{n}}}
\newcommand{\gnb}{\kappa}
\newcommand{\h}{\hbar}
\newcommand{\bi}{\gamma_{\textsc{bi}}}
\newcommand{\lpl}{\ell_{\textsc{pl}}}
\newcommand{\lbi}{\ell_{\gamma}}

\newcommand{\gr}{\textsc{gr}}

\newcommand{\qrc}{\textsc{qrc}}
\newcommand{\cdt}{\textsc{cdt}}
\newcommand{\lqg}{\textsc{lqg}}
\newcommand{\sfm}{\textsc{sf}}
\newcommand{\gft}{\textsc{gft}}
\newcommand{\sub}{\textsc{mf}} 
\newcommand{\mfgft}{\textsc{mf}}
\newcommand{\sgft}{}
\newcommand{\dwgft}{\textsc{dw-gft}}

\newcommand{\regge}{\text{Regge}}
\newcommand{\pr}{\textsc{pr}}

\newcommand{\bft}{\textsc{bf}}

\renewcommand{\mp}{M_{\textrm{kin}}}
\newcommand{\ma}{M_{\textrm{dyn}}}
\newcommand{\mpp}{M_{\textrm{emp}}}
\newcommand{\ia}{M_{\textrm{app}}}
\newcommand{\res}{r}

\newcommand{\stm}{\mathcal{M}}		
\newcommand{\sm}{\Sigma}			
\newcommand{\std}{D}				
\newcommand{\sd}{d}				
\newcommand{\br}{\partial}			

\newcommand{\m}{n}				
\newcommand{\mf}{\mathcal{M}}		
\newcommand{\surface}{\mathcal{S}}

\newcommand{\p}{p}					
\newcommand{\q}{q}		
\newcommand{\size}{N}				
\newcommand{\np}[1]{N^{[#1]}}

\newcommand{\cl}{c}
\newcommand{\clp}[1]{\cl_{#1}}
\newcommand{\lc}{\cl_{-1}}			
\newcommand{\gc}{\cl_{\m+1}}			
\newcommand{\cs}{\cl^\star}
\newcommand{\csp}[1]{\cs_{#1}}
\newcommand{\s}{\sigma}

\newcommand{\cp}{\mathcal{C}}		
\newcommand{\cps}{\cp^{\star}}		
\newcommand{\cpp}[1]{\cp^{[#1]}}
\newcommand{\cpsp}[1]{{\cp^{\star [#1]}}}
\newcommand{\cm}{\mathcal{C}}		
\newcommand{\cms}{\cm^\star}	

\newcommand{\cmsp}[1]{{\cm^{\star [#1]}}}

\newcommand{\dm}{dim}				
\newcommand{\bs}{\partial}			
\newcommand{\bsc}[2]{\epsilon_{#1,#2}}	
\newcommand{\bsi}{\partial^{\star}}
\newcommand{\dbs}{\partial^{\star}}
	
\newcommand{\clos}{\overline}
\newcommand{\bm}{\delta}			
\newcommand{\bmp}[1]{\delta_{#1}}		
\newcommand{\bc}[2]{\eta_{#1,#2}}		
\newcommand{\bcs}[2]{\eta^\star_{#1,#2}}

\newcommand{\simc}{\cp_\mathrm{sim}}	
\newcommand{\simcp}[1]{\cp_\mathrm{sim}^{[#1]}}	

\newcommand{\simcsp}[1]{{\cp_\mathrm{sim}^{\star [#1]}}}
\newcommand{\T}{\mathcal{T}}			
\newcommand{\Tp}[1]{\mathcal{T}^{[#1]}}
\newcommand{\at}{a_{\text{t}}}			
\newcommand{\as}{a_{\text{s}}}			

\newcommand{\polyc}{\cp_\mathrm{poly}}	

\newcommand{\gpolyc}{\cp_\mathrm{gen}}
\newcommand{\gpolycp}[1]{\cp_\mathrm{gen}^{[#1]}}

\newcommand{\ssub}{\Delta}			
\newcommand{\issub}{\Delta^{-1}}
\newcommand{\issubs}{\Delta^{-1\star}}
\newcommand{\subc}{\mathcal{S}}		
\newcommand{\subcp}[1]{\subc^{[#1]}}

\newcommand{\bra}{(}
\newcommand{\ket}{)}

\newcommand{\Nc}{ }
\newcommand{\Vc}{ }
\newcommand{\dl}{l^\star}				
\newcommand{\dist}[2]{D(#1,#2)}		

\newcommand{\V}{\mathcal{V}}
\newcommand{\Vb}{\overline{\mathcal{V}}}
\newcommand{\vb}{{\bar v}}
\newcommand{\Vh}{\widehat{\mathcal{V}}}
\newcommand{\vh}{{\hat v}}
\newcommand{\E}{\mathcal{E}}
\newcommand{\Eb}{\overline{\mathcal{E}}}
\newcommand{\eb}{{\bar e}}
\newcommand{\Eh}{\widehat{\mathcal{E}}}

\newcommand{\F}{\mathcal{F}}

\newcommand{\bulk}{\alpha}			
\newcommand{\bisec}{\beta}			
\newcommand{\gm}{\gamma}			

\newcommand{\bst}{\bs}

\newcommand{\copies}{k}
\newcommand{\ks}{{k^\star}}
\newcommand{\loopless}{\textsc{l}}
\newcommand{\rl}{{\copies,\loopless}}
\newcommand{\drl}{{\std,\loopless}}
\newcommand{\simplicial}{\textsc{s}}
\newcommand{\rs}{{\copies,\simplicial}}
\newcommand{\drs}{{{\std},\simplicial}}

\newcommand{\lnb}{\textsc{l}\textrm{-}\textsc{nb}}
\newcommand{\snb}{\textsc{s}\textrm{-}\textsc{nb}}
\newcommand{\sdec}{\std_{\copies,\textsc{l-s}}}

\newcommand{\sda}{\mathfrak{s}}
\newcommand{\coils}{\mathcal{\V_{\sda}}}
\newcommand{\coil}{\mathcal{\V_{\text{coil}}}}
\newcommand{\reroutings}{\E_\sda}

\newcommand{\hs}{\mathcal{H}	}			
\newcommand{\hkin}{\mathcal{H}_{\text{kin}}}	
\newcommand{\hk}[1]{\mathcal{H}_{\text{kin},#1}}	
\newcommand{\hphys}{\mathcal{H}_{\text{phys}}}
\newcommand{\phys}{\textrm{phys}}
\newcommand{\prjct}{\pi}
\newcommand{\alg}{\mathfrak{A}}
\newcommand{\spec}{\text{Spec}}

\newcommand{\qbra}{\langle}
\newcommand{\qket}{\rangle}

\newcommand{\G}{G}				
\newcommand{\dg}{{{\delta}_{_\G}}}
\newcommand{\la}{\mathfrak{g}}		
\newcommand{\dla}{\delta_\la}
\newcommand{\A}{\omega}			
\newcommand{\B}[1]{B^{^{#1}}}			
\newcommand{\eu}[1]{e^{^{#1}}}		
		
\newcommand{\rep}[1]{{j_{#1}}}
\newcommand{\intw}[1]{{\iota_{#1}}}
\renewcommand{\dj}[1]{\dim({#1})}
\newcommand{\threej}[4]{\begin{pmatrix}
 \rep{#1} & \rep{#2} & \rep{#3} \\ 
 #4_{#1} & #4_{#2} & #4_{#3}
\end{pmatrix}}

\newcommand{\sixj}[6]{\begin{Bmatrix}
 \rep{#1} & \rep{#2} & \rep{#3} \\ 
 \rep{#4} & \rep{#5} & \rep{#6}
\end{Bmatrix}}
\newcommand{\csu}{c_G}

\newcommand{\sn}{(\bg,\rep{\eb},\intw{\vb})}
\newcommand{\snr}{|\bg,\rep{\eb},\intw{\vb}\qket}

\newcommand{\sms}{{\bg,\rep{\eb}}}
\newcommand{\smr}{|\bg,\rep{\eb}\qket}
\newcommand{\sml}{\qbra \bg,\rep{\eb} |}

\newcommand{\gftvr}{| \emptyset \qket}

\newcommand{\bvol}{vol}

\newcommand{\so}{\mathfrak{so}}
\newcommand{\su}{\mathfrak{su}}

\newcommand{\SU}{\mathrm{SU}}

\newcommand{\fdiagram}{\Gamma}
\newcommand{\sym}{\mathrm{sym}}

\newcommand{\Mcal}{\Z_{M}}
\newcommand{\real}{\mathrm{rl}}
\newcommand{\virtual}{\mathrm{vl}}

\newcommand{\curv}{\gamma}
\newcommand{\bg}{\mathfrak{c}}
\newcommand{\bge}{\mathcal{G}}
\newcommand{\bgs}{\mathfrak{C}}
\newcommand{\bgt}{\widetilde{\mathfrak{c}}}
\newcommand{\bgst}{\widetilde{\mathfrak{C}}}

\newcommand{\bbg}{\mathfrak{b}}
\newcommand{\bbgs}{\mathfrak{B}}
\newcommand{\bbgt}{\widetilde{\mathfrak{b}}}
\newcommand{\bbgst}{\widetilde{\mathfrak{B}}}

\newcommand{\sfa}{\mathfrak{a}}
\newcommand{\sfas}{\mathfrak{A}}
\newcommand{\sfat}{\widetilde{\mathfrak{a}}}
\newcommand{\sfast}{\widetilde{\mathfrak{A}}}

\newcommand{\bp}{{\mathfrak{p}}}
\newcommand{\bps}{\mathfrak{P}}
\newcommand{\bpt}{{\widetilde{\mathfrak{p}}}}
\newcommand{\bpst}{\widetilde{\mathfrak{P}}}

\newcommand{\sfr}{\mathfrak{m}}
\newcommand{\sfrs}{\mathfrak{M}}
\newcommand{\sfrt}{\widetilde{\mathfrak{m}}}
\newcommand{\sfrst}{\widetilde{\mathfrak{M}}}

\newcommand{\sfo}{A} 

\newcommand{\sfls}{g} 

\newcommand{\sta}{\mathfrak{s}}

\newcommand{\Vbar}{\overline{\mathcal{V}}}

\newcommand{\Vhat}{\widehat{\mathcal{V}}}

\newcommand{\Phit}{\widetilde{\Phi}}
\newcommand{\phit}{\widetilde{\phi}}
\newcommand{\Obt}{\widetilde{\mathcal{O}}}
\newcommand{\obt}{\widetilde{O}}

\newcommand{\Acal}{\mathcal{A}}

\newcommand{\Dcal}{\mathcal{D}}

\newcommand{\Fcal}{\mathcal{F}}

\newcommand{\Ocal}{\mathcal{O}}
\newcommand{\Pcal}{\mathcal{P}}

\newcommand{\bbB}{\mathbf{B}}  

\newcommand{\Dbb}{\mathbf{D}}
\newcommand{\Ibb}{\mathbf{I}}
\newcommand{\Kbb}{\mathbf{K}}
\newcommand{\Pbb}{\mathbf{P}}
\newcommand{\Sbb}{\mathbf{S}}
\newcommand{\Vbb}{\mathbf{V}}

\newcommand{\bbBt}{\widetilde{\mathbf{B}}}
\newcommand{\Kbbt}{\widetilde{\mathbf{K}}}

\newcommand{\Vbbt}{\widetilde{\mathbf{V}}}

\newcommand{\Kbbb}{\overline{\mathbf{K}}}
\newcommand{\Vbbb}{\overline{\mathbf{V}}}
\newcommand{\bbBb}{\overline{\mathbf{B}}}

\newcommand{\dbb}{\mathbf{d}}
\newcommand{\ibb}{\mathbf{i}}
\newcommand{\pbb}{\mathbf{p}}
\newcommand{\vbb}{\mathbf{v}}

\def\Ds{D_{\textsc s}}
\def\ds{d_{\textsc s}}

\def\dh{d_{\textsc h}}
\def\dw{d_{\textsc w}}
\def\uv{\textsc{uv}}

\newcommand{\Vv}[1]{V_{v_0,#1}(r)}
\newcommand{\Vr}[1]{V_{#1}(r)}
\renewcommand{\P}[1]{P_{#1}(\tau)}
\newcommand{\X}[1]{\langle X^2 \rangle_{#1}(\tau)}
\newcommand{\Xv}[1]{\langle X^2 \rangle_{v_0,#1}(\tau)}

\newcommand{\ql}{\widehat l}
\newcommand{\qVv}[1]{\widehat {V_{v_0,#1}(r)}}
\newcommand{\qVr}[1]{\widehat {V_{#1}(r)}}

\newcommand{\qX}{\widehat{\langle X^2 \rangle(\tau)}}
\newcommand{\qXv}[1]{\widehat{\langle X^2 \rangle}_{v_0,#1}(\tau)}

\newcommand{\hl}[1]{{\boxplus_\size^#1}}
\newcommand{\hln}[2]{{\boxplus_{#2}^{#1}}}
\newcommand{\hlinf}[1]{{\boxplus_{\infty}^{#1}}}

\newcommand{\hxinf}[1]{\varhexagon_\infty^{#1}}

\newcommand{\olinf}[1]{{\octagon_{\infty}^{#1}}}
\newcommand{\dipole}[1]{{\ominus^{#1}}}

\newcommand{\cdiff}{c_{\text{diff}}}		
\newcommand{\ttau}{\tilde\tau}
\newcommand{\tphi}{\widetilde\phi}

\renewcommand{\digamma}{\psi}

\newcommand{\qs}{\psi}				
\newcommand{\qsc}{a}
\renewcommand{\csc}{\psi_\sms^{J_\eb,K_\eb}}
\newcommand{\cscp}{\psi_{\bg,j'_\eb}^{J'_\eb,K'_\eb}}
\newcommand{\jmin}{{j_{\text{min}}}}
\newcommand{\jmax}{{j_{\text{max}}}}
\newcommand{\rjc}{| j_c,\cm\qket}

\newcommand{\jc}{{j_c,\cm}}

\newcommand{\rj}{|j,\cm\rangle}
\newcommand{\lj}{\langle j,\cm |}
\newcommand{\rjb}{|j,\hl\sd \rangle}
\newcommand{\ljb}{\langle j,\hl\sd |}
\newcommand{\rsup}{| V_0,\jmin,\jmax \rangle}
\newcommand{\ssup}{{V_0,\jmin,\jmax}}

\newcommand{\sumint}{\sum}

\renewcommand\[{\begin{equation}}
\renewcommand\]{\end{equation}}
\newcommand{\ba}{\begin{eqnarray}}
\newcommand{\ea}{\end{eqnarray}}

\newcommand{\N}{\mathbb N}
\newcommand{\Z}{\mathbb Z}
\newcommand{\Q}{\mathbb Q}
\newcommand{\R}{\mathbb R}
\newcommand{\C}{\mathbb C}
\newcommand{\id}{\mathbb I}

\newcommand{\Tr}{\mathrm{Tr}}
\newcommand{\tr}{\mathrm{tr}}
\def\d{\mathrm{d}}
\def\da{\mathrm{d^\star}}
\def\D{\mathcal{D}}
\def\e{\textrm e}
\def\rme{\mathrm e}
\def\i{\mathrm i}
\def\rmi{\mathrm i}
\def\j{\jmath}
\def\sgn{sgn} 

\renewcommand{\ge}{\geqslant}
\renewcommand{\le}{\leqslant}

\mathchardef\ordinarycolon\mathcode`\:
\mathcode`\:=\string"8000
\begingroup \catcode`\:=\active
  \gdef:{\mathrel{\mathop\ordinarycolon}}
\endgroup

\newcommand{\cd}{\cdot}

\newcommand{\ra}{\rightarrow}

\newcommand{\lora}{\longrightarrow}

\newcommand{\us}{\underset}
\newcommand{\os}{\overset}
\newcommand{\ol}{\overline}
\newcommand{\In}{\subset}

\def\ie{{i.e.}~}
\def\eg{{e.g.}~}
\def\cf{{cf.}~}

\newtheoremstyle{mydef}
  {}
  {}
  {}
  {}
  {\bfseries}
  {.}
  { }
  {\thmname{#1}\thmnumber{ #2}\thmnote{ (#3)}}

\theoremstyle{mydef}
\newtheorem{defin}{Definition}[chapter]
\newtheorem{tdefin}[defin]{Theory explication}
\newtheorem{remark}[defin]{Remark}
\newtheorem{example}[defin]{Example}

\theoremstyle{plain}
\newtheorem{proposition}[defin]{Proposition}

\newtheorem{conjecture}[defin]{Conjecture}
\newtheorem{corollary}[defin]{Corollary}

\newcommand{\COTa}{Calcagni:2013ku}
\newcommand{\COTb}{Calcagni:2014ep}
\newcommand{\COTc}{Calcagni:2015is}
\newcommand{\ORT}{Oriti:2015kv}

\newcommand{\lqgT}{Thiemann:2007wt}
\newcommand{\lqgR}{Rovelli:2004wb,Rovelli:2011tk}
\newcommand{\sfOP}{Perez:2003wk,Perez:2013uz}
\newcommand{\gftFO}{Freidel:2005jy,Oriti:2012wt}
\newcommand{\cdtAJL}{Ambjorn:2000hp,Ambjorn:2007wl,Ambjorn:2012vc}
\newcommand{\qrcW}{Williams:1992kw,Williams:1997bn,Williams:2006iu,Williams:2007up,Hamber:2009wl}
\newcommand{\asR}{Niedermaier:2006up,Reuter:2012jx}
\newcommand{\hlH}{Horava:2009ho}

\newcommand{\gftcondensate}{Gielen:2013cr,Gielen:2014gv,Gielen:2014ca,Calcagni:2014jt,Gielen:2014vk,Sindoni:2014vs}

\newcommand{\BMS}{Balzer:1987tx}

\newcommand{\KKL}{Kaminski:2010ba}
\newcommand{\KLP}{Kisielowski:2012bo}
\newcommand{\EPRL}{Engle:2007em,Engle:2008ka,Engle:2008fj}
\newcommand{\FK}{Freidel:2008fv}
\newcommand{\BO}{Baratin:2012br}
\newcommand{\lostGS}{Gurau:2010iu,Smerlak:2011ea}
\newcommand{\gftrenorm}{Freidel:2009ek,BenGeloun:2013fw,BenGeloun:2013dl,BenGeloun:2013ek,BenGeloun:2013uf,Samary:2014bs}
\newcommand{\COR}{Carrozza:2014ee,Carrozza:2014bh,Carrozza:2014tf}

\newcommand{\largeN}{Bonzom:2011cs,Gurau:2012ek,Gurau:2012hl,Bonzom:2012bg,Baratin:2014bea}
\newcommand{\double}{Gurau:2011sk}

\newcommand{\cdtfractal}{Ambjorn:2005fj,Ambjorn:2005fh}
\newcommand{\asfractal}{Lauscher:2005kn}
\newcommand{\hlfractal}{Horava:2009ho}
\newcommand{\ncfractal}{Benedetti:2009fo,Alesci:2012jl}
\newcommand{\snfractal}{Modesto:2009bc}

\begin{document}

\title{Discrete quantum geometries and their effective dimension}

\publishers{
DISSERTATION\\ 

\

zur Erlangung des akademischen Grades\\

DOCTOR RERUM NATURALIUM\\ 

(Dr. rer. nat.)\\

im Fach Physik\\

\

eingereicht an der\\
Mathematisch-Naturwissenschaftlichen Fakult\"at\\
der Humboldt-Universit\"at zu Berlin\\

\

von\\
JOHANNES TH\"URIGEN

\vspace{1cm}

Pr\"asident der Humboldt-Universit\"at zu Berlin:\\
Prof. Dr. Jan-Hendrik Olbertz \\
Dekan der Mathematisch-Naturwissenschaftlichen Fakult\"at:\\
Prof. Dr. Elmar Kulke

\

Betreuer: 

\

\begin{tabular}{lll}
\hspace{2.1cm}	&  Dr. Daniele Oriti \hspace{1cm}	& Albert-Einstein-Institut Potsdam \\
			&  Dr. Gianluca Calcagni 			& Instituto de Estructura de la Materia\\  
			&							& (CSIC) Madrid\\
\end{tabular}

\

Gutachter: 

\

\begin{tabular}{lll}
1. & Prof. Dr. Hermann Nicolai 	& Albert-Einstein-Institut Potsdam \\
2. & Prof. Dr. Dirk Kreimer		& Humboldt-Universit\"at zu  Berlin \\
3. & Prof. Dr. Jerzy Lewandowski & Uniwersytet Warszawski \\
\end{tabular}

\vspace{1cm}

Tag der m\"undlichen Pr\"ufung: 2. Juli 2015
}
\date{}

\dedication{\emph{F\"ur Annika \& Nikolas}}

\maketitle


\clearpage

\frontmatter

\subsection*{Abstract}

The challenge of coherently combining general relativity and quantum field theory into a quantum theory of gravity is one of the main outstanding tasks in theoretical physics.
In several related approaches towards this goal, such as group field theory, spin-foam models, loop quantum gravity and simplicial quantum gravity, quantum states and histories of the geometric degrees of freedom turn out to be based on discrete space and spacetime.
The most pressing issue is then how the smooth classical geometries of general relativity arise from such discrete quantum geometries in some semiclassical and continuum limit.
This has to be expressed in terms of suitable geometric observables which should demonstrate that the desired features of smooth spacetime are recovered.

In this thesis I tackle the question of suitable observables focusing on the effective dimension of discrete quantum geometries, more specifically, the spectral, Hausdorff and walk dimension. These are also the relevant indicators of a possible fractal structure.
For this purpose I give an extensive and exhaustive, purely combinatorial description of the discrete structures which these geometries have support on.
As a side topic, this allows to present an extension of group field theory to cover the combinatorially larger kinematical state space of loop quantum gravity.
This can be realized with a particularly effective construction using the tensor model technique of a dual-weighting mechanism.

Then, I introduce a discrete calculus for arbitrary ($p$-form) fields on such fundamentally discrete geometries with a particular focus on the Laplacian.
This permits to define the effective-dimension observables for quantum geometries.
Preliminary, I study classical effects of topology, geometry and discreteness in a systematic way. 
This sets the stage to check whether quantum geometries reproduce the effective dimensions of classical geometries in an appropriate semiclassical regime and to identify quantum effects in a quantum regime.
I analyse the effective dimensions for various classes of quantum geometries, in particular (a) I develop and apply numerical techniques to be able to compute the spectral dimension of combinatorially large geometries in the precise setting of $(2+1)$-dimensional loop quantum gravity;  (b) I use analytic solutions for a particular model to analyse quantum geometries of arbitrary spatial dimension $\sd$.

As a general result I find that the spectral dimension is more sensitive to the underlying combinatorial structure than to the details of the additional geometric data thereon. 
Semiclassical (coherent) states (a) on a given complex turn out to approximate the classical geometries they are peaking on rather well and there are no indications for stronger quantum effects.
On the other hand (b), I do find such effects for states which are superposition over a large number of complexes: there is a flow of the spectral dimension from the topological dimension $\sd$ on low energy (IR) scales to a real number $0<\alpha<\sd$ on high energy (UV) scales for power function superposition coefficients which is related to their exponent.
The Hausdorff and walk dimension do not exhibit any particular quantum effect.
In the special case of $\alpha=1$ these results allow to understand the quantum geometry  as effectively fractal.
Moreover, in this case the spectral dimension indicates a flow of the spacetime dimension $\std$ to a UV dimension $\std^\uv=2$, in accordance with the findings in other approaches.
These results apply in particular to special superpositions of spin-network states providing more solid 
indications for a dimensional flow in this context.
Quantum-geometry properties like a fractal structure or a dimensional flow may have phenomenological consequences, for example in the early universe.


\subsection*{Zusammenfassung}

Allgemeine Relativit\"atstheorie und Quantenfeldtheorie in koh\"arenter Weise zu einer Quantentheorie der Gravitation zu verbinden, ist eine der gr\"o\ss ten offenen Aufgaben in der theoretischen Physik.
In mehreren miteinander in Beziehung stehenden Ans\"atzen, die dieses Ziel verfolgen, n\"amlich Gruppenfeldtheorie, Spinschaum-Modellen, Schleifenquantengravitation und simplizialer Quantengravitation, stellt sich heraus, dass Quantenzust\"ande und Quantenentwicklungen der geometrischen Freiheitsgrade auf einem dis\-kre\-ten Raum beziehungsweise einer diskre\-ten Raumzeit basieren.
Die dringendste Frage ist dann, wie die glatten klassischen Geometrien der Allgemeinen Relativit\"atstheorie aus solch diskreten Quantengeometrien im semiklassischen und Kontinuums-Limes hervorgehen.
Dies muss durch geeignete geometrische Beobachtungsgr\"o\ss en beschrieben werden, welche zeigen sollten, dass die gew\"unschten Merkmale einer glatten Raumzeit wiedergewonnen werden.

In der vorliegenden Dissertation nehme ich die Frage geeigneter Beobachtungsgr\"o\ss en mit einem Fokus auf die effektive Dimension diskreter Quantengeometrien, genauer gesagt, spektrale, Hausdorff- und (Random-)Walk-Dimension,  in Angriff.
Dies sind auch die relevanten Indikatoren f\"ur eine m\"ogliche fraktale Geometrie.
Zu diesem Zweck gebe ich eine ausf\"uhrliche und ersch\"opfende, rein kombinatorische Beschreibung der diskreten Strukturen, auf denen solche Geometrien basieren. 
Als ein Nebenthema erlaubt dies die Darlegung einer Erweiterung der Gruppenfeldtheorie, so dass diese den kombinatorisch gr\"o\ss eren kinematischen Zustandsraum der Schleifenquantengravitation abdeckt,
was sich mit einer besonders wirkungsvollen Konstruktion durch die Tensor-Modell-Methode eines Mechanismus dualer Gewichtung realisieren  l\"asst.

Daraufhin f\"uhre ich einen diskreten Differentialrechnungskalk\"ul f\"ur beliebige ($\p$-Form-) Felder auf solch fundamental diskreten Geometrien mit einem speziellen Augenmerk auf dem Laplace-Operator ein. Dadurch wird die Definition der Observablen der effektiven Dimensionen f\"ur Quantengeometrien m\"oglich.
Als Voruntersuchung betrachte ich sys\-te\-ma\-tisch klassische Effekte von Topologie, Geometrie und Diskretheit.
Dies ist die Voraussetzung f\"ur die \"Uberpr\"ufung, ob Quantengeometrien die effektiven Dimensionen klas\-sischer Geo\-metrien in einem geeigneten semiklassischen Bereich reproduzieren, und, um Quanteneffekte in einem Quantenregime zu identifizieren.
Ich analysiere die effektiven Dimensionen verschiedener Klassen von Quantengeometrien. 
Insbesondere ent\-wickele ich (a) numerische Techniken, um die spektrale Dimension von kombinatorisch gro\ss en Geo\-metrien in der pr\"azisen Situation $(2+1)$-dimensionaler Schleifenquantengravitation berechnen zu k\"onnen, und wende sie an; und ich verwende (b) ana\-ly\-tische L\"osungen eines speziellen Modells, um Quantengeometrien in beliebiger r\"aumlicher Dimen\-sion $\sd$ zu untersuchen.

Als ein allgemeines Resultat finde ich heraus, dass die spektrale Dimension st\"arker von der zugrunde liegenden kombinatorischen Struktur als von den Details der zus\"atzlichen geometrischen Daten darauf abh\"angt. 
Es stellt sich heraus, dass (a)  semiklassische (koh\"arente) Zust\"ande auf einem gegebenen Komplex die entsprechenden klassischen Geometrien ziemlich genau approximieren, und es gibt keine Anzeichen f\"ur st\"arkere Quanteneffekte. 
Andererseits entdecke ich solche Effekte f\"ur (b) Zust\"ande, die aus \"Uberlagerungen einer gro\ss en Anzahl von Komplexen bestehen: Ich finde einen Fluss der spektralen Dimension von der topologischen Dimension $\sd$ bei kleinen Energieskalen (IR) hin zu einem reellen Wert $0<\alpha<\sd$ bei hohen Energien (UV) f\"ur Potenzfunktions-\"Uberlagerungskoeffizienten, der von deren Exponent abh\"angt.
Hausdorff- und Walk-Dimension zeigen keine besonderen Quanteneffekte.
Im Spezialfall $\alpha=1$ erlauben diese Resultate, die Quantengeometrie als effektiv fraktal aufzufassen.
Des Weiteren deutet die spektrale Dimension in diesem Fall auf einen Fluss der Raumzeitdimension $\std$ zu einer UV-Dimension $\std^\uv=2$ hin, im Einklang mit Ergebnissen in anderen Ans\"atzen.
Die genannten Resultate lassen sich insbesondere auf spezielle \"Uberlagerungen von Spin-Netzwerk-Zust\"
anden \"ubertragen, womit sie solidere Hinweise auf einen Dimensionsfluss in diesem Kontext darstellen.
Quantengeometrische Eigenschaften von der Art einer fraktalen Struktur oder eines Dimensionsflusses k\"onnten ph\"anomenologische Auswirkungen haben, beispielsweise im fr\"uhen Universum.

\tableofcontents


\newpage

\mainmatter

\addchap[tocentry={}]{Introduction}
\addcontentsline{toc}{chapter}{\protect\numberline{}Introduction}

One of the biggest open challenges in theoretical physics is to bring together the two fundamental theories describing phenomena on the macroscopic scales of large masses, 
that is general relativity (GR), and on microscopic scales governed by high energies, that is quantum (field) theory. 
By now several approaches towards such a quantum theory of gravity have developed into sophisticated active research programs \cite{Oriti:2007uc}.

Indications for a breakdown of theory at very small length scales are present both in general relativity and in the relativistic quantum field theories describing elementary particles and their interactions: general relativity is challenged by the generic presence of spacetime singularities in black hole and cosmological solutions of the gravitational field equations,
while divergences at large frequencies and momenta render relevant quantum field theories mathematically not well defined.
As a consequence, almost all approaches to quantum gravity agree that at very small length scale 
continuum should effectively be replaced in favour for some kind of \emph{discretum} \cite{Nicolai:2014hy}.
There are two main strategies in terms of which such a discretum may arise: one may either establish a quantum theory of gravity based on a generalization of point particles to higher dimensional objects, or based on discrete geometries of space and spacetime.
This thesis focuses on the latter alternative.

To sketch the challenge of quantum gravity a bit more explicitly, one expects such a theory to provide expectation values of observables of the form (in a path integral description)
\[\label{gr-path-integral}
\qbra O_{\sm}[g,\Psi] \qket
 = \int_{\us{\br \stm = \sm}{\stm}} \D g \D\Psi\; O_{\sm}[g,\Psi] \;\e^{\frac{\rmi}{\hbar} (S_{\gr}[g] + S_{\textrm{matter}}[g,\Psi] )}\:
\]
where the observable $O_{\sm}[g,\Psi]=O[g|_{\sm},\Psi|_{\sm}]$ is a functional of the gravitational metric field $g$ and all kinds of matter fields $\Psi$ 
on a spatial slice $\sm$ of spacetime. 
Thereby, according to observation, space has $\sd =  3$ dimensions and spacetime $\std = \sd +1 = 4$. 
The integral is a sum over all field configurations $g,\Psi$ on a spacetime manifold $\stm$ with boundary $\bs\stm=\sm$, with formal measures $\D g$ and $\D\Psi$ and weighted by the exponential of the action which divides into a pure gravitational part $S_\gr$ and a matter part $S_{\textrm{matter}}$.

Since GR places no restriction to spacetime other than being some pseudo-Riemannian $\std$-manifold, one might consider not only the gravitational field $g$ on $\stm$ but also $\stm$ itself as variable.
In that case, according to the idea of the path integral as a sum over all possible intermediate configurations \cite{Feynman:1942va,Feynman:1948iw}, the path integral would contain an additional integral over a class of manifolds, formally
\[
\qbra O_{\sm}[g,\Psi] \qket
  = {\us{\br \stm = \sm}{\int}} \D \stm \us{\stm}{\int} \D g \D\Psi\; O_{\sm}[g,\Psi] \;\e^{\frac{\rmi}{\hbar} (S_{\gr}[g] + S_{\textrm{matter}}[g,\Psi] )}\;.
\]

There are several main directions how to define the formal observable \eqref{gr-path-integral}.
One is to treat the path integral perturbatively, as common in high energy physics. 
To this end, a partition of the gravitational action $S_\gr$ into kinetic and interaction part is only possible upon splitting the metric $g = \bar g + h$ where a background metric $\bar g$ is fixed and $h$ are small disturbances around it treated as the degrees of freedom, their quantum excitations called gravitons \cite{Dewitt:1967cs}.
Thus, eventually the perturbative treatment is quite different to the usual one in particle physics.
There, perturbations result in a restriction in the dynamics but not in the degrees of freedom.
Nevertheless, in this way pure gravity (neglecting $S_{\textrm{matter}}$) 
turns out to be non-renormalizable since at every finite order of perturbation counterterms of higher order are needed to render divergencies finite, leading to an infinite tower of counterterms \cite{Goroff:1985kg}
(though there are ideas how the structure of its Schwinger-Dyson equations could entail  renormalizability even so \cite{Kreimer:2008jm}).
Still it could be that inclusion of appropriate matter would render the theory finite.
Particular strong hopes in this direction are associated with maximal, $N=8$ supergravity which have been revived recently with the development of new efficient techniques for the calculation of amplitudes in maximally supersymmetric theories \cite{Bern:2007cx,Bern:2011kk}.

Another direction picks up the idea that gravitational and standard model interactions, extended and unified in an appropriate way, might solve the UV divergence issue.
String theory \cite{Green:1987tq,Polchinski:1998tm,Blumenhagen:2012wr} is the candidate which seems to accomplish this, replacing elementary point particles by extended objects.  Supersymmetric versions turn out to be anomaly free in $\std=10$ dimensions of target space which is thus usually understood as a prediction of string theory (though on the grounds of various dualities it is conjectured that these theories describe only certain regimes of a theory in $\std=11$ dimensions).
The spectrum of the string contains an infinite tower of string excitations, in particular a spin-2 particle which can be understood as the graviton since the corresponding part in low-energy effective actions are of the type of the Einstein-Hilbert action (with higher curvature corrections).

Yet another direction is characterized by questioning the setup of quantum gravity as a perturbative theory of metric disturbances around a fixed background metric, in this sense emphasizing the property of background independence in general relativity.
One example, still in the usual smooth metric setting, is the  idea of asymptotic safe quantum gravity where a nontrivial renormalization group fixed point may render the theory finite in a nonperturbative way \cite{\asR}.
Complementary thereto are approaches which are based on geometric degrees of freedom different but classically equivalent to the full metric $g_{\mu\nu}$ and which have support on a certain kind of discrete structures.
For this reason, their quantum states of geometry can be referred to as discrete quantum geometries, that is certain kind of combinatorial or cellular complexes with some kind of geometric data attached to their cells.
Most prominent examples of such approaches are loop quantum gravity (LQG) \cite{\lqgT,\lqgR}, spin-foam (SF) models \cite{Baez:2000kp,\sfOP}, group field theory (GFT) \cite{\gftFO,Krajewski:2012wm} and simplicial quantum gravity, either as a quantization of Regge calculus (QRC) \cite{\qrcW} or in terms of causal dynamical triangulations (CDT) \cite{\cdtAJL}.

These discrete quantum gravity approaches turn out to be related in many ways, not only with respect to their conceptual setup but also concerning the proposed dynamics.
This comes to a certain extent as a surprise since they originate from different ideas and they are distinct theory proposals from a systematic perspective.
Nevertheless, by now most of them%
\footnote{CDT plays to some extent a different role since therein discrete geometries are equilateral (up to the distinction of time-like and space-like edges) and thus effectively purely combinatorial.}
have evolved into a tightly connected web of theories, sometimes even regarded as a single theory which is then called ``loop quantum gravity" as well, though evidently in a generalized sense.
At times, this development is considered as a convergence of these theories. 
Apart from the pragmatic benefits of enabling to transfer insights and techniques from one approach into the other, the question in what sense one can speak of a genuine convergence process is interesting since theory convergence increases the degree of (epistemological) justification of theories. 
I will use the opportunity of introduction into the theoretical context of this thesis to discuss their relations and address this question in a systematic way. 
This will show that strict relations between the approaches occur mostly on the level of specific models.
On this basis I propose to understand the evolution into a web of theories rather as a crystallization process (a notion recently discussed in the philosophy of science \cite{Moulines:2010hi,Moulines:2013jn,Moulines:2013fy}) than theory convergence in a stricter sense.

A condition for the possibility of dynamical relations between the approaches concerns the compatibility of their discrete structures.
For example, covariant approaches such as SF models and GFT are originally defined most efficiently in the setting of simplicial complexes involving (dual) boundary graphs with vertices of fixed valency, while LQG, due to its formulation as a canonical quantization of continuum classical gravity, entails kinematical states technically defined on arbitrary embedded closed graphs.
While it has been shown recently that SF models as a proposal for LQG dynamics can be extended to combinatorially cover the whole kinematical LQG state space \cite{\KKL,\KLP}, it was not clear how such an extension could be possible in GFT.
Since GFT is a candidate both for a completion of SF models and a second quantized reformulation of LQG \cite{Oriti:2014wf}, this is an important question.

Here I will present how a generalization of GFT to combinatorially cover the kinematical LQG space can be accomplished in different ways.
First, there are neither mathematical nor conceptual obstructions to extend GFT in exactly the same way as has been done for SF models \cite{\KKL,\KLP} introducing multiple quantum fields to account for the combinatorial variety. However, such a theory is not expected to be of practical use since the large (and possibly infinite) number of fields and interactions would be hard to control (the same being true already in the SF case).
The particular advantage of GFT comes to the forefront in a second proposal: using field theoretical techniques one can reproduce the same state space and amplitudes with a standard simplicial GFT equipped with a dual weighting which is much more manageable and controllable.

\

Apart from such questions of the precise definition of kinematics and dynamics, in all the discrete quantum gravity approaches the major challenge is to find a relation to the continuum spacetime geometries of classical general relativity.
That is, one has to show that the latter emerge from the fundamental discrete quantum geometries of the theory in some approximation. This emergence has to be expressed in terms of suitable geometry observables, both classical and quantum, that should indicate that the desired features of smooth spacetimes are recovered.
The identification of at the same time meaningful and tractable  geometric observables is one of the main tasks in this endeavour.
In fact, it is a precondition for extracting physics from such quantum gravity proposals. 

Effective-dimension observables provide important information about the geometric properties of quantum states of space and spacetime histories in quantum gravity.
In particular, the spectral dimension $\ds$, which depends on the spectral properties of a geo\-metry through its definition as the scaling of the heat-kernel trace, has attracted special attention due to the observation of a dimensional flow (\ie the change of spacetime dimension across a range of scales \cite{Hooft:1993vl,Carlip:2009cy,Calcagni:2009fh}) in various approaches including the mentioned causal dynamical triangulations and the functional renormalization-group approach of asymptotic safety, as well as Ho\v{r}ava--Lifshitz gravity \cite{\hlH} among others. 

In all these approaches, the spectral dimension of spacetime exhibits a scale dependence itself, flowing from the topological spacetime dimension $\std$ in the infrared (IR) to $\Ds\simeq 2$ in the ultraviolet (UV) \cite{\asfractal,\hlfractal,Calcagni:2013jx,\cdtfractal,Benedetti:2009bi}.%
\footnote{More recent calculations in  CDT hint at $\Ds\simeq 3/2$ though \cite{Coumbe:2015bq}.} 
For smooth geometries, modified dispersion relations provide an obvious reason for this behaviour\cite{\asfractal,\hlfractal,\ncfractal,\snfractal}.
In contrast, in the case of discrete calculations as in the CDT approach \cite{\cdtfractal,Benedetti:2009bi, Coumbe:2015bq} the dimensional flow remains to be better understood. 
Causal dynamical triangulations provide a definition of the continuum path integral for quantum gravity via a regularization in terms of a sum over simplicial complexes 
weighted by the Regge action. While it is more difficult to identify the underlying reason for the dimensional flow in this context, such a flow is obtained in a very direct manner from the evaluation of the heat trace as a quantum geometric observable inside the CDT partition function. 

Here I choose a very similar direct approach 
for the evaluation of effective-dimension observables of quantum states of geometry as they appear in LQG, SF models and GFT.
Such quantum states can be expanded as superpositions of spin-network states, which are graphs labelled by algebraic data from the representation theory of $SU(2)$. 
Accordingly, there is an interplay between two types of objects and their corresponding discreteness: a combinatorial discreteness due to the underlying graph, as well as an algebraic discreteness due to the fact that the labels are half-integers corresponding to $SU(2)$ irreducible representations. Quantum effects in the evaluation of observables are thus to be expected from both these sources, in particular from superpositions of algebraic data as well as of graphs. 

Perhaps surprisingly, superpositions of quantum states supported on different graphs and complexes have not been considered much in the LQG literature so far. Instead, most analyses have involved only states based on one and the same complex. A first example of states based on superpositions of combinatorial structures are the simple condensate states with a homogeneous cosmology interpretation introduced recently in the GFT context \cite{\gftcondensate} and their generalization to states based on connected complexes in \cite{Oriti:2015tl}.

Technically, a particular challenge stems from the fact that the dimension observables are only meaningful on discrete quantum geometries of sufficiently large combinatorial size.
This challenge can be met in different ways.
On the one hand, I use numerics to directly evaluate the spectral dimension in the setup of semiclassical kinematical LQG states in $\std=2+1$ dimensions were such computations are still feasible.
On the other hand, I set up a model of a special class of superpositions based on a number of reasonable assumptions; this allows to use analytic solutions of the effective dimensions on the lattice for calculations on superpositions over a very large  number  (up to $10^6$) of complexes.

A fairly general results of these calculations is that the effective-dimension observables are considerably more sensitive to the combinatorial structures underlying the states than to the algebraic data associated with them. 
Furthermore, semiclassical states such as coherent LQG states on a single graph approximate well 
the discrete geometries they are peaking on.
More interesting quantum effects are found for large superpositions.
In particular, I have found strong evidence for a dimensional flow for superpositions with power function coefficients, where a value of spectral dimension in the UV below the topological dimension $\sd$ is found to depend on the exponent in the superposition coefficients.
Similarly, based on findings of analytic solutions for the walk dimension and Hausdorff dimension of lattice geometries, I do not
find any special properties in these observables for superpositions as compared to states defined on fixed complexes.

\

The outline of this thesis is the following. 
I will start in chapter \ref{ch:convergence} with an introduction to the discrete quantum gravity proposals guided by the question of convergence and their intertheoretical relations.
To this end, I will first present the five proposals of quantum Regge calculus, causal dynamical triangulations, loop quantum gravity, spin-foam models and group field theory in a brief but systematic manner as distinct theories. Then I will discuss in more detail their differences, conceptual similarities and dynamical relations. This allows finally to argue that their example is an instance of a crystallization process rather than theory convergence.

The purpose of chapter \ref{ch:discrete-spacetime} is to present an extensive and exhaustive treatment of the discrete space and spacetime structures entering in discrete quantum gravity. 
I will show that the combinatorial and differential structure of the relevant classes of complexes can be defined without any reference to topological or analytic structure in the first place, thus proposing combinatorial complexes as the most general common objects in all the approaches.
In particular, I will review in detail the 2-complex structure underlying the amplitudes in SF models and GFT, coined spin-foam molecules because of their structure as bondings of atomic building blocks.
Furthermore, I will introduce a differential calculus on combinatorial complexes and discuss the properties of the resulting Laplacian, necessary for the definition of effective dimensions as quantum observables.

Chapter \ref{ch:dqg} is devoted to the extension of GFT to cover discrete quantum geometries of the most general combinatorial type in its state space. 
To this end I will first give a more detailed introduction into SF models and GFT as proposals for quantum gravity with a particular focus on the way the 2-complex structure of spin-foam molecules arises.
I will then show how the most general spin-foam molecules give rise to GFT with multiple fields in a straightforward fashion. 
Finally I will introduce the dual-weighting mechanism  and demonstrate how it can be implemented on a standard simplicial GFT to reproduce the same amplitudes for general spin-foam molecules as the multi-field GFT. Furthermore this will open up the possibility to define new models which have arguably a better geometric interpretation.

Eventually, in chapter \ref{ch:dimensions}, I will come to the main topic of the thesis of effective dimensions. I will show how their definitions on smooth geometries naturally extend to discrete geometries based on the discrete calculus, and how one can define them as observables for quantum states of geometry.
Preliminary to the calculation of quantum states, I will systematically investigate effects of topology, geometry and discreteness of classical geometries in the effective-dimension quantities.
On these grounds, the results for the spectral dimension of numerical calculations of semiclassical, coherent states can be interpreted to agree well with their classical counterparts and exhibit only minor quantum corrections.
Significant quantum effects consisting of a dimensional flow are then presented for large superposition states. In a special case, such a discrete quantum geometry can be understood as fractal.

\ 

\begin{center}
***
\end{center}

\

This thesis is based on the publications \cite{\COTa,\COTb,\COTc,\ORT}. 
However, the analysis of relations between discrete quantum gravity approaches, chapter \ref{ch:convergence}, as well as the various concepts of combinatorial complexes in \sec{combinatorial-complexes} and some background material in \sec{dimension-definition} and  \sec{classical} are new.


\chapter{Convergence in quantum gravity?}\label{ch:convergence}

It is frequently stated in introductions to papers on results in quantum gravity approaches such as loop quantum gravity, spin-foam models, simplicial quantum gravity or group field theories that there is convergence between these theories.
Such a relation between research programs is certainly interesting for practical purposes allowing to transfer results and methods from one approach to the other.
But there is also another interesting, more philosophical aspect: if there is indeed convergence in one or the other direction, this might raise the epistemological value of the emerging theory. 
This is of particular interest in the case of quantum gravity where the standard criterion of predictivity for theory assessment does not apply due to the lack of observational accessibility of the quantum gravity regime.
Other criteria such as uniformity and coherence become then central for the evaluation of a theory's epistemological status, \ie the question on which grounds one is justified in one's conviction and belief in a theory.

In this chapter I want to use the opportunity to provide the research context for this thesis to address the question in what sense one can talk about convergence in the mentioned quantum gravity proposals.
That is, this chapter is meant to serve a twofold purpose. 
First of all, I will introduce the theoretical background for the presentation of results 
in later chapters introducing the relevant quantum gravity theories.
On the other hand, I will do so from a particular perspective with the question of theory convergence as a recurrent theme.

As a consequence, the presentation of theories will be slightly different from what one is used to. 
To provide the basis for an analysis of intertheoretical relations between proposals,
I will introduce them in a particular manner as common in the philosophy of science. 
This might seem a little artificial both in form and in choice of theory definitions. 
But it sets the stage to address then the questions about the relation between these theories:
what are important relations, what kind of relations are these, and finally, do they give rise to an interpretation of theory convergence?
Based on the result that the various equivalence and embedding relations are mostly between particular models and not the general theories as a whole, the answer will be that the notion of a crystallization process is more appropriate to account for the development of these quantum gravity theories.

Accordingly, this chapter is naturally divided into two parts:
in \sec{qg-theories} I will explain the conceptual framework to then present five distinct theories which are quantum Regge calculus, causal dynamical triangulations, loop quantum gravity, the theory of spin-foam models and group field theories.
This enables then, in \sec{relations}, the discussion of their relations, both on a conceptual, kinematical as well as dynamical level, and how these can be understood from the perspective of theory convergence.
Finally I will conclude with a brief remark on the epistemological relevance of this investigation.


\section{Proposals for a quantum theory of gravity \label{sec:qg-theories}}

As a basis for an analysis of the relations between quantum gravity approaches I will present the relevant theory proposals in a systematic manner.
That is, I will explicate the theories in a semiformal way motivated by a formal scheme of logical reconstructions according to a model-theoretic account of theories, called \emph{structuralism} in philosophy of science \cite{\BMS,Sneed:1971tc,Moulines:2002ww}.
An important aspect thereby is to consider the approaches as scientific, \ie empirical theories in contrast to mere formal, mathematical objects.
For such an investigation it is not necessary that the proposals already provide complete theories of quantum gravity.

By now there is rather a continuum of theories under investigation and a particular challenge here is to differentiate between distinct approaches.
This can be regarded as an attempt to distinguish a certain historical version of each theory. 
In the end, theory convergence and theory crystallization are processes demanding a diachronic view on the theories.
However, a detailed historical study is beyond the scope of this work
which is why I will focus on a systematic presentation. 

Let me  briefly sketch the most important aspects of the structure of a scientific theory. 
These are components appearing in every more involved assessment of empirical theories in the philosophy of science; 
however, they are explicated in a particularly obvious and systematic via using the structuralist conception. 

Structuralism is a set theoretic account to theory reconstruction, based on structures in the sense of model theory \cite{\BMS}.
This means that a theory is not understood as a class of statements, but as structures defined by predicates.
Therefore, by definition it is independent of language.
The crucial point here is that these structures should be understood as actual physical systems described by the theory.
That is, such a structure usually consists of physical objects together with their relations and associated quantities, possibly even depending on a specific gauge.
I will base my analysis on this formalism in a slightly informal, but more intuitive denomination in the following. 

According to a structuralist reconstruction, more precisely, a scientific theory T consists of four essential parts which are classes of structures (\fig{structuralism}):
\begin{itemize}[leftmargin=1.5cm]
\item [$(\mp^\textsc{t})$] The conceptual framework: elements of $\mp^\textsc{t}$ consist of all objects of the theory together with their typification and characterization.
Loosely speaking, from a physicist's point of view, they capture the kinematics of a theory, hence the intuitive name here (in the philosophical literature they are referred to as potential models).
\item [$(\ma^\textsc{t})$] The actual physical systems: 
the subclass $\ma^\textsc{t}\In\mp^\textsc{t}$ consists of those structures furthermore obeying the physical axioms and laws. 
While $\mp^\textsc{t}$ is the set of potential physical systems in the sense that they have the right conceptual structure, $\ma^\textsc{t}\In\mp^\textsc{t}$ is the subset of those systems which follow the dynamics of the theory T.

\item [$(\mpp^\textsc{t})$] The observable part: a restriction $\res:\mp\lora\mpp$ can be defined in terms of a distinction of objects in $\mp^\textsc{t}$ into theoretical and non-theoretical ones, ``projecting out" the theoretical components in the structures $\mp$.
For empirical theories this is a crucial part, though the distinction is very subtle and turns out to be possible only with reference to the theory T itself: roughly speaking, any quantity whose measurement presupposes the laws of the theory is theoretical with respect to the theory T \cite{Balzer:1980,\BMS}. 
\item [$(\ia^\textsc{t})$] The  specification of a class of physical systems $\ia$ which the theory is intended to be applied to, in this way providing the physical interpretation of its concepts.
\end{itemize}

Structuralism provides an even finer resolution of the inner building of theories.
For example, there is furthermore a concept of \emph{constraints} relating structures of a theory, \emph{links} relating structures of different theories and \emph{admissible blurs} accounting for the accuracy of theoretical descriptions \cite{\BMS}.
Moreover, there is a finer notion of an empirical theory as a tree-like net of so-called \emph{theory elements} of different degree of specialization \cite{\BMS}.
In the following, I will explicate in detail only the roots of such trees, the \emph{basic} theory elements, sometimes also called theory cores.
More specialized theory-elements are what one usually calls a ``model'' in physics and I will use this denomination when discussing specialized theory elements in an informal way.
For the scope of this thesis, a semiformal explication of the basic four parts $\mp$, $\ma$, $\mpp$ and $\ia$ is enough of formalities and, when relevant, I will discuss the other aspects in an intuitive way.

\begin{figure}
\centering
\begin{tikzpicture}
\filldraw [venn1]	(0,0) ellipse (1.5 and .25); 
\draw 		(0,0) ellipse (3 and .75); 
\filldraw [venn1]	(0,3) ellipse (2 and .5); 
\draw 		(0,3) ellipse (5 and 1); 
\path (2.5,3.2) node{ $\ma{}$}; 
\filldraw[venn2] (-1,-.2) ellipse (.3 and .5); 
\draw 		(0,3) ellipse (5 and 1);   
\draw 		(-1.5,0) -- (-2,3); 
\draw		(1.5,0) -- (2,3);
\draw [->]		(6,3) -- node[auto] {$\res$} (6,0);
\path (-3.5,3) node{ $\mp{}$}; 
\path (2,0) node{ $\mpp{}$}; 
\path (-2,-0.1) node{ $\ia{}$}; 
\end{tikzpicture}
\caption{A sketch of the basic components of a theory according to the structuralist view on scientific theories: on the upper level, a subclass $\ma\In\mp$ of physical systems is specified by the axioms and laws of a theory; restriction $\res$ to the non-theoretical, ``observable" domain allows to compare the empirical content of the theory with the physical systems $\ia$ which the theory is intended to be applied to.
}
\label{fig:structuralism}
\end{figure}
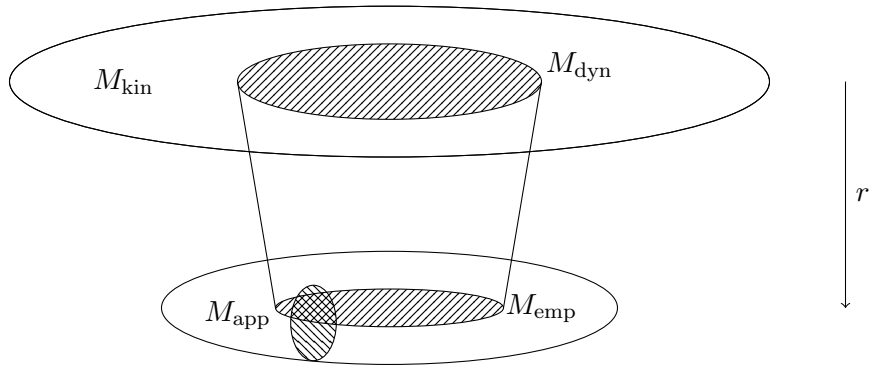

In this chapter I will use the scheme to, at the same time, introduce the theories in a more physicist's manner but also provide a model theoretic account for them in a semiformal way.
This is possible for the following reason. 
On the one hand, the definition of the logical reconstruction of a theory should capture all relevant information also from a physics' point of view. 
On the other hand, 
an informal treatment of logical reconstruction can nevertheless already capture most of the formalities lurking behind.

I will illustrate this account on scientific theories briefly with the example of general relativity.
This is also helpful in setting the stage for the discussion of candidates for a quantum theory of gravity.

\setcounter{defin}{-1}
\tdefin[]{\label{theory:gr}
General relativity (GR)
\begin{itemize}[leftmargin=1.5cm]
\item [$\mp^{\gr}:$] A structure of GR, \ie a potential physical system of GR,  consists of a representative $(\stm, g_{\mu\nu})$ of a Lorentzian geometry, that is a smooth 4-manifold $\stm$ and a pseudo-Riemannian metric $g_{\mu\nu}$ thereon%
\footnote{
In the structuralist reconstruction, single structures in $\mp^\gr$ correspond to specific metrics.
Diffeomorphism symmetry must be described by further constraints on the class of structures \cite{Bartelborth:1993vi} which I will further explain below.
},
together with various derived quantities such as a Levi-Civita connection, notions of curvature (Riemann tensor, Ricci tensor $R_{\mu\nu}$, Ricci scalar $R$), geodesics and their length.
Matter is described by a stress-energy tensor $T_{\mu\nu}$.
\item [$\ma^\gr:$] Actual physical systems are those obeying Einstein's field equations 
\[\label{eeq}
R_{\mu\nu} - \frac 1 2 R g_{\mu\nu} + \Lambda g_{\mu\nu} = 8\pi \gn T_{\mu\nu}\;
\]
where the strength of the coupling between geometry and matter is governed by the gravitational constant $\gn$ and $\Lambda$ is the cosmological constant.
Alternatively, as for any classical system, other formulations of the dynamics can be used such a Lagrangian, Hamiltonian, symplectic or least action formulation (among others).
The least action formulation, for example, is given by the generic least action principle $\delta S = 0$ and the specific action $S=S_{\textsc{eh}}+S_{\text{matter}}$ of the theory, \ie the Einstein-Hilbert action
\[\label{einstein-hilbert}
S_{\textsc{eh}} = \frac1{16\pi\gn} \int_\stm \d x^4 \sqrt{g} (R-2\Lambda) \;,
\]
where $\sqrt g$ denotes the square root of the metric determinant's modulus, 
together with some matter action $S_{\text{matter}}$.

\item [$\mpp^\gr:$] The basic notions of GR, $\stm$ and $g_{\mu\nu}$, are GR-theoretical, \ie their determination presupposes the field equations \eqref{eeq}.
Only some  derived geometric quantities are GR-non-theoretical, most importantly geodesic distances and trajectories of test particles and light on spacetime geometry.
\item [$\ia^\gr:$] As the fundamental classical theory of gravitational interaction, GR is intended to apply to all phenomena where gravitation plays a role. 
For the main part, this concerns physics on large length scales, from celestial mechanics (Mercury perihelion, gravitational lensing etc.) to cosmology. 
A phenomenon where the peculiarities of GR become eminently obvious are black holes.
Note that in all these instances indeed the observations concern only the motion of light and particles, not the gravitational metric field itself.
\end{itemize}
}
A remark is in order concerning the central notion of \emph{diffeomorphism invariance}. 
This requires an explication in terms of the mentioned notion of constraints in structuralism (not to be confused with constraints in the physics' sense).
More precisely, an equi\-valence constraint \cite[II.2]{\BMS} identifies all structures in $\mp^\gr$ which are equivalent under diffeomorphism symmetry \cite{Bartelborth:1993vi}.
In the discrete quantum gravity approaches, diffeomorphism invariance is a central notion and implemented more explicitly such that these types of constraints will not be needed in their reconstruction.

Another important property of GR which might be overlooked on first sight in the theory explication is \emph{background independence}.
While in any other physics theory space or spacetime appear as a background structure, in GR the geometry described by the metric $g_{\mu\nu}$ on a given smooth spacetime manifold $\stm$ is subject to the dynamics described by the field equations \eqref{eeq}.

\

All the five approaches to be discussed in the following share a modest strategy for combining GR with quantum theory: 
they aim to define a quantum theory of gravity using some well established framework of quantum theory which is specified by observables, states and dynamics capturing the content of GR, that is, which describe the geometric degrees of freedom of GR and unveil their dynamics in a semiclassical limit. 
Moreover, in all approaches GR's crucial property of background independence is emphasized, in some cases extended even further in the sense that also aspects of the manifold $\stm$ are considered as dynamical.

The description of $\ia$, \ie the physical systems a theory is meant to be applied to, is a little subtle in the case of quantum gravity. 
Usually, a physics theory has a very precise set of $\ia$ which are the concrete experiments and observations the theory is meant to apply to, up to a certain degree of accuracy.
For quantum gravity there are in our days only ideas and plans what possible phenomena the theory might be relevant for, \eg signatures of early universe cosmology in the cosmic microwave background,  or theoretically expected phenomena related to black holes such as Hawking radiation and their statistical physics \cite{Burgess:2004tw}, among others \cite{AmelinoCamelia:2013ct}.
All approaches to quantum gravity in principle intend to apply to all of these ideas;
but 
strictly speaking
there are no explicit experimental or observational data which clearly indicate a need for quantum-gravitational explanation. 
 With these caveats,  I will sketch under the label of $\ia$ in the following only examples of \emph{possible} applications which an approach already has some more detailed results on.


\subsection{Simplicial quantum gravity I: Quantum Regge calculus}\label{sec:qrc}

Simplicial quantum gravity can be taken as a generic term for theories which are based on triangulations $\T$, that is simplicial decompositions, of the Riemannian spacetime $\std$-manifold $\stm$ and which, more particularly, use an action as introduced by Regge \cite{Regge:1961ct}%
\footnote{Apart from the original Regge action as a function of edge lengths which I focus on here, there are various classically equivalent formulations 
(first order, area-angle variables etc. \cite{Barrett:1994ba,Barrett:1999fa,Makela:1994hm,Dittrich:2008hg,Caselle:1989cd,Gionti:2005gi}; see appendix \ref{sec:classical-expressions}). 
Accordingly, approaches to quantum gravity related to those formulations could also be seen as versions of simplicial quantum gravity.
}:
\[\label{regge-action}
S_{\textrm{Regge}} [l_{ij}^2;\T] = \frac1 {8\pi \gn} \sum_{h\in\Tp{\std-2}} \left(V_h
 \delta_h +   \Lambda V_h^{(\std)} \right)
\]
This is simply a discretization of the Einstein-Hilbert action of GR depending on a simplicial decomposition $\T$ as well as length variables $l_{ij}$ assigned to all edges $(ij)\in\Tp1$. 
Taking all $\std$-simplices as flat inside and on their bounding $(\std-1)$-simplices, this defines a piecewise flat manifold.
The Regge action is then a sum over $(\std-2)$-simplices $h\in\Tp{\std-2}$ with volumes $V_h$, in terms of which the spacetime measure is given by $\std$-volumes $V_h^{(\std)}$. 
This yields directly the term for the cosmological constant $\Lambda$.
The discretization of the Ricci scalar part of the action is the reason for favouring $(\std-2)$-simplices to sum over since curvature is concentrated on the these (thus often called \emph{hinges}).
It is expressed in terms of deficit angles 
\[
\delta_h = \sum_{\s > h} (2\pi - \theta_{h,\s})
\] 
where $\theta_{h,\s}$ is the angle between the two $(\std-1)$-simplices in the $\std$-simplex $\s$ meeting at the hinge $h$.

All the quantities $V_h$, $V_h^{(\std)}$ and $\delta_h$ in the action \eqref{regge-action} are understood as functions of the squared edge length variables $l_{ij}^2$ on the triangulation $\T$. 
A particular advantage of these variables is that 
there is no redundant coordinate dependence.


\ 

The idea of \emph{Quantum Regge calculus} is to understand the variables $l_{ij}^2$ on an appropriate triangulation as the relevant geometric degrees of freedom and use the Regge action to define dynamics in terms of a path integral \cite{\qrcW}:

\tdefin[]{\label{qrc-theory}
Quantum Regge calculus (QRC)
\begin{itemize}[leftmargin=1.5cm]
\item [$\mp^\qrc:$] A structure of QRC consists of piecewise flat Riemannian geometries based on a  (sufficiently fine%
)
fixed triangulation $\T$ of a spacetime 4-manifold $\stm$ specified by edge length variables $l_{ij}^2$ for all $(ij)\in\Tp1$,
as well as a measure $\mu(l_{ij}^2)$ on the space of edge length configurations.
Accordingly, states and observables are functions of these variables on 3-dimensional slices of $\T$.

\item [$\ma^\qrc:$] Dynamics is defined by the path integral 
\[\label{qrc-path-integral}
Z_\qrc (\T) = \int \d\mu(l_{ij}^2) \, \e^{-\frac1 \h S_\regge[l_{ij}^2;\T]}\;,
\]
\ie the expectation value of an observable $O(l_{ij}^2)$ is evaluated by insertion in the path integral.

\item [$\mpp^\qrc:$] 
As usual in quantum theory, only correlations and expectation values of observables are non-theoretical with respect to QRC.
The triangulation $\T$ and edge length configurations $\{l_{ij}^2\}$ thereon are QRC-theoretical.

\item [$\ia^\qrc:$] There is a range of applications carried out to a certain degree. 
Among others, these include analytic calculations of the graviton propagator in the weak field limit via perturbations around a highly symmetric fixed background \cite{Rocek:1984de}, various numerical simulations to explore the phase structure of the theory \cite{Hamber:1992jo,Hamber:2000ew,Riedler:1999bq}, and cosmological models \cite{Collins:1973hw}.
\end{itemize}
}
The precise form of the measure is a central issue in QRC and an on-going topic of debate \cite{Williams:2007up}.
For example, a class of models in QRC is specified by 
\[\label{qrc-measure}
\d\mu(l_{ij}^2) = \prod_{\s\in\Tp4} (V_\s)^\alpha \prod_{(ij)\in\Tp1} \d l_{ij}^2 \;\Theta(l_{ij}^2)	
\]
where $V_\s$ are 4-simplex volumes and $\Theta$ is a step function implementing all possible triangle inequalities and their higher dimensional generalizations. 
The rational for this measure is that it is the discrete version of the DeWitt measure (where $\alpha \in\R$ is an additional parameter) \cite{Hamber:1999cf}.

\subsection{Simplicial quantum gravity II: Causal dynamical triangulations}\label{sec:cdt}

A complementary way to define a quantum partition function based on the Regge action is to sum over a class of triangulations with fixed edge lengths $l_{ij}^2=a^2$, thus coined \emph{dynamical triangulations} \cite{\cdtAJL}.
To include the causal structure of Lorentzian spacetimes, triangulations of foliated spacetime are chosen and a distinction is made between the fixed length of time-like and space-like edges, $\at$ and $\as$, related by a parameter $\alpha$ such that $\at^2 = -\alpha \as^2$. 
It is then argued for an analytic continuation in $\alpha$ to obtain the Lorentzian theory from a Euclidean version in terms of a ``Wick rotation" 
\cite{Ambjorn:2001kb}
\[\label{cdt-wick}
\i S_{\textrm{Regge}}[\at^2=-\alpha \as^2, \as^2;\T] \mapsto S_{\textrm{Regge}}[\at^2=\alpha\as^2,\as^2;\T] \;.
\]
The main result of CDT is then that, contrary to purely Euclidean dynamical triangulations, there is a regime in the phase space of parameters $\gn, \Lambda, \alpha$ which can be understood as describing extended continuum geometries according to their properties as captured by various observables \cite{Ambjorn:2005fh}.

The theory of CDT, as a candidate for quantum gravity, can thus be explicated in the following way:
\tdefin[]{\label{cdt-theory}
Causal dynamical triangulations (CDT)
\begin{itemize}[leftmargin=1.5cm]
\item [$\mp^\cdt:$] A structure of CDT consists of an ensemble of abstract triangulations $\T$ with fixed space-like and time-like edge lengths $\as^2$ and $\at^2 = -\alpha \as^2$ of a rigidly foliated spacetime $\stm=[0,1]\times\sm$ with fixed spatial topology $\sm$. 
Three-dimensional slices in the triangulations are understood as spatial states \cite{Ambjorn:2012tj}.

\item [$\ma^\cdt:$] Correlations between spatial states are defined by a path integral over all interpolating triangulations, defined by the partition function 
\[\label{cdt-state-sum}
Z_\cdt 
=\lim_{\underset{\np4\ra\infty}{\as\ra0}}\sum_{\T_{\np4}}\frac{1}{\sym {\T_{\np4}}} \e^{-\frac1 \h S_{\textrm{Regge}}[\at^2=\alpha\as^2,\as^2;\T_{\np4}]}\;,
\]
Nevertheless, according to the calculations (Monte-Carlo simulations) explicitly done, for a given accuracy an ensemble of triangulations of a large enough fixed size (number of 4-simplices $\np4$) is sufficient \cite{Ambjorn:2012vc}.

\item [$\mpp^\cdt:$] In this sense, the triangulations $\T$ are clearly CDT-theoretical. 
Non-theoretical quantities with respect to CDT are geometric observables such as the volume or effective dimensions of spatial slices or the volume of the whole spacetime.

\item [$\ia^\cdt:$] The continuum phase of the CDT state sum \eqref{cdt-state-sum} shows an evolution of spatial slices which can be interpreted cosmologically as a de-Sitter universe \cite{Ambjorn:2005fh}.
\end{itemize}
}

\

Note again that CDT is different from QRC since the dynamical degrees of freedom are triangulations of fixed edge lengths while in QRC it is precisely the other way round with the triangulation being fixed and the edge lengths dynamical. 

One could understand simplicial quantum gravity even in a more general sense where both, triangulations and edge length variables, are dynamical; furthermore, there are also extensions of Regge calculus to other variables such as areas and angles \cite{Barrett:1994ba,Barrett:1999fa,Makela:1994hm,Dittrich:2008hg}, among others \cite{Caselle:1989cd,Gionti:2005gi}.
Nevertheless, in a more specific sense QRC and CDT are understood as explicated here and I will stick to these particular theory concepts.


\subsection{Loop quantum gravity}\label{sec:lqg}

Loop quantum gravity, understood in the traditional sense \cite{Rovelli:2004wb,\lqgT}, is a canonical (Dirac) quantization of GR as a gauge theory.
On a spacetime manifold of fixed topology $\stm=\sm\times\R$ an ADM splitting \cite{Arnowitt:1962wr} allows for a formulation of GR as a constrained Hamiltonian system.
On each spatial slice $\sm$, the constraints take a particularly convenient form in a formulation in terms of the Ashtekar-Barbero $SU(2)$ connection $A^i_a = \omega^i_a + \bi K^i_a$, a linear combination of gauge-fixed spin connection $\omega^i_a$ and extrinsic curvature $K^i_a$ governed by the Barbero-Immirzi parameter $\bi$.
Its conjugate momentum is  $E_i^a = \sqrt{\det e}\; e^a_i$, the densitized version of the inverse triad $e^a_i$, 
according to the Poisson algebra
\[
\left\{E_i^a(x),A_b^j(y)\right\} = 8\pi\bi\gn \delta^a_b \delta^j_i \delta(x,y) \;.
\]

On these grounds a sensible quantum theory is obtained by the canonical quantization description turning the Poisson algebra for the conjugate pair $(A^i_a,E_j^b)$ into an algebra of quantum operators. 
However, before quantization one extends the connection variables to the space $\Acal$ of generalized connections which associate with any curve $\curv\In\sm$ the parallel transport
\[
h_\curv[A] = \Pcal \e^{-\int_\curv A} \in SU(2)\;,
\]
and thus with any embedded closed graph $\bge\In\sm$ a set of $h_\eb[A]$ for each  $\eb\in\bge^{[1]}$, where $\bge^{[1]}$ is the set of graph edges.
Then, one can define wave functions $\qs=\qs_{\bge,f}$ on the generalized space of connections with support on such graphs $\bge$ in terms of functions $f:SU(2)^{|\bge^{[1]}|}\ra\C$ as 
\[
\qs[A] = f(h_{\eb}[A])\;,
\]
which are called cylindrical functions (with respect to the edge curves $\eb$).
On a single graph $\bge$, gauge invariant wave functions are thus simply obtained by an average over $SU(2)$ for each vertex $\vb\in\bge^{[0]}$ at which the gauge group acts on the parallel transports.
With the natural gauge invariant inner product 
\[
\qbra \qs_{\bge,f} | \qs'_{\bge,f'} \qket = \int \prod_{\eb\in\bge^{[1]} } \d h_\eb \;f^*(h_\eb) f'(h_\eb)\;
\]
in terms of the Haar measure $\d h$ on $SU(2)$, these form a Hilbert space 
\[
\hs_\bge=L^2(SU(2)^{|\bge^{[1]}|}/SU(2)^{|\bge^{[0]}|})
\]
upon completion with respect to the corresponding norm.

This construction can be extended to the space of functions cylindrical with respect to any graph $\bge\In\sm$ yielding the kinematical Hilbert space $\hkin^\lqg = L^2(\overline{\Acal},\d\mu_\textsc{al})$ upon completion, 
based on the Ashtekar-Lewandowski measure $\mu_\textsc{al}$ which is furthermore invariant under diffeomorphisms on $\sm$.
%
%
In fact, it turns out that under certain reasonable assumptions,  there is a unique representation with a gauge- and diffeomorphism invariant cyclic state of the corresponding holonomy-flux algebra $\alg_{h,E}$,
\[
\left\{{h_\eb[A]},{h_{\eb'}[A]}\right\} = 0 \, , \quad
\left\{{E_i(\surface)},{h_\eb[A]}\right\} \propto h_{\eb_1}[A] \tau_i h_{\eb_2}[A] \, , \quad 
\left\{{E_i(\surface)},{E_j(\surface)}\right\} \propto {\epsilon_{ij}}^k {E_k(\surface)}
\]
of holonomies $h_\eb[A]$ and fluxes $E(\surface)$, \ie a smeared version of densitized triads obtained from integrating over embedded surfaces $\surface\In\sm$,
which corresponds to the Hilbert space $ L^2(\overline{\Acal}, \d\mu_\textsc{al})$ \cite{Lewandowski:2006ev}
(where in the commutation relation between $E(\surface)$ and $h_\eb[A]$, the Pauli matrix $\tau_i$ splits the parallel transport $h_\eb$ on the edge $\eb$ at the point of its intersection with the surface $\surface$ into a part $h_{\eb_1}$ and $h_{\eb_2}$).


LQG can thus be explicated in the following way:
\tdefin[]{\label{lgq-theory}
Loop quantum gravity (LQG)
\begin{itemize}[leftmargin=1.5cm]
\item [$\mp^\lqg:$] A LQG structure consists of the holonomy-flux algebra $\alg_{h,E}$ based on a spacetime manifold $\R\times\sm$ with fixed spatial topology $\sm$ and the class of embedded graphs $\bge$ and surfaces $\surface$.
It furthermore contains the representation on the Hilbert space $\hkin^\lqg = L^2(\overline{\Acal}, \d\mu_\textsc{al})$ of generalized $G=SU(2)$ connections.


\item [$\ma^\lqg:$]  Physical states have to be determined via a projection $\prjct_\phys : \hkin^\lqg \ra \hphys^\lqg$ such that correlations are $\qbra \qs | \qs'\qket_\phys = \qbra \qs| \prjct_\phys \qs'\qket$.
There are various proposals how to define $\prjct_\phys$ in terms of the Hamiltonian constraint $\widehat C$. 
One idea is to find a mathematical definition for the formal expression \cite{Rovelli:2004wb}
\[\label{lqg-projector}
\prjct_\phys = \int \D N \e^{\i \int \d x^3 N_x \widehat{C_x} }\;.
\]
Another strategy (corresponding to a slightly altered constraint algebra) 
 is to consider directly the constraint 
\[\label{master-constraint}
\widehat{\mathbf{M}} = \int_\sm \d^3 x \frac{\widehat{C_x}^2}{\sqrt{\det(q_x)}}
\]
consisting of an average over all the infinitely many constraints in space points $x\in\Sigma$, thus called the ``master" constraint \cite{Thiemann:2006cx,Thiemann:2006ir,Thiemann:2007un}.
\item [$\mpp^\lqg:$] 
 The quantities that are non-theoretical with respect to the theory are the spectra of observables. 
Here, important examples are area and volume operators which have discrete spectra
\cite{Rovelli:1995gq}.
\footnote{But note that, strictly speaking, their usual construction \cite{Rovelli:1995gq,Ashtekar:1997bn,Ashtekar:1997we}, is not yet gauge invariant, \ie the operators do not commute with the gauge constraint. Thus it is not obvious whether a gauge-invariant version would have discrete spectrum as well \cite{Dittrich:2009gi}.
}
They form a maximal set of commuting operators. The area operators are diagonalized by spin-network states $|\bge,\rep{\eb},\intw{\vb}\qket$ labelled by $G$-representations $\rep{\eb}$ on the graph edges and corresponding intertwiners $\intw{\vb}$ at the graph vertices.

\item [$\ia^\lqg:$] There is hardly any application of the full theory (given that explicit nontrivial solutions of the dynamics are not known yet). On the other hand, whole research fields have developed for loop-quantized cosmology (LQC)\cite{Bojowald:2008wn,Bojowald:2011dn,Ashtekar:2011fh,Banerjee:2012fn,Ashtekar:2013fk} and black holes \cite{BarberoG:2011ur,BarberoG:2015ws}.  They are based on the quantization scheme of LQG applied to symmetry-reduced sectors of GR. 
There are a few exceptions, for example an attempt to identify a black hole horizon in the full theory \cite{Sahlmann:2011ch}.
\end{itemize}
}


\subsection{Spin-foam models}\label{sec:sfm}

The theory of spin-foam models is a covariant approach to quantum gravity usually (but not necessarily%
\footnote{
There are also attempts to derive spin-foam models without reference to BF theory (\eg directly from the Palatini-Holst action \cite{Baratin:2012gc}).
})
based on a classically  equivalent gauge-theory formulation as a constrained topological theory \cite{Baez:2000kp,Perez:2003wk,Rovelli:2011tk,Perez:2013uz}.
Topological quantum field theories \cite{Atiyah:1988id} are of interest for quantum gravity since they are naturally background independent, \ie observables do not depend on a (fixed) metric structure of the underlying manifold.
With the Plebanski action \cite{Plebanski:1977zz} and variants of it there are classically equivalent formulations of GR as a constrained topological BF theory \cite{Blau:1991ji}.
A strategy of spin-foam models is thus to use lattice gauge theory techniques to extend the results of quantum BF theory to gravity (see \sec{lgt-gravity} for a detailed discussion).

The fact that models of such a theory are called ``spin-foam" models reflects two of their essential features.
One is that the path integral over connections of the gauge theory can be transformed into a sum over representation (``spin") labels (though a formulation in terms of group or algebra elements is equally well possible).
The other feature is that the definition of the quantum partition function relies on a cellular decomposition $|\cp|$. However, all amplitudes depend only on the 2-skeleton of the dual $\cps$ to the underlying combinatorial complex $\cp$ eventually (\cf \sec{combinatorial-complexes}). Since such a 2-complex consists only of vertices, edges and faces, one could also call it a ``foam".

While the use of a cellular decomposition $|\cp|$ is very natural in topological QFT as the complex $\cp$ captures the  topology and thus all relevant information of the manifold, this is no longer true for the constrained, gravitational models.
It is an open question how one should remove this dependence, most prominent proposals being a refinement description or a sum over cellular decompositions (see \sec{gft}).
In the following theory explication I will thus leave this question open, indicating both possibilities:


\newcommand{\cfg}{{\chi}}

\tdefin[]{\label{sf-theory}
Spin-foam (SF) theory
\begin{itemize}[leftmargin=1.5cm]
\item [$\mp^\sfm:$] 
A structure of SF consists of a cellular decomposition $|\cp|$ of a topological spacetime manifold $\stm$ with combinatorial complex $\cp$ (or a class of such complexes) and an ensemble of configurations $\cfg$ of group elements of a Lie group $G$, or equivalently of representations or Lie algebra elements of $G$, associated to faces and edges of $\cps$.
A Hilbert space $\hkin^\sfm$ of states is given by functions of the respective data on the boundary $\bs\cp$ (square integrable with respect to an appropriate measure).

\item [$\ma^\sfm:$]  A physical inner product defining correlations between the states $\qs,\qs'\in\hkin$ is given by a path integral over all configurations on $\cp\setminus\bs\cp$, determined by a specification of amplitudes $A_f,A_e,A_v$ associated with faces, edges and vertices of the dual complex $\cps$ which depend on the configuration variables of neighbouring faces $f'$ and edges $e'$ such that
\[\label{sf-path-integral}
\qbra \qs | \qs' \qket_\phys = Z^\cp_\sfm(\qs,\qs')
= \sum_\cfg 
\prod_{f\in\cpsp2} A_f \prod_{e\in\cpsp1}A_e  \prod_{v\in\cpsp0} A_v \;. 
\]
(If the theory is defined by a summing description, the correlations are defined as a sum over all complexes in the class compatible with the boundary states,
\[\label{sf-path-integral-sum}
\qbra \qs | \qs' \qket_\phys = Z_\sfm(\qs,\qs')
= \sum_\cp w_\cp \sum_\cfg \prod_{f\in\cpsp2} A_f \prod_{e\in\cpsp1}A_e  \prod_{v\in\cpsp0} A_v \;,
\]
determined additionally by the specification of a set of weights $w_\cp$.)

\item [$\mpp^\sfm:$] Again, only expectation values of observables (which are defined as insertion into the path integral) are non-theoretical quantities. 
Common examples are Wilson loops, and more generally gauge invariant functions with support on a graph in $\cp$ (though there may be further restrictions according to diffeomorphism invariance).

\item [$\ia^\sfm:$] So far, most explicit calculations are done for very small complexes (number of internal vertices of order one). 
As far as one can trust such investigations, there are results for correlations of metric disturbances (``graviton propagation") \cite{Rovelli:2006dy,Bianchi:2006ka,Livine:2006dm,Bianchi:2009kc} as well as applications to cosmology \cite{Bianchi:2010ej,Ashtekar:2010eh,Rovelli:2010kz,Bianchi:2011bd,Kisielowski:2013du}.
\end{itemize}
}
Note that even if the theory is not defined by a sum over complexes \eqref{sf-path-integral-sum}, a cellular decomposition $|\cp|$ is only specified for a particular structure of SF theory. 
According to structuralism, a relation between the single structures that are meant to apply to the same measurement (but possibly with different levels of accuracy), as expected for example from a refinement description, would have to be specified by the mentioned model-theoretic constraint relations.



\subsection{Group field theory}\label{sec:gft-explication}

A different theory framework for a quantum gravity candidate exists with group field theory \cite{\gftFO,Krajewski:2012wm,Baratin:2012ge,Oriti:2009ur,Oriti:2014wf}.
In the first place, a GFT is simply a standard quantum field theory with a field $\phi$ defined on several $\copies$ copies of a Lie group $G$.
More specifically, a restriction to a certain class of nonlocal interactions is furthermore understood as a defining property (for a detailed introduction see \sec{gft} below).

The interpretation of GFT as a quantum gravity theory is very different to standard QFT on Minkowski (or another fixed) spacetime.
The group manifold $G^{\times\copies}$ is not seen as a spacetime background; 
rather, the field excitations $\phi(g)$ have to be understood as the excitation of a combination of $\copies$ group variables $g=(g_1\dots g_\copies)$.
These are interpreted as geometric data
associated to a building block of space. 
Then, the nonlocality of interaction kernels allows to understand the interaction vertices accordingly as atoms of spacetime.
Feynman diagrams in the perturbative expansion of the path integral can therefore be seen as discrete spacetimes built from atoms. Indeed they are equivalent to combinatorial manifolds (explained in detail chapter \ref{ch:discrete-spacetime}, in particular \rem{simplicial}).

In the light of this interpretation, GFT can be seen as a generalization of matrix models \cite{Francesco:1995ih}.
The perturbative expansion of matrix models has the structure of 2-dimensional combinatorial manifolds and they are related to gravitational theories in two dimensions.
GFT is a generalization of these in a two-fold way.
First, as a higher-dimensional extension since any $\copies>2$ variables are possible.
Second, as a generalization from purely combinatorial (or equilateral) manifolds to discrete geometries as the field group arguments can be related to metric degrees of freedom while, accordingly, a matrix is merely a function of two integers.

The logical structure of GFT is nevertheless similar to that of standard QFT:
\tdefin[]{\label{gft-theory}
Group field theory (GFT)
\begin{itemize}[leftmargin=1.5cm]
\item [$\mp^\gft:$] 
A structure of GFT consists of an algebra of field observables $O[\phi]$ built from a field $\phi: G^{\times\copies}\ra\R$ (or $\C$), where $G$ is a Lie group, and of a representation on its Fock space $\Fcal_\gft$ of field excitations $|\qs\qket = O_\qs[\phi] \gftvr$ from the GFT ``no-space" vacuum $\gftvr$.

\item [$\ma^\gft:$]  Propagation and interaction of field excitations are determined via a kinetic term $\Kbb$ and interaction terms $\Vbb_i$ in the GFT action
\[
S_\gft[\phi] = \frac12\int [\d g]\; \phi(g_1)\;\Kbb(g_1, g_2)\;\phi(g_2) + \sum_{i\in I}\lambda_i \int [\d g]\;\Vbb_i\big(\{g_j\}_{J_i}\big)\;\prod_{j\in J_i}\phi(g_{j})\;,
\]
with index sets $I,J_i$.
This defines correlations and observables as insertions into the path integral for that action,
\[\label{gft-path-integral}
\qbra \qs | \qs' \qket = \qbra O_{\qs,\qs'}\qket
= \frac{1}{Z_{\gft}}\int \Dcal\phi\; O_{\qs,\qs'}[\phi]\;e^{-S[\phi]}
= \sum_{\fdiagram} \frac{\lambda_i^{c_i}}{\sym(\fdiagram)} Z_\gft^\fdiagram(\qs,\qs')\;,
\]
where the perturbative expansion is a sum over diagrams $\fdiagram$ with order of automorphism group $\sym(\fdiagram)$.

\item [$\mpp^\gft:$] As in ordinary QFT, the non-theoretical quantities of GFT are expectation values of observables. There is obviously no direct access to the field states nor the spacetime diagrams occuring in the perturbative sum.

\item [$\ia^\gft:$] Recently, there are first results how to obtain phenomenological models from GFT, in particular based on coherent, condensate states whose effective dynamics have a cosmological interpretation \cite{\gftcondensate}.
\end{itemize}
}

\

In this section I have sketched the logical structure of five proposals for a theory of quantum gravity: quantum Regge calculus, causal dynamical triangulations, loop quantum gravity, SF models and group field theory. 
For a systematic treatment I have used a simplified scheme based on the structuralist view on scientific theories.
This analysis sets the stage for the investigation of their intertheoretical relations and the question of theory convergence in the next \sec{relations}.
It also serves as an introduction into the theory space relevant for the remainder of this thesis more generally.


\section{Relations between the approaches \label{sec:relations}}

Now that the theories and their logical structure are introduced, I will discuss various relations between them.
This will set the ground to tackle then the question of a convergence process.
To this end, one has to distinguish  conceptual relations and relations on the dynamical level.
The former are relations between the elements in the class $\mp$ of a theory, thus they concern merely the conceptual setup of the theories.
The latter apply to the actual, dynamical structures in the class $\ma$ and allow thus to compare the (theoretical) content of the theories.

Furthermore, an important preliminary step is to show that the five theories are indeed mutually distinct. I will make this clear comparing some specific elements of the theories, their geometric degrees of freedom and their background structure in \sec{differences}.
Then I will discuss conceptual relations in \sec{conceptual-relations} and dynamical relations in \sec{dynamical-relations}.
Based on these results I can then discuss in \sec{crystallization} an interpretation of these relations.
I will suggest that they are better understood as an instance of a crystallization process than a convergence.

\subsection{Major differences \label{sec:differences}}

From the logical reconstruction sketches of QRC, CDT, LQG, SF models and GFT in \sec{qg-theories} one can already see that there are many important differences rendering the five theories distinct.
Nevertheless, in practice there is rather a continuum between these theories nowadays.
Therefore, some researcher in the field might still question the distinctness on the grounds of their own conception of the theory they are investigating. 
Thus, it is useful to argue a bit more explicitly that the five theories, as explicated above, are indeed essentially different to each other.

To demonstrate the theories' distinctness I will focus here on two aspects: Their geo\-metric degrees of freedom and their background structure.
There might be other, more formal or mathematical, aspects which are important distinctions in a logical reconstruction but which are not as important from a physical point of view.
There is for example an important distinction between canonical versus covariant quantizations and it is a major challenge in itself to demonstrate the precise relations between such different theories \cite{Curiel:2009vo}.
But from the usual physicist's perspective these are just different formulations which are anyway expected to be equivalent, if only the formal details are put the right way. 
Furthermore, such an aspect is  not specific to the theories considered here but already an interesting conceptual and mathematical topic for simpler, well established theories.
A further advantage of choosing degrees of freedom and background structure as examples to illustrate the theories' distinctness is that these are at the same time points of relation and partial convergence to be explained in the next section.

\subsubsection*{Differences in background structure}

An important distinction between candidate theories of quantum gravity concerns their background structure.
As mentioned earlier, all the approaches considered here have in common that they take as a motivation and starting point the background independence of GR.
In the classical case of GR this means that the spacetime metric is fully dynamical, in contrast for example to theories of metric disturbances around a fixed background metric.
But this does not mean that there is no background structure at all, as for example a smooth manifold is needed to define a metric on.
Thus, it is an important question how much of such background structure is fixed in a quantum theory of gravity. 

Quite in general, one can identify roughly the following hierarchy in the spacetime structure of GR (\cf \eg \cite{Isham:1994db}): 
\begin{enumerate}[labelindent=\parindent, leftmargin=*, label=(\roman*), align=left]
\item a minimal structure is a set of events, \ie spacetime points; 
\item a topology provides a notion of neighbourhood and relation between points;
\item a topological manifold of dimension $\std$ provides a local equivalence to $\R^\std$ in the sense of homeomorphisms;
\item extension to local diffeomorphism equivalence with $\R^\std$ yields a smooth manifold;
\item finally, a Lorentzian (metric) structure allows to measure distances and volumes.
\end{enumerate}
Though conceptually natural from the point of view of Lorentzian geometries, this hierarchical ordering is not mandatory: most properties do not necessarily presuppose all structure of the respective lower level.
For example, a metric structure or a causal structure can be defined directly for sets; there are even notions of a manifold on a purely combinatorial level (as I will discuss in detail in \sec{combinatorial-complexes}).
Thus, there is no total ordering of spacetime theories with respect to background structure but rather a partial ordering \cite{Isham:1994db,Huggett:2013uu}.

In these terms, the five theories under investigation (\sec{qg-theories}) differ with respect to their background structure in the following way:
\begin{itemize}

\item LQG, in the traditional setting chosen here for explication, is (by motivation) very close to GR: on the grounds of the assumption of global hyperbolicity, \ie a causal structure, there is a foliated smooth spacetime manifold $\stm=\R\times\sm$ with fixed spatial topology $\Sigma$.

\item In CDT, a similar background manifold $\stm=\R\times\sm$ is present, but only in the topological (iii) and not in the smooth (iv) sense. 
A special notion of causal structure is defined directly on triangulations of that manifold.

\item QRC is a bit more liberal on the spacetime topology since the fixed spacetime manifold $\stm$ is not necessarily of the form $\R\times\sm$ in the first place.
On the other hand, a QRC structure fixes not only the topological manifold but also an explicit triangulation $\T$ thereof. While the triangulation is fixed in a particular instance, the theory as the set of such structures allows for any triangulations, though.

\item A SF structure, in the explication chosen above (\sec{sfm}), is similar to QRC in containing a fixed cellular decomposition $|\cp|$ of the spacetime $\stm$ in one formulation \eqref{sf-path-integral}, but it differs already there as $\cp$ is not restricted to be simplicial (which indeed is necessary in QRC).
In the summing description \eqref{sf-path-integral-sum}, even a particular structure relies already on a whole class of such complexes.

\item GFT is by far the most radical proposal with respect to a background structure: there is none at all.
Only the combinatorial structure of the spatial building blocks and spacetime atoms is fixed in terms of the number $\copies$ of arguments of the group field and the combinatorial structure of the interaction vertices, respectively.
Nevertheless, these structures specify effectively a class of complexes given by the Feynman diagrams in the perturbative series.
Particular dynamics may further result in a dominant behaviour of a subclass of complexes in a given regime.
\end{itemize}

This shows clearly that the five theories are distinct already at the conceptual level of structures in $\mp$, merely because of the difference in fixed background spacetime structure. 


Related to the issue of background structure, there is a possible further diversification of the theories concerning their (expected) relation to the classical smooth spacetime geometry of GR.
All the five approaches share that their fundamental degrees of freedom are based on discrete, combinatorial objects in some way.
It is then an essential, mostly open issue how such discrete quantum geometries relate to continuum geometries. 
Possibilities are that the graphs or complexes are either mere technical or approximation tools or indeed fundamental objects. Furthermore, in the first case an appropriate refinement limit might still be needed for a proper definition while, in both cases, a further summing prescription might result in (effective) continuum geometries, possibly through a phase transition.

In the discrete quantum gravity approaches discussed here, there are various proposals how a continuum limit can be defined.
The idea that discreteness is only a technical or approximation tool is already apparent in the explication of CDT and, on the kinematical level, in LQG. 
On the contrary, the GFT explication suggests that quantum geometries have to be understood as fundamentally discrete in GFT.
In QRC and SF this issue is usually addressed less explicitly which is why I have done so also in their explication.
In particular in these cases, but also in the others, there are currently various opinions on how to tackle the relation to continuum geometries. 
Would one decide that a theory should be explicated with one or the other strategy, if different for different theories, this would present a further distinction between them. 


\subsubsection*{Differences in geometric degrees of freedom}

Complementary to the issue of fixed background structure is the question in terms of what kind of quantities the geometric degrees of freedom are described in a quantum theory of gravity.
Obviously, the less background structure is fixed the more spacetime geometry structure should be dynamical. 
However, the question of degrees of freedom goes far beyond this simple relation. 
In particular, quantum theories based on different, but classically equivalent, degrees of freedom are in general distinct. 
This is important since none of the considered approaches is based on the metric degrees of freedom of the standard formulation of GR (theory explication \ref{theory:gr}).
The differences between the theories are thus the following:
\begin{enumerate}
\item In QRC, geometric degrees of freedom are captured by the edge length variables $l_{ij}^2$ on the given triangulation $\T$. 
The advantage of such simplicial geometries is that they are manifestly coordinate-invariant.

\item Though also based on simplicial geometries, CDT differs in that the edge length are all fixed but the triangulation is not. Degrees of freedom are thus triangulations of a given topology which are equilateral (up to the factor $\alpha$ between time-like and space-like edges).

\item On the contrary, LQG (as well as SFs and GFT) are based on a description of geometry in terms of a gauge field. 
For the possibility of the canonical formulation of LQG it is crucial that this is the Ashtekar-Barbero connection $A^i_a = \omega^i_a + \bi K^i_a$. 
Wave functions are defined on the generalized space of this connection on smooth spatial slices $\sm$, \ie on the holonomies $h_\eb[A]$ on edges $\eb\in\bge$ of  embedded graphs $\bge\In\sm$.

\item The degrees of freedom in SF theory are group elements which may be interpreted as parallel transports $h_e[\A]$ of a gauge connection $\A$ on the dual edges $e\in\cpsp1$ of a spacetime decomposition $|\cp|$. 
Models can be chosen to parallel LQG degrees of freedom by choosing $\A$ compatible with the Ashtekar-Barbero connection. 
Still, this presupposes that the general differences in background structure are overcome by choosing a spacetime topology such that it is possible to relate spatial slices in $|\cp|$ with embedded graphs $\bge$ as part of LQG.
Thus, in general, degrees of freedom in SF theory are different from LQG.

\item In GFT, degrees of freedom are excitations of the group field which are interpreted as parallel transports on a spatial building block. 
Similar to the SF case, it is possible to define models which can be related to LQG, as well as to SF models. The presupposition is here that the GFT model is chosen such that boundary states and Feynman diagrams in the perturbative expansion can be related to the graphs and complexes in the respective LQG or SF model.
But even then the degrees of freedom are still different because of the genuine QFT setting of GFT.
\end{enumerate}
Thus, I have demonstrated that all the five theories differ in an essential way also with respect to their degrees of freedom.


\subsection{Conceptual relations \label{sec:conceptual-relations}}


Having made clear that the five quantum gravity proposals explicated in \sec{qg-theories} are indeed distinct theories in the first place, I will now discuss some of their intertheoretical relations relevant to the question whether there is a convergence in these theories.
To this end, I focus in this section on conceptual relations. 


A very general relation concerns the common motivation of all five approaches; as mentioned in the beginning of the chapter, they have in common a focus on (metric) background independence  and a conservative conception of quantum theory.
A natural result is that they basically agree in a general setting on the empirical level, \ie with respect to structures in $\mpp$:
the theories are expected to define quantum correlations between states of geometry on a 3-dimensional boundary as well as observables of such geometries, either on the boundary or possibly also in the spacetime bulk.
This is even more explicit in all theories apart from LQG which use a path integral formulation, summing over an ensemble of possible intermediate quantum geometries of spacetime.

A further very general similarity is that in all five approaches discrete structures such as graphs and complexes play a central role.
As already discussed in the last section, these provide mostly the spatial or spacetime background structure which quantum geometries are defined on, though in CDT the triangulations are dynamical; in GFT,  combinatorial complexes occur as Feynman diagrams in the perturbative expansion.

The classes of complexes specific to each theory can be compared with respect to various criteria:
\begin{itemize}
\item Dimension: though all theories aim to describe spacetime of $\std=4$ dimensions, this dimension is not always manifest in their discrete structures. 
In LQG, only embedded graphs, \ie one-dimensional complexes, occur in the definition of spatial ($\sd=\std-1=3$) states.
Similarly, in SF models only the two-dimensional substructure of complexes is relevant for spacetime amplitudes. Accordingly, states defined on their boundary are essentially based on graphs as well.
In GFT the situation is also similar as the structure of Feynman diagrams is most naturally captured by stranded diagrams which are equivalent to the two-dimensional SF structures (\cf \rem{stranded}).
In QRC and CDT the triangulations already have the full structure of $\std=4$ dimensions.

\item Combinatorial structure: as mentioned earlier, QRC and CDT rely crucially on the simplicial structure of triangulations where faces are triangles, 3-cells are tetrahedra, and 4-cells are 4-simplices. 
While this is also true for common models of SFs and GFT, it is not a necessity in these theories and complexes can be of more general type, \eg polyhedral complexes (defined in \sec{polyhedral}).
On the lower dimensional substructure this is reflected in the valency of dual vertices; for example, a vertex can only be interpreted as dual to a tetrahedron if it has valency of four incident edges.
Accordingly, a group field with the interpretation of exciting a tetrahedron geometry needs to have four arguments.
In LQG embedded graphs of arbitrary valency occur.

\item Analytic/topological structure: a complex is a purely combinatorial object in the first place (\sec{combinatorial-complexes}) and indeed this is all the structure explicit in GFT Feynman diagrams. 
In SFs, the cellular decomposition complexes have furthermore a topological-manifold structure induced from a background manifold $\stm$, \ie their cells are homeomorphic to Euclidean space.
In QRC and CDT, moreover it is  assumed that triangulations have a local flat metric structure.
Finally, the graphs in LQG inherit a full (semi-)analytic structure through the embedding into spatial slices.
\end{itemize}

This comparison does not only show that the discrete structures of the theories are related in many ways. It also unveils which details in the definition of discrete structure could be further specified in a particular model of one theory to match a particular model of another theory.
To illustrate this point with an example, take coloured group field theory \cite{Gurau:2011dw,Gurau:2012hl} which is a particular model of GFT. Therein, an additional $(\std+1)$-colouring of the field has the effect that Feynman diagrams correspond to simplicial manifolds, and thus to triangulations of topological manifolds. Therefore, the coloured-GFT model agrees with (models of) SFs, QRC and CDT with respect to dimensionality and combinatorial structure.

\

 
Concerning combinatorial structure, one could go even a step further and identify some kind of partial convergence already at the conceptual level.
There are many cases where one finds that some additional specifications for the objects in the structures in $\mp$ improve the properties of a model, either of the kinematics, dynamics or the empirical content. 
In the quantum gravity approaches there is in general a tendency towards less analytic, and also less topological background structure (related to stronger notions of background independence, in LQG motivated by diffeomorphism invariance), but more combinatorial structure (mostly needed for better behaved dynamics).
I will give a few explicit examples of such a partial convergence relations:
\begin{itemize}
\item The original, most general kinematical Hilbert space $\hkin^\lqg$ of LQG is not separable. Nevertheless, one can specialize to LQG models by a minimal extension of the notion of  diffeomorphisms resulting in a seperable Hilbert space \cite{Fairbairn:2004db}. Such generalized diffeomorphism have to be smooth only up to a finite number of points in their domain. 
Then, the resulting effective class of embedded graphs $\bge\In\sm$ is basically insensitive to analytic structure.

\item Quite in general it is often observed that predictions (correlations, observables) are only sensitive to the combinatorial (as opposed to analytic) structure of complexes underlying quantum geometries. 
These combinatorics already cover the relevant topological information without the need for the complexes to be realized as decompositions of topological manifolds, \ie locally Euclidean. 
This is certainly true for GFT and SF models, also argued for in LQG \cite{\lqgR}, and even in CDT the triangulations are often taken to be ``abstract" triangulations \cite{\cdtAJL}.
If quantum space and spacetime cannot be resolved effectively below a certain scale, there is no need to introduce continuous structures in the first place, but combinatorial manifolds are sufficient.
\end{itemize}
On the other hand, higher dimensional combinatorial structures seem to be necessary for well-defined dynamics:
\begin{itemize}
\item A rather trivial remark concerns the relation between the combinatorics of spatial and spacetime geometry: if spatial states are defined on $\m$-dimensional complexes (which are not a fixed background structure), one expects spacetime dynamics defining correlations of these states to involve some $(\m+1)$-dimensional structures.
For example in LQG, where states have support on embedded graphs, one expects quite generally that an explicit definition of the Hamiltonian evolution should demand at least a 2-complex structure. 
This relates LQG with the other approaches with respect to spacetime combinatorics.

\item Furthermore there are reasons why combinatorial structures of all $\std$ spacetime dimensions are needed. 
In \cite{Bonzom:2012gw,Bonzom:2012gwa,Bonzom:2012tg} it has been shown that certain divergences in the path integral of SF and GFT models are rooted in the combinatorial (cohomology) structure of the 2-complexes that amplitudes are based on. They turn out to be much better behaved when the 2-complex is the 2-skeleton of a $\std$-complex.
This explains also the success of coloured GFT models \cite{Gurau:2011dw,Gurau:2012hl}.
\end{itemize}
In this sense one can speak of a partial convergence of spacetime structures in LQG, SFs and GFT towards purely combinatorial $\std$-manifolds.
This matches also with QRC and CDT.

Note that for this purpose it is not necessary to change the basic theories as explicated in \sec{qg-theories} as long as they cover specializations with the desired properties. 
Typically it is exactly this inner theory structure of a basic ``core" theory together with its particular models which allows both for a wide range of applicability and precise descriptions.

The result of the current discussion of a preference for purely combinatorial structures is not only relevant for the question of convergence in this chapter; it also constitutes the motivation for the definition of discrete quantum geometries in this thesis in general. 
For this reason, and because in the quantum gravity literature this aspect is nowhere made explicit (to the best of my knowledge), I will dedicate chapter \ref{ch:discrete-spacetime} to develop more precisely the relevant properties and classes of combinatorial complexes.

In this section I have focused on discrete spacetime structure as an example of conceptual relations between the quantum gravity approaches. 
A similarly rich net of such relations could be discussed for other components of the theories, for example concerning relations of their geometric degrees of freedom.
For the present purpose the discussed relations should suffice to make the point and I will rather go on with the discussion of dynamical relations.


\subsection{Relations on the dynamical level \label{sec:dynamical-relations}}

Relations and comparability between the conceptual frameworks (structures $\mp$) of theories are only the first step to compare theories; what one is usually really interested in are the relations of their theoretical ($\ma$) and empirical ($\mpp$) content, \ie how their predictions compare.
Most interesting are thereby equivalence or embedding (reduction) relations which are usually not expected to be exact but should hold approximately, to a given degree of accuracy.
From this point of view I will review four of the most important relations between the five theories under consideration which together already connect all of them.

\subsubsection*{QRC and SF models}

There is a very well known relation between QRC and SF models originally motivated by the asymptotics of the $6j$ symbol of $SU(2)$ representation theory for large spins $\rep a$,
\[\label{6j-asymptotics}
\sixj123456 \sim \frac1{\sqrt{12\pi V_t}}  \frac1 2 \left(\e^{\i\sum_a \rep a \theta_a +\pi/4} + \e^{-\i\sum_a \rep a \theta_a +\pi/4} \right)\;
\]
where $\theta_a$ are the dihedral angles in a tetrahedron at edges with length $l(\rep a) = \h(\rep a +1/2)$ and $V_t$ is its volume \cite{Ponzano:1968wi,Williams:2007up} (\cf \sec{bf} for details).
Thus, the so-called \emph{Ponzano-Regge model} \cite{Ponzano:1968wi} defined by the state sum 
\[\label{pr-model}
Z_\pr(\T) = \sum_{\rep{ij}}  \prod_{(ij)\in\Tp1} (-1)^{c(\rep{ij})} (2\rep{ij} + 1)  \prod_{t\in\Tp3} \{6j\}_t
\]
for 3-dimensional triangulations $\T$ (with $c(\rep{ij})$ some function of the spins) 
can be approximated in the asymptotic regime by the integral \cite{Williams:2007up}
\[\label{pr-asymptotics}
 \int \prod_{(ij)\in\Tp1} \left[ \d \rep{ij} (2\rep{ij} +1)\right]  \prod_{t\in\Tp3} V_t^{-1/2} \left(\e^{\i S_\regge[l_{ij}^2=l(\rep{ij})^2;\T] } + \e^{-\i S_\regge[l_{ij}^2=l(\rep{ij})^2;\T]} \right) \;.
\]
Up to the double occurrence of the action phase weight $\exp(\pm S_\regge)$ with different signs, this integral is precisely of the form of a QRC path integral \eqref{qrc-path-integral} as explicated in \sec{qrc}, specified by the measure \eqref{qrc-measure} with $\alpha=-1/2$. The step function $\Theta(l(\rep{ij})^2)$ implementing triangle inequalities is automatically taken care of by the Clebsch-Gordan conditions implicit in the $6j$ symbol.
Moreover, there are arguments that the two action phases with different signs can be identified under parity symmetry.
Then the Ponzano-Regge model is approximately equivalent to a model of QRC in the regime of large spins $\rep{ij}$.
This regime is usually interpreted as the semiclassical limit of the theory since it is equivalent to $\h\ra\infty$ for fixed $l(\rep{ij}) = \h(\rep{ij} +1/2)$.

On the other hand, the Ponzano-Regge model is clearly a SF model according to theory explication \ref{sf-theory}. More precisely, it is a model in $\std=3$ dimensions specified by a restriction to cellular decompositions of simplicial type, \ie triangulations $|\cp| = \T$, and amplitudes
\[
A_f(\rep{ij}) = (-1)^{c(\rep{ij})} (2\rep{ij} + 1) \, ,\quad A_e(\rep{ij}) = 1 \, , \quad A_v(\rep{ij}) = \{6j\}
\]
where one should remember that the SF amplitudes were defined on the dual complex $\cps$ such that each face $f$ is dual to an edge $(ij)$ and each vertex $v$ to a tetrahedron $t$. 
In this sense, the 3-dimensional model of QRC defined by the measure \eqref{qrc-measure} with $\alpha=-1/2$ is approximatively equivalent to the Ponzano-Regge SF model.%
\footnote{
Certainly, one might wish to establish this approximative equivalence also with respect to further components of the theory such as the kinematical state space.
However, note that from a physical point of view an equivalence on the level of empirical structures $\mpp$, \ie approximative identities for observables, would be perfectly sufficient.
}

In the spirit of the Ponzano-Regge model case, the approximate equivalence of QRC and SF theory can be extended to a larger set of models of the theory, in particular to 4-dimensional quantum gravity models \cite{Barrett:1999wu,Barrett:2011wc,Barrett:2009ci,Barrett:2011bb}.
In fact, taking the interpretation of the regime of approximation as a semi-classical limit seriously, this equivalence has been playing an important role in the development of SF models.
If the amplitude of a given SF model asymptotes to the exponential of the Regge action, this is often considered as a check for the approximate embedding of GR into SF models, which is a desideratum for any proposal for quantum gravity.
One has to note though that the full relation of a quantum gravity theory to GR should involve not only a semi-classical but also a continuum limit.

\subsubsection*{SF models and LQG dynamics}

SF models cannot only be related to QRC but also to LQG. 
In fact, one of the main motivations of research on SF models as a proposal for quantum gravity is as a covariant formulation of LQG dynamics. 
With an optimistic attitude to this relation, particular SF models are even referred to as ``covariant LQG" nowadays \cite{Rovelli:2011tk,Rovelli:2014vg}.
With a less optimistic attitude one might question the possibility of a strict relation of LQG to currently known SF models (though LQG dynamics might still turn out to be of the SF type, \ie provide a new SF model )\cite{Thiemann:2014fn}.
Here I will discuss briefly how LQG and SFs indeed can be related in this sense.

First, one has to match the theories' kinematical Hilbert spaces to be able to compare then the respective projections onto the physical Hilbert subspaces.
In the discussion of the distinctness of SF degrees of freedom in \sec{differences}, I have already indicated that under various further specifications (in particular on the LQG side) the kinematical Hilbert spaces of SFs and LQG could match in the case of particular models.
While this is rather straightforward for the algebraic components, \ie the particular gauge group and connection, the possibility of such choice of models is more restricted with respect to the discreteness structures (as discussed in \sec{conceptual-relations}).
Since the specification of a SF model cannot relax the condition (in the present theory explication \ref{sf-theory}) that amplitudes are based on cellular decompositions $|\cp|$ 
on whose boundary $|\bs\cp|$ kinematical states have support on, one has to specify a model of LQG where kinematical states are restricted to embedded graphs $\bge$ which agree with such boundaries $|\bs\cp|$.

Having the kinematical Hilbert spaces matched, it is possible to show for the simplicial model in $\std=3$ dimensions that the projector $\pi_\phys$ \eqref{lqg-projector} can be expanded explicitly 
resulting in a spacetime triangulation \cite{Noui:2005js,Perez:2013uz}; this LQG model matches exactly the Ponzano-Regge SF model \eqref{pr-model} discussed above.
The case of $\std=3$ spacetime dimensions is particularly convenient since the Hamiltonian constraint $\widehat{C}$ is simply the curvature of the gauge connection.
But similar explicit constructions are also possible in $\std=4$ dimensions \cite{Alesci:2008jg,Alesci:2012cu,Alesci:2013cd,Thiemann:2014fn}.


\subsubsection*{GFT as completion of SF models}

A particular strong relation exists between SFs and GFT where basically for any SF model one can find a corresponding GFT model. 
This relation is based on the observation that the amplitudes $Z_\gft^\fdiagram(\qs,\qs')$ in the perturbative expansion of the GFT path integral \eqref{gft-path-integral} in Feynman diagrams $\fdiagram$ have exactly the structure of a SF amplitude $Z^\cp_\sfm(\qs,\qs')$ \eqref{sf-path-integral} if $\fdiagram$ is equivalent to the dual 2-skeleton of $\cp$, gauge groups agree, and the kinetic and interaction terms $\Kbb$ and $\Vbb$ in the GFT model are chosen such that GFT amplitudes match the local amplitudes $A_f,A_e,A_v$ of the SF model. 
I will explain the relation of $\fdiagram$ and $\cp$ in detail in \sec{molecules} where the common combinatorial structures will be called ``spin-foam molecules''.
The matching of amplitudes will be discussed in more detail in \sec{molecule-formulation}.

In the explication of a SF structure with a single cellular decomposition \eqref{sf-path-integral}, this relation can be seen as an approximate embedding relation in the following sense.
Taking GFT seriously as a perturbative QFT, in particular assuming that all coupling constants $\lambda_i$ in the GFT action are small, the correlation function between to states $\qs,\qs'$ has at its lowest order in coupling constants a single amplitude $Z_\gft^\fdiagram(\qs,\qs')$ which is at the same time a correlation function of a SF model. 
Thus, the SF correlation approximates the GFT correlation. 
On the level of theory models this means that the SF model is approximately embedded into the GFT model.

In the explication of SF models in terms of a summing description \eqref{sf-path-integral-sum}, the embedding relation of SF models into GFT has a different connotation.
If the class of complexes in a particular SF structure matches the class of complexes generated by a particular GFT model, the weights $w_\cp$ as derived from the GFT perturbative series can be used to specify further the SF model.
Thus, GFT can be understood as a completion in the sense of specifying these weights.
In any case, the relation between SF and GFT amplitudes gives rise to an embedding relation of SF models into GFT.

\subsubsection*{Dynamical triangulations, tensor models and GFT}

An obvious similarity between perturbative GFT and CDT is that both give a path integral description for a sum over complexes, though quite different ones on first sight. 
This manifests itself both in the classes of complexes summed over and the amplitudes associated with each complex.
In CDT the complexes are simplicial and of fixed topology (allowing even for a specific foliation) and the amplitudes are given by the exponential of the Regge action $S_\regge$.
On the other hand, GFT Feynman diagrams correspond most generally to 2-complexes and their amplitudes result from integration of inverse kinetic and vertex kernels $\Kbb^{-1}$ and $\Vbb$, as usual in QFT.

This comparison spells out that GFT provides a more general framework which could cover CDT.
Leaving aside the causality aspect of CDT for the moment, this is indeed possible in an approximative sense.
To this end one has to specify first the GFT gauge group $G$ such that it generates trivial (equilateral) discrete geometries. This is easily obtained choosing, for example, $G=U(1)$ and trivial kinetic and interaction kernels and defining the path integral as depending on a sharp cutoff $\size$ in the representation space such that the full sum is recovered in the limit $\size\ra\infty$. 
This turns the group field $\phi(g_1,\dots,g_\copies)$ effectively into a tensor $\phi_{m_1\dots m_\copies}$ which is why such models are also called tensor models \cite{Gurau:2012hl}.
Thus, tensor models generate a sum over simplicial complexes which, in the absence of any further geometric data, can be considered as equilateral triangulations of some topological manifolds, with various topologies.

The crucial result is now that certain tensor models have a dominant contribution of triangulations of fixed spherical topology in their large-$\size$ expansion with Regge-action weights \cite{Gurau:2011aq,Bonzom:2011cs,Gurau:2012ek,Gurau:2012hl}. 
In this regime, such a GFT model is thus approximately equivalent to a model of dynamical triangulations. 
The ensemble of triangulations in the GFT model does not contain all spherical triangulations, but the same ensemble can be specified in a model of dynamical triangulations. 
Furthermore, the effective behaviour which is that of branched polymers is rather insensitive to the difference in the precise ensemble of triangulations.
Indeed one can check that characteristic quantities agree on both sides, such as for example the spectral and Hausdorff dimension \cite{Gurau:2013th}.

A relation of such a tensor model of GFT with CDT is possible in various ways.
First of all, the branched polymer regime of the tensor model can already be related to a branched polymer regime of CDT on the empirical level. 
In CDT one is observing three phases in parameter space  one of which is a branched polymer phase like in dynamical triangulations \cite{Ambjorn:2005fh,Ambjorn:2012vc}. Concerning observables, these phases basically agree.
Thus the described tensor model of GFT is also approximately equivalent to CDT in that regime of parameters on the empirical level.
On the other hand, current work in progress indicates that it is also possible to specify further tensor models such as to implement directly the CDT notion of causality.
 Approximate equivalence to CDT is then expected even on the theoretical level.


\subsection{Discrete quantum gravity as a crystallization process \label{sec:crystallization}}

Having discussed in detail the various relations between the five theories as explicated in \sec{qg-theories}, I can now address the question in what sense one can speak of a convergence process in discrete quantum gravity.
I will argue that the better concept to describe the net of theories is that of a crystallization process.
Still, this has potential relevance for the epistemic status of these theories.
 

\

Given the theory relations explained in the last \sec{dynamical-relations} one can clearly motivate an (oversimplified) idea of convergence:
when QRC and SF theory are approximatively equivalent and SF models can be understand as LQG dynamics, they basically seem to agree. 
Since SF models can be embedded into GFT and CDT might turn out as a special case of GFT as well, everything seems to indicate a convergence into a single theory, possibly described by the framework of GFT.

This picture of convergence is certainly imprecise and oversimplified for several reasons.
One reason concerns the generality of theories and in particular the distinction between the core, or basic framework, of a theory and its particular models. 
Without going into a detailed explication of the concept of theory convergence, or theory reduction, convergence is usually understood on the most general level of the theory. 
More precisely, in a case of convergence a theory T$_1$ which reduces to a theory T$_2$ is covered in full generality by T$_2$, usually in a particular regime or as a particular model of T$_2$.
For example, there is a convergence from Newton's theory of gravity to GR in the sense that Newton's gravitational-force law, in full generality, reduces approximately to GR in the special case of a weak gravitational field.
Such a relation cannot be identified clearly in any of the cases discussed. 
Only, the embedding of SF models into GFT might be considered as most similar to such a reduction since for any SF model there is a corresponding GFT model.
% ...

Another desideratum for convergence regards the accuracy of theories which is related to the very general challenge of absence of empirical guiding in quantum gravity.
The asymmetry of the convergence relation between two theories manifest itself not only in the level of generality but also in the fact that, even in the applications where the predictions of the reduced theory T$_1$ agree with those of T$_2$ up to some accuracy, T$_2$ is still expected to provide further corrections at a level of higher accuracy.
In the above example, even if Newtonian gravity and GR both describe most of celestial mechanics equivalently in some approximation, one needs GR to account for the perihelion of Mercury, for example.
For quantum gravity we do not have any observations which could differentiate between proposals. Thus, in the cases of approximate equivalences such as for example between certain QRC and SF models in a semi-classical regime, there is no way to judge in which way a potential reduction would be directed.
Moreover, even in cases of embedding relations such as for SF models into GFT, it is not clear whether higher order terms in the GFT couplings turn out to be indeed ``correction" terms to the lower order ones, \ie provide more accurate values to observables and correlations, or whether a simple SF model on a single complex already provides the more precise description.
Thus, observational data would be needed to settle this issue.

\



The idea of convergence understood in a broader sense, \ie the intuition that there is some kind of evolution process in the quantum gravity theories under investigation towards a single theory, 
can still be made more precise.
To this end I propose to consider it as a particular instance of theory crystallization \cite{Moulines:2010hi,Moulines:2013jn,Moulines:2013fy}:
\begin{conjecture}
There is a 
crystallization process in progress in the quantum gravity theories described in theory explication \ref{qrc-theory}--\ref{gft-theory}. 
\end{conjecture}

The proposal is conjectural insofar as I can only sketch the concept of crystallization process and its application to the quantum gravity theories under consideration here.
In its formal definition \cite{Moulines:2013jn}, successfully applied to the case of thermodynamics \cite{Moulines:2013fy},  the concept relies in particular on the finer notion of an empirical theory as a tree-like net consisting of a basic theory element and its specializations, the models in physics' jargon.
While I have explicated the basic theory elements semiformally, I have described relevant models only informally in the discussion of the theories' relations.

Rephrasing it in these terms a $\emph{crystallization process}$ is characterized by the following four properties \cite{Moulines:2013jn}:
\begin{itemize}
\item [(a)] There are different competing theories during one and the same period of the process which share some, but not all, of their structures' components.

\item [(b)] The theories (basic theory elements) are essentially different, \ie the structures in $\mp$, $\ma$ and $\mpp$ of each 
theory differ, though they share some components.

\item [(c)] Some of their models (specializations) are equal or nearly equal, in particular they share some physical systems $\ia$ they are intended to be applied to, as well as some theoretical concepts.

\item [(d)] The process of crystallization concludes with the establishment of one full theory (in particular only one basic theory-element).

\end{itemize}

Apparently, the last criterion (d) is presently not met in the case of the quantum gravity theories. 
For this reason I suggest a crystallization process \emph{in progress}, meant as a generalization characterized only by properties (a) -- (c).
Let me now argue that these are indeed met and come back to point (d) in the end.

The fulfilment of the general criterion (a) is rather obvious for the quantum gravity theories.
In the way I have explicated QRC, CDT, LQG, SF models and GFT they are different theories, as further argued in \sec{differences}.
As proposals for a quantum theory of gravity they are furthermore direct competitors.
This manifests itself in particular in their structures $\mpp$ which contain observables and correlations of 3-dimensional geometries.

To make the point (b) clear, I have argued in \sec{differences} in detail how their structures in $\mp$ are distinct to each other. I have explicitly discussed the differences in the respective spacetime background structure and concerning the gravitational degrees of freedom thereon. 
Dynamics do not alter these components such that classes $\ma$ are essentially different as well. Moreover, in each case dynamics are implemented in a distinct way.
However, the comparison in \sec{conceptual-relations}, in particular of the kind of discrete structures, has shown that some of their components agree, like for example the fixed topological background.

Finally, the intertheoretical relations discussed in \sec{dynamical-relations} support the statement (c). 
There, I have presented (approximate) equivalence and embedding relation between models of each theory. 
These are connecting the theories mostly on the level of $\ma$ and $\mpp$.
Then, the relations extend easily to concrete intended applications $\ia$ 
which are basically the same for all of them as explained in the beginning of \sec{qg-theories};
They are all meant to explain early universe cosmology as well as black hole physics.

Note that criterion (c) makes the crucial difference to a notion of convergence.
For theory crystallization it is enough if certain models agree (which furthermore share the same physical interpretation). 
Furthermore, the notion of crystallization does not imply any asymmetric relations between the crystallizing theories. In particular, it is not necessary to specify which of two approximately equivalent models provides the more accurate description of phenomena which is not possible for theories of quantum gravity at the present state of experiment.

\

The interesting question is now whether in the case of quantum gravity theories one can make any claims about the missing point (d).
Quite in general, if one finds furthermore any correlations or connecting schemes between the criterion of agreement of models (c) with criterion (d) this might allow to make some general statements about a potential crystallized theory even if the process is yet in progress.

In the present case one could make an educated guess about the details of the eventual, crystallized theory on the basis of the present relations:
\begin{itemize}
\item The discussion in \sec{conceptual-relations} has shown that purely combinatorial structures suffice as the discrete structures quantum states and spacetime histories have support on.
\item The dynamical relations between LQG, SF models and GFT presuppose a description of gravitational degrees of freedom based on parallel transports of the Ashtekar-Barbero connection, \ie corresponding to the extension of the Plebanski action by a Holst term in the classical action.
\item Relations between GFT and SF models indicate a preference for dynamics in\-volving a sum over a class of complexes.
\end{itemize}
Indeed, these findings serve as a partial motivation for the choice of theoretical framework in this thesis.
In the light of a lack of more precise hints for a unique crystallized quantum theory of gravity, the investigation motivates also to keep the setting general and applicable to all theories in the net if possible, as done in the phenomenological investigations in \sec{dimensional-flow}.

\

From a philosophical perspective, the probably most interesting aspect of the identification of a crystallization process in discrete quantum gravity approaches is the question of epistemological significance of such a process: 
does the fact that these theories might be evolving into a single theory already increase their epistemological status? 
This question is of high relevance because the most obvious criteria of theory evaluation such as predictivity do not apply to quantum gravity where clear observational evidence is missing.

One can certainly argue that a complete process of theory crystallization provides an additional value in epistemological justification of the theory.
Apart form standard criteria for theory evaluation such as a theory's \emph{systematization force}, determined by the number of successful applications, and its \emph{informativeness}, or predicitivity, accomplished for example in terms of restrictive dynamical laws and high accuracy, there is also the third important criterion of \emph{uniformity} which can be explicated as indivisibility of a theory in a precise sense \cite{Bartelborth:2002us,Bartelborth:1996uj}. 
This criterion is particularly relevant for a coherence-theoretical account of epistemological justification \cite{BonJour:1985tu}.
From this perspective, the partial intertheoretical relations defining the notion of theory crystallization provide an important improvement of theories, even if it is not the case that one theory is fully embedded in another which evidently raises the value according to the first two criteria (as is the case in theory convergence).
The remaining question in the particular case of quantum gravity theories, however, is whether the process will in fact find an end in terms of a successful crystallized theory.

\ 

In this section I have discussed the relation of the five quantum gravity proposals from various perspectives. 
In \sec{differences} I have made explicit that the theories are indeed distinct, explaining in detail their differences with respect to background structure and geometric degrees of freedom.

Against this background I have argued then in \sec{conceptual-relations} that, nevertheless, the objects in their conceptual setting are related in many ways. In particular they can be matched in many cases by further specializations corresponding to particular models of a theory which I have illustrated with concrete examples focusing on the kind of discrete spacetime structure present in the different theories.

This sets the stage for the analysis of relations on the dynamical level in \sec{dynamical-relations}.
Concentrating on four bilateral relations I have argued that there are QRC models and SF models which approximately agree in a semi-classical regime, that LQG models specified by a particular construction of the Hamiltonian constraint might match with SF models, that for every SF model there is a GFT model defining a sum over a class of complexes weighted by that SF amplitude, and, finally, that certain ensembles of dynamical triangulations can be generated by specific GFT models, also known as tensor models.

There are certainly more such relations that one could mention here.
For example, particular models of LQG can be reformulated as second quantized theories in terms of GFT  \cite{Oriti:2013vv}, giving rise to an embedding relation similar to that of quantum mechanics into quantum field theory.
However, for the purpose of the present investigation the relations discussed are sufficient as they already connect all of the five theories.

Finally, in \sec{crystallization}, I have argued that the theory relations satisfy neither the criteria of generality nor of asymmetry in accuracy essential for a convergence relation between theories.
Therefore, it turned out to be more appropriate to regard the net of quantum gravity theories as an instance of a crystallization process.

\chapter[Discrete space(time)]{Discrete space(time): combinatorial and differential structure of complexes \label{ch:discrete-spacetime}}


%
%
%
%

An essential property of all the approaches to quantum gravity discussed in the last chapter is that space and spacetime involve some sort of discreteness in the sense that quantum states and histories have support on certain kinds of complexes.
As discussed in \sec{conceptual-relations}, the precise structure of complexes may differ with respect to various properties such as dimension, simplicial structure in contrast to more general cellular or polyhedral one, as well as topological, piecewise linear or even analytic embedding structure.

While in most approaches these complexes are indeed still based locally on continuum structures, either through an embedding or a definition of cells as (homeomorphic to) pieces of Euclidean space, 
some approaches demand a more radical, purely combinatorial perspective in which cells have no inner structure at all.
The strongest case in this direction is made by GFT where degrees of freedom are necessarily fundamentally discrete in this sense because there simply is no background structure at all and the Feynman diagrams are in fact purely combinatorial objects.

Here I would like to propose that a similar, purely combinatorial account is possible in all other approaches as well.
Even in proposals which are understood as continuum theories in an essential way, as for example LQG, the resulting definitions and structure might eventually be recast in a purely combinatorial way.
The major obstacle seems to be that the notions of complexes mostly used in the quantum gravity literature are taken from well established algebraic-topological definitions and an appropriate combinatorial equivalent is simply not known.
For simplicial complexes there is in fact a well known purely combinatorial definition \cite{Kozlov:2008wc}, but even in this case it is employed in the quantum gravity literature very rarely (for example in \cite{Gurau:2010iu}).

The aim of this chapter is thus to introduce a purely combinatorial account for any structure of complexes as they occur in the quantum gravity literature.
To this end, I will extend the combinatorial definition of abstract simplicial complexes to polyhedral complexes and generalize these further to include also loops.
Moreover, I will specify the conditions under which topological information in the sense of homology can be obtained from such combinatorial structures.
While many aspects of such a purely combinatorial formulation can already be found in Reidemeister's original axiomatic approach to homology \cite{Reidemeister:1938vf}, the definition of abstract polyhedral complexes as I introduce them is to the best of my knowledge new.

This sets the stage to further understand the structure of GFT Feynman diagrams as subdivisions of generalized polyhedral complexes.
Based on \cite{\ORT}, I will give a systematic and exhaustive presentation of the molecular structure that GFT and SF amplitudes have support on and explain how this relates to simplicial, polyhedral and generalized complexes. 
Showing how arbitrary such spin-foam molecule complexes can be constructed from a simplicial subclass, this provides also the basis for the generalization of GFT presented in chapter \ref{ch:dqg}.

If combinatorial complexes are a serious proposal for spacetime structure, it should also be possible to equip them with a differential structure in order to be able to define physical objects such as fields thereon.
Given a homology structure this is indeed possible.
I will introduce a precise definition for fields on such combinatorial complexes and the appropriate kind of discrete exterior calculus.
With a further assignment of geometric data such as volumes to all cells, this extents even to a calculus on discrete geometries.
This is a straightforward extension of my work in \cite{\COTa} where such a calculus based on \cite{Desbrun:2005ug,Grady:2010wb} has been presented for discrete geometries of a simplicial type.
Furthermore, it is of relevance for the investigation of effective-dimension observables of discrete quantum geometries in chapter \ref{ch:dimensions} providing the definition of a Laplacian in that context.


\

This chapter is structured into three sections.
Section \ref{sec:combinatorial-complexes} introduces the notion of a combinatorial complex and various subclasses of these. 
In particular, I will present \emph{abstract polyhedral complexes} as a natural generalization from abstract simplices to abstract polytopes \cite{McMullen:2009ff}. 
Then I define \emph{generalized polyhedral complexes} based on the graded structure of simplicial subdivisions which may furthermore contain all kinds of loops,
and discuss the orientability of all these complexes.

Section \ref{sec:molecules} focuses on the structure of combinatorial complexes as they enter the stage in the SF models and group field theories.
Along the lines of \cite{\ORT} the molecular structure of complexes built from atomic parts, coined \emph{spin-foam molecules}, is systematically discussed. 
Especially, I present constructive proofs for the facts that any such spin-foam molecule can be obtained from regular loopless atoms and, even more, that for any boundary graph a molecule consisting of simplicial atoms exists.
This is the basis for the new group field theories to be presented in \sec{dw-gft}.

Finally, \sec{calculus} introduces a discrete exterior calculus on orientable combinatorial complexes.
The space of cochains provides a natural concept of exterior forms and an inner product is induced by a notion of Hodge duality based on the dual complex together with a geometric interpretation, \ie an assignment of volumes to cells.
I introduce a useful bra-ket notation for the resulting $L^2$ space of exterior forms.
The chain-complex structure and Stokes theorem provide then a natural definition of an exterior differential.
As the basis for the investigation of effective dimensions (chapter \ref{ch:dimensions}), I will finally discuss in detail the resulting Laplacian and its properties.

This whole chapter is written in a mathematics fashion, in contrast to the rest of the thesis, acknowledging the technical nature of the topic.
Section \ref{sec:molecules} and subsection \ref{sec:polyhedral} are based on the publication \cite{\ORT} and
section \ref{sec:calculus} on the publication \cite{\COTa}.


\section{Combinatorial Complexes}\label{sec:combinatorial-complexes}


The concepts of complexes usually used in the quantum gravity literature (more or less explicitly) are those which have proven most natural and useful in algebraic topology \cite{Hatcher:2002ut}.
These are either simplicial complexes based on $\p$-simplices as smallest convex subsets in $\R^\p$, thus relying in particular on a piecewise linear structure (at least homomorphically).
Or, far more generally, they are CW complexes, also referred to as cell complexes, which are obtained inductively gluing $\p$-balls along $(\p-1)$-balls in the disjoint union of their boundaries.

Neither of these two notions of complex is appropriate for GFT where the generated spacetimes are necessarily of fundamentally discrete type.
While the combinatorial structure of simplices is too restrictive for some versions of discrete quantum gravity, CW complexes are far too general. 
Moreover, both are defined in an explicit topological manner and are relying on metric spaces, either $\R^\p$ or balls and spheres. But a GFT provides only purely combinatorial structures for spacetime.
These facts motivate to propose a purely combinatorial class of complexes to be identified as the appropriate spacetime structure in discrete quantum gravity.
This is the topic of this section, extending \cite{\COTa} and \cite{\ORT}.

For the case of simplicial complexes there is already a well known combinatorial definition (\sec{simplicial}). 
In \sec{complex}, I recapitulate the very general combinatorial concept of a complex and its various features \cite{Reidemeister:1938vf} which seems to be sunken into oblivion despite providing the most appropriate mathematical framework for discrete space in quantum gravity.
In this setting, I propose the concept of \emph{polyhedral complexes} in \sec{polyhedral} which is based on the notion of abstract polytopes \cite{McMullen:2009ff,Danzer:1982dp}. 
In \sec{subdivision}, I discuss then the graded structure of simplicial subdivisions of combinatorial complexes. This structure allows to reconstruct the original complex from a simplicial one. But the result is not a polyhedral complex in every case, giving rise to the notion of \emph{generalized polyhedral complexes} which include various types of loops.
Thus, in absence of further restrictions, the most obvious draw-back of this is that they prevent orientability of the complex, and thus a chain-complex structure, which is the topic of the final \sec{orientation}.


\subsection{Abstract simplicial complexes}\label{sec:simplicial}

For simplicial complexes there is already a well known combinatorial definition. They are therefore an illustrative and paradigmatic example of a combinatorial complex to start with. 
Originally, \eg in \cite{Reidemeister:1938vf}, they are defined as complexes consisting of $\p$-dimensional cells which have boundary structure isomorphic to that of simplices in $\R^\p$. 
This rather implicit definition can be made explicit in a combinatorial way \cite{Kozlov:2008wc}:
\begin{defin}[{simplicial complex}]
\label{def:simplicial-complex}
A \emph{finite abstract simplicial complex} $\simc$ is a finite set of vertices $\simcp{0}=\{v_{1},v_{2},\dots ,v_{\np0}\}$ together with a collection (multiset) of ordered subsets of $\simcp{0}$ such that
\begin{description}
\item[(SC)\label{SC}] for every $\sigma\in \simc$ and $\sigma'\In\sigma$ also $\sigma'\in \simc$.%
\footnote{
Note that the empty set is a subset of any set. Thus it is considered as the unique (-1)-simplex, $\simcp{-1}=\{\emptyset\}\ne\emptyset$.
}
\end{description} 
A subset $\sigma\in \simc$ is called a \emph{simplex}. Each simplex has a dimension defined by its cardinality, $\dm(\sigma)=|\sigma|-1$.
\end{defin}

From a more general perspective on complexes, the essential structure of abstract simplicial complexes is their partial ordering \cite{Kozlov:2008wc}:
\begin{defin}[{poset representation}]
For a finite abstract simplicial complex $\simc$ the \emph{face poset} $\mathcal{F}(\simc)$ is the partially ordered set (poset) whose elements consist of all nonempty
simplices of $\simc$ and whose partial-ordering relation is the inclusion relation on the set of simplices.
\end{defin}
It will turn out in the following that partially ordered sets are the appropriate mathematical setting to extend from simplicial to  polyhedral complexes in a combinatorial way.


\subsection{A general notion of complex}\label{sec:complex}

%

While its hard to find a definition of \lq complex\rq\ in a most general sense in more recent textbooks and literature, it has been defined in the early days \cite{Reidemeister:1938vf}.
In this section I recall the definitions of \cite{Reidemeister:1938vf} which are necessary for the introduction of abstract polyhedral complexes and their generalizations in the following sections.

\defin[\bf{complex}]\label{def:complex}{
A \emph{complex}%
\footnote{There is a slight modification to the original definition here in that anti-symmetry of the ordering relation $<$ is not demanded in the original \cite{Reidemeister:1938vf}.}
$(\cp,\dm,\le)$ is a set $\cp$ of elements $c$, called \emph{cells}, to which one assigns a natural number $\dm(c)\in\{0,1,2,\dots\}$ and a partial ordering $\le$ on $\cp$,   called \emph{bounding}, with the property:
\begin{description}
\item[(CP)] If $c>c''$ and $\dm(c)-\dm(c'')>1$, then there is a cell $c'$ such that $c>c'>c''$.
\end{description}
In the following  a complex is often just denoted $\cp$, leaving $\dm$ and $\le$ implicit.
}

\begin{remark}[{Hasse diagrams}]
A good way to visualize posets are Hasse diagrams. 
These are graphs drawn on the plane where vertices represent the poset elements and edges the transitivity reduced ordering relations, \ie there is an edge for every two elements $c>c'$ for which there is no $c''$ in the poset such that $c>c''>c'$.  
A canonical way to draw a Hasse diagram for a complex $\cp$ is to use its graded structure  $\cp=\bigcup_{\p} \cpp{\p}$ induced by the dimension map $\dm$ and set all $\p$-cells at the same height in the plane, usually bottom to top with increasing dimension. (An example appears further below, figure. \ref{fig:hasse}).
\end{remark}


Already for the most general notion of complex one can introduce the concept of a dual complex. 
A \emph{dual poset} is the poset with inverted partial order and it is straightforward to define a dual dimension map for the dual of a complex such that it is a complex again:
\begin{defin}[\bf{dual complex}]
\label{def:dual-complex}
Let $(\cp,\dm,\le)$ be a complex with all cells $c\in\cp$ of finite dimension $0\le\dm(c)\le\m$. 
The $\m$-dual complex $\star(\cp,\dm,\le) = (\cps,\dm^\star,\le^\star)$ given by a bijective map $\star : \cp \ra  \cps$ is defined for all $c,c'\in \cp$ by the further properties
\begin{eqnarray}
\dm^\star(\star c)&:=&\m-\dm(c)  \\
\star \cl & \le^\star & \star \cl' \quad \textrm{ if and only if } \cl\ge \cl' .
\end{eqnarray}

It is straightforward to show that the condition (CP) is satisfied for $\cps$.
Moreover, the $\star$ map is unique up to isomorphisms preserving the structure of the complex.

Note that, while $\star \cp = \cps$, the grading is inverted: $\star \cpp{\p} = \cpsp{\m-\p}$ for every $\p$.
\end{defin}
Of particular interest is the $\m$-dual for a complex $\cp$ with cells of maximal dimension $\m$. 
Then one can call $\cps$ just the \emph{dual} complex.
In principle, though, one can consider the $\m$-dual also for a complex $\cp$ with $\cpp{\m}=\emptyset$. Then it is necessary to state explicitly that $\cps$ is meant to be the $\m$-dual.
For example in LQG and SF models, main objects are faces dual to graph edges or to 2-complex faces. These have to be understood as the $\sd$- or $\std$-dual respectively, even though higher cells are missing in these complexes.

\

Already for combinatorial complexes one can define a natural notion of a boundary which allows further to specify various classes of subcomplexes:

\begin{remark}[\bf{subcomplexes: boundary, bundle, section}]
For a subset of a complex $\cp$ the axiom (CP) is not fulfilled in general. Instead subcomplexes can be obtained in the following manner \cite{Reidemeister:1938vf}:

The \emph{boundary} of a cell $\cl\in\cp$, defined as the set of all bounding cells of $\cl$ in $\cp$
\[\label{cellboundary}
\bs\cl := \{\cl'\in\cp \mid \cl'<\cl \},
\]
is a complex, as well as its \emph{hull} $\ol c =\{c\}\cup\bs c$.

The other way round, also the \emph{bundle} $\bsi\cl$ (translated from German ``B\"uschel'', not ``B\"undel", in \cite{Reidemeister:1938vf})  of a cell consisting of all the cells in $\cp$ which $c$ is bounding, \ie
\[\label{cellbundle}
\bsi\cl := \{\cl'\in\cp \mid \cl'>\cl \},
\]
is a complex.
I have chosen the notation \lq $\bsi$\rq~  since obviously the boundary of the dual cell $\star\cl$ in $\cps$ is the dual of $\bsi\cl$, \ie $\bs(\star \cl) = \star(\bsi \cl)$.

A subset of cells $\cp'\In\cp$ is \emph{closed} if for every cell $\cl\in\cp'$ also $\bs\cl\in\cp'$. 
A closed subset is a complex, also called \emph{closed subcomplex}. 
Sums and intersections of closed subcomplexes are complexes  \cite{Reidemeister:1938vf}.

Note that the definitions of boundary and bundle are independent of the dimension map $\dm$.
Thus they can be understood on any poset.

In this sense, the \emph{section} of two elements of a poset $c,c'\in P$ (sometimes also called \emph{interval}) which is defined as \cite{McMullen:2009ff,Danzer:1982dp}
\[\label{section}
c/c' := \{c''\in P \mid c \le c'' \le c'\},
\]
is the intersection of their boundary and bundle in $P$ respectively together with $c,c'$ themselves,
\[
c/c' = (\bs c \cap \bsi c')\cup\{c,c'\}.
\]
If $P$ is a complex, each section $c/c' \In P$ is obviously a complex too.
Moreover, the dual of a section $\star(c/c') = \star c'/\star c \In P^\star$ is a section again.
\end{remark}

The bounding relation of combinatorial complexes permits further a natural notion of neighbours as incident cells:

\begin{defin}[\bf{paths, chains, boundary matrix}]
\label{def:paths}

Two cells $c,c'$ are \emph{incident} if either $c\le c'$ or $c\ge c'$.

A \emph{path} is a sequence 
\[
\label{path}
W = c^1c^2\dots c^k
\]
in which each successive pair of cells is incident and has dimension difference of one.

A 
\emph{boundary sequence}
is a totally ordered sequence of poset elements.
All boundary information can be encoded in boundary matrices 
\[
\bsc{\clp\p^i}{\clp{\p-1}^j}=\bsc{i}{j}^\p=1\quad\text{or}\quad 0
\]
depending on whether the $(\p-1)$-cell $\clp{\p-1}^j$ is bounding the $\p$-cell $\clp\p^i$ or not.
This yields an \emph{incidence number} for a path $W = c^1c^2\dots c^k$
\[\label{incidence-number}
\epsilon_{W} = \prod_{i=1,\cdots,k-1} \bsc{c^i}{ c^{i+1}} =1\quad\text{or}\quad 0 \;.
\]
according to whether start and end cell are incident or not.
\end{defin}

In quantum gravity, sometimes only a partial structure of a spacetime complex is specified, as in the case of the 2-complexes that SF and GFT amplitudes have support on.
The full dimensionality of a complex is specified if it has the following properties:

\begin{defin}[\bf{dimension, pure, flags, bounding polynomial}]
A complex $\cp$ is \emph{$\m$-dimensional} if there are no cells of higher dimension than $\m$ and each cell of lower dimension is incident to at least one $\m$-dimensional cell \cite{Reidemeister:1938vf}.

A complex $\cp$ is \emph{pure} if each cell of non-zero dimension is incident with at least one 0-dimensional cell.
If $\cp$ is pure, it is favourable to include also a unique cell $\clp{-1}$ of dimension $\dm(\clp{-1})=-1$ incident to all other cells.


In a pure complex each boundary sequence is a partial boundary sequence of a maximal boundary sequence consisting of $\m+2$ cells (including $\clp{-1}$).
A maximal boundary sequence is also called a \emph{flag} \cite{McMullen:2009ff,Danzer:1982dp}.

A pure complex is determined by the incidence numbers of its flags (if not pure, then the maximal boundary sequences) which can be encoded in the \emph{boundary} polynomial \cite{Reidemeister:1938vf}
\[
b(c)=\sum \bsc{\clp0^{i_0} \clp1^{i_1} \dots \clp\m^{i_\m}}{}\; \clp0^{i_0} \clp1^{i_1} \dots \clp\m^{i_\m}
\]
\end{defin}

\begin{remark}[{vertex representation}]
\label{rem:vertexrep}
The cells of a pure, countable complex $\cp=\bigcup_{\p} \cpp{\p}$, can be represented by a collection of (ordered) sets in the following way.
The 0-cells are labelled in an arbitrary way by natural numbers, $\cpp{0}=\{v_1,v_2,...\}$.
Then, every cell $c$ is represented by the ordered set $(v_{i_{1}},v_{i_{2}},\dots)$ consisting of all vertices $v_{i_j}\le c$.%
\footnote{Note that different cells might have the same vertex set, which is the reason why this representation is a collection, i.e. a multiset. To distinguish explicitly, an extra label is needed.

Furthermore, the vertex representation is meaningful only if for every cell $\cl$ with $\dm(\cl)>0$ there are at least two $(\dm(\cl)-1)$-cells bounding it. 
Otherwise, cells of different dimension have the same vertex set.
Abstract simplicial complexes and polyhedral complexes (introduced below, \dref{polyhedral-complex}) meet this requirement.
}

Obviously, every simplicial complex is pure and the vertex representation of its face poset is just the simplicial complex itself.
\end{remark}

\begin{remark}[\bf{connectedness}]
A complex is \emph{connected} if for any two cells $\cl^1, \cl^k$ there is a path $W = c^1c^2\dots c^k$ (not including $c_{-1}$).

If a complex is not connected, it is the sum of connected complexes, called its \emph{components}.

A complex is \emph{\p-dimensionally connected} if for every two $\p$-cells there is a path consisting of $\p$-cells and $(\p-1)$-cells connecting the two. 
\end{remark}


Already for combinatorial complexes one can apply a notion of a manifold \cite{Reidemeister:1938vf,Seifert:1980uo}: 

\begin{defin}[\bf{pseudo-manifold, boundary}]
\label{def:manifold}
A complex $\cp$ is a \emph{combinatorial $\m$-dimensional pseudo-manifold} if and only if it is
\begin{description}
\item[(M1)] $\m$-dimensional and pure 
\item[(M2)] $\m$-dimensionally connected and
\item[(M3)] \emph{non-branching}, \ie each $(\m-1)$-cell is incident to at most two $\m$-cells in $\cp$.
\end{description}
In general, for a complex $\cp$, an $(\m-1)$-cell incident to exactly one $\m$-cell is called a \emph{boundary cell} and the \emph{boundary} $\bs \cp$ is the closure of the set of boundary $(\m-1)$-cells.
Accordingly, a complex $\cp$ is \emph{closed} if and only if it has empty boundary.
\end{defin}

\begin{remark}[\bf{dual of pseudo-manifolds}]\label{rem:dual-manifold}
For later purposes it is helpful to point out the meaning of properties (M1) - (M3) of a combinatorial pseudo-$\m$-manifold $\cm$ for its dual complex $\cms$:
\end{remark}
\begin{itemize}[leftmargin=1.2cm]
\item[(M1$^\star$)] $\cms$ is pure (since $\cm$ is $\m$-dimensional) and it is $\m$-dimensional (since $\cm$ is pure).
\item[(M2$^\star$)] $\cms$ is $1$-dimensionally connected.
\item[(M3$^\star$)] All 1-cells $\cl\in\cmsp{1}$ are either edges or half-edges. An \emph{edge} is a 1-cell with two incident vertices and a \emph{half-edge} a 1-cell with one incident vertex. 
\end{itemize}
Note that (M3$^{\star}$), on the contrary, also states that the dual of a branching $\m$-dimensional complex has 1-cells incident to arbitrarily many vertices. 
These are not edges but hyper-edges \cite{Berge:1989wn}.

The half-edges in $\cps$ dual to boundary $(\m-1)$-cells in $\cp$ imply that there is no equivalent of the boundary $\bs\cp\In\cp$ in $\cps$. 
Instead, one has to define the \emph{dual boundary} $\dbs$ for a dual complex explicitly as the dual of the boundary of $\cp$,
\[\label{dualboundary}
\dbs\cps := \star(\bs\cp).
\]
Accordingly, the definition of the dual closure 
\[\label{dualclosure}
\clos\cps := \cps \cup \dbs\cps,
\]
includes bounding relations between boundary and bulk cells which are induced by those of the primal complex $\cp$.


\subsection{Abstract polyhedral complexes}\label{sec:polyhedral}

\begin{figure}
\begin{center}
\tikzsetnextfilename{hasse}
\begin{tikzpicture}
\node (empty)	[h, circle]				{$\emptyset$};
\node (3)		[above of=empty]	{\dots};
\node (2)		[h,left of=3]		{2};
\node (1)		[h,left of=2]		{1};
\node (4)		[h,right of=3]		{$n$-1};
\node (5)		[h,right of=4, circle,inner sep=1pt]	{$n$};
\node (12)		[hv,above of=1 ]		{12};
\node (23)		[hv,right of=12 ]		{23};
\node (34)		[right of=23 ]		{\dots};
\node (45)		[hv,right of=34 ]		{$n$-1$\;n$};
\node (51)		[hv,right of=45 ]		{$1n$};
\node (poly)	[he,above of=34]		{123...$n$};
\foreach \from/\to in {12/2,23/3,34/4,45/5,51/1}{
	\draw (poly) 
	 -- (\from);
	\draw (\from)
	 -- (\to);
	\draw (\to) 
	 -- (empty);
	}
\foreach \from/\to in {12/1,23/2,34/3,45/4,51/5}{
	\draw (\from) 
	 -- (\to);
	}
\begin{scope}[xshift=5cm]
\node (0)		[h, circle]					{$\emptyset$};
\node (a)		[h,above of=0,xshift=-15mm]	{$1$};
\node (b)		[h,right of=a]				{$2$};
\node (c)		[h,right of=b]				{$3$};
\node (d)		[h,right of=c]				{$4$};
\node (ab)		[h,above of=a]				{$12$};
\node (bc)		[h,right of=ab]				{$23$};
\node (cd)		[h,right of=bc]				{$34$};
\node (ad)		[h,right of=cd]				{$41$};
\node (f)		[h,above of=ab, xshift=15mm]	{$1234$};
\node (x)		[h,above of=0, xshift=40mm]	{$5$};
\node (ax)		[h,above of=x, xshift=-15mm]	{$15$};
\node (bx)		[h,right of=ax]				{$25$};
\node (cx)		[h,right of=bx]				{$35$};
\node (dx)		[h,right of=cx]				{$45$};
\node (abx)	[h,above of=ax]				{$125$};
\node (bcx)	[h,right of=abx]				{$235$};
\node (cdx)	[h,right of=bcx]				{$345$};
\node (adx)	[h,right of=cdx]				{$145$};
\node (p)		[h,above of=abx,xshift=15mm] 	{$12345$};
\foreach \from/\to in 
{a/0,b/0,c/0,d/0,x/0,
ab/b,bc/c,cd/d,ad/a,
f/ab,f/bc,f/cd,f/ad,
ax/x,bx/x,cx/x,dx/x,
abx/bx,bcx/cx,cdx/dx,adx/ax,
p/abx,p/bcx,p/cdx,p/adx,
abx/ab,bcx/bc,cdx/cd,adx/ad}
    \draw (\from) 
     -- (\to);
\foreach \from/\to in 
{ab/a,bc/b,cd/c,ad/d,
abx/ax,bcx/bx,cdx/cx,adx/dx,
ax/a,bx/b,cx/c,dx/d,
p/f}
    \draw (\from) 
     -- (\to);
\end{scope}
\end{tikzpicture} 
\caption{Hasse diagram of a complex which captures the combinatorics of the $\copies$-polygon (left) and a complex with the structure of a pyramid (right), both in the vertex representation (remark \ref{rem:vertexrep}). 
}
\label{fig:hasse}
\end{center}
\end{figure}
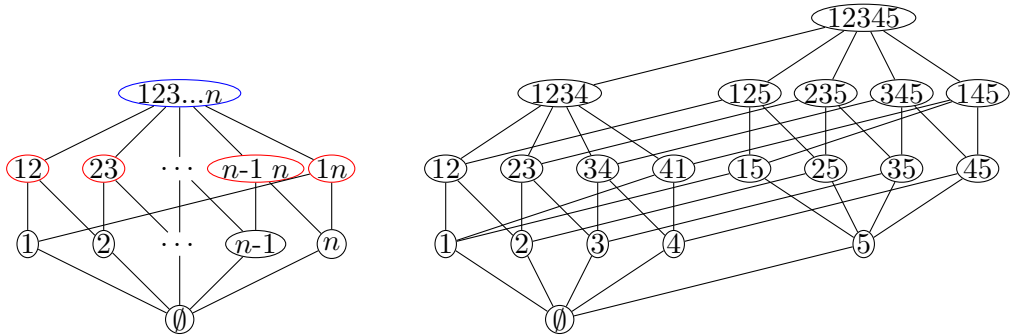


It is possible to define polyhedral complexes combinatorially as collections of abstract polytopes in the same spirit as simplicial complexes are collections of simplices.
Technically, the essential difference between the simplicial and the polyhedral case is the defining condition (SC) for simplicial complexes (\dref{simplicial-complex}) which guarantees that cells are indeed simplices with the full simplicial boundary structure. 
While this is given naturally in terms of subsets of vertex sets there, for polytopes the boundary cell structure has to be spelled out explicitly in terms of the partial-ordering relation.
Fortunately, there exists a purely combinatorial definition of abstract polytopes  as sections with some extra properties \cite{McMullen:2009ff,Danzer:1982dp}.
I will introduce these properties first.
Then I will define and discuss the new concept of polyhedral complexes.

\begin{remark}[\bf{properties of sections}]\label{rem:sections}
A section $f/f'\In P$ in a poset $P$ has two important properties:
\begin{description}
\item[(P0)] It contains a (unique) least and greatest face.
\item[(P1)] Each flag consists of $\m+2$ faces for some $\m\in\N$ 
\end{description}
Clearly, the least face in $f/f'$ is $f'$ and the greatest face is $f$. 
Thus, the flags define a notion of \emph{dimension} for sections in a poset as $\dm(f/f')=\m$.

Now, these two properties can be used to define an interesting class of posets in its own.
Let $P$ be a poset obeying (P0) and (P1) and denote $f_{-1}$ and $f_{\m}$ the least and greatest face.
For each face $f\in P$ a dimension is induced by $\dm(f):=\dm(f/f_{-1})$ turning $P$ into a complex since (P1) yields the defining property of complexes (CP) in definition \ref{def:complex}.
The complex $P$ is furthermore pure and $\m$-dimensional because of (P1).

Faces different from $f_{-1}$ and $f_{\m}$ are called \emph{proper} faces of $P$. A complex of this kind is \emph{connected} if for every two faces there is a path of proper faces connecting them.
\end{remark}

Now one can state the definition of an abstract polytope  \cite{McMullen:2009ff,Danzer:1982dp}:
\begin{defin}
\label{def:Polytope}
An \emph{abstract $\m$-polytope}, i.e. an abstract polytope of finite dimension $\m\ge -1$, is a poset $(P,<)$ obeying (P0), (P1) and:
\begin{description}
\item[(P2)] $P$ is \emph{strongly connected}, \ie every section of $P$ is connected.
\item[(P3)] All one-dimensional sections of $P$ are \emph{diamond-shaped}, \ie  for every $f''$ bounding $f$ with $\dm(f)-\dm(f'')=2$, there are exactly two faces $f'$ such that $f>f'>f''$.
\end{description}
\end{defin}

\begin{example}[{low dimensional polytopes}]
\label{rem:polyexamples}
Up to $\m=2$ there is a very manageable amount of abstract polytopes:
\begin{itemize}
\item Every 0-polytope is a single vertex, having the form $P=\{\emptyset, v\}$ with $\emptyset<v$.
\item Because of (P3), every 1-polytope consists of a single edge, $P=\{\emptyset, v_1, v_2, e\}$ with $\emptyset<v_i<e, i=1,2$.
\item Every finite 2-polytope is a polygon \cite{McMullen:2009ff} of the form shown in \fig{hasse}.
\end{itemize}
\end{example}

\begin{remark}[{duality}]
Abstract polytopes have a natural notion of a dual in terms of inversion of the partial order. 
Finite graded structure, connectedness and diamond shape of 1-sections guarantee that the dual poset is in fact an abstract polytope as well \cite{McMullen:2009ff}.
An example is given in \fig{hassedual}.
\end{remark}

\begin{figure}
\begin{center}
\tikzsetnextfilename{hassedual}
\begin{tikzpicture}[every node=rectangle]
\node (0)		[h, circle]					{$\emptyset$};
\node (a)		[hv, below of=0,xshift=15mm]	{$1$};
\node (b)		[hv, left of=a]				{$2$};
\node (c)		[hv, left of=b]				{$3$};
\node (d)		[hv, left of=c]				{$4$};
\node (ab)		[he,below of=a]				{$12$};
\node (bc)		[he,left of=ab]				{$23$};
\node (cd)		[he,left of=bc]				{$34$};
\node (ad)		[he,left of=cd]				{$41$};
\node (f)		[h,below of=ab, xshift=-15mm]	{$1234$};
\node (x)		[hv, below of=0, xshift=-40mm]	{$5$};
\node (ax)		[he,below of=x, xshift=15mm]	{$15$};
\node (bx)		[he,left of=ax]				{$25$};
\node (cx)		[he,left of=bx]				{$35$};
\node (dx)		[he,left of=cx]				{$45$};
\node (abx)	[h,below of=ax]				{$125$};
\node (bcx)	[h,left of=abx]				{$235$};
\node (cdx)	[h,left of=bcx]				{$345$};
\node (adx)	[h,left of=cdx]				{$145$};
\node (p)		[h,below of=abx,xshift=-15mm] 	{$12345$};
\foreach \from/\to in 
{a/0,b/0,c/0,d/0,x/0,
ab/b,bc/c,cd/d,ad/a,
f/ab,f/bc,f/cd,f/ad,
ax/x,bx/x,cx/x,dx/x,
abx/bx,bcx/cx,cdx/dx,adx/ax,
p/abx,p/bcx,p/cdx,p/adx,
abx/ab,bcx/bc,cdx/cd,adx/ad}
    \draw (\from) -- (\to);
\foreach \from/\to in 
{ab/a,bc/b,cd/c,ad/d,
abx/ax,bcx/bx,cdx/cx,adx/dx,
ax/a,bx/b,cx/c,dx/d,
p/f}
    \draw (\from) -- (\to);
\end{tikzpicture}
\caption{Hasse diagram of the dual of the pyramid complex in \fig{hasse} which can itself be represented as a pyramid. 
}
\label{fig:hassedual}
\end{center}
\end{figure}
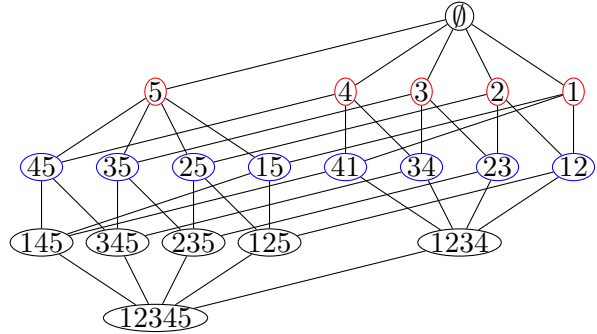

An abstract polyhedral complex can now be defined as a collection of abstract polytopes in analogy to abstract simplicial complexes as collections of abstract simplices:

\begin{defin}[{polyhedral complex}]
\label{def:polyhedral-complex}
An \emph{abstract polyhedral complex} $\polyc$ is a poset with the properties that
\begin{description}
\item[(P0')] $\polyc$ contains a least face, denoted $f_{-1}$, and
\item[(PC)] for every face $f\in\polyc$ the section $f/f_{-1}\In\polyc$ is an abstract polytope.
\end{description}
Obviously $\polyc$ is indeed a complex because it consists of sections which are already proven to be complexes (remark \ref{rem:sections}). 
\end{defin}

Since an abstract simplicial complex is a collection of abstract simplices and abstract simplices are abstract polytopes, it is obvious that abstract simplicial complexes are a special case of abstract polyhedral complexes.
In particular, in both cases complexes are pure but not necessarily $\m$-dimensional, \ie the largest cell a vertex is bounding can be of different dimension for different vertices in the complex.

\subsection{Simplicial subdivisions}\label{sec:subdivision}

It is straightforward to notice that (the face poset) of every simplicial complex is a polyhedral complex.
More interestingly, a simplicial subdivision%
\footnote{Even in the context of combinatorial topology this is often called barycentric subdivision \cite{Kozlov:2008wc}, even though there is no notion of centre in the abstract setting. For this reason, and to highlight that it is a subdivision into abstract simplices, I prefer to call it ``simplicial subdivision".}
of polyhedral complexes can be defined purely combinatorial as well.
I propose a definition analogous to the case of simplicial complexes \cite{Kozlov:2008wc}, by defining vertices for every (proper) face and simplices for every boundary sequence, effectively subdividing all polytopes into simplices.
Eventually, there is no reason to restrict even to polyhedral complexes:

\begin{defin}[{simplicial subdivision}]
\label{def:subdivision}
The \emph{simplicial subdivision} of a finite complex $\cp$ is the simplicial complex generated by all its boundary sequences, 
\[\label{subdivision}
\ssub\cp := \left\{ (c_1,c_2,\dots,c_t) | c_1>c_2>\dots>c_t, c_i\in\cp, t\ge1 \right\}.
\]
The subdivision $\ssub\cp$ is indeed a simplicial complex since in a poset each subset of a boundary sequence is a boundary sequence of the poset as well, proving (SC).
\end{defin}
 
\begin{remark}[{grading of subdivisions}]
The vertex set of a simplicial decomposition $\subc=\ssub\cp$ has an $(\m+1)$-grading 
$\subcp0 = \bigcup_{\p=0}^\m \subcp0_\p$ 
induced by the original finite complex $\cp$ (with cells of dimension up to $\m$) according to $\subcp0_\p \cong \cpp\p$.
Similarly, all simplex sets $\subcp\p$ of higher dimension $\p$ have an $\binom{\m+1}{\p+1}$ grading into levels $\subcp\p_{q_0\dots q_\p} \cong \cpp{q_0}\times\dots\times\cpp{q_\p}$ (\fig{subdivision}).
(Combining this grading with the graded structure of the dimension map $\dm$, the simplicial subdivision $\subc$ thus is equipped with a $\sum_{\p=0}^\m \binom{\m+1}{\p+1} = 2^\m$ grading.)

Let us call such a grading on a simplicial $\m$-complex a \emph{$\ssub$-grading}.

\begin{figure}
\begin{center}
\tikzsetnextfilename{subdivision}
\begin{tikzpicture}
\draw [|->] (4,0) -- node[label=above:$\ssub$] {} (5,0);

\begin{scope}[scale=.4]
\node [c]		(a)	at (0,0)		{};
\node [c]		(b)	at (8,0)		{};
\node [c]		(c)	at (3,6)		{};
\node [c]		(d)	at (-3,3)		{};
\node [c]		(e)	at (-6,-2)		{};
\node [c]		(f)	at (2,-5)		{};
\node [c]		(g)	at (9,-5)		{};
\path [f]
(b) -- (c) -- (d) -- (e) -- (f) -- (g) -- cycle;
\foreach \i/\j in {a/b,a/c,a/d,a/e,a/f,b/c,c/d,d/e,e/f,f/g,g/b,f/b}
\draw [eb] 
(\i) node [vh]{}	-- (\j) node [vh]{};
\end{scope}

\begin{scope}[xshift=8cm,scale=.4]
\node [c]		(a)	at (0,0)		{};
\node [c]		(b)	at (8,0)		{};
\node [c]		(c)	at (3,6)		{};
\node [c]		(d)	at (-3,3)		{};
\node [c]		(e)	at (-6,-2)		{};
\node [c]		(f)	at (2,-5)		{};
\node [c]		(g)	at (9,-5)		{};
\foreach \i/\j in {a/b,a/c,a/d,a/e,a/f,b/c,c/d,d/e,e/f,f/g,g/b,f/b}{
\draw [eb] 	(\i) node [vh]{}	-- (\j) node [vh]{};
}
\node [v]	(abc)		at (3.7,2)		{};
\node [v]	(acd)		at (0,3)		{};
\node [v]	(ade)		at (-3,.33)		{};
\node [v]	(aef)		at (-1.33,-2.33)	{};
\node [v]	(abf)		at (3.33,-1.66)	{};
\node [v]	(bfg)		at (6.33,-3.33)	{};
\node [vb]	(ab)		at (4,0)		{};
\node [vb]	(ac)		at (1.5,3)		{};
\node [vb]	(ad)		at (-1.5,1.5)	{};
\node [vb]	(ae)		at (-3,-1)		{};
\node [vb]	(af)		at (1,-2.5)		{};
\node [vb]	(bc)		at (5.5,3)		{};
\node [vb]	(cd)		at (0,4.5)		{};
\node [vb]	(de)		at (-4.5,.5)		{};
\node [vb]	(ef)		at (-2,-3.5)		{};
\node [vb]	(fg)		at (5.5,-5)		{};
\node [vb]	(bf)		at (5,-2.5)		{};
\node [vb]	(bg)		at (8.5,-2.5)	{};
\foreach \i/\j/\k in {a/b/c,a/c/d,a/d/e,a/e/f,a/b/f,b/f/g}{
\path	[f] 	(\i)	-- (\j) -- (\k) -- cycle;
\draw [eh] 	(\i) 	-- (\i\j\k);
\draw [eh] 	(\j)	-- (\i\j\k);
\draw [eh] 	(\k) 	-- (\i\j\k);
\draw [e] 	(\i\j) 	-- (\i\j\k);
\draw [e] 	(\j\k)	-- (\i\j\k);
\draw [e] 	(\i\k) 	-- (\i\j\k);
}
\end{scope}
\end{tikzpicture}
\caption{Simplicial subdivision $\subc = \ssub\cp$ of a polyhedral 2-complex $\cp$.
The grading of $\subc$ is visualized by colours: Blue is used for elements of $\subcp0_0=\cpp0$, red for $\subcp0_1$ and black for $\subcp0_2$. The $\binom{3}{1}=3$ kinds of edges are green for $\subcp1_{01}$, red for $\subcp1_{12}$ and dotted for $\subcp1_{02}$.
There is only one kind of triangles, $\subcp2_{012}$.
}
\label{fig:subdivision}
\end{center}
\end{figure}

This structure allows to define also the inverse $\issub$ to the subdivision map $\ssub$. 
Let $\subc$ be a simplicial $\m$-complex with a grading as described. 
Then
\[\label{inversesubdivision}
\issub\subc \equiv \bigcup_{\p=0}^\m \left(\issub\subc\right)^{[\p]} 
:= \bigcup_{\p=0}^\m \subcp0_\p 
\]
and the bounding relation $<$ on $\ssub\subc$ is defined for all $c\in\subcp0_i$, $c'\in\subcp0_j$ by
\[\label{inversesubdivisionrelation}
c<c'\quad \textrm{ if and only if } (c,c')\in\subcp1 \textrm{ and } i<j .
\]
It is straightforward to check that $\issub$ is indeed the inverse to the subdivision map $\ssub$.

A geometrical intuition behind this formal definition (at least in the case of combinatorial pseudo-manifolds) is to understand the $\p$-cells $\cl\in\left(\issub\subc\right)^{[\p]}$ as reconstructed from gluing together the $\p$-simplices which the original cell is divided into.
That is, every $\cl \in \subcp0_\p$ can be related to the``gluing" (e.g. the union) of all $\p$-simplices $(\clp0,\dots,\clp\p)\in\subcp\p_{01\dots\p}$ containing the subdivision vertex $\cl=\clp\p$.
\end{remark}

\begin{remark}[{subdivision and the dual}]
\label{rem:subdivision}
The geometric intuition behind simplicial subdivisions provide also a complementary understanding of  duality of complexes.
Since the definition of the subdivision \eqref{subdivision} is invariant under inverting the bounding relation $<$, the subdivision of a complex $\cp$ and its dual $\cps = \star\cp$ coincide,
\[
\ssub\cp = \ssub\cps .
\]
It is thus also straightforward to define a dual version $\issubs$ of the inverse of the subdivision map by merely setting $(\issubs\subc)^{[\p]}:=\subcp0_{\m-\p}$ 
(instead of \eqref{inversesubdivision}) and inverting the bounding relation in \eqref{inversesubdivisionrelation} accordingly (\ie $c<c'$ for $i>j$).

Obviously this is equivalent to the definition of the dual (\dref{dual-complex}), that is
\[\label{dualfromsub}
\cps = \issubs \ssub \cp .
\]

More interesting is the geometric intuition behind it. Similarly to the case of $\issub$, one can understand the dual cells 
$\cs \in (\issubs\subc)^{[\p]}$ as constructed by ``gluing" all $\p$-simplices $(\clp{\m-\p},\dots,\clp\m)\in\subcp\p_{\m-\p\dots\m}$ containing that $\cs = \clp{\m-\p}$.
In this sense, \eqref{dualfromsub} indeed tells us how to construct the cells of a dual complex from the subdividing simplices (\fig{dualfromsub}).
\end{remark}

\begin{figure}
\begin{center}
\tikzsetnextfilename{dualfromsub}
\begin{tikzpicture}
\draw [|->] (4.5,0) -- node[label=above:$\issubs$] {} (5.5,0);
\draw [|->] (11.7,0) -- node[label=above:closure] {} (12.8,0);

\begin{scope}[xshift=8cm,scale=.4]
\node [c]		(a)	at (0,0)		{};
\node [c]		(b)	at (8,0)		{};
\node [c]		(c)	at (3,6)		{};
\node [c]		(d)	at (-3,3)		{};
\node [c]		(e)	at (-6,-2)		{};
\node [c]		(f)	at (2,-5)		{};
\node [c]		(g)	at (9,-5)		{};
\path [f]	(b) -- (c) -- (d) -- (e) -- (f) -- (g) -- cycle;
\node [v]	(abc)		at (3.7,2)		{};
\node [v]	(acd)		at (0,3)		{};
\node [v]	(ade)		at (-3,.33)		{};
\node [v]	(aef)		at (-1.33,-2.33)	{};
\node [v]	(abf)		at (3.33,-1.66)	{};
\node [v]	(bfg)		at (6.33,-3.33)	{};
\node [c]	(bc)		at (5.5,3)		{};
\node [c]	(cd)		at (0,4.5)		{};
\node [c]	(de)		at (-4.5,.5)		{};
\node [c]	(ef)		at (-2,-3.5)		{};
\node [c]	(fg)		at (5.5,-5)		{};
\node [c]	(bf)		at (5,-2.5)		{};
\node [c]	(bg)		at (8.5,-2.5)	{};
\foreach \i/\j/\k in {a/b/c,a/c/d,a/d/e,a/e/f}{
\draw [e]	(\j\k)	-- (\i\j\k);
}
\path [e] 	(abc) -- (acd) -- (ade) -- (aef) -- (abf) -- (abc);
\draw [e]	(abf)	-- (bfg);
\draw [e]	(fg)	-- (bfg);
\draw [e]	(bg)	-- (bfg);
\end{scope}

\begin{scope}[xshift=15.2cm,scale=.4]
\node [c]		(a)	at (0,0)		{};
\node [c]		(b)	at (8,0)		{};
\node [c]		(c)	at (3,6)		{};
\node [c]		(d)	at (-3,3)		{};
\node [c]		(e)	at (-6,-2)		{};
\node [c]		(f)	at (2,-5)		{};
\node [c]		(g)	at (9,-5)		{};
\node [c]	(bc)		at (5.5,3)		{};
\node [c]	(cd)		at (0,4.5)		{};
\node [c]	(de)		at (-4.5,.5)		{};
\node [c]	(ef)		at (-2,-3.5)		{};
\node [c]	(fg)		at (5.5,-5)		{};
\node [c]	(bf)		at (5,-2.5)		{};
\node [c]	(gb)		at (8.5,-2.5)	{};
\path [f]	(f) \foreach \i/\j in {f/g,g/b,b/c,c/d,d/e,e/f}{ .. controls (\i) .. (\i\j)} -- cycle;
\node [v]	(abc)		at (3.7,2)		{};
\node [v]	(acd)		at (0,3)		{};
\node [v]	(ade)		at (-3,.33)		{};
\node [v]	(aef)		at (-1.33,-2.33)	{};
\node [v]	(abf)		at (3.33,-1.66)	{};
\node [v]	(bfg)		at (6.33,-3.33)	{};
\foreach \i/\j/\k in {a/b/c,a/c/d,a/d/e,a/e/f}{
\draw [e]	(\j\k)	-- (\i\j\k);
}
\path [e] 	(abc) -- (acd) -- (ade) -- (aef) -- (abf) -- (abc);
\draw [e]	(abf)	-- (bfg);
\draw [e]	(fg)	-- (bfg);
\draw [e]	(gb)	-- (bfg);
\foreach \i/\j/\k in {b/c/d,c/d/e,d/e/f,e/f/g,f/g/b,g/b/c}{
\draw [eb] 	(\i\j) node [vb]{}	.. controls (\j) .. (\j\k) node [vb]{};
}
\end{scope}

\end{tikzpicture}
\caption{The dual $\cps$ to the complex of \fig{subdivision} obtained from its subdivision via $\issubs$ (colouring as before) and the dual closure \eqref{dualclosure}.
The vertices are $\cpsp0 = \subcp0_2$, 
edges are obtained from (joining all edges in $\subcp1_{12}$ incident to a given vertex in) $\subcp0_1$ and faces from (joining all triangles in $\subcp2$ incident to a given vertex in) $\subcp0_0$. 
The edges $e=\star\cl$ dual to boundary cells $\cl\in(\bs\cp)^{[1]}$ are half-edges. 
Thus, the complex $\cps$ becomes polyhedral only under closure.
}
\label{fig:dualfromsub}
\end{center}
\end{figure}



\begin{remark}[{$\issub$-inverse of $\ssub$-gradings}]
\label{rem:issub-inverse}
  It is important to note that for a $\ssub$-graded simplicial complex $\subc$ the inverse $\issub\subc$ (and thus $\overline{\issubs\subc}$, by duality) is not a polyhedral complex, in general.
  The main reason for this is the following: the simplices in $\subc$ joined into a cell of $\cl\in\issub\subc$ may have common boundary cells (in the vertex representation $\cl$ is a multi-set). 
  This violates the defining condition (P3).
  An example is given in \fig{2torus}.

  In SF models and GFT, a common instance of this is the case of a shared boundary $(\std-1)$-cells, resulting in a dual edge starting and ending at the same vertex. 
  As a consequence there are divergences in the amplitudes called wrapping singularities \cite{\lostGS}.
  Still, in the most general SF and GFT amplitudes such non-polyhedral complexes $\issub\subc$ occur.
  
  It is not possible to explicitly define these complexes directly in the purely combinatorial framework, \ie without reference to an underlying $\ssub$-graded complex.%
  \footnote{In algebraic topology it is standard to generalize simplicial complexes to include such special cases, the so-called ``$\Delta$-complexes" \cite{Hatcher:2002ut}.}
  (which is demanded by GFTs). 
  The reason is that on the level of posets there is no way to distinguish, for example, between half-edges (1-cells incident to exactly one vertex) and loops (1-cells incident to this vertex twice).
  However, this information is contained in the $\ssub$-graded simplicial complex in terms of which $\issub\subc$ is defined (again \fig{2torus}).
\end{remark}

\begin{figure}
\begin{center}
\tikzsetnextfilename{2torus}
\begin{tikzpicture}
\draw [|->] (3.5,1.5) -- node[label=above:$\issubs$] {} (4.5,1.5);
\begin{scope}[xshift=.5cm]
\node [c,label=left:$\vh$]		(1)	at (0,0)		{};
\node [c,label=right:$\vh$]		(2)	at (3,0)		{};
\node [c,label=right:$\vh$]		(3)	at (1.5,2.6)	{};
\node [c,label=left:$\vh$]		(4)	at (-1.5,2.6)	{};
\path [f]	(1) -- (2) -- (3) -- (4) -- cycle;
\path [eb]	(1) -- (2) -- (3) -- (1) -- (4) -- (3);
\node [vb,label=below:$\vb_1$]	(a)	at (1.5,0)		{};
\node [vb,label=right:$\vb_2$]	(b)	at (2.25,1.3)	{};
\node [vb,label=above:$\vb_1$]	(c)	at (0,2.6)		{};
\node [vb,label=left:$\vb_2$]	(d)	at (-.75,1.3)	{};
\node [vb,label=above:$\vb_3$]	(e)	at (.75,1.3)	{};
\node [v,label=above:$v_1$]	(A)	at (1.5,.87)	{};
\node [v,label=below:$v_2$]	(B)	at (0,1.73)		{};
\foreach \i/\j in {A/a,A/b,A/e,B/c,B/d,B/e}{
\draw [e]	(\i) -- (\j);
}
\foreach \i/\j in {A/1,A/2,A/3,B/3,B/4,B/1}{
\draw [eh]	(\i) -- (\j) node[vh]{};
}
\end{scope}

\begin{scope}[xshift=6.5cm]
\node [c,label=left:$\vh$]		(1)	at (0,0)		{};
\node [c,label=right:$\vh$]		(2)	at (3,0)		{};
\node [c,label=right:$\vh$]		(3)	at (1.5,2.6)	{};
\node [c,label=left:$\vh$]		(4)	at (-1.5,2.6)	{};
\path [f]	(1) -- (2) -- (3) -- (4) -- cycle;
\path [eb]	(1) node[vh]{} -- (2) node[vh]{} -- (3) node[vh]{} -- (1) node[vh]{} -- (4) node[vh]{} -- (3);
\node [c,label=below:$e_1$]	(a)	at (1.5,0)		{};
\node [c,label=right:$e_2$]	(b)	at (2.25,1.3)	{};
\node [c,label=above:$e_1$]	(c)	at (0,2.6)		{};
\node [c,label=left:$e_2$]		(d)	at (-.75,1.3)	{};
\node [c,label=above:$e_3$]	(e)	at (.75,1.3)	{};
\node [c,label=above:$f_1$]	(A)	at (1.5,.87)	{};
\node [c,label=below:$f_2$]	(B)	at (0,1.73)		{};
\end{scope}

\begin{scope}[xshift=12cm,scale=.1]
\node (0)		[h, circle]					{$\emptyset$};
\node (1)		[he,above of=0, yshift=-3mm]	{$\vh$};
\node (b)		[hv,above of=1, yshift=-3mm]	{$\vb_2$};
\node (a)		[hv,left of= b]				{$\vb_1$};
\node (e)		[hv,right of=b]				{$\vb_3$};
\node (A)		[h,above of=b, xshift=-5mm,yshift=-3mm]	{$v_1$};
\node (B)		[h,right of=A]				{$v_2$};
\node (T)		[h,above of=b, yshift=4mm]	{$T^2$};
\foreach \from/\to in 
{1/0,1/b,1/a,1/e,A/a,A/b,A/e,B/a,B/b,B/e,T/A,T/B}
    \draw (\from)     -- (\to);
\end{scope}

\end{tikzpicture}
\caption{A $\ssub$-graded simplicial complex $\subc$ for which $\issub\subc$ is not a polyhedral complex since its 1-cells are loops, \ie they are incident to only one vertex (violating (P3)), but twice. The latter information is not contained in $\issub\subc\cong\subcp0$ as such, as can be seen in its Hasse diagram (on the right), but is only induced by its construction in terms of $\subc$.
}
\label{fig:2torus}
\end{center}
\end{figure}
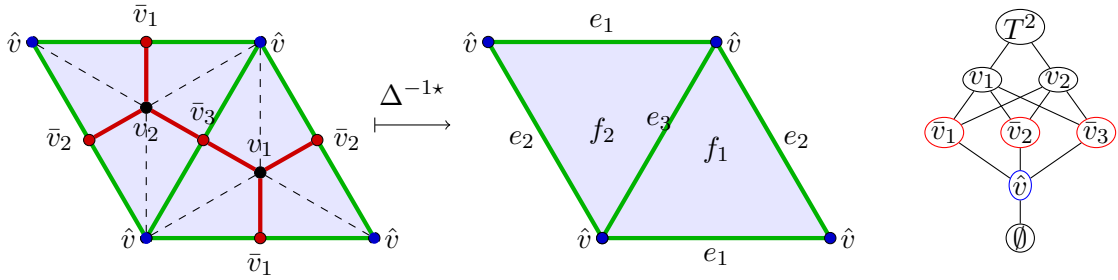

This remark therefore motivates to define a generalization of polyhedral complexes including loops and their higher dimensional analogues:
\begin{defin}[{generalized polyhedral complex}]
\label{def:generalized-polyhedral}
  A generalized polyhedral complex $\gpolyc$ is a complex for which there is a non-branching 
  $\ssub$-graded simplicial complex $\subc$ such that $\gpolyc = \issub\subc$.
  
  In the vertex representation of a $\gpolyc$, multiple bounding relations are made explicit in the following way. 
  A $\p$-cell $\cl\in\gpolycp\p$ is represented by a multiset containing each vertex $v$ as many times as there are $\p$-simplices in $\subc$ containing both $\cl$ and $v$. 
  Thus, in general, the vertex representation of a generalized polyhedral complex is not only a multiset but it also contains multisets.
\end{defin}


\begin{remark}[{dual 2-complexes}]
The minimal structure needed in quantum gravity (in particular SF dynamics) are the cells up to dimension two of the dual complex of spacetime dimension $\m=\std$ (the so-called dual 2-skeleton), together with 1-dimensional (graph) structure of the boundary. 
This information is conveniently encoded in the vertex sets of the simplicial subdivision $\subc$.
In the following I will use the notation
\[\label{vertex-sets}
  \begin{array}{lll}
    \V = \subcp0_\m, &  \Vb = \subcp0_{\m-1}, & \Vh = \subcp0_{\m-2} \\
    \E = \subcp1_{\m-1,\m-2}, & \Eb = \subcp1_{\m,\m-1}, & \Eh = \subcp1_{\m,\m-2}
  \end{array}
\]
(as already anticipated in \fig{2torus}).
\end{remark}
 

\subsection{Orientation} 
\label{sec:orientation}

For later purposes, in particular the possibility to define an exterior calculus (\sec{calculus}), it will be necessary to have the notion of homology defined for combinatorial complexes. 
This is related to orientability of the cells in the complex.
It can be defined in a standard way and I will follow again \cite{Reidemeister:1938vf} to introduce the relevant notions in this section.

\begin{defin}[orientation]
\label{def:orientation}
In a pure complex $\cp$, an \emph{orientation} of a cell $\cl\in\cp$ is an assignment 
\[\label{orientation}
\sgn (\lc\clp0\dots\clp{\p-1}\clp{\p} ) = \pm 1
\]
to all its flags $\lc\clp0\dots\clp{\p-1}\clp{\p}$ in $\cp$
such that any two 
flags differing in exactly one $\q$-cell $\clp\q\ne\clp\q'$ have opposite sign:
\[\label{orientation-property}
\sgn (\lc\clp0\dots\clp\q\dots\clp{\p-1}\clp{\p}) = - \sgn (\lc\clp0\dots\clp\q'\dots\clp{\p-1}\clp{\p}) .
\]
An \emph{orientation} of a pure complex $\cp$ is an assignment \eqref{orientation} to the flags of all its cells obeying \eqref{orientation-property}.
\end{defin}

Note that a necessary condition for the orientation property \eqref{orientation-property} is that every one-dimensional section $\clp{\p+1}/\clp{\p-1}$ contains at most two $\p$-cells. 
If there are exactly two for all one-dimensional sections, \ie all cells are abstract polytopes, there is an important equivalent property \cite{Reidemeister:1938vf}:
\begin{proposition}[oriented boundary matrices]
\label{prop:orientation}
A polyhedral complex $\polyc$ is orientable if and only if there is a system of \emph{oriented boundary matrices}, \ie 
\[
\bc{\clp\p}{\clp{\p-1}} = \pm \bsc{\clp\p}{\clp{\p-1}}
\]
such that for all one-dimensional sections $\clp{\p+1}/\clp{\p-1}\In\polyc$
\[\label{oriented-matrix-property}
\sum_{\clp{\p+1}>\clp\p>\clp{\p-1} } \bc{\clp{\p+1}}{\clp{\p}} \,\bc{\clp\p}{\clp{\p-1}} = 0 .
\] 
\end{proposition}

\begin{proof}
The proof follows \cite[\S13.3]{Reidemeister:1938vf}:
let $\clp{\p+1}/\clp{\p-1} \In \polyc$ be a one-dimensional section in an oriented polyhedral complex
and define for any $\bsc{\cl}{\cl'}\ne 0$ 
\[
\bc\cl{\cl'}:= \sgn(\lc\clp0\dots\cl') \sgn(\lc\clp0\dots\cl' \cl) .
\]
This is well-defined since any change of the sequence $\lc \clp0 \dots \cl'$ yields the same sign change in both terms.
Then 
$\bc{\clp{\p+1}}{\clp{\p}} \,\bc{\clp\p}{\clp{\p-1}} = \sgn(\lc\clp0\dots\clp{\p-1}) \sgn(\lc\clp0\dots \clp\p \clp{\p+1})$
and, because of (P3), the sum in \eqref{oriented-matrix-property} runs over exactly two $\p$-cells $\clp\p,\clp\p'\in\clp{\p+1}/\clp{\p-1}$ with opposite signs in the second term because of \eqref{orientation-property}. 
This proves \eqref{oriented-matrix-property}.

The other way round, let $\cp$ be a pure complex with a system of oriented boundary matrices $\{\bc\cl{\cl'}\}$.
Assign to all flags $\lc\clp0\dots\clp\p$ of all cells $\clp\p\in\polyc$
\[
\sgn(\lc\clp0\dots\clp\p):= \prod_{\q=0}^\p \bc{\clp\q}{\clp{\q-1}} .
\]
Then \eqref{oriented-matrix-property} yields that any two flags differing in exactly one $\q$-cell $\clp\q\ne\clp\q'$ have relative sign 
$ \bc{\clp{\q+1}}{\clp\q} \bc{\clp\q}{\clp{\q-1}} / \bc{\clp{\q+1}}{\clp\q'} \bc{\clp\q'}{\clp{\q-1}} = -1 $.
\end{proof}


\begin{example}[orientation of simplicial complexes]
  Every abstract finite simplicial complex $\simc$ is orientable. 
  The canonical way \cite{Kozlov:2008wc} to define a system of oriented boundary matrices is to use the ordering of each $\p$-simplex vertex set $(i_1 i_2 \dots i_\p) \equiv (v_{i_1},v_{i_2},\dots,v_{i_\p})\in\simcp\p$ to define
  \[
  \bc{(i_1 i_2 \dots i_\p)}{(i_1 \dots \hat{\iota}_\q \dots i_\p)} := (-1)^\q
  \]
  where $(i_1 \dots \hat{\iota}_\q \dots i_\p)$ is the $(\p-1)$-simplex consisting of all vertices except for $v_{i_\q}$.
  See \fig{oriented-simplex} for an example (where orientations are visualized in the Hasse diagram using directed edges). 
  \begin{figure}
    \begin{center}
      \tikzsetnextfilename{oriented-simplex}
\begin{tikzpicture}[shorten >=3pt,shorten <=3pt]

\node (ab)		[h]						{$12$};
\node (ac)		[h,left of=ab]				{$13$};
\node (ad)		[h,left of=ac]				{$14$};
\node (bc)		[h,left of=ad]				{$23$};
\node (bd)		[h,left of=bc]				{$24$};
\node (cd)		[h,left of=bd]				{$34$};

\node (a)		[h,below of=ab,xshift=-5mm]	{$1$};
\node (b)		[h,below of=ad,xshift=2mm]	{$2$};
\node (c)		[h,below of=bd,xshift=7mm]	{$3$};
\node (d)		[h,below of=cd,xshift=5mm]	{$4$};

\node (abc)	[h,above of=ab,xshift=-5mm]	{$123$};
\node (abd)	[h,above of=ad,xshift=2mm]	{$124$};
\node (acd)	[h,above of=bd,xshift=7mm]	{$134$};
\node (bcd)	[h,above of=cd,xshift=5mm]	{$234$};

\node (p)		[h,above of=acd,xshift=8mm] 	{$1234$};
\node (0)		[h,circle, below of=c,xshift=8mm]{$\emptyset$};

\foreach \from/\to in 
{ab/a,ac/a,ad/a,bc/b,bd/b,cd/c,
abc/ad,acd/ad,abd/ad,bcd/bd,
p/abc,p/acd}
\draw (\from) [->, line width=.8pt] -- (\to);

\foreach \from/\to in 
{a/0,b/0,c/0,d/0,
ab/b,ac/c,ad/d,bc/c,bd/d,cd/d,
abc/ab,acd/ac,abd/ab,bcd/bc,abc/bc,acd/cd,abd/bd,bcd/cd,
p/abd,p/bcd}
    \draw (\from) [<-,line width=.8pt] -- (\to);

\end{tikzpicture}
      \caption{Canonical orientation on the abstract tetrahedron $(1234)$ with arrow down (up) for positive (negative) coefficient in the oriented boundary matrix.}
      \label{fig:oriented-simplex}
    \end{center}
  \end{figure}
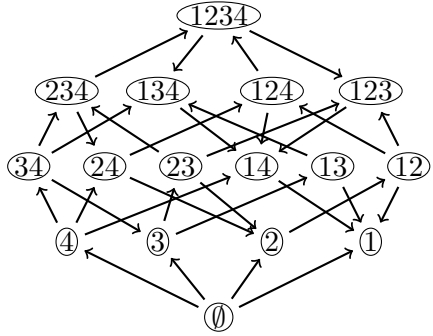
\end{example}

\begin{remark}[orientable manifolds]
  Note that orientability of a complex is a local concept. 
  For combinatorial (pseudo-)manifolds it is still a weaker property than the global orientability of the manifold.
  For example, all kinds of non-orientable smooth surfaces are triangulable with the underlying combinatorial structure of abstract simplicial complexes (which are orientable, as shown above).
  Thus, even if all $\m$-cells in an $\m$-dimensional pseudo-manifold $\cm$ are orientable, it is not yet guaranteed that they give a consistent global orientation to $\cm$.
  
  This can be made explicit by formally adding to an $\m$-dimensional pseudo-manifold $\cm$ a unique greatest $(\m+1)$-dimensional element $\gc$ incident to all $\m$-cells $\clp\m\in\cm$.
  Thus, the whole manifold is the section $\gc/\lc = \cm$. Then one can define:
  \begin{description}
    \item A combinatorial pseudo-manifold $\cm = \gc/\lc$ is \emph{orientable} if $\gc$ is orientable. 
  \end{description}
\end{remark}

A system of oriented boundary matrices is the basis for chain complexes and homology \cite{Reidemeister:1938vf,Hatcher:2002ut}:
\begin{defin}[{chain complex}]
\label{def:chain-complex}
  Given a complex $\cp$, \emph{$\p$-chains} are formal linear combinations of $\p$-cells in $\cp$  with integer coefficients, \ie elements of the free Abelian group over $\cpp\p$, denoted as $C_\p(\cp)\equiv C(\cp,\Z)$.
  For each $\p$, a boundary map $\bmp\p: C_\p(\cp) \lora C_{\p-1}(\cp)$ is defined by its action on each $\p$-cell
  \[
  \bmp\p(\clp\p) := \sum_{\clp{\p-1}\in\cpp{\p-1}} \bc{\clp\p}{\clp{\p-1}} \clp{\p-1} .
  \]
  If $\cp$ is pure and finite and the coefficients $\bc{\clp\p}{\clp{\p-1}}$ of the boundary map are a system of oriented boundary matrices, it follows that
  \[
  \bmp\p \bmp{\p+1} = 0 .
  \]
  Such a sequence of Abelian groups $C(\cp)=\bigcup_{\p}C_\p(\cp)$ obviously is a complex as well.
  Together with the boundary maps
  \[
  0 \os{\bmp{\m+1}}{\lora}  C_\m(\cp) \os{\bmp{\m}}{\lora}  C_{\m-1}(\cp) \os{\bmp{\m-1}}{\lora} \dots \os{\bmp2}{\lora}  C_1(\cp) \os{\bmp1}{\lora}  C_0(\cp) \os{\bmp0}{\lora}  C_{-1} (\cp) = \Z\,\lc
  \]
  it is called the \emph{chain complex} $(C(\cp),\bm)$.
  The well-defined quotients $H_\p:= \text{Ker}\,\bmp\p / \text{Im}\,\bmp{\p+1}$ are the \emph{homology groups}.
\end{defin}



\

To sum up this \sec{combinatorial-complexes}, starting from the well known combinatorial definition of abstract simplicial complexes I have recalled the most general definition of a complex as a partially ordered set with a bounding relation and a notion of dimension.
I have shown how polyhedral complexes can be defined in this purely combinatorial setting and how this class can be further generalized on the basis of subdivision complexes. 
Finally, I have introduced the proper notion of orientability of complexes and explained its relation to homology.

\newpage


\section{Complexes from atoms}
\label{sec:molecules}


The combinatorial complexes which spin-foam models are based on have a molecular structure in the sense that they can be constructed from atomic building blocks on which their amplitudes have support.
Most basic models are derived on a cellular decomposition of a smooth spacetime manifold and their amplitudes turn out to be based on the 2-skeleton of the dual complex \cite{Baez:2000kp,Perez:2013uz}.
From a genuine state sum perspective, emphasizing the intermediate spin-network states summed over, the $\std$-cell structure of the primal complex is reflected on the dual in terms of the atomic parts in the subdivision, as has been noted and explained in detail in \cite{\KLP}.
This becomes obvious also in various reformulations \cite{Bahr:2013ek} of the models.

The molecular structure is most explicit as well as necessary in the GFT formulation.
This is induced by the way variables and local amplitudes in each spin-foam amplitude in the perturbative expansion of a GFT state sum appear.
Moreover, since a GFT is independent of any a-priori manifold structure, molecules necessarily become the primary spacetime structure, defining the type of combinatorial complexes which are generated. 

The forthcoming section, based on \cite{\ORT}, includes a self-contained description of these structures, one that increases its utility within the group field theory framework.  
Thereby, a particular aim is to show how the most general molecules can be obtained from rather special atoms.
On the one hand, any spin-foam molecule can be obtained from only regular loopless atoms. 
On the other hand, for any boundary graph there are spin-foam molecules which are derived using only one type of atom, the simplicial atom.

Given the technical nature of this section, a synopsis of the various structures involved might help with orientation.

\begin{minipage}{0.65\textwidth}
\begin{displaymath} \nonumber
  \xymatrix{
   &\bps \ar[d] \\ 
   \bgs \ar[r]^{\bisec} &\bbgs \ar@/^/[r]^{\bulk}  &\sfas \ar@/^/[l]^{\bs} \ar@{.>}[rr]&&\sfrs\ar@/_{3pc}/@{.>}[lll]_{\bs}\\
   &\\
   \bgst_\rl \ar[r]\ar[uu]^{\pi_\rl}  &\bbgst_\rl \ar@/^/[r]\ar[uu]  &\sfast_\rl \ar@/^/[l]\ar[uu]\ar@{.>}[r] &\sfrst_\rl\ar[r]&\sfrst_{\copies,\lnb}\ar[uu]^{\Pi_{\copies,\lnb}}\ar[dd]_{\sdec}&\\
   &\bpst_{\copies}\ar[u] \ar[d] & & & &\\
 \bgst_\rs\ar[r]^{\widetilde{\bisec}}&
    \bbgst_\rs\ar@/^/[r]^{\widetilde{\bulk}}& \sfast_\rs\ar@/^/[l]^{\bst}\ar@{.>}[r]&\sfrst_\rs\ar[r]& \sfrst_{\copies,\snb}\ar@/^/@{.>}[uulll]^{\bst}\ar@/_{2pc}/[uuuu]_{\Pi_{\copies,\snb}}
 }
\end{displaymath}
\end{minipage}
\begin{minipage}{0.35\textwidth}\
\small
\begin{tabular}{cl}
 $\bgs$ 	& boundary graphs \\
$\bbgs$	& bisected graphs \\
$\sfas$ 	& spin-foam atoms \\
$\bps$	& boundary patches \\
$\sfrs$	& spin-foam molecules \\
\\
$\sim$	& labelled \\
$\copies$	& $\copies$-regular \\
$\loopless$& loopless \\
$\simplicial$& simplicial \\ 
$\textsc{nb}$& non-branching\\
\\
$\bisec$	& bisection map\\
$\bulk$	& bulk map \\
$\bs$	& boundary map\\
$\pi,\Pi$	& projection maps from \\
		& labelled to unlabelled\\
$D$		& decomposition map \\ 
\end{tabular}
\end{minipage}
  
For the definition of spin-foam molecules, one starts with a set of boundary graphs $\bgs$ that provide support for 
LQG states. 
For a graph $\bg\in\bgs$, one arrives at the corresponding bisected boundary graph  $\bbg= \bisec(\bg)\in\bbgs$ by bisecting each of its edges. The graph $\bbg$ can be augmented to 
a 2-dimensional spin-foam atom $\sfa = \bulk(\bbg) \in\sfas$.  This spin-foam atom $\sfa$ is the simplest spin-foam structure with $\bbg$ as a boundary: $\bbg= \bs\sfa$.  
Moreover, the bisected boundary graph $\bbg$ can be decomposed into boundary patches $\bp\in\bps$.  
The boundary patches are important because it is along these patches that atoms are bonded to form composite structures, known as spin-foam molecules $\sfrs$. 
The boundary of these molecules are (generically a collection of) graphs in $\bbgs$. Moreover, the molecules are the objects generated in the perturbative expansion of the group field theory. 

From the GFT  perspective, however,  one looks for as concise a way as possible to generate such structures.  
It turns out that the complexity of the GFT  generating function can be substantially reduced by considering $\copies$-regular, loopless ($\loopless$) graphs $\bgs_\rl$. 
For this set of objects, one can then follow an analogous procedure to generate $\bbgs_\rl$, $\sfas_\rl$ and $\sfrs_\rl$.    
From these, arbitrary molecules $\sfrs$ can be obtained 
in terms of an additional labelling ($\sim$), distinguishing \emph{real} from \emph{virtual} edges.
Each graph in $\bgs$ can then be represented by a class of labelled graphs in $\bgst_\rl$ in terms of a surjection $\pi_\rl:\bgst_\rl\ra\bgs$.
This map can be extended to labelled bisected graphs $\bbgst_\rl$ and atoms $\sfast_\rl$ but not the labelled molecules $\sfrst_\rl$. 
However, one can identify a subset $\sfrst_{\copies,\lnb}\In \sfrst_\rl$, for which $\pi_\rl$ can be extended to a surjection $\Pi_{\copies,\lnb}:\sfrst_{\copies,\lnb}\ra\sfrs$.  Thus, every molecule in $\sfrs$ is represented by a class of molecules in $\sfrst_{\copies,\lnb}$.

To obtain molecules with arbitrary boundary graphs $\bgs$ the strategy is the following.
One can pick out a finite subset of simplicial $\copies$-graphs $\bgst_\rs\In \bgst_\rl$,
that are based on the complete graph over $\copies+1$ vertices,
inducing $\bbgst_\rs$, $\sfast_\rs$ and $\sfrst_\rs$ follow as before. 
While $\bgst_\rs$, $\bbgst_\rs$ and $\sfast_\rs$ are finite sets, the set of simplicial spin-foam molecules $\sfrst_\rs$ is infinite and contains a subset $\sfrst_{\copies,\snb}$ whose elements reduce properly to molecules in $\sfrs$.
But the set $\sfrst_{\copies, \snb}$ does not cover $\sfrs$ itself through some surjection. 
To cover all of $\sfrs$, one needs $\sfrst_{\copies,\lnb}$. 
Still, all possible boundary graphs are covered since \textit{i}) there is a decomposition map $\sdec:\sfrst_{\copies,\lnb}\lora\sfrst_{\copies,\snb}$ and \textit{ii}) every graph or collection of graphs from $\bbgst_\rl$ arises as the boundary of some molecule in $\sfrst_{\copies,\snb}$. 
As a result, $\sfrst_{\copies,\snb}$ is sufficient to support a spin-foam dynamics for arbitrary kinematical LQG states. 

\

The forthcoming construction is separated into three parts, accordingly.
The first, \sec{atoms}, catalogues the basic building blocks or atoms, along with the set of possible bonds that may arise between pairs of atoms.  
These structures are drawn directly from those used in loop quantum gravity. 
In \sec{regular-loopless} the way these molecules are obtained from regular loopless atoms is detailed.
Then, in \sec{simplicial-structures} it is proven that any boundary graph arises as the boundary constructed from labelled simplicial atoms.

After all these technicalities I will finally discuss the relation of the 2-dimensional spin-foam molecules to combinatorial pseudo-manifolds of spacetime dimension $\std$.
%
%


\subsection{Spin-foam atoms and molecules}\label{sec:atoms} 

This part focusses on defining the structure underlying LQG and SF models, in particular the procedure by which atoms bond to form composite structures:
\begin{displaymath} \nonumber
  \xymatrix{
   &\bps \ar[d] \\ 
   \bgs \ar[r]^{\bisec} &\bbgs \ar@/^/[r]^{\bulk}  &\sfas \ar@/^/[l]^{\bs} \ar@{.>}[r]&\sfrs\ar@/_{3pc}/@{.>}[ll]_{\bs} 
   }
\end{displaymath}

\begin{defin}[{bisected/boundary graph}]
  \label{def:atomicbg}
  A \emph{boundary graph} is a double $(\V_\bg,\E_\bg)$, where $\V_\bg$ is the set of vertices and $\E_\bg$ is the multiset of edges which are unordered pairs $(\vb_1 \vb_2)\in\V_\bg\times\V_\bg$,
  subject to the condition that the graph is connected. 

  A \emph{bisected boundary graph} is a double, $(\V_{\bbg}, \E_{\bbg})$, constituting a bipartite graph with vertex partition $\V_{\bbg}=\Vb\cup\Vh$, such that the vertices $\vh\in\Vh$ are bivalent. 
\end{defin}
The set of boundary graphs is denoted by $\bgs$. 
Indeed this is just the set of connected multigraphs.
The set of bisected boundary graphs is denoted by $\bbgs$.  

\begin{remark}
  Boundary graphs may contain multi-edges (multiple edges joining two vertices), loops (edges whose two vertices coincide) and even 1-valent vertices (vertices with only one incident edge). Thus, $\bgs$ constitutes a very large set.  
  However, such graphs arise within loop quantum gravity and can be incorporated within the group field theory framework. 
  From this perspective they are the natural objects to start with.
  
  The graphs are called boundary graphs because they are meant to be the part of a spacetime boundary which quantum states have support on.
  Indeed they are defined to have the generic structure of the (dual) 1-skeleton of generalized polyhedral complexes  (\dref{generalized-polyhedral}). 
  For this reason it is necessary to define all objects in the vertex representation.
\end{remark}

\begin{proposition}
  \label{prop:bisectedcor}
  There is a bijection $\bisec: \bgs \lora \bbgs$.
\end{proposition}
\begin{proof}
The \emph{bisection map} $\bisec$ is the subdivision map $\ssub$ defined explicitly on the level of vertex represented complexes, thereby appropriately  treating loops:
given a boundary graph $\bg\in \bgs$, $\bisec$ acts on each edge $(\vb_1\vb_2)\in\E_\bg$, replacing it by a pair of edges $(\vb_1\vh),(\vb_2\vh)$, where $\vh$ is a newly created bivalent vertex effectively bisecting the original edge. 
Thus, $\bisec$ maps
\begin{description}
    \item $\V_\bg\lora  \Vb\cup\Vh$, where $\Vb=\V_\bg$ and  $\Vh$ is the set of vertices bisecting the edges of $\bg$, and
    \item $\E_\bg\lora  \bigcup_{e\in\E_\bg} \{(\vb_1 \vh),(\vb_2 \vh) \mid e =  (\vb_1\vb_2)\}$, the multiset of newly bisected edges.
  \end{description}
  This clearly results in an element of $\bbgs$ and the constructive nature of the map assures its injectivity. 

  Given a graph $\bbg\in\bbgs$, removing the vertex subset $\Vh$ and replacing the edge pair $(\vb_1\vh),(\vb_2\vh)$ by $(\vb_1\vb_2)$ results in an element $\bg\in\bgs$ such that $\bisec(\bg) = \bbg$. Thus, $\bisec$ is surjective.  
\end{proof}
  
  A graph $\bg\in\bgs$ and its bisected counterpart $\bbg\in\bbgs$ are presented in \fig{boundary}. 

\begin{figure}[htb]
  \centering
  \tikzsetnextfilename{boundary}

\begin{tikzpicture}[scale=1.3]

\draw [|->] (2,0) -- node[label=above:$\bisec$]{} (3,0);

\begin{scope}
\draw [eb] (1.5,1) circle (0.25cm);
\node [vb]		(a)	at (0,-1)		{};
\node [vb]		(b)	at (1.25,1)		{}; 
\node [vb]		(c)	at (0.5,0.5) 	{}; 
\node [vb]		(d)	at (-0.5,0.5)	{}; 
\node [vb]		(e)	at (-1.25,1)	{}; 
\node [vb]		(f)	at (1,-1)		{}; 
\foreach \i/\j in {a/b,a/c,a/d,a/e,b/c,b/e,c/d,d/e,a/f}{
 \draw [eb] (\i) -- (\j);
  }
\end{scope}

\begin{scope}[xshift=4.5cm]
\draw [eb] (1.5,1) circle (0.25cm);
\node [vb]		(a)	at (0,-1)		{};
\node [vb]		(b)	at (1.25,1)		{}; 
\node [vb]		(c)	at (0.5,0.5) 	{}; 
\node [vb]		(d)	at (-0.5,0.5)	{}; 
\node [vb]		(e)	at (-1.25,1)	{}; 
\node [vb]		(f)	at (1,-1)		{}; 
\node [vh]		(g)	at (1.75,1)		{}; 
\foreach \i/\j in {a/b,a/c,a/d,a/e,b/c,b/e,c/d,d/e,a/f}{
 \draw [eb] (\i) -- node[vh] {} (\j);
  }
\end{scope}

\end{tikzpicture}
  \caption{\label{fig:boundary} A boundary graph $\bg$ and its bisected counterpart $\bbg$.}
\end{figure}


\begin{defin}[{spin-foam atom}]
  \label{def:sf-atom}
  A \emph{spin-foam atom} is a triple, $\sfa = (\V_{\sfa},\E_{\sfa},\mathcal{F}_{\sfa})$,  of vertices, edges and faces. It is constructed from the pair $(\bbg, \bulk)$, where 
  $\bbg  \in{\bbgs}$ and $\bulk$ is a \emph{bulk map} sending $\bbg = (\V_\bbg,\E_\bbg) = (\Vb\cup\Vh,\Eb)$ to
  \begin{description}
    \item $\V_{\sfa} = \V\cup\Vb\cup\Vh$, where $\V = \{v\}$  is the one-element set of a single \emph{bulk} vertex $v$,
    \item $\E_\sfa = \E \cup \Eb \cup \Eh$, where $\E = \bigcup_{\vb\in\Vb} \{ (v \vb) \}$ and $\Eh = \bigcup_{\vh\in\Vh} \{ (v \vh ) \}$,
    \item $\mathcal{F}_{\sfa} = \bigcup_{\vh\in\Vh}\{(v\bar v\hat v) \mid (\bar v\hat v)\in \E \}$. 
  \end{description}
\end{defin}

  The set of spin-foam atoms is denoted by $\sfas$. 
  By definition they are $\ssub$-graded simplicial 2-complexes (cf. equation \eqref{vertex-sets}).
  They are called ``atoms" since they contain only a single bulk vertex. 
  For this reason, spin-foam atoms are uniquely specified by their boundary $\bbg = \bs \sfa$, and vice-versa.
  Indeed, $\bulk:\bbgs\lora\sfas$ is bijective by construction and $\bs\sfa = \bulk^{-1}(\sfa)$.

Thus, as a result of the bijective property of the maps $\bulk$ and $\bisec$, the following holds:

\fbox{
\begin{minipage}[c][][c]{0.95\textwidth}
\begin{proposition}
  \label{prop:atoms-graphs}
  The set $\sfas$ of spin-foam atoms is catalogued precisely by the set $\bgs$ of boundary graphs. 
\end{proposition}
\end{minipage}
}

An illustrative example of such a structure is presented in \fig{atom}.

\begin{figure}[htb]
  \centering
  \tikzsetnextfilename{apyr}
  \begin{tikzpicture}

\draw [|->] (-4,0.3) -- node[label=above:$\bulk$] {} (-3,0.3);
\draw [<-|] (-4,-.3) -- node[label=below:$\bs$] {} (-3,-.3);

\begin{scope}[xshift=-6.5cm, scale=1.5]
\node [vb]		(a)	at (0,-1)		{};
\node [vb]		(b)	at (1.25,1)		{}; 
\node [vb]		(c)	at (0.5,0.5) 	{}; 
\node [vb]		(d)	at (-0.5,0.5)	{}; 
\node [vb]		(e)	at (-1.25,1)	{}; 
\foreach \i/\j in {a/b,a/c,a/d,a/e,b/c,b/e,c/d,d/e}{
 \draw [eb] (\i) -- node[vh] {} (\j);
  }
\end{scope}

\begin{scope}[yshift=.5cm, scale=2]
\node [c]		(v)	at (0,-.12)		{};
\node [c]		(1)	at (-0.56,0.32) 	{}; 
\node [c]		(2)	at (-0.17,-0.07)	{}; 
\node [c]		(3)	at (0.6,0) 		{}; 
\node [c]		(4)	at (.14,.37)	{}; 
\node [c]		(5)	at (0,-.6)		{};
\node [c]		(12)	at (-.72,.32)	{};
\node [c]		(23)	at (.44,.09)	{};
\node [c]		(34)	at (.57,.52)	{};
\node [c]		(14)	at (-.3,.61)		{};
\node [c]		(15)	at (-1,-.44)		{};
\node [c]		(25)	at (-.28,-.97)	{};
\node [c]		(35)	at (1.02,-.76)	{};
\node [c]		(45)	at (.28,-.23)	{};
\foreach \i/\j in {1/2,1/4,2/3,3/4,1/5,2/5,3/5,4/5}{
 \path	[f] 	(\i) -- (\i\j) -- (\j) -- (v) -- cycle;
 }
 \foreach \i in {1,2,3,4,5}{
  \draw [e] (\i) -- (v);
  }
\foreach \i/\j in {1/2,1/4,2/3,3/4,1/5,2/5,3/5,4/5}{
 \draw 	[eh]	(v)		-- (\i\j);
 \draw	[eb] 	(\i) node[vb] {} -- (\i\j) node[vh] {};
 \draw 	[eb] 	(\j) node[vb] {} -- (\i\j) node[vh] {};
 }
\draw [e] (3) node[vb] {} -- (v) node[v] {};
%
\end{scope}

\end{tikzpicture}
  \caption{\label{fig:atom} A spin-foam atom and its (bisected) boundary graph (same colouring as in \fig{subdivision}, according to equation \eqref{vertex-sets}.}
\end{figure}



\

In the spirit of such a constructive way to define 2-complexes, it will be useful to go even one step further and identify the structure of graphs as built from 1-vertex parts:
\begin{defin}[{boundary patch}]
  A \emph{boundary patch} is a double $\bp = 
  (\V_{\bp},\E_{\bp})$, where
  \begin{description}
    \item $\V_{\bp} = \{\vb\} \cup  \Vhat_{\bp}$, $\Vhat_{\bp}\ne\emptyset$,
    \item $\E_{\bp} = \{(\vb\vh) \mid \vh\in\Vhat_{\bp}\}$ is a multiset of edges where each $(\vb\vh)$ occurs 
    at most twice.
  \end{description} 
\end{defin}
The set of boundary patches is denoted by $\bps$.

\begin{remark}
  Boundary patches are useful since they arise as the doubles $\bp_{\vb}(\bbg) = (\V_{\vb}, \E_{\vb})$, formed as the closure of the star of $\vb\in\Vb$, within $\bbg\in\bbgs$.
  Thus,
  \[
  \V_{\vb} = \{\vb\} \cup  \{\vh\mid(\vb\vh)\in\E_{\bbg}\} \quad\text{and}\quad
  \E_{\vb} = \{(\vb\vh) \mid \vh\in\V_{\bbg}\}.
  \]
  In words, a boundary patch $\bp_{\vb}(\bbg)$ is a graph containing $\vb$ itself, all bisected boundary edges containing $\vb$, 
  as well as the endpoints of these edges. 
  A simple example is depicted in \fig{star}. 

  \begin{figure}[htb]
  \centering
  \tikzsetnextfilename{patch}
  \begin{tikzpicture}[scale=1.5]
  \draw [eb] (0.75,0.5) circle (.25cm);
  \node [vb]		(vb)	at (.5,.5) 	 	{}; 
  \node [vh]		(1)	at (0,.5)		{};
  \node [vh]		(2)	at (.25,1)		{}; 
  \node [vh]		(3)	at (.25,0)		{}; 
  \node [vh]		(4)	at (1,.5)		{};
  \foreach \i/\j in {vb/1,vb/2,vb/3}
  \draw [eb] (\i) -- (\j);
  \end{tikzpicture}
  \caption{\label{fig:star} A boundary patch.}
\end{figure}
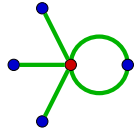

  Furthermore, when considering the boundary graphs as dual 1-skeletons of the boundary $\bs\cp$ of a spacetime $\m$-complex $\cp$, a patch is the part of the graph dual to a single $(\m-1)$-cell in $\bs\cp$.

\end{remark}

\begin{remark}[{generators}]
  \label{rem:generator}
  For some subset of patches, $\bps'\In \bps$,  the set of graphs \emph{generated} by $\bps'$, denoted $\sigma(\bps')$, is the set of all possible graphs that are composed only of patches from $\bps'$.
\end{remark}

Then, it is apparent that for every bisected boundary graph $\bbg\in\bbgs$  there is some subset $\bps'\In\bps$ generating $\bbg$ from which it directly follows that:
\nopagebreak
\begin{proposition}\label{prop:generator}
$\bbgs = \sigma(\bps)$.
\end{proposition}
\begin{remark}[{bondable}]
  \label{rem:bondable}
  Two patches, $\bp_{\vb_1}(\bbg_1)$ and $\bp_{\vb_2}(\bbg_2)$,  whether or not $\bbg_1$ and $\bbg_2$ are distinct,  
  are said to be \emph{bondable}, if $|\V_{\vb_1}|= |\V_{\vb_2}|$ and $|\E_{\vb_1}|= |\E_{\vb_2}|$ (and thus, they have the same number of loops).  
\end{remark}

\begin{defin}[{bonding map}]
  \label{def:bonding}
  A \emph{bonding map}, $\gm:\bp_{\vb_1}(\bbg_1)\lora \bp_{\vb_2}(\bbg_2)$, is a map identifying, elementwise, two bondable patches such that
  \begin{equation}
    \vb_1    \mapsto    \vb_2\,,\quad\quad
    \V_{\vb_1} \setminus \{\vb_1\}   \lora  \V_{\vb_2} \setminus \{\vb_2\}\,,\quad\quad
  \E_{\vb_1}  \lora  \E_{\vb_2}
  \end{equation}
  and with the compatibility condition that for each identified pair of vertices $\vh_1\mapsto \vh_2 $ 
  the corresponding pair of edges is identified accordingly, \ie $(\vb_1\vh_1) \mapsto (\vb_2\vh_2)$. 
  
  More particular, if $\bbg_1\ne\bbg_2$, the map $\gm$ is also called a \emph{proper} bonding map. 
  For $\bbg_1 = \bbg_2$, on the contrary, it is called a \emph{self}-bonding map. 
\end{defin}
A simple example is illustrated in \fig{bonding}, introducing a graphical notation for bonding maps. 
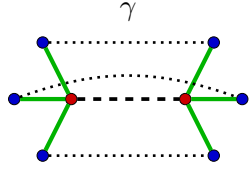
\begin{figure}[htb]
  \centering
  \tikzsetnextfilename{bonding}
 \begin{tikzpicture}[scale=1.5]
\node [vb]		(a)	at (-0.5,0)		{}; 
\node [vh]		(1)	at (-.75,.5)		{};
\node [vh]		(2)	at (-.75,-.5)	{}; 
\node [vh]		(3)	at (-1,0)		{}; 
\node [vb]		(b)	at (0.5,0)		{}; 
\node [vh]		(4)	at (.75,.5)		{};
\node [vh]		(5)	at (.75,-.5)		{}; 
\node [vh]		(6)	at (1,0)		{}; 
\foreach \i/\j in {a/1,a/2,a/3,b/4,b/5,b/6}
 \draw [eb] (\i) -- (\j);
\path	(1) 	edge [bh] node[label=above:$\gm$] {} (4)
	(2)	edge [bh]		 		(5)
	(3)	edge [bh, bend left=20]	(6)
	(a)	edge [bb]				(b);
\end{tikzpicture}
  \caption{\label{fig:bonding} A bonding map $\gm$ identifying two bondable patches.} 
\end{figure}

\begin{remark}
  The compatibility condition ensures that loops are bonded to loops. In principle, slightly more general gluing maps can be incorporated within the group field theory framework, corresponding to loop edges bonding to non-loop edges.  However, these gluings are absent from the LGQ and SF theories.  Thus, there is no motivation to include them here.
%

  Certainly, for two bondable patches, there are many bonding maps that satisfy the compatibility condition.   However, all may be obtained from a given one by applying compatible permutations to the sets $\V_{\vb_1}$ and $\E_{\vb_1}$.   
\end{remark}

\begin{defin}[{spin-foam molecule}]
  \label{def:molecule}
  $\quad$ A \emph{spin-foam molecule} is a triple $\sfr = (\V_{\sfr}, \E_{\sfr}, \mathcal{F}_{\sfr})$  constructed from a collection of spin-foam atoms quotiented by a set of bonding maps.
\end{defin}

\begin{remark}[{bonding example}]
  \label{rem:bonding-example}
  It is worth considering the simple example of two spin-foam atoms $\sfa_1$ and $\sfa_2$, with respective bisected boundary graphs $\bbg_1 = \bs\sfa_1$ and $\bbg_2 = \bs\sfa_2$ and two bondable patches $\bp_{\vb_1}(\bbg_1)$ and $\bp_{\vb_2}(\bbg_2)$. Quotienting the pair $\sfa_1$, $\sfa_2$ by a bonding map $\gm:\bp_{\vb_1}(\bbg_1)\lora \bp_{\vb_2}(\bbg_2)$ results in a spin-foam molecule $\sfr  \equiv \sharp_\gm\,\{\sfa_1,\sfa_2\}$ consisting of 
  \begin{equation}  
    \V_{\sfr} = \sharp_\gm\,\{\V_{\sfa_1},\V_{\sfa_2}\}\;,
     \quad\quad
     \E_{\sfr}=\sharp_\gm\,\{\E_{\sfa_1},\E_{\sfa_2}\}\;,
     \quad\quad
     \mathcal{F}_{\sfr}=\sharp_\gm\,\{\mathcal{F}_{\sfa_1},\mathcal{F}_{\sfa_2}\}\;,
  \end{equation}
  where $\sharp_\gm$ denotes the union of the relevant sets \emph{after} the identification of the elements of $\V_{\vb_1}\In \V_{\sfa_1}$ and $\E_{\vb_1}\In \E_{\sfa_1}$ with those of $\V_{\vb_2}\In  \V_{\sfa_2}$ and $\E_{\vb_2}\In \E_{\sfa_2}$.
  Thus, the patch 
  $\bp_{\vb_1}(\bbg_1) \cong \bp_{\vb_2}(\bbg_2)$ along which the atoms are bonded is part of the resulting molecule. 
  This is because a molecule is still a $\ssub$-graded simplicial 2-complex.
  Only, by construction, edges in the patch are now incident to two faces and therefore internal, \ie not on the boundary $\bs\sfr$.
  Under the action of $\issubs$ they are removed.
  
  Note that the bonding construction thus guarantees conditions (M2$^\star$) and (M3$^\star$) (remark \ref{rem:dual-manifold}) in the resulting 2-complex. 
  This sets the stage for an interpretation of molecules as dual to combinatorial manifolds (\sec{std-complexes}).
  
  An instance of the above example is presented in \fig{molecule}. 
\end{remark}

\begin{figure}[htb]
  \centering
  \tikzsetnextfilename{molecule}
\begin{tikzpicture}[scale=1.8]
\path (.6,0) edge [bb]	(1.5,0);

\node [c]		(v)	at (0,-.12)		{};
\node [c]		(6)	at (-0.56,0.32) 	{}; 
\node [c]		(7)	at (-0.17,-0.07)	{}; 
\node [c]		(3)	at (0.6,0) 		{}; 
\node [c]		(8)	at (.14,.37)	{}; 
\node [c]		(5)	at (0,-.6)		{};
\node [c]		(67)	at (-.72,.32)	{};
\node [c]		(37)	at (.44,.1)		{};
\node [c]		(38)	at (.7,.5)		{};
\node [c]		(68)	at (-.3,.61)		{};
\node [c]		(56)	at (-1,-.44)		{};
\node [c]		(57)	at (-.28,-.97)	{};
\node [c]		(35)	at (.8,-.7)		{};
\node [c]		(58)	at (.28,-.23)	{};
\foreach \i/\j in {6/7,6/8,3/7,3/8,5/6,5/7,3/5,5/8}{
 \path	[f] 	(\i) -- (\i\j) -- (\j) -- (v) -- cycle;
 }
 \foreach \i in {6,7,3,8,5}{
  \draw [e] (\i) -- (v);
  }
\foreach \i/\j in {6/7,6/8,3/7,3/8,5/6,5/7,3/5,5/8}{
 \draw 	[eh]	(v)	-- (\i\j);
 \draw	[eb] 	(\i) node[vb] {} -- (\i\j) node[vh] {};
 \draw 	[eb] 	(\j) node[vb] {} -- (\i\j) node[vh] {};
 }
\draw [e] (3) node[vb] {} -- (v) node[v] {};

\begin{scope}[xshift=2cm]
\node [c]		(v)	at (0,-.15)		{};
\node [c]		(1)	at (.5,0)	 	{}; 
\node [c]		(2)	at (.12,.2)		{}; 
\node [c]		(3)	at (-.5,0)		{}; 
\node [c]		(4)	at (0,-.6)		{};
\node [c]  		(12)	at (.65,.5)		{};
\node [c]	 	(13)	at (-.1,.1)		{};
\node [c]		(14)	at (.7,-.8)	{};
\node [c]		(23)	at (-.65,.5)		{};
\node [c] 		(24)	at (.3,-.3)		{};
\node [c]		(34)	at (-.6,-.7)		{};
\foreach \i/\j in {1/2,1/3,1/4,2/3,2/4,3/4}{
 \path	[f] 	(\i) -- (\i\j) -- (\j) -- (v) -- cycle;
 }
 \foreach \i in {1,2,3,4}{
  \draw [e] (\i) -- (v);
  }
\foreach \i/\j in {2/4,1/2,1/3,1/4,2/3,3/4}{
 \draw 	[eh]	(v)		-- (\i\j);
 \draw	[eb] 	(\i) node[vb] {} -- (\i\j) node[vh] {};
 \draw 	[eb] 	(\j) node[vb] {} -- (\i\j) node[vh] {};
 }
\draw [e] (1) node[vb] {} -- (v) node[v] {};
\end{scope}

\path	(38)	edge [bh] node[label=above:$\gamma$] {} (23)
	(37)	edge [bh]		(13)
	(35)	edge [bh] 		(34);

\draw [|->] (3,0) -- node[label=above:$\sharp_{\gamma}$] {} (3.7,0);

\begin{scope}[xshift=5cm]
\node [c]		(w)	at (0,-.12)		{};
\node [c]		(6)	at (-0.56,0.32) 	{}; 
\node [c]		(7)	at (-0.17,-0.07)	{}; 
\node [c]		(8)	at (.14,.37)	{}; 
\node [c]		(5)	at (0,-.6)		{};
\node [c]		(67)	at (-.72,.32)	{};
\node [c]		(37)	at (.44,.1)		{};
\node [c]	 	(13)	at (.44,.1)		{};
\node [c]		(38)	at (.7,.5)		{};
\node [c]		(23)	at (.7,.5)		{};
\node [c]		(68)	at (-.3,.61)		{};
\node [c]		(56)	at (-1,-.44)		{};
\node [c]		(57)	at (-.28,-.97)	{};
\node [c]		(35)	at (.5,-.7)		{};
\node [c]		(34)	at (.5,-.7)		{};
\node [c]		(58)	at (.28,-.23)	{};
\begin{scope}[xshift=1cm]
\node [c]		(v)	at (0,-.15)		{};
\node [c]		(1)	at (.5,0)	 	{}; 
\node [c]		(2)	at (.12,.2)		{}; 
\node [c]		(3)	at (-.5,-.13)	{}; 
\node [c]		(4)	at (0,-.6)		{};
\node [c]  		(12)	at (.65,.5)		{};
\node [c]		(14)	at (.7,-.8)		{};
\node [c] 		(24)	at (.3,-.3)		{};
\end{scope}
\foreach \i/\j in {6/7,6/8,3/7,3/8,5/6,5/7,3/5,5/8}{
 \path	[f] 	(\i) -- (\i\j) -- (\j) -- (w) -- cycle;
 }
 \foreach \i in {6,7,3,8,5}{
  \draw [e] (\i) -- (w);
  }
\foreach \i/\j in {6/7,6/8,3/7,3/8,5/6,5/7,3/5,5/8}{
 \draw 	[eh]	(w)	-- (\i\j);
 \draw	[eb] 	(\i) node[vb] {} -- (\i\j) node[vh] {};
 \draw 	[eb] 	(\j) node[vb] {} -- (\i\j) node[vh] {};
 }
\draw [e] (3) node[vb] {} -- (w) node[v] {};

\foreach \i/\j in {1/2,1/3,1/4,2/3,2/4,3/4}{
 \path	[f] 	(\i) -- (\i\j) -- (\j) -- (v) -- cycle;
 }
 \foreach \i in {1,2,3,4}{
  \draw [e] (\i) -- (v);
  }
\foreach \i/\j in {2/4,1/2,1/3,1/4,2/3,3/4}{
 \draw 	[eh]	(v)		-- (\i\j);
 \draw	[eb] 	(\i) node[vb] {} -- (\i\j) node[vh] {};
 \draw 	[eb] 	(\j) node[vb] {} -- (\i\j) node[vh] {};
 }
\draw [e] (1) node[vb] {} -- (v) node[v] {};

\end{scope}

\end{tikzpicture} 
\caption{\label{fig:molecule} The bonding $\sharp_{\gm}$ of two atoms along an identification of patches $\gm$.}
\end{figure}
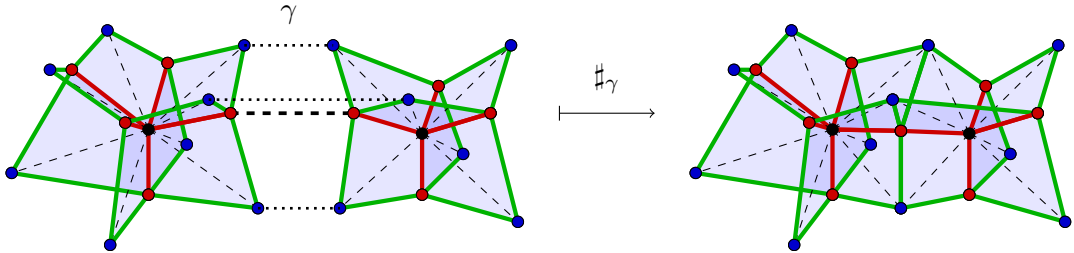

%
%


\begin{remark}[{Stranded representation}]
\label{rem:stranded}
  In matrix models, tensor models and GFT, stranded diagrams are a common alternative representation of molecules.
 
  A \emph{stranded atom} is the double, $\sda = (\coils, \reroutings)$, such that
  \begin{description}
    \item $\coils$ is a set of vertices together with a partition into subsets $\V_i$ called \emph{coils}. This set $\coils$ has an even number of elements. 
    \item $\reroutings$ is the set of \emph{reroutings}, where a rerouting is an edge, referred to quite frequently as a \emph{strand}, joining a pair of {\bf distinct} vertices in $\coils$.  
    This set of reroutings saturates the set of vertices,  in the sense that each vertex is an endpoint of exactly one strand.  
    A rerouting joining a pair of vertices in the same coil is called \emph{retracing}.
  \end{description}


  A bijection between stranded atoms and spin-foam atoms is constructed in the following way.
  Consider a spin-foam atom $\sfa=(\V_\sfa,\E_\sfa,\F_\sfa)\in\sfas$.  As was shown in proposition \ref{prop:atoms-graphs}, it is completely determined by its boundary graph $\bg = (\V_\bg, \E_\bg)\in\bgs$. From $\bg$, one constructs a stranded graph $\sda=(\coils,\reroutings)$ by \lq\lq exploding\rq\rq\ the vertices $\vb\in\V_\bg$. More precisely, for each edge $\eb = (\vb_1\vb_2)\in\E_\bg$, one creates two vertices in $\coils$ (one for each endpoint) and a strand in $\reroutings$ joining them. The subset of vertices in $\coils$ created from a given endpoint vertex in $\V_\bg$ constitutes a coil. 

  The reverse operation is equally simple. Given a stranded atom $\sda$, one constructs a boundary graph $\bg$ by identifying the vertices within each coil. 

  These operations are clearly inversely related and are illustrated for a simple example in \fig{explosion}.
 
  \begin{figure}[htb]
    \centering
     \tikzsetnextfilename{strandedvertex}
      \begin{tikzpicture}[scale=1.3]

\draw [<->] (2,0)-- (3,0);
\begin{scope}

\node [vb]		(a)	at (0,-1)		{};
\node [vb]		(b)	at (1.25,1)		{}; 
\node [vb]		(c)	at (0.5,0.5) 	{}; 
\node [vb]		(d)	at (-0.5,0.5)	{}; 
\node [vb]		(e)	at (-1.25,1)	{}; 
\foreach \i/\j in {a/b,a/c,a/d,a/e,b/c,b/e,c/d,d/e}{
 \draw [eb] (\i) -- node[vh] {} (\j);
  }
\end{scope}

\begin{scope}[xshift=5cm]
\foreach \i in {0,60,120,180}{
\draw [cs, rotate=\i] (1.1,-.3) rectangle (.9,.3);
}
\draw [cs]	 (-.4,-1.1) rectangle (.4,-.9);
\node [vs]		(14)	at (-.35,.95)	{};
\node [vs]		(15)	at (-.5,.86)		{};
\node [vs]		(12)	at (-.65,.77)	{};
\node [vs]		(21)	at (-1,.2)		{};
\node [vs]		(23)	at (-1,0)		{};
\node [vs]		(25)	at (-1,-.2)		{};
\node [vs]		(52)	at (-.3,.-1)		{};
\node [vs]		(51)	at (-.1,-1)		{};
\node [vs]		(54)	at (.1,-1)		{};
\node [vs]		(53)	at (.3,-1)		{};
\node [vs]		(34)	at (1,.2)		{};
\node [vs]		(32)	at (1,0)		{};
\node [vs]		(35)	at (1,-.2)		{};
\node [vs]		(41)	at (.35,.95)	{};
\node [vs]		(45)	at (.5,.86)		{};
\node [vs]		(43)	at (.65,.77)	{};
\path
\foreach \i/\j in {14/41,21/12,43/34}{
  (\i) edge [es,bend right=70] (\j)
  }
\foreach \i/\j in {52/25,35/53}{
  (\i) edge [es,bend right=50] (\j)
  }
\foreach \i/\j in {51/15,45/54}{
  (\i) edge [es,bend right=16] (\j)
  }
(23) edge [es] (32);
\end{scope}
\end{tikzpicture}
      \caption{\label{fig:explosion} An example of the relation between (bisected) boundary graphs and stranded diagrams. While faces of atoms (and molecules) are in 1-to-1 correspondence to bisection vertices in the graph description, in the stranded diagrams they are uniquely represented by the strands.
            }
  \end{figure}
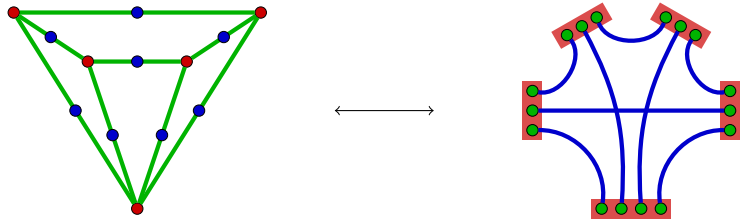
  
  Equivalently to spin-foam molecules, one can then bond stranded atoms to form stranded molecules. 
  The stranded counterparts of the relevant objects are the following:
  \begin{description}
    \item A \emph{stranded patch} is a coil $\V_i\In \coils$ along with retracings within that coil. 
    Two stranded patches are \emph{bondable} if they have the same number of vertices and the same number of retracings.  Knowledge of the retracing are necessary to capture the loop information of a boundary patch. 
    \item A \emph{stranded bonding map} identifies the vertices within two bondable stranded patches, with the compatibility condition that the vertices associated to a retracing in one patch are identified with the vertices associated to a retracing in the other. This is illustrated in \fig{StrandedBonding}.	 
  \begin{figure}[htb]
    \centering
    \tikzsetnextfilename{strandedbonding}
    \begin{tikzpicture}[scale=1.5]
  \draw [<->] (1.5,0) -- (2,0);
\node [vb]		(a)	at (-0.5,0)		{}; 
\node [vh]		(1)	at (-.75,.5)		{};
\node [vh]		(2)	at (-.75,-.5)	{}; 
\node [vh]		(3)	at (-1,0)		{}; 
\node [vb]		(b)	at (0.5,0)		{}; 
\node [vh]		(4)	at (.75,.5)		{};
\node [vh]		(5)	at (.75,-.5)		{}; 
\node [vh]		(6)	at (1,0)		{}; 
\foreach \i/\j in {a/1,a/2,a/3,b/4,b/5,b/6}{
 \draw [eb] (\i) --  (\j);
  }
\path	(1) 	edge [bh] 				(4)
	(2)	edge [bh]		 		(5)
	(3)	edge [bh, bend right=20]	(6)
	(a)	edge [bb]				(b);
  \begin{scope}[xshift=3cm]
  \foreach \i in {0,180}{
    \draw [cs, rotate=\i] (.6,-.3) rectangle (.4,.3);
    }
  \node [vs]	(11)	at (-.5,.2)	{};
  \node [vs]	(12)	at (-.5,0)	{};
  \node [vs]	(13)	at (-.5,-.2)	{};    
  \node [vs]	(21)	at (.5,.2)	{};
  \node [vs	]	(22)	at (.5,0)	{};
  \node [vs]	(23)	at (.5,-.2)	{};  
  \path
  \foreach \i in {1,2,3}{
    (1\i) edge [bh]  (2\i)
    };
  \end{scope}
  \end{tikzpicture}		
    \caption{Stranded bonding map.}
    \label{fig:StrandedBonding}
  \end{figure}
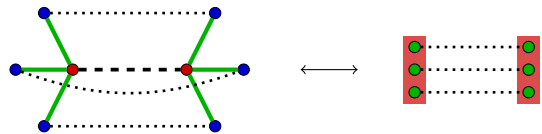
  \item A \emph{stranded molecule} is a set of stranded atoms quotiented by a set of stranded bonding maps, as drawn in \fig{StrandedMolecule}. 
  \begin{figure}[htb]
    \centering
    \tikzsetnextfilename{strandedmolecule}
    \begin{tikzpicture}[scale=1.3]

\draw [|->] (4.5,0) -- node[label=above:$\sharp_{\gamma}$] {} (5.5,0);

\begin{scope}
\foreach \i in {0,60,120,180}{
\draw [cs, rotate=\i] (1.1,-.3) rectangle (.9,.3);
}
\draw [cs]	 (-.4,-1.1) rectangle (.4,-.9);
\node [vs]		(14)	at (-.35,.95)	{};
\node [vs]		(15)	at (-.5,.86)		{};
\node [vs]		(12)	at (-.65,.77)	{};
\node [vs]		(21)	at (-1,.2)		{};
\node [vs]		(23)	at (-1,0)		{};
\node [vs]		(25)	at (-1,-.2)		{};
\node [vs]		(52)	at (-.3,.-1)		{};
\node [vs]		(51)	at (-.1,-1)		{};
\node [vs]		(54)	at (.1,-1)		{};
\node [vs]		(53)	at (.3,-1)		{};
\node [vs]		(34)	at (1,.2)		{};
\node [vs]		(32)	at (1,0)		{};
\node [vs]		(35)	at (1,-.2)		{};
\node [vs]		(41)	at (.35,.95)	{};
\node [vs]		(45)	at (.5,.86)		{};
\node [vs]		(43)	at (.65,.77)	{};
\path
\foreach \i/\j in {14/41,21/12,43/34}{
  (\i) edge [es,bend right=70] (\j)
  }
\foreach \i/\j in {52/25,35/53}{
  (\i) edge [es,bend right=50] (\j)
  }
\foreach \i/\j in {51/15,45/54}{
  (\i) edge [es,bend right=16] (\j)
  }
(23) edge [es] (32);
\end{scope}

\begin{scope}[xshift=1.5cm]
  \node [vs]	(11)	at (-.5,.2)	{};
  \node [vs]	(12)	at (-.5,0)	{};
  \node [vs]	(13)	at (-.5,-.2)	{};    
  \node [vs]	(21)	at (.5,.2)	{};
  \node [vs	]	(22)	at (.5,0)	{};
  \node [vs]	(23)	at (.5,-.2)	{};  
  \path
  \foreach \i in {1,2,3}{
    (1\i) edge [bh]  (2\i)
    };
   \node at (0,.5) {$\gamma$} ;
\end{scope}

\begin{scope}[xshift=3cm] 
  \foreach \i in {0,90,180,270}{
    \draw [cs, rotate=\i] (1.1,-.3) rectangle (.9,.3);
    }
  \node [vs]	(12)	at (1,.2)	{};
  \node [vs]	(13)	at (1,0)	{};
  \node [vs]	(14)	at (1,-.2)	{};    
  \node [vs]	(21)	at (.2,1)	{};
  \node [vs]	(24)	at (0,1)	{};
  \node [vs]	(23)	at (-.2,1)	{};    
  \node [vs]	(32)	at (-1,.2)	{};
  \node [vs]	(31)	at (-1,0)	{};
  \node [vs]	(34)	at (-1,-.2)	{};    
  \node [vs]	(41)	at (.2,-1)	{};
  \node [vs]	(42)	at (0,-1)	{};
  \node [vs]	(43)	at (-.2,-1)	{};    
  \path
  (13) edge [es] 	 (31)
  (24) edge [es]  (42)  
  \foreach \i/\j in {14/41,21/12,32/23,43/34}{
  (\i) edge [es,bend right=50]  (\j)
  };
  \end{scope}
  
\begin{scope}[xshift=7cm]
\foreach \i in {60,120,180}{
\draw [cs, rotate=\i] (1.1,-.3) rectangle (.9,.3);
}
\draw [cs]	 (-.4,-1.1) rectangle (.4,-.9);
\node [vs]		(14)	at (-.35,.95)	{};
\node [vs]		(15)	at (-.5,.86)		{};
\node [vs]		(12)	at (-.65,.77)	{};
\node [vs]		(21)	at (-1,.2)		{};
\node [vs]		(23)	at (-1,0)		{};
\node [vs]		(25)	at (-1,-.2)		{};
\node [vs]		(52)	at (-.3,.-1)		{};
\node [vs]		(51)	at (-.1,-1)		{};
\node [vs]		(54)	at (.1,-1)		{};
\node [vs]		(53)	at (.3,-1)		{};
\node [vs]		(34)	at (1,.2)		{};
\node [vs]		(32)	at (1,0)		{};
\node [vs]		(35)	at (1,-.2)		{};
\node [vs]		(41)	at (.35,.95)	{};
\node [vs]		(45)	at (.5,.86)		{};
\node [vs]		(43)	at (.65,.77)	{};
\path
\foreach \i/\j in {14/41,21/12,43/34}{
  (\i) edge [es,bend right=70] (\j)
  }
\foreach \i/\j in {52/25,35/53}{
  (\i) edge [es,bend right=50] (\j)
  }
\foreach \i/\j in {51/15,45/54}{
  (\i) edge [es,bend right=16] (\j)
  }
(23) edge [es] (32);
\end{scope}

\begin{scope}[xshift=9cm] 
  \foreach \i in {0,90,180,270}{
    \draw [cs, rotate=\i] (1.1,-.3) rectangle (.9,.3);
    }
  \node [vs]	(12)	at (1,.2)	{};
  \node [vs]	(13)	at (1,0)	{};
  \node [vs]	(14)	at (1,-.2)	{};    
  \node [vs]	(21)	at (.2,1)	{};
  \node [vs]	(24)	at (0,1)	{};
  \node [vs]	(23)	at (-.2,1)	{};    
  \node [vs]	(32)	at (-1,.2)	{};
  \node [vs]	(31)	at (-1,0)	{};
  \node [vs]	(34)	at (-1,-.2)	{};    
  \node [vs]	(41)	at (.2,-1)	{};
  \node [vs]	(42)	at (0,-1)	{};
  \node [vs]	(43)	at (-.2,-1)	{};    
  \path
  (13) edge [es] 	 (31)
  (24) edge [es]  (42)  
  \foreach \i/\j in {14/41,21/12,32/23,43/34}{
  (\i) edge [es,bend right=50]  (\j)
  };
  \end{scope}

\end{tikzpicture}
    \caption{Stranded representation of the bonding $\sharp_{\gm}$ in \fig{molecule}.}
    \label{fig:StrandedMolecule}
   \end{figure}
 \end{description}
  One particular advantage of stranded diagrams as compared to bondings of boundary graphs is that the full internal bonding structure, including the ordering of bondings of faces along patches, is represented in these diagrams in terms of the strands. This is not explicit in bondings of boundary graphs.
  
  All the additional structure to spin-foam molecules defined and discussed in the remainder of this section has a straightforward equivalent in the language of stranded molecules.

\end{remark}


\subsection{Molecules from regular, loopless graphs \label{sec:regular-loopless}}


  
  The set of spin-foam atoms $\mathfrak{A}$ is efficiently catalogued by their boundary graphs $\bgs$. However, this is a large collection of objects and thus motivates one to seek out sub-atomic building blocks that are more concisely presented but can nevertheless resemble all of $\bgs$.  

This search is divided into two stages.  This first stage examines the boundary graphs in terms of their constituent boundary patches. The set of such patches is very large. 
Thus, the first stage will focus on manufacturing a manageable set of patches, with which, nonetheless, one may encode all the boundary graphs in $\bgs$.  

Having accomplished this, the next stage examines the boundary graphs from the perspective of generating them by bonding boundary graphs from a more manageable set. 
 This part focusses on defining a projection $\pi$ which relates labelled graphs to unlabelled ones by contracting and deleting the virtual edges, as well as its restriction to the labelled, $\copies$-regular, loopless structures,
$\pi_\rl$, $\Pi_\rl$ and $\Pi_{\copies,\lnb}$, which can be shown to still map surjectively to arbitrary graphs and molecules:
\begin{displaymath} \nonumber
  \xymatrix{
   \bgs \ar[rr]^{\bisec}& &\bbgs \ar@/^/[r]^{\bulk}  &\sfas \ar@/^/
   [l]^{\bs} \ar@{.>}[rr]&&\sfrs \\
   &\\
   \bgst_\rl \ar[rr]^{\widetilde{\bisec}} \ar[uu]^{\pi_\rl} & &\bbgst_\rl  \ar@/^/[r]^{\widetilde{\bulk}}\ar[uu]  
   &\sfast_\rl \ar@/^/[l]^{\bst}\ar[uu]_{\Pi_\rl}\ar@{.>}[r] &\sfrst_\rl\ar[r] &\sfrst_{\copies,\lnb}\ar[uu]^{\Pi_{\copies,\lnb}}
 }
  \end{displaymath}

\begin{defin}[{loopless structures}]
\label{def:loopless}
Loopless structures are specified by:

A \emph{loopless boundary graph}, $\bg\in\bgs_{\loopless}$, is a $\bg=(\V_\bg,\E_\bg)\in\bgs$ without edges from any vertex $\vb\in\V_\bg$ to itself, that is for every $\vb\in\V_\bg$:  $(\vb\vb)\not\in\E_\bg$. 

Their images under the bisection map $\bisec$ and thereafter the bulk map $\bulk$ straightforwardly define \emph{loopless bisected boundary graphs} $\bbgs_{\loopless}$ and \emph{loopless atoms} $\sfas_{\loopless}$, respectively.

For a graph in $\bbgs_{\loopless}$, all of its patches are obviously loopless. In fact, the \emph{loopless patches} are uniquely specified by $\copies$, the number of edges. 
Therefore, one can call it a \emph{\copies-patch}, $\bp_\copies$, such that $\bps_{\loopless}=\bigcup_{\copies=1}^{\infty}\{\bp_{\copies}\}$. 
Moreover, $\bbgs_{\loopless} = \sigma(\bps_{\loopless})$, the loopless graphs are generated by loopless patches. 

Finally, spin-foam molecules constructed from collections of loopless atoms are called \emph{loopless spin-foam molecules} $\sfrs_{\loopless}$.
\end{defin}

Even though the complex $\clos{\issubs\sfr}$ of a loopless molecule $\sfr\in\sfrs_\loopless$, by definition, contains no loops on the boundary $\bs\sfr$, self-bondings contribute internal loops rendering $\clos{\issubs\sfr}$ still a generalized polyhedral complex.
Only loopless molecules constructed exclusively by {\bf proper} bondings (\dref{bonding}) are indeed subdivisions of polyhedral complexes (in the strict sense of \dref{polyhedral-complex}).

Another important restriction concerns the valency of boundary graph vertices: 
\begin{defin}[{$\copies$-regular structures}]
  A \emph{$\copies$-regular boundary graph} $\bg\in\bgs_{\copies}$ is a double $\bg=(\V_\bg,\E_\bg)\in\bgs$, for which every vertex $\vb\in\V_\bg$ is $\copies$-valent, \ie incident to exactly $\copies$ edges in $\E_\bg$. 
  Analogous to definition \ref{def:loopless}, the notion of their bisected counterparts $\bbgs_{\copies}$, the related \emph{$\copies$-regular atoms} $\sfas_{\copies}$, as well as $\copies$-regular molecules $\sfrs_{\copies}$, is straightforward.
\end{defin}
\begin{remark}[{$\copies$-regular and loopless}]
	Combining these restrictions, one arrives at much simpler sets of graphs $\bbgs_\rl$, atoms $\sfas_\rl$ and molecules $\sfrs_\rl$. In particular, $\bbgs_\rl = \sigma(\bp_\copies)$, \ie a single patch generates the whole set. 
Since the structure of a GFT  field is determined by a patch, these structures will play a role in single field GFTs, explained in detail in \sec{gft}.
\end{remark}

Nevertheless, the simplest GFT  is not only defined in terms of one field, but also only one interaction term of simplicial type. This motivates the following definition:

\begin{defin}[{$\copies$-simplicial molecules}]
\label{def:simplicial}
The set of \emph{$\copies$-simplicial molecules} $\sfrs_\rs$ consists of all molecules, which are bondings of the single spin-foam atom $\sfa_\rs$ obtained from the complete graph with $\copies +1$ vertices $\bg_\rs = K_{\copies+1}$ (\fig{complete}),
\[
\sfa_\rs:=\bulk(\bbg_\rs):=\bulk(\bisec(\bg_\rs)).
\]
\end{defin}

\begin{figure}
  \centering
    \tikzsetnextfilename{complete}
\begin{tikzpicture}[scale=1.5]
\node [vb]		(a)	at (0,0.82)		{};
\node [vb]		(b)	at (-1,0) 		{}; 
\node [vb]		(c)	at (-0.64,-1) 	{}; 
\node [vb]		(d)	at (0.64,-1)	{}; 
\node [vb]		(e)	at (1,0)		{}; 
\foreach \i in {a,b,c,d,e}{
  \foreach \j in {a,b,c,d,e}{
    \draw [eb] (\i) -- (\j);
    }
  }
\end{tikzpicture}
  \caption{\label{fig:complete} The complete graph over $\copies+1$ vertices ($\copies = 4$).}
\end{figure}


\begin{remark}
  The construction presented here is effectively very similar to the \emph{operator spin network} approach devised in \cite{\KKL,\KLP}, which in turn is based upon the language of \emph{operator spin-foams} \cite{Bahr:2011ey,Bahr:2012iu}.

  For clarity, it is worth setting up a small dictionary between the two descriptions. To begin, loopless boundary patches  correspond to \emph{squids}. Then \emph{squid graphs} are defined as gluings of such patches where gluing vertices of a patch to itself is allowed. 
  Thus, these are what I call bisected boundary graphs. Here, the definition of patches including loops in general is necessary from a GFT perspective.
  Moreover, the set of squid graphs considered in \cite{\KLP} corresponds to that subset of boundary graphs $\bg\in\bgs$ without 1-valent vertices $\vb\in\V_\bg$. However, this is a choice and is easily generalized. 

Squid graphs encode \emph{1-vertex spin foams} 
just as boundary graphs encode spin-foam atoms.  
After that,  1-vertex spin foams are glued together by identifying pairs of squids, just like boundary patches are bonded during the construction of spin-foam molecules.

\end{remark}



\begin{defin}[{labelled structures}]
\label{def:labelled}
  A \emph{labelled boundary graph}, $\bgt$ is a boundary graph augmented with a label for each edge drawn from the set $\{real,virtual\}$. 

The set of such graphs is denoted by $\bgst$ and is much larger than the set $\bgs$, since for a graph $\bg = (\V_\bg,\E_\bg)\in\bgs$, there are $2^{|\E_\bg|}$ labelled counterparts in $\bgst$.

  There are some trivial generalizations:
  \begin{description}
    \item \emph{Labelled bisected boundary graphs}, denoted by $\bbgt\in\bbgst$, are obtained using a bisection map $\widetilde{\bisec}$ that maintains edge labelling. 
    Thus, if $(\vb_1\vb_2)\in\bgt$ is a real (virtual) edge, then $\{\vh,(\vb_1\vh),(\vb_2\vh)\}\In \bbgt=\widetilde{\bisec}(\bgt)$ is a real (resp.\ virtual) subset, where $\vh$ is the bisecting vertex.
    \item \emph{Labelled spin-foam atoms}, denoted by $\sfat\in\sfast$, are obtained using a bulk map $\widetilde{\bulk}$, such that if $(\vb \vh)$ is real  (virtual), then so is $(v\vb \vh)$.
    In other words,  the faces inherit their label from the boundary $\bs\sfat=\bbgt$.
    \item \emph{Labelled boundary patches}, denoted by $\bpt\in\bpst$, are bonded pairwise using bonding maps $\widetilde{\gm}$ that ensure real (virtual) elements bonded to real (resp.\ virtual) elements.
    \item \emph{Labelled spin-foam molecules} $\sfrt\in\sfrst$ follow immediately with these bonding maps.
  \end{description}
\end{defin}



    One can naturally identify the unlabelled boundary graphs $\bgs$ with the subset of labelled graphs that possess only real edges $\bgst_{real}\In \bgst$. However, one would like to go further and utilize the unlabelled graphs to mark classes of labelled graphs.  From another aspect, one would think of this class of labelled graphs as encoding an underlying (unlabelled) subgraph $\bg\in\bgs$. 

    To uncover this structure, one defines certain moves on the set of labelled graphs:
\begin{defin}[{reduction moves}]
  \label{def:moves}
  Given a graph $\tilde{\bg}\in\bgst$, there are two moves that reduce the virtual edges of the graph:

  \vspace{-0.2cm}

  \begin{description}
    \item[-] given two vertices, $\vb_1$ and $\vb_2$, such that $(\vb_1\vb_2)$ is a virtual edge of $\tilde{\bg}$, a \emph{contraction move}, removes this virtual edge and identifies the vertices $\vb_1$ and $\vb_2$;
    \item[-] given a vertex $\vb$ such that $(\vb\vb)$ is a virtual loop, a \emph{deletion move} is simply the removal of this edge.
  \end{description}

  \vspace{-0.2cm}

  These inspire two counter moves:

  \vspace{-0.2cm}

  \begin{description}
    \item[-] given a vertex $\vb$, an \emph{expansion move} partitions the edges, incident at $\vb$, into two subsets. In each subset, $\vb$ is replaced by two new vertices $\vb_1$ and $\vb_2$, respectively, and a virtual edge $(\vb_1\vb_2)$ is added to the graph.\footnote{There is subtlety for loops, in that both ends are incident at $\vb$ and may (or may not) be separated by the partition.} 
    \item[-] given a vertex $\vb$, a \emph{creation move} adds a virtual loop to the graph at $\vb$.
  \end{description}
  These moves are illustrated in \fig{moves}.

  \begin{figure}
\centering
 \tikzsetnextfilename{moves}
\begin{tikzpicture}[scale=1.5]

\draw [eb, dotted] (0,.25) circle (0.25cm);
\node [vb]  (v) 	at (0,0) {};
\draw [eb]  (v) -- (-.7,0);
\draw [eb]  (v) -- (-.5,-.5);
\draw [eb]  (v) -- (.5,-.5);
\draw [eb]  (v) -- (.7,0);

\draw [|->] (1.5,0) -- node[label=above:$\pi$] {} (2,0);

\begin{scope}[xshift=3.5cm]
\node [vb]  (v) 	at (0,0) {};
\draw [eb]  (v) -- (-.7,0);
\draw [eb]  (v) -- (-.5,-.5);
\draw [eb]  (v) -- (.5,-.5);
\draw [eb]  (v) -- (.7,0);
\end{scope}

\begin{scope}[xshift=0cm,yshift=1.5cm]
\node [vb]  (v) 	at (0,0) {};
\node [vb]  (v2) 	at (.5,0) {};
\draw [eb, dotted] (v) -- (v2);
\draw [eb]  (v) -- (-.7,0);
\draw [eb]  (v) -- (-.5,-.5);
\draw [eb]  (v) -- (-.5,.5);
\draw [eb]  (v2) -- (1,-.5);
\draw [eb]  (v2) -- (1,.5);

\draw [|->] (1.5,0) -- node[label=above:$\pi$] {} (2,0);
\end{scope}

\begin{scope}[xshift=3.5cm,yshift=1.5cm]
\node [vb]  (v) 	at (0,0) {};
\draw [eb]  (v) -- (-.7,0);
\draw [eb]  (v) -- (-.5,-.5);
\draw [eb]  (v) -- (.5,-.5);
\draw [eb]  (v) -- (-.5,.5);
\draw [eb]  (v) -- (.5,.5);
\end{scope}

\end{tikzpicture}
\caption{\label{fig:moves} Contraction/expansion and deletion/creation moves.}
\end{figure}
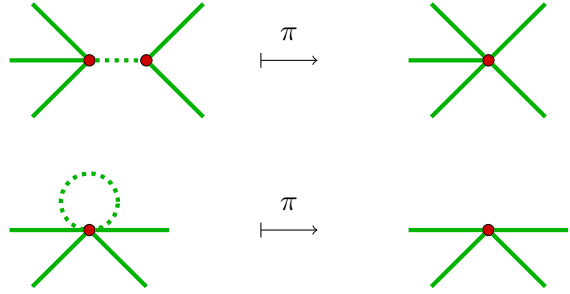

\end{defin}
\begin{remark}[{projector}]
  \label{rem:projector}
This allows one to define a projection $\pi:\bgst\lora\bgs$, which captures the complete removal of virtual edges through contraction and deletion. It is well-defined, in the sense that contraction and deletion eventually map to an element of $\bgs$ (that is, the graph remains connected) and the element $\bg\in\bgs$ acquired from $\bgt\in\bgst$ is independent of the sequence of contraction and deletion moves used to reduce the graph. In turn, this means that the $\pi^{-1}(\bg)$ partition $\bgst$ into classes. 
\end{remark}

  In fact, one is interested only in the $\copies$-regular ($\copies>2$), loopless subset $\bgst_\rl$.
One denotes the restriction of $\pi$ to these subsets as $\pi_\rl$. Note that the $\pi_\rl$ are no longer projections, since $\pi_\rl(\bgt)$ with $\bgt\in\bgst_\copies$ need no longer be $\copies$-valent and might contain loops. 

\begin{proposition}[{surjections}]
  \label{prop:surjections}
  The maps $\pi_\rl$ have the following properties:
  \begin{description}
    \item If $\copies$ is odd, the map $\pi_\rl:\bgst_\rl \lora \bgs$ is surjective.
    \item If $\copies$ is even, the map $\pi_\rl:\bgst_\rl \lora \bgs_{even}\In \bgs$ is surjective, where $\bgs_{even}$ is the subset of boundary graphs with only even-valent vertices. 
  \end{description}
\end{proposition}
\begin{proof}
  Let us first prove the statements for the restriction of $\pi$ to $\bgst_\copies$, starting with the lowest nontrivial values of $\copies$. 
  For $\copies=3$, consider an $l$-valent vertex ($l>3$) in a graph $\bg\in\bgs$.  
  Such a vertex can be expanded into a sequence of 3-valent vertices joined by a string of virtual edges.  
  A 2-valent vertex in $\bg$ can be expanded inserting a virtual double edge. 
  Finally, for a 1-valent vertex, one simply creates a virtual loop. See \fig{3decomp} for an illustration of these three cases processes.
  \begin{figure}[htb]
    \centering
    \tikzsetnextfilename{3val}
    $\quad\quad$

\begin{tikzpicture}[scale=1.5]

\draw [eb] (0.7,0.0)-- (0.0,0.0);
\draw [eb] (-0.7,0.0)-- (0.0,0.0);
\draw [eb] (0.0,0.0)-- (-0.5,-0.5);
\draw [eb] (0.5,-0.5)-- (0.0,0.0);
\draw [eb] (0.0,0.25) circle (0.25cm);

\draw [fill=red] (0.0,0.0) circle (1.5pt);

\draw [<-|] (1.5,0) -- node[label=above:$\pi_3$] {} (2,0);

\begin{scope}[xshift=4cm]

\draw [dotted, eb] (-0.5,0.0)-- (0.0,0.0);
\draw [dotted, eb] (1.0,0.0)-- (0.0,0.0);
\draw [eb] (-1.0,-0.5)-- (-0.5,0.0);
\draw [eb] (-0.5,0.0)-- (-1.2,0.0);
\draw [eb] (1.7,0.0)-- (1.0,0.0);
\draw [eb] (1.5,-0.5)-- (1.0,0.0);
\draw [shift={(0.25,0.1875)},eb]  plot[domain=-0.64:3.785,variable=\t]({1.0*0.3125*cos(\t r)+-0.0*0.3125*sin(\t r)},{0.0*0.3125*cos(\t r)+1.0*0.3125*sin(\t r)});
\begin{scriptsize}
\draw [fill=red] (-0.5,0.0) circle (1.5pt);
\draw [fill=red] (0.0,0.0) circle (1.5pt);
\draw [fill=red] (1.0,0.0) circle (1.5pt);
\draw [fill=red] (0.5,0.0) circle (1.5pt);
\draw [fill=red] (-0.5,0.0) circle (1.5pt);
\draw [fill=red] (1.0,0.0) circle (1.5pt);
\end{scriptsize}
\end{scope}
\begin{scope}[xshift=0cm,yshift=-1.5cm]
\draw [eb] (0.7,0.0)-- (0.0,0.0);
\draw [eb] (-0.7,0.0)-- (0.0,0.0);
\draw [fill=red] (0.0,0.0) circle (1.5pt);

\draw [<-|] (1.5,0) -- node[label=above:$\pi_3$] {} (2,0);
\end{scope}

%

\begin{scope}[xshift=3.5cm,yshift=-1.5cm]
\draw [eb] (-0.7,0.0)-- (0.0,0.0);
\draw [dotted, eb] (0.5,0.0)-- (0.0,0.0);
\draw [eb] (1.2,0.0)-- (0.5,0.0);
\draw [dotted, shift={(0.25,0.1875)},eb]  plot[domain=-0.64:3.785,variable=\t]({1.0*0.3125*cos(\t r)+-0.0*0.3125*sin(\t r)},{0.0*0.3125*cos(\t r)+1.0*0.3125*sin(\t r)});
\draw [fill=red] (0.5,0.0) circle (1.5pt);
\draw [fill=red] (0.0,0.0) circle (1.5pt);
\end{scope}

\begin{scope}[xshift=0cm,yshift=-3cm]
\draw [eb] (-0.7,0.0)-- (0.0,0.0);
\draw [fill=red] (0.0,0.0) circle (1.5pt);
\draw [<-|] (1.5,0) -- node[label=above:$\pi_3$] {} (2,0);
\end{scope}

\begin{scope}[xshift=3.5cm,yshift=-3cm]
\draw [eb] (-0.7,0.0)-- (0.0,0.0);
\draw [dotted, eb] (0.25,0) circle (0.25cm);
\draw [fill=red] (0.0,0.0) circle (1.5pt);
\end{scope}

\end{tikzpicture}
    \caption{\label{fig:3decomp} The expansion and creation moves to arrive at a 3-valent graph.}
  \end{figure}
  
  For $\copies$ even, note that $\pi_\copies$ maps into $\bgs_{even}$ since contraction and deletion both preserve the evenness of the vertex valency.  
  Specializing to the case of $\copies=4$, consider a graph $\bg\in\bgs_{even}$.  Once again, examining an $l$-valent vertex in $\bg$ ($l$ even), such a vertex can be expanded into a sequence of 4-valent vertices joined by a string of virtual edges.  
  For a 2-valent vertex, one may simply add a virtual loop. See \fig{4decomp} for an illustration. 
  \begin{figure}[htb]
    \centering
    \tikzsetnextfilename{4val}

\begin{tikzpicture}[scale=1.5]

\draw [eb] (0.7,0.0)-- (0.0,0.0);
\draw [eb] (-0.7,0.0)-- (0.0,0.0);
\draw [eb] (0.0,0.0)-- (-0.5,-0.5);
\draw [eb] (0.5,-0.5)-- (0.0,0.0);
\draw [eb] (0.0,0.25) circle (0.25cm);

\draw [fill=red] (0.0,0.0) circle (1.5pt);

\draw [<-|] (1.5,0) --  node[label=above:$\pi_{4}$] {} (2,0);

\begin{scope}[xshift=3.5cm]

\draw [dotted, eb] (0.5,0.0)-- (0.0,0.0);
\draw [eb] (-0.7,0.0)-- (0.0,0.0);
\draw [eb] (0.0,0.0)-- (-0.5,-0.5);
\draw [eb] (1.0,-0.5)-- (0.5,0.0);
\draw [eb] (0.5,0.0)-- (1.2,0.0);
\draw [shift={(0.25,0.1875)},eb]  plot[domain=-0.64:3.785,variable=\t]({1.0*0.3125*cos(\t r)+-0.0*0.3125*sin(\t r)},{0.0*0.3125*cos(\t r)+1.0*0.3125*sin(\t r)});
\begin{scriptsize}
\draw [fill=red] (0.0,0.0) circle (1.5pt);
\draw [fill=red] (0.5,0.0) circle (1.5pt);
\end{scriptsize}
\end{scope}

\begin{scope}[xshift=0cm,yshift=-1.5cm]
\draw [eb] (0.7,0.0)-- (0.0,0.0);
\draw [eb] (-0.7,0.0)-- (0.0,0.0);
\draw [fill=red] (0.0,0.0) circle (1.5pt);

\draw [<-|] (1.5,0) -- node[label=above:$\pi_{4}$] {} (2,0);
\end{scope}

\begin{scope}[xshift=3.5cm,yshift=-1.5cm]
\draw [eb] (0.7,0.0)-- (0.0,0.0);
\draw [eb] (-0.7,0.0)-- (0.0,0.0);
\draw [dotted, eb] (0.0,0.25) circle (0.25cm);
\draw [fill=red] (0.0,0.0) circle (1.5pt);

\end{scope}

\end{tikzpicture}
    \caption{\label{fig:4decomp} The expansion and creation moves to arrive at a 4-valent graph.}
  \end{figure}
  To generalize to arbitrary $\copies$ odd (even), then one needs only to create $(\copies-3)/2$ (resp.\ $(\copies-4)/2$) virtual loops at each vertex.

  The proof of the generalization of the statement from $\bgst_\copies$ to $\bgst_\rl$ goes as follows. 
  There exists a sequence of expansion and creation moves that effect a \emph{1-\copies-move}, \ie, for a $\copies$-valent graph vertex, the insertion of the complete graph $\bg_\rs$ (the $\copies=4$ example is depicted in \fig{onefour}).
  Consider a boundary graph $\bgt\in\bgst_\copies$.
  Applying a 1-$\copies$-move to a vertex incident to (up to $\lfloor \copies/2\rfloor$) loops removes all of them.  
  This process can be iterated until all loops in $\bgt$ are removed.
  \begin{figure}[htb]
    \centering
     \tikzsetnextfilename{14move}

\begin{tikzpicture}[scale=1.5]

\draw [eb] (0.0,0.0)-- (-0.5,-0.5);
\draw [eb] (0.5,-0.5)-- (0.0,0.0);
\draw [eb] (0.0,0.25) circle (0.25cm);
\draw [fill=red] (0.0,0.0) circle (1.5pt);

\draw [<-|] (1.5,0) --  node[label=above:$\pi_{4,\loopless}$] {} (2,0);

\begin{scope}[xshift=3.5cm,yshift=.25cm]

\draw [dotted, eb] (0.5,0.0)-- (0.0,0.0);
\draw [dotted, eb] (0.5,-0.5)-- (0.0,-0.5);
\draw [dotted, eb] (0,-0.5) -- (0.0,0.0);
\draw [dotted, eb] (0.5,0.0)-- (0.5,-0.5);
\draw [dotted, eb] (0.5,-0.5)-- (0.0,0.0);
\draw [dotted, eb] (0.5,0.0)-- (0.0,-0.5);

\draw [eb] (0.0,-0.5)-- (-0.5,-1);
\draw [eb] (1.0,-1)-- (0.5,-0.5);
\draw [shift={(0.25,0.1875)},eb]  plot[domain=-0.64:3.785,variable=\t]({1.0*0.3125*cos(\t r)+-0.0*0.3125*sin(\t r)},{0.0*0.3125*cos(\t r)+1.0*0.3125*sin(\t r)});
\begin{scriptsize}
\draw [fill=red] (0.0,0.0) circle (1.5pt);
\draw [fill=red] (0.5,0.0) circle (1.5pt);
\draw [fill=red] (0.0,-0.5) circle (1.5pt);
\draw [fill=red] (0.5,-0.5) circle (1.5pt);
\end{scriptsize}
\end{scope}

\end{tikzpicture}
   \caption{\label{fig:onefour} Use of a 1-$\copies$-move on an $\copies$-valent vertex with loop to create a loopless graph ($\copies=4$ in the example)}
  \end{figure}
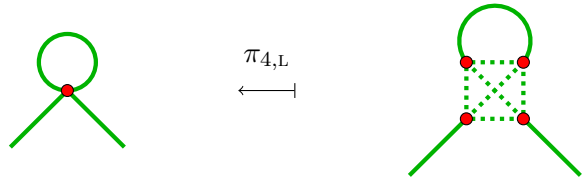
\end{proof}

\begin{remark}
  In effect, one has encoded the unlabelled graphs in $\bgs$ in terms of labelled $\copies$-regular loopless graphs in $\bgst_\rl$. The surjectivity result above implies that for $\copies$ odd (even), each graph $\bg\in\bgs$ (resp.\ $\bgs_{even}$) labels a class $\pi_\rl^{-1}(\bg)$ of graphs in $\bgst_\rl$. 
\end{remark}

\begin{remark}[{atomic reduction}]
\label{rem:atomicred}
  There is an obvious and natural extension of the contraction/expansion and deletion/creation moves, defined for $\bgt\in\bgst$ in definition \ref{def:moves}, to labelled spin-foam atoms $\sfat\in\sfast$:
  \begin{description}
    \item  {On a virtual edge $(\vb_1\vb_2)\in \bgt$ the contraction move translates to 
      \item \textit{i}) the deletion of the virtual subset
     $\{\vh$, $(\vb_1\vh)$, $(\vb_2\vh)$, $(v\vh)$,  $(v\vb_1\vh)$, $(v\vb_2\vh)\}\In \sfat$, and
      \item\textit{ii}) the identifications $\vb_1=\vb_2$ and   $(v\vb_1)= (v\vb_2)$.  
    }
    \item {On a virtual loop $(\vb\vb)\in\bgt$ the deletion move translates to the deletion of the virtual subset
     $\{\vh, (\vb\vh), (\vb\vh), (v\vh), (v\vb\vh), (v\vb\vh)\}\In \sfat$. 
    }
  \end{description}
  The expansion and creation moves are similarly extended.
  For an example see \fig{atommoves}.  
  \begin{figure}
    \centering
     \tikzsetnextfilename{apyr2}
     \begin{tikzpicture}[scale=2]

\draw [|->] (1.5,0) --  node[label=above:$\Pi_{3,\loopless}$] {} (2,0);

\begin{scope}
\node [c]		(v)	at (0,-.12)		{};
\node [c]		(1)	at (-0.56,0.32) 	{}; 
\node [c]		(2)	at (-0.17,-0.07)	{}; 
\node [c]		(3)	at (0.6,0) 		{}; 
\node [c]		(4)	at (.14,.37)	{}; 
\node [c]		(5a)	at (-.1,-.6)		{};
\node [c]		(5b)	at (.3,-.5)		{};
\node [c]		(12)	at (-.72,.32)	{};
\node [c]		(23)	at (.44,.09)	{};
\node [c]		(34)	at (.57,.52)	{};
\node [c]		(14)	at (-.3,.61)		{};
\node [c]		(15a)	at (-1,-.44)		{};
\node [c]		(25a)	at (-.28,-.97)	{};
\node [c]		(35b)	at (1.02,-.76)	{};
\node [c]		(45b)	at (.28,-.23)	{};
\node [c]		(5a5b)at (.1,-.55)	{};
\foreach \i/\j in {1/2,1/4,2/3,3/4,1/5a,2/5a,3/5b,4/5b,5a/5b}{
 \path	[f] 	(\i) -- (\i\j) -- (\j) -- (v) -- cycle;
 }
 \foreach \i in {1,2,3,4,5a,5b}{
  \draw	[e] 	(\i) -- (v);
  }
\foreach \i/\j in {1/2,1/4,2/3,3/4,1/5a,2/5a,3/5b,4/5b}{
 \draw 	[eh]	(v)		--	 (\i\j);
 \draw	[eb] 	(\i) node[vb] {} -- (\i\j) node[vh] {};
 \draw 	[eb] 	(\j) node[vb] {} -- (\i\j) node[vh] {};
 }
\draw 	[eh]	(v)		 	 -- (5a5b);
\draw[dotted,eb](5a) node[vb] {} -- (5a5b) node[vh] {};
\draw[dotted,eb](5b) node[vb] {} -- (5a5b) node[vh] {};
\draw 	[e] 	(3) node[vb] {} -- (v) node[v] {};
\end{scope}
\begin{scope}[xshift=3.5cm]
\node [c]		(v)	at (0,-.12)		{};
\node [c]		(1)	at (-0.56,0.32) 	{}; 
\node [c]		(2)	at (-0.17,-0.07)	{}; 
\node [c]		(3)	at (0.6,0) 		{}; 
\node [c]		(4)	at (.14,.37)	{}; 
\node [c]		(5)	at (0,-.6)		{};
\node [c]		(12)	at (-.72,.32)	{};
\node [c]		(23)	at (.44,.09)	{};
\node [c]		(34)	at (.57,.52)	{};
\node [c]		(14)	at (-.3,.61)		{};
\node [c]		(15)	at (-1,-.44)		{};
\node [c]		(25)	at (-.28,-.97)	{};
\node [c]		(35)	at (1.02,-.76)	{};
\node [c]		(45)	at (.28,-.23)	{};
\foreach \i/\j in {1/2,1/4,2/3,3/4,1/5,2/5,3/5,4/5}{
 \path	[f] 	(\i) -- (\i\j) -- (\j) -- (v) -- cycle;
 }
 \foreach \i in {1,2,3,4,5}{
  \draw [e] (\i) -- (v);
  }
\foreach \i/\j in {1/2,1/4,2/3,3/4,1/5,2/5,3/5,4/5}{
 \draw 	[eh]	(v)		-- 	(\i\j);
 \draw	[eb] 	(\i) node[vb] {} -- (\i\j) node[vh] {};
 \draw 	[eb] 	(\j) node[vb] {} -- (\i\j) node[vh] {};
 }
\draw [e] (3) node[vb] {} -- (v) node[v] {};
\end{scope}

\end{tikzpicture}
    \caption{\label{fig:atommoves} A contraction move on an atom.}
  \end{figure}
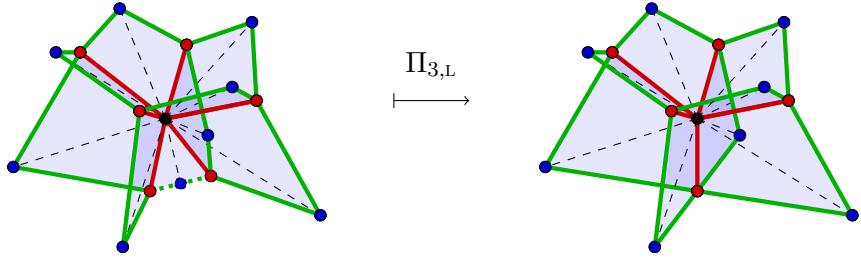

  Quite trivially, one may extend the map $\pi_\rl$ of proposition \ref{prop:surjections} to $\sfat\in\sfast_\rl$. 
  This map 
  \[
  \Pi_\rl := \bulk \circ\bisec\circ\pi_\rl\circ\widetilde{\bisec}^{-1}\circ\widetilde{\bulk}^{-1}:\sfast_\rl\lora\sfas
  \]
  is surjective.  
  Thus, each $\sfa\in\sfas$ marks a non-trivial class $\Pi_\rl^{-1}(\sfa)\in\sfast_\rl$.
\end{remark}

\begin{remark}[{molecule reduction}]
  \label{rem:moleculered}
  While the bonding of atoms in $\sfast$ just follows \dref{labelled},  the reduction of a labelled spin-foam molecule possesses certain subtleties. 
  Within a spin-foam molecule, two scenarios arise for a virtual vertex $\vh\in\sfrt$:
  \begin{description}
    \item[$\vh\notin\bst\sfrt$:]{
      Consider a virtual vertex $\vh$ with $\ks\in\N$ virtual edges and $2\ks$ virtual faces incident at $\vh$, here denoted by $(\vb_1\vh)$, ..., $(\vb_\ks\vh)$ and $(v_{12}\vb_1\vh),(v_{12}\vb_2\vh)$, ..., $(v_{\ks 1}\vb_\ks \vh)$, $(v_{\ks1}\vb_1\vh)$, respectively. 
      Following the rules laid out in remark \ref{rem:atomicred}, a contraction move applied to that virtual substructure 
      \begin{enumerate}
        \item [\textit{i})] deletes $\{\vh\}$, as well as all edges and faces incident at $\vh$ and 
        \item [\textit{ii})] identifies $\vb\equiv\vb_i$ and pairwise $(v_{ii+1}\vb) \equiv (v_{ii+1}\vb_i) = (v_{ii+1}\vb_{i+1})$, for all $i\in\{1,\dots,\ks\}$.
      \end{enumerate}
      }
      This contraction only behaves well when $\ks = 2$, that is, there are two virtual edges of type $(\vb\vh)$ incident at $\vh$ (\fig{moleculecontraction}). 
      As illustrated in \fig{moleculecontraction2}, for other values of $\ks$, the resulting structure does not lie within $\sfrst$ and therefore ultimately, it lies  outside $\sfrs$; the reason is that in a $\sfr\in\sfrs$ there are precisely two edges of type $e=(v\vb)\in\E$ incident at each vertex $\vb\notin\bs\sfr$ while in the reduction of a $\sfrt\in\sfrst$ in general there occur any $\ks\ge 2$ edges at a vertex $\vb\notin\bst\sfrt$.
   
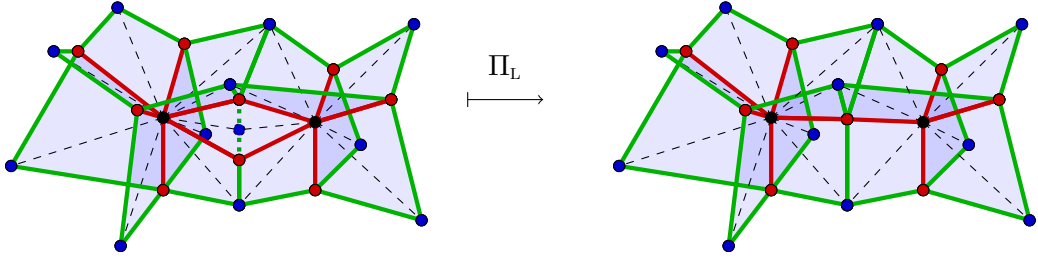
\begin{figure}
    \centering
        \tikzsetnextfilename{moleculecontraction}
        \begin{tikzpicture}[scale=2]

\draw [|->] (2,0) --  node[label=above:$\Pi_{\loopless}$] {} (2.5,0);

\node [c]		(w)	at (0,-.12)		{};
\node [c]		(6)	at (-0.56,0.32) 	{}; 
\node [c]		(7)	at (-0.17,-0.07)	{}; 
\node [c]		(8)	at (.14,.37)	{}; 
\node [c]		(5)	at (0,-.6)		{};
\node [c]		(67)	at (-.72,.32)	{};
\node [c]		(3a7)	at (.44,.1)		{};
\node [c]	 	(13a)	at (.44,.1)		{};
\node [c]		(3a8)	at (.7,.5)		{};
\node [c]		(23a)	at (.7,.5)		{};
\node [c]		(68)	at (-.3,.61)		{};
\node [c]		(56)	at (-1,-.44)		{};
\node [c]		(57)	at (-.28,-.97)	{};
\node [c]		(3b5)	at (.5,-.7)		{};
\node [c]		(3b4)	at (.5,-.7)		{};
\node [c]		(58)	at (.28,-.23)	{};
\begin{scope}[xshift=1cm]
\node [c]		(v)	at (0,-.15)		{};
\node [c]		(1)	at (.5,0)	 	{}; 
\node [c]		(2)	at (.12,.2)		{}; 
\node [c]		(3a)	at (-.5,0)		{}; 
\node [c]		(3b)	at (-.5,-.4)		{};
\node [c]		(3a3b)at (-.5,-.2)	{}; 
\node [c]		(4)	at (0,-.6)		{};
\node [c]  		(12)	at (.65,.5)		{};
\node [c]		(14)	at (.7,-.8)		{};
\node [c] 		(24)	at (.3,-.3)		{};
\end{scope}

\path [eb,dotted] (3a) -- node[vh] {} (3b);

\foreach \i/\j in {6/7,6/8,3a/7,3a/8,5/6,5/7,3b/5,5/8,3a/3b}{
 \path	[f] 	(\i) -- (\i\j) -- (\j) -- (w) -- cycle;
 \draw 	[eh]	(w) -- (\i\j);
 }
 \foreach \i in {6,7,3a,3b,8,5}{
  \draw [e] (\i) -- (w);
  }
\foreach \i/\j in {6/7,6/8,3a/7,3a/8,5/6,5/7,3b/5,5/8}{
 \draw	[eb] 	(\i) node[vb] {} -- (\i\j) node[vh] {};
 \draw 	[eb] 	(\j) node[vb] {} -- (\i\j) node[vh] {};
 }
\draw [e] (3a) node[vb] {} -- (w) node[v] {};
\draw [e] (3b) node[vb] {} -- (w) node[v] {};

\foreach \i/\j in {1/2,1/3a,1/4,2/3a,2/4,3b/4,3a/3b}{
 \path	[f] 	(\i) -- (\i\j) -- (\j) -- (v) -- cycle;
 \draw 	[eh]	(v) -- (\i\j);
 }
 \foreach \i in {1,2,3a,3b,4}{
  \draw [e] (\i) -- (v);
  }
\foreach \i/\j in {1/2,1/3a,1/4,2/3a,2/4,3b/4}{
 \draw	[eb] 	(\i) node[vb] {} -- (\i\j) node[vh] {};
 \draw 	[eb] 	(\j) node[vb] {} -- (\i\j) node[vh] {};
 }
\draw [e] (1) node[vb] {} -- (v) node[v] {};

\begin{scope}[xshift=4cm]
\node [c]		(w)	at (0,-.12)		{};
\node [c]		(6)	at (-0.56,0.32) 	{}; 
\node [c]		(7)	at (-0.17,-0.07)	{}; 
\node [c]		(8)	at (.14,.37)	{}; 
\node [c]		(5)	at (0,-.6)		{};
\node [c]		(67)	at (-.72,.32)	{};
\node [c]		(37)	at (.44,.1)		{};
\node [c]	 	(13)	at (.44,.1)		{};
\node [c]		(38)	at (.7,.5)		{};
\node [c]		(23)	at (.7,.5)		{};
\node [c]		(68)	at (-.3,.61)		{};
\node [c]		(56)	at (-1,-.44)		{};
\node [c]		(57)	at (-.28,-.97)	{};
\node [c]		(35)	at (.5,-.7)		{};
\node [c]		(34)	at (.5,-.7)		{};
\node [c]		(58)	at (.28,-.23)	{};
\begin{scope}[xshift=1cm]
\node [c]		(v)	at (0,-.15)		{};
\node [c]		(1)	at (.5,0)	 	{}; 
\node [c]		(2)	at (.12,.2)		{}; 
\node [c]		(3)	at (-.5,-.13)	{}; 
\node [c]		(4)	at (0,-.6)		{};
\node [c]  		(12)	at (.65,.5)		{};
\node [c]		(14)	at (.7,-.8)		{};
\node [c] 		(24)	at (.3,-.3)		{};
\end{scope}
\foreach \i/\j in {6/7,6/8,3/7,3/8,5/6,5/7,3/5,5/8}{
 \path	[f] 	(\i) -- (\i\j) -- (\j) -- (w) -- cycle;
 }
 \foreach \i in {6,7,3,8,5}{
  \draw [e] (\i) -- (w);
  }
\foreach \i/\j in {6/7,6/8,3/7,3/8,5/6,5/7,3/5,5/8}{
 \draw 	[eh]	(w)	-- (\i\j);
 \draw	[eb] 	(\i) node[vb] {} -- (\i\j) node[vh] {};
 \draw 	[eb] 	(\j) node[vb] {} -- (\i\j) node[vh] {};
 }
\draw [e] (3) node[vb] {} -- (w) node[v] {};

\foreach \i/\j in {1/2,1/3,1/4,2/3,2/4,3/4}{
 \path	[f] 	(\i) -- (\i\j) -- (\j) -- (v) -- cycle;
 }
 \foreach \i in {1,2,3,4}{
  \draw [e] (\i) -- (v);
  }
\foreach \i/\j in {2/4,1/2,1/3,1/4,2/3,3/4}{
 \draw 	[eh]	(v)		-- (\i\j);
 \draw	[eb] 	(\i) node[vb] {} -- (\i\j) node[vh] {};
 \draw 	[eb] 	(\j) node[vb] {} -- (\i\j) node[vh] {};
 }
\draw [e] (1) node[vb] {} -- (v) node[v] {};

\end{scope}

\end{tikzpicture}
    \caption{\label{fig:moleculecontraction} Contraction move with respect to a vertex $\vh$ incident to two virtual edges in a molecule.}
 \end{figure} 
      
  \begin{figure}[htb]
    \centering
      \tikzsetnextfilename{moleculecontraction2}
      \begin{tikzpicture}

\draw [|->] (2,0) --  node[label=above:$\sharp_{\{\gamma_1,\gamma_2,\gamma_3\}}$] {} (3.5,0);
\draw [|->] (6.5,0) --  node[label=above:$\Pi_{\loopless}$] {} (7.5,0);

\begin{scope}[xshift=.5cm,yshift=.4cm]
\node [c]		(fu)	at (0,0)		{};
\node [c]		(u)	at (1,.58)	 	{};  
\node [c]		(uv)	at (0,.58)		{};
\node [c]		(uw)	at (.5,-.3)		{};
\end{scope}
\begin{scope}[xshift=-.5cm,yshift=.4cm]
\node [c]		(fv)	at (0,0)		{};
\node [c]		(v)	at (-1,.58)		{};
\node [c]		(vu)	at (0,.58)		{};
\node [c]	 	(vw)	at (-.5,-.3)		{};
\end{scope}
\begin{scope}[yshift=-.5cm]
\node [c]		(fw)	at (0,0)		{};
\node [c]		(w)	at (0,-1.16)	{}; 
\node [c]		(wu)	at (.5,-.3)		{};
\node [c]		(wv)	at (-.5,-.3)		{};
\end{scope}
\foreach \i/\j/\k in {u/v/w,v/w/u,w/u/v}{
\path		[f] 	(\i) 	-- (\i\j) -- (f\i) -- (\i\k) -- cycle;
\draw[dotted,eb](f\i) 	-- (\i\j);
\draw[dotted,eb](f\i) 	-- (\i\k);
\path		(\j\k) 	edge [bb]  (\k\j)
		(f\i)	edge [bh] 	(f\j);
\draw	[e]	(\i)	-- (\i\j) node[vb]{};
\draw	[e]	(\i)  node[v]{} 	-- (\i\k) node[vb]{};
\draw	[eh]	(\i) 	-- (f\i) node[vh]{};
}

\begin{scope}[xshift=5cm]
\node [c]		(fu)	at (0,0)		{};
\node [c]		(u)	at (1,.58)	 	{};  
\node [c]		(uv)	at (0,.58)		{};
\node [c]		(uw)	at (.5,-.3)		{};
\node [c]		(fv)	at (0,0)		{};
\node [c]		(v)	at (-1,.58)		{};
\node [c]		(vu)	at (0,.58)		{};
\node [c]	 	(vw)	at (-.5,-.3)		{};
\node [c]		(fw)	at (0,0)		{};
\node [c]		(w)	at (0,-1.16)	{}; 
\node [c]		(wu)	at (.5,-.3)		{};
\node [c]		(wv)	at (-.5,-.3)		{};
\foreach \i/\j/\k in {u/v/w,v/w/u,w/u/v}{
\path		[f] 	(\i) 	-- (\i\j) -- (f\i) -- (\i\k) -- cycle;
\draw[dotted,eb](f\i) 	-- (\i\j);
\draw[dotted,eb](f\i) 	-- (\i\k);
\draw	[e]	(\i)	-- (\i\j) node[vb]{};
\draw	[e]	(\i)  node[v]{} 	-- (\i\k) node[vb]{};
\draw	[eh]	(\i) 	-- (f\i) node[vh]{};
 }
\end{scope}

\begin{scope}[xshift=9cm]
\node [vb]		(f)	at (0,0)		{};
\node [v]		(u)	at (1,.58)	 	{};  
\node [v]		(v)	at (-1,.58)		{};
\node [v]		(w)	at (0,-1.16)	{}; 
\foreach \i in {u,v,w}{
\draw [e] (f) -- (\i);
}
\end{scope}
\end{tikzpicture}
      \caption{Contraction move with respect to a vertex $\vh$ adjacent to three virtual edges as consequence of three bondings. The contraction identifies three boundary vertices and the resulting vertex is incident to three bulk edges. 
      This is not possible in a molecule $\sfr\in\sfrs$.}
      \label{fig:moleculecontraction2}
  \end{figure}
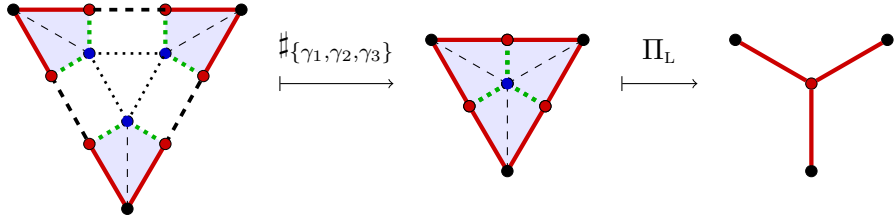

If $\vh$ corresponds to an internal loop edge, $\ks=2$ follows automatically since both half-edges incident to $\vh$ are part of the same patch along which it is bonded.

    \item[$\vh\in\bst\sfrt$:]  In this case, a similar argument reveals the necessity for precisely one virtual edge of type $(\vb\vh)$ incident at $\vh$ to obtain a molecule $\sfr\in\sfrs$ upon reduction.
  \end{description}
  
  As already noted in remark \ref{rem:bonding-example}, an important motivation for the molecule property which is essential here (\ie that each $\vb$ is incident to at most two edges $(v\vb)$ in $\sfr$) is the resulting manifold property (M3$^\star$) for $\issubs\sfr$.
  Restricting to $\sfrt\in\sfrst_\rl$ with $\ks\le2$ exactly guarantees that, for the reduced molecule $\sfr = \Pi_\rl \sfrt$, 1-cells in $\issubs\sfr$ are indeed edges (or half-edges).
  Thus the complex is the dual of a non-branching complex. 
  If $\ks>2$, on the other hand, the resulting complex would be dual to a branching one which cannot be a pseudo-manifold.
\end{remark}
 
\begin{remark}[{non-branching molecules}]
\label{rem:dually} 
  According to \rem{moleculered}, one is not interested in the whole of $\sfrst_\rl$, 
  but rather in the subset that possesses vertices $\vh\in\Vh$ with at most two 
  virtual edges incident at a vertex $\vh$,
  because only these correspond to non-branching complexes under reduction $\Pi_\rl$.
  Thus, this set is denoted by $\sfrst_{\copies,\lnb}$. 
  Fortunately, the expansion and creation moves act each time on a single vertex $\vb$, so that one may define a surjective map $\Pi_{\copies,\lnb}:\sfrst_{\copies,\lnb}\lora\sfrs$.  In words, 
  \begin{description}
  \item each spin-foam molecule $\sfr\in\sfrs$ is represented in $\sfrst_{\copies,\lnb}$.
  \end{description}
\end{remark}

%
\begin{remark}
  \label{prop:generated}
  Anticipating the GFT  application, it should be emphasized that the whole construction is based only on a single kind of labelled patches,  the $\copies$-patch.
In the  labelled case this is not unique but there are $2^\copies$ $\copies$-patches and one denotes their set as $\bpst_\copies$.
Thus  $\bbgst_\rl = \sigma(\bpst_\copies)$.
\end{remark}


\subsection{Molecules from simplicial structures}
\label{sec:simplicial-structures} 

Finally, one can show that it is even possible to use only molecules obtained from bonding labelled atoms of simplicial type to recover unlabelled molecules with any graph in $\bgs$ as boundary in terms of reduction:
\begin{displaymath} \nonumber
  \xymatrix{
   \bgs \ar[r]^{\bisec} &\bbgs &&&\sfrs\ar@/_/@{.>}[lll]_{\bs}\\
   \bgst_\rl \ar[r]^{\widetilde{\bisec}} & \bbgst_\rl \ar[u] & & &\sfrst_{\copies,\lnb}\ar[dd]_{\sdec}&\\
   &\bpst_{\copies}\ar[u] \ar[d] & & & &\\
 \bgst_\rs\ar[r]^{\widetilde{\bisec}}&
    \bbgst_\rs\ar@/^/[r]^{\widetilde{\bulk}}& \sfast_\rs\ar@/^/[l]^{\bst}\ar@{.>}[r]&\sfrst_\rs\ar[r]& \sfrst_{\copies,\snb}\ar@/^/@{.>}[uulll]^{\bst}\ar@/_{2pc}/[uuu]_{\Pi_{\copies,\snb}}
 }
\end{displaymath}

In proposition \ref{prop:surjections}, it was shown that all boundary graphs can be encoded in terms of labelled, $\copies$-regular, loopless graphs.  Moreover, from the spin-foam point of view these graphs occur as the boundaries of labelled spin-foam atoms $\sfast_\rl$ \emph{and} labelled spin-foam molecules $\sfrst_{\copies,\lnb}$ (see remarks \ref{rem:moleculered} and \ref{rem:dually}).  However, one would also like to show that all possible boundary graphs arise as the boundary of molecules composed of atoms drawn from a small finite set of types. 

This can be achieved using the labelled version of simplicial graphs and atoms.
\begin{remark}[{labelled $\copies$-simplicial structures}]
Due to the label on each edge, there are $2^{(\copies+1)(\copies+2)/2}$ \emph{labelled $\copies$-simplicial boundary graphs}, denoted $\bgst_\rs$. 

Through the maps $\widetilde{\bisec}$ and $\widetilde{\bulk}$, defined in \dref{labelled}, one can rather easily obtain the \emph{labelled bisected $\copies$-simplicial graphs} $\bbgst_\rs$ and \emph{labelled $\copies$-simplicial atoms} $\sfast_\rs$, respectively. 

Furthermore, label-preserving bonding maps $\widetilde{\gm}$ give rise to \emph{labelled $\copies$-simplicial molecules} $\sfrst_\rs$, and their subclass $\sfrst_{\copies,\snb}$ according to remark \ref{rem:dually}.
\end{remark}

\begin{remark}[{atoms from patches}]
  \label{rem:atoms-from-patches}
  One can use a $\copies$-patch $\bpt_\copies\in\bpst_\copies$ as the foundation for a bisected $\copies$-simplicial graph $\bbgt\in\bbgst_\rs$ in the following manner:  
  \begin{description} 
    \item  A $\copies$-patch consists of a single $\copies$-valent vertex $\vb$,  1-valent vertices $\vh^i$ with $i\in I_{\vb}$ a $\copies$-element index set, and labelled edges $(\vb\vh^i)$.  
    \item For each $i$, one creates a new vertex $\vb^i$, along with an edge $(\vb^i\vh^i)$ with the same label as $(\vb\vh^i)$. 
    \item For each pair of new vertices $\vb^i$ and $\vb^j$ with $i\neq j$, one creates a new vertex $\vh^{ij}$, along with a pair of {\bf real} 
    edges $(\vb^i\vh^{ij})$ and $(\vb^j\vh^{ij})$.
  \end{description}
  The result is a $\copies$-simplicial graph 
  denoted $\bbgt_{\vb}$.
  Also, $\bpt(\vb) \equiv \bpt_{ \vb}(\bbgt_{\vb})$ denotes the original $\copies$-patch, and the new patches are $\bpt_{\vb_i}(\bbgt_{\vb_i})$ for $i\in I$.
\end{remark}

The aim is summarized in the statement:

\fbox{
  \begin{minipage}[c][][c]{0.97\textwidth}
\begin{proposition}
  \label{prop:all-graphs}
  Every graph in $\bgst_\rl$ arises as the boundary graph of a non-branching molecule composed of simplicial $\copies$-atoms.  
\end{proposition}
\end{minipage}
}
\begin{proof}
  The basic argument is fairly straightforward and goes as follows: given a graph $\bgt\in\bgst_\rl$, one bisects it and thereafter cuts it into its constituent patches; one uses remark \ref{rem:atoms-from-patches} to construct a $\copies$-simplicial atom from each patch; one supplements this set of atoms with bonding maps that yield a molecule with $\bgt$ as boundary.  
The procedure is also sketched in \fig{decomposition}.
\begin{figure}[htp]
    \centering
        \tikzsetnextfilename{decomposition}
        \begin{tikzpicture}[scale=1.5]
\node [vb,label=left:$\vb_i$]	(1)	at (0,-.2)	{};
\node [vh,label=above:$\vh_{ij}$](12)	at (1,-.1)	{};
\node [vh]					(13)	at (.4,.4)	{};
\node [vh]					(14)	at (.5,-.5)	{};
\node [vb,label=right:$\vb_j$]	(2)	at (2,0)	{};
\node [vh]					(21)	at (1,-.1)	{};
\node [vh]					(23)	at (1.6,.5)	{};
\node [vh]					(24)	at (1.5,-.4)	{};
\foreach \i/\j in {1/2,1/3,1/4,2/1,2/3,2/4}{
    \draw [eb] (\i) -- (\i\j);
    }
\draw [->] (2.8,0) -- (3.5,0);

\begin{scope}[xshift=5cm,yshift=-.2cm]
  \node [vb,label=left:$\vb_i$]	(3)	at (-1.1,-.2){};
  \node [vb]				(4)	at (0,-.8)	{};
  \node [vb,label=above:$\vb_i^j$](1)	at (1,0) 	{};
  \node [vb]				(2)	at (.2,1)	{};
  \foreach \i/\j in {1/2,2/4,3/4,2/3,1/3,1/4}{
    \draw [eb] (\i) -- node[vh] {} (\j);
    }
\end{scope}

\begin{scope}[xshift=8cm]
  \node [vb,label=above:$\vb_j^i$](5)	at (-1.1,-.2){};
  \node [vb]				(4)	at (0,-.8)	{};
  \node [vb,label=right:$\vb_j$]	(6)	at (1,0) 	{};
  \node [vb]				(2)	at (.2,1)	{};
  \foreach \i/\j in {6/2,2/4,5/4,2/5,6/5,6/4}{
    \draw [eb] (\i) -- node[vh] {} (\j);
    }
\end{scope}

\path		(1) 	edge [bb] (5) 	(5,-.3) 
			edge [bh, bend right=15] node[label=below:$\gamma_{ij}$] {} (8,-.1);
\node at (4.9,-.05) {$\vh_i^j$};
\node at (7.9,.15) {$\vh_j^i$};

\end{tikzpicture}
    \caption{\label{fig:decomposition} Decomposition of an atom with boundary graph $\bgt\in\bgst_\rl$ into simplicial atoms, sketched for the patches of two connected vertices in $\bgt$ and $\copies=3$.}
 \end{figure}
 
 More precisely, consider a labelled, loopless, $\copies$-regular graph $\bgt\in\bgst_\rl$. 
   It is useful to index the vertex set by $\vb_i$ with $i\in\{1,\dots,|\V_{\bgt} |\}$. 
   This induces an index for the edges; an edge joining $\vb_i$ to $\vb_j$ is indexed by $e_{ij}^{(a)}$, where a non-trivial index $(a)$ arises should multiple edges join the two vertices. 
  The graph $\bgt$ has a bisected counterpart $\widetilde\bisec(\bgt) = \bbgt = (\V_{\bbgt}, \E_{\bbgt})$.  
  The vertex set is $\V_{\bbgt} = \Vb\cup\Vh$, where $\Vb = \V_{\bgt}$ and $\Vh$ is the set of bisecting vertices.
  A vertex in $\Vh$ is indexed by $\vh_{ij}^{(a)}$ if it bisects the edge $e_{ij}^{(a)}$ of $\bgt$.
   
  The boundary patches in $\bbgt$ are $\bpt_{\vb_i}(\bbgt)$ with $i\in\{1,\dots, |\Vb|\}$. The patch $\bpt_{\vb_i}(\bbgt)$ is comprised of the vertex $\vb_i$, the  $\copies$ vertices $\vh_{ij}^{(a)}$ and $\copies$ edges $(\vb_i\vh_{ij}^{(a)})$. 
  The indices of type $j(a)$, attached to the $\copies$ elements $\vh_{ij}^{(a)}$, form a $\copies$-element index set $I_{\vb_i}$. 
  Each bisecting vertex $\vh_{ij}^{(a)}\in\Vh$ is shared by precisely two patches.
  
  Now one cuts the graph along each bisecting vertex and considers each patch in isolation.
  This cutting procedure sends each $\bpt_{\vb_i}(\bbgt) \lora\bpt(\vb_i)$, where $\bpt(\vb_i)$ is a 
   $\copies$-patch
 comprising of a vertex $\vb_i$, $\copies$ vertices $\vh_{i}^{j(a)}$ and $\copies$ edges $(\vb_i\vh_{i}^{j(a)})$.
 Thus, after cutting, a bisecting vertex $\vh_{ij}^{(a)}$ is represented by $\vh_{i}^{j(a)}$ in $\bpt(\vb_i)$ and $\vh_{j}^{i(a)}$ in $\bpt(\vb_j)$. 
  For each patch $\bpt(\vb_i)$, the $\copies$ superscript indices $j(a)$ are the indexing set $I_{\vb_i}$. 
  Thus, one may use remark \ref{rem:atoms-from-patches} to construct, from $\bpt(\vb_i)$, a simplicial  $\copies$-graph $\bbgt_{\vb_i}$  and there after a simplicial  $\copies$-atom $\sfat_{\vb_i}$.

  Through this process, one obtains a set of simplicial  $\copies$-atoms, $\sfat_{\vb_i}$ with $i\in\{1,\dots,|\Vb|\}$.  This set is denoted by $\sfast_{\Vb}$, since the atoms are in one-to-one correspondence with the vertices $\Vb$ of $\bbgt$.   They will be used to form a spin-foam molecule whose bisected boundary graph is $\bbgt$.
 

  For each pair $\vb_{i}^{j(a)}\in\bbgt_{\vb_i}$, $\vb_j^{i(a)}\in\bbgt_{\vb_j}$, define a bonding map
  \begin{eqnarray}
\gm_{ij}^{(a)}:\bpt_{\vb_{i}^{j(a)}}(\bbgt_{\vb_i}) & \lora & \bpt_{ \vb_j^{i(a)}}(\bbgt_{\vb_j})\\
\vb_i^{j(a)} & \lora & \vb_j^{i(a)} \\
\vh_i^{j(a)} & \lora & \vh_j^{i(a)}
\end{eqnarray}
while the remaining $\copies-1$ vertices in each patch are paired in an arbitrary way:\footnote{As an aside, the bonding maps are specified only up to permutations of these $\copies-1$ vertex pairings, leading to ${\copies-1 \choose 2}$ choices for each bonding map. However, the resulting spin-foam molecules possess the same boundary.}
\begin{eqnarray}
  \left\{ \vh_i^{j(a)k(b)} : k(b)\in I_{\vb_i} \setminus \{j(a)\}\right\} & \lora & \left\{ \vh_j^{i(a)l(c)} : l(c)\in I_{\vb_j} \setminus \{i(a)\}\right\}\;.
\end{eqnarray}
The set of bonding maps is denoted $\Gamma_{\Vh}$, since the maps are in one-to-one correspondence with the bisecting vertices $\Vh$ of  $\bbgt$.

Then, in the molecule $\sfrt = \sharp_{\Gamma_{\Vhat}}\sfat_{\Vbar}$, the only patches that remain unbonded are the original $\bpt(\vb_i)$ for $i\in\{1,\dots,|\Vb|\}$.  Moreover, after relabeling the identified vertices $\vh_{ij}^{(a)} \equiv \vh_{i}^{j(a)} = \vh_{j}^{i(a)}$ 
the boundary of $\sfrt$ 
satisfies the relation $\bst \sfrt = \bbgt$.
  From remark \ref{rem:atoms-from-patches}, one notices that all edges added in the construction are {\bf real}.  Thus, the molecule $\sfrt\in\sfrst_{\copies,\snb}$.
\end{proof}

Proposition \ref{prop:all-graphs} has the following consequence:  
\begin{corollary}[{molecule decomposition}]
\label{rem:refmolreduction}
There is a decomposition map $\sdec:\sfrst_{\copies,\lnb}\lora\sfrst_{\copies,\snb}$. 
\end{corollary}
\begin{proof}
  Consider $\sfrt\in\sfrst_{\copies,\lnb}$. 	
  By proposition \ref{prop:all-graphs}, one can decompose each of its atoms, leading to the image of the molecule $\sfrt$ itself under decomposition map $\sdec$. 
\end{proof}

There is an important limitation.  
\begin{proposition}\label{prop:limitation}
  The projection $\Pi_{\copies,\snb}:\sfrst_{\copies,\snb}\lora\sfrs$ is {\bf not} surjective.
\end{proposition}
A sketch of the reasoning is as follows.
  Consider a generic $\sfr\in\sfrs$ and let $\sfrt$ be a representative in the class $\Pi_{\copies,\lnb}^{-1}(\sfr)$. Then, $\sfrt$ consists of bonded spin-foam atoms drawn from the set $\sfast_\rl$. According to proposition \ref{prop:all-graphs}, every atom $\sfat\in\sfast_\rl$ has a decomposition into simplicial atoms of $\sfast_\rs$.  Just like in the decomposition utilized in \ref{prop:all-graphs},  it is possible to show that any decomposition requires one to add {\bf real} structures in order to maintain the integrity of the boundary graph under reduction.  However, if one adds in real structures, then one does not arrive back to the original atom/molecule after reduction, since reduction just amounts to contraction and deletion of virtual structures.



\subsection{From molecules to $\std$-dimensional complexes}
\label{sec:std-complexes}


Since spin-foam molecules are meant to act as spacetime structures in quantum gravity state sums, it is important to understand their relation to $\std$-dimensional manifolds.
Molecules are 2-complexes and  do not posses any higher dimensional information in the first place.
To extend molecules to complexes of spacetime dimension $\std$, some further information has to be specified. 
Another important question is then, whether the resulting $\std$-complexes are (orientable) pseudo-manifolds.
In this section, these issues are discussed first for the special case of simplicial molecules and then for regular, loopless and more general molecules.

\begin{remark}[{generalized simplicial pseudo-$\std$-manifolds}]  
\label{rem:simplicial}
  The single $\copies$-regular, simplicial atom $\sfa_\rs$ has a natural extension to the dual of an abstract $\copies$-simplex (\fig{simplicial}).
  Thus, $\copies$-regular simplicial molecules $\sfrs_\rs$ can be understood as bondings of $\copies$-simplices which yields $\std$-complexes for $\copies=\std$. 
\begin{figure}
  \centering
    \tikzsetnextfilename{simplicial}
\begin{tikzpicture}
\begin{scope}
\node (ab)		[he]						{$12$};
\node (ac)		[he,right of=ab]				{$13$};
\node (ad)		[he,right of=ac]				{$14$};
\node (bc)		[he,right of=ad]				{$23$};
\node (bd)		[he,right of=bc]				{$24$};
\node (cd)		[he,right of=bd]				{$34$};

\node (a)		[hv,above of=ab,xshift=5mm]	{$1$};
\node (b)		[hv,above of=ad,xshift=-2mm]	{$2$};
\node (c)		[hv,above of=bd,xshift=-7mm]	{$3$};
\node (d)		[hv,above of=cd,xshift=-5mm]	{$4$};

\node (abc)	[h,below of=ab,xshift=5mm]	{$123$};
\node (abd)	[h,below of=ad,xshift=-2mm]	{$124$};
\node (acd)	[h,below of=bd,xshift=-7mm]	{$134$};
\node (bcd)	[h,below of=cd,xshift=-5mm]	{$234$};

\node (p)		[h,below of=acd,xshift=-8mm] 	{$1234$};
\node (0)		[h,inner sep=1pt,circle, above of=c,xshift=-8mm]{$v$};

\foreach \from/\to in 
{ab/a,ac/a,ad/a,bc/b,bd/b,cd/c,
abc/ad,acd/ad,abd/ad,bcd/bd,
p/abc,p/acd}
\draw (\from) -- (\to);

\foreach \from/\to in 
{a/0,b/0,c/0,d/0,
ab/b,ac/c,ad/d,bc/c,bd/d,cd/d,
abc/ab,acd/ac,abd/ab,bcd/bc,abc/bc,acd/cd,abd/bd,bcd/cd,
p/abd,p/bcd}
    \draw (\from) -- (\to);
\end{scope}


\begin{scope}[xshift=-4.cm, scale=1.5]
\node [c,label=-45:$v$]		(v)	at (0,-.19)		{};
\node [c,label=0:$\vb_1$]		(1)	at (.5,0)	 	{}; 
\node [c,label=90:$\vb_2$]	(2)	at (-.12,.12)	{}; 
\node [c,label=180:$\vb_3$]	(3)	at (-.32,-.14)	{}; 
\node [c,label=-90:$\vb_4$]	(4)	at (0,-.84)		{};
\node [c,label=0:$\vh_{12}$]	(12)	at (.55,.55)	{};
\node [c
				]		(13)	at (.21,0)		{};
\node [c,label=0:$\vh_{14}$]	(14)	at (.77,-.83)	{};
\node [c,label=135:$\vh_{23}$]	(23)	at (-.73,.31)	{};
\node [c
				] 		(24)	at (-.17,-.51)	{};
\node [c,label=200:$\vh_{34}$]	(34)	at (-.5,-1.06)	{};
\foreach \i/\j in {1/2,1/3,1/4,2/3,2/4,3/4}{
 \path	[f] 	(\i) -- (\i\j) -- (\j) -- (v) -- cycle;
 }
 \foreach \i in {1,2,3,4}{
  \draw [e] (\i) -- (v);
  }
\foreach \i/\j in {1/2,1/3,1/4,2/3,2/4,3/4}{
 \draw 	[eh]	(v)		-- 	(\i\j);
 \draw	[eb] 	(\i) node[vb] {} -- (\i\j) node[vh] {};
 \draw 	[eb] 	(\j) node[vb] {} -- (\i\j) node[vh] {};
 }
\draw [e] (3)	node[vb] {} 				-- (v) node[v] {};
\draw (0,1.35)	node[label=right:{\small 123}] {}	-- (.43,-1.37);
\draw (0,1.35)							-- (1.11,-.27);
\draw (0,1.35)							-- (-1.44,-.73);
\draw (.43,-1.37) node[label=right:{\small134}] {}-- (1.11,-.27);
\draw (-1.44,-.73) node[label=below:{\small234}] {}-- (.43,-1.37);
\draw (1.11,-.27) node[label=right:{\small124}] {}-- (-1.44,-.73);
\end{scope}

\end{tikzpicture}	
    \caption{Tetrahedron from the the 3-simplicial atom $\sfa_{3,\simplicial}$, in a geometric realization (left) and its combinatorics (Hasse diagram, right).
    The construction is based on the canonical extension of complete graphs $\bg_\rs = \bisec^{-1} (\bs \sfa_\rs)$ to $\copies$-simplices and their self-duality.
    }
\label{fig:simplicial}
\end{figure}
 The bonding of atoms along patches is constructed exactly in this way as the $\std$-patches correspond to $(\std-1)$-simplices along which the $\std$-simplices are bonded. 

  Even more generally, given such an extension of all single atoms $\{\sfa_i\}\In\sfas$ in a molecule $\sfr = \sharp_{\{\gm_j\}}\{\sfa_i\}$ to a primal $\std$-polytope each, the associated complex $\cm$ consisting of these $\std$-polytopes with bondings induced by the maps $\gm_j$ is a combinatorial pseudo-$\std$-manifold.
  This can be easily seen in the following way.
  Each $\std$-polytope is pure and $\std$-dimensional and there is no cell in $\cm$ which does not belong to a polytope associated to an atom. Thus $\cm$ itself is pure and $\std$-dimensional.
  The pairwise bonding construction then guarantees $\std$-dimensional connectedness (M2) and non-branching (M3) because the dual properties (M2$^\star$), (M3$^\star$) already apply to the molecule $\sfr$ (remark \ref{rem:bonding-example}).
  

  Still, the combinatorial pseudo-$\std$-manifold obtained from bonding $\std$-simplices according to the structure of a $\std$-simplicial molecule is not an abstract simplicial complex in general \cite{Gurau:2010iu}. 
  The most obvious reason is that identification of vertices under bonding may result in simplices given by multisets (\eg  when bonding $(\std-1)$-simplices on the boundary of the same $\std$-simplex).
  While the set of simplices $\simc$ in an abstract simplicial complex (\dref{simplicial-complex}) may be a multiset, the simplex sets $\s\in\simc$ themselves are not allowed to. 
  
  The deeper reason for this shortcoming, as explained in remark \ref{rem:issub-inverse}, is that such bondings effect loop edges (and higher dimensional equivalents). 
  These are violating property (P3) of polyhedral complexes.
  A further consequence is then that no chain-complex structure is possible.
  
  Therefore, the pseudo-$\std$-manifolds obtained from a $\std$-dimensional enhancement of $\std$-regular simplicial atoms are only generalized polyhedral complexes (in the sense of definition \ref{def:generalized-polyhedral}).
  Since their building blocks are simplices, one could also call them \emph{generalized simplicial complexes}.%
  \footnote{Note that this does not imply that they cannot be given a topological meaning \cite{Smerlak:2011ea} since they do correspond to the so-called ``$\Delta$-complexes" well known in algebraic topology \cite{Hatcher:2002ut}.}
\end{remark}

\begin{remark}[{Coloured simplicial molecules}]
\label{rem:coloured}
  A strategy to overcome the shortcomings of plain simplicial molecules $\sfrs_\drs$ as spacetime manifolds
  that has gained a lot of traction in recent years is based on a $(\std+1)$-colouring of the patches in the underlying simplicial atom $\sfa_\drs$ \cite{Gurau:2011dw,Gurau:2010iu,Gurau:2012hl}.
  Assigning a colour $i=0,1,...,\std$ to every boundary vertex $\vb\in\Vb$ in $\sfa_\drs$ (and, thus, to every patch) induces a unique ${\std+1 \choose\p+1}$-colouring of the $\p$-simplices in the associated $\std$-simplex which $\sfa_\drs$ is the dual of (see again \fig{simplicial}, reading the labels as colour indices).
  Colour preserving bonding of these atoms then results in molecules which are dual to abstract simplicial complexes in the strict sense because vertices in a simplex have different colour and thus cannot be identified \cite{Gurau:2011dw}.
  
\end{remark}

\begin{remark}[{$\std$-manifolds from polytopes}]
  For an atom $\sfa\in\sfas$ which is not simplicial it is neither obvious whether there exist a $\std$-polytope with the corresponding graph $\bg=\bisec^{-1}(\bs\sfa)$ as its dual boundary graph, nor whether this is unique. 
  Certainly, only loopless atoms $\sfa_\loopless \in \sfas_\loopless$ are candidates. Any loop in the boundary graph would violate the defining property (P3) for polytopes.
  
  Even for a subset of $\sfas_\loopless$ with a unique interpretation as polytopes, the set of molecules generated from these atoms is not equivalent to a set of polyhedral $\std$-complexes.
  Along the same general argument carried out in remark \ref{rem:simplicial} in the simplicial case, the local polytope interpretation provides the molecules with the structure of combinatorial pseudo-$\std$-manifolds; but various possible kinds of loops again violate (P3) such that the complexes are not polyhedral complexes in the strict sense of definition \ref{def:polyhedral-complex} but only generalized polyhedral complexes (\dref{generalized-polyhedral}).
  
  Unfortunately, there is no colouring similar to the simplicial case at hand in general  for improving this situation.
  Still, such a colouring can be used to enhance a subclass of $\std$-regular loopless molecules $\sfas_\drl$ to polyhedral pseudo-$\std$-manifolds \cite{Bonzom:2012bg}.
  The essential idea is to use the possibility to construct simplicial atoms from regular, loopless patches (\rem{atoms-from-patches}) to obtain a subdivision of a regular loopless atom.
  If this subdivision, possibly after a further finite number of subdivision moves, is itself a $(\std+1)$-colourable simplicial molecule (as described in remark \ref{rem:coloured}), the atom has an interpretation as a ($(\std+1)$-coloured) simplicial complex, and, removing the subdivision, as a coloured abstract polytope with simplicial boundary.
   This boundary, and thus the (dual) graph $\bg_\drl$, can even be shown to be $\std$-colourable \cite{Bonzom:2012bg,\ORT}.
   Finally, colour-preserving bondings of such atoms result in molecules equivalent to abstract polyhedral complexes in the strict sense (for the same reasons as in the plain simplicial case, remark \ref{rem:coloured}).
\end{remark}

Using the reduction mechanism from labelled regular molecules to arbitrary ones, one could in principle attempt to make a more ambitious statement. 
By showing the existence, for $\std$ odd (even), of a surjective map 
from labelled $\std$-regular $\std$-coloured boundary graphs to $\bgs$ (resp.\ $\bgs_{even}$), 
 one could conjecture the following:
\begin{conjecture}
$\std$-coloured graphs capture all of $\bgs$  ($\bgs_{even}$).  
\end{conjecture}
In order to prove this statement it is necessary to show that in every class $\pi_\drl^{-1}(\bg) \In  \bgst_\drl$, there is a graph that is $\std$-colourable.

The benefit would be that in this way one could, for arbitrary molecules $\sfr$, specify the subclass whose molecules allow for a subdivision into the colourable subclass of $\sfrst_{\copies,\snb}$.
Thus, all these molecules would have a well-behaved topological structure as polyhedral pseudo-$\std$-manifolds. 

\

The purpose of this \sec{molecules} has been a constructive definition of the 2-complex structure underlying SF and GFT amplitudes unveiling their molecular structure.
A focus has been on the way arbitrary such molecules can be obtained from simpler subclasses in terms of a distinction of cells in real and virtual ones. This allows in particular to construct any molecules from $\copies$-regular atoms, and to obtain any closed graph as the boundary graph of a molecule constructed from simplicial atoms.
Finally, I have discussed the conditions under which these 2-complexes can be extended to complexes of larger dimension, allowing for a spacetime interpretation.


\section{Calculus on combinatorial complexes}\label{sec:calculus}

\renewcommand{\clp}{\cl}

After the detailed investigation of the combinatorial spacetime structures relevant for quantum gravity, the topic of this section which is based on \cite{\COTa,Thurigen:2015ba} is the definition of fields (that is, most generally, exterior forms) thereon.
Furthermore, I introduce a discrete calculus in this combinatorial setting, generalizing \cite{Desbrun:2005ug}.
This is not only relevant for a proper derivation of discrete quantum gravity models from established continuum formulations of GR. 
The main aim in the context of this thesis is rather to provide an adequate formalism entailing a definition of a Laplacian on a discrete geometry to lay the ground for the definition of effective dimension observables on quantum geometries which are based on the Laplacian's spectrum (\sec{dimension-definition}).

I start in \sec{forms} by defining $\p$-form fields on a complex as cochains and discuss their geometric meaning. 
For a geometric realization of a complex as a cellular decomposition, geometry is induced from the ambient space.
Equally well, relevant geometric data, \ie volumes of cells, can simply be assigned to a combinatorial complex which I call a \emph{geometric interpretation}.
A geometric interpretation is the prerequisite for a definition of Hodge duality in terms of the dual complex.

In \sec{inner-product}, I use Hodge duality to define an inner product on $\p$-forms analogously to (continuum) differential geometry, turning the space of $\p$-forms into an $L^2$-space. 
This gives rise to a convenient bra-ket formalism for $\p$-form fields.

Section \ref{sec:laplacian} then introduces the exterior differential, induced from taking Stokes theorem as a definition. 
With both a derivative and Hodge duality at hand, the definition of the Laplacian is then straightforward.
Finally, I discuss at length the properties of the Laplacian acting on scalar fields on the dual complex, which is the case relevant for the effective dimension observables.
The geometric interpretation of the dual plays a particular role here and I argue that a barycentric interpretation is most appropriate from a physical, field-theory point of view.

\newpage

\subsection{Exterior forms and Hodge duality on complexes}
\label{sec:forms}

One can identify a natural concept of \emph{discrete forms} in terms of the dual of chains on complexes in the linear-algebraic sense \cite{Desbrun:2005ug,Albeverio:1990ii, Adams:1996ul, Teixeira:2013ee, Grady:2010wb}.
In contrast to the cited literature where typically cellular decompositions of some ambient space are the starting point,
here the definitions are set in the purely combinatorial context. 
This involves in particular that, where necessary, an orientation is not induced from ambient space but defined in the sense of definition \ref{def:orientation}.

\begin{defin}[{$\p$-forms}]
\label{def:p-forms}
On a finite combinatorial complex $\cp$, for $\p\ge0$, the space of $\C$- (or $\R$-)valued \emph{\p-forms} $\Omega^\p(\cp)$ is defined as the space of $\p$-cochains $C^\p(\cp):=C^\p(\cp,\C)$ 
(or $C^\p(\cp,\R)$)
which is the dual space $C^\p(\cp,\C)\equiv \textrm{Hom}(C_\p(\cp),\C)$ to the space of $\p$-chains $C_\p(\cp)$ (definition \ref{def:chain-complex}).

Using a bra-ket notation, I write cells $\clp\in\cpp\p$ understood as basis elements of $C_\p(\cp)$  as $|\clp\ket$.
Accordingly, $\bra \clp |$ denotes the dual basis elements in $C^\p(\cp)$, induced by the pairing
\[\label{braket1}
\bra \clp|\clp' \ket = \Nc\delta_{\clp,\clp'}
\]
%
Thus, a $\p$-form field $\phi \in \Omega^\p(\cp)$ has an expansion in the cochain basis as
\[
\bra \phi |
= \sum_{\clp\in \cpp\p} 
 \bra \phi | \clp \ket \bra\clp|\,.
\]
\end{defin}

\begin{remark}[{Interpretation and comparison to the continuum}]
The meaning of defining fields on a complex as cochains becomes clear when comparing with the continuum \cite{Desbrun:2005ug,Grady:2010wb}.
Take, for example, a finite triangulation $\T$ of a Riemannian manifold $(\mf,g)$ which is a geometric realization $\T=|\simc|$ of an abstract simplicial complex $\simc$ (\dref{simplicial-complex}).
Then, $\p$-cochains can be naturally interpreted as discretized $\p$-forms $\phi\in\Omega^{\p}(\cp)\cong C^{\p}(\cp)$ by smearing the continuous form $\phi_{\rm cont}\in\Omega^{\p}(\mf)$ over $\p$-surfaces $\surface\In|\simc|\In \mf$ represented by chains $|\surface\ket=\sum_i 
|\clp^{i}\ket\in C_{\p}(\simc)$ in the triangulation: 
\[\label{surface-integral}
  \phi(\surface):= \bra \phi|\surface \ket 
  = \sum_i \bra \phi|\clp^{i} \ket
  = \sum_i V_{\clp^{i}} \phi_{\clp^i}
  = \sum_i \int_{|\clp^{i}|}\phi_{\rm cont}
  = \int_{\surface}\phi_{\rm cont}\,,
\]
where $V_{\clp}$ denotes the $\p$-volume of the realization $|\clp|$ of $\clp$ in $|\simc|$ and $\phi_{\clp}$ denotes the (averaged) field value of $\phi_{\rm cont}$ over $\clp$.
In particular, for 
a single $\p$-simplex $\clp$ represented by $|\clp \ket $ 
one has
\begin{equation}
  \phi({\clp}) = \bra \phi | \clp \ket 
  = V_{\clp} \phi_{\clp}
  = \int_{|\clp|}\phi_{\rm cont}.
\label{eq:primal-smearing}
\end{equation}
Therefore, the 
coefficients $\bra\phi|\clp\ket $  
carry the interpretation of integrated field values, while the coefficients $\phi_{\clp}$ are the physical discrete field components.
\footnote{
  For this reason, in \cite{\COTa} I had chosen the convention $\phi_{\clp} = \bra \phi | \clp \ket = \phi(\clp)/V_{\clp}$ instead of \eqref{surface-integral} where the volume factors are then part of the co/chain basis. 
  There is also a third convention \cite{Itzykson:1983vc}, used in random lattice field theory \cite{Christ:1982kr,Christ:1982hv,Christ:1982bn}, where only the Hodge dual fields are densities carrying the whole $\m$-volume.
 } 
Obviously, this requires an embedding of the abstract simplicial complex into the continuum manifold in terms of a geometric realization.

However, note that, even though motivated by discretization, \dref{p-forms} works perfectly well for finite combinatorial complexes $\cp$. 
A geometric interpretation in terms of $\p$-volumes $V_{\clp}$ as induced by the ambient space $\mf$ in the case of triangulations is not needed at this stage, as long as one is only interested in the forms $\bra\phi|$ themselves and not in quantities $\phi_{\clp}$. 
\end{remark}

%

Nevertheless, even a combinatorial complex without any realization as a decomposition of some smooth space can be given a geometric meaning:
\begin{defin}[{geometric interpretation}]
For $\left(\cp,\dm,\ge \right)$ a combinatorial complex, a \emph{geometric interpretation} is an assignment of a $\dm(c)$-volume $V_\cl\in\C$ to every cell $\cl\in\cp$
(where all vertices $\cl\in\cpp0$ are assigned trivial volumes $V_\cl=1$).
It is called a \emph{proper} geometric interpretation if $V_\cl\in\R^+ \setminus\{0\}$ 
for every $\cl\in\cp$.
\end{defin}

\begin{remark}
Note that for a geometric realization $|\cp|$ the set of volumes $\{V_\cl\}_{\cl\in\cp}$ is usually not independent. 
For example, if a complex $\simc$ has a realization as a piecewise linear simplicial complex $|\simc|$, there are various sets of variables which determine all volumes, such as the set of edge lengths, or in $\m=4$ dimensions the set of areas and dihedral angles (examples relevant in quantum gravity are discussed in appendix \ref{sec:classical-expressions}).
Thus, a combinatorial complex together with such variables has an induced geometric interpretation.
\end{remark}

\

To define an inner product of $\p$-forms similar to the continuum case, 
a discrete version of Hodge duality is necessary.
Most approaches to a discrete calculus \cite{Albeverio:1990ii, Adams:1996ul, Sen:2000cr, Teixeira:2013ee}
use the Whitney embedding map to define the Hodge dual. 
But this is not available for purely combinatorial complexes since there is no ambient space 
from which such structure can be induced.
Alternatively, in \cite{Desbrun:2005ug} a definition of the Hodge dual is given solely in terms of a dual complex (but still in a setting of embedded complexes).
Based on the general notion of a dual complex for any combinatorial complex (\dref{dual-complex}) this definition can be generalized:
%


\begin{defin}[{Hodge duality}]\label{def:hodge-duality}
Let $\cp$ be a finite $\m$-dimensional 
complex with 
geo\-metric interpretations $\{V_{\clp}\}$ and $\{V_{\star\clp}\}$ for the primal complex $\cp$ as well as the dual complex $\cps=\star\cp$.
For each $\p=0,1,\dots,\m$, the \emph{Hodge star operator} $\ast$ is an isomorphism 
$\Omega^\p(\cp)\cong C^\p(\cp) \lora \Omega ^{\m-\p}(\cps) \cong C^{\m-\p}(\cps)$ between $\p$-forms on $\cp$ to $(\m-\p)$-forms on $\cps$ induced by the $\star$-duality, defined by its action on field components as the complex conjugate of the $\star$-dual form,
\[
\phi_{\clp} \os\ast\mapsto  (*\phi)_{\star\clp} := \phi^*_{\clp} .
\]
\end{defin}

\begin{remark}
In the bra-ket notation,  chains on the dual complex are notated as bras $\bra \star\cl|$ while their cochains are notated as kets $| \star\cl \ket$.
Thus, the Hodge-dual forms are also kets,
\ba
\label{bra-hodge-ket}
& \bra \phi | & = \sum_{\clp\in\cpp\p} V_{\clp} \phi_{\clp} \bra \clp | 
 = \sum_{\clp\in\cpp\p} \bra\phi | \clp \ket  \bra \clp | 
 \\ 
&\downarrow \ast & \\
|\phi\ket := & |\ast\phi\ket & = \sum_{\clp\in\cpp\p} V_{\star\clp} \phi_{\clp}^* |\star\clp\ket
 \equiv \sum_{\cs\in\cpsp{\m-\p}} \bra \cs |\phi \ket |\cs\ket\,.
\ea
This implies for the dimensionless coefficients
\[
\frac1{V_{\clp}}\bra \phi | \clp \ket \os\ast\lora \frac1{V_{\star\clp}} \bra \star\clp | \phi \ket^*
\]
which is the reason for the necessity of a geometric interpretation in the definition.
The appearance of volumes is the equivalent to the metric factor $\sqrt g$ necessary in the continuum analogue of the Hodge dual on a Riemannian manifold.
%
%

Thus, one can take Hodge duality as two different perspectives to look
at the same discrete field $\phi$: either as a $\p$-form $\bra\phi|$ on the primal
complex or an $(\m-\p)$-form $|\phi\ket$ on the dual complex.

Note that due to the relative orientations of the complexes the duality holds only up to a sign.
The dual of a dual cell as a basis element of $C_\p(\cp)$ has a sign 
$
| \star\star \clp \ket = (-1)^{\p(\m-\p)} | \clp \ket
$
and thus
\cite{Desbrun:2005ug,
Mattiussi:1997jp,Teixeira:2013ee,Teixeira:1999hv} 
\[\label{eq:sign}
| \ast\ast \phi \ket = (-1)^{\p(\m-\p)} | \phi \ket .
\]

\end{remark}


\begin{remark}[{Geometric interpretation of the dual}]
\label{rem:geometric-dual}
In the case of a realization of the complex $\cp$ as a cellular decomposition $|\cp|$ of a Riemannian manifold $(\mf,g)$, 
one might still view a dual field $\ast\phi$ as smearing the Hodge-dual $\ast\phi_\textrm{cont}$ of a continuous field $\phi_\textrm{cont}$, now smeared over dual cells $\star\clp\in\cps$:
\begin{equation}
  \label{dual-integration}
  \ast\phi(\surface_{\star\clp}) = V_{\star\clp}(\ast\phi)_{\star\clp}
  = \int_{|\star\clp |}\ast\phi_{\rm cont}
\end{equation}


where $V_{\star\clp}$ are the volumes of dual cells.
Since a cellular decomposition $|\cp|$ of $(\mf,g)$ does not yet uniquely determine a dual cellular decomposition $|\cps|$, there are various choices for the definition of the dual cell volumes $V_{\star\clp}$.
%

In the case of a simplicial complex $\cp=\simc$, the most common choices in the literature are circumcentric  \cite{Desbrun:2005ug} and barycentric \cite{Albeverio:1990ii, Adams:1996ul, Teixeira:2013ee} constructions of $|\cps|$ (\fig{circum-bary}).
\begin{figure}
  \begin{centering}
    \tikzsetnextfilename{circum-bary}
    \tikzstyle{eb} =	[draw=jgreen,line width=.6pt]

\begin{tikzpicture}
\begin{scope}[scale=.4]
\node  			at (-6,6)		{a)};
\node [c]		(a)	at (0,0)		{};
\node [c]		(b)	at (8,0)		{};
\node [c]		(c)	at (3,6)		{};
\node [c]		(d)	at (-3,3)		{};
\node [c]		(e)	at (-6,-2)		{};
\node [c]		(f)	at (2,-5)		{};
\node [c]		(g)	at (9,-5)		{};
\foreach \i/\j in {a/b,a/c,a/d,a/e,a/f,b/c,c/d,d/e,e/f,f/g,g/b,f/b}{
\draw [eb] 	(\i) 	-- (\j);
}
\node [v]	(abc)		at (4,1.75)		{};
\node [v]	(acd)		at (0.5,3.5)	{};
\node [v]	(ade)		at (-3.25,-.25)	{};
\node [v]	(aef)		at (-2.1,-3.7)	{};
\node [v]	(abf)		at (4,-1.3)	{};
\node [v]	(bfg)		at (5.5,-3.1)	{};
\node [c]	(ab)		at (4,0)		{};
\node [c]	(ac)		at (1.5,3)		{};
\node [c]	(ad)		at (-1.5,1.5)	{};
\node [c]	(ae)		at (-3,-1)		{};
\node [c]	(af)		at (1,-2.5)		{};
\node [c]	(bc)		at (5.5,3)		{};
\node [c]	(cd)		at (0,4.5)		{};
\node [c]	(de)		at (-4.5,.5)		{};
\node [c]	(ef)		at (-2,-3.5)		{};
\node [c]	(fg)		at (5.5,-5)		{};
\node [c]	(bf)		at (5,-2.5)		{};
\node [c]	(bg)		at (8.5,-2.5)	{};
\foreach \i/\j/\k in {a/b/c,a/c/d,a/d/e,a/e/f,a/b/f,b/f/g}{
\path	[f] 	(\i)	-- (\j) -- (\k) -- cycle;
\draw [e] 	(\i\j) 	-- (\i\j\k);
\draw [e] 	(\j\k)	-- (\i\j\k);
\draw [e] 	(\i\k) 	-- (\i\j\k);
}
\end{scope}

\begin{scope}[xshift=8cm,scale=.4]
\node  			at (-6,6)		{b)};
\node [c]		(a)	at (0,0)		{};
\node [c]		(b)	at (8,0)		{};
\node [c]		(c)	at (3,6)		{};
\node [c]		(d)	at (-3,3)		{};
\node [c]		(e)	at (-6,-2)		{};
\node [c]		(f)	at (2,-5)		{};
\node [c]		(g)	at (9,-5)		{};
\foreach \i/\j in {a/b,a/c,a/d,a/e,a/f,b/c,c/d,d/e,e/f,f/g,g/b,f/b}{
\draw [eb] 	(\i) 	-- (\j);
}
\node [v]	(abc)		at (3.7,2)		{};
\node [v]	(acd)		at (0,3)		{};
\node [v]	(ade)		at (-3,.33)		{};
\node [v]	(aef)		at (-1.33,-2.33)	{};
\node [v]	(abf)		at (3.33,-1.66)	{};
\node [v]	(bfg)		at (6.33,-3.33)	{};
\node [c]	(ab)		at (4,0)		{};
\node [c]	(ac)		at (1.5,3)		{};
\node [c]	(ad)		at (-1.5,1.5)	{};
\node [c]	(ae)		at (-3,-1)		{};
\node [c]	(af)		at (1,-2.5)		{};
\node [c]	(bc)		at (5.5,3)		{};
\node [c]	(cd)		at (0,4.5)		{};
\node [c]	(de)		at (-4.5,.5)		{};
\node [c]	(ef)		at (-2,-3.5)		{};
\node [c]	(fg)		at (5.5,-5)		{};
\node [c]	(bf)		at (5,-2.5)		{};
\node [c]	(bg)		at (8.5,-2.5)	{};
\foreach \i/\j/\k in {a/b/c,a/c/d,a/d/e,a/e/f,a/b/f,b/f/g}{
\path	[f] 	(\i)	-- (\j) -- (\k) -- cycle;
\draw [e] 	(\i\j) 	-- (\i\j\k);
\draw [e] 	(\j\k)	-- (\i\j\k);
\draw [e] 	(\i\k) 	-- (\i\j\k);
}
\end{scope}
\end{tikzpicture}

\tikzstyle{eb} =	[draw=jgreen,line width=1.6pt]	
\par\end{centering}
\caption{Circumcentric (a) and barycentric (b) geometric realization of the dual to the same simplicial 2-complex (the one of \fig{subdivision}, with the same colouring).
}
\label{fig:circum-bary}
\end{figure}
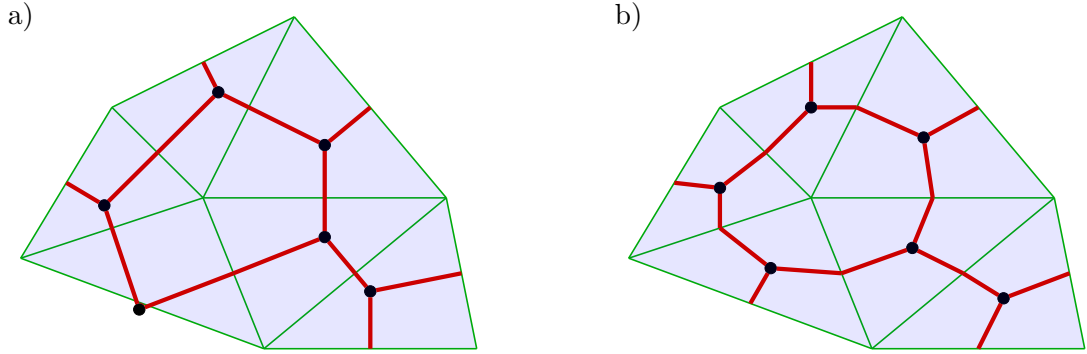
The \emph{barycentric} dual is defined by a realization of the vertices of the subdivision $\ssub\cp$ as barycenters of the related simplex.
Then $|\cps|$ follows from geodesically connecting these according to the procedure laid out in remark \ref{rem:subdivision}.
For constructing the \emph{circumcentric} dual one chooses the circumcenters of the $\m$-simplices in $|\simc|$ as dual vertices and builds up higher dimensional cells 
geodesically connecting them according to the combinatorics of $\cps$. 
\end{remark}

\begin{remark}[{Circumcentric dual and Voronoi decomposition}]
If a geo\-metric realization $|\simc|$ is a Delaunay triangulation, the circumcentric dual complex is a Voronoi decomposition. 
A \emph{Delaunay triangulation} is obtained by constructing $\m$-sim\-plices  from a set of points in a metric space such that no point is in the interior of the circumsphere of any $\m$-simplex. From the same set of points, an $\m$-cell of a \emph{Voronoi decomposition} associated with some point $x$ is constructed as the set of points closer to $x$ than to any other in the set.

For this reason the circumcentric dual is often also called Voronoi dual. But this is meaningful only for Delaunay triangulations. 
For an arbitrary triangulation the circumcentric dual complex and the Voronoi decomposition with respect to the vertex set of the triangulation are different. 
In fact, the Voronoi decomposition does not have the structure of a dual complex for triangulations which are not Delaunay. 
This is particularly important in a fundamentally discrete setting (like in GFT), where the simplicial pseudo-manifold is constructed as a spin-foam molecule (\ie a bonding of $\m$-simplices) and the geometry of each simplex is defined independently of its neighbours.   
The difference is further detailed in the discussion of the Laplacian in \sec{laplacian} and in \fig{circum} below.
\end{remark} 


\subsection{Inner product and bra-ket formulation}\label{sec:inner-product}

With Hodge duality on complexes defined, one can now equip the space $\Omega^\p(\cp)$ of discrete $\p$-forms with an $L^2$-space structure as common in the continuum.
On a Riemannian $\m$-manifold $(\mf,g)$, 
Hodge duality induces an inner product of $\p$-forms $\phi,\psi\in\Omega^\p(\mf)$ 
by pairing $\phi$ with the dual form $\ast\psi$, 
\begin{equation}
\left( \phi,\psi\right) =\int_\mf\phi\wedge\ast\psi
= \int_\mf \phi_{i_{1}\dots i_\p} (\ast\psi)_{i_{p+1}\dots i_{d}}\sqrt{g}\,\d x^{i_{1}}\wedge\dots \d x^{i_{\m}}\,.\label{eq:ScalarProductForms}
\end{equation}
This defines an $L^2$-space of forms $L^2\Omega^\p(\mf)$ \cite{Rosenberg:1997to}. 
An extension of the bra-ket $(\cdot | \cdot)$ between chains and cochains  
will serve the same purpose in the discrete setting of combinatorial complexes.


\begin{defin}[{inner product on $\Omega^\p(\cp)$}]
Let $\cp$ be a finite pure $\m$-dimensional complex. 
An inner product $\Omega^\p(\cp) \times \Omega^\p(\cp) \lora \C$ is then defined by
\[
\bra\phi,\psi\ket:= 
 \sum_{\clp\in\cpp\p} \Vc \bra\phi|\clp\ket \bra\star\clp|\psi\ket
\]
%
With this inner product, the field space $\Omega^\p(\cp)\cong\Omega^{\m-\p}(\cps)$ is the $L^2$ space of $\p$-forms since its dimension 
is the number of $\p$-cells in the finite complex $\cp$, 
\[\label{dimo}
\dim\Omega^\p(\cp)=\dim\Omega^{\m-\p}(\cps)={\rm Card}(\cpp\p)<\infty\,.
\]
\end{defin}

\begin{remark}[{bra-ket formalism}]
One can furthermore define a pairing between $\p$-chains on the finite complex $\cp$ and $(\m-\p)$-chains on its dual $\cps$ for the basis elements as 
\[\label{braket2}
\bra\star\clp|\clp'\ket:=\delta_{\clp,\clp'}\,. 
\]
Because the space of chains is of finite dimension, the completeness relation
\[\label{completeness}
\sum_{\clp\in\cpp\p} \Vc|\clp\ket\bra\star\clp|=\id
\]
follows directly. 
Together with the usual pairing of chains and cochains \eqref{braket1} such a resolution of the identity yields the inner product as a pairing of primal and dual cochains
\[
\bra\phi | \psi\ket 
\os{\eqref{braket2}} =  \bra\phi|\sum_{\clp\in\cpp\p} \Vc |\clp\ket \bra\star\clp|\psi\ket
= \sum_{\clp\in\cpp\p} \Vc \phi_{\clp} \psi_{\star\clp}^*.
\]



To define a formalism with unique types of bras and kets, one can go one step further (beyond \cite{Desbrun:2005ug}) and notationally identify primal chains with dual cochains and dual chains with primal cochains, \ie for all $\cl\in\cp$ 
\begin{equation}
|\clp\ket\equiv|\star\clp\ket\,,\qquad \bra\clp|\equiv\bra\star\clp|\,.
\end{equation}
Then one can write orthonomality and completeness relations
\[
\bra \clp|\clp' \ket = \delta_{\clp,\clp'}  
\quad , \quad
\sum_{\clp\in\cpp\p} \Vc|\clp\ket\bra\clp|=\id
\]
with a four-fold meaning (for the action of chains on cochains on $\cp$ as well as on $\cp$ \eqref{braket1}, and for the pairing of primal with dual chains as well as cochains \eqref{braket2}).

With this identification, a notation of Hodge duality  (\dref{hodge-duality}) can be defined on the level of coefficients as
\begin{equation}
\left\bra \ast\phi|\clp\right\ket :=\left\bra \star\clp|\phi\right\ket \equiv \left\bra \clp|\phi\right\ket =\left\bra \phi|\clp\right\ket ^{*}.\label{eq:HodgeDuality}
\end{equation}

The following commutative diagram shows the identifications and dualities
by which the discrete $L^{2}$ position function space is defined:
\[
\xymatrix{\Omega^\p(\cp)\ar@{<->}[d]^{\cong}\ar@{<->}[r]^{\ast} & \Omega^{\m-\p}(\cps)\ar@{<->}[d]^{\cong}\\
C^\p(\cp)\ar@{<->}[rd]^{\sim}\ar@{<->}[r]^{\star}\ar@{<->}[d]^{\equiv} & C^{\m-\p}(\cps)\ar@{<->}[d]^{\equiv}\\
C_{\m-\p}(\cps)\ar@{<->}[ru]_{\ \ \ \sim}\ar@{<->}[r]_{\star} & C_\p(\cp)
}
\]
In the case of triangulations,  all the maps are already well known \cite{Desbrun:2005ug,Grady:2010wb} except for the last identification denoted as `$\equiv$', which makes it possible to have a Dirac position space notation.
Furthermore, I have generalized the formalism from simplicial decompositions (\ie triangulations) to arbitrary finite pure $\m$-dimensional combinatorial complexes.


\end{remark}

\renewcommand{\clp}[1]{c_{#1}}

\subsection{Exterior derivative and the Laplacian}\label{sec:laplacian}

One can easily introduce the exterior differential operator on discrete forms on a finite pure complex with chain complex $(C(\cp),\bm)$ using Stokes theorem as a definition \cite{Desbrun:2005ug,Grady:2010wb}.
For the integration of the differential of a form $\phi\in\Omega^{\p-1}(\cp)$ over one orientable $\p$-cell $\clp\p$ in a complex $\cp$ with realization as the cellular decomposition $|\cp|$ of a pseudo-manifold, 
the theorem states that 
\[
  \d\phi({\clp\p}) = \int_{|\clp\p|}\d\phi_{\text cont}=\int_{|\partial\clp\p|}\phi_{\text cont}=\phi({\bm\clp\p})\,.
\]
Thus, the differential operator is induced by 
the coboundary operator, which is the operator adjoined to the boundary operator $\bm$ with respect to the duality between chains and cochains. 

\begin{defin}[{differential}]\label{def:differential}
  On a finite pure complex $\cp$ with chain complex $(C(\cp),\bm)$ 
  the \emph{differential} is a linear map $\d: \Omega^{p-1}(\cp) \lora \Omega^\p(\cp)$ defined by its action on single cells
  \begin{equation}\label{differential}
    \d\phi(\clp\p) 
    :=\phi(\bm\clp\p) = \sum_{\clp{\p-1}<\clp\p} \bc{\clp\p}{\clp{\p-1}} \phi ( \clp{\p-1} )
  \end{equation}
  or, equivalently, in the bra-ket notation
  \[
  \bra \d\phi | \clp\p \ket := \bra \phi | \bm \clp\p \ket 
  = \sum_{\clp{\p-1}<\clp\p} \bc{\clp\p}{\clp{\p-1}} \bra \phi | \clp{\p-1} \ket
  \]
  The sign factors $\bc{\clp\p}{\clp{\p-1}}$ take into account the orientation of the faces $\clp{\p-1}$ relatively to the cell $\clp\p$ (proposition \ref{prop:orientation}).

\end{defin}

\begin{remark}
  One can check that the differential on the dual complex is indeed the adjoint to the differential on the primal one, $\bra \d\phi|\psi \ket = \bra \phi|\d\psi \ket$. 
  More precisely, if one does not write the inner product directly as a pairing of a bra and a ket but as a bilinear form on either $\Omega^\p(\cp)$ or $\Omega^{\m-\p}(\cps)$, the adjoint operator of the differential $\da$ as usual is
\[
\da:=(-1)^{\m(p+1)+1}\ast\d\ast\,,
\]
taking into account the sign of multiple Hodge operations \cite{Desbrun:2005ug}.
\end{remark}



Using the above notions of discrete differential and codifferential, one can now simply define the discrete Laplacian using the standard definition of the Hodge-Laplace-Beltrami operator in the well-known form \cite{Rosenberg:1997to}:

\begin{defin}[{Laplace operator}]
Under the assumptions of \dref{differential}, the Laplacian is defined as 
\[\label{Delta}
\boxd{
\Delta=\da\d+\d\da\,.}
\]
\end{defin}

\begin{example}[{dual scalar Laplacian}]
  For later purpose (chapter \ref{ch:dimensions}), the action of the Laplacian on dual scalar fields $\phi\in\Omega^{0}(\cps)\cong\Omega^{d}(\cp)$,
  or equivalently, on fields living on $\m$-cells on a closed combinatorial pseudo-manifold $\cp$ is of particular interest:%
  \footnote{In the Regge calculus literature, a Laplacian of the same form is derived for a primal scalar field (\ie a scalar field living on the vertices of the primal simplicial complex) in the circumcentric case \cite{Hamber:2009wl}. Then the dual Laplacian $\Delta $ is guessed to have exactly the form \eqref{scalar-laplacian}.}
  \begin{eqnarray}
    \label{scalar-laplacian}
    (-\Delta \phi )_v & = & - \bra v | (-1)^{\m(1+1)+1} \ast\d\ast\d\phi \ket \nonumber \\
    & = & (-1)^{\m(\m-\m)} \frac{1}{V_{\clp{\m}}} \bra \d\ast\d\phi|\clp{\m} \ket \nonumber \\
    & = & \frac{1}{V_{\clp{\m}}} \sum_{\clp{\m-1}<\clp{\m}} \bc{\clp\m}{\clp{\m-1}}  \bra \ast\d\phi|\clp{\m-1} \ket \nonumber \\
    & = & \frac{1}{V_{\clp{\m}}} \sum_{\clp{\m-1}<\clp{\m}} \bc{\clp\m}{\clp{\m-1}} \frac{V_{\clp{\m-1}}}{V_{\csp{1}}} \bra \csp{1}|\d\phi \ket \nonumber \\
    & = & \frac{1}{V_{\clp{\m}}} \sum_{\clp{\m-1}<\clp{\m}} \bc{\clp\m}{\clp{\m-1}} \frac{V_{\clp{\m-1}}}{V_{\csp{1}}} \sum_{\csp0 <^\star \csp{1}} \bcs{\csp1}{\csp0} \bra \csp0|\phi \ket \nonumber \\
    & = & \frac1{V_{\star v}} \sum_{v' \sim v} \frac{V_{\star(vv')}}{V_{vv'}}  (\phi_v -\phi_{v'} )\,.
  \end{eqnarray}
  In the first line, the vanishing of $\da\propto\ast\d\ast$ on 0-forms is used, while in the next four lines the differential and Hodge star operator are applied one after the other. 
  In the second line, $\clp\m$ denotes the dual cell $\star v$ and, in the fourth line, $\csp1$ denotes  the dual cells $\star\clp{\m-1}$ accordingly.  
  The last line is just a reordering of terms.
  The dual volumes $V_{\csp1} = V_{vv'}$ in the denominator are the lengths of the dual edges $\csp1 =(vv') = \star(\clp\m \clp\m')$ between the vertex $v$ where the dual Laplacian is evaluated at and its neighbours $v'\sim v$. 
  An alternative, more intuitive notation for these is therefore $\dl_{\clp\m \clp\m'} \equiv \dl_{vv'}\equiv  V_{vv'}$.
  
\renewcommand{\cl}{v}  
  
  The action of the Laplacian on a scalar field ket 
  \begin{eqnarray}
    - | \Delta \phi\ket & = & - \sum_{\cl} |\cl\ket \bra \cl | \Delta \phi \ket 
     \os{\eqref{scalar-laplacian}}{=} \sum_\cl |\cl\ket \frac1{V_{\star v}} \sum_{\cl'\sim\cl} \frac{V_{\star(\cl\cl')}}{\dl_{\cl\cl'}} \left[ \bra \cl | \phi \ket - \bra \cl' | \phi \ket \right] \nonumber\\
    & = & \left[ \sum_\cl V^{-1}_{\star v} \left(\sum_{\cl'\sim\cl} w_{\cl\cl'} \right)| \cl \ket \bra \cl | \right] |\phi\ket - \left[\sum_\cl V^{-1}_{\star v} \sum_{\cl'\sim\cl} w_{\cl\cl'} |\cl\ket\bra\cl'| \right] |\phi\ket \nonumber\\
    & \equiv &  D|\phi\ket-A|\phi\ket\,.\label{DAeq}
  \end{eqnarray}
  is of the general type of a graph Laplace matrix \cite{Chung:1997tk}: up to the inverse volume factor $V^{-1}_{\star v}$, on the 1-skeleton graph of the dual complex it is a difference of an off-diagonal adjacency matrix $A$ in terms of weights
\[\label{weights}
w_{\cl\cl'}:=\frac{V_{\star(\cl\cl')}}{\dl_{\cl\cl'}}
\]
and a diagonal degree matrix $D$ with entries 
$
\sum_{\cl'\sim\cl}w_{\cl\cl'}
$.
Thus, in the trivial case of constant volumes over the complex, \eg equilateral triangulations, the Laplacian is just proportional to the combinatorial graph Laplacian of the dual 1-skeleton of the complex. 

\end{example}

\renewcommand{\cl}{v}  

By definition, such discrete (graph) Laplacians obey three desirable properties \cite{Chung:1997tk,Wardetzky:2008kk}: 

\begin{remark}[{properties of the scalar Laplacian}]
\label{rem:laplacian-properties}
  The Laplacian \eqref{scalar-laplacian} has the following properties:
  \begin{enumerate}
    \item[1.] \emph{Null condition}: $(\Delta \phi )=0$ if, and only if, $\phi$ is constant. 
    This is obvious because $\Delta \phi$ is the difference of position values of $\phi$. 
    A zero mode in the spectrum of $\Delta$ reflects the fact that $\cp$ corresponds to a closed pseudo-manifold. 
    \item[2.] \emph{Self-adjointness}: The Laplace operator is self-adjoint with respect to the inner product 
    \[
    \bra \phi|\Delta\psi \ket =\bra \Delta\phi|\psi \ket\,.
    \]
    This is reflected by the symmetry of the weights $w_{\cl\cl'}$.
    \item[3.] \emph{Locality}: The action of $\Delta$ at any given position, $(\Delta\phi )_\cl$, is affected only by field values $\phi_{\cl'}$ at neighbouring positions, \ie vertices $\cl'$ incident to $\cl$. 
    In discrete calculus, this comes directly from the definition of the Laplacian as a second-order differential operator. 
  \end{enumerate}
  In the case of a cellular decomposition $|\cp|$ of a pseudo-manifold $\mf$, a further natural condition which is built into the formalism from the start (by the definition of differentials via Stokes theorem) is the 
  \begin{enumerate}
    \item[4.] Convergence to the continuum Laplacian under refinement of triangulations. 
  \end{enumerate}
  To see this, consider a region $\Omega\in M$ large compared to the scale $a\sim V_{\clp\p} ^{1/\p}$ of cells $\clp\p\in \cp$, in which the function $\phi$ and its derivatives do not vary strongly. 
  Since products 
  $V_{\star(\cl\cl')} \dl_{\cl\cl'}/\m$ 
  provide a local $\m$-volume measure, one has
  \begin{eqnarray}
    \underset{\cl\in\Omega}{\sum}V_{\cl}(-\Delta \phi)_{\cl} =\underset{\cl\in\Omega}{\sum}\underset{\cl'\sim\cl}{\sum}\frac{V_{\star(\cl\cl')}}{\dl_{\cl\cl'}} (\phi_{\cl}-\phi_{\cl'} )
 \approx2\m\ \mbox{Vol}(\Omega)\underset{\csp{1}\in\Omega}{\sum}\frac{\phi_{\cl}-\phi_{\cl'}}{a^{2}}\,.
  \end{eqnarray}
  Summing over all the dual edges $\csp{1}\in\Omega$ gives effectively a rotationally invariant expression. 
  In particular, it is an average over hypercubic lattices and the difference term can readily be identified as the Laplacian in the continuum limit, just like in standard lattice field theory with hypercubic lattice size $a$. 
  Because $\phi_{\cl+ae_{\mu}}\underset{a\ra0}{\lora}\phi_{\cl}+a\cl(\partial^{\mu}\phi )_{\cl}e_{\mu}+\mathcal{O}(a^{2})$, the difference term gives 
  \ba \label{continuum-limit}
    \sum^{2\m}_{\cl'} \frac{\phi_{\cl}-\phi_{\cl'}}{a^{2}} 
    & = & -\overset{\m}{\underset{\mu=1}{\sum}}\frac{1}{a} \left(\frac{\phi_{\cl+a e_{\mu}}-\phi_{\cl}}{a}-\frac{\phi_{\cl}-\phi_{\cl-a e_{\mu}}}{a} \right) \nonumber\\
    &\underset{a\ra0}{\lora} &  -\overset{\m}{\underset{\mu=1}{\sum}}\frac{(\partial_{\mu}\phi )_{\cl}-(\partial_{\mu}\phi )_{\cl-ae_{\mu}}}{a}e^{\mu}\approx-\overset{d}{\underset{\mu=1}{\sum}}(\partial^{\mu}\partial_{\mu}\phi )_{\cl}\,.
  \ea
\end{remark}

\begin{remark}[{circumcentric vs. barycentric dual}]
  Despite the validity of the above properties, one has to expect that it is not possible to preserve all the features of the continuum Laplacian in the discrete setting. 
  This is expected on general grounds and has been shown for example in the case of two-dimensional triangulations \cite{Wardetzky:2008kk}. 
  As a result, the definition of a discrete counterpart of the continuum Laplacian cannot be unique. 
  In the present case, it is therefore natural to wonder which properties of the continuum Laplacian are not preserved by the discrete scalar Laplacian $\Delta$ \eqref{scalar-laplacian}.
  
  The answer turns out to depend also on the specific choice of the geometry of the dual complex, that is, on the choice of its geometric realization as compared to the primal complex. 
  The two distinguishing features are linear precision and positivity. 
  \begin{enumerate}
    \item[5.] \emph{Linear precision}: On a piecewise linear (``straight-line" \cite{Wardetzky:2008kk}) polyhedral decomposition $|\polyc|$ of flat space $\mf \In \R^\m$, it holds that $(\Delta \phi)_\cl = 0$ if $\phi$ is a linear function $\phi(x^{\mu})=a+\sum_{i=1}^d a_{\mu}x^{\mu}$ in Cartesian coordinates $x^{\mu}$. 
    By linearity, this is equivalent to a vanishing Laplacian $(\Delta x)_{\cl}=0$ of the coordinate
function $x$ (considered as a bunch of scalars $x^{\mu}$).
  \end{enumerate}
  Linear precision holds for circumcentric dual geometries, in which case the dual lengths are $\dl_{\cl\cl'} = |x_{\hat{\cl}}-x_{\hat{\cl}'} |$ and (with unit face normals $\hat{n}_{\cl\cl'}=$ $\frac{x_\cl - x_{\cl'}}{|x_\cl-x_{\cl'} |}$) 
  \[
    (\Delta x)_\cl \propto \sum_{\cl'\sim\cl} \frac{V_{\star(\cl\cl')}}{\dl_{\cl\cl'}} l(x_\cl - x_{\cl'} )
    = \sum_{\cl'\sim\cl}  V_{\star(\cl\cl')} \hat{n}_{\cl\cl'} = 0
  \]
  is true because these are exactly the closure relations for the polyhedral cell $\star\cl$ dual to the vertex $\cl$.
  This property fails, on the other hand, in the barycentric case. 
  One can  understand this  heuristically by noting that generically $\dl_{\cl\cl'}\ne |x_{\hat{\cl}}-x_{\hat{\cl}'} |$ in any dimension for the barycentric dual edges, such that $(\Delta x)_{\cl}$ reduces to a sum over normals of a set of modified faces, which cannot be expected to close, in general. 

  The second property is 
  \begin{enumerate}
    \item[6.] \emph{Positivity} of the weights: $w_{\cl\cl'} > 0$ for all edges $(\cl\cl')\in\cpsp1$. 
    It is also called Markov property \cite{Kigami:2001wk} and is directly related to Osterwalder-Schrader positivity. 
    The latter is crucial for a Euclidean quantum field theory to yield unitarity in the corresponding Lorentzian theory after Wick rotation \cite{Osterwalder:1973hq}. 
  \end{enumerate}
  Positivity holds if all the volumes in the weights are positive. 
  This is generically true for barycentric duals. 
  For circumcentric duals the situation is less general. 
  Positivity does hold for circumcentric duals of regular complexes (where the circumcenters lie in the simplices).
  However, this is not the case for irregular circumcentric duals. 
  When a circumcenter does not lie inside the simplex, the part of the dual length associated with this simplex is negative such that in some cases the sum of the two parts is negative (see \fig{circum}), inducing negative Laplace matrix weights. 
  \begin{figure}[htb]
    \begin{centering}
     \tikzsetnextfilename{circum}
     \tikzstyle{eb} =	[draw=jgreen,line width=.6pt]

\begin{tikzpicture}
\begin{scope}[scale=.4]
\node  						at (-6,6)		{a)};
\node [c,label=45:$ 2$]		(a)	at (0,0)		{};
\node [c,label=right:$ 1$]		(b)	at (8,0)		{};
\node [c]					(c)	at (3,6)		{};
\node [c]					(d)	at (-3,3)		{};
\node [c]					(e)	at (-6,-2)		{};
\node [c,label=below:$ 3$]		(f)	at (2,-5)		{};
\node [c,label=right:$ 4$]		(g)	at (9,-5)		{};
\foreach \i/\j in {a/b,a/c,a/d,a/e,a/f,b/c,c/d,d/e,e/f,f/g,g/b,f/b}{
\draw [eb] 	(\i) 	-- (\j) ;
}
\node [v]	(abc)		at (4,1.75)		{};
\node [v]	(acd)		at (0.5,3.5)	{};
\node [v]	(ade)		at (-3.25,-.25)	{};
\node [v]	(aef)		at (-2.1,-3.7)	{};
\node [v,label=10:$v_{123}$]	(abf)		at (4,-1.3)		{};
\node [v,label=-45:$v_{134}$]	(bfg)		at (5.5,-3.1)	{};
\node [c]	(ab)		at (4,0)		{};
\node [c]	(ac)		at (1.5,3)		{};
\node [c]	(ad)		at (-1.5,1.5)	{};
\node [c]	(ae)		at (-3,-1)		{};
\node [c]	(af)		at (1,-2.5)		{};
\node [c]	(bc)		at (5.5,3)		{};
\node [c]	(cd)		at (0,4.5)		{};
\node [c]	(de)		at (-4.5,.5)		{};
\node [c]	(ef)		at (-2,-3.5)		{};
\node [c]	(fg)		at (5.5,-5)		{};
\node [vb,label=left:$\vb_{13}$](bf)	at (5,-2.5)		{};
\node [c]	(bg)		at (8.5,-2.5)	{};
\foreach \i/\j/\k in {a/b/c,a/c/d,a/d/e,a/e/f,a/b/f,b/f/g}{
\path	[f] 	(\i)	-- (\j) -- (\k) -- cycle;
\draw [e] 	(\i\j) 	-- (\i\j\k);
\draw [e] 	(\j\k)	-- (\i\j\k);
\draw [e] 	(\i\k) 	-- (\i\j\k);
}
\end{scope}

\begin{scope}[xshift=8cm,scale=.4]
\node  			at (-6,6)		{b)};
\node [c]		(a)	at (0,0)		{};
\node [c,label=right:$ 1$]		(b)	at (7,2)		{};
\node [c]		(c)	at (3,6)		{};
\node [c]		(d)	at (-3,3)		{};
\node [c]		(e)	at (-6,-2)		{};
\node [c,label=below:$ 3$]		(f)	at (.5,-4)		{};
\node [c]		(g)	at (9,-5)		{};
\foreach \i/\j in {a/b,a/c,a/d,a/e,a/f,b/c,c/d,d/e,e/f,f/g,g/b,f/b}{
\draw [eb] 	(\i) 	-- (\j) ;
}
\node [v]	(abc)		at (3.2,2.2)	{};
\node [v]	(acd)		at (0.5,3.5)	{};
\node [v]	(ade)		at (-3.25,-.25)	{};
\node [v]	(aef)		at (-2.55,-2.35)	{};
\node [v,label=10:$v_{123}$]	(abf)		at (4.2,-1.5)	{};
\node [v,label=-45:$v_{134}$]	(bfg)		at (5,-2.4)		{};
\node [c]	(ab)		at (3.5,1)		{};
\node [c]	(ac)		at (1.5,3)		{};
\node [c]	(ad)		at (-1.5,1.5)	{};
\node [c]	(ae)		at (-3,-1)		{};
\node [c]	(af)		at (.25,-2)		{};
\node [c]	(bc)		at (5,4)		{};
\node [c]	(cd)		at (0,4.5)		{};
\node [c]	(de)		at (-4.5,.5)		{};
\node [c]	(ef)		at (-2.75,-3)	{};
\node [c]	(fg)		at (4.75,-4.5)	{};
\node [vb,label=left:$\vb_{13}$]	(bf)  at  (3.75,-1)	{};
\node [c]	(bg)		at (8,-1.5)		{};
\foreach \i/\j/\k in {a/b/c,a/c/d,a/d/e,a/e/f,a/b/f,b/f/g}{
\path	[f] 	(\i)	-- (\j) -- (\k) -- cycle;
\draw [e] 	(\i\j) 	-- (\i\j\k);
\draw [e] 	(\j\k)	-- (\i\j\k);
\draw [e] 	(\i\k) 	-- (\i\j\k);
}
\node [v]	(abf)		at (4.2,-1.5)	{};
\end{scope}

\begin{scope}[yshift=-6cm,scale=.4]
\node  			at (-6,6)		{c)};
\node [c]		(a)	at (0,0)		{};
\node [c]		(b)	at (7,2)		{};
\node [c]		(c)	at (3,6)		{};
\node [c]		(d)	at (-3,3)		{};
\node [c]		(e)	at (-6,-2)		{};
\node [c]		(f)	at (.5,-4)		{};
\node [c,label=right:$ 4$]		(g)	at (7.5,-1.5)	{};
\foreach \i/\j in {a/b,a/c,a/d,a/e,a/f,b/c,c/d,d/e,e/f,f/g,g/b,f/b}{
\draw [eb] 	(\i) 	-- (\j) ;
}
\node [v]	(abc)		at (3.2,2.2)	{};
\node [v]	(acd)		at (0.5,3.5)	{};
\node [v]	(ade)		at (-3.25,-.25)	{};
\node [v]	(aef)		at (-2.55,-2.35)	{};
\node [v,label=right:$v_{123}$]	(abf)		at (4.2,-1.5)	{};
\node [v,label=left:$v_{134}$]	(bfg)		at (3.1,-.3)		{};
\node [c]	(ab)		at (3.5,1)		{};
\node [c]	(ac)		at (1.5,3)		{};
\node [c]	(ad)		at (-1.5,1.5)	{};
\node [c]	(ae)		at (-3,-1)		{};
\node [c]	(af)		at (.25,-2)		{};
\node [c]	(bc)		at (5,4)		{};
\node [c]	(cd)		at (0,4.5)		{};
\node [c]	(de)		at (-4.5,.5)		{};
\node [c]	(ef)		at (-2.75,-3)	{};
\node [c]	(fg)		at (4,-2.75)	{};
\node [vb,label=left:$$] (bf) 	at  (3.75,-1)	{};
\node [c]	(bg)		at (7.25,.25)		{};
\foreach \i/\j/\k in {a/b/c,a/c/d,a/d/e,a/e/f,a/b/f,b/f/g}{
\path	[f] 	(\i)	-- (\j) -- (\k) -- cycle;
\draw [e] 	(\i\j) 	-- (\i\j\k);
\draw [e] 	(\j\k)	-- (\i\j\k);
\draw [e] 	(\i\k) 	-- (\i\j\k);
}
\end{scope}

\begin{scope}[xshift=8cm,yshift=-6cm,scale=.4]
\node  			at (-6,6)		{d)};
\node [c]		(a)	at (0,0)		{};
\node [c]		(b)	at (7,2)		{};
\node [c]		(c)	at (3,6)		{};
\node [c]		(d)	at (-3,3)		{};
\node [c]		(e)	at (-6,-2)		{};
\node [c]		(f)	at (.5,-4)		{};
\node [c]		(g)	at (7.5,-1.5)	{};
\foreach \i/\j in {a/b,a/c,a/d,a/e,a/f,b/c,c/d,d/e,e/f,f/g,g/b,a/g}{
\draw [eb] 	(\i) 	-- (\j) ;
}
\node [v]	(abc)		at (3.2,2.2)	{};
\node [v]	(acd)		at (0.5,3.5)	{};
\node [v]	(ade)		at (-3.25,-.25)	{};
\node [v]	(aef)		at (-2.55,-2.35)	{};
\node [v,label=45:$v_{124}$] 		(agb)		at (3.85,-.24)	{};
\node [v,label=-135:$v_{234}$]	(agf)		at (3.6,-1.6)	{};
\node [c]	(ab)		at (3.5,1)		{};
\node [c]	(ac)		at (1.5,3)		{};
\node [c]	(ad)		at (-1.5,1.5)	{};
\node [c]	(ae)		at (-3,-1)		{};
\node [c]	(af)		at (.25,-2)		{};
\node [c]	(bc)		at (5,4)		{};
\node [c]	(cd)		at (0,4.5)		{};
\node [c]	(de)		at (-4.5,.5)		{};
\node [c]	(ef)		at (-2.75,-3)	{};
\node [c]	(gf)		at (4,-2.75)	{};
\node [c,label=left:$$] (ag) 	at  (3.75,-.75)	{};
\node [c]	(gb)		at (7.25,.25)		{};
\foreach \i/\j/\k in {a/b/c,a/c/d,a/d/e,a/e/f,a/g/b,a/g/f}{
\path	[f] 	(\i)	-- (\j) -- (\k) -- cycle;
\draw [e] 	(\i\j) 	-- (\i\j\k);
\draw [e] 	(\j\k)	-- (\i\j\k);
\draw [e] 	(\i\k) 	-- (\i\j\k);
}

\end{scope}
\end{tikzpicture}

\tikzstyle{eb} =	[draw=jgreen,line width=1.6pt]	
     \end{centering}
     \caption{Examples for the circumcentric dual edge in a simplicial 2-complex: 
     In the first picture (a), the length of the dual edge $(v_{123}v_{134})$ (being the sum of the distance of both vertices to $\vb_{13}$) is positive. 
     Shifting the vertices 1 and 3 (b), the dual vertex $v_{123}$ lies outside the triangle $(123)$. Its distance to $\vb_{13}$ therefore is negative, but the dual length  of $(v_{123}v_{134})$ is still positive. 
     Shifting vertex 4 (c), both $v_{123}$ and $v_{134}$  lie outside their simplices such that the dual length is negative. 
     Exactly when this happens the triangulation cannot be Delaunay because the circumcenter is closer to the neighbour than to its own triangle. 
     The forth picture (d) shows the Delaunay triangulation for the primal points of (c), where the triangles $(123)$ and $(124)$ originally considered, and thus their dual vertices, are replaced with $(124)$ and $(234)$.}
     \label{fig:circum}
  \end{figure}
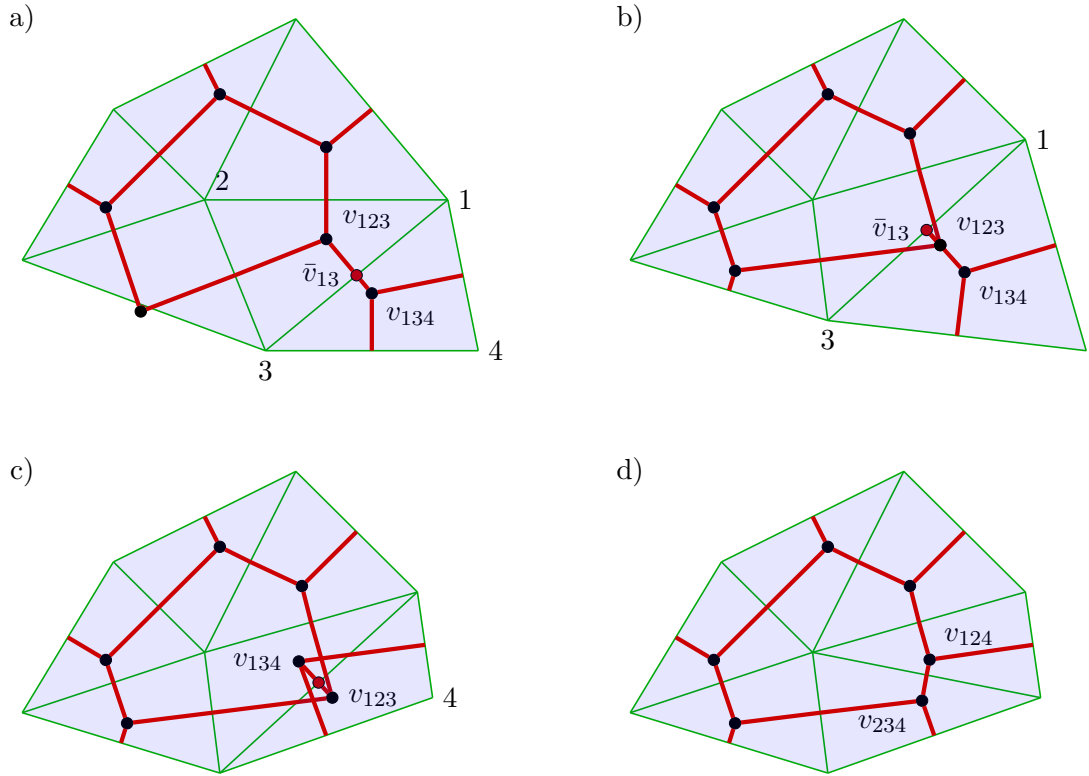
\end{remark}
  
Therefore, as anticipated, the choice of geometry of the dual complex is crucial, yielding different properties for the discrete Laplacian. 
In quantum gravity, in particular in the investigation of the spectral dimension, the barycentric dual is to be preferred:

\begin{remark}[{physics choice for the Laplacian}]
  Quite in general, the null-condition, symmetry and positivity are required properties for any Laplacian.
  They are even taken as {\it the} defining properties in fractal spectral theory \cite{Kigami:2001wk} (see appendix \ref{sec:classical-expressions}). 
  On the contrary, from the physics' perspective it could be expected on general grounds that standard locality and linear precision might be violated. 
  
  Furthermore, the relation of the notion of locality for $\Delta $ in the combinatorial context and the continuum is not immediate. 
  Indications of a breakdown of standard locality actually exist in several approaches to quantum gravity (\eg \cite{Giddings:2001fq,Giddings:2004bc,Giddings:2006cn,Calcagni:2013jx}). 
In fractional calculus, which can be used as an effective description of fractal and other anomalous spacetimes, the Laplacian may be composed by fractional integro-differential operators, which are non-local (by the dependence on non-neighbouring points) \cite{Calcagni:2012rm,Calcagni:2012kj,Calcagni:2011sz,Calcagni:2012vd}.
  
  Linear precision is not needed either, because the combinatorial manifolds considered are not flat in general and its only relevance is as an asymptotic property in the continuum limit to flat spaces. 
  But as argued, this is already fulfilled up to higher-order corrections. 
  A reason why this works despite the lack of linear precision is that the average difference between circumcentric and barycentric dual lengths is only of higher order in the scale of refined cellular decompositions.
  Thus, as far as quantum gravity is concerned, linear precision is not a necessary property since it does not seem reasonable to enforce properties of the continuum flat-space Laplacian exactly in the discrete theory.  
  Also, fractional spacetimes are a continuum example where this property is viola\-ted, in all self-adjoint Laplacians (even in the second-order one, due to the presence of a measure weight to the right of the derivatives) \cite{Calcagni:2012rm,Calcagni:2012zj}.
  
  As for why positivity should then be satisfied, instead, the reasons are the following. 
  One reason is simply by exclusion: when linear precision can be dropped, it makes sense to try to enforce as many as possible of the other properties. 
  A second reason is that all quantum gravity approaches under consideration are phrased as standard quantum theories of discrete spacetime, as discussed in detail in \sec{conceptual-relations}.
  Though reflection positivity in a Euclidean formulation has not been yet directly related to unitarity in this context, one still expects such relation to exist, even if it is not realized by a simple Wick rotation. 
  Therefore, it seems preferable to maintain it in the definition of the discrete theory. 
  
  The main reason in the context of this thesis concerns the application of the discrete Laplacian operator 
  as the central quantity to define the spectral and walk dimension (chapter \ref{ch:dimensions}).
  The calculation of these observables uses the discrete Laplacian operator for defining a test diffusion process taking place on the discrete structures defining quantum gravity states and histories. 
  Positivity of the Laplacian is then a necessary requirement for a  properly defined diffusion process and thus for a sensible notion of observables. 
\end{remark}

\

With this \sec{calculus} I have concluded the discussion of combinatorial complexes, showing how discrete calculus can be defined in this context and for discrete geometries consisting in an additional assignment of geometric data to the complex cells.
I have explained how the dual complex gives rise to a definition of Hodge duality for $\p$-form fields which are  understood as cochains on the complex. 
Given an exact sequence of boundary maps $\bm$ on the space of chains it is then possible to use Stokes theorem to define a differential structure on discrete $\p$-forms. 
Finally I have discussed the properties of the resulting scalar Laplacian and argued that a geometric interpretation using the barycentric dual complex should be chosen to maintain positivity of the Laplacian,  needed for a meaningful definition of spectral and walk dimension.

\renewcommand{\cl}{c}

\chapter{Group field theory for all discrete quantum geometries \label{ch:dqg}}

The purpose of this chapter is a detailed presentation of GFT as a completion of SF models.
In particular I will introduce group field theories, \ie models of GFT, which allow for boundary states based on arbitrary boundary graphs and are thus compatible with LQG and SF models, further strengthening the relation between the theories and thus contributing to the theory crystallization discussed in \sec{crystallization}.
This is based on the publication \cite{\ORT}.

As discussed in \sec{relations} already to some detail, there are various differences between LQG, SF models and GFT in the first place,  concerning both kinematics and dynamics.
Nevertheless, SF models and GFT are rather closely related in the sense that for any SF model there is GFT model which generates the SF amplitudes in the perturbative expansion of the GFT path integral.
So far this has been true up to one subtlety: while in both approaches there is a main focus on the simplicial setting, \ie on the definition of amplitudes based on triangulations, SF models have recently been extended to a more general class of complexes \cite{\KKL}; so far, in GFT the question of such an extension has not been addressed, up to some brief remark on the possibility in principle \cite{Reisenberger:2001hd}.
In this chapter I will show how to close this gap and define GFT models for the extended SF models.

The challenge to generalize GFT to cover arbitrary combinatorial structures is even more pressing when viewed from the perspective of LQG.
Due to the setting of LQG with kinematical states defined in terms of embedded graphs $\bge\In\sm$ coming from curves in a smooth spatial background manifold $\sm$, these graphs allow for arbitrary combinatorics, \ie their combinatorial structure is given by the most general closed graphs $\bg\in\bgs$, in particular with vertices of arbitrary valence. 
In a simplicial setting, on the other hand, states are defined on the dual boundary graph of a simplicial complex of spacetime dimension $\std$ which is thus  $\std$-regular, \ie all vertices are incident to $\std$ edges.
For compatibility with LQG, the KKL-extension of SF models \cite{\KKL} is constructed precisely on the basis of a space of boundary states with support on arbitrary graphs $\bg\in\bgs$. 

The aim here is thus to define a GFT  framework that can accommodate, both kinematically and dynamically,  all the states and histories that one might expect to appear in SF models and which are compatible with LQG.
The combinatorial structures introduced in the preceding chapter set the stage for such a generalization. 

There are two ways how GFT can be extended towards this goal.
The first proposal constitutes a very formal (and thereby somewhat trivial) generalization of GFT to a formalism based on an infinite number of fields. 
Such a multi-field GFT  generates series catalogued by arbitrary spin-foam molecules, such that arbitrary graphs label quantum states. 
It is a direct counterpart of the KKL-extension of gravitational SF models. 
The main point in defining this theory is showing the absence of any fundamental obstruction to accommodating arbitrary combinatorial structures. 
However, such a field theory is not expected to be useful since the sum over complexes seems hard to be tamed.
 
More interesting and much more manageable is a second construction. 
The standard simplicial GFT contains an interaction based on the simplicial spin-foam atom which can be interpreted as the dual 2-skeleton of a $D$--simplex. 
As shown in \sec{simplicial-structures}, the simplicial atom equipped with a labelling of real and virtual cells is sufficient to construct a general class of molecules covering the whole class of closed  graphs on their boundary
Thus, simplicial GFT is sufficient to generate 2--complexes with arbitrary boundary graphs.
The challenge, however, is to assign correct amplitudes.  
This is solved by a mild extension consisting in an augmentation of the data set over which the field is defined, providing control over the combinatorial structures generated by the theory. 
Such a mechanism is known as dual-weighting \cite{DiFrancesco:1992cn,Kazakov:1996et,Kazakov:1996fq,Benedetti:2012ed}  and permits to tune the theory to a regime in which the perturbative sum is catalogued by appropriately weighted arbitrary spin-foam molecules.
This allows to give an explicit GFT formulation of the KKL-extension of gravitational SF models; but even more it allows to propose a new class of models incorporating similar constraints which are arguably better motivated from the geometric point of view. 

\

The plan of this chapter is the following.
In order to motivate GFT as a completion of (gravitational) SF models, 
I will start in \sec{lgt-gravity} with a review of SF theory as a proposal for quantum gravity, \ie as a discrete gauge theory providing a path integral based on a variation of quantum BF theory.
Thereby, I will highlight in particular how the combinatorial structures as introduced in \sec{molecules} arise in this framework and show that the theory can be defined on such purely combinatorial structure from the beginning.

In \sec{gft} I will then introduce GFT, in particular with respect to its definition as a quantum field theory generating SF amplitudes.
While GFT Feynman diagrams are usually considered as stranded diagrams, I will recast the GFT formalism in a way showing explicitly the spin-foam molecule structure encoded.
This renders the first generalization strategy of a multi-field GFT a straightforward task.
In this manner, any spin-foam molecule can be generated, albeit in a rather formal manner, with an infinite set of GFT  fields.  
    
In \sec{dw-gft} I will thus 
move over to labelled structures, which permit a much simpler class of GFTs, based on a single GFT  field over a larger data domain. 
The technique of a dual weighting, standard in tensor models, allows then to generate dynamically the class of non-branching simplicial spin-foam molecules $\sfrst_{\copies,\snb}$.  
 Drawing upon the results of \sec{simplicial-structures}, these can be related to a very general class of molecules in $\sfrs$ covering arbitrary boundary graphs. 
Finally I will show how the GFT propagator and interactions can be devised such that the resulting models generate weights for the molecules in $\sfrs$ and that effectively assign to them the amplitudes of 4-dimensional quantum gravity SF theory. 

Sections \ref{sec:gft} and \ref{sec:dw-gft} are based on the publication \cite{\ORT}.


\renewcommand{\rep}[1]{{\rho_{#1}}}

\section[Quantum gravity as discrete gauge theory]
{Quantum gravity as discrete gauge theory
\label{sec:lgt-gravity}}

The essential idea of SF models is to obtain a GR path integral \eqref{gr-path-integral} from the classically equivalent Plebanski-Holst formulation on a discrete manifold instead of a smooth one, extending the well-understood example of quantum BF theory.
The Plebanski-Holst action
\[\label{plebanski-holst-action} 
S_{\textsc{ph}}[\A,B] = S_{\bft} [\A,B]  + S_{\text{simp}}[\A,B]  + S_{\text{Holst}} [\A,B]  
\]
consist of three parts depending most generally on a connection 1-form $\A$ for some group $G$ and a $(\std-2)$-form $B$ valued in the Lie algebra $\la$ of $G$.
The first two parts, the BF action $S_\bft$ together with the constraint term $S_{\text{simp}}$, define the Plebanski action
\[\label{plebanski-action} 
S_{\text{Pl}}[\A,B] = S_{\bft} [\A,B]  + S_{\text{simp}}[\A,B]  
= \frac1 {2\gnb} \int_\stm  \epsilon_{_{IJKL}} \B{IJ} \wedge F^{^{KL}}[\A]
+ \frac1\gnb \int_\stm \lambda_{_{IJKL}} \B{IJ} \wedge \B{KL} 
\]
where $\gnb = 16\pi \gn$ and $G$ is further specified to the local symmetry group of GR such that one can represent $B$ and the curvature $F[\A]$ as antisymmetric tensors on a local frame according to   $\Lambda^2(\R^4)\cong \la = \text{Lie}(G)$.
The Holst action
\[\label{holst-term}
S_{\text{Holst}} [\A,B]  =  \frac1 {\gnb\bi} \int_\stm \B{IJ} \wedge F_{_{IJ}}[\A]
\]
provides a further topological term which allows to relate the phase space of the theory to loop quantum gravity and which turns out to improve the quantum dynamics.


In this section I will review how a path integral $Z(\cp)$ of the Plebanski-Holst theory discretized on a combinatorial manifold $\cp$ reduces to a SF state sum 
\ba
\label{ph-sum}
&& Z(\cp\,) = \int_\cp \D\A \D B \, \e^{\i S_{\textsc{ph}} [\A,B]} \\
\label{sf-sum}
&\lora& Z(\sfr) = \int [\d\sfls] \prod_{\vh\in\Vhat} \sfo_{\vh}(\sfls_{\vh}) \prod_{\vb\in\Vbar} \sfo_{\vb}(\sfls_{\vb}) \prod_{v\in\V}\sfo_v(\sfls_v) \;.
\ea
on the spin-foam molecule $\sfr = \sfr_\cp$ (\dref{molecule}) which is the subdivision of the 2-skeleton of the dual $\cps$ (\cf theory explication \ref{sf-theory}).


The presentation  will be along the following lines.
First, I will give the derivation for the discrete topological (BF) path integral in the way common for SF models, but formulated in the more precise language of discrete exterior calculus 
(\ref{sec:bf}).
Then, in \sec{molecule-formulation}, I will reformulate this path integral revealing the underlying molecular structure of a SF state sum.
Finally, in \sec{gravitational}, I discuss the inclusion of the gravitational part, that is the application of simplicity constraints on the quantum state sum.


\subsection{Discrete BF theory}
\label{sec:bf}

The definition of SF models relies essentially on a quantization method for topological BF theory based on a cellular decomposition of the underlying spacetime manifold $\stm$, \ie a definition of the BF partition function
%
%
%
\[\label{bf-partition-function}
Z_\bft(\stm) = \int \D \A \D B \;\e^{\i\int_\stm tr(B\wedge F[\A])} \;.
\]
In the following I will demonstrate how this partition function can be defined directly for combinatorial manifolds. 

According to the discrete exterior calculus for fields on complexes (\sec{calculus}), let us consider a combinatorial $\std$-dimensional manifold $\cp$.
Then, the $B$ field takes values on $(\std-2)$-cells $\clp{\std-2} \in \cpp{\std-2}$, 
\[\label{discreteB}
B_x \mapsto B_{\clp{\std-2}} = \bra B | \clp{\std-2} \ket .
\]
The wedge product $B\wedge F$ is commonly understood in the discrete such that the 2-form $F$ lives on 2-cells $f=\csp2 = \star\clp{\std-2}$ in $\cps$, the dual complex%
\footnote{In contrast to this, attempts to define a wedge product in the mathematical literature of discrete exterior calculus \cite{Desbrun:2005ug} are build on the cup product of algebraic topology.
In that case all terms in the product are exterior forms on the primal complex.},
\[\label{discreteF}
F_x \mapsto \bra f | F \ket = \bra \clp{\std-2} | F \ket
\]
such that one obtains a fundamentally discrete form of the BF action as a functional of combinatorial $\std$-manifolds $\cp$:
\[
\int_\stm \tr(B\wedge F) \mapsto
 \sum_{\clp{\std-2} \in \cp} \tr \left[ \bra B | \clp{\std-2} \ket \bra \clp{\std-2} | F \ket \right]
= \sum_{f\in\cpsp2} \tr\left[\bra B | f \ket \bra  f | F \ket \right] .
\]
This defines the BF partition function \eqref{bf-partition-function} more explicitly specifying the formal integration measures $\D\A$ and $\D B$.
For consistency with the definition of its curvature $F$ \eqref{discreteF},
the connection $\A$ is defined as a discrete 1-form on dual edges $e=\csp1\in\cpsp1$ 
such that (with the slight abuse of notation $B_\cl = \bra B | \cl \ket$ etc., in contrast to \eqref{eq:primal-smearing})
\begin{eqnarray}
Z_{\bft}(\cp) 
&=& \prod_{e\in\cpsp{1}} \int \d\A_{e} \prod_{f\in\cpsp{2}}
\left(\int \d B_{f}\, \e^{ \i \, \tr(B_{f} F_{f}[\A_{e}])}\right) \,.
\end{eqnarray}
The field $B$  can formally be integrated out
\[
Z_{\bft}(\cp) = \prod_{e\in\cpsp{1}} \int \d\A_{e} \prod_{f\in\cpsp{2}} \dla( F_{f}[\A_{e} ])
\]
showing that the integration is effectively only over flat connections, implemented here by the delta function $\dla$ on the Lie algebra $\la$.


For the combinatorial $\std$-manifold $\cp$ a change of variables from connections to holonomies proves especially practical.
On a smooth manifold $\stm$, a holonomy is the path ordered exponentials of the connection integrated along a closed curve $\gamma$ starting and ending at a point $x\in\stm$,
\[
H_x [\gamma,\A] = \Pcal \e^{-\oint_\gamma \A} .
\]
At each point $x$ these form a subgroup of the gauge group $G$.
In these terms, the curvature at $x$ is the differential of holonomy at identity (according to the Ambrose-Singer theorem \cite{Nakahara:2003vx}).
That is, infinitesimally, \ie for curves $\gamma$ along which the connection $\A$ is varying slowly,
\[
H_x [\gamma,\A] \approx \id - F_x [\A] .
\]
In the discrete analogue, the smallest closed curves on the dual complex $\cps$ are those around faces $f\in\cpsp2$ 
such that the delta function for curvature translates into a delta function $\dg$ on the group $G$
\[\label{deltaH}
\dla( F_{f}[\A ]) \lora \dg(H_{f}[\A])
\]
for the discrete version of holonomy
\[\label{curvature-holonomy}
H_{f}[\A] := H_v [\bs f,\A]
= \Pcal_v \e^{-\sum_{e < f} \bc{f}{e} \A_{e} }
= \Pcal_v \prod_{e < f} \left(\e^{- \A_{e} }\right)^{\bc{f}{e}}
\]
where the discrete path-ordered product  $\Pcal_v$ is simply determined by a start vertex $v<f$, though the delta function \eqref{deltaH} is independent of this starting point. 
This motivates the change to group variables on edges $e\in\cpsp1$
\[
h_{e} =  \e^{- \A_{e} }\in \G
\]
which are called \emph{holonomy} variables as well in quantum gravity. 
The state sum running over these degrees of freedom is thus
\[\label{BF-group}
\boxd{
Z_{\bft}(\cp) = \prod_{e\in\cpsp{1}} \int \d h_{e} \prod_{f\in\cpsp{2}}  \dg \left( {\prod}_{e < f} {h_{e}}^{\bc{f}{e}} \right)
}
\]
where $\d h_e$ is the Haar measure on $G$ and the path ordering in the product $\prod_{e<f}$ of group elements 
is left implicit here and in the following.
This state sum has the standard form of a lattice gauge theory \cite{Oeckl:2005wg} with particular face weight $w_f = \dg$.

The crucial idea for further simplification of this state sum is to transform into the representation space of the group $G$. 
According to the Peter-Weyl theorem, square integrable functions on a compact topological group can be decomposed into a direct sum of its irreducible unitary representations $\rep{}$.
This works in particular  for the delta function
\[
\dg(h) = \sum_\rep{} \dj{\rep{}}\, \tr_{\rep{}} (D^\rep{}(h) )
\]
where $\dj{\rep{}}$ is the dimension of the irreducible unitary representation $\rep{}$ and $D^{\rep{}}(h)$ denote its representation matrices.
Since the matrices $D^{\rep{}}(h)$ are furthermore group homomorphisms, they commute with group multiplication such that
\[\label{delta-expansion}
 \dg \left( {\prod}_{e < f} h_{e} \right)	
= \sum_{\rep{f}} \dj{\rep{f}} \, \tr_{\rep{f}} \left[D^{\rep{f}} \left({\prod}_{e} h_{e} \right) \right]
= \sum_{\rep{f}} \dj{\rep{f}} \, \tr_{\rep{f}} \left[ \left({\prod}_{e} D^{\rep{f}}(h_{e}) \right) \right]
\]
where here and in the following the relative orientations $\bc{f}{e}$ are left implicit.
This factorization (inside the traces) allows to separate the path integral into a product of single integrals
\[\label{projector}
P_{e} (\{\rep{f}\}):= \int \d h_{e} \prod_{f > e} D^{\rep{f}}(h_{e})
\]
which are projectors on the $G$-invariant subspace of the tensor product of representation due to the group averaging. 
This simplifies the state sum to
\[\label{bf-rep}
Z_{\bft}(\cp) = \sum_{\{\rep{f}\}} \prod_{f\in\cpsp{2}} \dj{\rep{f}} \, \tr_{\rep{f}} \prod_{e\in \cpsp1} P_{e}(\{\rep{f}\})\;.
\]

Finally, an expansion of the projectors $P_e$ in the intertwiner basis of the respective tensor product of representations results in a factorization of the face traces $\tr_{\rep f}$ which leads to a SF state sum \eqref{sf-sum} with local amplitudes. 
To illustrate this 
and explain the relevance of its combinatorial and algebraic ingredients 
I discuss as explicit examples briefly the standard case of simplicial BF theory in three and four dimensions.

\subsubsection*{Simplicial BF theory in three dimensions}

In $\std = 3$ dimensions the SF state sum of BF theory with group $G=SO(3)$ is particularly simple.
On a simplicial 3-complex $\cp$ the edges $e\in\cps$ are dual to triangles.
Therefore the corresponding integral \eqref{projector} contains three representations for adjacent dual faces $f = 1,2,3$ and can be evaluated using $3j$ symbols
\[
\int \d h_{e}  D^{\rep1}_{m_1 n_1}(h_{e})D^{\rep2}_{m_2 n_2}(h_{e})D^{\rep3}_{m_3 n_3}(h_{e})
 =  \threej 1 2 3 m \threej 1 2 3 n \,.
\]
In this case $P_{e}$ is obviously a projector and the $SO(3)$-invariant subspace of the triple tensor product is one-dimensional.
Moreover one observes a factorization of the integral into two parts, one for each vertex of the edge.
Thus, the traces in the state sum reduce to pairings of the two representations of a face adjacent to the same vertex. 

On the simplicial complex this results in a combination of four $3j$ symbols according to the structure of the simplicial boundary graph $\bbg_{3,\simplicial} = \bs\sfa_{3,\simplicial}$ (\dref{simplicial}, see \fig{dualtet}) which is the $6j$ symbol \cite{NIST}:
\begin{figure}
  \centering
  \tikzsetnextfilename{dualtet}
  \begin{tikzpicture}


  \begin{scope}[xshift=0cm,scale = 1.8]
  \node [vb, label=left:3]
  (3)	at (-1.1,-.2){};
  \node [vb, label=-90:4]
  (4)	at (0,-.8)	{};
  \node [vb, label=right:1]
  (1)	at (1,0) 	{};
  \node [vb, label=90:2]
  (2)	at (.2,1)	{};
  \foreach \i/\j in {1/2,2/4,2/3,1/4}{
    \draw [eb] (\i) -- node[vh, label=0:$\rep{\i\j}$] {} (\j); 
    }
  \foreach \i/\j in {1/3,3/4}{
    \draw [eb] (\i) -- node[vh, label=below:$\rep{\i\j}$] {} (\j); 
    }
  \end{scope}

%

\begin{scope}[xshift=-6.cm,scale=1.5]
\node [c]					(v)	at (0,-.19)		{};
\node [c,label=0:1]			(1)	at (.5,0)	 	{}; 
\node [c,label=90:2]			(2)	at (-.12,.12)	{}; 
\node [c,label=180:3]		(3)	at (-.32,-.14)	{}; 
\node [c,label=-90:4]			(4)	at (0,-.84)		{};
\node [c,label=0:(12)]		(12)	at (.55,.55)	{};
\node [c
				]		(13)	at (.21,0)		{};
\node [c,label=0:(14)]		(14)	at (.77,-.83)	{};
\node [c,label=135:(23)]		(23)	at (-.73,.31)	{};
\node [c
				] 		(24)	at (-.17,-.51)	{};
\node [c,label=200:(34)]		(34)	at (-.5,-1.06)	{};
\foreach \i/\j in {1/2,1/3,1/4,2/3,2/4,3/4}{
 \path	[f] 	(\i) -- (\i\j) -- (\j) -- (v) -- cycle;
 }
 \foreach \i in {1,2,3,4}{
  \draw [e] (\i) -- (v);
  }
\foreach \i/\j in {1/2,1/3,1/4,2/3,2/4,3/4}{
 \draw 	[eh]	(v)		-- 	(\i\j);
 \draw	[eb] 	(\i) node[vb] {} -- (\i\j) node[vh] {};
 \draw 	[eb] 	(\j) node[vb] {} -- (\i\j) node[vh] {};
 }
\draw [e] (3) node[vb] {} -- (v) node[v] {};
\draw (0,1.35)		-- (.43,-1.37);
\draw (0,1.35)		-- (1.11,-.27);
\draw (0,1.35)		-- (-1.44,-.73);
\draw (.43,-1.37)	-- (1.11,-.27);
\draw (.43,-1.37)	-- (-1.44,-.73);
\draw (1.11,-.27)	-- (-1.44,-.73);
\end{scope}

\end{tikzpicture}
\caption{\label{fig:dualtet} Left: a tetrahedron with the 2-skeleton of its dual which is the spin-foam atom $\sfa_{3,\simplicial}$. 
Right: the corresponding boundary graph $\bbg_{3,\simplicial}=\bs\sfa_{3,\simplicial}$ with face representation $\rep{ij}$ associated with edges on this boundary graph.}
\end{figure}
\ba\label{6j}
\sixj{12}{13}{14}{34}{24}{23} 
& = & \sum_{m_{ij}} (-1)^{\rep{34}+\rep{24}+\rep{23}+m_{34}+m_{24}+m_{23}}  \threej{12}{13}{14}m \\
& & \times \threej{12}{23}{24}m \threej{13}{23}{34}m \threej{14}{24}{34}m \nonumber \;.
\ea
Upon these modifications the state sum simplifies to
\[\label{BF-3d}
Z_{\bft}(\cp) = \sum_{\{\rep f \}} \prod_{f \in\cpsp{2}} (-1)^{c(\rep f)} \dj{\rep{f}} \prod_{v\in\cpsp0} \{6j \}_v
\]
where $\{ 6j \}$ refers to the above $6j$ symbol \eqref{6j} and $c(\rep f)$ is some linear combination of representation labels.
This partition function is explicitly of the form of a SF state sum $\eqref{sf-sum}$.
It is the partition function of the  Ponzano-Regge model \cite{Ponzano:1968wi} for $G=SO(3)$, already discussed in the form \eqref{pr-model} in \sec{dynamical-relations} from the point of view of quantum Regge calculus.
In particular, it is the state sum of Euclidean quantum gravity in $\std=3$ since the BF action is classically equivalent to GR, the only change being the choice of ``Euclidean" gauge group $SO(3)$  instead of the 3-dimensional Lorentz group $SO(2,1)$.

\subsubsection*{Simplicial BF theory in four dimensions}

Algebraically, GR-related BF theory in $\std=4$ does not differ substantially from the $\std = 3$ case when one chooses $G=Spin(1,3)\cong SL(2,\C)$ or $Spin(4)\cong SU(2)\times SU(2)$ which leads to classically equivalent formulations in the Plebanski formulation \cite{Baez:2000kp}.
For a simple model, I will thus stick to $G=SO(3)\cong SU(2)$ here.

Combinatorially, on the other hand, already for the simplest case of an abstract simplicial 4-complex $\cp$, the intertwiners are not trivial but there is an infinite basis to be summed over.
Group averaging four $SO(3)$ representation matrices yields
\ba
\int \d h_{e}  D^{\rep1}_{m_1 n_1}(h_{e})D^{\rep2}_{m_2 n_2}(h_{e}) D^{\rep3}_{m_3 n_3}(h_{e}) D^{\rep4}_{m_4 n_4}(h_{e}) &  \\ \nonumber
=  \sum_{\rep e}  \sum_{m_e,n_e}  \threej 1 2 e m \threej 3 4 e m & \threej 1 2 e n  \threej 3 4 e n \;.
\ea
Reordering of face traces $\tr_{\rep f}$ in \eqref{bf-rep} into contractions along the dual boundary graphs $\bbg_{4,\simplicial}$ of the 4-simplices in the simplicial complex (\fig{15j}) leads to a contraction of ten $3j$ symbols, two for each of the five edge projectors $P_e$.
Such a contraction at a dual vertex $v$ is a $15j$ symbol $\{15j\}_v$ as it depends on fifteen representations, ten for the faces and five for the edge intertwiners. 
Accordingly, the resulting partition function
\[\label{BF-4d}
Z_{\bft}(\cp) = \sum_{\{\rep f \}} \sum_{\{\rep e \}} \prod_{f \in\cpsp{2}} \dj{\rep{f}} \prod_{v\in\cpsp0} \{15j \}_v
\]
has a similarly compact form as in $\std=3$ dimensions.

\begin{figure}
  \centering
  \tikzsetnextfilename{15j}
  \begin{tikzpicture}[scale=2.5]
    \node [vb,label=90:$\rep1$]	(1)	at (0,0.82)		{};
    \node [vb,label=180:$\rep2$]	(2)	at (-1,.1) 		{}; 
    \node [vb,label=-90:$\rep3$]	(3)	at (-0.64,-1) 	{}; 
    \node [vb,label=-90:$\rep4$]	(4)	at (0.64,-1)	{}; 
    \node [vb,label=0:$\rep5$]	(5)	at (1,.1)		{}; 
    \node at (.3,-.6) 	{$\rep{35}$};
    \foreach \i/\j in {2/5,3/4,1/2,1/3,1/5,2/3,2/4,4/5,1/4}
      \draw [eb] (\i) -- node [auto] 	{$\rep{\i\j}$} 	(\j);
    \foreach \i/\j in {2/5,3/4,3/5,1/2,1/3,1/5,2/3,2/4,4/5,1/4}
      \draw [eb] (\i) -- node [vh] 	{}			(\j);
  \end{tikzpicture}
  \caption{The fifteen representations associated with the dual boundary graph $\bbg_{4,\simplicial}$ of a combinatorial 4-simplex.}
  \label{fig:15j}
\end{figure}
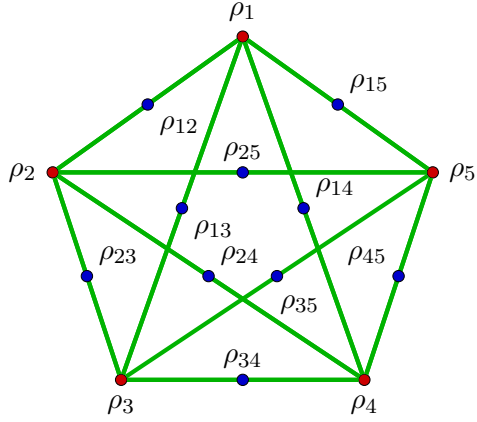

\

Apart from the concise form of the partition function there is no need for a restriction to simplicial complexes for SF models.
The same modification of the BF path integral \eqref{bf-rep} into a state sum with vertex amplitudes \eqref{BF-3d}, \eqref{BF-4d} are possible for any polyhedral complex, and even more, fore any (dual) 2-complex.
For an arbitrary number of representations, the projector \eqref{projector} evaluates to a sum over $3j$ symbols associated with either one of the edge ends. 
These symbols are contracted at each vertex according to its combinatorial type which is just the spin-foam atom as introduced in \sec{atoms}. 
Thus, a generic BF partition function actually depends only on the spin-foam molecule $\sfr = \sfr_\cp$ of a combinatorial $\std$-complex, which is the subdivision of its dual 2-skeleton.
Moreover, for this reason there is a straightforward generalization of the partition function to arbitrary SF molecules $\sfr\in\sfrs$.  
In the next section, I will present a formulation of SF models which shows the spin-foam-molecule structure of the partition function more explicitly from the beginning.

\subsection{Spin-foam molecule formulation}
\label{sec:molecule-formulation}


It is illuminating to reformulate the state sum's construction emphasizing the underlying spin-foam-molecule structure (cf. \cite{\KKL}).
This can be done in a constructive way according to the constructive nature of their combinatorial structure as introduced in \sec{atoms}.

Take as a starting point the Hilbert space of gauge-invariant boundary states of the present discrete gauge theory
\[\label{state-space}
\hs = \bigoplus_{\bg\in\bgs} \hs_\bg \quad \textrm{ where } \hs_\bg = L^2(G^{|\E_\bg|} / G^{|\V_\bg|})\;.
\]
As mentioned before (\sec{dynamical-relations}), this space is very similar to the kinematical state space 
in LQG with the main difference that it is a simple direct sum of Hilbert spaces (instead of a projective limit, or the union of spaces related by cylindrical consistency) over abstract graphs (without any embedding into a background manifold $\sm$).

A complete and orthogonal basis of \eqref{state-space} is given by spin-network states \cite{Rovelli:2004wb}.
A \emph{spin network} $(\bg,\rep{\eb},\intw{\vb})$ is defined by assigning $\G$-representations $\rep{\eb}$ to all edges $\eb$ and intertwiners $\intw{\vb}$ to all vertices $\vb$ in a closed graph $\bg\in\bgs$ (equipped with some orientation $\bc{\eb}{\vb}$ such that each $\intw{\vb}$ is a G-invariant map from representation spaces on edges $\eb$ with $\bc{\eb}{\vb}=-1$ to those with $\bc{\eb}{\vb}=+1$). 

Contracting all the intertwiners according to the graph structure yields an amplitude
\[\label{sn-amplitude1}
\sfo(\bg,\rep{\eb},\intw{\vb}) := \tr_{\rep{\eb}} \bigotimes_{\vb\in\bg} \intw{\vb}\;.
\]
This spin network amplitude is invariant under change of graph orientation.
Quite generally, observables in the theory are orientation-independent, though orientability of graphs and 2-complexes is necessary.

Spin networks can be generated from intertwiners like (bisected) boundary graphs  are generated from patches ($\bbgs = \sigma(\bps)$, prop. \ref{prop:generator}).
Each intertwiner $\intw{\vb}\equiv\intw{\bp_\vb}$ inherits the combinatorial structure of a patch $\bp_\vb$ when considered as  mapping between representations $\rep{\vb\vh}$ on half edges $(\vb\vh)\in\bp_\vb$. 
A set of intertwiners $\{\intw{\bp_\vb}\}$ generates then a spin network $(\bbg,\rep{\vb\vh},\intw{\bp_\vb})$ 
if $\bbg=\sigma(\{\bp_\vb\})$ and if representations are pairwise identical, $\rep{\vb_1\vh}=\rep{\vb_2\vh}\equiv\rep{\vh}$, along all bisecting vertices $\vh\in\bbg$:
\begin{center}
\tikzsetnextfilename{delta1}
\begin{tikzpicture}
\draw [eb] 	(-2,0) node [vb,label=below:$\vb_1$] {}  --  node[auto] {$\rep{\vb_1\vh}$} 
			(0,0) node[vh,label=below:$\vh$] {} -- node[auto] {$\rep{\vb_2\vh}$}  
			(2,0) node[vb,label=below:$\vb_2$] {};
\draw [->] 	(3,0) -- node[auto] {
} (5,0);
\draw [eb] 	(6,0) node [vb,label=below:$\vb_1$] {}  --  
			(8,0) node[vh,label=below:$\vh$,label=above:$\rep{\vh}$] {} -- 
			(10,0) node[vb,label=below:$\vb_2$] {};
\end{tikzpicture}
\end{center}

This construction is most explicit in the second quantized reformulation of LQG in terms of GFT where one-particle field excitations are exactly such patches labelled with algebraic data, and identifications of data are effected by appropriate superposition states \cite{Oriti:2013vv,Kittel:2014vo}.
In this notation, the spin network amplitude \eqref{sn-amplitude1} reads
\[\label{sn-amplitude2}
\sfo(\bbg,\rep{\vb\vh},\intw{\bp_\vb}) = \tr_{\rep{\vh}} \bigotimes_{\vb\in\bbg} \intw{\bp_\vb}\;.
\]

\

Spin-foam amplitudes are then constructed from spin-network amplitudes as the underlying spin-foam molecules are obtained from bonding atoms. 
The spin network amplitude \eqref{sn-amplitude2} is precisely the generic form of a vertex amplitude in the BF state sum, attached to the spin-foam atom $\sfa=\bulk(\bbg)$ via the bulk map $\bulk$ (definition \ref{def:sf-atom}). 
To distinguish them as such, all variables may further be labelled by the atom's vertex $v$, 
\ie representations $\rep{v\vb\vh}$ and intertwiners $\intw{\bp_{v\vb}}=\intw{\bp_\vb (\bbg)}$, as well as the atom and graph $\sfa_v = \bulk(\bbg_v)$ itself.
The $\{6j\}$ and $\{15j\}$ symbols in \eqref{BF-3d} and \eqref{BF-4d} are just the special cases of simplicial, \ie complete, graphs $\bbg_{3,\simplicial}$ and $\bbg_{4,\simplicial}$, 
\[\label{simplicial-sn-amplitudes}
\{6j\}_v 
= \sfo(\bbg_{3,\simplicial},\rep{v\vb\vh}) \quad , \quad 
\{15j\}_v 
= \sfo(\bbg_{4,\simplicial},\rep{v\vb\vh},\intw{\bp_{v\vb}}) \;.
\]

A set of spin networks $\{(\bbg_v,\rep{v\vb\vh},\intw{\bp_{v\vb}})\}_{v\in V}$ 
defines then a \emph{spin-foam} $(\sfr,\rep{v\vb\vh},\intw{\bp_{v\vb}})$ 
if there are bonding maps $\gm_j$ such that $\sfr = \sharp_{\{\gm_j\}} \{\bbg_v\}$ and if intertwiners are compatible with these bondings depending on the particular model. 
For BF theory the compatibility condition requires 
\[\label{BF-compatibility}
\intw{\bp_{v_1\vb}}=\intw{\bp_{v_2\vb}}\equiv\intw{\vb} \quad \textrm{ for every bonding } 
\gm:\bp_{v_1\vb}\ra\bp_{v_2\vb} \text{ in } \{\gm_j\} \;:
\]
\begin{center}
\tikzsetnextfilename{delta2}
\begin{tikzpicture}[scale=1.5]
\node [vb,label=above:$\vb$]		(a)	at (-0.5,0)		{}; 
\node [vh,label=above:$\vh$]		(1)	at (-.75,.5)		{};
\node [vh]		(2)	at (-.75,-.5)	{}; 
\node [vh,label=above:$\intw{\bp_{v_1\vb}}$]		(3)	at (-1,0)		{}; 
\node [vb,label=above:$\vb$]		(b)	at (0.5,0)		{}; 
\node [vh,label=above:$\vh$]		(4)	at (.75,.5)		{};
\node [vh]		(5)	at (.75,-.5)		{}; 
\node [vh,label=above:$\intw{\bp_{v_2\vb}}$]		(6)	at (1,0)		{}; 
\foreach \i/\j in {a/1,a/2,a/3,b/4,b/5,b/6}{
 \draw [eb] (\i) --  (\j);
  }
\path	(1) 	edge [bh] 				(4)
	(2)	edge [bh]		 		(5)
	(3)	edge [bh, bend right=20]	(6)
	(a)	edge [bb]				(b);
\draw [->] 	(2,0) -- node[auto] {
} (3,0);	
\begin{scope}[xshift=5cm]
\node [vb,label=-80:$\vb$,,label=80:$\intw{\bp_{\vb}}$]		(a)	at (-0.5,0)		{}; 
\node [vh,label=above:$\vh$]		(1)	at (-.75,.5)		{};
\node [vh]		(2)	at (-.75,-.5)	{}; 
\node [vh]		(3)	at (-1,0)		{}; 
\foreach \i/\j in {a/1,a/2,a/3}{
 \draw [eb] (\i) --  (\j);
  }
\end{scope}
\end{tikzpicture}
\end{center}
Note that this implies in particular that all representations on the patches are identified pairwise $\rep{v_1\vb\vh}=\rep{v_2\vb\vh}\equiv\rep{\vb\vh}$ according to the bonding $\gm$.
With the identifications of representations on each spin network this results in a single representation $\rep{\vh}$ for each $\vh\in\Vh$, in analogy to \eqref{bf-rep}.
Thus, one can define the \emph{BF spin-foam amplitude} for a BF spin foam $(\sfr,\rep{v\vb\vh},\intw{\bp_{v\vb}})$ to be
\[ 
\sfo_\bft(\sfr,\rep{v\vb\vh},\intw{\bp_{v\vb}}) = \prod_{\vh\in\Vhat} \dj{\rep\vh} \prod_{v\in V} \sfo(\bbg_v,\rep{v\vb\vh},\intw{\bp_{v\vb}}) \;
\]
and the BF partition function is simply the sum of these amplitudes over identified re\-presentations $\rep{\vh}$ and intertwiners $\intw{\vb}$.
Since the partition function in this form is explicitly a sum over internal spin networks, it is a ``state sum" in a particularly meaningful way.%
\footnote{Partition functions for topological field theories and SF models are usually called ``state sums" in the first place simply because they are considered as statistical partition functions where then states correspond to \emph{spacetime} configurations; the present formulation adds another meaning since the path integral here is really a sum over internal $\sd$-dimensional (``\emph{spatial}") states.
}

\

It is useful to render the model-dependent compatibility conditions \eqref{BF-compatibility} explicit in the state sum, thus summing over generic, model-independent spin foams without any a-priori relations between the basic variables $\rep{v\vb\vh}$ and $\intw{\bp_{v\vb}}$.
Since these conditions relate variables on patches along bondings, they are formulated in terms of edge amplitudes $\sfo_\vb(\intw{\bp_{v_1\vb}},\intw{\bp_{v_2\vb}})$.
In the case of BF theory these are simply delta functions 
\[
\sfo_{\vb,\bft}(\intw{\bp_{v_1\vb}},\intw{\bp_{v_2\vb}}) = \delta(\intw{\bp_{v_1\vb}},\intw{\bp_{v_2\vb}})
\]
which identify at the same time the representations which the intertwiners  are defined on, as well as the particular intertwiner maps themselves.
The generalization of \eqref{bf-rep} to arbitrary spin-foam molecules is then
\[
Z_{\bft}(\sfr) = \sum_{\{\rep{v\vb\vh}\}} \sum_{\{\intw{\bp_{v\vb}}\}}
\left( \prod_{\vh\in\Vhat} \dj{\rep\vh} \right)
\left( \prod_{\vb\in\Vbar}\delta(\intw{\bp_{v_1\vb}},\intw{\bp_{v_2\vb}})
\right)	
\left( \prod_{v\in\V}		\sfo(\bbg_v,\rep{v\vb\vh},\intw{\bp_{v\vb}}) 
\right) 
\]
where $\rep\vh$ in the face amplitude denotes the resulting variable after identification of all $\rep{v\vb\vh}$ at $\vh$.

An equivalent description of these models is the operator spin foam formulation \cite{Bahr:2011ey}.
There, a spin foam is an assignment of irreducible $\G$-representations to faces $f\in\cpsp2$ inducing $\G$-invariant intertwiner Hilbert spaces on edges $e\in\cpsp1$, together with a class of operators $P_e$ on these spaces.
These operators should be understood exactly as the maps instantiating the compatibility conditions \eqref{BF-compatibility} along bondings.
Again, this formalism extends \eqref{bf-rep} by allowing for operators $P_e$ different from the BF projectors \eqref{projector} in which intertwiners are simply identified.
In this way it allows relaxations of the flatness condition.

\subsubsection*{Alternative formulations and variables}

There are various similar formulations of SF models using slightly different variables each, but which are all based on the spin-foam molecule structure.
In the above derivation it is the change from holonomy to representation variables \eqref{delta-expansion} which effects the transition from edge variables to face variables, generalized further to variables on the boundary of spin-foam atoms.
Similar localization to variables on the atoms are equally well possible in holonomy variables (and also in $B$ field Lie algebra variables \cite{Baratin:2010in,Baratin:2012br}).
This is not only interesting for graviational models 
but in particular relevant for the field-theoretical description in terms of GFT as well as for compatibility with lattice gauge theory tools such as of coarse graining methods \cite{Dittrich:2012he,Dittrich:2012ba,Bahr:2013ek}.



The so called holonomy spin-foam formulation \cite{Bahr:2013ek} takes the BF path integral over holonomy variables \eqref{BF-group} as a starting point and implements relaxations of the flatness condition by introducing additional group variables.
Instead of one holonomy $h_{e}$ per edge $e=(v_1v_2)\in\cpsp1$, there are two holonomies $h_{v_1\vb}$, $h_{v_2\vb}$ for the corresponding half-edges $(v_1\vb), (v_2\vb)$ as well as a group variable $g_{\vb\vh}$ for each incident face with subdivision vertex $\vh$.
While the holonomies are simply doubled in this way, the variables $g_{\vb\vh}$ are controlled by a factorizing edge amplitude $A_\vb(g_{\vb\vh}) = \prod_{(\vb\vh)}E(g_{\vb\vh})$ such that the path integral takes the general form
\[\label{holonomy-sf}
Z(\sfr)
=  \int [\d h_{v\vb}]  \int [\d g_{\vb\vh} ]
\prod_{(\vb\vh)} E(g_{\vb\vh}) 
\prod_\vh \dg \left( \prod_{\vb<\vh} (h_{v_1\vb} g_{\vb\vh} {h_{v_2\vb}}^{-1})^{\bc{\vh}{\vb}} \right) 
\]
where $\bc{\vh}{\vb}$ is just an alternative notation for the relative face-edge orientations in terms of their corresponding subdivision vertices.
This integral is already genuinely based on the subdivision of the (2-skeleton) of $\cps$, and thus on a molecule $\sfr=\sfr_\cp$.
Integrating out the group variables $g_{\vb\vh}$ yields a discrete gauge theory with effective face weights $w_\vh$ which trivialize to BF theory with $w_\vh=\dg$ for the special case of edge amplitudes with $E=\dg$. 
The holonomy spin-foam formulation is a special case of the operator spin-foam formulation due to the factorizing property of the edge amplitude \cite{Bahr:2013ek}.

From the GFT perspective, the most interesting form which the state sum \eqref{holonomy-sf} can be recast into \cite{Bahr:2013ek}  is one involving only edge and vertex amplitudes.
A change to group variables $g_{v\vb\vh}$ allows to rewrite the state sum as
\[\label{gft-sf}
Z(\sfr)
=  \int [\d g_{v\vb\vh} ]
\prod_{\vb} \Pbb(g_{v_1\vb},g_{v_2\vb}) 
\prod_v \Vbb(g_v) 
\]
where $\Pbb(g_{v_1\vb},g_{v_2\vb})$ is an edge amplitude depending on all $g_{v\vb\vh}$ incident to the edge bisecting vertex $\vb$, thus involving two adjacent vertices $v_1$ and $v_2$, and $\Vbb(g_v) $ is a vertex amplitude depending on all $g_{v\vb\vh}$ incident to the vertex $v$.
Similar to the formulation in representation variables \eqref{bf-rep}, the edge function $\Pbb$ effectively depends only on combinations of variables $g_{v_1\vb\vh}g_{v_2\vb\vh}^{-1}$ and the vertex function $\Vbb$ depends only on combinations $g_{v\vb_1\vh}g_{v\vb_2\vh}^{-1}$. 
All details will be discussed in \sec{gft} where I will explain the field-theoretic generation of SF state sums  \eqref{gft-sf} in terms of GFT.


\subsection{Gravitational models \label{sec:gravitational}}


In 4-dimensional spacetime, GR dynamics are recovered from the BF action when adding simplicity constraints.
These are encoded in the Plebanski action \eqref{plebanski-holst-action} in terms of the Lagrange-multiplier term quadratic in the field $B$.
There is a range of strategies how the resulting constraints can be applied in the quantum theory, most prominently described by the the models associated with the names Engle-Pereira-Rovelli-Livine (EPRL) \cite{\EPRL}, Freidel-Krasnov (FK) \cite{\FK} and Baratin-Oriti (BO) \cite{\BO}.
Moreover, any of these models is shown to permit in principle an extension to arbitrary spin-foam molecules $\sfrs$  \cite{\KKL}, called its KKL-extension.

For the present purpose, \ie the definition of GFT models for these gravitational SF models, the particular choices in each model are not as relevant as the general form of the simplicity constraints in the resulting SF amplitude.
For this reason I will explain in this section only the relevant generic aspects and illustrate them with one particular example.
%

The constraint equations of motion for the Lagrange multiplier resulting from the classical smooth Plebanski action \eqref{plebanski-action} are  
\[
\B{IJ} \wedge \B{KL} = \bvol\,\epsilon^{^{IJKL}}
\]
with  volume 4-form $\bvol = \frac1{4!} \epsilon_{_{MNOP}} \B{MN} \wedge \B{OP}$.
These are the classical simplicity constraints.
They have four sectors of solutions:
\ba
\B{IJ} & = & \pm {\epsilon^{^{IJ}}}_{_{KL}} \eu K \wedge \eu L \;, \\
\B{IJ} & = & \pm \eu I \wedge \eu J \;.
\ea
Inserting the first solution (with `+' sign) into the Plebanski action \eqref{plebanski-action} yields then the Palatini action of GR. In this sense they are equivalent.

On a simplicial molecule $\sfr \in \sfrs_{4,\simplicial}$,  the simplicity constraints can be implemented locally on each simplicial atom $\sfa_{4,\simplicial} \In \sfr$ using the canonical extension to a 4-simplex (remark \ref{rem:simplicial}). 
Accordingly, the field $\B{}$ lives on triangles $\s_2$ which are in one-to-one correspondence to the bisecting boundary vertices $\vh\in \sfa_{4,\simplicial}$. 
Moreover, corresponding to a Lie-algebra-valued 2-form they have the meaning of area 2-forms in the local frame of the 4-simplex. 
Indeed, the discrete simplicity constraint equations effect that in this sense the discrete version of the field $\B{}$ provides the geometry of a 4-simplex in flat spacetime, often called a ``geometric 4-simplex" \cite{Barrett:1998fp,Barrett:2000fr}. 
That is, the discrete field $B_\vh=(B|\vh)$ \eqref{discreteB} satisfying the discrete simplicity constraints induces a geometric interpretation to each $\sfa_{4,\simplicial} \In \sfr$, and thus to the simplicial molecule $\sfr$ as a generalized simplicial complex (remark \ref{rem:simplicial}) as well.


For SF models, an alternative, slightly stronger version of these discrete quadratic simplicity constraints is typically used: the linear simplicity constraints \cite{Freidel:2008fv,Alexandrov:2008im,Gielen:2010ek} which, applied on each tetrahedron dual to a vertex $\vb$ on the boundary of atoms $\sfa_{4,\simplicial}$, read
\[\label{linear-simplicity}
\sum_{\vh \in\sfa_{4,\simplicial}} \B{IJ}_\vh N_{_J} (\vb) = 0
\]
where ${N_{_J}}(\vb)$ is the normal directing the tetrahedron in internal space.

The addition of the Holst term \eqref{holst-term} in the Plebanski-Holst action \eqref{plebanski-holst-action} results simply in a change to variables $\B{IJ} - \bi \frac1 2 {\epsilon^{^{IJ}}}_{_{KL}} \B{KL}$ (for $\bi>0$) in the linear simplicity constraints \eqref{linear-simplicity}.
These constraints are applied on the states in the state sum as operator equations which leads to a restriction of allowed representations to sum over, or to non-commutative $\delta$-functions in the $B$-field representation.

As a concrete example, consider the Euclidean version with local symmetry group $G=SO(4)$.
The restriction on the representations of $\mathfrak{g} = Lie(G) = \so(4) \cong \su(2)_+\times\su(2)_-$ is then then given by 
 \begin{equation}
      \mathcal{J} = \left\{J\in \textrm{Irrep}(\mathfrak{g}): J = (j_+,j_-) \;\textrm{with}\; \tfrac{j_-}{j_+} = \tfrac{|1-\bi|}{1+\bi}\;\textrm{and} \; j_\pm\in \N/2\right\}\;.
    \end{equation}
which are called $\bi$-simple representations. 
Choosing with hindsight the spin-foam formulation \eqref{gft-sf}, this yields for each pair of 
edges $(v_1\vb),(v_2\vb)$ in the spin-foam molecule $\sfr$ the amplitude
\begin{equation}
  \label{edge-amplitude} 
  \sfo_{\vb}(g_{\vb}) \equiv \Pbb(g_{v_1\vb}, g_{v_2\vb})
  = \int_{G^{2}} \d h_{v_1\vb}\, \d h_{v_2\vb}\;\prod_{(\vb\vh)} 
   \pbb (g_{v_1\vb\vh},g_{v_2\vb\vh}; h_{v_1\vb},h_{v_2\vb})
   \]
that factorizes across the boundary edges $(\vb\vh)$ incident to the vertex $\vb$ into operators
  \[\label{edge-operator}
  \pbb (g_{v_1\vb\vh},g_{v_2\vb\vh}; h_{v_1\vb},h_{v_2\vb})
=  \sum_{J_{\vh}\in \mathcal{J}} \tr_{J_{\vh}}\left(g_{v_1\vb\vh}\;h_{v_1\vb}^{-1}\; \Sbb_{J_{\vh},N_0} h_{v_2\vb}\; g_{v_2\vb\vh}^{-1}\right) \;.
\]
Here,
  $\Sbb_{J,N}$ is the gauge-covariant simplicity operator
\[
\Sbb_{J,N} = \left\{
    \begin{array}{ll}
      \displaystyle \dj{J} \int_{\mathcal{S}_{N}^2}\d \vec{n}\; |j_+ \vec{n}\qket {|j_- \vec n\qket}\, \qbra j_+ \vec n| {\qbra j_- \vec n |}\;,& (\bi > 1)\\[0.4cm]
      \displaystyle \dj{J}\int_{\mathcal{S}_{N}^2}\d \vec{n}\; |j_+ \vec{n}\qket \overline{|j_- \vec n\qket}\, \qbra j_+ \vec n| \overline{\qbra j_- \vec n |}\;,& (\bi < 1)
    \end{array}
    \right.
\]
which is defined in terms of the representation dimension $\dj{J} = (2j_+ +1)(2j_- +1)$ 
and using $\su(2)$ coherent states%
\footnotetext{$\su(2)$ coherent states are defined as
$ |j\vec{n}\qket = n|jj\qket $
where $|jj\qket$ is the highest weight state in the irreducible representation $j$ of $\su(2)$ and $n = exp(\vec{n}\cdot\vec{\sigma})\in \SU(2)$ in terms of generators $\vec{\sigma}$.
}
 $|j_\pm \vec{n}\qket$ \cite{Livine:2007bq} which are summed by integration over $\mathcal{S}^2_{N}$, the unit 2-sphere in the 3-dimensional hypersurface perpendicular to the 4-vector $N$, in particular $N_0 = (1,0,0,0)$.


Since both the simplicity operator   $\Sbb_{J,N}$ and the gauge averaging are already implemented in the edge amplitude \eqref{edge-amplitude}, the vertex amplitude for this model is simply
  \begin{equation}
    \label{vertex-amplitude}
    \sfo_v(g_v)  = \Vbb_{\bbg}(g_v) =   \prod_{(\vb_1,\vh),(\vb_2,\vh)\in\E_{\bbg} }
    \delta(g_{v\vb_1\vh},g_{v\vb_2\vh})\;,
  \end{equation}
  where $\bbg=(\V_{\bbg},\E_{\bbg})$ is the bisected boundary graph associated to the spin-foam atom in $\sfr$ containing $v$. 
Packing everything together, one obtains a particular gravitational model
\[  \label{gravitational-sf}
Z(\sfr) = \int [\d g]\prod_{\vb\in\Vbar} \sfo_{\vb}(g_{\vb}) \prod_{v\in\V}\sfo_v(g_v)\;.
\]

The relevance of these amplitudes becomes clear when attaching a 4-dimensional re\-ference frame to each vertex $v$, $\vb$ and $\vh$ in each spin-foam atom $\sfa$ in the molecule $\sfr$.
Then, each holonomy $h_{v\vb}$ has the meaning of a parallel transport from $v$ to $\vb$, while the group elements $g_{v\vb\vh}$ are the parallel transports from $\vb$ to $\vh$ on the boundary of each spin-foam atom, hence the $v$-dependence. 
The ordered product of elements $h_{v\vb}$ arising in faces containing $\vh$ constitutes the holonomy representation of the curvature tensor \eqref{curvature-holonomy}.  
Because of the nontrivial edge amplitude \eqref{edge-amplitude} applying the simplicity constraints (via the operator $\Sbb_{J,N}$)  it is obvious that curvature is not restricted to be flat in this model such that geometric degrees of freedom are propagating.

\

Note that there is a straightforward extension of gravitational SF models to arbitrary spin-foam molecules $\sfr\in\sfrs$.
The restriction to simplicial molecules is only needed to obtain the discrete version \eqref{linear-simplicity} of the simplicity constraints on each simplicial atom corresponding to a 4-simplex. 
The resulting amplitudes \eqref{edge-amplitude}, however, are insensitive to this specification of combinatorics since they factorize over boundary edges. 
Therefore, they are defined for any spin-foam molecule $\sfr\in\sfrs$. 
This is indeed the rational of the KKL-extension for SF models as introduced in \cite{\KKL}.

There is however a caveat: the meaning of the simplicity constraints \eqref{linear-simplicity} is tightly related to the geometry of the 4-simplex as determined by area 2-forms on its boundary. 
When generalizing to other polytopes such a geometric interpretation via areas is nontrivial and a modification of the simplicity constraint, depending on the combinatorial structure of each specific polytope, is expected. 
It is not obvious at all whether simplicial edge amplitudes of the type  \eqref{edge-amplitude} applied to non-simplicial spin-foam molecules encode the information of simplicity constraint in the SF state sum in an appropriate way.

\

In this \sec{lgt-gravity} I have introduced SF models in three steps: first I have  reviewed their derivation as discrete path integrals for a BF-type action of GR. Then I have discussed in detail the origin and relevance of their molecular structure and discussed various equivalent formulations in alternative variables. Finally, I have sketched how simplicity constraints lead to gravitational amplitudes and presented their form with an explicit example.




%


\section[Generating functions for spin foams]
{Generating functions for spin-foam molecules and amplitudes}
\label{sec:gft}

The crucial issue left open in SF models as such is their triangulation dependence, or more generally, their dependence on the underlying spin-foam molecule. 
While in topological theories, and in particular in BF theory which is the basis for SF models, the partition function depends by definition only on the underlying topology, this is no longer true for gravitational models.
Since one does not observe any effect of a preferred discrete spacetime structure in experiment, a reasonable requirement for a theory of quantum gravity is the smoothness and diffeomorphism invariance of GR in some limit or regime, coined the \emph{continuum limit}.
Thus arises the question how SF models can account for such a continuum limit.

In principle there are two possibilities:  
(a) one can show complex independence also for gravitational models, either directly in an approximative sense 
 (\ie that a finite (set of) measurements is sufficient to determine observables for a given degree of accuracy) 
 or using some notion of refinement of the complex-dependent amplitudes \cite{Dittrich:2012ba,Dittrich:2014ui}; 
(b) the dependence on complexes is removed by average over a class of them in terms of some summing description. 


In this thesis I focus on the second alternative. 
This means, according to the theory explication \ref{sf-theory} in \sec{qg-theories}, that I understand SF models in the summing description \eqref{sf-path-integral-sum}:
a \emph{spin foam (SF) model} is a model of a quantum theory defined by a molecule independent partition function
\begin{equation}\label{eq:fullsf}
  Z = \sum_{\bs\sfr=\emptyset} w_\sfr\, Z(\sfr)\;,
\end{equation}
which is a sum of molecule dependent partition functions $Z(\sfr)$ \eqref{sf-sum} over a subclass of closed spin-foam molecules $\sfr\in\sfrs$ with weights $w_\sfr$.
As mentioned above (\sec{conceptual-relations}), GFT provides a complete prescription for both $Z(\sfr)$ and $w_\sfr$. 

In this section I will review the way GFT provides such a complete SF model description for various subclasses of $\sfr$.
First, I present the usual GFT set-up as common in the literature (\sec{gftmain}). 
In \sec{generating-sf}, I recast the formalism using the combinatorial language of atoms and molecules.
This will clarify how the graphs supporting LQG states occur, as well as the complexes supporting spin-foam amplitudes.

Then (\sec{simplicial-gft}) I will present the class of GFT  models that are standard in the literature. 
These are based on a single field and generate series catalogued by a specific subset of the unlabelled spin-foam molecules $\sfrs$.  
Via the interpretation given in \sec{std-complexes}, these are associated to $\copies$-dimensional simplicial structures.
Then, one can attempt to directly generate (larger subsets of) $\sfrs$. In this context, the first generalization is effected simply by broadening the type of interaction terms in the theory while keeping a single field. Such models are already common in the GFT  literature. \cite{Bonzom:2012bg,\gftrenorm,\COR}.

To generate all molecules in $\sfrs$ in the most straightforward fashion, the GFT field space has to be infinitely extended. 
I will make this explicit in \sec{multi-field}, before I will present in the next \sec{dw-gft} a much more elegant and efficient class of GFTs which reproduce the same results in terms of a single group field.


\subsection{Group field theory}
\label{sec:gftmain}

The general definitions and components of GFTs, as one finds them in the literature \cite{\gftFO,Krajewski:2012wm,Baratin:2012ge,Oriti:2009ur,Oriti:2014wf}, are the following. 
  A \emph{group field} is a function of $\copies\in\N$ copies of a Lie group $G$
  \begin{equation}
    \label{eq:field}
    \phi:G^{\times\copies}\lora\R\;.
  \end{equation}
%
A \emph{group field theory} is a quantum field theory for a group field, defined by a partition function 
\begin{equation}
 \label{eq:partition}
Z_{\gft} = \int \Dcal\phi\; e^{-S[\phi]}\;,
\end{equation}
where $\Dcal\phi$ denotes a (formal) measure on the space of group fields, while the action functional takes the form
\begin{equation}
   S[\phi] = \frac12\int [\d g]\; \phi(g_1)\;\Kbb(g_1, g_2)\;\phi(g_2) + \sum_{i\in I}\lambda_i \int [\d g]\;\Vbb_i\big(\{g_j\}_{J_i}\big)\;\prod_{j\in J_i}\phi(g_{j})\;.
 \label{eq:action}
\end{equation}
$\Kbb$ is the kinetic kernel, $\Vbb_i$ are interaction kernels of a certain combinatorially non-local type,  while $I$ and $J_i$ are finite sets indexing the interactions and the number of fields in the $i$th interaction, respectively.  Meanwhile, $[\d g]$ represents the appropriate number of copies of the measure on $G$ and $\{\lambda_i\}_{I}$ is the set of coupling constants.\footnote{There is an analogous set of actions for complex group fields and of course, one can define models involving several such fields.}  

  The kinetic kernel is a real function with domain $G^{\times 2n}$ that (in some model dependent manner) pairs arguments according to $(g_{1a},g_{2a})$ with $a \in \{1,\dots, \copies\}$:
  \begin{equation}
    \Kbb(g_1, g_2)  = \Kbb(g_{11},g_{21};\dots;g_{1\copies},g_{2\copies})
  \end{equation}

The {\it combinatorial non-locality} satisfied by the GFT interaction kernels is a property effected through  pairwise convolution of the field arguments.  It is the main peculiarity of GFTs in contrast to local quantum field theories on spacetime. In more detail, the GFT  interaction kernels do not impose coincidence of the points,  in the group space $G^{\times\copies}$,  at which the interaction fields are evaluated.  Rather, the totality of field arguments from the smaller group space $G$ occurring in a given action term (that is $\copies\times|J|$ for an interaction term with $|J|$ group fields) is partitioned into pairs and the kernels convolve such pairs, \ie
  \begin{equation}
  \Vbb\big(\{g_j\}_J\big)  = \Vbb\big( \{g_{ja}g_{kb}^{-1}\}\big)
  \end{equation}
  where $j,k\in J$, $a,b\in\{1,\dots,\copies\}$ and $(ja, kb)$ is an element of the pairwise partition of the set $J\times\{1,\dots,\copies\}$. 
  The specific combinatorial pattern of such pairings determines the combinatorial structure of the Feynman diagrams of the theory.  It will be one of the main foci in later discussions, 
both in the standard GFT  models and, later on, in the generalized class of models.

Besides this combinatorial peculiarity, the structure of  GFT is similar to ordinary QFT. 
  \emph{Quantum observables}, $O[\phi]$,  are functionals of the group field. 
  In particular, the kinetic and interaction terms are quantum observables. Due to their functional form, they motivate interest in a subset of polynomial functionals of the field:
a \emph{\trace\ observable} is a polynomial functional of the group field that satisfies combinatorial non-locality (since all group elements are traced over pairwise).  Thus, they have the generic form:
\begin{equation}
  \label{eq:obsdef}
  O[\phi] \equiv \int [dg]\; \bbB\big(\{g_j\}_{J}\big)\; \prod_{j\in J} \phi(g_j)\;,\quad\quad\textrm{where}
  \quad\quad \bbB\big(\{g_{j}\}_{J}\big) = \bbB\big(\{g_{ja}g_{kb}^{-1}\}\big)
\end{equation}
and $(ja,kb)$ is an element of the pairwise partition of the set $J\times \{1,\dots,\copies\}$.

  One can estimate expectation values of quantum observables using perturbative techniques.  
  In this way, the evaluation of an observable $O[\phi]$, expanded with respect to the coupling constants $\{\lambda_i\}_I$, leads to a series of Gaussian integrals evaluated through Wick contraction.  
  The patterns of contractions are catalogued by Feynman diagrams $\fdiagram$
\begin{eqnarray}
 &&\qbra O\qket_{\gft} = \frac{1}{Z_{\gft}}\int \Dcal\phi\; O[\phi]\;e^{-S[\phi]} \nonumber \\
  &&= \frac{1}{Z_{\gft}} \int \Dcal\phi\;O[\phi] \sum_{\{c_i\}_I} \prod_{i\in I}\frac{\lambda_i^{c_i}}{c_i!}\Big[\int [\d g]\Vbb_i (\{g_j\}
  \prod_{j\in J_i}\phi(g_{j})\Big]^{c_i} \e^{-\frac12\int[\d g]\phi(g_1)\Kbb(g_1,g_2)\phi(g_2)}\nonumber\\
&&= \sum_{\fdiagram} \frac{1}{\sym(\fdiagram)} A(\fdiagram;\{\lambda_i\}_I)\;,
\label{eq:expand}
\end{eqnarray}
where $\sym(\fdiagram)$ are the combinatorial factors related to the automorphism group of the Feynman diagram $\fdiagram$  and $A(\fdiagram;\{\lambda_i\}_I)$ is the weight of $\fdiagram$ in the series. The Feynman amplitudes $A(\fdiagram)$ are constructed by convolving (in group space) propagators $\Pbb = \Kbb^{-1}$ and interaction kernels. 

  The stranded-diagram representation of the Feynman diagrams $\fdiagram$ is immediate. With reference to remark \ref{rem:stranded},  one associates a coil $\coil$ with $\copies$ vertices to each occurrence of the field $\phi(g)$ in the integrand.
  In an interaction term, the fields represent a set of coils $\coils$, while the combinatorial non-locality property of the interaction kernel encodes the set of reroutings $\reroutings$.  Thus, each interaction term represents a stranded atom $\sta = (\coils, \reroutings)$. 
The kinetic term, through its involvement in the Wick contraction, is responsible for the bonding of these stranded atoms. Then, the perturbative expansion is quite clearly catalogued by stranded molecules. 
Through the bijection outlined in remark \ref{rem:stranded} the stranded diagrams map to spin-foam atoms and molecules. 
An explicit formulation of GFT in terms of the latter from scratch follows in a moment.

\

While one can define a GFT for any of the gravitational SF models described in \sec{gravitational}, an alternative interpretation of GFT as a candidate for quantum gravity is to understand the theory as providing models of quantum or random geometry.  
The Feynman diagrams $\fdiagram$ generated by a GFT, being stranded molecules, can already be understood as combinatorial complexes.
Even though these are two-dimensional, they can be extended to $\std$ dimensions (although this enhancement is a subtle issue, see \sec{std-complexes}).
  
Keeping to $\std$-dimensional language, the group fields correspond to $(\std-1)$-dimensional building blocks of $(\std-1)$-dimensional topological structures, the \trace\ observables. 
In a similar manner, the interaction terms in the action correspond to the $\std$-dimensional building blocks for $\std$-dimensional topological structures cataloguing the terms of the perturbative expansions. 

  Then, the estimation of observables $\qbra O_{1}\dots O_{N}\qket$ via perturbative expansion yields a sum over $\std$-dimensional complexes whose boundaries are precisely the $N$ complexes of $\std-1$ dimensions  encoded by observables. In other words, one is calculating the correlation of the $N$ $(\std-1)$-dimensional complexes. 

The intention of both the data contained in the group $G$ and the kernels (boundary $\bbB$, kinetic $\Kbb$ and interaction $\Vbb$) is to transform all these combinatorial-topological statements above into quantum geometrical ones.  More precisely, using results from LQG and SF  models they may be interpreted as one of the following: the discrete gravitational connection, the discrete fluxes of the conjugate triad or the eigenvalues of fundamental quantum geometric operators like areas and volumes. 


\subsection{Generating spin-foam molecules and amplitudes}
\label{sec:generating-sf}

According to the equivalence to stranded molecules (\ref{rem:stranded}) one can recast this GFT formalism in terms of spin-foam molecules, detailed in \sec{molecules}:  
  \begin{enumerate}
    \item A set of group fields is indexed by a set of patches
      \begin{equation}
	\Phi = \{\phi_{\bp}\}_{\bps}\;,\quad\quad\textrm{where}\quad\quad \phi_{\bp}:G^{\times |\E_{\bp}|} \lora \R\;,
      \end{equation}
      and $\bp = (\{\vb\}\cup\Vhat_{\bp},\E_{\bp})$. 
\item  The set of \trace\ observables is indexed by the set of bisected boundary graphs
  \begin{equation}
    \Ocal = \{O_{\bbg}\}_{\bbgs}\;,\quad\quad\textrm{where}\quad\quad O_{\bbg}[\Phi] = \int [\d g]\; \bbB_{\bbg}\big(\{g_{\vb}\}_{\Vbar}\big)\prod_{\vb\in\Vbar} \phi_{\bp_{\vb}}(g_{\vb})
  \end{equation}
  and $\bbg = (\Vbar\cup\Vhat, \E_{\bbg})$.  The patches of $\bbg$ are in correspondence with vertices of $\Vbar$ and one has that $g_{\vb} = \{g_{\vb\vh}: (\vb\vh)\in\E_{\bbg}\}$. Combinatorial non-locality is realized using the bisecting vertices $\Vhat$. Each such vertex has a pair of incident edges and thus they encode a pairwise partition of the data set $\{g_{\vb}\}_{\Vbar}$. Conversely, a pairwise partition of this data set determines a graph $\bbg$.   
  Thus, the graphs in $\bbgs$ catalogue the combinatorially non-local configurations.

\item[ ] Likewise, the set of vertex interactions is indexed by $\bbgs$
  \begin{equation}
    \lambda_{\bbg}\int [\d g]\; \Vbb_{\bbg}\big(\{g_{\vb}\}_{\Vbar}\big)\prod_{\vb\in\Vbar} \phi_{\bp_{\vb}}(g_{\vb})\;.
  \end{equation}
  As a result of the bijection in proposition \ref{prop:atoms-graphs}, the interaction terms can be interpreted as generating spin-foam atoms $\sfa = \bulk(\bbg)$.

\item[ ]
  The kinetic term, through its role in the Wick contractions occurring in later perturbative expansions,  is responsible for the bonding of patches compatible according to the compatibility condition
$\bp\equiv\bp_1\cong\bp_2$ of definition \ref{def:bonding}.
Thus,
  \begin{equation}
  \label{eq:kinetic}
    \frac12\int[\d g]\; \phi_{\bp}(g_{\vb_1})\;\Kbb_{\bp}(g_{\vb_1}, g_{\vb_2})\;\phi_{\bp}(g_{\vb_2})\;,\quad\quad\textrm{where}\quad \Kbb_{\bp}(g_{\vb_1},g_{\vb_2}) = \Kbb(\{g_{\vb_1\vh}, g_{\vb_2\vh} \})\;
  \end{equation}
  is a function of group elements for each $(\vb_i\vh)\in\E_{\bp}$.
\item
  Then, the partition sum
  \begin{equation}
    Z = \int \Dcal\Phi \;e^{-S[\Phi]}
  \end{equation}
  defines a generic model specified by an action
\end{enumerate}
  \begin{equation}
  \boxd{
    S[\Phi] = \frac12\int[\d g]\; \phi_{\bp}(g_{\vb_1})\Kbb_{\bp}(g_{\vb_1}, g_{\vb_2})\phi_{\bp}(g_{\vb_2}) + 
    \sum_{\bbg\in\bbgs}\lambda_{\bbg}\int [\d g] \Vbb_{\bbg}\big(\{g_{\vb}\}_{\Vbar}\big)\prod_{\vb\in\Vbar} \phi_{\bp_{\vb}}(g_{\vb})\;.
    }
  \end{equation}
\begin{enumerate}  
\item[4.]
Sums and products of \trace\ observables can be estimated perturbatively, generating series of the type:
\begin{eqnarray}
  \qbra O_{\bbg_1}\dots O_{\bbg_l}\qket &=& \frac{1}{Z}\int \Dcal\Phi\; O_{\bbg_1}[\Phi]\dots O_{\bbg_l}[\Phi]\;e^{-S[\Phi]} \nonumber \\
  &=& \sum_{\substack{\sfr\in\sfrs\\[0.1cm] \bs\sfr = \sqcup_{i = 1}^{l}\bbg_{i}}} \frac{1}{\sym(\sfr)} A(\sfr;\{\lambda_{\bbg}\}_{\bbgs})\;.
\end{eqnarray}
Thus, the Feynman diagrams generated by GFTs are actually better characterized as spin-foam molecules. 
\end{enumerate}
Using the above index, one can catalogue the generalized classes of GFT  models that make contact with the set of spin-foam molecules $\sfrs$.

\

  It is worth noting some advantages 
  of such a generalized concept of  GFT  models: 
First, as one can see above, there is no technical obstacle whatsoever, within the GFT  formalism,  to passing from a single-field GFT  to a multi-field GFT  (indexed by some set of patches) and new interaction terms (indexed by some set of bisected boundary graphs).  Such choices generate broader classes of spin-foam molecules, as one might wish from an LQG  perspective. 
  
Given the facility with which such generalized GFTs are defined, a real issue is rather to pinpoint some criterion, for selecting one model over another. Other important issues centre on settling  \textit{i}) whether or not one is able to control analytically or numerically the dynamics of such generalized GFTs  and  \textit{ii}) whether or not such control is improved by one choice of combinatorics over another. Indeed, these issues should also be posed from the SF perspective. 

A common choice in the SF and GFT  literature is to restrict to spin-foam atoms and molecules with a $\std$-dimensional simplicial interpretation.  This choice could be motivated as being more  \lq fundamental\rq, in the sense that one can triangulate more general complexes but not vice versa, and as being simpler than other alternatives.

  Moreover, generalized GFTs already exist in the literature. Indeed, so-called \emph{invariant tensor models}, which are in essence single-field GFTs with a specific subset of generalized interactions \cite{Bonzom:2012bg}, have been the setting for most studies on GFT  renormalization \cite{\gftrenorm,Carrozza:2014ee,Carrozza:2014bh,Carrozza:2014tf} and for analysis using tensor model techniques \cite{Gurau:2012hl}. 


Finally, even in models starting with simplicial interactions only, one should expect the quantum dynamics to generate new effective interactions with generalized combinatorics. In turn, these new interaction terms should then be taken into account in the renormalization flow of the simplicial models. Again, the issue is not whether such combinatorial generalizations can be considered, but how one should deal with them in the quantum dynamics of the theory.


\subsection{Simplicial group field theory}\label{sec:simplicial-gft}

It is worth presenting in more detail the special case of GFT based on
$\copies$-simplicial structures $\sfas_\rs$, thus generating molecules in $\sfrs_\rs$. 
As I have argued in \sec{std-complexes}, such structures have a simplicial interpretation. They correspond to a particularly simple choice of combinatorics for the GFT  action and represent a class of models that are by far the most used in the quantum-gravity literature. 

To generate $\std$-dimensional spacetime structures, the group field must depend on $\copies = \std$ copies of the group. 
Thus, simplicial group field theory is defined in the following way :
\begin{enumerate}
  \item The group field corresponds to the unique 
   $\std$-patch $\bp_\std$,
    \begin{equation}
      \phi\equiv\phi_{\bp_\std}:G^\std \lora \R\;.
    \end{equation}
  \item The pairing of field arguments in the interaction kernel is based upon the unique 
  simplicial $\std$-graph $\bbg_\drs$ (\ie the complete graph over $D+1$ vertices)
    \begin{equation}
      \label{simplicial-vertex}
      \Vbb_{\bbg_\drs}(g) =  \Vbb(\{g_{ij}{g_{ji}}^{-1}\})\;, \quad\quad \textrm{with}\quad\quad i<j\;.
    \end{equation}
   where $i,j = 1,\dots,D+1$ index the $D+1$ vertices $\Vbar$ and thus the patches of $\bbg_\drs$. The bisecting vertices are labelled by $(ij)$. 
    The edge joining the vertex $i$ to the vertex $(ij)$ is denoted by $ij$, while the edge joining the vertex $j$ to the vertex $(ij)$ is denoted by $ji$.    


  \item The action is therefore specified by
    \begin{equation}
      \label{eq:SimplicialAction}
S[\phi] = \frac12\int [\d g]\; \phi(g_1)\;\Kbb(g_1, g_2)\;\phi(g_2) 
+\lambda \int [\d g]\; \Vbb_{\bbg_\drs}(g) \prod_{j=1}^{D+1}\phi(g_j).
    \end{equation}
    where data indices are abbreviated to $g_j\equiv g_{\vb_j}$.
    The combinatorics of propagator and simplicial vertex kernel, both as boundary graphs and stranded diagrams,  are illustrated in \fig{simplicialdiagrams} for $\std=3$ (compare also \fig{dualtet}). 
    
\item There is a distinguished subclass of \trace\ observables $O_{\bbg}[\phi]$ indexed by  $\std$-regular loopless graphs $\bbg \in\bbgs_\drl$. This stems from the property that each graph in $\bbgs_\drl$ arises as the boundary of some spin-foam molecules in $\sfrs_\drs$, while the boundary of every spin-foam molecule in $\sfrs_\drs$ is a collection of graphs in $\bbgs_\drl$.
\item[] The perturbative expansion of the partition function (the GFT  vacuum expectation value) leads to a series catalogued by closed spin-foam molecules $\sfr\in\sfrs_\drs$, $\bs\sfr = \emptyset$. Meanwhile, the evaluation of a generic observable $O_{\bbg}[\phi]$ leads to a series catalogued by spin-foam molecules with boundary $\bbg$, \ie $\sfr\in\sfrs_\drs$ with $\bs \sfr = \bbg$. 
\end{enumerate} 

\begin{figure}
  
  \centering
  \tikzsetnextfilename{simplicialdiagrams}
  \begin{tikzpicture}[scale=1.3]
  \draw [<->] (1.4,0) -- (1.8,0);
\node [vb, label=135:1]	(a)	at (-0.5,0)		{}; 
\node [vh]		(1)	at (-.75,.5)		{};
\node [vh]		(2)	at (-.75,-.5)	{}; 
\node [vh]		(3)	at (-1,0)		{}; 
\node [vb, label=45:2]	(b)	at (0.5,0)		{}; 
\node [vh]		(4)	at (.75,.5)		{};
\node [vh]		(5)	at (.75,-.5)		{}; 
\node [vh]		(6)	at (1,0)		{}; 
\foreach \i/\j in {a/1,a/2,a/3,b/4,b/5,b/6}{
 \draw [eb] (\i) --  (\j);
  }
\path	(1) 	edge [bh] node[label=above:$\gamma$] {} (4)
	(2)	edge [bh]		 		(5)
	(3)	edge [bh, bend right=20]	(6)
	(a)	edge [bb]				(b);
  \begin{scope}[xshift=2.7cm]
  \foreach \i in {0,180}{
    \draw [cs, rotate=\i] (.6,-.3) rectangle (.4,.3);
    }
  \node [vs]			(11)	at (-.5,.2)	{};
  \node [vs, label=180:1]	(12)	at (-.5,0)	{};
  \node [vs]			(13)	at (-.5,-.2)	{};    
  \node [vs]			(21)	at (.5,.2)	{};
  \node [vs, label=0:2]	(22)	at (.5,0)	{};
  \node [vs]			(23)	at (.5,-.2)	{};  
  \path
  \foreach \i in {1,2,3}{
    (1\i) edge [bh]  (2\i)
    };
  \end{scope}

  \begin{scope}[xshift=5.5cm]
    \draw [<->] (1.4,0) -- (1.8,0);
  
  \node [vb, label=90:3]	(3)	at (-1.1,-.2){};
  \node [vb, label=0:4]	(4)	at (0,-.8)	{};
  \node [vb, label=90:1]	(1)	at (1,0) 	{};
  \node [vb, label=-0:2]	(2)	at (.2,1)	{};
  \foreach \i/\j in {1/2,2/4,1/4}{
    \draw [eb] (\i) -- node[vh, label=0:(\i\j)] {} (\j);
    }
  \draw [eb] (3) -- node[vh, label=below:(34)] {} (4);  
  \draw [eb] (2) -- node[vh, label=above:(23)] {} (3);
  \draw [eb] (1) -- node[vh] {} (3);
  \end{scope}  
  \begin{scope}[xshift=8.8cm] 
  \foreach \i in {0,90,180,270}{
    \draw [cs, rotate=\i] (1.1,-.3) rectangle (.9,.3);
    }
  \node [vs]			(12)	at (1,.2)	{};
  \node [vs, label=0:1]	(13)	at (1,0)	{};
  \node [vs]			(14)	at (1,-.2)	{};    
  \node [vs]			(21)	at (.2,1)	{};
  \node [vs, label=90:2]	(24)	at (0,1)	{};
  \node [vs]			(23)	at (-.2,1)	{};    
  \node [vs]		 	(32)	at (-1,.2)	{};
  \node [vs, label=180:3]	(31)	at (-1,0)	{};
  \node [vs]			(34)	at (-1,-.2)	{};    
  \node [vs]			(41)	at (.2,-1)	{};
  \node [vs, label=270:4]	(42)	at (0,-1)	{};
  \node [vs]			(43)	at (-.2,-1)	{};    
  \path
  (13) edge [es]  	 (31) 
  (24) edge [es]   	 (42)  
  \foreach \i/\j in {14/41,21/12,32/23,43/34}{
  (\i) edge [es,bend right=50] node [auto] {(\j)} (\j)
  };
  \end{scope}
  \end{tikzpicture} 
  \caption{
  Equivalent representation of combinatorics of propagator and simplicial interaction for a $D=3$ GFT  in terms of bisected boundary graphs as well as stranded diagrams.}
  \label{fig:simplicialdiagrams}
  
\end{figure}
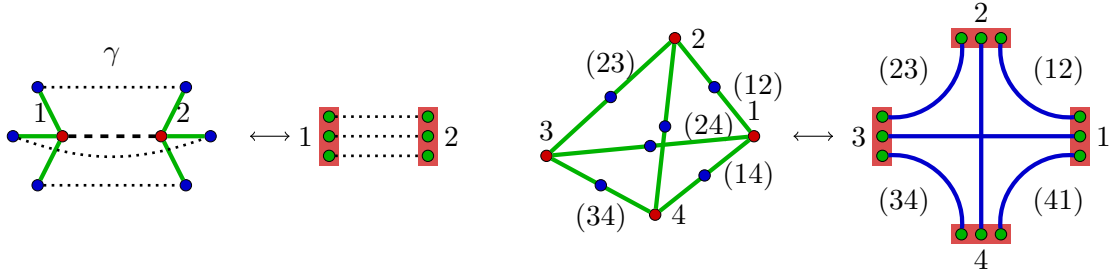

  To translate the points made in \sec{std-complexes} into GFT  language, notice that the spin-foam molecules $\sfrs_\drs$ are interpretable as locally simplicial in $\std$ dimensions.  Thus, the group field corresponds to a $(\std-1)$-simplex,  the interaction term corresponds to a $\std$-simplex, while the kinetic term, through its role in Wick contraction, corresponds to the gluing of $\std$-simplices along shared $(\std-1)$-simplices.

  Note that the GFT  action prescribes only the bonding of the spin-foam atoms along patches. 
  This corresponds to rules for identifying boundary $(\std-1)$- and $(\std-2)$-simplices.  
  The induced full $\std$-dimensional information in general implies identifications of vertices in single $\std$-complexes such that the resulting spin-foam molecule does not correspond to a simplicial pseudo-manifold in the strict sense.
  As mentioned in remark \ref{rem:coloured}, 
  a way around this limitation is to restrict to $(\std+1)$-coloured simplicial atoms with a colour for each patch in the boundary graph. 
  This is very natural from the GFT  point of view, since for any given (simplicial) GFT  model generating elements of $\sfrs_\drs$, there is an associated model generating the restricted subclass of coloured $\std$-simplicial molecules (see the review \cite{Gurau:2012hl} for details). 

It is straightforward to generalize the class of interaction terms to include those based on graphs from the set $\bbgs_\drl$. Since these are composed of unlabelled $\std$-patches, the GFT  remains dependent on a single group field, with the general action
  \begin{equation}
   \label{eq:SimplicialActionGeneralized}
  \boxd{
S[\phi] = \frac12\int [\d g]\; \phi(g_1)\;\Kbb(g_1, g_2)\;\phi(g_2) 
+\sum_{\bbg\in\bbgs_\drl}\lambda_{\bbg} \int [\d g]\; \Vbb_{\bbg}(g) \prod_{\vb\in\Vbar}\phi(g_{\vb})\;
  }
  \end{equation}
  where $\bbg = (\Vbar\cup\Vhat, \E_{\bbg})$.
Understanding the group field $\phi$ once more as a $(\std-1)$-simplex, the spin-foam atoms could still be given the interpretation of encoding $\std$-dimensional building blocks with locally simplicial $(\std-1)$-dimensional boundaries. All spin-foam molecules generated by this GFT  have boundaries in $\bbgs_\drl$.  

  Thus, it is clear that the generalized simplicial GFT \eqref{eq:SimplicialActionGeneralized} is insufficient for the purpose of defining dynamics for all LQG  quantum states with support in the larger space $\bbgs$. 
  

\subsection{Multi-field group field theory}
\label{sec:multi-field}

An obvious strategy for generating series catalogued by (larger subsets of) $\sfrs$ is simply to increase the number of field species entering the model.  Such a scenario was already anticipated at the outset of the GFT approach to SF models \cite{Reisenberger:2000fj,Reisenberger:2001hd}. 
However, from a field-theoretic viewpoint, it is a rather unattractive strategy, since the more one wishes to probe quantum states on arbitrary boundary graphs in $\bbgs$, the larger the number of field species and interaction terms required.  
Thus, the resulting formalism is not easily controlled using QFT methods.  

Having said that, with appropriate kinetic and interaction kernels, a multi-field GFT assigns amplitudes to these broader classes of spin-foam molecules in the same manner as they are defined in the generalized constructions for SF models \cite{\KKL,\KLP}.
As a result, these GFT  models are at the same level of formality.
Multi-field GFTs are illustrated here simply for demonstration of the absence of any  impediment  {\it in principle}  to having such a GFT  formulation for the quantum dynamics of all LQG  states.

Explicitly, a multi-field group field theory is devised in the following manner:
\begin{enumerate}
  \item A subset of patches $\bps_{\sub}\In \bps$ defines a class of group fields 
    \begin{equation}
      \label{eq:multifield}
      \Phi_{\sub} = \{\phi_{\bp}\}_{\bps_{\sub}}  .
    \end{equation}
  \item
    The bisected boundary graphs $\bbgs_{\sub} = \sigma(\bps_{\sub})\In \bbgs$ generated by $\bps_{\sub}$ induce a distinguished class of \trace\ observables 
    \begin{equation}
      \label{eq:multiobservable}
      \Ocal_{\sub}  =\{O_{\bbg}\}_{\bbgs_{\sub}} \In \Ocal 
    \end{equation}
    which, in particular,  can be utilized as interaction terms in the action.

  \item
A class of action functionals is then specified by
\begin{equation}
  \label{eq:multiaction}
  S[\Phi_{\sub}] 
  = \sum_{\bp\in \bps_{\sub}}  \frac1 2 \int [\d g]\; \phi_{\bp}(g_1) \Kbb_{\bp}(g_1,g_2) \phi_{\bp}(g_2) 
  + \sum_{\bbg\in\bbgs_{\sub}}  \lambda_{\bbg} \int [\d g]\;\Vbb_{\bbg}(\{g_{\vb}\}
  \prod_{\vb\in\Vbar}\phi_{\bp}(g_{\vb})
\end{equation}
where $\bbg = (\Vbar\cup\Vhat,\E_{\bbg})$.

  \item 
    In such a theory, the expectation value of an arbitrary product of observables takes the form
\end{enumerate}
\begin{equation}
  \label{eq:multipartition}
  \qbra O_{\bbg_1}\dots O_{\bbg_l}\qket_{\mfgft} 
  \equiv  \int \Dcal\Phi_{\sub}\; O_{\bbg_1}[\Phi]\dots O_{\bbg_l}[\Phi]\; e^{-S[\Phi_{\sub}]}  = \sum_{\substack{\sfr\in\sfrs_{\sub}\\ \bs\sfr   = \sqcup_{i = 1}^l \bbg_i}}  \frac{1}{\sym(\sfr)} A(\sfr;\{\lambda_{\bbg}\}_{\bbgs_{\sub}})\;.
\end{equation}

In this multi-field setting, the natural connection between a class of models and a particular value of spacetime dimension $\std$ is lost. 
Without doubt, it is difficult to identify precisely generalized classes of spin-foam molecules, such that the reconstruction of a $\std$-complex is always possible (and unique). In this non-simplicial setting, the restriction to \emph{coloured} structures is not available (to the best of my knowledge). Moreover, the set of gluing rules that one would need to specify at the outset grows with the generality of the boundary graphs and spin-foam atoms.

As in illustrative example, consider the  following particular multi-field GFT  model 
which has a natural 3-dimensional interpretation:
  \begin{enumerate}
    \item
      Choose $\copies$-patches with $3\leq \copies \leq \copies_{\text{max}}$ for some $\copies_{\text{max}}\in\N$:
      \begin{equation}
      \bps_{\sub} = \{\bp_\copies:3\leq \copies\leq\copies_{\text{max}}\}\;.
    \end{equation} 
    Then, the group field $\phi_{\bp_\copies}$ could be viewed as representing a 2-dimensional $\copies$-gon.  
  
  \item 
    A distinguished class of \trace\ observables is indexed by $\bbgs_{\sub} = \sigma(\bps_{\sub})$ and they may be interpreted as surfaces composed of polygons (up to the caveats discussed in \sec{std-complexes} where I have already stressed that reconstructing these surfaces is a subtle topic and extra information must be put in by hand).
    As a specific example, consider the \trace\ observable
    \begin{equation}
\label{eq:pyramid}
O_{\bbg}[\phi_{\bp_3},\phi_{\bp_4}] 
= \int [\d g]\; \bbB_{\bbg}(g)\; \phi_{\bp_4}(g_1)\,\phi_{\bp_3}(g_2)\,\phi_{\bp_3}(g_3)\,\phi_{\bp_3}(g_4)\,\phi_{\bp_3}(g_5) 
\end{equation}
where
\begin{equation}
\bbB_{\bbg}(g) = \bbB(g_{12}g_{21}^{-1},g_{23}g_{32}^{-1},g_{34}g_{43}^{-1},g_{14}g_{41}^{-1},g_{15}g_{51}^{-1},g_{25}g_{52}^{-1},g_{35}g_{53}^{-1},g_{45}g_{54}^{-1})\;.
\end{equation}
As illustrated in \fig{pyramid}, one could associate a pyramid with a square base to the graph $\bbg$. 

\item An action functional of the type given in \eqref{eq:multiaction}, along with the partition function \eqref{eq:multipartition} generate spin-foam molecules that may be interpreted as 3-dimensional objects composed of such building blocks.  
\end{enumerate}

\begin{figure}
  
  \centering
  \tikzsetnextfilename{pyramid}
  \begin{tikzpicture}[scale=1.7]
 \draw [<->] (1.2,0) -- (1.7,0);
 \draw [<->] (-1.2,0) -- (-1.7,0);
\begin{scope}[xshift=-2.8cm]
\node [vb,label=180:1]	(1)	at (-1.25,1)	{}; 
\node [vb,label=90:2]		(2)	at (-0.5,0.5)	{}; 
\node [vb,label=90:3]		(3)	at (0.5,0.5) 	{}; 
\node [vb,label=0:4]		(4)	at (1.25,1)		{}; 
\node [vb,label=0:5]		(5)	at (0,-1)		{};
\foreach \i/\j in {1/4,1/2,2/3,3/4}{
 \draw [eb] (\i) -- node[vh,label=above:(\i\j)] {}  (\j);
  }
  \foreach \i/\j in {1/5,2/5,3/5,4/5}{
 \draw [eb] (\i) -- node[vh,label=left:(\i\j)] {}  (\j);
  }
\end{scope}

\begin{scope}
\node [c]					(v)	at (0,-.12)		{};
\node [c,label=88:1]			(1)	at (-0.56,0.32) 	{}; 
\node [c,label=185:2]			(2)	at (-0.17,-0.07)	{}; 
\node [c,label=0:3]			(3)	at (0.6,0) 		{}; 
\node [c,label=90:4]			(4)	at (.14,.37)	{}; 
\node [c,label=-10:5]			(5)	at (0,-.6)		{};
\node [c,label=left:(12)]		(12)	at (-.72,.32)	{};
\node [c,label=above:(23)]		(23)	at (.44,.09)	{};
\node [c,label=0:(34)]		(34)	at (.57,.52)	{};
\node [c,label=5:(14)]		(14)	at (-.3,.61)		{};
\node [c,label=below:(15)]		(15)	at (-1,-.44)		{};
\node [c,label=below:(25)]		(25)	at (-.28,-.97)	{};
\node [c,label=below:(35)]		(35)	at (1.02,-.76)	{};
\node [c,label=-10:(45)]		(45)	at (.28,-.23)	{};
\foreach \i/\j in {1/2,1/4,2/3,3/4,1/5,2/5,3/5,4/5}{
 \path	[f] 	(\i) -- (\i\j) -- (\j) -- (v) -- cycle;
 }
 \foreach \i in {1,2,3,4,5}{
  \draw [e] (\i) -- (v);
  }
\foreach \i/\j in {1/2,1/4,3/4,1/5,2/5,3/5,4/5,2/3}{
 \draw 	[eh]	(v)		-- 	(\i\j);
 \draw	[eb] 	(\i) node[vb] {} -- (\i\j) node[vh] {};
 \draw 	[eb] 	(\j) node[vb] {} -- (\i\j) node[vh] {};
 }
\draw [e] (3) node[vb] {} -- (v) node[v] {};

\end{scope}
\begin{scope}[xshift=3cm]
\node [c]					(v)	at (0,-.12)		{};
\node [c,label=88:1]			(1)	at (-0.56,0.32) 	{}; 
\node [c,label=185:2]			(2)	at (-0.17,-0.07)	{}; 
\node [c,label=0:3]			(3)	at (0.6,0) 		{}; 
\node [c,label=90:4]			(4)	at (.14,.37)	{}; 
\node [c,label=-10:5]			(5)	at (0,-.6)		{};
\node [c,label=left:(12)]		(12)	at (-.72,.32)	{};
\node [c
				]		(23)	at (.44,.09)	{};
\node [c,label=0:(34)]		(34)	at (.57,.52)	{};
\node [c,label=5:(14)]		(14)	at (-.3,.61)		{};
\node [c,label=below:(15)]		(15)	at (-1,-.44)		{};
\node [c,label=below:(25)]		(25)	at (-.28,-.97)	{};
\node [c,]					(35)	at (1.02,-.76)	{};
\node [c
				]		(45)	at (.28,-.23)	{};
\foreach \i/\j in {1/2,1/4,2/3,3/4,1/5,2/5,3/5,4/5}{
 \path	[f] 	(\i) -- (\i\j) -- (\j) -- (v) -- cycle;
 }
 \foreach \i in {1,2,3,4,5}{
  \draw [e] (\i) -- (v);
  }
\foreach \i/\j in {1/2,1/4,2/3,3/4,1/5,2/5,3/5,4/5}{
 \draw	[eh]	(v)		-- 	(\i\j);
 \draw	[eb] 	(\i) node[vb] {} -- (\i\j) node[vh] {};
 \draw 	[eb] 	(\j) node[vb] {} -- (\i\j) node[vh] {};
 }
\draw [e] (3) node[vb] {} -- (v) node[v] {};

\draw (-0.58,-0.14)-- (0,1.37);
\draw (-0.58,-0.14)-- (1.15,-0.32);
\draw (-0.58,-0.14)-- (-1.45,-0.74);
\draw (0.89,-1.2)-- (1.15,-0.32);
\draw (1.15,-0.32)-- (0,1.37);
\draw(0,1.37)--  (0.89,-1.2);
\draw (-1.45,-0.74)-- (0.89,-1.2);
\draw(0.89,-1.2)-- (1.15,-0.32);
\draw (-1.45,-0.74)-- (0,1.37);
\draw(0,1.37)--  (0.89,-1.2);
\draw (0.89,-1.2)-- (-1.45,-0.74);
\end{scope}

\end{tikzpicture} 
  \caption{
The pyramid graph $\bbg$, the associated spin-foam atom $\sfa = \bulk(\bbg)$ and finally the pyramid 3-cell constructed around it.}
  \label{fig:pyramid}
  
\end{figure}


  GFT  models based upon proper 
  subsets $\Phi_{\sub}\In  \Phi$ probe only subsets $\bbgs_{\sub}\In \bbgs$ and $\sfrs_{\sub}\In \sfrs$ and thus, only subsets of the LQG  states and SF dynamics. One could consider examining a model based on all of $\Phi$ and all of $\Ocal$. In this manner, one would probe all of $\bbgs$ and $\sfrs$,  as one might expect in the traditional LQG  context. 
 The resulting construction, however, is likely to remain at a formal level. In fact, the multi-field GFT  realization depends upon infinitely many fields and, in order to have non-trivial dynamics for each field, infinitely many interaction terms. This likely renders any field theoretic analysis rather impracticable.

Having said that, with appropriate choices for kinetic and interaction kernels, the multi-field GFT  based upon $\Phi$ and $\Ocal$ generates series probing all the spin-foam molecules of $\sfrs$,  weighted by amplitudes coinciding with the KKL extension of the SF models and propagating LQG  states on graphs in $\bbgs$.

\

In this \sec{gft} I have given a detailed introduction to GFT and explicated the structure of the theory's Feynman diagrams. According to their equivalence to spin-foam molecules I have recast the definition of GFT in this language associating fields with boundary patches and interaction vertices with spin-foam atoms. 
After illustrating this structure for the case of simplicial GFT, I have shown how this formulation gives rise to a generalization to a multi-field GFT corresponding to the KKL extension of SF models in a straightforward way.


\section{Dually weighted group field theories}\label{sec:dw-gft}

 To the extent that GFTs are currently analytically tractable, one is motivated to repackage the structures generated in the above multi-field GFT and devise a class of GFT  models that encode the quantum dynamics of arbitrary LQG  states, while remaining more practically useful. This means managing to encode arbitrary boundary graphs using a single or at least a (small) finite number of GFT  fields and interactions. 
 The key to achieving this result, which is the topic of this section, lies in the use of labelled structures and the dynamical implementation of reduction moves in terms of a dual weighting.


Thus, I start in \sec{labelled-simplicial-gft} defining the labelled version of GFT with a focus on the simplest case, \ie simplicial GFT. 
Then I introduce in \sec{dual-weighting} the dual weights and discuss in detail the mechanism how virtual structures are effectively projected out in the dynamics.
Finally, in \sec{dw-gravitational}, I discuss how a combination of dual-weighting and gravitational operators defines a theory which yield sum-over-spin-foams dynamics for gravitational models such as the model discussed in \sec{gravitational}.


\subsection{Labelled simplicial group field theory}
\label{sec:labelled-simplicial-gft}

Utilizing the labelled simplicial structures $\bpst_{\copies}$, $\bbgst_\rs$ and $\sfast_\rs$ to generate spin-foam molecules $\sfrst_\rs$ (defined in \sec{regular-loopless}) is a simple generalization of the simplicial model presented in section \ref{sec:simplicial-gft};
the labelling distinguishes between cells considered as real and cells considered as virtual in each molecule $\sfrt\in\sfrst_\rs$.

\begin{enumerate}
  \item The set of group fields is indexed by the set of labelled $\copies$-patches:
    \begin{equation}
      \Phit = \{\phit_{\bpt}\}_{\bpst_{\copies}}\;,\quad\quad\textrm{where}\quad\quad \phit_{\bpt}:G^{\times |\E_{\bpt}|} \lora \R\;,
    \end{equation}
    Note that this is a finite set of fields: $|\Phit| = |\bpst_{\copies}| = 2^\copies$. Also, $|\E_{\bpt}| = \copies$. 
  \item  The set of \trace\ observables is indexed by the set of labelled $\copies$-regular, loopless graphs $\bbgst_\rl$:
    \begin{equation}
      \Obt = \{\obt_{\bbgt}\}_{\bbgst}\;,\quad\quad\textrm{where}\quad\quad \obt_{\bbgt}[\Phit] = \int [\d g]\; \bbBt_{\bbgt}\big(\{g_{\vb}\}_{\Vbar}\big)\prod_{\vb\in\Vbar} \phi_{\bpt}(g_{\vb})\;,
    \end{equation}
    where $\bbBt_{\bbgt}$ implicitly depends on the edge labels drawn from $\{real,virtual\}$. 

  \item [] The set of vertex interactions is indexed by labelled simplicial $\copies$-graphs $\bbgst_\rs$.  Since these are all based on the complete graph over $\copies+1$ vertices, one can utilize the vertex labelling of equation \eqref{simplicial-vertex}: 
    \begin{equation}
      \lambda_{\bbgt}\int [\d g]\; \Vbbt_{\bbgt}\big(g\big)\prod_{j = 1}^{\copies+1} \phit_{\bpt}(g_j)
    \end{equation}
    This is a finite set of interactions: $|\bbgst_\rs| = 2^{{\copies+1\choose 2}}$.
    
    Of course, the set of interaction terms can be extended to those indexed by $\bbgst_\rl = \sigma(\bpst_\copies)$, and one should probably expect them to be generated during the renormalization process.  However, the point is that the small set $\bbgst_\rs$ is rich enough to generate spin-foam molecules that could provide non-trivial correlations for all of $\bbgst_\rl$ (proposition \ref{prop:all-graphs}), and so is a well-chosen minimal model to take at the outset. 

    Again, using  the bijection in proposition \ref{prop:atoms-graphs}, the interaction terms can be interpreted as generating spin-foam atoms $\sfat = \bulk(\bbgt)$.

  \item []
    The kinetic term 
    is responsible for the bonding of patches just as in the unlabeled case \eqref{eq:kinetic}:  
    \begin{equation}
      \frac12\int[\d g]\; \phit_{\bpt}(g_1)\;\Kbbt_{\bpt}(g_1, g_2)\;\phit_{\bpt}(g_2)\;,\quad\quad\textrm{where}\quad\quad \Kbbt_{\bpt}(g_1,g_2) = \Kbbt(g_{\vb_1}, g_{\vb_2} )\;.
    \end{equation}

  \item
    Then, the class of labelled simplicial GFTs is defined via:
    \begin{equation}
      Z_{\sgft} = \int \Dcal\Phit \;e^{-S[\Phit]}
    \end{equation}
    with:
    \begin{equation}
      S[\Phit] = \frac12\sum_{\bpt\in\bpst_\copies} \int[\d g]\;\phit_{\bpt}(g_1)\;\Kbbt_{\bpt}(g_1, g_2)\;\phit_{\bpt}(g_2) + 
      \sum_{\bbgt\in\bbgst_\rs}\lambda_{\bbgt}\int [\d g]\; \Vbbt_{\bbgt}\big(g\big)\prod_{j = 1}^{\copies+1} \phit_{\bpt}(g_j)
    \end{equation}

  \item
    Trace observables can be estimated perturbatively, generating series of the type:
    \begin{eqnarray}
      \qbra \obt_{\bbgt_1}\dots \obt_{\bbgt_l}\qket_{\sgft} &=& \frac{1}{Z}\int \Dcal\Phit\; \obt_{\bbgt_1}[\Phit]\dots\obt_{\bbgt_l}[\Phit]\;e^{-S[\Phit]} \nonumber \\
      &=& \sum_{\substack{\sfrt\in\sfrst_\rs\\[0.1cm] \bs\sfrt = \sqcup_{i = 1}^{l}\bbgt_{i}}} \frac{1}{\sym(\sfrt)} A(\sfrt;\{\lambda_{\bbgt}\}_{{\bbgst}_\rs})\;,
    \end{eqnarray}
\end{enumerate}

  As pointed out in remark \ref{rem:moleculered}, not all molecules in $\sfrst_\rs$ reduce to a molecule in $\sfrs$. It is rather the subset of non-branching molecules $\sfrst_{\copies,\snb}\In \sfrst_\rs$ that possesses this property (\rem{dually}).  As a result, one needs a mechanism at the GFT  level that isolates this subset.  
  This is possible with a dual-weighting mechanism and I will explain it in the next section.


\subsection{Dual weighting}\label{sec:dual-weighting}

It turns out that employing a simple technique at the field theory level allows one to extract directly the subclass of non-branching structures $\sfrst_{\copies,\snb}\In \sfrst_\rs$. This technique, dubbed \emph{dual-weighting} in the matrix model literature \cite{DiFrancesco:1992cn,Kazakov:1996et,Kazakov:1996fq}, assigns parameterized weights to the vertices $\Vhat$ of the spin-foam atoms $\sfat\in\sfast_\rs$ and, through the bonding mechanism, of the spin-foam molecules $\sfrt\in\sfrst_\rs$.\footnote{The \emph{dual-weighting} denomination stems from the fact that in 2d these vertices are in one-to-one correspondence with the vertices of the dual complex. In that context, these parametrized weights can be interpreted as coupling parameters for dual vertices.} 
These weights can be tuned so that only virtual interior/boundary vertices in $\Vhat$ with precisely two virtual faces/one virtual face incident survive. This is precisely the condition pinpointing the configurations in $\sfrst_{\copies,\snb}$.

The ground for a dual-weighting mechanism on a simplicial GFT is an extension of the elementary data set of the field from $G$ to $G\times\Mcal$, where $\Mcal = \{0,1,\dots, M\}$.
The integer $M$ can be regarded as a free parameter of the theory.
Since these data sets are associated to edges of both patches and boundary graphs, they permit a new encoding of the edge labels $\{real, virtual\}$. The \emph{real} label is encoded as the zero element $0\in\Mcal$, while the \emph{virtual} label is encoded by the non-zero elements $m\in\Mcal \setminus\{0\}$.  
\begin{enumerate}
  \item In this way, one can repackage the $2^{\copies}$ fields $\phit_{\bpt}$ ($\bpt\in\bpst_\rs$) into a single field 
    \begin{equation}
      \phi:(G\times\Mcal)^{\copies}\lora\R
    \end{equation}
    based on the unique \emph{un}labelled $\copies$-patch $\bp_\copies\in\bps_\rs$.
    This stems from the fact that these patches have the same combinatorics, differing only in the choice of labels $\{real, virtual\}$ assigned to their edges.
  
  \item In principle, the \trace\ observables are indexed once again by labelled, $\copies$-regular, loopless graphs $\bbgt\in\bbgst_\rl$. Encoding the labelling as above, one can re-index observables by \emph{un}labelled, $\copies$-regular, loopless graphs $\bbgs_\rl$:
   \begin{equation}
      \Ocal = \{O_{\bbg}[\phi]\}_{\bbgs_\rl}\;,\quad \textrm{where}\quad\quad O_{\bbg}[\phi] = \int [\d g]\sum_{[m]} \bbB_{\bbg}\big(\{g_{\vb}; m_{\vb}\}_{\Vbar}\big)\prod_{\vb\in\Vbar} \phi(g_{\vb};m_{\vb})\;
    \end{equation}
    where the \emph{combinatorial non-locality} extends to the $\Mcal$ variables. In effect, the observable $O_{\bbg}$ incorporates all $2^{|\mathcal{\bbg}|}$ labelled observables $\obt_{\bbgt}$ with support on that graphical structure.

    However, combinatorial non-locality, in conjunction with this novel label-encoding, places a restriction on $\bbB_{\bbg}(\{g_{\vb};m_{\vb}\}_{\Vbar})$. To detail this, 
one uses the same indexing of vertices and edges as in \eqref{simplicial-vertex} with an extra label $(a)$ to number multi-edges. A simple illustration for a bisected edge between two vertices $i,j\in\Vbar$ looks like:
\begin{center}
\tikzsetnextfilename{edge}
\tikz \draw [eb]  (-2,0) node [vb,label=left:$i$] {}  --  node[auto] {$ij(a)$} 
			(0,0) node[vh,label=above:$(ij)(a)$] {} -- node[auto] {$ji(a)$}  
			(2,0) node[vb,label=right:$j$] {};
\end{center}
Then the graph $\bbg$ dictates that the boundary kernel has the form
\begin{equation}
  \bbB_{\bbg}(\{g_{\vb};m_{\vb}\}_{\Vbar}) = \bbB(\{g_{ij(a)}g_{ji(a)}^{-1}; m_{\vb}\})\;.
\end{equation}  
For labelled boundary graphs, both edges $ij(a)$ and $ji(a)$ are marked by the same label $\{real, virtual\}$.  This translates to the restriction that $\bbB_{\bbg}(\{g_{\vb};m_{\vb}\}_{\Vbar}) = 0 $ when $m_{ij(a)} = 0$ but $m_{ji(a)}\in\Mcal \setminus\{0\}$ or vice versa. 
The other way round, $\bbB_{\bbg}(\{g_{\vb};m_{\vb}\}_{\Vbar}) \neq 0$ only when both $m_{ij(a)}  = m_{ji(a)} = 0$ or $m_{ij(a)}, m_{ji(a)} \in \Mcal \setminus \{0\}$.    

The $2^{{\copies+1\choose 2}}$ interaction terms are indexed by labelled simplicial graphs $\bbgt\in\bbgst_\rs$.  As with the boundary kernels, the labelling can be re-encoded in terms of the new data set.  As a result, one can capture all the interaction terms using the unique \emph{un}labelled, $\copies$-regular, loopless graph $\bbg\in\bbgs_\rs$:
    \begin{equation}
      \lambda\int [\d g]\sum_{[m]} \Vbb_{\bbg}(g; m)\prod_{j = 1}^{\copies+1} \phi(g_j; m_j)\;
      \end{equation}
      where in analogy to \eqref{simplicial-vertex} 
     \begin{equation}
       \Vbb_{\bbg}(g;m) = \Vbb\big(\{g_{ij}g_{ji}^{-1};m_{\vb}\}\big)\;, \quad\quad \textrm{with}\quad\quad i<j\;,
     \end{equation} 
     and the markers $i,j\in\{1,\dots,\copies+1\}$ index the $\copies+1$ vertices $\Vbar\In  \bbg$ and thus the patches of $\bbg$. Meanwhile, the pair $ij$ (with $j\neq i$) indexes the edge joining the vertex $i$ to the bisecting vertex $(ij)$.
Combinatorial non-locality imposes an analogous constraint on this interaction kernel.

    Just as for labelled simplicial GFTs,  the set of interaction terms could be extended to those indexed by $\bbgs_\rl = \sigma(\bps_\rs)$, while still invoking the dual weighting mechanism. In terms of labelled structures, this means that one could isolate $\sfrst_{\copies,\lnb}\In \sfrst_\rl$. 

   \item[] The kinetic term takes the form
     \begin{equation}
       \frac12\int[\d g]\sum_{[m]}\phi(g_1;m_1)\;\Kbb(g_1,g_2;m_1,m_2)\;\phi(g_2;m_2)\;
     \end{equation}
     with kernel
   \begin{equation}
     \Kbb(g_1,g_2;m_1,m_2) = \Kbb\big(g_{\vb_1},g_{\vb_2};m_{\vb_1},m_{\vb_2}\big)\;.
     \end{equation}
    Since the kinetic term is responsible for the bonding of the patches and bonding respects labelling, then $\Kbb \neq 0$ only when both $m_{\vb_1\vh} = m_{\vb_2\vh} = 0$ or both $m_{\vb_1\vh},m_{\vb_2\vh}\in\Mcal \setminus \{0\}$.  
 \end{enumerate}

In order to define dually-weighted GFT, one has to specify the precise form of the $\Mcal$-sector of the various kernels. The $G$-sector will be left unspecified for the moment, dealing with specific cases in \sec{dw-gravitational}.
The crucial ingredient for the $\Z_M$-sector is 
a \emph{dual-weighting matrix sequence} $\{\Acal_M\}_{M>0}$, defined as a sequence of invertible matrices (where $M$ denotes the size of $\Acal_M$) which satisfy the condition
\begin{equation}
\label{dw-matrix}
\boxd{
  \lim_{M\rightarrow\infty}
  \tr\left( (\Acal_M)^k \right) = \delta_{k,2}.
  }
\end{equation}
The trace invariant information contained within an $M\times M$ matrix $\Acal_M$ can be characterized in a number of ways,   perhaps most familiarly through its $M$ eigenvalues, which arise as the roots of the characteristic equation. 
In the large-$M$ limit, it is clear therefore that one may impose an infinite number of conditions on the matrix traces.

For concreteness, consider as a specific sequence $\{\Acal_M\}$ the diagonal matrices 
\[
\left(\Acal_M\right)_{mm'} = (-1)^m M^{-1/2} \delta_{mm'}\;. 
\]
These fulfill the conditions Eq.\,\eqref{dw-matrix} since for odd $k$ 
\begin{equation}
\tr \left( (\Acal_M)^k \right) = - M^{-k/2} \us{M\ra\infty}\lora 0
\end{equation}
and for even $k$ 
\begin{equation}
\tr \left( (\Acal_M)^k \right) = M^{1-k/2} 
\end{equation}
which equals one for $k=2$ and tends to zero in the large-$M$ limit for $k>2$.
For more uses of the dual-weighting mechanism see \cite{Benedetti:2012ed} and references therein. 

 The implementation 
  of the dual-weighting mechanism places certain restrictions on the kinetic, interaction and boundary kernels.
The kinetic kernel takes the form:
\[\label{dw-kinetic}
\Kbb(g_1,g_2;m_1,m_2) = \Kbbb(g_1,g_2;m_1,m_2)\;\Dbb^{-1}(m_1,m_2)\;,
\]
where $\Kbbb$ is constant across $m_{1j}, m_{2j}\in\Mcal \setminus\{0\}$, for each $j\in\{1,\dots,\copies\}$.  In other words, $\Kbbb$ only depends on whether the edges are \emph{real} or \emph{virtual}.
 Meanwhile, $\Dbb$ factorizes across the edges:
\[\label{dw-matrix2}
   \Dbb(m_1,m_2) = \prod_{j = 1}^\copies \dbb_{m_{1j},m_{2j}}\;,
   \quad\quad\textrm{with}\quad\quad 
   \dbb = \left(
    \begin{array}{c|c}
      1&0\\ \hline
      0&\Acal_M
    \end{array}
    \right)\;.
 \end{equation}
 The condition on $\Kbbb$ means that the value it attains only depends on whether the edges are \emph{real} or \emph{virtual}. 
The zero entries in the $\dbb$-matrix are the manifestation of non-mixing of \emph{real} and \emph{virtual} edges. The dual-weighting matrix $\Acal_M$ is the truly significant player, as it will be responsible for restricting the spin-foam molecules in the large-$M$ limit. 

The interaction kernel takes the form
\[\label{dw-interaction}
  \Vbb_{\bbg}(g;m) = \Vbbb_{\bbg}(g;m)\;\Ibb(m)\;
  \quad\textrm{where}\quad 
  \Ibb(m) = \prod_{(ij)} \ibb_{m_{ij},m_{ji}} 
  \quad\textrm{and}\quad
  \ibb =  \left(
    \begin{array}{c|c}
      1&0\\ \hline
      0&\id_M
    \end{array}
    \right)\;
\end{equation}
where $\id_M$ is the $M\times M$ identity matrix and the function $\Vbbb$ only depends on whether the edges are \emph{real} or \emph{virtual}.   
Accordingly, the boundary kernels take the similar form
\begin{equation}
  \bbB_{\bbg}(\{g_{\vb};m_{\vb}\}_{\Vbar}) = \bbBb_{\bbg}(\{g_{\vb};m_{\vb}\}_{\Vbar})\;\Ibb(\{m_{\vb}\}_{\Vbar})\;
  \quad\textrm{where}\quad
  \Ibb(\{m_{\vb}\}_{\Vbar}) = \prod_{(ij)(a)} \ibb_{m_{ij(a)},m_{ji(a)}} 
\end{equation}
and $\bbBb$ only depends on whether the edges are \emph{real} or \emph{virtual}.

Of course, in true GFT  style, one could shift the dual-weighting matrix to the interaction kernels, that is, swapping $\id_M$ for $\Acal_M$ in the interaction kernel, while simultaneously swapping $\Acal_M$ for $\id_M$ in the kinetic kernel. 
Accordingly, in this realization, the boundary kernels should also contain the dual-weighting matrix to ensure the correct propagation of virtual edges. 

Now, back to the definition of the dually-weighted GFTs:
\begin{enumerate}[resume]
\item 
The class of dually weighted GFTs is defined in terms of the partition function
    \begin{equation}
      Z_{\dwgft} = \int \Dcal\phi \;e^{-S[\phi]}
    \end{equation}
based on the action
\ba\label{dw-action}
  S[\phi] &=& \frac12\int[\d g]\sum_{[m]}\phi(g_1;m_1)\;\Kbb(g_1, g_2;m_1,m_2)\;\phi(g_2;m_2) \nonumber\\
  &&+ \lambda\int [\d g]\sum_{[m]} \Vbb_{\bbg}(g;m)\prod_{j =1}^{\copies+1} \phi(g_j;m_j)
\ea

  \item 
   A \trace\ observable can be estimated perturbatively, generating series of the type:
    \begin{eqnarray}
    \label{dw-perturbative}
      \qbra O_{\bbg_1}\dots O_{\bbg_l}\qket_{\dwgft} &=& \frac{1}{Z_{\dwgft}}\int \Dcal\phi\; O_{\bbg_1}[\phi]\dots O_{\bbg_l}[\Phi]\;e^{-S[\phi]} \nonumber \\
      &=& \sum_{\substack{\sfrt\in\sfrst_\rs\\[0.1cm] {\bs}\sfrt = \sqcup_{i = 1}^{l}\bbg_{i}}} \frac{1}{\sym(\sfrt)} A(\sfrt;\lambda)\;,
    \end{eqnarray}
The series in the second line can be explicitly catalogued by labelled simplicial $\copies$-molecules $\sfrt\in\sfrst_\rs$ according to the label information encoded in the $\Z_M$ in the integrand of each amplitude.
\end{enumerate}

Now, the crucial property of these perturbative series is that only amplitudes attached to the subclass of non-branching molecules survive in the large-$M$ limit:

\fbox{
  \begin{minipage}[c][][c]{0.97\textwidth}
\begin{proposition}[{large-$M$ limit}]
\label{prop:mlimit}
In the large-$M$ limit, the observable expectation values of dually weighted GFT possess perturbative expansions in terms of simplicial $\copies$-molecules within the non-branching subclass $\sfrst_{\copies,\snb}$ (remark \ref{rem:dually})
  \begin{eqnarray}
    \lim_{M\rightarrow\infty} \qbra O_{\bbg_1}\dots O_{\bbg_l}\qket_{\dwgft} &=&  \sum_{\substack{\sfrt\in\sfrst_{\copies,\snb}\\[0.1cm] {\bs}\sfrt = \sqcup_{i = 1}^{l}\bbgt_{i}}} \frac{1}{\sym(\sfrt)} A(\sfrt;\lambda)\;.
    \end{eqnarray}
\end{proposition}
\end{minipage}
}
\begin{proof}
 According to the definition of the amplitudes $A(\sfrt;\lambda)$, with $\sfrt\in\sfrst_\rs$, in the perturbative sum \eqref{dw-perturbative}, the dual-weighting part of the amplitude factorizes across the vertices $\vh\in\Vhat$ and for each vertex it takes one of two values:
  \begin{equation} 
    \begin{array}{cl}
      1 & \textrm{for $\vh$ real.}\\[0.3cm]
   \displaystyle \tr\Big(\prod_{(\vb\vh)}\Acal_M\Big) & \textrm{for $\vh$ virtual.}
  \end{array}
\end{equation}
Due to the dual-weighting property \eqref{dw-matrix},  the second contribution vanishes in the limit $M\lora\infty$ unless there are precisely two/one edge(s) incident at this internal/boundary vertex $\vh$.  This is exactly the defining property of $\sfrst_{\copies,\snb}$ (remark \ref{rem:dually}).
\end{proof}

Proposition \ref{prop:mlimit} allows to recast the perturbative series generated by the dually weighted GFT in the large-$M$ limit as a series catalogued by more generic molecules 
according to the projection map $\Pi_{\copies,\snb}:\sfrst_{\copies,\snb}\lora\sfrs$ (remark \ref{rem:refmolreduction}). From proposition \ref{prop:limitation}, $\Pi_{\copies,\snb}$ does not cover the whole of $\sfrs$, but only a subset $\sfrs' = \Pi_{\copies,\snb}(\sfrst_{\copies,\snb})$.
Thus, the perturbative series can be rewritten as a sum over such generic spin-foam molecules with effective amplitudes
\[\label{dw-effective}
\lim_{M\rightarrow\infty} \qbra O_{\bbg_1}\dots O_{\bbg_l}\qket_{\dwgft} 
= \sum_{\substack{\sfr\in\sfrs'\In \sfrs\\[0.1cm] {\bs}\sfr = \sqcup_{i = 1}^{l}\pi_{\copies,\snb}(\bbgt_{i})}} \frac{1}{\sym^\textbf{eff}(\sfr)} A^\textbf{eff}(\sfr;\lambda)\;.
\]
In particular, \emph{every} collection of boundary graphs drawn from $\bbgs$ can be evolved within this dually weighted GFT according to proposition \ref{prop:all-graphs} together with proposition \ref{prop:surjections}.
Should one wish to include also a larger set of effective molecules from $\sfrs$, this could be easily obtained by incorporation of more interactions with support on $\bbgs_\rl = \sigma(\bps_\rs)$. 

In the quantum gravity context, 
this means that dually weighted GFT is effectively a GFT describing physical inner products and correlations of various quantum gravity states with support on arbitrary graphs, which are estimated using perturbative series catalogued by spin-foam molecules of the most general combinatorics.

\ 


To illustrate this dual-weighting mechanism further, consider again the particular \trace\ observable \eqref{eq:pyramid} with a pyramid interpretation, now from the dual-weighting viewpoint.
Since the pyramid is 3-dimensional, the dually weighted GFT has to be defined in terms of the field $\phi:(G\times\Mcal)^3\lora\R$.
The pyramid observable 
has 
the following realization composed out of six fields according to the graph $\bbg$ which is illustrated in \fig{dw3dgraph} (cf. the corresponding atom, \fig{atommoves}):
\begin{equation}
\label{eq:pyramidAgain}
O_{\bbg}[\phi] = \int [\d g]\sum_{[m]} \bbB_{\bbg}(g;m)\,
\phi(g_1;m_1)\,\phi(g_2;m_2)\,\phi(g_3;m_3)\,\phi(g_4;m_4)\,\phi(g_5;m_5)\;\phi(g_6;m_6)\,.
\end{equation}
The explicit non-local convolution structure according to the graph \fig{dw3dgraph} is
\[
  \bbB_{\bbg}(g;m) = \bbB(g_{12}g_{21}^{-1},g_{23}g_{32}^{-1},g_{34}g_{43}^{-1},g_{14}g_{41}^{-1},g_{15}g_{51}^{-1},g_{25}g_{52}^{-1},g_{36}g_{63}^{-1},g_{46}g_{64}^{-1},g_{56}g_{65}^{-1}; m_{ij}).
\]
This packages together a number of observables, depending on the $\Mcal$-labelling.
The observable of interest is precisely the configuration where the labels $m_{56}$ and $m_{65}$ are non-zero (indicating a virtual edge), while the rest are zero (indicating real edges).  Upon reduction of this virtual edge, one obtains a graph in $\sfrs$ with one 4-valent patch and four 3-valent patches just as in \eqref{eq:pyramid}. 
In this way, the square base of the observable \eqref{eq:pyramid} is represented in terms of two triangles in the dually-weighted model.

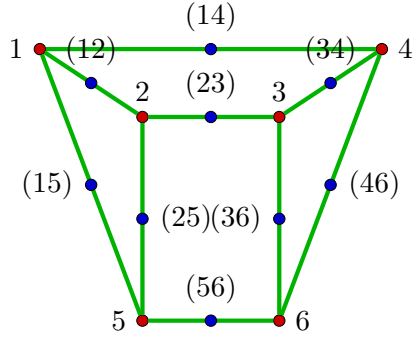
\begin{figure}
  \centering
  \tikzsetnextfilename{pyramiddw}
  \begin{tikzpicture}[scale=1.8]

\node [vb,label=180:1]	(1)	at (-1.25,1)	{}; 
\node [vb,label=90:2]		(2)	at (-0.5,0.5)	{}; 
\node [vb,label=90:3]		(3)	at (0.5,0.5) 	{}; 
\node [vb,label=0:4]		(4)	at (1.25,1)		{}; 
\node [vb,label=180:5]	(5)	at (-0.5,-1)		{};
\node [vb,label=0:6]		(6)	at (0.5,-1)		{};
\foreach \i/\j in {1/4,1/2,2/3,3/4}{
 \draw [eb] (\i) -- node[vh,label=above:(\i\j)] {}  (\j);
  }
\foreach \i/\j in {1/5,3/6}{
  \draw [eb] (\i) -- node[vh,label=left:(\i\j)] {}  (\j);
  }
\foreach \i/\j in {2/5,4/6}{
  \draw [eb] (\i) -- node[vh,label=right:(\i\j)] {}  (\j);
  }
\draw [eb] (5) -- node[vh,label=above:(56)] {}  (6);
%
%

\end{tikzpicture}
  \caption{\label{fig:dw3dgraph}The graph $\bbg$ composed of six unlabelled 3-patches.}
\end{figure}

\

 The dual-weighting mechanism can be applied to other classes of models.
The construction above works for arbitrary valence $\copies$, and can clearly be extended to multiple GFT  fields, if so wished. 
In models for $\std$-dimensional gravity, one would like the valence of the graphs associated to quantum states to be $\copies=\std$ like in the simplicial case, and for the same reasons. 
One should note, however, that in even dimensions $\std$, only effective nodes of even valency are then obtained after the dynamical contraction of virtual links. This combinatorial restriction has been noticed already in the previous section (proposition \ref{prop:surjections}).
If one wants to generate graphs of truly arbitrary valence, using the same mechanism, one can easily do so by incorporating a single odd-valent field species, endowed with a $\std$-dimensional interpretation. This doubling of fields does not change the general features of the construction. 




\subsection{Dually weighted gravitational models}
\label{sec:dw-gravitational}

To conclude the presentation of dually weighted GFT, I will show in this section how the amplitudes of gravitational SF models  are obtained from dually weighted simplicial GFT models.
In \sec{gravitational}, I have explained how the imposition of simplicity constraints leads to a special form of local spin-foam amplitudes, illustrated with a particular example. 
Here I will now give the details how the kinetic and vertex kernels in dually weighted GFT can be defined such that these amplitudes emerge effectively in the large-$M$ limit \eqref{dw-effective}.
Accordingly, the group field $\phi$ is defined on $\copies = \std = 4$ group elements in this section.

As mentioned before, GFT amplitudes are convolutions of propagators $\Pbb=\Kbb^{-1}$ and interaction kernels $\Vbb$, thus giving rise to spin-foam amplitudes in the form $\eqref{gft-sf}$.
The way the gravitational part of a dually weighted kinematic kernel $\Kbb$  is defined is thus most explicit on its inverse, \ie the propagator. 
To control the implementation of simplicity constraints along the dual weights, it is defined as 
\[\label{dw-edge-operator}
\Pbb(g_{v_1\vb},g_{v_2\vb};{m_{v_1\vb},m_{v_2\vb}}) =  \int \d h_{v_1\vb}\, \d h_{v_2\vb}\prod_{(\vb\vh)\in\E_{\bp_4}} \pbb(g_{v_1\vb\vh}, g_{v_2\vb\vh}; h_{v_1\vb}, h_{v_2\vb};m_{v_1\vb\vh},m_{v_2\vb\vh})
\]
where the integrand is a product over the edges in the unique unlabelled 4-patch of
\ba\label{eq:extprop}
\pbb(g_{v_1\vb\vh}, g_{v_2\vb\vh}; h_{v_1\vb}, h_{v_2\vb};m_{v_1\vb\vh},m_{v_2\vb\vh})  = 
\pbb^{\real}(g_{v_1\vb\vh}, g_{v_2\vb\vh}; h_{v_1\vb}, h_{v_2\vb})
\dbb^{\real}_{m_{v_1\vb\vh},m_{v_2\vb\vh}}  \nonumber\\ +
\pbb^{\virtual}(g_{v_1\vb\vh}, g_{v_2\vb\vh}; h_{v_1\vb}, h_{v_2\vb})
\dbb^{\virtual}_{m_{v_1\vb\vh},m_{v_2\vb\vh}}
\ea
which allows to define the gravitational edge operator in terms of $\pbb^\real$ solely on the real edges while the virtual operators $\pbb^\virtual$ are left trivial. 
More explicitly, for the example \eqref{edge-operator} this means that
\ba
    \pbb^{\real}(g_{v_1\vb\vh}, g_{v_2\vb\vh}; h_{v_1\vb}, h_{v_2\vb})
 &=& \sum_{J_{\vh}\in \mathcal{J}} \tr_{J_{\vh}}\left(g_{v_1\vb\vh}\;h_{v_1\vb}^{-1}\; \Sbb_{J_{\vh},N_{0}}\; h_{v_2\vb}\; g_{v_2\vb\vh}^{-1}\right)\;,\\
     \pbb^{\virtual}(g_{v_1\vb\vh}, g_{v_2\vb\vh}; h_{v_1\vb}, h_{v_2\vb})
 &=& \delta(g_{v_1\vb\vh}\;h_{v_1\vb}^{-1})\;\delta(h_{v_2\vb}\;g_{v_2\vb\vh}^{-1})\;. \label{virtual-edge-operator}
\ea
Meanwhile, the dual-weighting factors are
  \begin{equation}
    \dbb^{\real}= 
    \left(
    \begin{array}{c|c}
      1&0\\ \hline
      0&0
    \end{array}
    \right)
    \quad\quad,\quad\quad
    \dbb^{\virtual} = 
    \left(
    \begin{array}{c|c}
      0&0\\ \hline
      0&\Acal_M
    \end{array}
    \right)\;.
  \end{equation}
Notice that dual-weighting factors satisfy $\dbb = \dbb^{\real} + \dbb^{\virtual}$.  
Thus, the \emph{real} and \emph{virtual} labels encode the conditions of the dual-weighting mechanism as in \eqref{dw-matrix2}.  
In this way, the simplicity constraints are effectively only applied on real edges via $\pbb^\real$ while
the virtual edge factor $\pbb^\virtual$ decouples the information assigned to the edge in one atom, from that assigned with the other. 

The spin-foam vertex operator 
given by the GFT interaction kernel has in general an analogous structure
\[\label{dw-vertex-operator}
\Vbb_{\bbg}(g_v,m_v) = \prod_{\vh\in\Vh} \left[ 
\vbb^{\real}(g_{v\vb_1\vh},g_{v\vb_2\vh}) \ibb^{\real}_{m_{v\vb_1\vh},m_{v\vb_2\vh}} + 
\vbb^{\virtual}(g_{v\vb_1\vh},g_{v\vb_2\vh}) \ibb^{\virtual}_{m_{v\vb_1\vh},m_{v\vb_2\vh}}\right]\;
\]
for $\bbg=\bbg_{4,\simplicial}$ the unique 4-simplicial bisected boundary graph.
Since all gravitational information can be encoded in the edge operator the factors are simply
    \begin{equation}
    \vbb^{\real}(g_{v\vb_1\vh},g_{v\vb_2\vh}) = \vbb^{\virtual}(g_{v\vb_1\vh},g_{v\vb_2\vh})=  \delta(g_{v\vb_1\vh},g_{v\vb_2\vh}) \;
    \end{equation}
and
    \begin{equation}
	    \ibb_{real} =
	    \left(
    \begin{array}{c|c}
      1&0\\ \hline
      0&0
    \end{array}
    \right)
    \quad\quad\quad\quad
    \ibb_{virtual} = 
    \left(
    \begin{array}{c|c}
      0&0\\ \hline
      0&\id_M
    \end{array}
    \right)\;.
  \end{equation}

%

One can confirm now that these operator assignments lead to the correct effective amplitude. 
\begin{proposition}
In the large-$M$ limit, the effective amplitude assigned by a dually-weighted gravitational model defined by the (inverse) kinetic kernel \eqref{dw-edge-operator} and interaction kernel \eqref{dw-vertex-operator} coincides with that of the original gravitational model \eqref{gravitational-sf}.
\end{proposition}
\begin{proof}
Employing proposition \ref{prop:mlimit}, in the large-$M$ limit, the contributing molecules are restricted to those, for which their virtual vertices $\vh \in \Vh \In \V_{\sfrt}$ lie in precisely four virtual faces $f\in\mathcal{F}_{\sfrt}$. 
The amplitude in the perturbative expansion \eqref{dw-perturbative} is
\[\label{dw-amplitude}
A(\sfrt;\lambda) = \lambda^{|\V|} \int [\d g]\sum_{[m]}\prod_{\vb\in\Vbar} 
\Pbb_{\vb}(g_{\vb};m_{\vb}) 
\prod_{v\in\V}\Vbb_\bbg(g_v;m_v)\;,
	\end{equation}
  Then, the key calculation examines the effect of integrating out the variables associated to components in the neighbourhood of this vertex $\vh$. More precisely, in the amplitude $A(\sfrt;\lambda)$ there is for each such vertex a factor 
  \begin{multline}
C_\vh(g_{v_1\vb_{12}\vh}, g_{v_2\vb_{12}\vh}; h_{v_1\vb_{12}}, h_{v_2\vb_{12}}) =
\vbb^{\virtual}(g_{v_1\vb_{12}\vh},g_{v_1\vb_{21}\vh})\;  \pbb^{\virtual}(g_{v_1\vb_{12}\vh}, g_{v_2\vb_{12}\vh}; h_{v_1\vb_{12}}, h_{v_2\vb_{12}})\\
\times \vbb^{\virtual}(g_{v_2\vb_{12}\vh},g_{v_2\vb_{21}\vh})\; \pbb^{\virtual}(g_{v_1\vb_{21}\vh}, g_{v_2\vb_{21}\vh}; h_{v_1\vb_{21}}, h_{v_2\vb_{21}})\;
	  \label{eq:eval}
\end{multline}
whose configuration is illustrated in \fig{contraction}.
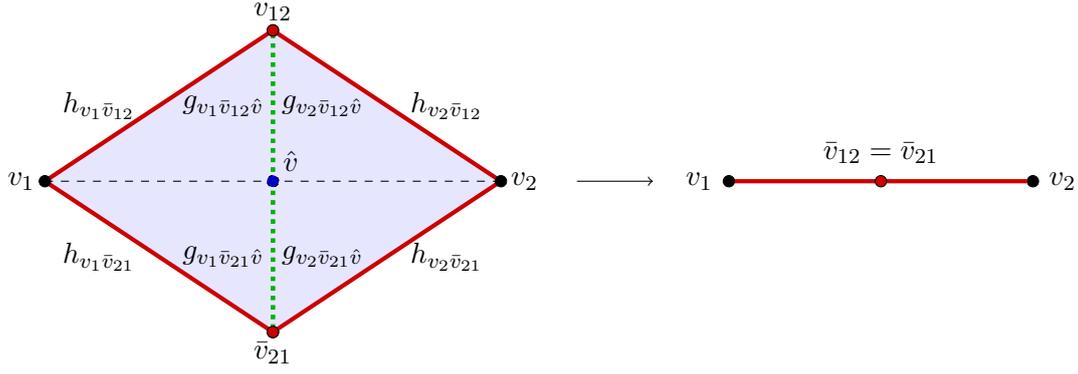
\begin{figure}
  \centering
  \tikzsetnextfilename{contraction}
  \begin{tikzpicture}[scale=2]

\begin{scope}
\node [c,label=left:$v_1$]		(1)	at (0,0)		{};
\node [c,label=right:$v_2$]	(2)	at (3,0)	 	{};  
\node [c,label=45:$\vh$]		(vh)	at (1.5,0)		{};
\node [c,label=above:$\vb_{12}$](12)at (1.5,1)		{};
\node [c,label=below:$\vb_{21}$](21)	at (1.5,-1)		{};
\foreach \i/\j/\k in {1/12/21,2/12/21}{
\path		[f] 	(\i) 	--  (\j) -- (vh) --  (\k) -- cycle;
\draw	[eh]	(\i)  	-- (vh);
\draw[dotted,eb](vh) 	-- (\j);
\draw[dotted,eb](vh) node[vh]{}	-- (\k);
 }
\draw [e] 	(1) 			-- node [label=left:$h_{v_{1}\vb_{12}}$,label=right:$g_{v_{1}\vb_{12}\vh}$]{} (12) node[vb]{};
\draw [e] 	(1) node[v]{}	-- node [label=left:$h_{v_{1}\vb_{21}}$,label=right:$g_{v_{1}\vb_{21}\vh}$]{} (21) node[vb]{};
\draw [e] 	(2) 			-- node [label=right:$h_{v_{2}\vb_{12}}$,label=left:$g_{v_{2}\vb_{12}\vh}$]{} (12) node[vb]{};
\draw [e] 	(2) node[v]{}	-- node [label=right:$h_{v_{2}\vb_{21}}$,label=left:$g_{v_{2}\vb_{21}\vh}$]{} (21) node[vb]{};
\end{scope}

\draw [->] (3.5,0) -- (4,0);

\begin{scope}[xshift=4.5cm]
\node [v,label=left:$v_1$]		(1)	at (0,0)		{};
\node [v,label=right:$v_2$]	(2)	at (2,0)	 	{};  
\draw [e] (1) -- node[vb,label=above:${\vb_{12} = \vb_{21}}$] {} (2);
\end{scope}
\end{tikzpicture}
  \caption{\label{fig:contraction} Integrating out a virtual face.}
\end{figure}
Applying the integrals with respect to the elements the set $g_{\vh}$ contained in the amplitude \eqref{dw-amplitude} these factors \eqref{eq:eval} simplify to
\[\label{eq:evalResult}
\int \d^4 g_{\vh} \; C_\vh(g_{v_1\vb_{12}\vh}, g_{v_2\vb_{12}\vh}; h_{v_1\vb_{12}}, h_{v_2\vb_{12}}) = \delta(h_{v_1\vb_{12}}\;h_{v_1\vb_{21}}^{-1})\;\delta(h_{v_2\vb_{12}}\;h_{v_2\vb_{21}}^{-1})\;.
\]
   This integration can be completed since the elements in $g_{\vh}$ only occur within the four factors \eqref{eq:eval}.
   Now, one is free to use these two $\delta$-functions to integrate out the variables $h_{v_1\vb_{21}}$ and $h_{v_2\vb_{21}}$, setting $h_{v_1\vb_{21}}=h_{v_1\vb_{12}}$ and $h_{v_2\vb_{21}}=h_{v_2\vb_{12}}$ in the remaining factors within \eqref{dw-amplitude} to arrive at the gravitational amplitude \eqref{gravitational-sf} assigned to the molecule $\sfr = \Pi_{\copies,\snb}(\sfrt)$ obtained from $\sfrt$ by molecular reduction along the virtual structure (illustrated in \fig{contraction}):
\[\label{sf-gft-relation}
\lambda^{|\V|} Z(\sfr) = A(\sfrt;\lambda) \;.
\]
Note that the exponent $|\V|$ in the the spin-foam weight $w_\sfr = \lambda^{|\V|}$ still refers to the number of internal bulk vertices $v\in\V$ of the original labelled molecule $\sfrt$.
\end{proof}

\

While the assignment of trivial edge operators $\pbb^\virtual$ \eqref{virtual-edge-operator} to virtual edges is the way to obtain the KKL-extension \cite{\KKL} of gravitational SF amplitudes,
the dually weighted GFT broadens the framework to allow for other reasonable models.
For example, another tempting proposal is to impose the simplicity constraints on both real \emph{and} virtual structures.
The motivation is that, assuming a polyhedral interpretation is available, the polyhedra corresponding to the states of the model will now be decomposed into \emph{geometric} simplices, the geometricity of each being ensured by the imposition of the simplicity constraints.	
This amounts to altering the propagator \eqref{eq:extprop}, now defining
\ba
\label{eq:strandfinal}
\pbb^{\virtual}(g_{v_1\vb\vh}, g_{v_2\vb\vh}; h_{v_1\vb}, h_{v_2\vb})
& = & \pbb^{\real}(g_{v_1\vb\vh}, g_{v_2\vb\vh}; h_{v_1\vb}, h_{v_2\vb}) \nonumber\\
& =  &\sum_{J_{\vh}\in \mathcal{J}} \tr_{J_{\vh}}\left(g_{v_1\vb\vh}\;h_{v_1\vb}^{-1}\; \Sbb_{J_{\vh},N_{0}}\; h_{v_2\vb}\; g_{v_2\vb\vh}^{-1}\right) .
\ea
This defines, a priori, a \emph{different} SF model, with an expected higher degree of geometricity.	
A motivation for this change stems from the logic that polytopes which are constructed from geometric simplices are likely to be more physically viable than polytopes constructed from simplices that are only partially geometric (in the sense that the simplicity constraints are not imposed on some of their virtual sub-cells). 

It is worth clarifying that the resulting model can be interpreted in two ways, depending on whether the large-$M$ limit is taken or not:
\begin{itemize}[leftmargin=*]
\item The perturbative series are catalogued by molecules in $\sfrst_{4,\simplicial}$. In the large-$M$ limit, the surviving molecules are again those of $\sfrst_{4,\snb}$.   Within this model, reduction leads to effective amplitudes that do not coincide with those assigned by KKL-extended SF models to the generic spin-foam molecules $\sfrs' = \Pi_{4,\snb}(\sfrst_{4,\snb})$.

\item The perturbative series are catalogued by molecules in $\sfrs_{4,\simplicial}$.
Due to the coincidence of the strand factors in \eqref{eq:strandfinal}, the dual-weighting part of the amplitude factorizes completely from the gravitational part, as well as over the vertices $\vh$ yielding a factor
\[\label{eq:factorfinal}
w_\vh = \tr\Big(\prod_{(\vb\vh)}\dbb\Big) + 1 =\tr\Big(\prod_{(\vb\vh)}\Acal\Big)\;,
\]
for each $\vh\in\Vh$, where the product is over those edges $(\vb\vh)$ incident at $\vh$. Thus, without taking the large-$M$ limit, one gets the original \emph{simplicial} gravitational model, with a slight modification of the weights by a factor $w_\vb^{|\Vh|}$. 
Since the number of vertices $|\V|$ and the number of faces $|\Vh|$ are generically independent for simplicial molecules $\sfrs_{4,\simplicial}$, these factors cannot be simply absorbed in a re-definition of the coupling $\lambda$ which appears with a power of $|\V|$ in \eqref{sf-gft-relation}.
\end{itemize}

\

To sum up, the topic of this \sec{dw-gft}, has been the  definition of a class of GFTs admitting boundary states supported on arbitrary closed graphs.
Rather than increasing the number of fields and interactions, I have extended the data set of the usual simplicial GFT. 
Employing these extra arguments to invoke a dual-weighting matrix, the effective dynamical content of arbitrary spin foams and boundary states is then obtained in a limit of the theory.

There are a number of advantages of this strategy as compared with the multi-field extension discussed in \sec{multi-field}.
Quite generally, the underlying simplicial GFT with just one field and one interaction term is much more easily controllable. This is not only relevant from the pragmatic point of view where one aims at explicit calculations of estimates of observables but also from the more mathematical perspective where one seeks for a proper definition of the GFT path integral.

Also, notice that the {\it coloured extension} of the GFT  formalism can be directly applied to the dually-weighted model. 
This can be done either by choosing also a simplicial GFT  interaction, which brings one back to the standard simplicial setting, or by choosing as GFT  interactions only the tensor invariant ones. The result of doing so is, in both cases, a set of GFT  Feynman diagrams dual to combinatorial complexes whose full homological structure can be reconstructed from the colour information. 

Furthermore, the class of dually weighted GFTs opens up the theory space.
The precise choice of GFT  action, incorporating the dual-weighting mechanism, is then a matter of model building. In particular, depending on the choice of interaction kernels, some classes of graphs, present in the kinematical Hilbert space, can be suppressed dynamically. 

Therefore, one possible criterion for model building stems from the wish to suppress or enhance specific combinatorial structures. 
Conversely, one may want to start from the simplest set of GFT  interactions that ensures that {\it all} kinematical states participate to the quantum dynamics. 
Simplicial spin-foam atoms, in the context of the dually-weighted theories, satisfy this criterion as proven in this section.

Another criterion that might determine the choice of GFT  interaction combinatorics emerges from the correspondence between the interaction kernels and the matrix elements of a canonical LQG  projector operator in the Fock representation, emphasized in \cite{Oriti:2013vv}. Prescribing the latter implies a choice for the former. In general though, one should expect there to be infinitely many non-trivial matrix elements, meaning infinitely many GFT  interaction kernels, unless these are restricted by very strong symmetry requirements. As a result, the real quest centres on pinpointing the subsets of interactions that are \emph{physically} relevant at different scales, in particular,  to define the theory in some deep UV or IR regime. In other words, the problem becomes that of \emph{GFT  re\-normalization} \cite{\gftrenorm,\COR}. In fact, one should expect that the renormalization group flow will select a finite set of GFT  interactions to define a renormalizable GFT  model. Moreover, this  dictates which new terms are relevant for the quantum dynamics at different scales. In turn, this prescribes a renormalizable LQG  dynamics.


\renewcommand{\rep}[1]{{j_{#1}}}

\chapter
{Effective dimensions of quantum geometries}\label{ch:dimensions}

In this chapter I am going to address the challenge of relating the discrete quantum geometries in the theory laid out in the last chapter to the continuum geometries one observes. 
To this end I will focus on a special case: 
I will investigate global geometric properties of quantum geometry.
In particular, based on the publications \cite{\COTb,\COTc}, I will consider here one of the more feasible cases, that is effective dimension observables, more precisely the spectral, walk and Hausdorff dimension.
While the focus will be on the spectral dimension, I will investigate all three dimension in the most interesting case of a class of superposition states with a dimensional flow, in order to be able to also make some more precise statements about a possible fractal structure in quantum gravity.

The spectral dimension of quantum geometries has been a rather hot topic in recent quantum gravity research. It is one of the indicators of topological and geometric properties of quantum spacetime as described by different approaches, as well as a way to check that an effective semi-classical spacetime is obtained in appropriate sectors of the theory.
One of the main goals of any such analysis is to prove that, among all the states and histories that appear in the theory (most of which far from describing any smooth spacetime geometry), configurations which do so are either dominant or the approximate result of averaging over the others. 
A further goal is to study whether the effective dimension of spacetime remains the same at different scales or if it is subject to a dynamical reduction in the ultraviolet regime \cite{Hooft:1993vl,Carlip:2009cy,Calcagni:2009fh}.

In dynamical triangulations \cite{Ambjorn:2010kv,Ambjorn:2012vc}, the spectral dimension has been important as one of the few observables available to classify phases of spacetime ensembles. 
In particular, in CDT in four dimensions it was found that, while very close to the spacetime topological dimension $\std=\sd+1=4$ at large scales, at least in the geometric phase, it takes smaller values at short distances, approaching the value $\Ds\simeq 2$ in the ultraviolet limit \cite{Ambjorn:2005fj,Ambjorn:2005fh}. Similar results are obtained in three dimensions \cite{Benedetti:2009bi}. Two is, of course, an interesting value for the effective dimension in the ultraviolet, because gravity is perturbatively renormalizable in two dimensions. In fact, also in the asymptotic safety scenario \cite{Niedermaier:2006up,Reuter:2012jx} one can find a reduction of the spectral dimension from the topological dimension $\std$ to half of it under the renormalization group flow using analytical methods \cite{Lauscher:2005kn,Reuter:2013ji,Calcagni:2013jx}. Other formalisms manifest similar behaviour. Ho\v{r}ava--Lifshitz gravity has a built-in dynamical dimensional reduction, due to the defining anisotropic scaling \cite{Calcagni:2013jx,Horava:2009ho}. 
Modified dispersion relations in non-commutative spacetimes \cite{Benedetti:2009fo,Alesci:2012jl}, super-renormalizable non-local gravity \cite{Modesto:2012er}, black-hole effective physics \cite{Arzano:2013jj} and other scenarios (including string theory \cite{Calcagni:2014ig}) give rise to a dimensional flow as well.

Apart from the intrinsic interest in this effect, its occurrence in approaches to quantum gravity which are otherwise hard to compare has raised hope that the spectral dimension might help to relate such approaches to one another, or at least to confront their results \cite{Calcagni:2012rm}. Although the spectral dimension is a rather coarse tool in comparison with more refined geometric indicators \cite{Calcagni:2013jx,Calcagni:2013ds}, its determination in full-fledged quantum geometries is still a technically challenging problem.

All the above considerations give good reasons to investigate the spectral dimension also for discrete quantum geometries as they appear in LQG. 
Using the spectrum of the area operator in LQG to guess an effective dispersion relation for propagating particles, some evidence was found for a dimensional flow of kinematical LQG states \cite{Modesto:2009bc} and their spacetime spin-foam dynamics \cite{Caravelli:2009td,Magliaro:2009wa}. Here I go beyond the arguments used in previous analyses and investigate the spectral dimension of LQG states in detail, at least at the level of kinematical states. 
In particular, I will study the spectral dimension for states defined on large graphs or complexes, and of superpositions over such complexes.

LQG differs from all the above-mentioned approaches in that the degrees of freedom of the theory are discrete geometric data, \ie a combination of holonomies, fluxes or Lorentz group representations \cite{Rovelli:2004wb,Baratin:2011hc} associated with combinatorial structures such as graphs and combinatorial complexes. This combination poses new challenges.
So far, the spectral dimension has been investigated in more or less traditional smooth settings, in terms of either modified differential structures (multi-scale spacetimes \cite{Calcagni:2013ds}) or modified dispersion relations on a smooth geometric manifold. 
Examples are asymptotic safety, where anomalous scaling is a general feature at the non-Gaussian fixed point \cite{\asfractal}; non-commutative spacetimes, where modified dispersion relations are a consequence of a deformed Poincar\'e symmetry \cite{\ncfractal}; or Ho\v{r}ava--Lifshitz gravity, where the Laplacian is multi-scale by construction \cite{\hlfractal}.
Also previous work on the LQG case \cite{Modesto:2009bc} has been limited to an approximation in which the scaling of the Laplacian is first extracted from the area spectrum and then effectively treated as continuous. In other approaches, the setting is purely combinatorial: such is the case of CDT \cite{\cdtfractal,Benedetti:2009bi, Coumbe:2015bq}  
or graph models \cite{Giasemidis:2012kq,Giasemidis:2012er,Giasemidis:2013wz}.

Therefore, the first important step is a definition of the spectral, walk and Hausdorff dimension on the type of discrete quantum geometries LQG is built on. 
To this end, the generalization of discrete exterior calculus \cite{Desbrun:2005ug} to be applicable to such discrete geometries described in \sec{calculus} is a crucial precondition.


The combination of combinatorial discreteness and additional (pre-)geometric data obviously must result in an interplay of their respective effects. To the best of my knowledge, the latter have never been classified in the literature. A major aim of this chapter is to understand this interplay in the case of discrete quantum geometries. 
To this end, I will first have to study systematically the topological and geometric effects in the spectral dimension of smooth spheres and tori as well as the discreteness effects of combinatorial complexes with geometric realizations as cubulations or triangulations of these.
This groundwork having been done, I can then compare the effective dimensions of quantum states with those of the discrete geometries they are semi-classical approximations or superpositions of. 
For the spectral dimension, I find as a general tendency that the effect of the underlying combinatorics dominates and that a relatively large size of the base complexes is needed to obtain a geometric regime of spatial dimension at all.

To construct explicit LQG states on large complexes and calculate the corresponding expectation values of global observables is a major and seldom accepted challenge. In the case of the spectral dimension, however, this becomes inevitable. 
As a line of attack, I set up a way to numerically construct large abstract simplicial complexes and define geometries thereon. I follow \cite{Bell:2011wu} to derive their combinatorial properties needed for studying the action of the Laplacian operator.

While the preparatory work on the classical geometries is done in arbitrary dimensions, I restrict the detailed quantum analysis to kinematical states of (2+1)-dimensional Euclidean LQG. There are two reasons for this restriction. 

The first is merely technical. The number of degrees of freedom for given assignments of algebraic data grows with the combinatorial size $\size$ of the considered complexes as $\np0 \sim \size^\sd$. The Laplacian operator defined thereon is an $\np0\times \np0$ matrix, which makes the computational effort more severe as $d$ increases. Thus, for a given size $\size$ one can calculate much larger $\sd=2$ discrete geometries than in $\sd=3$. On the other hand, a sufficient combinatorial size turns out to be crucial for the physical interpretation of the spectral dimension as a notion of positive-definite dimension somewhere close to $\sd$. This limitation could be easily overcome by using a more powerful computational environment for numerical analysis.

The second reason lies in the structure of quantum gravity itself. In 2+1 dimensions, the spin-network basis of LQG simultaneously diagonalizes all edge-length operators, permitting a straightforward definition of the Laplacian and a direct identification of deviations of the quantum spectral dimension from its classical counterpart as quantum corrections. In 3+1 dimensions, one would have to deal further with effects due to the non-commutativity of the full set of geometric operators needed to define the Laplacian, as well as with the role of non-geometric configurations which are present in the standard $SU(2)$-based Hilbert space of the theory. Its construction thus becomes much more involved and it is much harder to be implemented efficiently for numerical calculations. 
While these challenges do not pose any obstacle in principle, I prefer to concentrate here on the more straightforward 2+1 case for detailed computations.

Motivated by the findings in the concrete $(2+1)$-dimensional LQG case, I will then use a few reasonable assumptions to set up a model for superpositions over complexes in arbitrary spatial dimension $\sd$ which is based on analytical results for the single terms.
Surmounting the numerical limitations in this way, I can then calculate much larger superpositions with up to $10^6$ complexes summed over.
Furthermore it becomes possible to investigate a much broader class of states with respect to the functional form of the superposition coefficients.

In this way I find a special class of such superposition states with power function coefficients  characterized by a real-valued parameter $\alpha$ which have a dimensional flow in the spectral dimension towards $\ds=\alpha$ in an intermediate UV regime.
Furthermore, also walk and Hausdorff dimensions can be evaluated.
This allows to single out a subclass of states for which the effective geometry can be considered as fractal in a more precise sense.


The investigation of effective dimensions is quite explorative in nature. 
There are no constraints on the kind of complexes which discrete  quantum geometry states could be defined on. 
Thus, the only guidance may come from classical geometries one might want to approximate semi-classically.
Furthermore, there is no unique way to parametrize the `quantum-ness' of states. 
In fact, to search for genuine quantum effects of discrete quantum geometries, it is not enough to pick one state and look only at the value of its spectral dimension at small scales, as is usually done when some kind of dimensional flow is expected. 
Instead, I compute the spectral dimension for those states considered as ``more quantum'" which are in the first place coherent spin-network states peaked at smaller representation labels and with larger spread and, more generally, states highly randomized in the intrinsic geometry or superposed with respect to various geometric degrees of freedom.
The more interesting quantum effects like a dimensional flow occur in states of high ``quantumness" in the sense that  they are superpositions over a large set of different complexes.

\

The structure of this chapter is as follows.
In section \sec{dimension-definition}, the concepts of spectral, walk and Hausdorff dimension are introduced by their common definitions in a smooth context.
I present then definitions on discrete geometries which is the basis whereupon to define the quantum dimensions. 

As a preparation for quantum geometries, in section \ref{sec:classical} I will discuss first classical geometric and topological features, both analytically and numerically, in particular for the simple examples of spheres and tori.
Similarly I address then features of the effective dimensions on discrete structures, exemplified through various cases.

In section \ref{sec:coherent-states} I present numerical calculations of the spectral dimension of kinematical LQG states in $\std=2+1$ dimensions.
Here the focus is on investigating the relevance, for the spectral dimension, of the geometric data additional to the geometry coming from purely combinatorial structures. 
The spectral dimension of semi-classical states is compared with the classical one to identify quantum corrections, which turn out to be rather small. 
Furthermore, superpositions of coherent states peaked at different spins will be discussed.

Complementary to these superpositions in geometry, in section \ref{sec:superpositions} I analyse the spectral dimension of superpositions of states on two classes of torus triangulations. It turns out that the quantum spectral dimension in LQG is more sensitive to the underlying graph structure of states than to the associated geometric data.
In particular I find a dimensional flow for a fairly general class of superposition states over lattices for fixed overall volume which depends on the exponent parameter of their power function superposition coefficients.

Sections \ref{sec:classical} and \ref{sec:coherent-states} are based on the publication \cite{\COTb}, though the subsection \ref{sec:hypercubic} contains to a large extent new material. Section \ref{sec:dimensional-flow} is based on the publication \cite{\COTc}.
The definitions presented in \sec{dimension-definition} are a more detailed and improved version of the ones presented in \cite{\COTb} and \cite{\COTc}, respectively.


\section{Discrete and quantum definition of effective dimensions}
\label{sec:dimension-definition}

In this section, I start by introducing the concepts of Hausdorff, spectral and walk dimensions in the usual context of  Riemannian manifolds (\sec{riemannian}). 
Then I show in \sec{discrete-geometries} how these definitions can be transferred to combinatorial complexes. 
This sets the stage to define these quantities in \sec{quantum} on quantum states of geometry with variables supported on such complexes.


\subsection{Three concepts of effective dimension}
\label{sec:riemannian}

The paradigmatic example of an effective dimension is the Hausdorff dimension $\dh$.
It allows to assign a notion of dimension to a metric space $X$  in terms of the infimum of positive real numbers $d\in\R^+_0$ for which the $d$-dimensional Hausdorff measure of $X$ vanishes (which itself is defined as a covering of $X$ in terms of $d$-dimensional balls).
On these grounds, the more applicable definition of $\dh$, common in the quantum gravity and random geometry literature where a notion of volume is at hand, is directly as the scaling of the volume $V(r) \propto r^{\dh}$ of elements at radius distance up to $r$. 
While the original definition is in the limit of small $r$, it is informative to generalize $\dh$ to a scale dependent function 
\[\label{dh}
d_{H}(r):=\frac{\partial}{\partial\ln r} \ln V(r).
\]

Spectral and walk dimension are similarly defined as scalings of geometric quantities, but they are based on the spectral properties of a space equipped with certain differential operators 
instead of mere metric properties.
In more physical terms, they are related to the scaling of the propagator of a massless test particle field or, equivalently, to a diffusion process on that space. 
This is defined via the heat kernel $K(x,x_0;\tau)$ which describes the diffusion of the particle in fictitious time $\tau$ from point $x_0$ to $x$. 
The heat kernel is the solution of the diffusion equation
\begin{equation}\label{diffusion-equation}
\partial_{\tau}K(x,x_0;\tau)-\Delta_{x}K(x,x_0;\tau)=0\,,
\end{equation}
where $\Delta$ is the Laplacian with respect to the underlying space.
Since the diffusion constant in front of the Laplacian is set to one here, $\tau$ has dimension of a squared length.
 
Alternatively, equation \eqref{diffusion-equation} can be regarded as a running equation establishing how much one can localize the point particle, placed at some point $x_0$, if one probes the geometry with resolution $1/\sqrt{\tau}$ \cite{Calcagni:2014uv}. 
The length scale $\sqrt{\tau}$ represents the minimal detectable separation between points. 
While in the diffusion interpretation $K$ is the probability to find the particle at point $x$ after diffusing for some time $\tau$ from point $x_0$, in the resolution interpretation $K$ is the probability to see the particle in a neighbourhood of point $x$ of size $\sqrt{\tau}$ if the geometry is probed with resolution $1/\sqrt{\tau}$. 
For infinite resolution ($\tau=0$) and a delta initial condition, the particle is found where it was actually placed (at $x_0$), while the smaller the resolution (\ie larger diffusion time) the wider the region where it can be seen. There is no practical difference between the two interpretations but the latter is more suitable in diffeomorphism-invariant theories where the notion of diffusion time makes no sense (while a Lorentz-invariant scale does).

A formal solution to the problem is provided by
\begin{equation}\label{Kformal}
K(x,x_0;\tau)=\left\bra x_0|\exp(\tau\Delta)|x\right\ket\,, 
\end{equation}
where the initial condition is incorporated as the definition of the inner product $\left\bra \cdot|\cdot\right\ket $ of the function space. 
Usually, the particle is taken to be concentrated at $\tau=0$ in the point $x_0$, but it could start with other distributions. 

The traditional context where this is  defined are (pseudo-) Riemannian $\m$-manifolds $(\mf,g)$. 
The relevant differential operator is then the Beltrami-Laplace operator, or equally the Hodge-Laplace-operator on functions, 
\begin{align}
\Delta^{g} & =\Tr\nabla^{2}=g^{\mu\nu}\nabla_{\mu}\nabla_{\nu}=g^{\mu\nu}\left(\partial_{\mu}\partial_{\nu}-\Gamma_{\mu\nu}^{\rho}\partial_{\rho}\right)=\frac{1}{\sqrt{g}}\partial_{\mu}\sqrt{g}g^{\mu\nu}\partial_{\nu}\,.
\label{eq:LaplaceBeltrami}
\end{align}
which is a functional of the metric $g$.
With the initial condition
\begin{equation}
K(x,x_0;0^{+}) = \bra x_0 | x \ket = \frac{1}{\sqrt{g}}\delta(x-x_0)\,,
\end{equation}
the heat kernel is
\begin{equation}
K(x,x_0;\tau)=\frac{\rme^{-\frac{\dist x{x_0}^2}{4\tau}}}{(4\pi\tau){}^{\frac{\m}{2}}},
\label{eq:Kheatg}
\end{equation}
where $\dist x{x_0}$ is the geodesic distance. 
On $(\mf,g)$ there is a well-defined (the so called Seeley-DeWitt) expansion of the heat kernel around the solution
on a flat manifold $K_{0}(x,x_0;\tau)$ with the metric distance being the Euclidean norm $\dist{x}{x_0}=\sqrt{|x-x_0|^2}$, 
\begin{equation}
K(x,x_0;\tau)=K_{0}(x,x_0;\tau)\underset{k=0}{\overset{\infty}{\sum}}b_{k}(x,x_0)\,\tau^{k},
\end{equation}
where the $b_{k}(x,x_0)$ are computable in terms of geometric invariants \cite{Vassilevich:2003fw}.

The spectral dimension probed by the particle can now be defined as the scaling of the trace of the heat kernel over $\mf$
\begin{equation}\label{heat-trace}
P(\tau) = \Tr_\mf K(x,x_0;\tau) = \int_\mf \d x_0\, \sqrt{g}\, K(x_0,x_0;\tau)\,,
\end{equation}
which is also called \emph{return probability}. 
While in the flat case $P(\tau) = (4\pi\tau)^{-\m/2} V$, where $V=\int_\mf \d x \sqrt{g}$, the general heat-kernel expansion yields a power series 
\begin{equation}
P(\tau)=\frac{1}{(4\pi\tau)^{\frac{\m}{2}}}\overset{\infty}{\underset{k=0}{\sum}}a_{k}\tau^{k}
\end{equation}
where  
the coefficients $a_{k}=\Tr_\mf\,b_{k}$ consist of geometric invariants. 

Motivated by the flat case with scaling $P(\tau) \propto (\sqrt\tau)^{-\m}$, one defines the \emph{spectral dimension} $\ds$ as the scaling of the heat trace,
\[\label{ds-scaling}
P(\tau) \propto (\sqrt\tau)^{\ds}
\]
such that $\ds=\m$ for the flat case, while in general there are corrections related to global properties described by the geometric invariants $a_{k}$. 
Like in the case of the Hausdorff dimension, 
one can generalize this implicit definition (\ref{ds-scaling}) which is usually in the first place understood in the small-$\tau$ limit to a function of (diffusion) scale
\begin{equation}\label{ds}
\ds(\tau):=-2\frac{\partial\ln P(\tau)}{\partial\ln\tau}\;.
\end{equation}
By definition (via heat trace $P$, heat kernel $K$ and Laplace operator $\Delta$) the spectral dimension is a functional of the geometry $\left[g\right]$ on
the manifold $\mf$.



A closely related observable is the \emph{walk dimension} $\dw$ which is also based on the heat kernel.
It is defined via the scaling of the mean square displacement implicitly as 
\[\label{displacement}
\langle X^{2}\rangle _{x_0}(\tau)=\int \d x \, \dist x{x_0}^2 K(x,x_0;\tau)\propto\tau^{2/\dw},
\]
which can be turned into an explicit, scale-dependent definition 
\[\label{dw}
\dw(\tau):=2\left(\frac{\partial\ln\langle X^{2}\rangle _{x_0}(\tau)}{\partial\ln\tau}\right)^{-1}.
\]
In this work I will always use the scale-dependent definition of each effective dimension. 

While defined on Riemannian manifolds in the first place, spectral and walk dimension are interesting in the context of quantum gravity because their definition resting on the Laplacian can be generalized to other classes of geometries as well.
In particular, I will explain in the following how to define them on discrete quantum geometries. 
To this end, the first step is their definition on combinatorial complexes with a geometric interpretation which is the topic of the next section.


\subsection{Effective dimension definitions for discrete geometries}
\label{sec:discrete-geometries}

In this section I will now define the effective dimensions 
on combinatorial complexes with geometric interpretation. 
In the case of Hausdorff dimensions $\dh$ this is very much straightforward since the only notions needed are volumes and distance.
For the spectral and walk dimension $\ds$ and $\dh$ the discrete exterior calculus introduced in section \sec{calculus} is needed.
Through the dependence of the discrete Laplacian on the geometric volumes, both dimensions eventually depend on these volumes as well, but in a more involved way.
In all cases I understand the points between which to measure distances and on which scalar fields are living as the vertices of the dual complex. 
Eventually, all definitions rely combinatorially just on the dual 1-skeleton.

\subsubsection*{Discrete Hausdorff dimension}

For the definition of discrete Hausdorff dimension, consider an $\m$-dimensional complex $\cp$ with dual $\cps$ and an assignment of $\m$-volumes $V_\cl$ to all $\m$-cells $\cl\in\cpp\m$ and edge lengths $\dl_e$ to all $e\in\cpsp1$.
Then, one can define the distance between two vertices $v_a,v_b\in\cpsp0$ in terms of the shortest path (definition \ref{def:paths}, equation \eqref{path}) between them,
\[\label{discrete-distance}
\dist{v_a}{v_b} := \min\{\sum_{e\in \cpsp1 \textrm{ in }W} \dl_e \mid W = v_a e_{a1} v_1 e_{12}\dots v_k e_{kb} v_b \textrm{ 1-dimensional path in } \cps \}  
\]
(If $v_a$ and $v_b$ are not 1-dimensionally connected one sets $\dist{v_a}{v_b} := \infty$.)
Using the primal volume $V_{\star v}$ as the natural $\m$-volume attached to the dual vertices, one arrives at the definition of the volume $\Vv\cp$ of a ball of radius $r$ around a vertex $v_0$ in $\cps$ as
\[\label{discreteV}
\Vv\cp := \sum_{\dist{v}{v_0}\le r} V_{\star v} \;.
\]
The Hausdorff dimension $\dh^\cp(r)$ of $\cp$ is defined through \eqref{dh} using as the volume
\[\label{discreteV2}
\Vr\cp := \frac1{\np\m} \sum_{v_0\in\cpsp0} \Vv\cp
\]
 obtained from averaging over the $\np\m$ centre vertices $v_0\in\cpsp0\cong\cpp\m$.

\subsubsection*{Discrete spectral and walk dimension}


For spectral and walk dimension all relevant quantities can be defined using the discrete calculus introduced in \sec{calculus} in the following way.
The diffusing scalar test field on $\cp$ is a field $\phi\in\Omega^0(\cps)\cong\Omega^\m(\cp)$, \ie living on dual vertices $v_a,v_b,... \in\cpsp0$ or equivalently on the $\m$-cells $\cl_a,\cl_b,... \in\cpp\m$.
Now, assume there is a geometric interpretation assigning next to the $\m$-volumes $V_a\equiv V_{\cl_a}$ and dual lengths $\dl_{ab}\equiv \dl_{(v_a v_b)}$ also $(\m-1)$-volumes $V_{ab}$ to the common boundary cells of pairs of $\m$-cells $c_a,c_b\in\cpp\m$.
Assume furthermore that there is a boundary map $\bm_\m$, \ie relative orientations between $\m$-cells and $(\m-1)$-cells, or equivalently dual vertices and edges (definition \ref{def:chain-complex}).
Then there is a discrete Laplacian acting on the fields $\phi\in\Omega^0(\cps)$, denoted $\Delta_\cp$ from now on.
According to \eqref{scalar-laplacian} it takes the form
\ba\label{discrete-laplacian}
-(\Delta_\cm \phi)_a &=& \sum_{b\sim a} (\Delta_\cm)_{ab} (\phi_a - \phi_b)\nonumber\\
&=& \frac1{V_a} \sum_{b\sim a} \frac{V_{ab}} {\dl_{ab}} (\phi_a - \phi_b)\,,
\ea
where the sum runs over $\m$-cells $c_b$ adjacent to $c_a$. 



The formal heat kernel solution \eqref{Kformal} gains a well-defined meaning directly in terms of the bra-ket formalism for the field space $\Omega^0(\cps)$ introduced in \sec{inner-product}.
As such it is
\[\label{discrete-heat-kernel}
K_\cp(v_a,v_b;\tau) = \bra v_a | \e^{\tau \Delta_\cp} | v_b \ket \;.  
\]
Accordingly, the heat trace on $\cp$ is given by the sum over dual vertices $\Tr_\cp := \sum_{v\in\cpsp0}$, that is
\[\label{discreteP}
\P\cp = \Tr_\cp K_\cp(v_a,v_b;\tau) = \sum_{v\in\cpsp0} \bra v | \e^{\tau \Delta_\cp} | v \ket .
\]
This defines the spectral dimension $\ds^\cp(\tau)$ of $\cp$ in terms of \eqref{ds}.
Similarly, the mean square displacement on $\cp$ is the sum
\[\label{discreteX}
\Xv\cp = \sum_{v\in\cpsp0} \dist {v}{v_0}^2 K(v,v_0;\tau)
\]
which defines the walk dimension $\dw^\cp(\tau)$ according to \eqref{dw}, possibly after an averaging over centre points $v_0$ like in in the case of Hausdorff dimension \eqref{discreteV2}.

In practice it is often useful to transform into the the basis of eigenfunctions of the Laplacian $\Delta_\cp$ to calculate $\ds^\cp$ or $\dw^\cp$.  
The matrix coefficients $ (\Delta_\cm)_{ab} = V_{ab} /(\dl_{ab}V_a)$ 
are finite and well defined if the volumes therein associated with the complex $\cm$ are finite and non-degenerate, in particular non-vanishing.
Then, the Laplacian is just a linear map in the finite vector space $\Omega^0(\cp)$ and  is diagonalized by its eigenfunctions $e^{\lambda}$ with coefficients $e^\lambda_a = \bra v_a | \lambda \ket$ corresponding to eigenvalues $\lambda$,
\begin{equation}
(-\Delta_\cp e^{\lambda})_a = -\bra v_a | \Delta_\cp | \lambda \ket = \lambda \bra v_a | \lambda \ket = \lambda e_a^{\lambda},
\end{equation}
and these form a complete orthonormal basis defining momentum space. 
In this basis the heat trace on $\cp$ simplifies to
\begin{equation}\label{ht-from-spectrum}
\P\cp = \Tr_\cp \bra \lambda' | \e^{\tau \Delta_\cp} | \lambda \ket 
= \sum_{\lambda \in \spec(\Delta_\cp)}  \e^{-\tau \lambda} 
\end{equation}
(where the label $\lambda$ is meant to run not only over eigenvalues but also over their multiplicities).
Since $\Delta_\cp$ is symmetrizable (\rem{laplacian-properties}), for real geometric volume coefficients in equation (\ref{discrete-laplacian}) the spectrum $\spec(\Delta_\cp)$ and thus the heat trace are real valued.



\renewcommand{\clp}{c}

While the spectrum of the Laplacian gives a closed expression for $\P\cp$ in many cases, for numerical computations of combinatorially very large complexes it can alternatively be treated as a random walk. 
In this case, the matrix elements $(\Delta_\cp)_{ab}$ of the Laplacian provide local probabilities for jumping from one point $v_a$ to a neighbour $v_b$.
This is the more practical concept used in dynamical triangulations \cite{\cdtfractal,Benedetti:2009bi}, random combs and multi-graphs  \cite{Durhuus:2006bd,Atkin:2011bk,Giasemidis:2012kq,Giasemidis:2012er}.


\subsection{Effective dimension definitions for discrete quantum geometries}
\label{sec:quantum}

Now the ground is laid for a definition of the dimension concepts on discrete quantum geometries.
Quite generally, by this I mean a Hilbert space $\hs$ of quantum states representing geometric degrees of freedom with an orthonormal, complete basis of states $\rjc$ which are given by a combination of algebraic data $\{j_\cl\}$ and a complex $\cp$ whose cells $\cl\in \cp$ these are based on.
For the states $\rjc$ to represent spatial geometry, the underlying complex $\cp$ is supposed to be $\sd$-dimensional (but not necessarily pure, since all relevant quantities \eqref{discreteV}, \eqref{discreteP} and \eqref{discreteX} depend only on quantities attached to the dual graph).
Obviously, the state space basis of discrete quantum gravity as introduced in chapter \ref{ch:dqg} in terms of spin-network states $\snr$ 
is an example of such states, but I prefer to keep the setting here as general as possible.

The strategy for promoting the relevant quantities to quantum observables is to use their functional dependence on local geometry, \ie the various volumes attached to cells for which a quantum version is either known or can be itself derived from known operators.
Most prominent examples are edge, area and volume operators on spin networks with the $j$'s identifying irreducible representations of $SU(2)$.
In three spacetime dimensions, the spatial ($\sd=2$) spin-network states diagonalize the length operators $\widehat{l}_e$ associated with primal edges $e=\star\eb$ that are dual to the graph edges $\eb \in\bg$
\[\label{length-spectrum}
\widehat{l}_e \snr = l(j_\eb) \snr  \propto \sqrt{ j_\eb(j_\eb+1)+ \csu} \,\snr\;,
\]
{with a free parameter $\csu\in\R$ due to a quantization ambiguity} for the Euclidean theory (as well as for timelike edges in the Lorentzian theory, spacelike ones being instead assigned a continuous positive variable) \cite{Freidel:2003kx,Achour:2014gr}.
In four spacetime dimensions ($\sd=3$), spin-network states have the same spectrum for area operators on primal faces $f=\star\eb$ dual to the graph edges $\eb \in \bg$ \cite{Rovelli:1995gq,DePietri:1996en,Ashtekar:1997bn}
\[\label{area-spectrum}
\widehat{A}_f \snr = A(j_\eb) \snr  \propto \sqrt{ j_\eb(j_\eb+1)+ \csu} \,\snr\;.
\]
In general one can obtain operators for other volumes,
in particular for the dual edge lengths $\dl$ of interest for distances \eqref{discrete-distance} and the Laplacian \eqref{discrete-laplacian}, 
using their classical functional dependence on such preferred quantities.
This can be done in different ways, depending on the geometric variables that are most convenient in the specific quantum geometric context that is chosen (appendix \ref{sec:classical-expressions}).

Generic quantum geometry states are superpositions over states on different complexes which poses  a particular challenge on defining observables, rarely addressed in the literature so far.\footnote{Recently this issue has been investigated in the group field condensate framework \cite{\gftcondensate}.}
Quantizing discrete geometric quantities results in the first place in operators $\widehat O_\cp$ acting only on states on a given complex $\cp$.
To extend these to operators $\widehat O$ on the full state space $\hs = \bigoplus_\cm \hs_\cm$ it will be necessary to use the family of orthogonal projections 
\[\label{orthogonal-projections}
\pi_\cp : \hs \lora \hs_\cm
\]
in one or another way depending on the specific observable.

\subsubsection*{Quantum Hausdorff dimension}

The quantum Hausdorff dimension relies both on distance and volume observables. 
Assume that dual edge lengths $\dl_e$ and primal $\sd$-volumes $V_{\star v}$ have both well-defined  quantum operator versions $\widehat{\dl_e}$ and $\widehat V_{\star v}$ on each space $\hs_\cp$.
Then, the distance \eqref{discrete-distance} has a straightforward quantum analogue $\widehat{D(v,v_0)}$ in terms of the quantum length operators $\widehat{\dl_e}$ because of its linearity in $\widehat{\dl_e}$.
Thus, the volume $\Vv\cp$ \eqref{discreteV} which is also linear in $V_{\star v}$ can be promoted to a quantum operator acting on states $\rjc \in \hs_\cp$ pointwise for each $r\in\R^+$,
\[\label{quantumV}
\qVv\cp \rjc  := \sum_{\qbra \widehat{D(v,v_0)} \qket_\jc < r  } \widehat V_{\star v} \rjc\,. 
\]
The extension to the averaged observable 
\[\label{quantumV2}
\qVr\cp := \frac1{\np\sd} \sum_{v_0\in\cpsp0} \qVv\cp
\]
in analogy to \eqref{discreteV} is then unproblematic due to its linear form.

For the definition of quantum Hausdorff dimension on the full state space $\hs = \bigoplus_\cm \hs_\cm$ an extension of these volume definitions using the orthogonal projections $\pi_\cp$ \eqref{orthogonal-projections} is necessary.
With the operator 
\[\label{quantumV3}
\Vr{} := \sum_\cp \pi_\cp \qVr\cp \pi_\cp \;
\]
the definition of quantum Hausdorff dimension, \ie Hausdorff dimension of a quantum state $|\qs\qket \in \hs$, is
\[\label{qdh}
\boxd{
\dh^\qs := \frac{\partial}{\partial\ln r} \ln \qbra \qVr{} \qket_\qs\;.
}
\]
It is meaningful to take the volume $V(r)$ as the primary quantity to be quantized, instead of the Hausdorff dimension \eqref{dh} itself, since the very idea of the effective dimension concept is that it is scaling of this geometric quantity.

\subsubsection*{Quantum spectral dimension}

In a similar way one can define the quantum spectral dimension, though the nonlinear dependence of the heat kernel \eqref{discrete-heat-kernel} on volumes and lengths makes it a slightly more subtle issue.
Assume that there is a well-defined quantum operator version $\widehat{(\Delta_\cm)}_{ab}$ 
on each space $\hs_\cp$ for all the coefficients of the discrete Laplacian $\Delta_\cp$ \eqref{discrete-laplacian}.
With slight abuse of notation this defines a map  $\widehat{\Delta}_\cm$ sending states in the Hilbert space $\mathcal H_\cm$ 
to a linear combination of such states with coefficients that are discrete Laplacians.%
\footnote{
Note that only the coefficients of $\widehat\Delta_\cm$ 
are quantum operators in the usual sense, \ie maps from the Hilbert space $\mathcal{H_\cm}$ to itself.
To promote $\Delta_\cm$ to a proper quantum operator it would be necessary to define it on a coupled Hilbert space of geometry and test fields.
Here it is not necessary to consider quantum states of test fields,
since the relevant object $\widehat {P(\tau)}$ to define the spectral dimension is a functional of pure geometry
and as such  it can eventually be defined as a quantum operator in the strict sense.
}

The operators $\widehat{P(\tau)}$ and $\widehat \Delta$  on the full Hilbert space $\mathcal H = \bigoplus_\cm \mathcal H_\cm$ have to be defined in terms of the family of orthogonal projections \eqref{orthogonal-projections}.
In this way, the Laplacian acting on generic quantum states of geometry in $\hs$ is the formal sum
\[\label{quantum-laplacian}
\widehat \Delta := \sum_\cm \pi_\cm \widehat \Delta_\cm \pi_\cm \;.
\]
With the appropriate notion of a trace $\overline\Tr := \sum_\cm \Tr_\cm \pi_\cm$, 
based on the trace $\Tr_\cm$ over discrete field space on $\cm$ introduced above for \eqref{discreteP}, the heat trace is defined pointwise for each $\tau\in\R^+$ as
\[\label{quantumP}
\widehat{P(\tau)} := \overline\Tr\, \e^{\tau \widehat \Delta} \;.
\]
Note that $\widehat{P(\tau)}$ is a map from $\mathcal H$ on itself, and thus a quantum operator in the strict  sense. 
Then, the spectral dimension $\ds^{\psi}(\tau)$ of a geometry state $|\psi\rangle\in\mathcal H$ is the scaling of the expectation value of  $\widehat {P(\tau)}$
\[\label{qds}
\boxd{
\ds^{\psi}(\tau) := -2\frac{\partial}{\partial\ln\tau} \ln\langle \widehat{P(\tau)}\rangle _{\psi} \;.
}
\]
Note that it depends only on pure geometry, since the relevant operators are acting on pure-geometry states.

The evaluation of the quantum spectral dimension simplifies substantially upon a choice of diagonalizing basis.
Choose now quantum labels $\{j_c\}$ such that states $\rjc$ present a diagonalizing basis of (the coefficients) of $\widehat \Delta_\cm$ for each complex $\cp$. 
A generic state $|\qs\qket \in \hs$ is then a superposition $|\qs\qket = \sum_\jc \qsc_\jc \rjc$.
In terms of this expansion the heat trace on this quantum geometry simplifies to 
\begin{eqnarray}
\label{ht-expansion}
\langle \widehat{P(\tau)}\rangle _{\qs} 
& = & \left( \sumint_{j'_c,\cm'} \qsc^*_{j'_c,\cm'} \qbra j',\cm| \right) \!\!\!
\left( \sum_\jc  \qsc_\jc \, \Tr_\cm\, \e^{\tau\widehat{\Delta}_\cm} \rjc \right) \nonumber\\
& = & \sumint_\jc \left|\qsc_\jc \right|^{2} \Tr_\cm\, \e^{\tau \qbra \widehat\Delta_\cm \qket_\jc}
\end{eqnarray}
where in the last step the eigenvector property of the basis $\rjc$ is used to simplify the exponential.
That is, with an appropriate choice of basis the expectation value of the heat trace observable on a state $\qs$ is simply a sum over the heat trace on discrete geometries like \eqref{discreteP}, weighted by the quantum state coefficients $\qsc_\jc$ in that basis.

Among the several possibilities to define a quantum version of the spectral dimension, the one chosen here \eqref{qds} is not only natural but some alternatives are simply not possible for discrete quantum geometries.
In the first place, there are various inequivalent choices which quantity in the definition of the spectral dimension to promote to an operator, apart from the usual issue of operator ordering ambiguities. 
Due to their mutual nonlinear dependence it makes a difference whether the Laplacian, the heat kernel, the heat trace or the spectral dimension itself is taken to be the primary quantity to be promoted to a quantum observable. 
The present choice of the heat trace is natural because the concept of effective dimensions is essentially based on scaling properties, here on the scaling of the heat trace \eqref{ds-scaling}.
But a quantum definition of spectral dimension based on the Laplacian, as common in a smooth setting \cite{Lauscher:2005kn,Reuter:2013ji,Calcagni:2013jx,Horava:2009ho,Benedetti:2009fo,Modesto:2009bc}, is also not possible (even though \eqref{ht-expansion} might suggest the contrary).
The reason is again the structure of state space as a direct sum over complexes.
To derive the spectral dimension it would be necessary to solve the heat kernel equation 
\begin{equation}
\partial_{\tau}K(v_a,v_b;\tau)= (\qbra \widehat{\Delta}\qket _{\qs}K)(v_a,v_b;\tau)\,.
\end{equation}
But this is a well-defined equation only if the quantum state has support on a single complex $\cp$, \ie $|\qs\qket\in\hs_\cp$.
Otherwise, the expectation value of the Laplacian \eqref{quantum-laplacian} would be a sum over discrete Laplacians for which no action on a heat kernel function is defined.
For the same reason a definition of path integral insertion analogue to \eqref{qds} is used  in CDT  \cite{Ambjorn:2005fj,Ambjorn:2005fh,Benedetti:2009bi}.


\subsubsection*{Quantum walk dimension}

The quantum definition of walk dimension is similar to the spectral dimension, summing over the product of distance and heat kernel operator instead of only the latter.
That is, under the assumptions for both quantum Hausdorff and walk dimension, one can directly define the mean square displacement operator on $\hs_\cp$ 
\[\label{quantumX}
\qXv\cp := \sum_{v\in\cpsp0} \widehat{\dist {v}{v_0}}^2 \bra v | \e^{\tau\widehat\Delta_\cp} | v_0\ket
\]
and its averaged version
\[
\widehat{\langle X^{2}\rangle}_\cp(\tau) := \sum_{v_0\in\cpsp0} \qXv\cp \;.
\]
The extension of the averaged mean square displacement operator to the whole state space $\hs$ demands again the orthogonal projections $\pi_\cp$ \eqref{orthogonal-projections} such that 
\[
\widehat{\langle X^2 \rangle} (\tau) := \sum_\cp \pi_\cp \widehat{\langle X^{2}\rangle}_\cp(\tau) \pi_\cp\;.
\]
Using this operator, the quantum walk dimension of a state $|\qs\qket \in \hs$ is defined as 
\[\label{qdw}
\boxd{
\dw^\qs(\tau):=2\left(\frac{\partial\ln \qbra \widehat{\langle X^{2}\rangle}(\tau) \qket_\qs}{\partial\ln\tau}\right)^{-1} \;.
}
\]

\

In this \sec{dimension-definition}, I have given and discussed the definition of the Hausdorff, spectral and walk dimension in the smooth, discrete and quantum cases,
thereby setting the stage for its calculation. 
Since the quantum properties of such observables depend on their interplay with classical topological and discreteness effects, in the next section \ref{sec:classical} I will first investigate these effects separately before addressing the quantum case in sections \ref{sec:coherent-states} and \ref{sec:dimensional-flow}.


%


\section{Effective dimensions of smooth and discrete geometries}\label{sec:classical}

In order to understand quantum effects in effective dimension observables of discrete quantum geometries, control over classical effects of the underlying complexes is essential.
Effective dimensions are significant observables with nontrivial properties already for classical geometries. 
There are generic properties stemming from topology such as a decrease  of dimension to zero on large scales, as well as discreteness effects on lattice scales and effects related to the combinatorial structure of complexes more in general.
Thus, it is important for interpretations of quantum effects to have a detailed understanding of these classical effect.

For this reason I present here a systematic overview of effective dimensions for relevant examples of classical smooth and discrete geometries.
I start in \sec{smooth} with smooth geometries, discussing the spectral dimension of spheres and tori.
Then, I investigate discreteness effects in all the three dimensions calculating them for regular lattices in \sec{hypercubic}.
Finally, in \sec{ds-simplicial}, I investigate effects from various classes of combinatorial structures in simplicial complexes, focusing again on the spectral dimension.


\subsection{Spectral dimension of smooth manifolds\label{sec:smooth}}


On a Riemannian manifold $(\stm,g)$, the dominant effects in the behaviour of effective dimensions at large scales are due to the topology of $\stm$ while the details are governed by the geometry $[g]$.
I will exemplify this calculating the spectral dimension of the circle $S^{1}$ and the generalization to the $\m$-torus $T^{\m}$ and the $\m$-sphere $S^{\m}$.
For $T^{\m}$ I even find analytic solutions in closed form for the heat trace.
Similar effects are also expected in the Hausdorff dimension.

Qualitatively, a compact topology leads to a fall-off of the spectral dimension $\ds(\tau)$ to zero.
The (diffusion) scale $\tau$ where this happens is related to the geometry (\ie the curvature radii). 
In the light of the heat trace's random walk interpretation, this can be easily understood: the return probability approaches one after diffusion times $\tau$ larger than the circumferences (closed geodesics).
The resolution interpretation is also easy to spell out: when the resolution is lowered to the degree where different points cannot be distinguished at all due to the limited compact geometry of the set, the particle appears trivially localized in the same region, \ie on the whole manifold.

\

The spectral dimension of the circle $S^1$ has an analytic solution.
The spectrum of the Laplacian $-(k/R)^{2}$, $k\in\Z$, depends on the circle's radius $R$, yielding a heat trace
\begin{equation}
\P{S^{1}}=\underset{k\in\Z}{\sum}\rme^{-\left(\frac{k}{R}\right)^{2}\tau}=\theta_{3}\left(0,\,\rme^{-\frac{1}{R^2}\tau}\right)=\theta_{3}\left(0\Big|\,\frac{1}{R^2}\frac{\rmi\tau}{\pi}\right)
\end{equation}
in terms of $\theta_3$, the third theta function%
\footnote{Here and in the following the notation for theta functions is as defined in \cite{NIST}.
}.
While $\ds(\tau)=\m=1$ for $\tau<R^{-2}$, due to the periodicity there is a decrease  to zero for $\tau>R^{-2}$ (\fig{dSsmooth}).
\begin{figure}
\centering
\includegraphics[width=7.5cm]{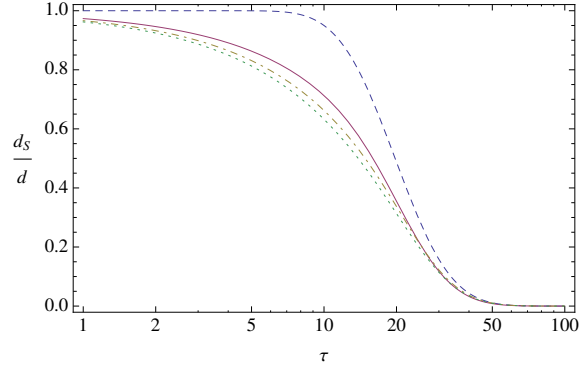}
\caption{Spectral dimension of $S^\m$ rescaled to volume $V_{S^\m}=1$,
for comparison divided by $\m$, for $\m=1,2,3,4$ ($S^1$ dashed line, $S^2$ solid, $S^3$ dot-dashed and $S^4$ dotted). The inverse square of the corresponding radii sets the scale of the topological effect. 
At the same time, it is a comparison with the case of $T^\m$ with radii $R=1/(2\pi)$ as these tori are equivalent to the case of $S^{1}$.
\label{fig:dSsmooth}}
\end{figure}

The geometry defined by the constant curvature radius $R$ governs only the scale at which the topological effect takes place. In that sense, the fall-off of the spectral dimension  can be understood as due to the combination of geometry and topology.

\subsubsection*{The torus \texorpdfstring{$T^\m$}{}}

The torus $T^\m=\left(S^{1}\right)^{\times \m}$ with radii $R_{i}$ generalizes the case of the circle straightforwardly with spectrum
$-\sum_{i=1}^\m k_i/R_i$, 
 $\vec{k}\in\Z^\m$, such that
\begin{equation}
\P{T^\m}=\underset{\vec{k}\in\Z^\m}{\sum}\rme^{-\sum\left(\frac{k_{i}}{R_{i}}\right)^{2}\tau}=\theta\left(0 \mid \frac{\rmi\tau}{\pi} \begin{pmatrix}
R_{1}^{-2} &&\\
 & \ddots & \\
 &  & R_{d}^{-2}
\end{pmatrix}\right),
\end{equation}
where $\theta$ is the (multi-dimensional) Riemann $\theta$-function \cite{NIST}.
In the case where all radii are equal the spectral dimension turns out to be just $\m$ times the spectral dimension of $T^1 = S^1$, since 
\begin{equation}\label{torus-circle-relation}
\P{T^\m}\propto\left[\P{S^{1}}\right]^\m.
\end{equation}
In general, for different constant curvature radii $R_{i}$, the geometry affects not only the scale at which the decay starts, but also accounts for various geometric regimes (\fig{dsT2}). 
If the radii are ordered as $R_1\geqslant R_2\geqslant\dots\geqslant R_\m$, the spectral dimension is constantly the topological dimension for $\tau<1/R_1^2$ and zero for $\tau>1/R_\m^2$. 
In intermediate regimes, it can have plateaux at heights corresponding to intermediate integer dimensions, if the radii are of sufficiently different order of magnitude. This can be easily understood: When $k$ radii are much smaller than a given scale, the $\m$-torus effectively appears as a $(\m-k)$-torus for the diffusion process at that scale. 
\begin{figure}
\centering
\includegraphics[width=7.5cm]{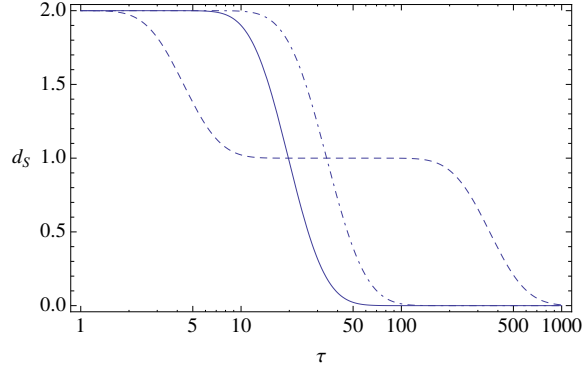}
\caption{Spectral dimension of $T^2$ with various geometries, $(R_1,R_2)=(1,1)$, $(R_1,R_2)=(1,\sqrt{3}/2)$, $(R_1,R_2)=(3,1/3)$ (solid, dash-dotted and dashed curve)}
\label{fig:dsT2}
\end{figure}

\subsubsection*{The sphere \texorpdfstring{$S^\m$}{}}

The spectrum of the Laplacian acting on functions on the sphere $S^{2}\cong SU(2)/U(1)$ of radius $R$ is proportional to the $SU(2)$ Casimir, $-C(j)/R^{2}=-j(j+1)/R^{2}$, for $j\in\N$ with $(2j+1)$-fold degeneracy, leading to a heat trace \cite{Benedetti:2009fo}
\begin{equation}
\P{S^{2}}=\underset{j=0}{\overset{\infty}{\sum}}(2j+1)\,\rme^{-\frac{j(j+1)}{R^{2}}\tau}.
\end{equation}
More generally, for the $\m$-sphere $S^\m$, $\m\geqslant 1$, the spectrum
is given by the eigenvalues $-j(j+d-1)/R^{2}$ with multiplicities such
that
\begin{equation}
\P{S^\m}=1+\underset{j=1}{\overset{\infty}{\sum}}\left[\binom{d+j}{d}-\binom{d+j-2}{d}\right]\,\rme^{-\frac{j(j+d-1)}{R^{2}}\tau}\,.
\end{equation}
The spectral dimension of $S^\m$ has a slower fall-off in comparison with the torus $T^\m$, with $\m\geqslant 2$.
Thus, $\m$-spheres are examples where the topological dimension is only obtained in the limit $\tau \ra 0$.
In contrast with the torus, spheres do not exhibit any multi-scale behaviour since they are governed by only one geometric parameter.


Here I have discussed the generic effects of spherical and toroidal topology and of the geometric curvature parameters. All compact topologies have a spectral dimension going to zero at scales of the order of the curvature radii. 
Below this scale, $\ds$ agrees with the topological dimension in the case of the tori, while for spheres such an accordance holds only in the small-scale (or infinite resolution) limit $\tau\to 0$ (\fig{dSsmooth}). 
Furthermore, in general an $\m$-torus is given by $\m$ geometric parameters setting the intermediate scales at which the torus has an effective lower-dimensional behaviour (\fig{dsT2}).


\subsection{Effective dimensions of hypercubic lattices\label{sec:hypercubic}}

Now I consider classical effects of manifold discreteness in the effective dimensions. 
To this aim, I calculate the spectral, walk and Hausdorff dimension of regular lattices in this section.
It turns out that there is a generic effect of a fall-off to zero of the spectral dimension below the lattice scale while the precise form of a peak larger than the topological dimension at that scale depends on the structure of the lattice.
On the other hand, the walk dimension does not show any discreteness effects.
The Hausdorff dimension has a fall-off to dimension one.

\subsubsection*{Spectral dimension on hypercubic lattices\label{sub:Hypercubic}}

A simple, purely combinatorial case to start with are finite and infinite hypercubic lattices.
I will denote the finite $\m$-dimensional hypercubic lattice of finite size $\size$ in each direction with torus topology as $\cp=\hl\m$. 
Its vertex is equivalent to $\m$-tupels of integers, $\hl\m^{[0]}\cong(\Z_\size)^\m$, which I will use to label the vertices.
The infinite hypercubic lattice is then the large-size limit $\hlinf\m$.

The spectral dimension of hypercubic lattices can be derived directly from the spectrum of the Laplacian.
In the one-dimensional case $\hl1$, also known as the cycle graph with $\size$ vertices, the eigenvalues of the Laplacian are \cite{Chung:1997tk}
\begin{equation}\label{distor}
\lambda_{k}=1-\cos\left(\frac{2\pi k}{\size}\right)=2\sin^{2}\left(\frac{\pi k}{\size}\right),\qquad k=1,\dots,\size\,,
\end{equation}
such that
\begin{equation}\label{ht-hl}
\P{\hl1} = \rme^{-\tau}\sum_{k=1}^{\size}\rme^{\tau\cos\left(\frac{2\pi k}{\size}\right)}
\end{equation}
and 
\begin{equation}
\ds^{\hl1}(\tau)=2\tau\left[1-\frac{\sum_k\cos\left(\frac{2\pi k}{\size}\right)\,\rme^{\tau\cos\left(\frac{2\pi k}{\size}\right)}}{\sum_k \rme^{2\tau\cos\left(\frac{2\pi k}{\size}\right)}}\right].
\end{equation}
Using trigonometric relations, a given sum over exponentials of cosines can be further simplified to a sum over hyperbolic cosines; \eg for $\size=8$
\begin{equation}
d_{S}^{\hln1 8}(\tau)=2\tau\left[1-\frac{\sinh\left(2\tau\right)+\sqrt{2}\sinh\left(\sqrt{2}\tau\right)}{1+\cosh\left(2\tau\right)+2\cosh\left(\sqrt{2}\tau\right)}\right].
\end{equation}
Like in the case of smooth tori \eqref{torus-circle-relation}, the heat trace of the $\m$-dimensional lattice $\hl\m$ is simply related to the one-dimensional case as
\begin{equation}
\P{\hl\m} \propto \left( \P{\hl1} \right)^\m \;,
\end{equation}
yielding just a pre-factor of $\m$ for the spectral dimension:
\begin{equation}
\ds^{\hl\m}(\tau)=\m\, \ds^{\hl1}(\tau).
\end{equation}

The salient property of this lattice spectral dimension is a peak at the lattice scale $a$ together with a fall-off to zero below that scale (\fig{dsZn}).
The fact that the dimension is zero below the lattice scale can be easily understood: in the random walk picture, the return probability is constantly zero for diffusion times up to the time needed to walk to a nearest neighbour and back; in the resolution interpretation, if the resolution is sharper then the size of a single cell, a particle cannot be localized at all.
For the peak at the lattice scale with a value of $\ds^{\textrm{max}} \approx 1.22\m$ at $\tau\approx1.70$ a similarly explanation is lacking. For this reason it is considered as a discreteness artefact.

The topological dimension $\m$ is obtained as a plateau for lattices of sufficiently large size $\N$.
Due to the torus topology the finite-lattice spectral dimension has exactly the same decrease  at the scale set by the curvature radii which are in this case directly proportional to $\size$. 
Thus, $\ds^{\hl\m}(\tau) = \m $ only for $1 \ll \tau \ll \size^2$ (\fig{dsZn}).

\begin{figure}
\centering
\includegraphics[width=7.5cm]{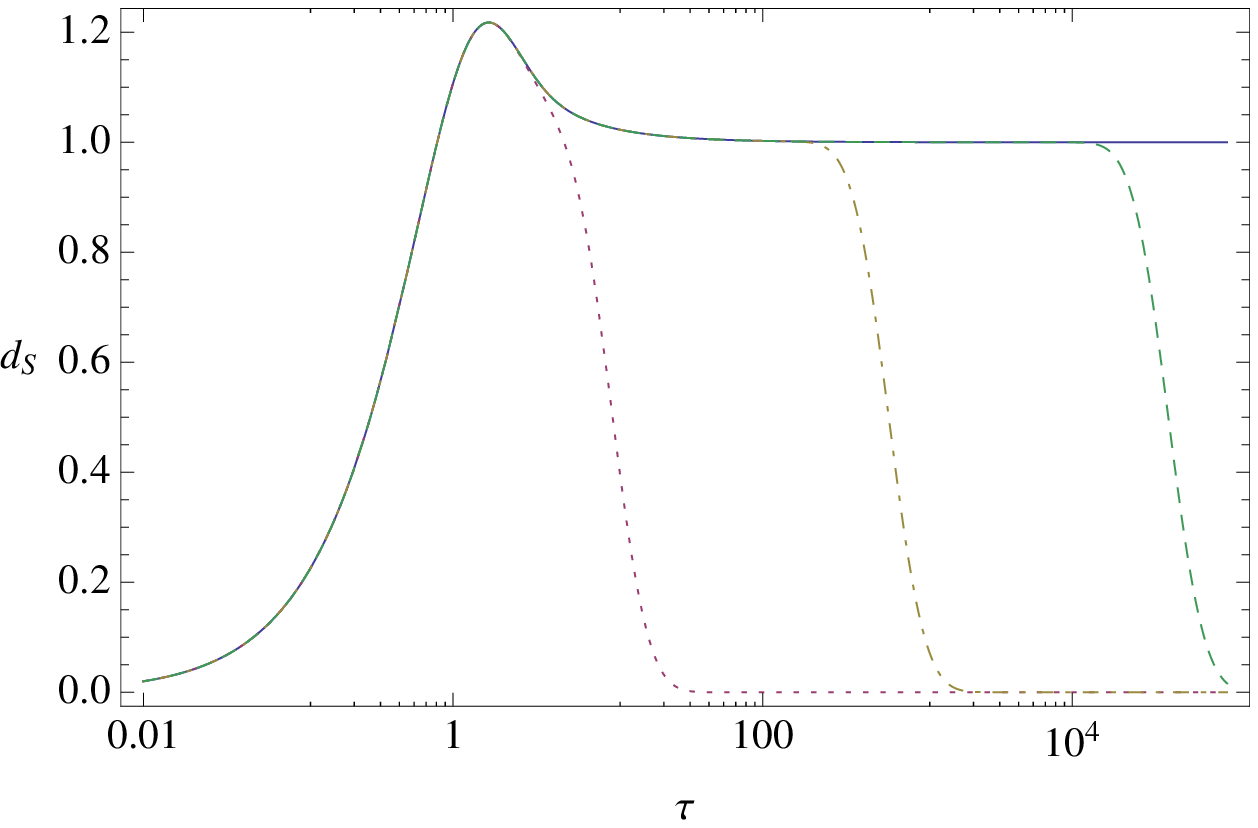}
\caption{The spectral dimension of the infinite one-dimensional lattice $\hlinf1$, compared with the finite $\hln1{8}$ (dotted curve), $\hln{1}{64}$ (dot-dashed) and $\hln{1}{512}$ (dashed).}
\label{fig:dsZn}
\end{figure}

\

These numerical results become more precise in the large-$\size$ limit where I can derive analytic solutions for the spectral dimension.
For the infinite hypercubic lattice $\cp = \hlinf\m$ it can be obtained in the following way.
In the limit $\size\ra\infty$, the discrete spectrum (\ref{distor}) extends to a spectrum parametrized by a continuous parameter $p\in[0,\pi[\In\R$ in the same interval,
\begin{equation}
\lambda_k = 2\sin^{2}\left(\frac{\pi k}{n}\right)\to \Delta(p) = 2\sin^{2}p\,.
\end{equation}
From this spectrum the heat trace can be calculated analytically as
\begin{equation}
\P{\hlinf1} = \Tr_{[0,\pi[}\,\rme^{-\tau2\sin^{2}p}
= \int_{0}^{\pi}\d p\,\rme^{-\tau2\sin^{2}p}=\rme^{-\tau}I_{0}(\tau)\,
\end{equation}
where $I_k$ is the modified Bessel function of the $k$'th kind.
Introducing a lattice spacing $a$ (compare with lattice gauge theory,
\eg \cite{Wiese:2009ub}) compatible with this spectrum,
\begin{equation}
\Delta(p)=2\left(\frac{2}{a}\sin\frac{ap}{2}\right)^{2}\overset{a=2}{=}2\sin^{2}p\,,
\end{equation}
where now $p\in[0,\frac{2\pi}{a}[$, this can even be connected to the continuous case
\begin{equation}
2\left(\frac{2}{a}\sin\frac{ap}{2}\right)^{2}\underset{a\ra0}{\ra}2p^{2}.
\end{equation}

The straightforward extension to $\m$ dimensions is given by
\begin{equation}
\Delta(p)=\underset{i=1}{\overset{d}{\sum}}2\left(\frac{2}{a}\sin\frac{ap_{i}}{2}\right)^{2}\overset{a=2}{=}\underset{i=1}{\overset{d}{\sum}}2\sin^{2}p_{i}
\end{equation}
 such that there is again the same simple relation to the one-dimensional case
\ba
\P{\hlinf\m} & = &  \Tr\; \e^{-\tau \sum_{i=1}^{\m} 2\sin^{2}p_{i}}
= \prod_{i=1}^{\m} \left(\int_{0}^{\pi} \d p_{i}\right)  \prod_{i=1}^{\m} \e^{-\frac{\tau}{\m}2\sin^{2}p_{i}} \\
\label{ht-hlinf}
 & =  &  \left(\e^{-\tau}I_{0}({\tau})\right)^\m 
 = \left( \P{\hlinf1}\right)^\m \;.
\ea
The spectral dimension has then the simple analytic solution 
\begin{equation}\label{ds-hlinf}
\ds^{\hlinf\m}(\tau) = \m\,\ds^{\hlinf1}(\tau) = 2 \m \tau\left[1-\frac{I_{1}(\tau)}{I_{0}(\tau)}\right]\;. 
\end{equation}

Indeed, this analytic solution to the infinite-lattice spectral dimension perfectly agrees with the results for the finite lattices both at the discreteness scale and on the scale of the topological-dimension plateau (\fig{dsZn}).
In particular it confirms the shape of the discreteness artefact at lattice scale.
As expected, the only difference to finite lattices is that the plateau in the infinite-lattice spectral dimension extends to infinity.
This can be shown analytically in terms of the large-$\tau$ asymptotic behaviour of the modified Bessel function $I_0$, that is \cite{Ambjorn:1999in}
\[
\P{\hlinf1} = \rme^{-\tau} I_{0}(\tau) \sim \sqrt\tau^{-1}[1+\mathcal{O}(1/\tau)].
\]
This shows that the infinite-lattice solution \eqref{ht-hlinf} really provides the large-$\size$ limit to the finite case \eqref{ht-hl}.

\subsubsection*{Spectral dimension on other lattices} 

The essential properties of the lattice spectral dimension are independent of the specific kind of the lattice. 
To demonstrate this I present analytic results of $\ds^\cp$ also for the hexagonal lattice $\cp = \hxinf2$ and for a combinatorial octagonal  lattice $\cp = \olinf2$.
The necessary technique to derive such analytic solutions for various lattice combinatorics is provided by generating functions \cite{Wilf:1994to}.

To explain this technique I will start again with the infinite hypercubic lattice, thereby providing an alternative to the above derivation of the solution \eqref{ht-hlinf} and \eqref{ds-hlinf} in terms of the large-$N$ limit of the finite-$N$ spectrum of the Laplacian.
The starting point is now the diffusion equation \eqref{diffusion-equation} with the discrete Laplacian \eqref{discrete-laplacian}. 
For any lattice $\cp$ one can use translational symmetry to simplify the heat trace \eqref{discreteP} to
\[
\P{\cp} = \Tr_\cp K_\cp(v_a,v_b;\tau) = \sum_{v\in\cpsp0} K_\cp(v,v;\tau) = K_\cp(v_0,v_0;\tau)
\]
for an arbitrary chosen lattice origin $v_0\in\cpsp0$. 
Thus it is sufficient to consider the components 
\[
\phi_{j}(\tau) := K(v_j,v_0;\tau)
\]
of the heat kernel.

The general strategy is then first to simplify the equation, redefining the diffusion parameter $\ttau := \cdiff\tau$ to absorbe the effective diffusion constant 
\[\label{diffusion-constant}
\cdiff := 2 (\Delta_\cp)_{ab} = 2 \frac{V_{ab}}{V_a \dl_{ab}}
\]
given by the constant coefficients $(\Delta_\cp)_{ab}$ of the Laplacian \eqref{discrete-laplacian}.
Then, the diffusion has the form
\begin{equation}\label{step1}
\left(\partial_{\ttau} + \frac{k}{2}\right)\phi_{j}(\ttau) = \frac{1}{2} \sum_{l\sim j} \phi_{l}(\ttau)
\end{equation}
where $k$ is the valency of the vertices and the sum runs over nearest neighbours.
In the second step, the diagonal part depending on the valency $k$ can be absorbed by redefining the heat kernel function
\[
\widetilde{\phi}_j(\ttau)=e^{k\ttau/2}\phi_j(\ttau)
\]
such that the diffusion equation simplifies to a pure sum over nearest neighbours
\[\label{step2}
\partial_{\ttau}\tphi_{j}(\ttau) = \sum_{l\sim j} \tphi_{l}(\ttau)\;.
\]
The solution to such equations can be obtained in terms of generating functions which I will explain now in the case of the hypercubic-lattice example.

For the one-dimensional lattice $\hlinf1$ with lattice spacing $a$ the diffusion equation
\begin{equation}\label{lattice1d}
\partial_{\tau}\phi_{j}(\tau) = \left(\Delta\phi\right)_{j}(\tau) 
= -\frac{1}{a}\sum_{+/-}\frac{1}{a}\left[ \phi_{j}(\tau) - \phi_{j\pm1}(\tau) \right]
= \frac{1}{a^{2}} \left[ -2\phi_{j}(\tau) + \phi_{j+1}(\tau) + \phi_{j-1} (\tau) \right]
\end{equation}
simplifies upon the redefinition with $\cdiff=2/a^{2}$ and valency $k=2$, in this case, to
%
\[\label{1dlattice}
\partial_{\ttau}\tphi_{j}(\ttau) = \tphi_{j+1}(\ttau) + \tphi_{j-1}(\ttau)\;.
\]
The solution of this equation can be obtained from the corresponding generating function. 
In the present case of the lattice $\hlinf1$, one can easily check that solutions to the diffusion equation \eqref{1dlattice} are generated by the function
\[\label{generating-1dlattice}
G(u;\ttau)=e^{\frac{1}{2}\left(u+\frac{1}{u}\right)\ttau}=\sum_{j\in\Z}u^{j}\tphi_{j}(\ttau)\;.
\]
This is the generating function of the modified Bessel function $\tphi_{j}(\ttau)=I_{j}(\ttau)$.
The heat kernel solution is therefore
\begin{equation}
K_{\hlinf1}(v_j,v_0;\tau)\equiv \phi_{j}(\tau)=e^{-2\tau/a^{2}}I_{j}({2\tau}/{a^{2}})
\end{equation}
which shows that the precise lattice spacing implicitly assumed before in \eqref{ht-hlinf} is $a=\sqrt2$.


The generating function technique allows now to easily identify solutions to the diffusion equation for lattices with other combinatorics. 
The generalization to higher dimensional hypercubic lattices $\hlinf\m$ with vertex valency $k=2\m$, for example, is encoded in the generating function 
\[
G(u_1,...,u_\m;\ttau) = e^{\frac{1}{2}\sum_{i=1}^{\m} \left(u_{i}+\frac{1}{u_{i}}\right)\ttau}
= \sum_{\vec\j \in\Z^\m}\left(\prod_{i=1}^\m u_{i}^{j_i} \right)\tphi_{\vec\j}(\ttau)
\]
where vertices are labelled by $\vec\j = (j_{1}\dots j_{\m})\in\Z^\m$.
The effective diffusion constant $\cdiff=2/a^\m \cdot a^{\m-1}/a=2/a^{2}$ is the same as for $\m=1$.
Since the generating function itself factorizes, it is meaningful to make the factorizing ansatz also for its expansion, $\tphi_{\vec\j}(\ttau)=\prod_{i=1}^\m\tphi_{j_{i}}^{(i)}(\ttau)$.
Then, the generating function  
\[\label{generating-hypercubic}
G(u_1,...,u_\m;\ttau) = \prod_{i=1}^\m G(u_i;\ttau) 
= \prod_{i=1}^\m e^{\frac{1}{2}\left(u_{i}+\frac{1}{u_{i}}\right)\ttau}=\prod_{i=1}^\m\left(\sum_{j_i\in\Z}u_{i}^{j_i}\tphi_{j_{i}}^{(i)}(\ttau)\right)
\]
is just a multiple of the above generating function $G(u;\ttau)$, \eqref{generating-1dlattice}.
Thus, the heat kernel solution is 
\begin{equation}\label{hk-hlinf}
K_{\hlinf\m}(v_{\vec\j},v_0;\tau) \equiv \phi_{\vec\j}(\tau) = \e^{-2\m\tau/a^2}\prod_{i=1}^\m I_{j_{i}}({2\tau}/{a^{2}})\;,
\end{equation}
in agreement with \eqref{ht-hlinf}.


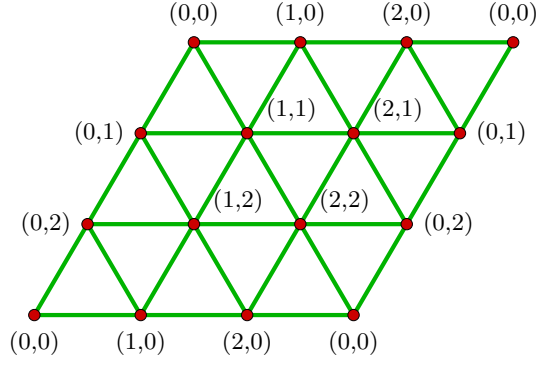
\begin{figure}
\centering
\begin{tikzpicture}[scale=0.7 ]
\draw [eb] (0,0) -- (6,0);
\draw [eb] (1,1.722) -- (7,1.722); 
\draw [eb] (2,3.444) -- (8,3.444); 
\draw [eb] (3,5.166) -- (9,5.166); 
\draw [eb] (0,0) -- (3,5.166); 
\draw [eb] (2,0) -- (5,5.166);
\draw [eb] (4,0) -- (7,5.166);
\draw [eb] (6,0) -- (9,5.166); 
\draw [eb] (2,0) -- (1,1.722); 
\draw [eb] (4,0) -- (2,3.444);
\draw [eb] (6,0) -- (3,5.166); 
\draw [eb] (7,1.722) -- (5,5.166);
\draw [eb] (8,3.444) -- (7,5.166);
\node at (0,0) [vb, label=below:{\footnotesize(0,0)}]{}; 
\node at (2,0) [vb, label=below:{\footnotesize(1,0)}]{}; 
\node at (4,0) [vb, label=below:{\footnotesize(2,0)}]{}; 
\node at (6,0) [vb, label=below:{\footnotesize(0,0)}]{}; 
\node at (1,1.722) [vb, label=left:{\footnotesize(0,2)}]{};  
\node at (3,1.722) [vb, label={[label distance=1pt]5:\footnotesize(1,2)}]{};  
\node at (5,1.722) [vb, label={[label distance=1pt]5:\footnotesize(2,2)}]{};  
\node at (7,1.722) [vb, label=right:{\footnotesize(0,2)}]{}; 
\node at (2,3.444) [vb, label=left:{\footnotesize(0,1)}]{}; 
\node at (4,3.444) [vb, label={[label distance=1pt]5:\footnotesize(1,1)}]{}; 
\node at (6,3.444) [vb, label={[label distance=1pt]5:\footnotesize(2,1)}]{}; 
\node at (8,3.444) [vb, label=right:{\footnotesize(0,1)}]{}; 
\node at (3,5.166) [vb, label=above:{\footnotesize(0,0)}]{}; 
\node at (5,5.166) [vb, label=above:{\footnotesize(1,0)}]{}; 
\node at (7,5.166) [vb, label=above:{\footnotesize(2,0)}]{}; 
\node at (9,5.166) [vb, label=above:{\footnotesize(0,0)}]{};
\end{tikzpicture}
\caption{Smallest abstract simplicial complex with a realization as an equilateraleral triangulation of the torus $T^{2}$, given by 
$\size^2 = 3^2$ vertices.
\label{fig:T2triang}}
\end{figure}

Along the same lines one can derive the heat kernel on a 2-dimensional hexagonal lattice $\cp = \hxinf2$.
The heat kernel $K_{\hxinf2}(v_{ij},v_{00};\tau)$ is based on vertices $v_{ij}\in\cpsp0$ of the dual complex which is triangular and has valency $k=6$ (\fig{T2triang}).
The vertex set $\cpsp0$ is still equivalent to $\Z^2$ such that vertices can be labeled by a pair of integers $i,j$.
In terms of the hexagonal lattice spacing $a$ the geometric factor is 
$\cdiff = \frac2{3\sqrt3 a/2} \frac a {\sqrt3 a} = \frac4{9a^2}$.
The generating function has the form 
\[
G(u,v;\ttau)= \e^{\frac{1}{2}\left(u+\frac{1}{u}+v+\frac{1}{v}+uv+\frac{1}{uv}\right)\ttau}
= \sum_{i,j\in\Z}u^{i}v^{j}\tphi_{ij}(\ttau)
\]
since on the dual triangular lattice vertices $v_{ij}$ are neighbour to vertices $v_{i+1,j+1}$ on top of the incidence relation of the quandrangular lattice.
This function is known to be the generating function of the two-index modified Bessel function \cite{Dattoli:1994iu}
\[
\tphi_{ij}(\ttau)=I_{ij}^{(+)}(\ttau) :=\sum_{s\in\Z}I_{i-s}(\ttau)I_{j-s}(\ttau)I_{s}(\ttau)\;.
\]
The hexagonal heat kernel solution is therefore
\[
K_{\hxinf2}(v_{ij},v_{00};\tau) = \e^{-k\,\cdiff \tau/ 2} \tphi_{ij}(\cdiff \tau)  = \e^{-4\tau/3a^{2}}I_{ij}^{(+)}({4\tau}/{9a^{2}})\;.
\]


Finally, there is a similar solution for the combinatorial lattice $\olinf2$ defined by the property that each face has eight neighbouring faces, thus coined orthogonal lattice.
It is merely combinatorial since there is no tilling of the plane with such a combinatorial pattern. 
Therefore, the geometric factor $\cdiff$ has no derivation from the Euclidean geometry of the plane as a function of lattice spacing $a$ like in the previous cases.

The dual of $\olinf2$ is a quadrangular lattice where vertices are furthermore adjacent to diagonal neighbours, \ie a vertex $v_{ij}$ is also adjacent to $v_{i+1,j+1}$ and $v_{i-1,j-1}$.
Their valency is thus $k=8$ and the heat kernel combinatorics are captured by the generating function
\[
G(u,v;\ttau)=e^{\frac{1}{2}\left(u+\frac{1}{u}+v+\frac{1}{v}+uv+\frac{1}{uv}+\frac{u}{v}+\frac{v}{u}\right)\ttau}=\sum_{i,j\in\Z}u^{i}v^{j}\tphi_{ij}(\ttau)\;.
\]
This is the generating function of another kind generalization of modified Bessel function \cite{Dattoli:1994iu}
\[
\tphi_{ij}(\tau)=I_{ij}^{(+,-)}(\tau) 
:=\sum_{s\in Z}I_{i-s,j-s}^{(+)}(\tau)I_{s}(\tau) 
\equiv\sum_{s,t\in Z}I_{i-t-s}(\tau)I_{j-t-s}(\tau)I_{t}(\tau)I_{s}(\tau)
\]
such that the heat kernel on the octagonal lattice $\olinf2$ is
\[
K_{\olinf2}(v_{ij},v_{00};\tau) 
 = \e^{-8 \cdiff \tau} I_{ij}^{(+,-)}(\cdiff \tau)\;.
\]

\begin{figure}[htb]
\centering
\includegraphics[width=7.5cm]{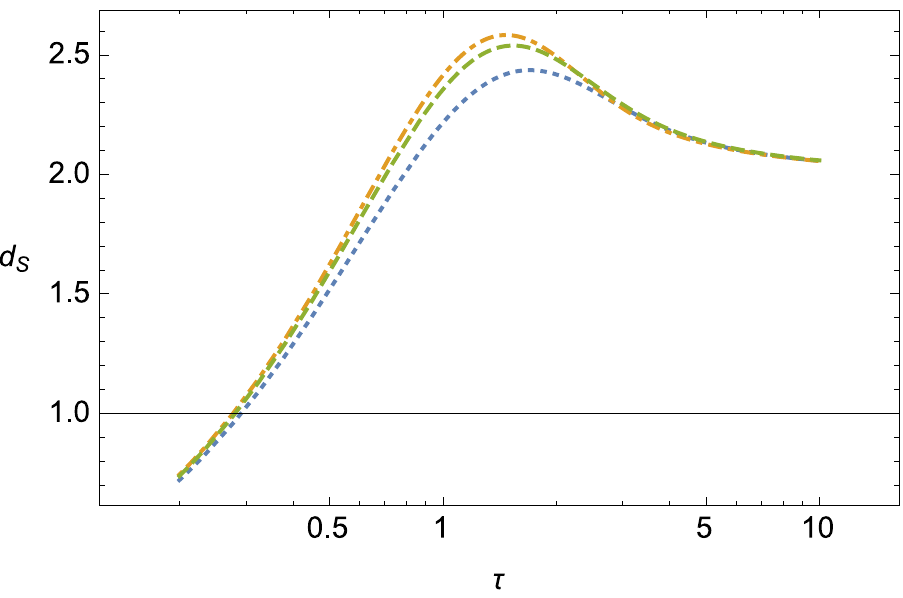}
\caption{Comparison of the spectral dimension of different 2-dimensional lattices: Hypercubic $\hlinf2$ (dotted curve), hexagonal $\hxinf2$ (dashed), and octagonal $\olinf2$ (dot-dashed, with $\cdiff=1$).
For comparison, all lattice spacings on the dual lattices are equally set to $a^\star = \sqrt{2}$.
}
\label{fig:ds-lattices}
\end{figure}

Now, the relevance of these calculations is that the spectral dimension of all the lattices considered differs only marginally. 
The only difference is the specific form of the discreteness artefact at the lattice scale, as shown in \fig{ds-lattices}. 
In all other regimes, the spectral dimension of different lattices perfectly agrees.
These results show that the specific combinatorics of a regular lattice, \ie a complex with translational symmetry, are not relevant for such global observables like the spectral dimension.
One is thus free to choose the lattice which fits best the practical needs in a specific context.

\subsubsection*{Walk dimension on the lattice}

Given the heat kernel \eqref{hk-hlinf} on the hypercubic lattice $\hlinf\m$, it is straightforward to calculate the walk dimension. 
Again, translational symmetry reduces the averaged mean square displacement $\X{\hlinf\m}
$ to the mean square displacement $\Xv{\hlinf\m}$ at the arbitrarily chosen origin $v_0$.
This is
\begin{eqnarray}
\X{\hlinf\m} = \Xv{\hlinf\m}
& = & \sum_{\vec\j\in\Z^\m} \dist{v_{\vec\j}}{v_0}^{2}  K_{\hlinf\m}(v_{\vec\j},v_0;\tau) \\
& = & \sum_{\vec\j\in\Z^\m} \left(\sum_{i=1}^\m j_i^{2}\right) \e^{-\tau}\prod_{i=1}^\m I_{j_i}(\tau)\,.
\end{eqnarray}
This can be evaluated using standard relations of the Bessel functions $I_j(\tau)$:
\begin{eqnarray}
\X{\hlinf\m} 
& \propto & \e^{-\m\tau}\sum_{i=1}^\m \sum_{\vec\j\in\Z^\m} j_i^{2}\prod_{i=1}^\m I_{j_i}(\tau)\\
 & = & \e^{-\m\tau}\sum_{i=1}^\m \left[\sum_{j_i\in\Z} j_i^{2}I_{j_i}(\tau)\right] \prod_{l\ne i}^\m\left[\sum_{j_l\in\Z}I_{j_l}(\tau)\right] \\
 & = & \e^{-\m\tau} \m \left[\frac{\tau}{2}\sum_{j\in\Z}I_{j-1}(\tau)+I_{j+1}(\tau)\right]\left(\e^{\tau}\right)^{\m-1} \\
 & = & \m\,\tau\,. \label{Xhlinf}
\end{eqnarray}
Thus, the walk dimension on the lattice is
\begin{eqnarray}
\dw^{{\hlinf\m}}(\tau)=2\,,
\end{eqnarray}
which is the same as in the continuum, \eg on flat space.
The result is thus that there are no discreteness effects in the walk dimension.

%
%

\subsubsection*{Hausdorff dimension on the lattice}

For the calculation of Hausdorff dimension on an equilateral lattice $\cp$ it is convenient to split the volume into a sum of shells 
\[
S_{v_0,\cp}(r) := \sum_{\dist v{v_0} = r} V_{\star v}
\]
over points $v\in\cpsp0$ at a given distance $r$ to the centre $v_0$.
In units of the fixed lattice spacing $a^\star$ of the dual complex, the sum over shells is indeed a discrete sum
\[
\Vv\cp = 
\sum_{j=0}^{\lfloor {r/a^\star} \rfloor } S_{v_0,\cp}(a^\star j)\;
\]
where the floor function $\lfloor\cdot\rfloor$ is needed to allow arbitrary real radius distance $r$.
%
%
%
%

For the $\m$-dimensional hypercubic lattice $\hlinf\m$ the $\m$-volumes are just $V_{\star v} = a^\m$ and $a^\star = a$ since $\hlinf\m$ is dual to itself.
Due to translational symmetry averaging over centre points is again redundant, $\Vr\cp = \Vv\cp$.
Now, the part of a shell at radius distance $r$ belonging to an $\m$-dimensional hyper-quadrant equals the volume of a $(\m-1)$-dimensional hyper-quadrant such that 
\[
\frac{1}{(2a)^\m }S_{\hlinf\m}(r)=\frac{1}{(2a)^{\m-1}}\Vr{\hlinf{\m-1}} \;.
\]
This can be used recursively to find
\[\label{latticeV}
V_{\hlinf\m}(r)= \frac{(2a)^\m}{\m} \frac r a \binom{\frac r a + \m-1}{\m} 
\]
which yields  for the lattice Hausdorff dimension the closed expression
\[\label{dh-lattice}
\dh^{\hlinf\m}(r) = r \sum_{l = 0}^{\m-1}\frac{1}{r + l a} 
= \frac r a \left[\digamma(r/a + d) - \digamma(r /a )\right]
\]
 where $\psi$ is the digamma function here. 
Both in \eqref{latticeV} and \eqref{dh-lattice} I have generalized from the exact discrete dependence on $\lfloor r/a \rfloor$ to the continuous $r/a$ 
which is necessary for an appropriate notion Hausdorff dimension as scaling of volume.

The Hausdorff dimension of $\m$-dimensional hypercubic lattices is characterized by a transition from topological dimension $\m$ at larger scales towards a constant value of one below the lattice scale.
The regime of topological dimension $\m$ is found for
 \[
\dh^{\hlinf\m}(r) \us{r/a \ra \infty}{\lora} \sum_{l=0}^{\m-1} 1 = \m \;.
 \]
On small scales, on the other hand,
 \[
\dh^{\hlinf\m}(r)  = 1 + \sum_{l=1}^{\m-1}\frac{r}{r + la} \us{r/a \ra 0}{\lora}  1.
  \]
The Hausdorff dimensions for lattices in various dimensions are shown in \fig{dh-lattice}.
\begin{figure}
\begin{centering}
\includegraphics[width=7.5cm]{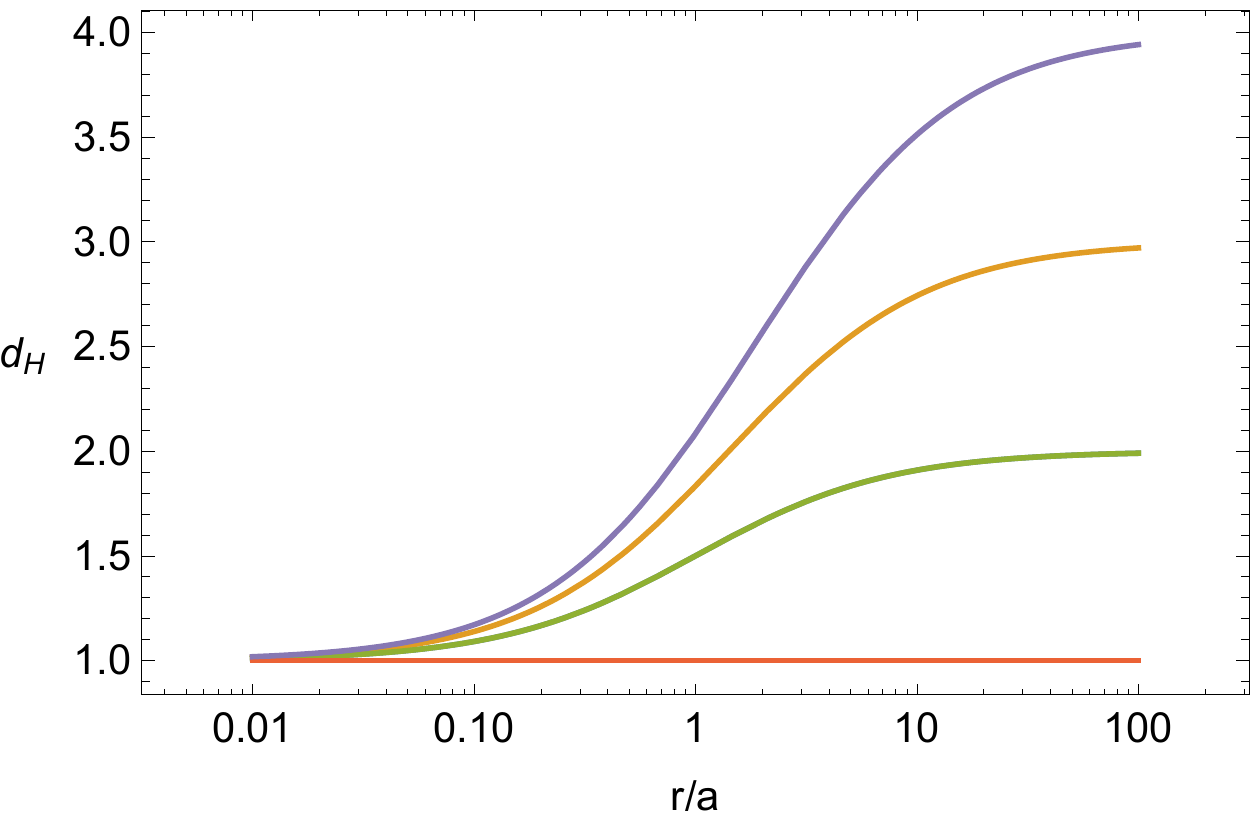}
\caption{Hausdorff dimension $\dh^{\hlinf\m}(r)$ of hypercubic lattices in dimension $\m=1,2,3,4$, easily distinguished by the large-scale regime where they approach the topological dimension $\m$.}
\label{fig:dh-lattice}
\end{centering}
\end{figure}


\subsection{Spectral dimension on simplicial complexes}
\label{sec:ds-simplicial}

Now that I have analysed in detail the generic discreteness effect in the spectral dimension of complexes, the focus in this section is on specific effects of the combinatorial structure of complexes, in particular complexes without the regularity of lattices.
To this end I will calculate numerically the spectral dimension of various kinds of simplicial complexes: equilateral triangulations of the sphere and torus as well as various global and random subdivisions of the latter.
To carry out the calculations I have implemented an algorithm for the discrete calculus (\sec{calculus}) to compute the Laplacian in each case. 
After diagonalization of the Laplacian, the spectral dimension is then obtained directly from the heat trace in the form \eqref{ht-from-spectrum}.

This investigation is preliminary to the quantum case in a twofold way. 
First, it is important for the distinction of effects of the underlying discrete structure from additional quantum effects. 
Second, from a practical perspective, one can then choose those complexes with smaller and more controllable discreteness features. 
Only when discreteness effects are under control quantum effects of states on complexes can be distinguished.
In the end, regular torus triangulations turn out to be the best candidates for this purpose.

\subsubsection*{Triangulations of $S^{2}$.}

As a first kind of complexes I consider triangulations of the sphere to compare them to the smooth case (\sec{smooth}).
Obvious triangulations of the 2-sphere are the boundaries of the three triangular platonic solids (tetrahedron, octahedron, icosahedron).
These are the only equilateral triangulations of the smooth 2-sphere in terms of simplicial complexes (in the strict sense of \dref{simplicial-complex}).
The results of calculating their spectral dimension are shown in \fig{S2triang}.
The larger the triangulation, the taller the height of the peak. 
In particular, only the spectral dimension of the icosahedral triangulation can be seen as providing a good approximation for the topological dimension $\m=2$, provided one defines it to coincide with the height of the peak (and assuming this would become a plateau for larger complexes, as seems to be the case in feasible calculations).
\begin{figure}
\centering
\includegraphics[width=7.5cm]{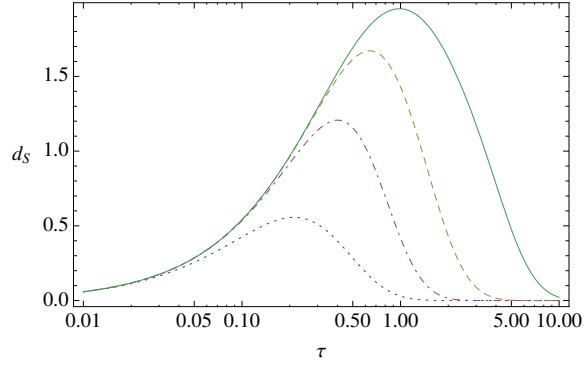}
\caption{$\ds$ of equilateral triangulations of $S^{2}$ in terms of the boundary of the (triangular) platonic solids, icosahedron, octahedron and tetrahedron (solid, dashed and dash-dotted curve) and of the dipole (dotted curve).  \label{fig:S2triang}}
\end{figure}

The extent to which too-small triangulations fail to capture the topological dimension can be seen even more drastically in the case of the triangulation of the $\m$-sphere in terms of just two simplices,  called dipole or (super)melon, 
which is an example of a generalized simplicial complex (\dref{generalized-polyhedral}).
Independently of the dimension $\m$, its heat trace is (see Appendix C in \cite{\COTa})
\begin{equation}
P_{\dipole\m}(\tau) = P_{\dipole\m}(t=\alpha \tau)=1+\rme^{-t},\label{eq:Pdipol}
\end{equation}
where the only geometric factor $\alpha$ can be absorbed into the diffusion parameter. This yields a spectral dimension 
\begin{equation}
\ds^{\dipole\m}(t)=\frac{2t}{1+\rme^t}\,.\label{eq:dSdipol}
\end{equation}
From its derivative 
\begin{equation}
\frac{\d}{\d t}\ds^{\dipole\m}(t)=2\frac{1-\rme(t-1)\,\rme^{t-1}}{(1+\rme^t)^{2}}\,,
\end{equation}
one finds that the maximum is at $t^{\rm max}=W_{0}(1/\rme)+1\approx1.278$
(where $W_{0}$ is the upper branch of the real Lambert $W$-function) and has value
$
\ds^{\rm max}=\ds(t^{\rm max})\approx 0.56\,,
$
 independently of dimension $\m$ and the parameter $\alpha$. Only the position of the maximum is rescaled by $\alpha$.

\subsubsection*{Regular equilateral triangulations of $T^\m$.}

\begin{figure}
\centering
\includegraphics[width=7cm]{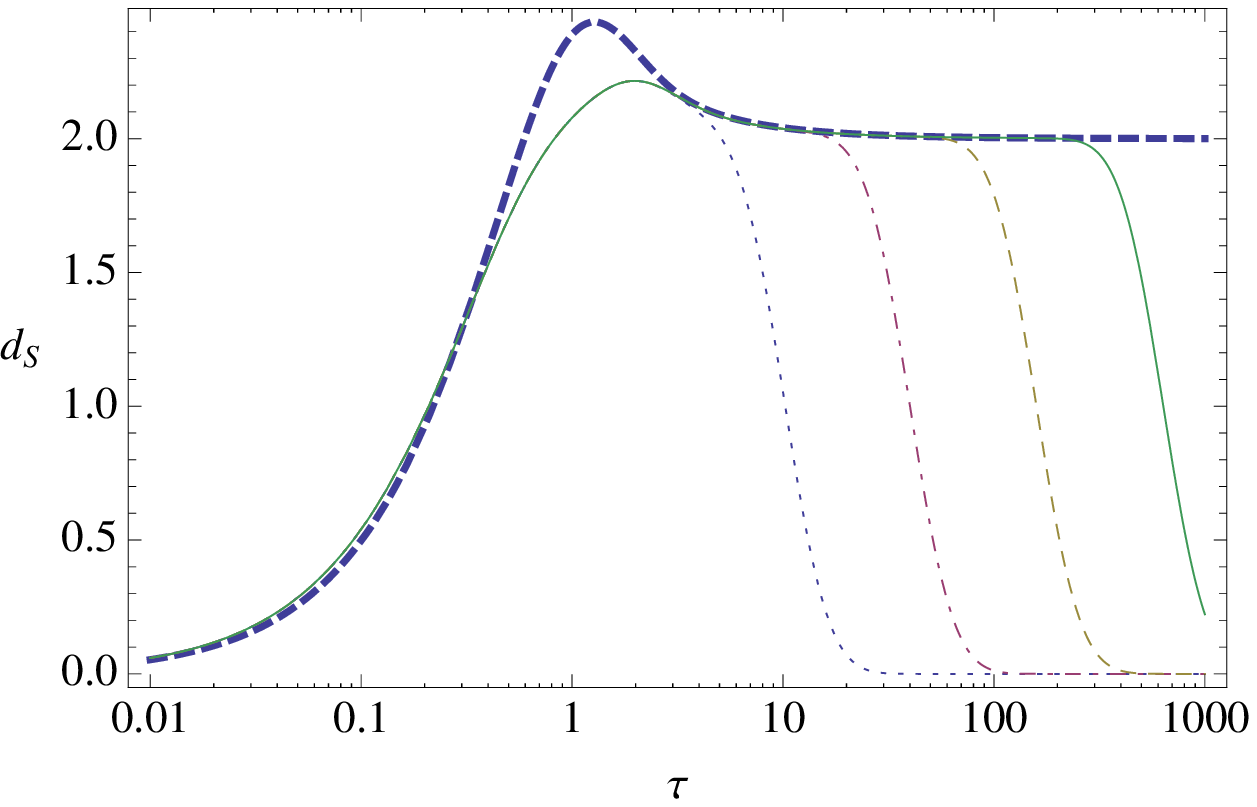}
\includegraphics[width=7cm]{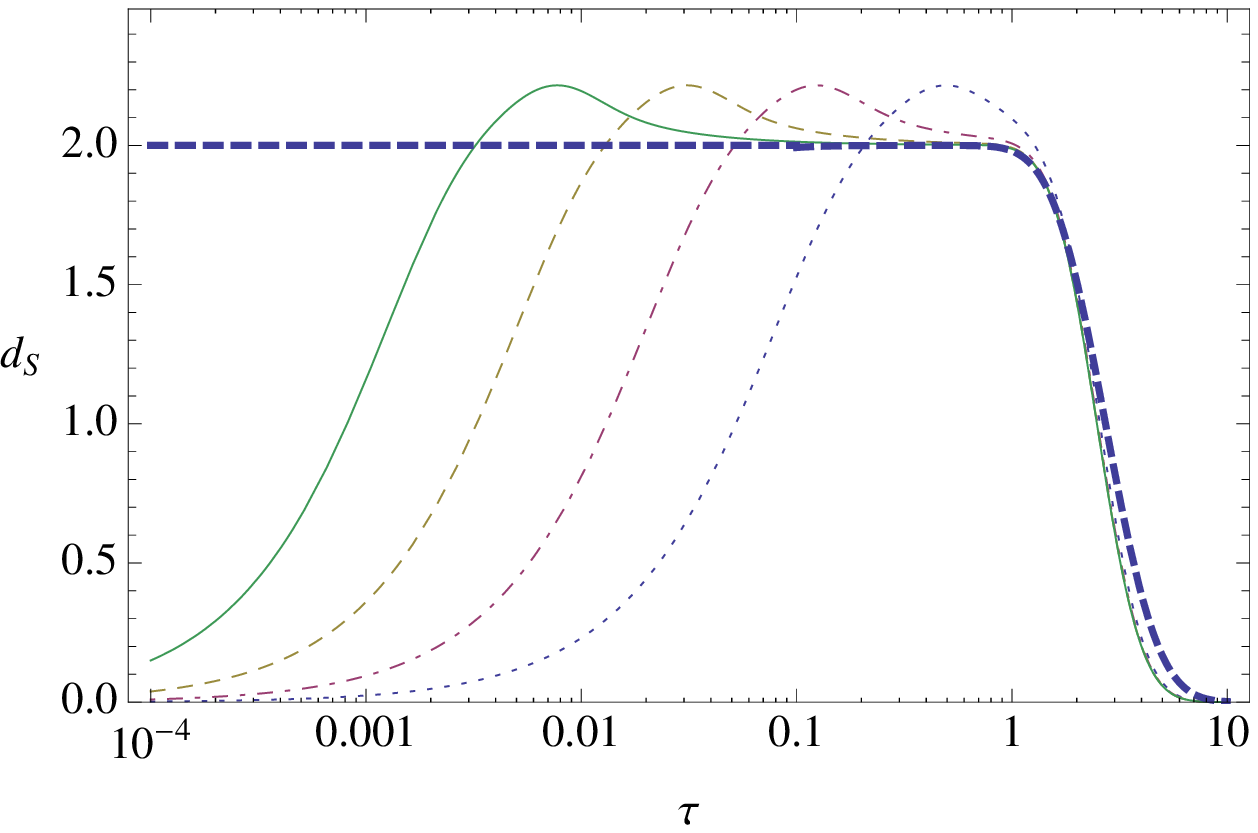}
\caption{Left: (a) unrescaled equilateral $T^2$ triangulation $\T_{2,\size}$ with $\np0 = \size^2 = (3\cdot 2^k)^2$ vertices, $k=0,1,2,3$ from left to right (dotted, dot-dashed, dashed thin, solid curve) indicating a convergence to the topological dimension for large $\tau$ in the $\size\to\infty$ limit, compared with the infinite quadrangular lattice (dashed thick).
Right: (b) rescaled triangulations ($k=0,1,2,3$ from right to left) indicating a convergence to the smooth $T^2$ (dashed thick) for $\size\to\infty$.
\label{fig:dsT2triang}}
\end{figure}

While there are no further, larger equilateral triangulations of $S^{2}$, for the $\m$-torus $T^\m$ there are regular equilateral triangulations $\T_{\m,\size}$ of arbitrary combinatorial size, \ie number of vertices $\np0 = |\Tp{0}_{\m,\size}| = \size^\m$. 
These are obtained from the hypercubic lattices via standard triangulation of each hypercube
\cite{Brehm:2011kp} and they are simplicial complexes (in the strict sense of \dref{simplicial-complex}) for $\size\ge 3$.
In two dimensions, these are triangulations of the flat torus with radii ratio $R_{1}/R_{2}=\sin(\pi/3) = \sqrt{3}/2$ (\fig{T2triang}). They are the dual to the finite, toroidal version of the hexagonal lattice.
Even though they are as such just another kind of lattice, I consider them here in detail because finite simplicial complexes are most appropriate for the quantum states later.
Furthermore, they will be the seeds for various kinds of irregular complexes in this section.

When comparing triangulations of different combinatorial size $|\Tp{0}|$ there are two possibilities: one can either
\begin{enumerate}
\item  [(a)] fix the edge lengths to some scale $a$ such that the geometric size of the triangulations is growing with the combinatorial size, or one can 
\item [(b)] rescale them to $a_\size = a/\size$ according to the combinatorial size to keep the overall geometric size fixed. 
\end{enumerate}
Thus, in the limit $\size\to\infty$ the case (a) gives a triangulated plane $\R^{2}$, while case (b) approximates the smooth flat $T^{2}$ geometry.
Indeed, the calculations in \fig{dsT2triang} (for various finite $\size$) indicate that the spectral dimension of these triangulations capture both limits.
Moreover, numerical calculations of the simplicial case show again that the only difference between lattices of various combinatorics consists in a qualitatively different discretization artefact.

The same analysis can be repeated in higher dimensions. 
For example, in terms of the standard triangulation of the cube, there is an equilateral triangulation $\T_{3,\size}$ of the 3-torus with radii ratios $R_{2}/R_{1}=\sqrt{3}/2$ and $R_{3}/R_{1}\approx0.752$.
The calculations in \fig{dsT3triang} again indicate the correct behaviour in the  infinite size limit.
The discretization effect here is slightly more marked, in that below the relatively small peak there is a small regime of relatively weak slope before the usual steeper fall-off at small $\tau$ sets in. 
\begin{figure}
\centering
\includegraphics[width=7cm]{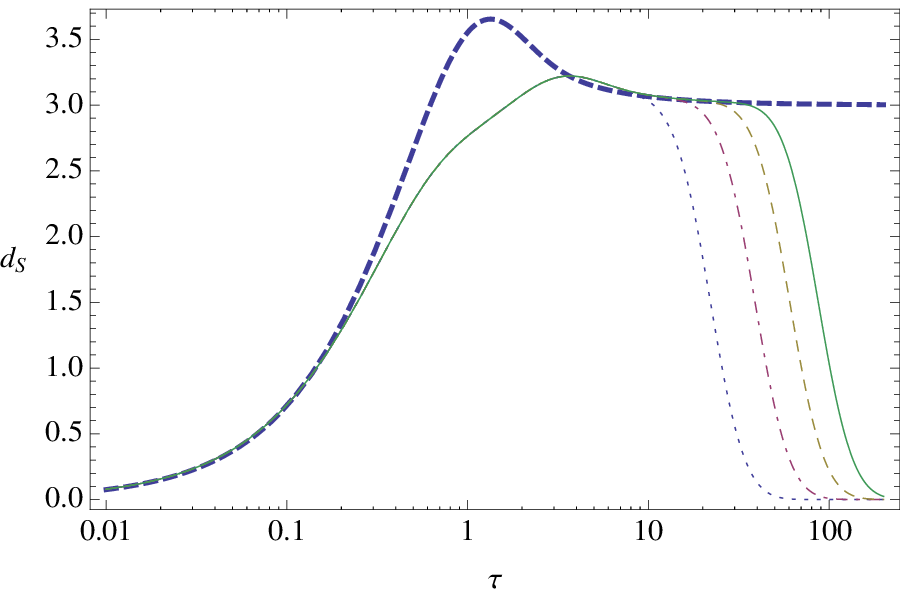}
\includegraphics[width=7cm]{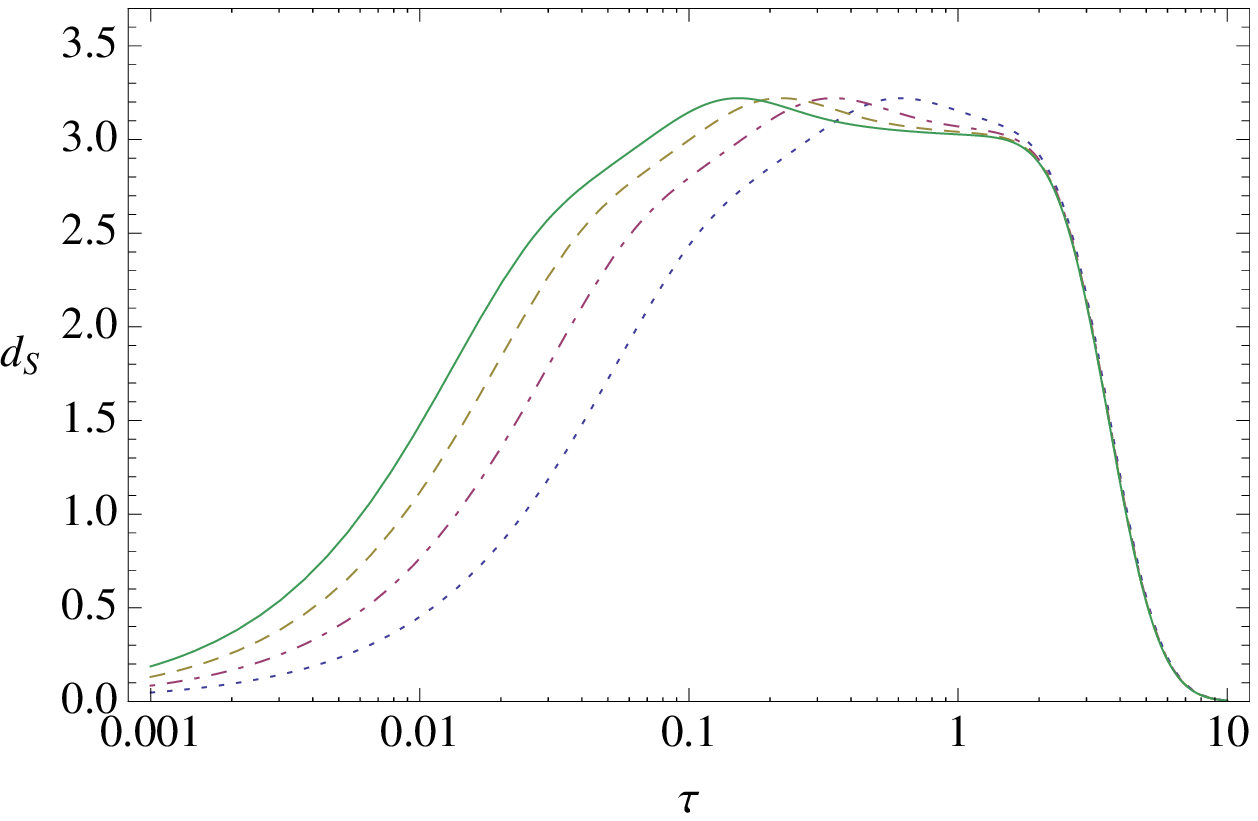}
\caption{Left: (a) equilateral $T^{3}$ triangulation in terms of $\np0 = \size^{3}$ vertices,
$\size=6,8,10,12$ from left to right (dotted, dot-dashed, dashed thin, solid curve), compared with the cubic lattice (dashed thick). Right: (b) rescaled triangulations.\label{fig:dsT3triang}}
\end{figure}

\subsubsection*{Subdivisions of triangulations of $T^\m$.}

An alternative strategy for obtaining combinatorially larger simplicial manifolds of a given topology is to subdivide them. A natural way to do so is the Pachner 1-($\m$+1) move where one $\m$-simplex is subdivided into $\m+1$ simplices by inserting a vertex in the middle of the original simplex and connecting it to all its vertices. 

Again, concerning the geometric realization of such a complex, one may either
\begin{enumerate}
\item [(a)] consider it as an equilateral triangulation (although this is not a triangulation of a torus with flat geometry anymore) or 
\item [(b)] rescale the edge lengths such that it is a triangulation of the flat torus.
\end{enumerate}
More precisely, such a rescaling is obtained in the following way.
The vertex $v_0$  inserted by the Pachner 1-($\m$+1) move can be realized as the barycentre of an $\m$-simplex with vertex set $(v_1 v_2 \dots v_{\m+1})$. 
The new edges $(v_0 v_a)$, $a=1,2,...,\m+1$, then need to have squared lengths
\begin{equation}
l_{0a}^{2}=\frac{1}{\left(\m+1\right)^{2}}\left(\m\underset{b}{\sum}l_{ab}^{2}-\underset{(bc)}{\sum}l_{bc}^{2}\right)
\end{equation}
to preserve the flat geometry approximated by the triangulation (see appendix \sec{length-variables}; sums are running over vertices $v_b$ and edges $(v_b v_c)$ of the simplex not containing the vertex $v_a$).

One can obtain new simplicial complexes applying these subdivisions in various ways.
I consider two cases, a global and a random subdivision:
\begin{enumerate}
\item [(i)] First, I subdivide the above $T^{2}$ triangulation $\T_{2,3}$ applying the 1-($\m$+1) move simultaneously on all triangles; 
\item [(ii)] second, I apply it randomly.
\end{enumerate}
Both cases are considered either as equilateral (a) or as rescaled triangulations (b) and the results are the following:

\begin{figure}
\centering
\includegraphics[width=7cm]{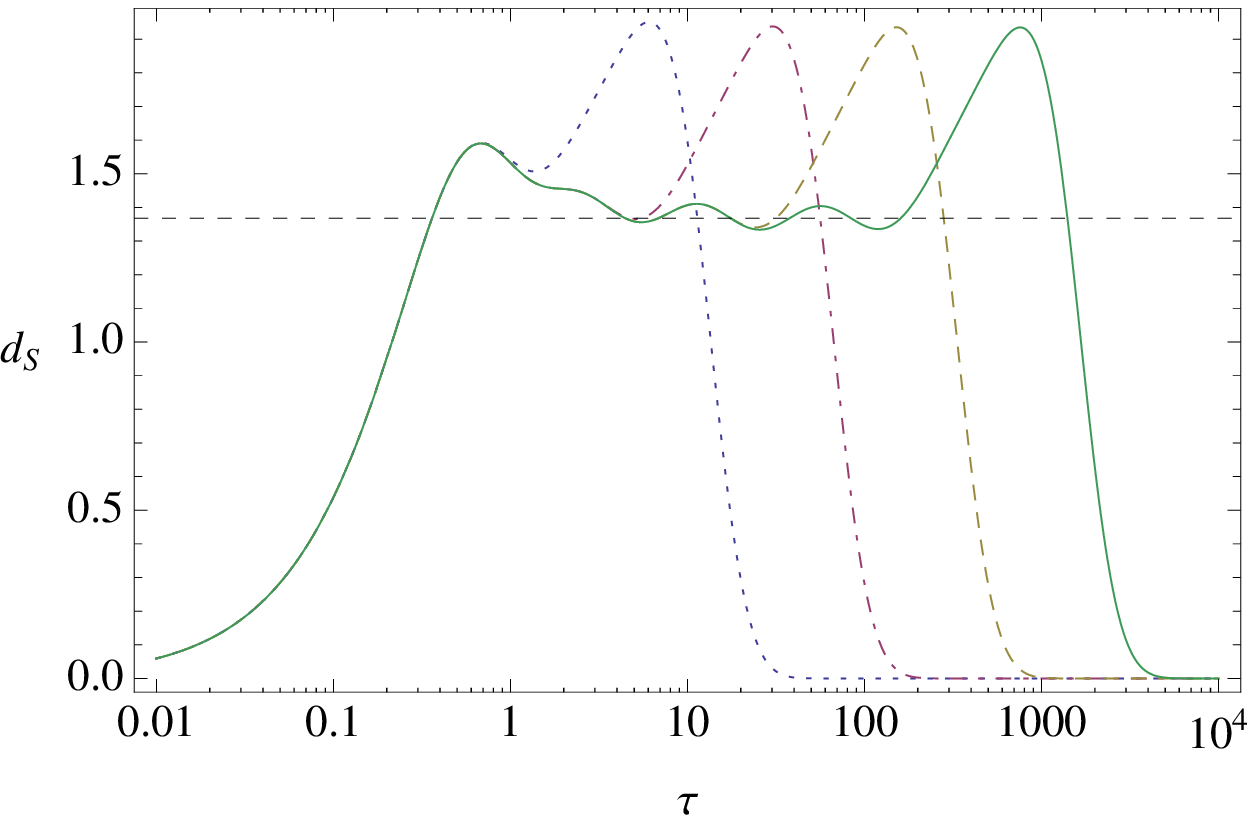}
\includegraphics[width=7cm]{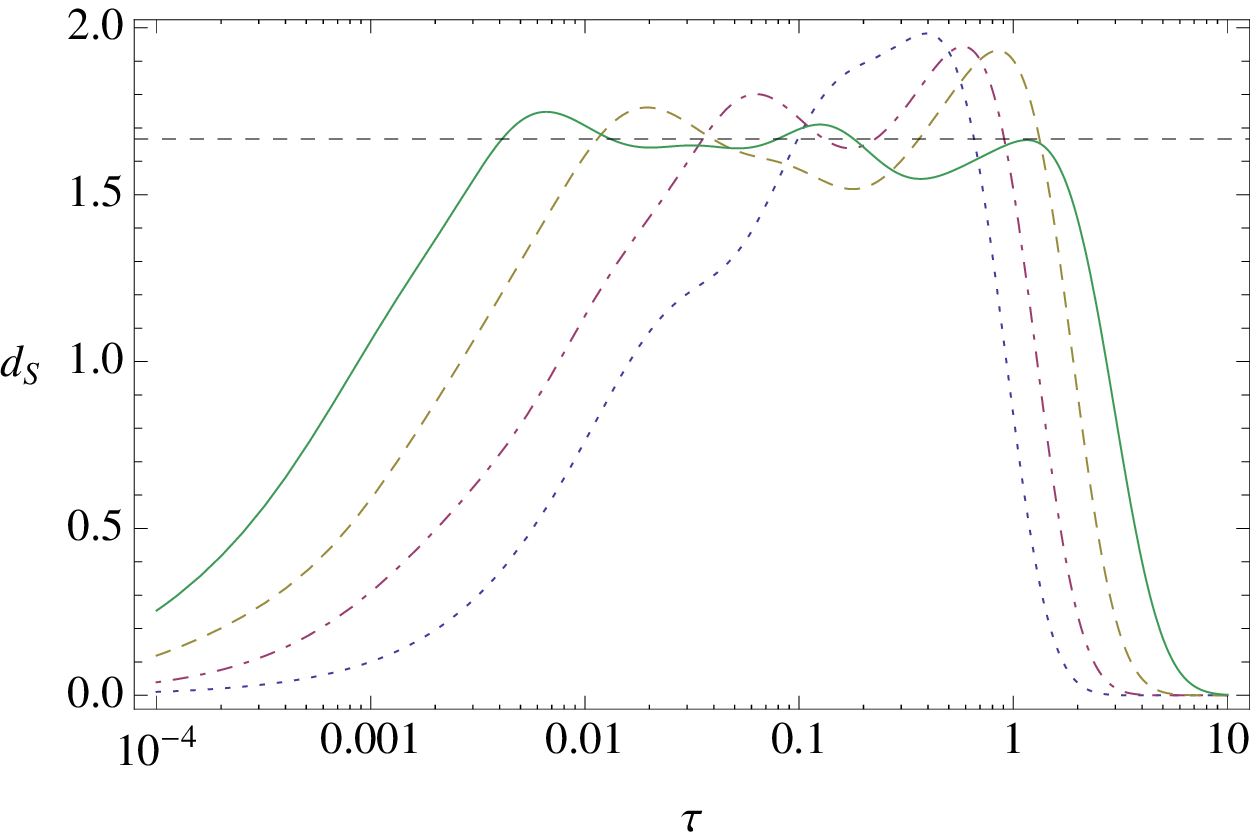}
\caption{Spectral dimension of $k$ global subdivisions of the smallest regular $T^{2}$-triangulation $\T_{2,3}$, for $k=1,2,3,4$ (dotted, dot-dashed, dashed and solid curve).
Left: (a) purely combinatorial complex, with a dashed line at $\ds=3/4$. 
Right: (b) complex with rescaled edge lengths triangulating the flat torus, with a dashed line at $d_S=5/3$.}
\label{fig:dsT2subdiv}
\end{figure}

\begin{enumerate}
\item[(i-a)]
In the first case of global subdivisions (\fig{dsT2subdiv}), considered (a) as equilateral triangulations, there is a peak slightly lower than $\m=2$ at the diffusion time scale of the size of the triangulation. 
But at smaller intermediate scales I find small oscillations around a value of about $\ds\approx 3/4$, obtained after integrating over a period. 
It is not surprising that there is a deviation from the topological dimension, since these equilateral triangulations have geometric realizations only in terms of a torus curved at various scales in a specific manner.
Indeed, the approximate value of $\ds\approx 3/4$ suggests that these simplicial complexes provide an instance of branched polymers.
These are an example of discrete geometries found in many cases in (causal) dynamical triangulations \cite{\cdtAJL} and in tensor models \cite{Gurau:2013th}.

\item [(i-b)] What might be more surprising is that also the complex with barycentrically rescaled edges triangulating the flat torus has a spectral-dimension plot quite different from the equilateral triangulations and cubulations considered above. 
There is a more complicated oscillatory behaviour but now at a value around $5/3$. 
Furthermore, the fall-off at small $\tau$ is much less steep. 

The rescaled Pachner subdivisions are thus an example of a triangulation which substantially deviates with respect to the spectral dimension in the corresponding continuum geometry, due to the particular combinatorial structure.
In particular, in the finite case there is no plateau at the value of the topological dimension, nor is it expected that the topological dimension be recovered in the large-size limit.
Thus, an important lesson is that the spectral dimension of a triangulation of a given smooth geometry does depend on the combinatorics of the chosen triangulation.

\begin{figure}
\centering
\includegraphics[width=7cm]{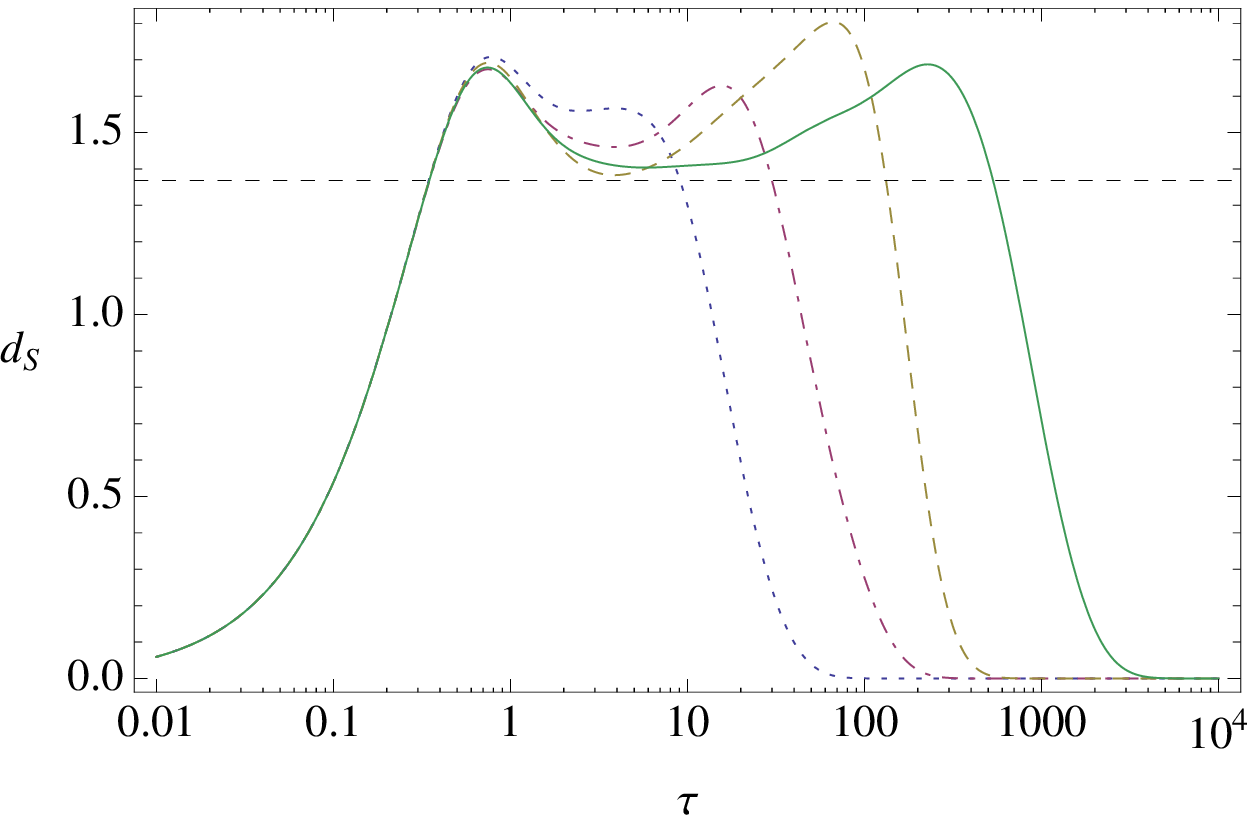}
\includegraphics[width=7cm]{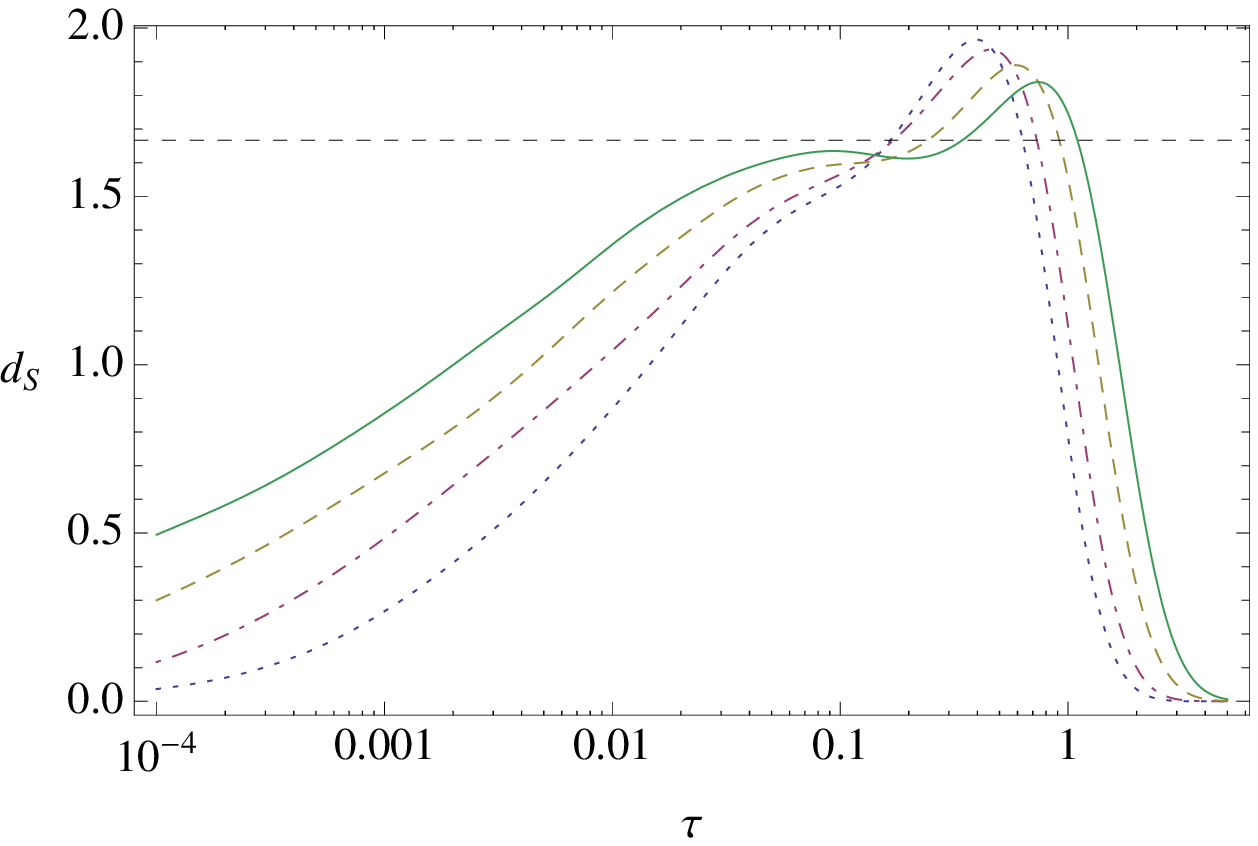}
\caption{Spectral dimension of the simplicial complex obtained from the 
smallest regular $T^{2}$-triangulation $\T_{2,3}$ by $6\times 3^{k}$ single random subdivisions, $k=1,2,3,4$ (dotted, dot-dashed, dashed and solid curve)
taken (a) as equilateral (left, with dashed line at $d_S=1.37$) or (b) rescaled triangulations of the flat torus (right, with dashed line at $d_S=3/5$). \label{fig:dsT2subran}}
\end{figure}
\end{enumerate}

In the second case, I construct subdivisions by choosing (with uniform random distribution) one triangle to subdivide, performing the subdivision, and then re-iterating the process. 
The results (\fig{dsT2subran}) hint at some kind of averaging effect. 
\begin{enumerate}
\item[(ii-a)] 
In the equilateral case, the oscillations are washed away and the height of the peak at the volume-size scale seems to be dependent on the particular elements chosen in the random ensemble. 
The closer this choice to a global subdivision, the more pronounced the peak.

\item [(ii-b)]
In the rescaled case, the regime around $5/3$ is much smaller than in the global subdivided case and the fall-off is even less steep. This might be explained by the fact that the random subdivision effectively averages over both the regime corresponding to a plateau and the regime of low-$\tau$ fall-off.
\end{enumerate}

Finally, I consider a very peculiar element in the random ensemble, namely the repeated subdivision around a single vertex of the triangulation.
This is interesting because it shows that the spectral dimension is very much dependent on the combinatorics.
The result is a spectral dimension which could be conjectured to run to $\ds\ra1$ in the large-size limit, as suggested by the calculations at finite sizes (\fig{dsT2frac}). 
Since this property of the example is actually independent of the global structure of the complex, one can have the same result starting with only one triangle.
The only difference is that the spectral dimension goes to infinity for large $\tau$ since there is no zero eigenvalue in the spectrum of the Laplacian due to the boundary of the single triangle (\fig{dsT2frac}). 
\begin{figure}
\centering
\includegraphics[width=7cm]{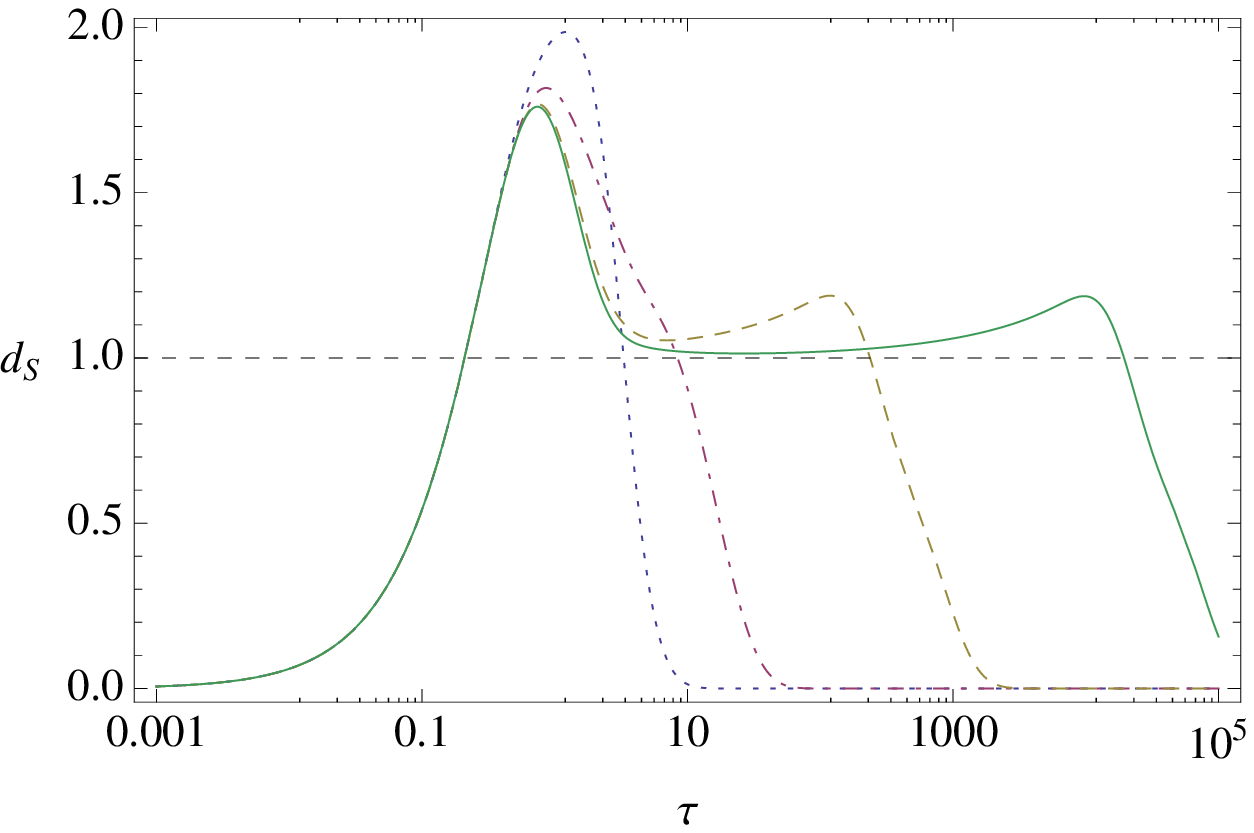}
\includegraphics[width=7cm]{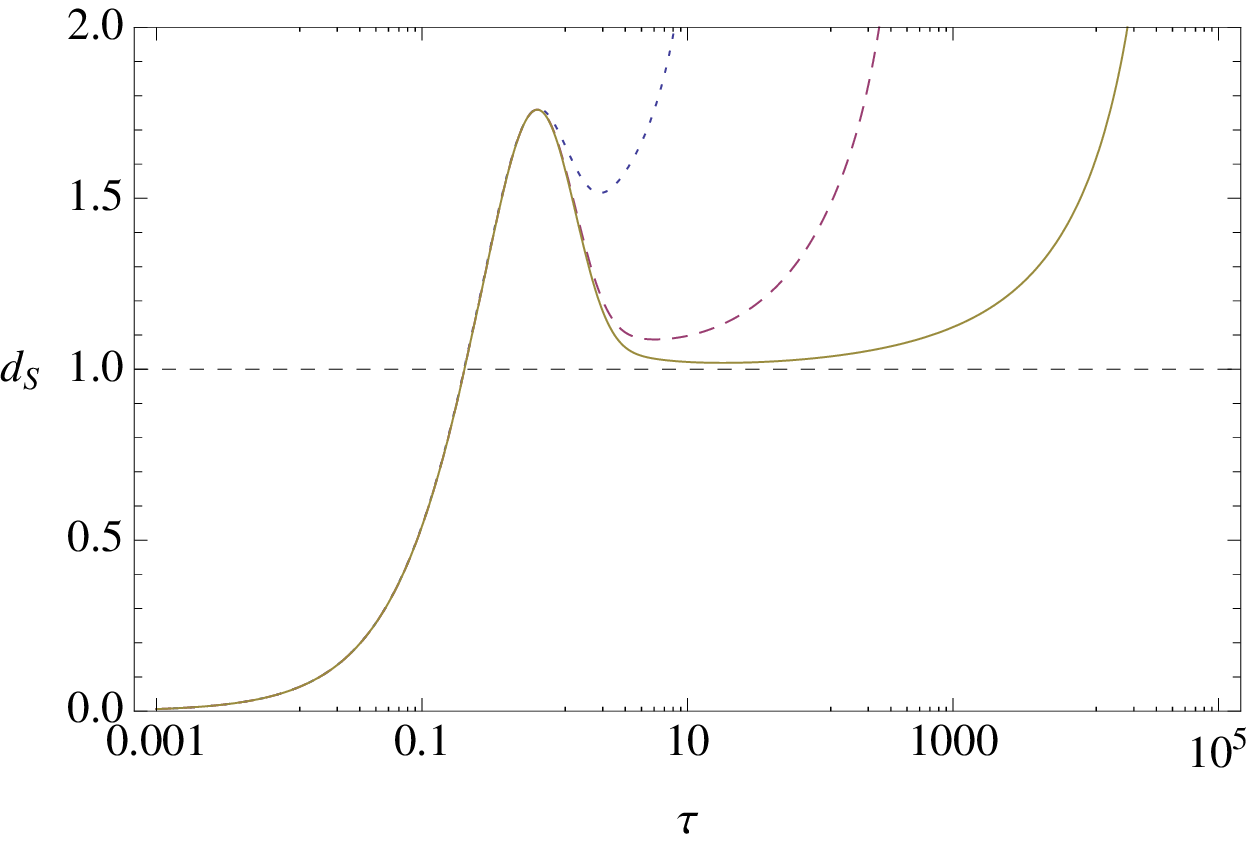}
\caption{Spectral dimension of the simplicial complex obtained by $10^{k}$ single subdivisions, $k=0,1,2,3$, around the same vertex of the torus triangulation $\T_{2,3}$ (left) and of a single triangle (right), both unrescaled.
In both cases, $\ds$ is independent of the global structure in the large complex size limit.
\label{fig:dsT2frac}}
\end{figure}

\

In this section \ref{sec:classical}, I have investigated discreteness effects in the effective dimensions for various classes of complexes. 
For lattices, the walk dimension does not show any discreteness effect and the Hausdorff dimension has a generic effect of a fall-off to $\dh=1$ below the lattice scale. 
In the case of spectral dimension, there are three characteristic features:
\begin{enumerate}
\item[(A)] a zero spectral dimension at scales below the lattice spacing, $\ds\simeq 0$, coming from the fact that the test particle feels the discrete spacing of the lattice: for too-short or infinitesimal times $\delta\tau$, the probe does not have the chance to diffuse from a given node to another;

\item[(B)] a peak larger than the topological dimension at the lattice scale; 

\item[(C)] at larger scales, there is agreement with the smooth spectral dimension in the case of lattices of various combinatorics, such as hypercubulations (\fig{dsZn}), triangulations (\fig{dsT2triang}) as well as hexagonal and octagonal lattices (\fig{ds-lattices}).
In particular, there is a plateau with height close to the topological dimension, the more extended the larger the complex. This cannot be found for spherical equilateral triangulations, which only exist for certain sizes too small to establish a plateau (\fig{S2triang}).  
\end{enumerate}

Furthermore, already at the classical level the spectral dimension is very sensitive to the precise structure of the discrete manifolds combinatorics.
This is illustrated exemplarily by triangulations from subdivisions via Pachner 1-($\m$+1)-moves.
These do not reproduce the topological dimension neither in the equilateral nor in the rescaled case (figures \ref{fig:dsT2subdiv} and \ref{fig:dsT2subran}).


\section{Spectral dimension on LQG coherent states\label{sec:coherent-states}}

Now the ground is set to calculate the effective dimension of various examples of discrete quantum geometries.
The focus of this section is on the spectral dimension of states of 2+1 Euclidean quantum geometry. 
The reason is that in this case numerical calculations of meaningful examples are feasible without further approximations.
The aim is to identify the quantum corrections to the classical spectral dimension.

States of particular interest  are coherent states, that is states which are peaked both on an intrinsic, two-dimensional spatial geometry as well as on its conjugate extrinsic geometry. 
Quantum corrections of $\ds$ depend on the parameters of the coherent states in this case.
For these states the quantum spectral shows only qualitative and small deviations from the classical case. 
To uncover stronger quantum effects, I investigate also superpositions of coherent states. 
The resulting spectral dimension turns out to be an average of the superposed cases, thereby showing a more distinct quantum behaviour.
To determine the features of this effect more precisely, it turns out that the numerical techniques used in this section are too limited.
To determine these more in detail, I will present a model for which calculations are more tractable in the following section \ref{sec:dimensional-flow}.

The structure of this section is the following.
I will start in \sec{defLQG} with a discussion of the definition of the heat kernel trace in the particular setting of kinematical LQG states in $2+1$ dimensions.
Then I will present the methods and results for the computation of the spectral dimension of coherent states on a given triangulation in that context in section \ref{coherent}.
Finally, I will give the numerical results for superposition of such states, both over geometric data which the states are peaking at (\sec{summing-coherent-states}) as well as over underlying triangulations (\sec{superpositions}).


\subsection{Heat trace operator on kinematical LQG states
\label{sec:defLQG}}

As introduced in section \sec{molecule-formulation}, LQG states are functions of spatial geometric variables (holonomies, flux variables or representation labels) based on closed graphs $\bg\in\bgs$ which are orientable.
As anticipated and motivated in the introduction to this chapter, I will restrict here to LQG in $\std = 2+1$ dimensions.

Now, the elements of $\hkin = \bigoplus_{\bg\in\bgs} \hs_\bg$ are functions of the holonomies on the edges $\eb\in\E_\bg$ of the graph $\bg$ which is usually considered as the dual 1-skeleton of a complex $\cp$ corresponding to a $\sd=2$ spatial slice.
For simplicity, let me focus furthermore on the subspace of states 
\[\label{hkin3}
\hk3 := \bigoplus_{\bg\in\bgs_3} \hs_\bg
\]
based on 3-regular graphs.
These may be interpreted as dual to triangulations (up to the caveats discussed in \sec{std-complexes}).
In particular, elements $\sn$ in the gauge-invariant spin-network basis have unique trivial intertwiners $\intw{\vb}$, the Clebsch-Gordan symbols.
Thus, they are denoted simply by $(\sms)$.

Quantum geometries in this space can be directly related to discrete geometries in the following sense. 
As explained before \eqref{length-spectrum}, the spin-network basis $\{\sn\}$ diagonalizes the length operators $\widehat{l}_e$ associated to edges $e=\star\eb$ dual to $\eb\in\bg$ with a spectrum given in terms of the $SU(2)$ Casimir operator \cite{Achour:2014gr,Freidel:2003kx},
\[\label{spede}
\ql_e \smr = l(j_\eb) \smr  = \lbi\sqrt{ j_\eb(j_\eb+1)+ \csu} \,\smr\;,
\]
where the scale $\lbi = 8\pi\bi\lpl $ is set by the Planck length and the Barbero-Immirzi parameter and $\csu$ is a constant dependent on the quantization map chosen for the Casimir operator \cite{Freidel:2003kx,Guedes:2013cc}.
If $\csu>0$, the Clebsch-Gordan conditions for the representations on edges $\eb=1,2,3$
\begin{equation}
\left|j_{1}-j_{2}\right|\leqslant j_{3}\leqslant j_{1}+j_{2},
\end{equation}
implicit in the intertwiners on vertices of $\bg$ yield triangle inequalities for the length expectation values of the associated edges $e=1,2,3$ on the primal complex $\cp$,
\begin{equation}\label{trianineq}
l(j_1) + l(j_2) > l(j_3)\;.
\end{equation}
This can be seen from the inequality 
\ba
\sqrt{j_{1}(j_{1}+1)+\csu} + \sqrt{j_{2}(j_{2}+1)+\csu} 
& > &\sqrt{\left(j_{1}+j_{2}\right)(j_{1}+j_{2}+1)+\csu}\nonumber\\
& \ge & \sqrt{j_{3}(j_{3}+1)+\csu}.
\ea
Only for $\csu = 0$ there are degenerate spin configurations where the inequality on the lengths is not strict (\eg $(j_1,j_2,j_3) = (0,1,1)$).
Even in this case one could obtain the triangle inequalities by  restricting to those states in $\mathcal{H}_{\rm kin}$ corresponding to non-degenerate configurations. 
This restriction is not expected to be significant for the calculation of the quantum spectral dimension. The results presented in the following are for $\csu = 1/4$, but alternative calculations 
 indicate that $\ds$ is not sensitive to such choice (\eg test of the $c=0$ case give very similar results).

On the basis of this relation of simplicial spin networks $(\sms)$ to edge-length geometries, an obvious choice is to take an explicit form of the discrete Laplacian in terms of edge-length variables \eqref{edge-length-laplacian}.
Promoting this expression to the corresponding map $\widehat\Delta_\bg\equiv\widehat\Delta_\cp$,  its action on spin-network states $\smr\in\hs_\bg$ returns a discrete Laplacian $\Delta_\bg(j_\eb)\equiv\Delta_\bg[l(j_\eb)] $,
\[\label{lqg-laplacian}
\widehat{\Delta}_\bg \smr 
= \widehat{\Delta_\bg[l_\eb^2]} \smr
= \Delta_\bg[\ql_e^{2}] \smr
= \Delta_\bg[l(j_\eb)] \smr \,,
\end{equation}
which is a function of representations $\rep{\eb}$ according to its dependence \eqref{edge-length-laplacian} on the lengths $l(j_\eb)$.
While the form of the primal volumes in the Laplacian, \ie triangle areas 
and the edge lengths themselves, is straightforward, 
there is in principle some freedom in choosing the form of the dual edge lengths as functions of primal lengths.
Here I choose a barycentric dual as argued for in \sec{laplacian}.


Finally, one has to check whether the formal expression for the heat-trace quantum observable \eqref{quantumP} in terms of the specific Laplacian \eqref{lqg-laplacian} is well defined in this context, \ie if the heat trace is a self-adjoint operator on $\hk3$ \eqref{hkin3}.
Thus, it has to be checked whether for every $(\sms),(\bg',j'_\eb) \in \hk3$ it holds that
\begin{equation}\label{usef}
\sml \widehat{P(\tau)} \bg',j'_\eb \qket \stackrel{?}{=} \qbra\widehat{P(\tau)} \bg,j_\eb | \bg',j'_\eb\qket\,.
\end{equation}
With the Laplacian \eqref{lqg-laplacian}, the left-hand side is
\ba
\sml \widehat{P(\tau)} \bg',j'_\eb \qket 
&=& \sml \Tr_{\bg'}\,\rme^{\tau\widehat{\Delta}_{\bg'}} |\bg',j'_\eb \qket
 = \sml \Tr_{\bg'}\,\rme^{\tau\Delta_{\bg'}(j'_\eb)} | \bg',j'_\eb \qket\\
&=& \Tr_{\bg}\,\rme^{\tau\Delta_{\bg}(j_\eb)}\delta_{\bg,\bg'}\delta_{j_\eb,j'_\eb}\,,
\ea
where $\delta_{\bg,\bg'}$ refers to the identity of graphs $\bg\in\bgs$ up to automorphisms. 
The reason is that states defined on distinct graphs are orthogonal according to  the usual inner product of LQG \cite{\lqgT,Kittel:2014vo}. 
The expression is equal to the right-hand side of (\ref{usef}) if, and only if, the spectrum of $P(\tau)$ is real since accordingly
\begin{equation}
\qbra\widehat{P(\tau)} \sms | \bg',j'_\eb \qket
= [\Tr_\bg\,\rme^{\tau\Delta_\bg(j_\eb)}]^{*}\delta_{\bg,\bg'}\delta_{j_\eb,j'_\eb}\,.
\end{equation}

In general the coefficients $(\Delta_\bg)_{ab}$ are real on geometric states where triangle inequalities (\ref{trianineq}), or equivalently closure constraints (which are geometricity conditions) are satisfied. 
Then, $\Tr_\bg\,\rme^{\tau\Delta_\bg(j_\eb)}$ is real as well.
Hence $\widehat{P(\tau)}$ is a good quantum observable on the kinematical states of 2+1 gravity with the operator ordering in the length operators $\ql_e$ chosen \cite{Freidel:2003kx,Guedes:2013cc} such that $\csu>0$, which is chosen here.



\subsection{Spectral dimension on coherent spin-network states}\label{coherent}

As a first example, I present the results for spectral dimension calculations of $(2+1)$-dimensional LQG coherent states.
The reason for considering coherent states here is twofold. 
First, as they are semi-classical states peaking at classical geometries, it is interesting to check if their properties are reflected in the spectral dimension. 
Second, if the quantum spectral dimension was comparable with the spectral dimension of the classical geometries peaked at, the difference should be understood as due to quantum corrections, which can be studied in a controlled manner in terms of the parameters of the coherent states.
The results show clear evidence for the first point.
Quantum corrections turn out to be small and can be described qualitatively.

In the LQG literature, group coherent states in the holonomy basis are the usual starting point for the construction of semiclassical, coherent states in $\hkin$.
Coherent states peaked at phase space points $(g,x=0)\in T^*G\cong G\times\mathfrak{g}$ can be obtained from the heat kernel $K_G^{\sigma}$ on the group $G=SU(2)$ with Peter-Weyl expansion:
\begin{equation}
\qs_{(g,0)}^{\sigma}(h)=K_{SU(2)}^{\sigma}(hg^{-1})=\sum_{j} \dj{j}\,\rme^{-\frac{C_{j}}{2 \sigma^2}}\chi^{j}(hg^{-1})\;.
\end{equation}
Here $C_j$ is the Casimir operator, $\chi^j$ refers to the character of $G$ representations labeled by $j$ and $\sigma \in \R^+$ is the spread.
There are two known constructions for a generalization including a non-trivial $\mathfrak{g}$-dependence: 
either by analytic continuation to $\rme^{\rmi x}g\in G^{\C}$ such that \cite{Hall:1994gg,Sahlmann:2001bw,Bahr:2009bc}
\begin{equation}
\qs_{(g,x)}^{\sigma}(h)=K^{\sigma}(hg^{-1}\rme^{-\rmi x})\;,
\end{equation}
or using the flux representation \cite{Oriti:2012kx} where the heat kernel appears to be a ($\kappa$-non-commutative) Gaussian 
\begin{equation}
\tilde{\qs}_{(g,x)}^{\sigma}(y)=\mathcal{F}[\qs_{g}^{\sigma}](x-y)\propto \rme_{\star}^{-\frac{1}{2\kappa^{2}}\frac{(x-y)^{2}}{2\sigma^2}}\star \rme^{\frac{\rmi}{\kappa}|P(g)|(x-y)}.
\end{equation}
Here, $\kappa=\lpl^{-1}$ and plane waves $e_g(x):=\rme^{\frac{\rmi}{\kappa}|P(g)| x}$ use group coordinates $\vec{P}(g) = \sin[\theta(g)] \vec{k}$ in the usual coordinates in which the group element is parametrized as $g = \rme^{\rmi\theta \vec{k} \cdot \vec{\sigma}}$, with $\vec{\sigma}$ the Pauli matrices. 
The notation $\rme_\star$ indicates the non-commutative exponential defined as a power series expansion of $\star$-monomials \cite{Guedes:2013cc}. 
Transforming back to group space, this results in an additional plane wave factor,
\begin{equation}
\qs_{(g,x)}^{\sigma}(h)=\qs_{(g,0)}^{\sigma}(h)\,e_{h}(-x)\,.
\end{equation}

Since the Laplacian \eqref{lqg-laplacian} is diagonal in the spin representation, the spin expansion of these coherent states is needed here.
In both cases, there is a limit (for large enough spins) in which these can be described \cite{Hall:1994gg,Sahlmann:2001bw,Oriti:2012kx,Bianchi:2010fp} as Gaussian-type states.
Thus, a coherent state on a graph $\bg = (\V_\bg,\E_\bg) \in\bgs$ is peaked at spin representation labels of intrinsic geometry $\{J_\eb\}_{\E_\bg}$ and angles of extrinsic geometry $\{K_\eb\}_{\E_\bg}$:
\begin{equation}
|\qs_{\bg}^{J_\eb,K_\eb}\qket=\frac{1}{N_\qs}\sum_{\{j_\eb\}}\csc \smr\,,
\end{equation}
with spin-network basis coefficients 
\begin{equation}
\csc \propto \prod_{\eb\in\E_\bg}\rme^{-\frac{(J_\eb-j_\eb)^{2}}{2\sigma^{2}}+\rmi K_\eb j_\eb}.
\end{equation}
In fact, following \cite{Bianchi:2010fp,Oriti:2012kx}, in the large-$x$ approximation, one finds that the $J_\eb$ can be identified (up to a factor dependent on $\sigma$) with the modulus $x_\eb$ of the fluxes,
and the $K_\eb$ are angles in the representation of the group elements $g_\eb$ (in the plane orthogonal to the fluxes $x_\eb$).
For the details I refer to \cite{Bianchi:2010fp} and \cite{Oriti:2012kx} respectively, since in the following only the intrinsic curvature as captured by the $J_\eb$ will be relevant. Their dependence on $\sigma$ does not play a role for fixed $\sigma$ or small variations of it, with respect to (assumed large) $x_\eb$.

The heat-trace expectation value can then be evaluated as
\ba
\qbra \widehat{P(\tau)}\qket _{\qs_{\bg}^{J_\eb,K_\eb}}
& = &\frac{1}{N_\qs^2} \underset{\{j_{\eb}\}}{\sum}\left|\csc\right|^{2}\qbra \bg,j_{\eb}|\Tr_\bg\,\rme^{\tau\widehat{\Delta}_\bg}|\bg,j_{\eb}\qket \\
& \propto &\underset{\{j_{\eb}\}}{\sum}\left| \prod_{\eb\in\E_\bg}\rme^{-\frac{(J_\eb-j_\eb)^{2}}{2\sigma^{2}}+\rmi K_{\eb}j_{\eb}}\right|^{2}\Tr_\bg\,\rme^{\tau\qbra \widehat{\Delta}_\bg \qket_\sms }\\
& = &\underset{\{j_{\eb}\}}{\sum}\left[\prod_{\eb\in\E_\bg}\rme^{-\frac{(J_\eb-j_\eb)^{2}}{\sigma^{2}}}\right]\Tr_\bg\,\rme^{\tau\Delta_{\bg}(j_{\eb})},\label{eq:htSuperposed}
\ea
from which the spectral dimension is derived according to \eqref{qds}.
As the spatial Laplacian does not depend on the extrinsic curvature, it is natural that also the phase of the coherent state drops out of the expectation value.


The interpretation of the semi-classical limit for the spectral dimension in terms of the parameters, \ie the spins ${J_\eb}$, the classical extrinsic geometries ${K_\eb}$, the spread $\sigma$ as well as the graph $\bg$, 
is rather subtle.
As far as the spread is concerned, with respect to the whole coherent state $\{J_{\eb},K_{\eb}\}$ there is a value where Heisenberg inequalities are minimized. 
On the other hand, since the spectral dimension on the spatial state is not dependent on the extrinsic curvature, obviously the deviation from the classical spectral dimension vanishes for $\sigma=0$, \ie for states sharply peaked at the intrinsic curvature but totally random in the extrinsic one. 
Nevertheless, these states are highly quantum. 
More interesting are quantum effects of those states randomizing the intrinsic curvature, that is states with large spread $\sigma$.

Furthermore, from the physical interpretation of the geometric spectra with corresponding eigenbasis in terms of the spin representations, the limit $J_\eb\ra\infty$ is often seen as another semi-classical limit. 

Finally, one should not forget the dependence of the states on the underlying graphs. At least with respect to the spectral dimension, classically I have already shown in \sec{ds-simplicial} the important dependence on the combinatorial structure. 
At the level of kinematical states, this dependence is still poorly understood in the literature.

In the following, I consider as spin-network graphs $\bg$ the dual 1-skeletons of the previous finite torus triangulations $\T_{\sd,\size}$ (\fig{T2triang}). 
They are parametrized by the number of nodes $\np0 = |\Tp0_{\sd,\size}| = \size^\sd$ and their spectral dimension converges to the topological spatial dimension $\sd$ if they are large or fine enough, \ie in the limit $\size \ra \infty$ (\sub{ds-simplicial}).
One can therefore interpret the quantum corrections as actual deviations from the topological dimension in such a geometric regime.
Accordingly, I will consider states peaked at all equal $J_\eb=J$ for $\eb\in \E_\bg$. 


For the numerical computations a sampling technique is needed as a direct implementation of equation (\ref{qds}) is unfeasible here. 
The reason is that the number $N_{j_\eb} = \sum_{j_\eb}$ of terms in the quantum sum over representations $\{j_{\eb}\}$, even with cutoffs $j_{\min}$ and $j_{\max}$ 
\begin{equation}
N_{j_\eb} = (j_{\max}-j_{\min})^{\np1}
\end{equation}
grows exponentially fast with the number of edges $\np1 = |\Tp1_{\sd,\size}|$.
I have already shown that only for large classical triangulations (\eg for $T^{2}$ of the order $10^{3}$; see \fig{dsT2triang}) a geometric regime is obtained. 
Furthermore, although the state amplitude for many spin configurations vanishes due to the Clebsch-Gordan conditions implicit in the intertwiners, the resulting effective space of spin configurations is highly non-trivial and, in general, not well understood \cite{Smerlak:2011vt}.

Alternatively, the quantum sum can be approximated by summing over some number of configuration samples $\{j_{\eb}\}$ chosen randomly according to the coefficients $\left|\csc\right|^{2}$ (where the norm $N_\qs$ must be included to give a proper probability density).
If the space of representations is discrete, as for the spins of $SU(2)$, this measure needs to be discrete. Here one can choose the binomial distribution $B$. For large enough $J_{\eb} = J$ and $\sigma$, this is in turn well approximated by the the normal distribution $\mathcal{N}$,
\begin{equation}
B\left(\left\lfloor\frac{2J^{2}}{J-\sigma^{2}}\right\rfloor,\frac{J-\sigma^{2}}{J}\right)\simeq\mathcal{N}\left(J,\frac{\sigma}{\sqrt{2}}\right)\,,\label{eq:BinomDis}
\end{equation}
where the floor function $\lfloor\cdot\rfloor$ is needed since the first argument must be an integer. 
The variable $X(J,\sigma)=B$ is a random field dependent on $J$ and $\sigma$. 
By virtue of the approximation (\ref{eq:BinomDis}), the probability density function associated with it is the Gaussian profile $(\pi\sigma^2)^{-1/2}\,\rme^{-(J-x)^2/\sigma^2}$.


To compare $\ds$ for various peaks $J$, it is meaningful to choose a scale such that $a_J=l^{2}(J)$ (the scale of the Laplacian set by the spectra \eqref{spede}) is kept fixed.
This is obtained including a rescaling factor $l^{2}(J)$ such that
\begin{equation}
\qbra \widehat{P(\tau)}\qket _{\qs_\bg^{J_\eb,K_\eb}} \propto \underset{\{j_{\eb}\}}{\sum}\left[\prod_{\eb\in\E_\bg} \rme^{-\frac{(J_\eb-j_\eb)^{2}}{\sigma^{2}}}\right]\Tr\,\rme^{\tau l^{2}(J)\Delta_{\bg}(j_{\eb})}.
\end{equation}
Without this rescaling, as in the classical cases, one observes a shift $\ln\tau\to\ln\tau-2\ln l(J)$ of the $\ds$ plot due to the $\Delta\to\Delta/l^{2}(J)$ scaling of the Laplacian. 
Here I include this rescaling factor in the following calculations, mainly to allow for a more direct comparison of the spectral dimension of states.

\

Starting with the coherent states' dependence on the spin $J$ peaked at, the results are the following.
As expected, the spectral dimension function of these quantum states does not differ much from the classical version: 
for $J_{\eb}=J=\mathcal{O}(10)$ and $\sigma=\mathcal{O}(1)$, the deviation is at most of order $10^{-2}$ (\fig{coherJs}). This is an important consistency check. 
Note, though, that all the known coherent states are peaked at discrete geometries. Therefore, strictly speaking, it is not the spatial topological dimension $\sd$ but the particular spectral-dimension profile of these discrete geometries to be approximated well by the coherent states. 
What can be considered as quantum corrections in the spectral dimension is thus the difference between the quantum spectral dimension and the spectral dimension of discrete geometries.
\begin{figure}
\centering
\includegraphics[width=7cm]{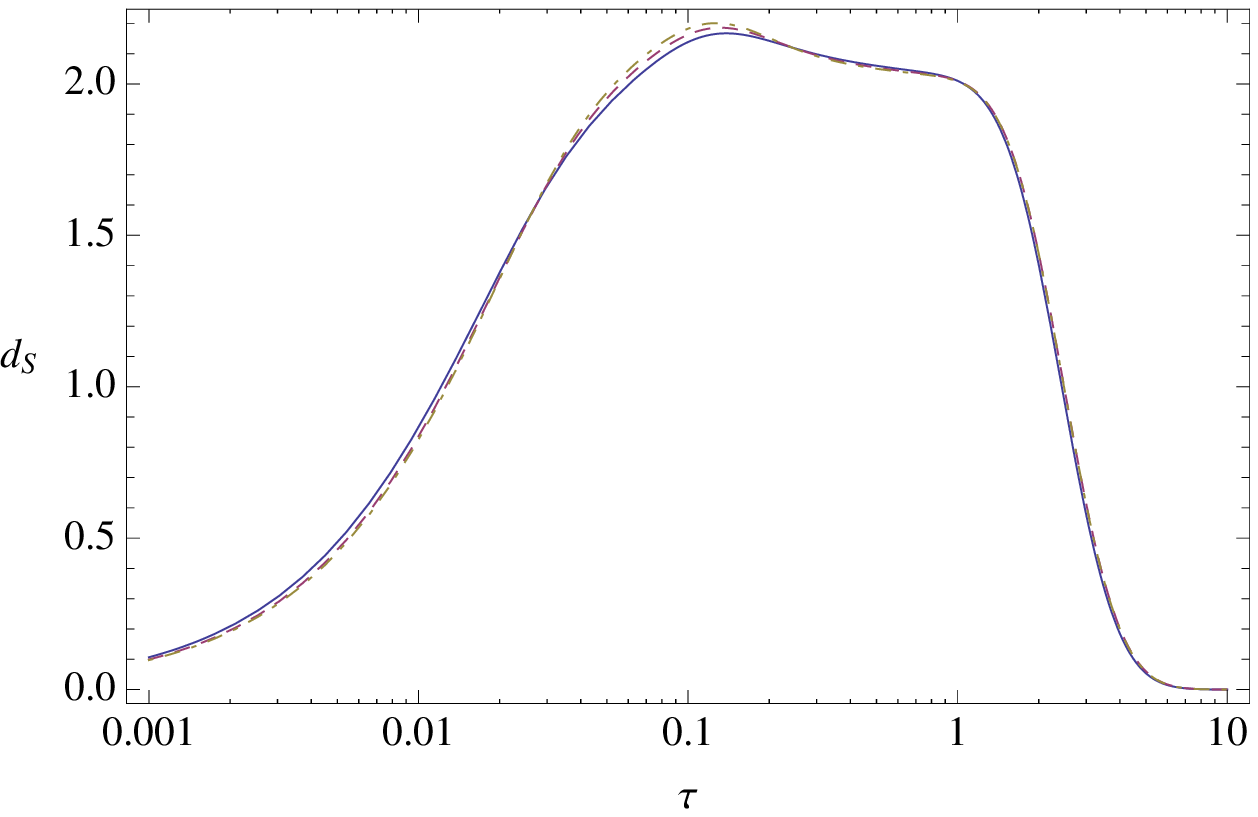}
\includegraphics[width=7.4cm]{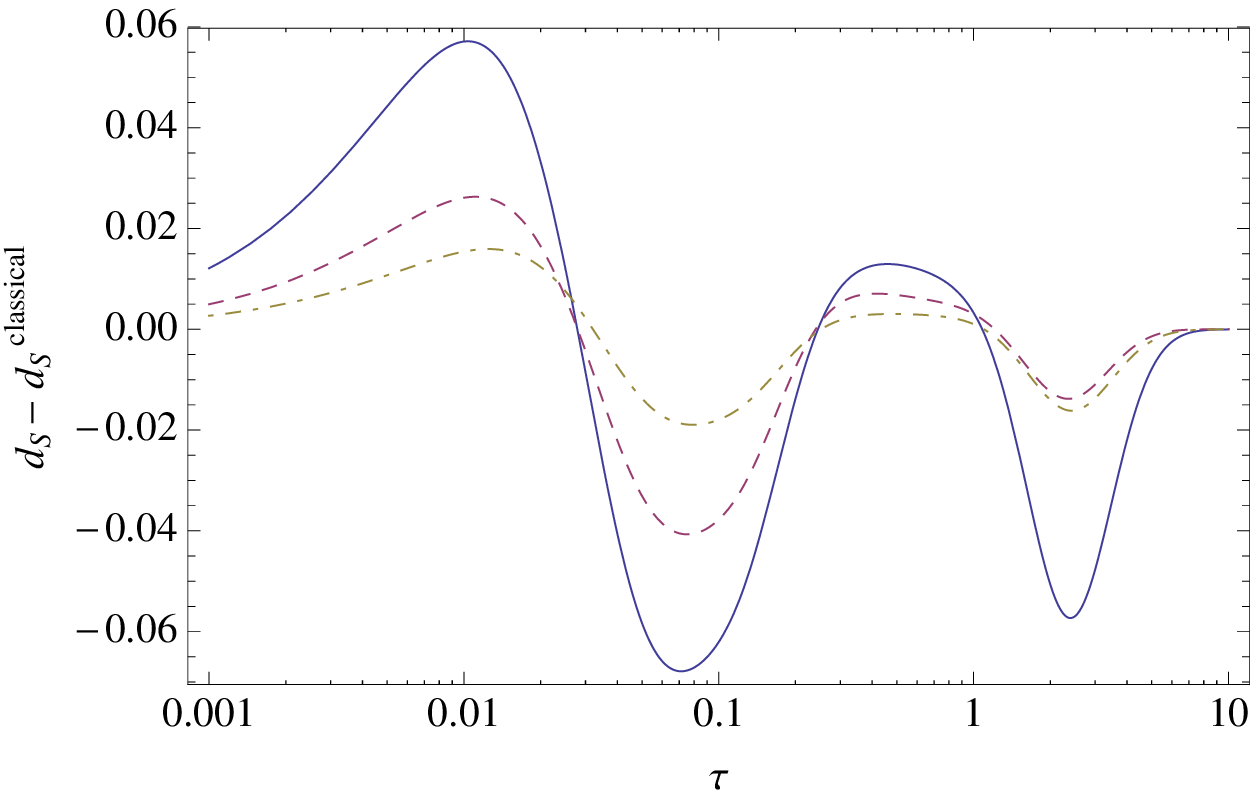}
\caption{Left: $\ds$ for coherent states peaked at $l(J)=J+1/2=16,32,64$ (solid, dashed, dot-dashed curve) on the regular torus triangulation $\T_{2,4}$ (\ie with $\np2 =18\times4^{2}=288$ triangles) with spread $\sigma=\sqrt{J-1/2}$. Right: deviation of $\ds$ from the classical case.
\label{fig:coherJs}}
\end{figure}

Second, I test the dependence of states with spread $\sigma$. 
Since the binomial distribution approximation (\ref{eq:BinomDis}) is only defined for $\sigma^{2}<2J$, it is difficult to probe the regime of larger $\sigma$ within this method. 
Probing the $\sigma$ dependence for $J_{\eb}=10$, I observe increasing quantum corrections, but still of order $\mathcal{O}(10^{-2})$ (\fig{coherSigmas}).

\begin{figure}
\centering
\includegraphics[width=7cm]{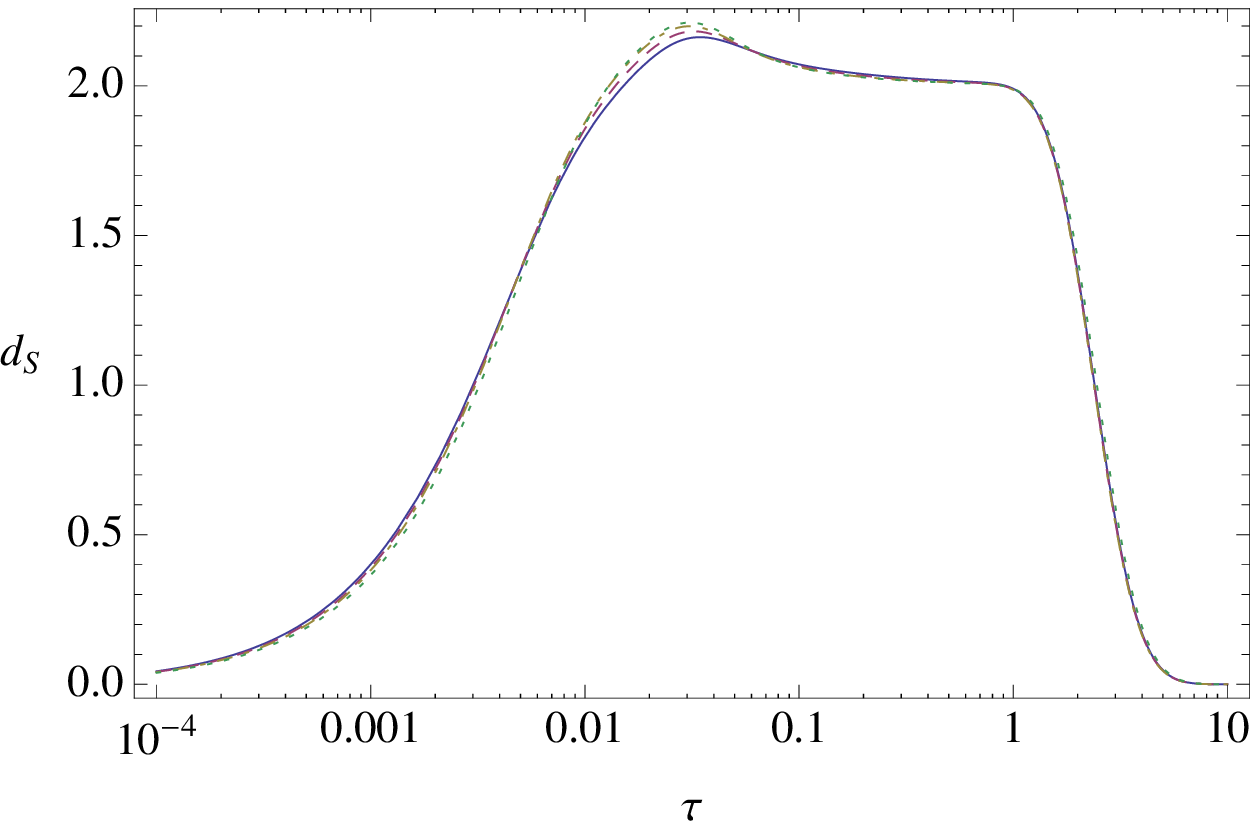}
\includegraphics[width=7.4cm]{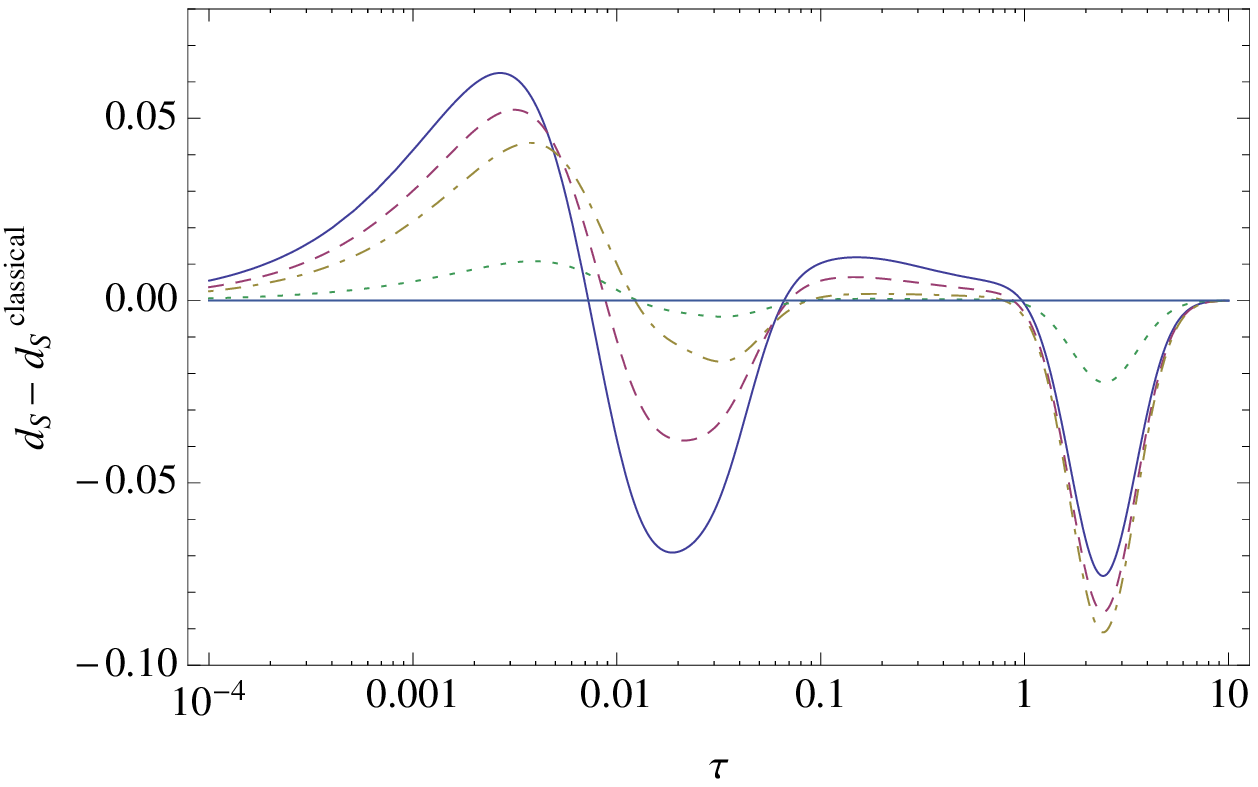}
\caption{Left: $\ds$ for coherent states peaked at $l(J)=J+1/2=16$ with spread $\sigma=1,2,3,\sqrt{15}$ (dotted, dot-dashed, dashed and solid curve) on the regular torus triangulation $\T_{2,8}$ (\ie with $\np2 = 18\times 8^2 = 1152$ triangles).
Right: deviation of $\ds$ from the classical case.\label{fig:coherSigmas}}
\end{figure}


The main challenge when extending the calculations to larger spreads is to deal with the highly non-trivial space of group representations due to the Clebsch-Gordan conditions, as noted before.
However, there is a very straightforward way to define pure states (or their superpositions) by bounding the range of spins to an interval $I=[j_{\rm min},j_{\rm max}]$ such that the Clebsch-Gordan conditions are trivially fulfilled for any combinatorial (\ie simplicial complex) structure of the states. 
In terms of the difference $\Delta J=j_{\rm max}-j_{\rm min}$, one could for instance construct states uniformly randomized over the interval
\begin{equation}
I_{J}=[\tfrac{1}{3}(2J+1),\tfrac{1}{3}(4J-1)]\cap\tfrac{1}{2}\N\,,
\end{equation}
on which any three elements satisfy the Clebsch-Gordan conditions.
As an example, I consider uniformly distributed spins for the same $J$ as in the above cases of coherent states (\fig{dsiid}).
Again, the difference with respect to the classical triangulation randomized around is of order $\mathcal{O}(10^{-2})$.
\begin{figure}
\includegraphics[width=7cm]{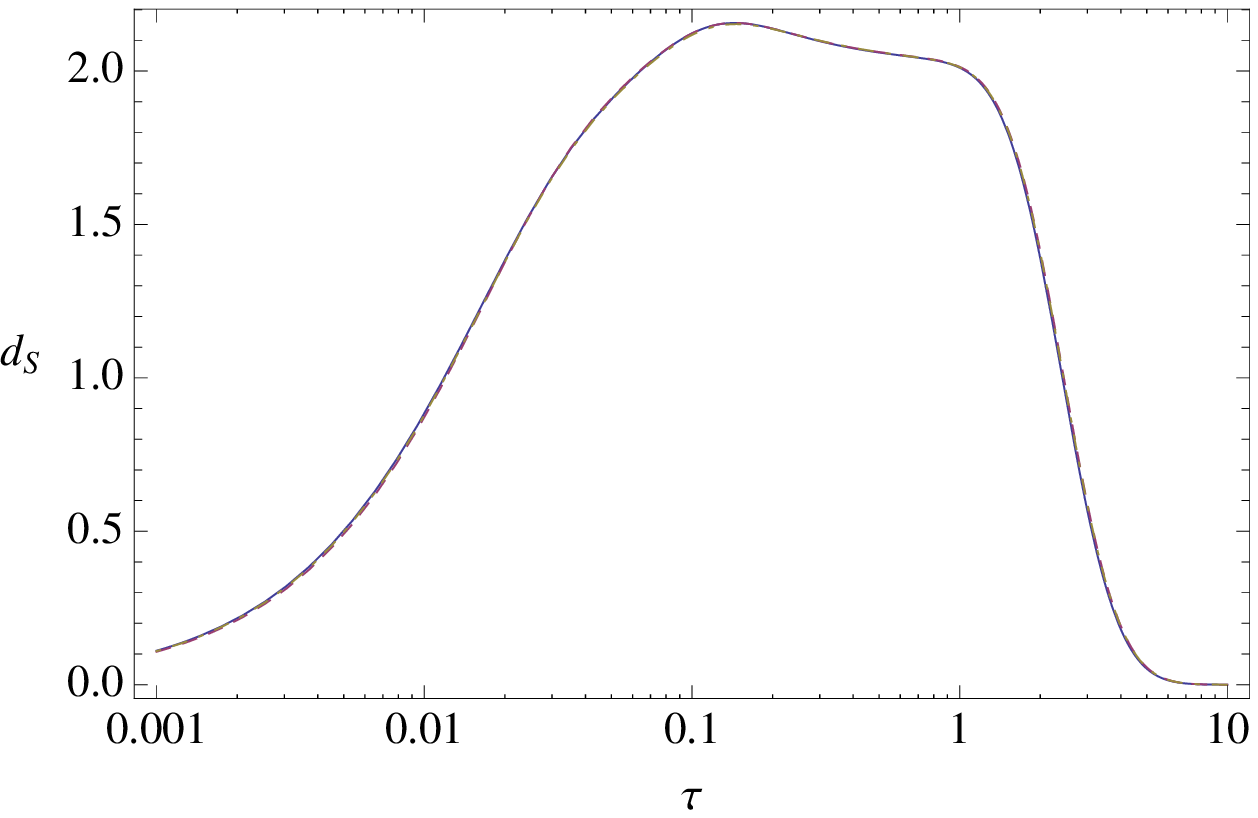}
\includegraphics[width=7.4cm]{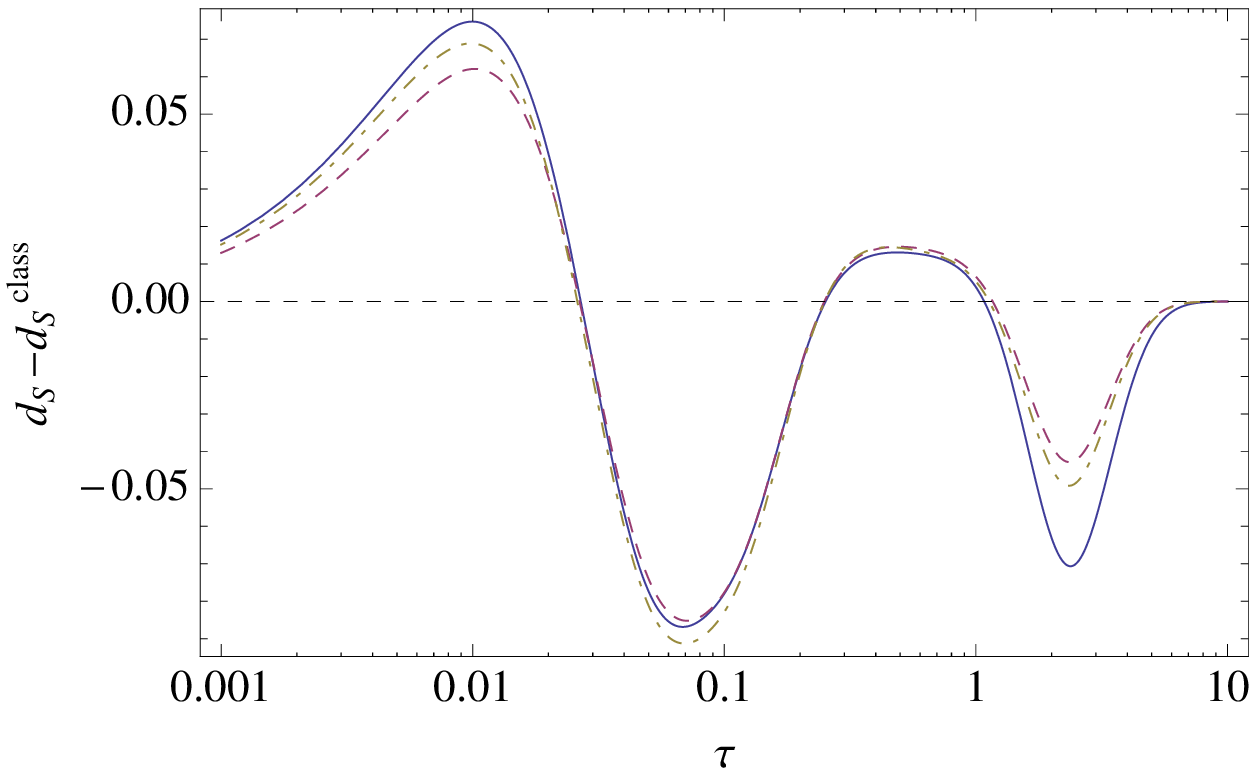}
\caption{Left: $\ds$ for a sum over seven random states with spins in $I_J$ for $J+1/2=16,32,64$ (solid, dashed and dot-dashed curve)  on the triangulation $\T_{2,4}$ ($\np2=288$). 
Right: deviation of $\ds$ from the classical case.}
\label{fig:dsiid}
\end{figure}

It is worth noticing that, in all the above examples of quantum states, the difference with the corresponding classical state is rather marginal. 
The key result of the calculations for coherent states on a given regular triangulation is that the quantum corrections are very small, since by varying spins $J$ and spread $\sigma$ one gets deviations from the classical spectral dimension only of order $\sim 10^{-2}\ds$.

\subsection{Summing semi-classical states \label{sec:summing-coherent-states}}

A basic feature of quantum mechanics is the superposition principle. 
The coherent and randomized states already are a typical example in this sense, as they are superpositions in the spin-network basis;
but one can probe the quantum features of spatial geometry also by constructing other kinds of superposition states.
One obvious strategy is to superpose coherent states themselves. There are various ways one could do so.
Possible choices would be sums in the coherent-state labels $J$, $K$ and $\sigma$,
\begin{equation}\label{sup1}
|\bg,\{c_{J_{\eb},K_{\eb},\sigma}\}\qket=\sum_{J_{\eb},K_{\eb},\sigma}c_{J_{\eb},K_{\eb},\sigma}|\psi_{\bg}^{J_\eb,K_\eb}\qket\,,
\end{equation}
or over different complexes and their dual graphs $\bg$:
\begin{equation}\label{sup2}
|\{\qsc_{\bg}\},J_{\eb},K_{\eb},\sigma\qket=\sum_{\bg}\qsc_{\bg}|\psi_{\bg}^{J_\eb,K_\eb}\qket\,.
\end{equation}
In the following, the focus is on superposed states on the same complex \eqref{sup1}. Superpositions over complexes \eqref{sup2} are the topic of the next \sec{superpositions}, and in a more general setting of \sec{dimensional-flow}.

In the case of superpositions on a fixed graph $\bg$ (\ref{sup1}), the expectation value of the heat trace does not simplify to a single sum over the expectation values of squared coefficients. 
Assuming trivial intertwiners,
\begin{align}\label{eq:htCrossterms}
\qbra \widehat{P(\tau)} \qket _{\bg} 
&= \sum_{J_\eb,K_\eb,\sigma} \sum_{J'_{\eb},K_{\eb}',\sigma'} \sum_{j_\eb,j'_\eb}  
c_{J_{\eb},K_{\eb},\sigma}^* ({\csc})^* c_{J'_{\eb},K_{\eb}',\sigma'}\cscp \qbra \bg,j_{\eb}|\Tr_\bg\rme^{\tau\widehat{\Delta}_\bg}|\bg,j_{\eb}'\qket 
\nonumber \\
 &= \sum_{J_{\eb},K_{\eb},\sigma}\sum_{J'_{\eb},K_{\eb}',\sigma'}c_{J_{\eb},K_{\eb},\sigma}^*c_{J'_{\eb},K_{\eb}',\sigma'}
 \sum_{j_{\eb}}\prod_{\eb\in\E_\bg} \rme^{-\frac{(J_\eb-j_\eb)^2}{2\sigma^2} - \frac{(J'_\eb-j_\eb)^2}{2\sigma'^2} - \rmi(K_{\eb}-K'_{\eb})j_{\eb}}\Tr_\bg\rme^{\tau\Delta_{\bg}(j_{\eb})}.
\end{align}
Assuming a fixed extrinsic curvature $K_{\eb}$ there is no phase causing interferences.
Furthermore, for sharply peaked coherent states the cross terms with $J_{\eb}\ne J'_{\eb}$ are suppressed. 
In this case, one would expect that 
\begin{equation}
\qbra \widehat{P(\tau)} \qket _{\bg}  \approx \sum_{J_{\eb},\sigma}\underset{\{j_{\eb}\}}{\sum}\left|c_{J_{\eb},\sigma}\right|^{2}\left[\prod_{l\in\bg}\rme^{-\frac{(J_\eb-j_\eb)^{2}}{\sigma^{2}}}\right]\Tr\,\rme^{\tau\Delta_{\bg}(j_{\eb})} \;
\end{equation}
is still a good approximation.

\begin{figure}
\centering{}
\includegraphics[width=7.5cm]{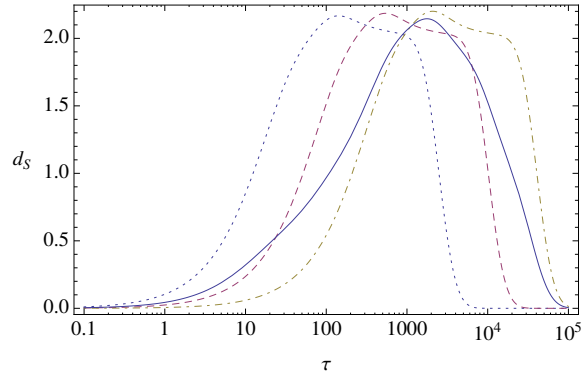}
\caption{$\ds$ for the superposition of coherent states $l(J)^2=J+1/2=16,32,64$ on a regular torus triangulation $\T_{2,4}$ ($\np2=288$) 
compared to the states summed over (dashed curve; \cf \fig{coherJs}).}
\label{fig:SupJs}
\end{figure}

When summing over coherent states peaked at different spins, it makes a crucial difference whether one considers them as peaked at classical geometries of different size or at the same classical geometry obtained by fixing the scale $a=a_J =l(J)$. 
In both cases, the spectral dimension turns out to be an average of the spectral dimensions of the parts summed over. 
However, in the first case (\fig{SupJs}) these individual profiles are shifted with respect to one another such that the superposition has dimension of order one only in the regime where most of them overlap. 
In the second case (\fig{SupJref}), all profiles have features at the same scales, so that the superposition yields a $\ds$ plot close to the classical geometry peaked at; the quantum correction is even less pronounced, \ie smaller than for the individual coherent states (\fig{coherJs}).

\begin{figure}
\includegraphics[width=7cm]{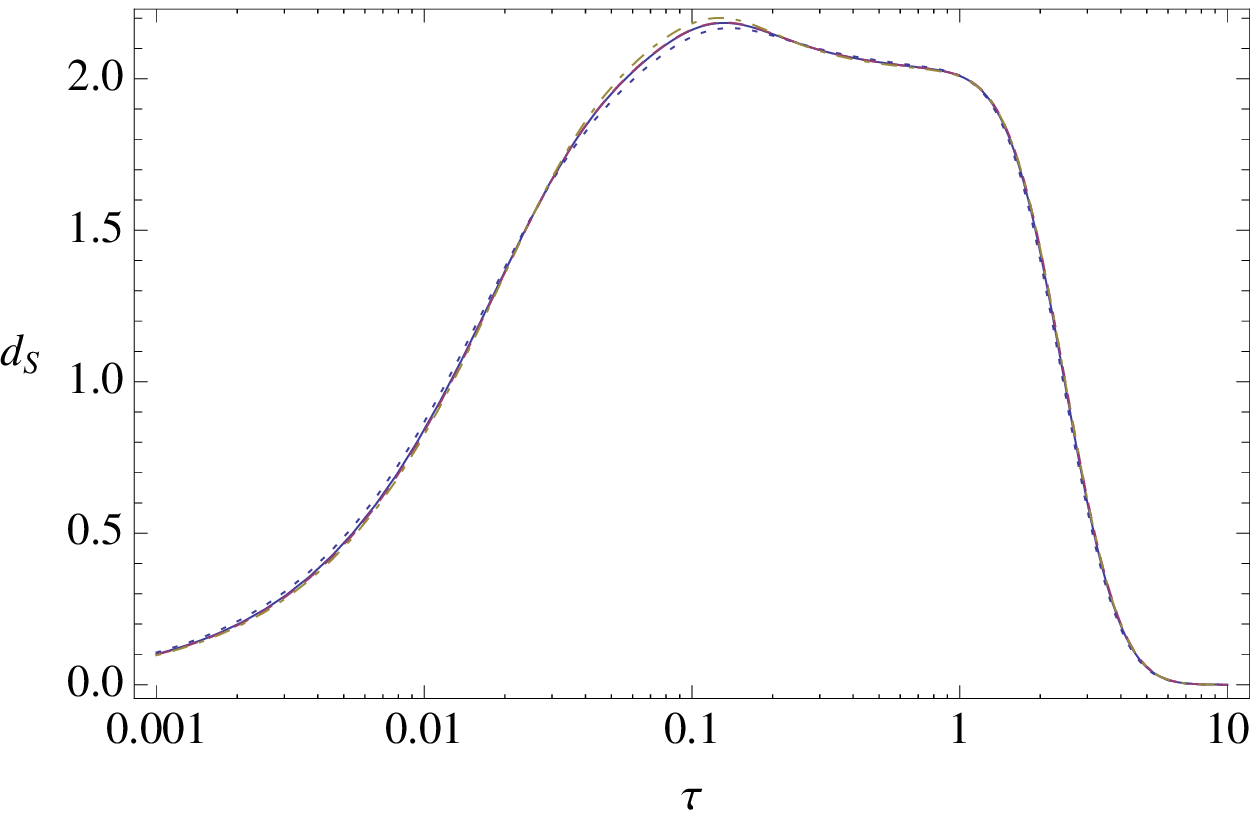}
\includegraphics[width=7.4cm]{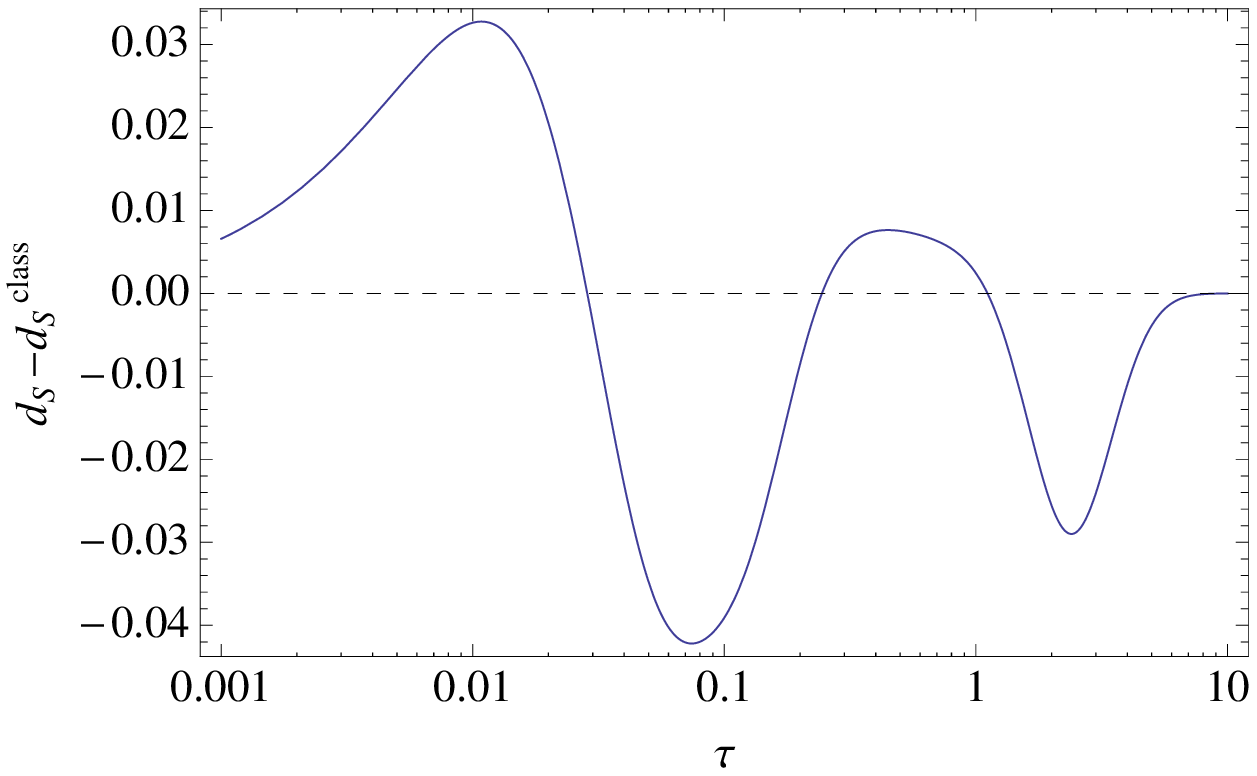}
\caption{Left: $\ds$ for the superposition of coherent states $l(J)^2=J+1/2=16,32,64$ on the  regular torus triangulation $\T_{2,4}$ ($\np2=288$) with rescaled Laplacian compared to the states summed over (dashed curve; see \fig{coherJs}). Right: deviation of $\ds$ from the classical case.}
\label{fig:SupJref}
\end{figure}


\subsection{Superpositions of complexes}\label{sec:superpositions}

The focus of this section is the spectral dimension of superpositions of states on distinct combinatorial complexes. 
In general, the heat trace expectation value of such states \eqref{sup2} is of the form
\begin{equation}
\qbra \widehat{P(\tau)} \qket \propto \sum_{\bg} \sum_{j_{\eb}} \left|\qsc_{\bg}\right|^{2}\left|\csc\right|^2 \sml \Tr_\bg \rme^{\tau\widehat{\Delta}_\bg} \smr \;.
\label{eq:SupGraphs}
\end{equation}
Comparing with the heat-trace expectation value of the superposition of coherent states on the same complex \eqref{eq:htCrossterms}, the sum here does not contain cross terms since the Hilbert space of states \eqref{hkin3} is a direct sum over graphs.
The numerical calculations in this section focus on states with trivial superposition coefficients $\qsc_\bg=1$, though I have numerically tested a few other cases as well.
More general classes of states can be tested with an analytically tractable model which I will present in \sec{dimensional-flow}.


Furthermore, for reasons of feasibility, here I approximate the coherent states $|\qs_{\bg}^{J_\eb,K_\eb} \qket$ to be summed over by the corresponding classical geometries.
This can be understood as the extreme case of minimal spread $\sigma$, that is, states sharply peaked at the intrinsic curvature but totally randomized in the extrinsic one.
But even more generally, the results in the last two sections have already shown that the effect of quantum fluctuations in coherent states is only of order $10^{-2}$. 
It is thus reasonable to expect that the effect of superposing truly quantum coherent states is reproduced by the superposition of discrete geometries associated with their peak values. 
Indeed, the results will show that effects of superposing graphs are of order higher than $10^{-2}$.
This justifies the approximation of the full sum, which in the case of superposition of graphs is considerably more challenging from a numerical point of view.

In effect, the setting here is very similar to CDT.
There, equilateral triangulations are the only dynamical degrees of freedom. 
The fact that the spectral dimension of the spacetime sum-over-histories is scale dependent in the CDT ensemble \cite{Ambjorn:2005fj,Benedetti:2009bi} is thus a consequence of (and eventually needs to be explained by) the way it is summed over a class of simplicial manifolds.
Although the setting here is restricted to kinematical spatial states and their superpositions, when summing with certain weights over simplicial manifolds one is in a setting quite comparable with CDT, and similar results could be expected. 
The main difference is that, while in CDT there is a precise description for the integration measure given by the exponential of the Regge action of equilateral triangulations, in the context of kinematical states of quantum geometry there is no unique prescription for how these states should be superposed on different complexes, in order to obtain some approximately smooth geometry. Therefore, the present investigation is somewhat explorative and driven mainly by the aim of unveiling generic features of superpositions of complexes.

\

Coming to the results of the spectral dimension calculations, as before it turns out to be crucial whether the superposed complexes are taken as purely combinatorial or as refinements of the same smooth geometry through rescaling the edge lengths (\fig{sumT2}). 
In both cases the result can be described qualitatively as an averaging of the spectral dimension plots of the single states summed over (depending on the weights $\qsc_{\bg}$)
The difference is that in the first case, this effect takes place at larger diffusion scales, in contrast with the small-diffusion-scale regime of the second.
\begin{figure}
\centering{}
\includegraphics[width=7.2cm]{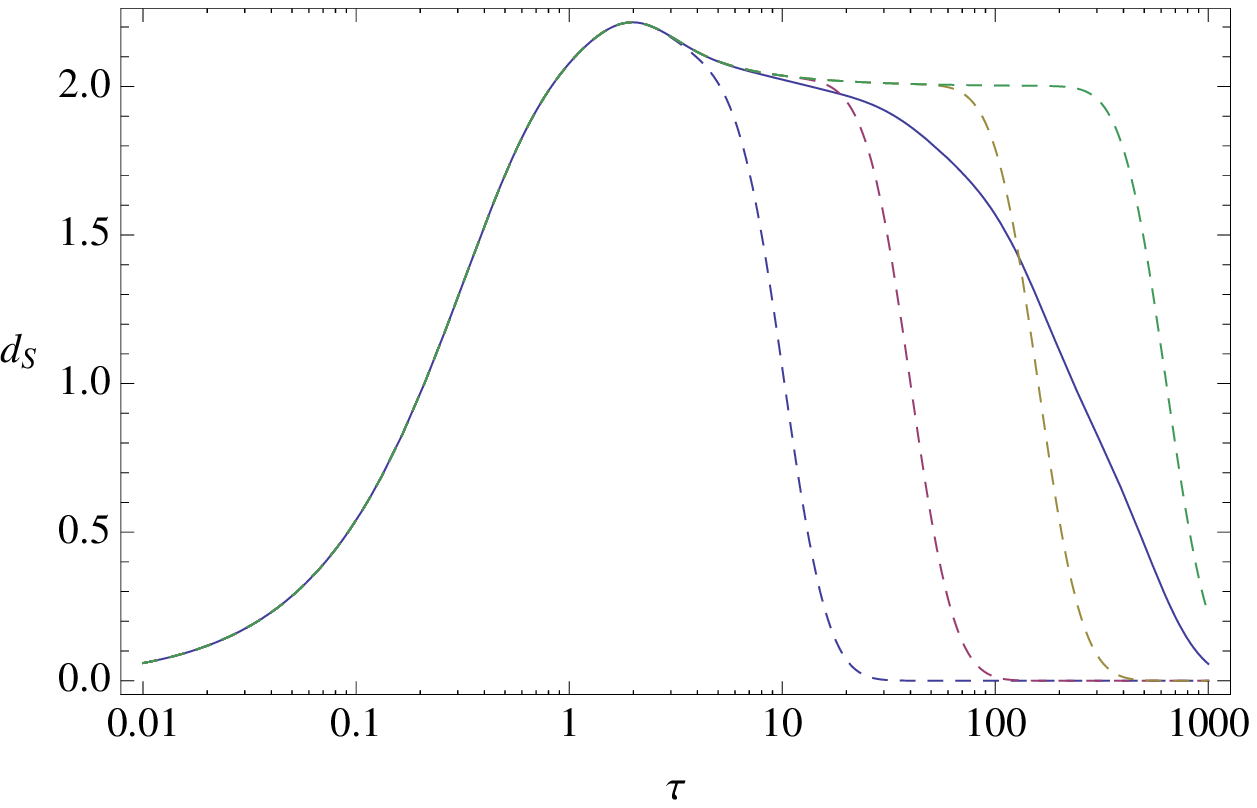}
\includegraphics[width=7cm]{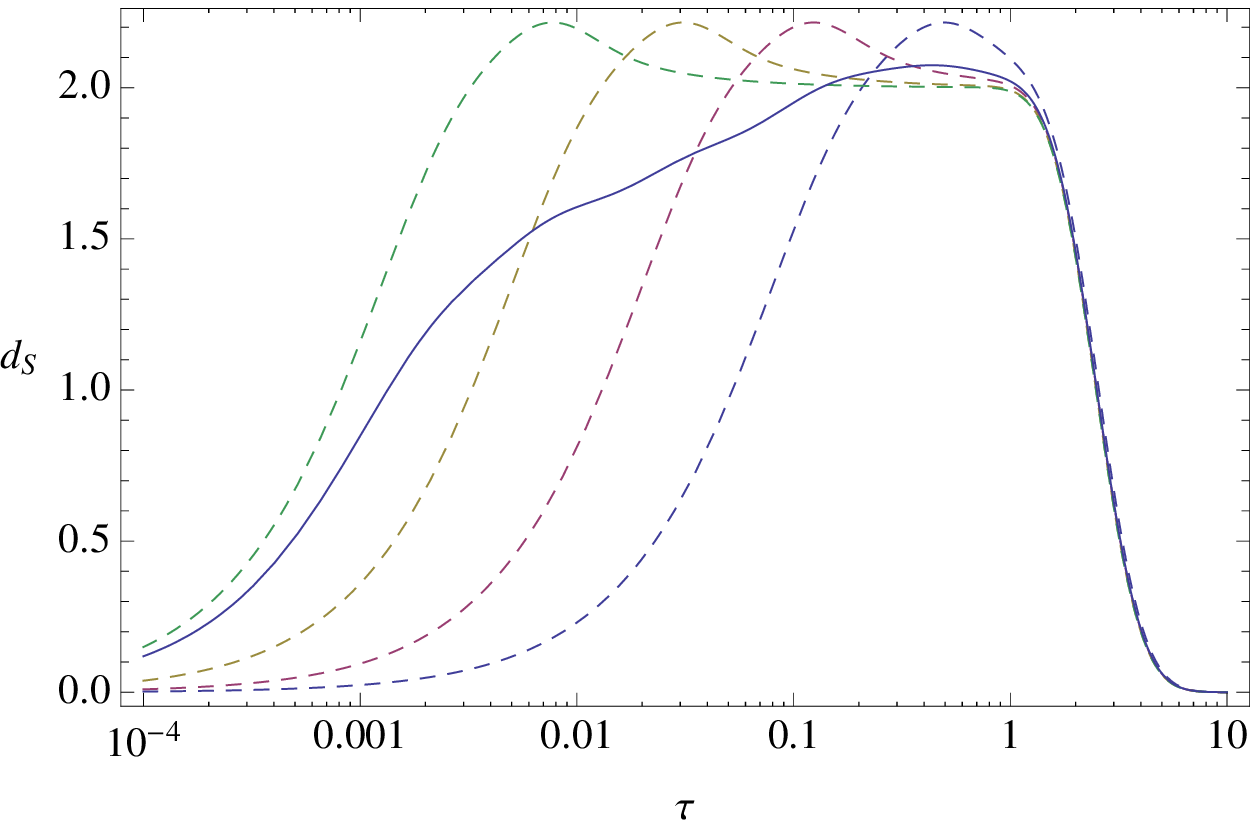}
\caption{Sum over the $T^{2}$ triangulations $\T_{2,\size}$ with $\size=3,6,12,24$ (included as dashed curves for comparison, \cf \fig{dsT2triang}) with trivial coefficients $\qsc_\bg=1$, unrescaled (left) and rescaled (right).
\label{fig:sumT2}}
\end{figure}

The second (rescaling) case shows some more interesting features, in particular when summing over a larger set of complexes.
Since all complexes summed over are triangulations of the same smooth geometry, it can be interpreted as a special case of a semi-classical state incorporating a (kinematical) continuum approximation.
Though in the numerical context here, it can only be implemented up to some finite order, though. 
A first interesting feature is that the discreteness artefact of a peak does not appear for these states.
Calculating not only the superposition of a few rescaled triangulations but of all regular triangulations $\T_{2,\size}$ of the type described in section \sec{ds-simplicial} up to some maximal size $\size_{\max}$ (\fig{sumT2pmax}),
the result is a more extended plateau. 
Thus, from these calculations one would expect the appearance of a plateau at the topological dimension for sums over more and larger triangulations, but without the discretization effect of a peak at the characteristic lattice scale.
Indeed, this numerical result will be confirmed by the far more general analytic model in the next \sec{dimensional-flow}.

\begin{figure}
\centering{}
\includegraphics[width=7.5cm]{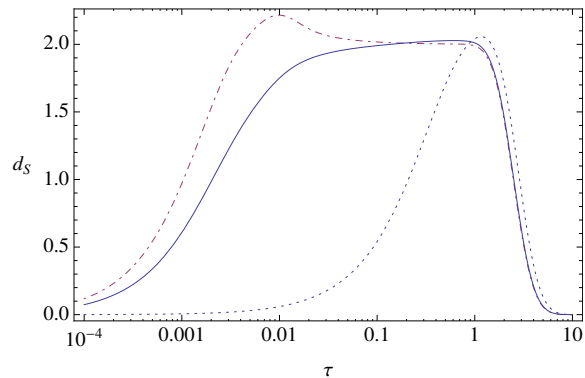}
\caption{Sum over all rescaled regular $T^{2}$ triangulations $\T_{2,\size}$ of size $\size=3,4,\dots,43$ (solid curve) and, for comparison, the individual $\size=3$ and $\size=43$ cases (dotted and dot-dashed curve).
\label{fig:sumT2pmax}}
\end{figure}

In the explorative spirit of this section, the same can finally be done for the randomly subdivided triangulations. The effect of summing is, again, an averaging of the $\ds$ profiles.
Qualitatively, these are quite different from the regular triangulations and one can hardly conjecture any more specific properties for the general case of larger superpositions.
\begin{figure}
\centering{}
\includegraphics[width=7cm]{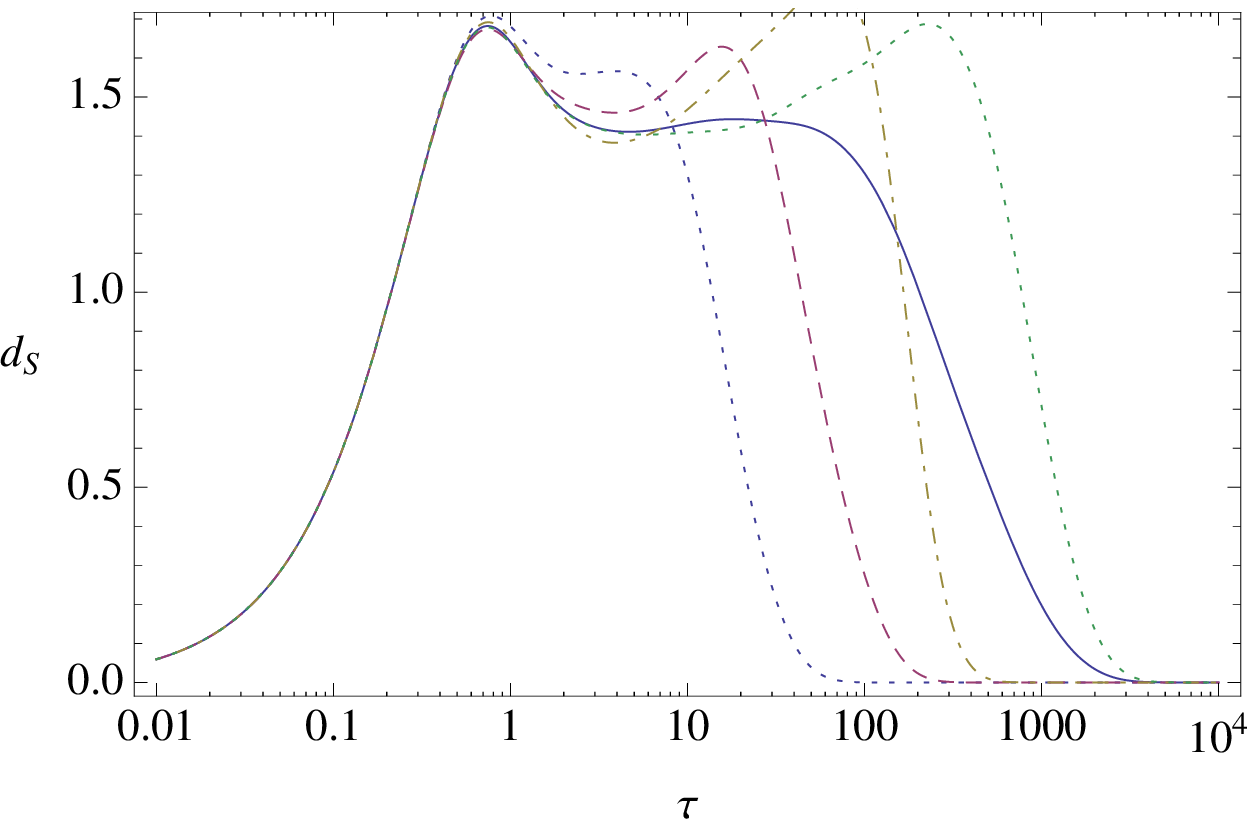}
\includegraphics[width=7cm]{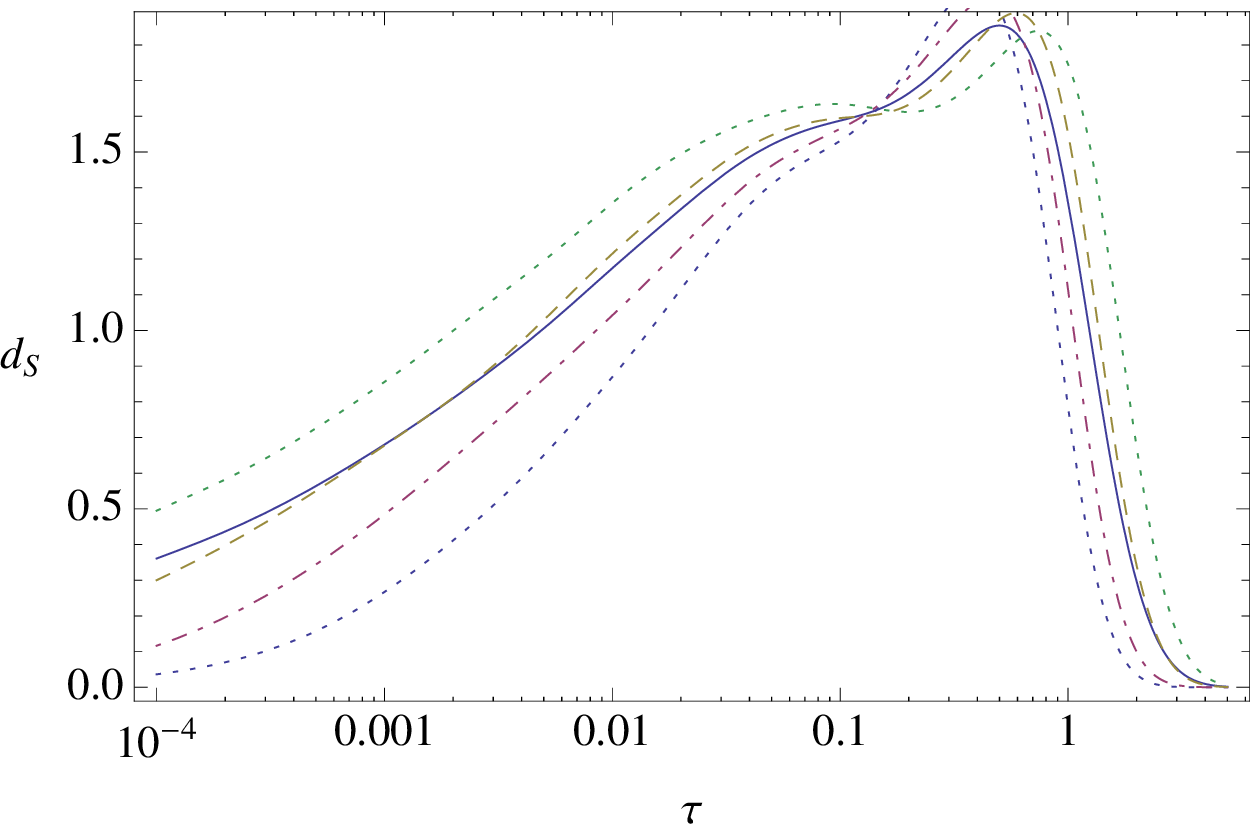}
\caption{Sum over randomly subdivided $T^{2}$ triangulations (dashed line, for comparison; see \fig{dsT2subran}), unrescaled (left) and rescaled (right).
\label{fig:dsT2randomsum}}
\end{figure}


\

In this \sec{coherent-states}, I have investigated the spectral dimension of coherent states and superpositions of these on toroidal complexes. 
Quantum corrections turn out to be small in the small-spin regime (\fig{coherJs}), for large spreads (\fig{coherSigmas}) and even for total randomization in a given spin interval (\fig{dsiid}). 
On a fixed complex, superpositions of coherent states peaked at the same geometry smoothen out the quantum deviations even more (\fig{SupJref}). 
Only superpositions of states peaked at different geometries show a distinct behaviour. This happens when the spectral dimension of the individual states is noticeably different from one another (\fig{SupJs}).

On the other hand, the superposition of states on distinct combinatorial complexes results in a more interesting behaviour of the spectral dimension. 
For a superposition of torus triangulations approximating the same smooth geometry but of different combinatorial size, I find that the discreteness artefact of the the peak, appearing in classical cases, disappears (\fig{sumT2}) and a plateau with the topological dimension is obtained (\fig{sumT2pmax}).

Still, there are no hints for an effective smaller dimension at smaller scales in one superposition state. The running of the spectral dimension in the profiles of the figures is, in fact, mainly due to discreteness artefacts ($\ds\to 0$ at small $\tau$) and topological effects ($\ds\to 0$ at large $\tau$).
On the other hand, superpositions of 1-3 Pachner subdivisions have a $\ds$ plateau at a height smaller than the topological dimension (\fig{dsT2randomsum}). 
Summarizing, the averaging effect stemming from superposition states is the only manifestation of additional geometric data in the profile of the spectral dimension of the states investigated here, which would be otherwise reproduced by a classical triangulation.

Nevertheless one has to note that, due to the involved and expensive numerical techniques used for the calculations, the class of states feasible to calculate their spectral dimension is limited in a twofold way.
On the one hand, it has been possible to consider states based on complexes only up to a combinatorial size of $\np2 = \Ocal(10^4)$ cells.
On the other hand, in the superposition over complexes only simple choices of superposition coefficients such $\qsc_\bg=1$ have been chosen.
Both limitations can be overcome by a simple analytic model which covers a much broader class of applications than semi-classical (2+1)-dimensional LQG states. 
Its surprising results are the topic of the next section.


\section{Dimensional flow in discrete quantum geometries}\label{sec:dimensional-flow}

The aim of this section is to present a special class of superpositions of discrete quantum geometries with a dimensional flow, characterized by a real-valued parameter $\alpha$. 
This parameter controls the scale-dependent values taken by the spectral dimension, and therefore the dimensional flow. 
Such states can be considered as a generalization of the LQG superposition states \eqref{sup1} and \eqref{sup2} of the previous section to arbitrary dimension. 
Motivated by the detailed numerical computations described there, the states considered here are based on a number of assumption such that the analytical results of \sec{classical} can be used.

For simplicity, the superposition states are sums over regular complexes corresponding to hypercubic lattices to which a single quantum label is assigned uniformly to all cells of a certain dimension. 
Such states occur indeed in the Hilbert space $\hkin$ of the discrete quantum gravity theory considered here (chapter \ref{ch:dqg}).
Because of the uniform labelling, these superpositions are also similar to the discrete geometries in CDT, although the complexes are not considered as regularization tools for physically smooth geometries here but as fundamentally discrete structures with their own physical interpretation. Contrary to the CDT setting, I interpret the superposed complexes as defining quantum gravity states, not histories, and the coefficients in the superposition have no immediate dynamical content. However, it is important to point out that this interpretation enters only minimally in the actual calculations and it could be generalized.

Using the known analytic expression \eqref{ht-hlinf} for the spectral dimension of single members in the superposition, I compute numerically superpositions over up to $10^6$ discrete geometries. On these grounds, one finds strong evidence for a dimensional flow, characterized by the parameter $\alpha$.

For these states one can furthermore use the analytic solutions for the walk dimension and Hausdorff dimension presented in \sec{hypercubic} and perform again numerical calculations of superpositions. 
However, these observables do not display any special properties for superpositions as compared to states defined on fixed complexes.

I will start  in section \ref{lattice-superpositions} with a brief discussion of the class of states under consideration, presenting then the setup for the spectral dimension calculation and discuss its result in \sec{spectral-flow}. 
Similarly, the outcomes of the calculation of the Hausdorff and walk dimension will be the topic of \sec{walk-Hausdorff}. 
In \sec{fractal-structure}, I will conclude with a discussion of these results with a particular focus on the question of a fractal structure.


\subsection{A general class of superposition states\label{lattice-superpositions}}

Let me first explain in detail the construction of superposition states of interest, and the calculation of their spectral, walk and Hausdorff dimension.


%
%


Generalizing the case of spin-network states in the LQG Hilbert space, here I denote as a \emph{discrete quantum geometry} a state $\rjc$ which is given by an assignment of some quantum numbers $j_c$ to a certain subset of cells $c\in\cm$ of a combinatorial complex $\cm$, 
diagonalizing volume operators of these cells
\[\label{spectra}
\widehat {V_{c'}^{(\p)} } \rjc \propto l^\p(j_{c'}) \rjc \,.
\]
Thus, spin-network states in LQG are an example of such states, based on the 1-skeleton of the dual complex $\cm^{\star}$, with the $j$'s identifying irreducible representations of $SU(2)$. 
They diagonalize length operators \eqref{length-spectrum} in $\std=2+1$ and area and volume operators \eqref{area-spectrum} in $\std=3+1$ spacetime dimensions as explained above.
Generic quantum-geometry states are superpositions of the discrete quantum geometries $\rjc$. 

In the following, I will restrict to superpositions with nonzero coefficients only for states $\rj$ 
labelled by a single quantum number $j_c = j$ for all cells. 
Thus, one can consider the individual states $\rj$ as corresponding to equilateral lattices.
Generic superpositions are then of the form 
\[\label{sups}
|\qs\rangle = \sumint_{j,\cm} \qsc_{j,\cm} \rj \,.
\]
In particular I am interested in states constrained to a fixed overall volume $V_0$ which can be interpreted as continuum states:
\[\label{fixedvolume}
|\qs,V_0\rangle = \sumint_{j,\cm} \qsc_{j,\cm} \,\delta(\lj \widehat{V} \rj, V_0) \,\rj \,,
\]
where the delta is a Kronecker delta. 

Furthermore I restrict the sum to the hypercubic lattices $\cp=\hl\sd$ of spatial dimension $\sd$ and size $\size$ introduced in \sec{hypercubic}.
In this case, the fixed volume condition is explicitly
\[
V_0  = \ljb \widehat V \rjb \propto \size^\sd\, l^\sd(j)\,,
\]
which fixes the lattice size $\size = \size(j)$ for a given $j$.

In general, there are three scales involved in such superposition states, denoted
\[\label{states}
\rsup := \sumint_{j=\jmin}^{\jmax} \qsc_j |j,\hln{\sd}{\size(j)} \rangle\,,
\]
summing (or for a continuous label $j$, integrating) over a finite range from  $\jmin$ to $\jmax$:
A minimal length scale $l(\jmin)$, an intermediate scale $l(\jmax)$ and the overall volume size $V_0^{1/\sd} \propto N(\jmin)l(\jmin) = N(\jmax)l(\jmax)$.
Note that a finite volume $V_0$ bounds also possible cutoffs $\jmax$ (since $N$ is a positive integer).


One can also consider the limit of noncompact geometries $N(\jmin) \ra \infty$, where all complexes in the superposition state \eqref{states} converge to the infinite lattice ${\hlinf\sd}$.
Thus, such a state is technically the same as a superposition on the fixed complex ${\hlinf\sd}$ (similar to \eqref{sup1} in \sec{summing-coherent-states}), although the physical interpretation is different. 
Due to the combinatorial simplicity, the results of infinite-size calculations in \sec{classical} can be directly applied to the finite-volume case. 

Having defined the special class of superposition states, I can move on to the evaluation of the geometric observables of interest, \ie the effective dimensions.


\subsection{Dimensional flow in the spectral dimension \label{sec:spectral-flow}}

For the discrete quantum geometries $\rjc$ it is reasonable to assume that they are eigenvectors of $\widehat \Delta_\cm$, based on the definition of these labels \eqref{spectra} and on then general form of $\widehat \Delta_\cm$ \eqref{quantum-laplacian}.
In the case of $(2+1)$-dimensional LQG states I have shown this in detail (\sec{defLQG}).
Then, according to \eqref{ht-expansion}, the heat-trace expectation value of the states \eqref{sups} of interest here is
\[\label{eq:dsexpand}
\langle \widehat{P(\tau)}\rangle _{\qs} 
 =  \sumint_{j,\cm} \left|\qsc_{j,\cm} \right|^{2} \Tr_\cm\, \e^{\tau \lj \widehat\Delta_\cm \rj}\,.
\]

One simplifying assumption is however needed in order to proceed with systematic computations on extended complexes. 
Motivated by the results in the specific LQG case in \sec{coherent-states}, assume that the expectation values of the coefficients of the Laplacian $\widehat\Delta_\cm$ scale as
\[\label{assumption}
\lj \widehat\Delta_\cm \rj_{ab} \propto l^{-2}(j)\,(\Delta_\cm)_{ab}\,,
\]
where $\Delta_\cm$ is the combinatorial Laplacian \eqref{discrete-laplacian} on the complex $\cm$.
This assumption is sensible if the Laplacian can be expressed as a function of the volumes \eqref{spectra} which is possible in most quantum gravity cases (\cf appendix \ref{sec:classical-expressions}).
A similar ansatz is, in fact, made in \cite{Modesto:2009bc}, although in that work it is not justified on the basis of a detailed analysis of the underlying graph and on the complete expression for the Laplacian, such as the one presented here \eqref{scalar-laplacian}.

The spectral dimension on the lattice superposition states has now a simple expression.
Under the assumption \eqref{assumption}, the expectation value of the return probability further simplifies to
\[
\langle \widehat{P(\tau)}\rangle _{\qs} \propto \sumint_{j,\cm} \left|\qsc_{j,\cm} \right|^{2} \Tr_\cm\, \e^{\tau l^{-2}(j) \Delta_\cm} .
\]
This expression can be computed most efficiently considering the limit of infinite lattices, for which the analytic solution \eqref{ht-hlinf} for the heat trace is available. 
Thus, the spectral dimension of a state $\rsup$, \eqref{states}, in the limit $\size(\jmin)\ra\infty$, is given by the scaling of
\[\label{htsuperposition}
\langle \widehat{P(\tau)} \rangle_\ssup
\propto \sumint^\jmax_{j=\jmin} \left|\qsc_j \right|^2 \left\{\e^{l^{-2}(j) \tau} I_0[l^{-2}(j)\, \tau]\right\}^\sd .
\]

\

For various classes of coefficient functions $\qsc_j$, values of spatial dimension $\sd$ and cutoffs $\jmax$ I have carried out numerical calculations of the spectral dimension using asymptotic power-law spectra 
\[\label{beta}
l(j) \simeq  j^\beta\,,
\]
where $\beta>0$ as usual in LQG, \eg \eqref{length-spectrum} and \eqref{area-spectrum}.
In the examples presented here $\jmin =1$, though calculations with lower cutoffs of the same order (\eg  $\jmin =1/2$) have been computed with similar results.
Notice also that the same finite minimal value for the geometric spectra could be obtained in correspondence with a quantum label $j=0$, for choices of quantization map that give a nonzero value for $\csu$ in \eqref{area-spectrum}. 

A dimensional flow in the spectral dimension is found for a class of coefficients with power-law functions,
\[\label{gamma}
\qsc_j \propto j^\gamma \,.
\]
Depending on the range of values of the parameter
\[\label{alpha}
\alpha := -\frac{2\gamma+1}{\beta}\,,
\] 
the dimensional flow is found for  $0<\alpha<\sd$.
More precisely, I find the following behaviour for the  spectral dimension of the state under consideration in this parameter regime:
\begin{itemize}

\item[(a)] In the IR, \ie for large length scales $\tau\gg l^2(\jmax)$, $\ds(\tau)=\sd$ (\fig{SuperposD}). This is of course a consistency check for the validity of the formalism, since at large scales one recovers the topological dimension of the space the quantum states are supposed to represent. It is however already a nontrivial test, as identifying quantum states with the right semiclassical continuum properties at large scales is no small task in background-independent quantum gravity. 

\begin{figure}
\includegraphics[width=7cm]{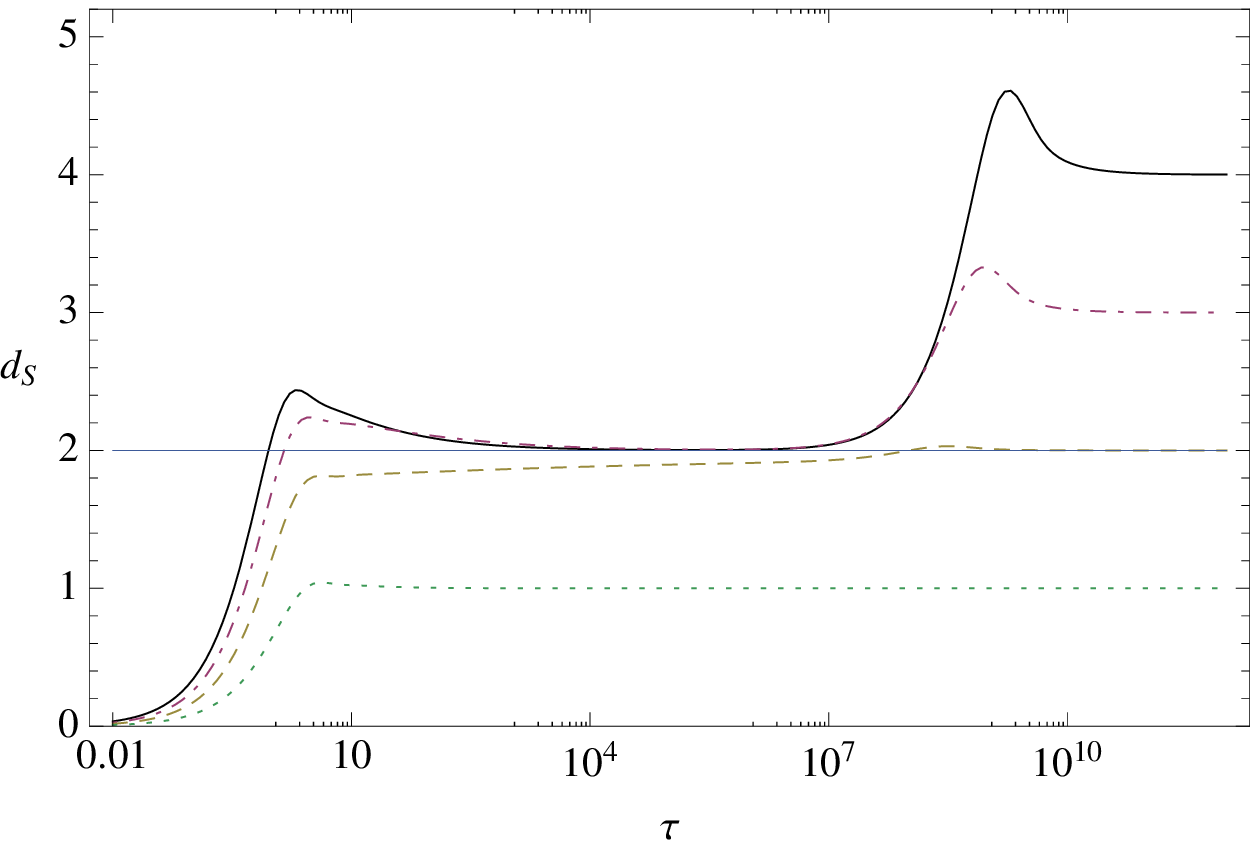}\,\,\,
\includegraphics[width=7.2cm]{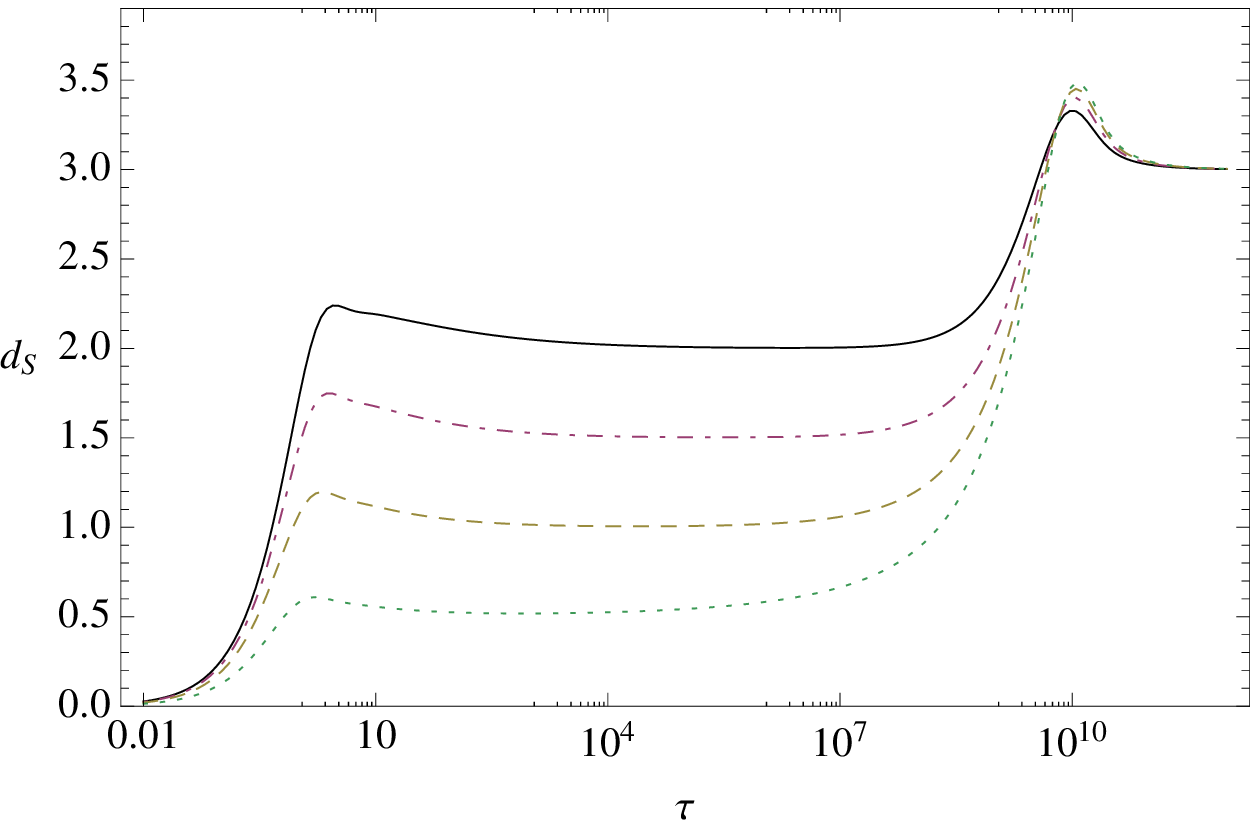}
\caption{Left: spectral dimension of a superposition with $\alpha = 2$ in $\sd=1,2,3,4$ (dotted, dashed, dot-dashed, solid curve) with cutoff $\jmax = 10^4\sd$.\\
Right: spectral dimension of superpositions with $\alpha = 1/2, 1, 3/2, 2$ (dotted, dashed, dot-dashed, solid curve) in $\sd=3$ with cutoff $\jmax = 10^5$.
}
\label{fig:SuperposD}
\end{figure}

\item[(b)] Below the smallest lattice scale, \ie for $\tau\ll l^2(\jmin)$, $\ds(\tau) = 0$.  
This is the usual discreteness effect also found for individual lattices (\cf \sec{classical}), which remains at the Planck scale for discrete spectra induced by holonomies valued in compact groups \cite{Rovelli:1995gq,DePietri:1996en,Ashtekar:1997bn,Ashtekar:1997we,Achour:2014gr}. For noncompact groups, spectra are typically continuous and no volume discreteness effect at Planck scale occurs, as $\jmin \ra 0$ \cite{Freidel:2003kx}.

\item[(c)] Between these scales, there is a plateau with value $\ds(\tau) = \alpha$ (\fig{SuperposD}). This plateau indicates a regime in which the effective dimension is physically smaller than the topological one, and thus a proper dimensional flow. 
In the light of the results in the previous \sec{coherent-states}, which are performed on a particular instance of these quantum states, but without considering large superpositions of lattice structures, this result can be regarded as a genuine quantum effect stemming from the superposition of states $\rj$ with geometric spectra on different scales and based on complexes of different size. 
It is interesting that, at such intermediate scales, the effective dimension depends on the superposition coefficients but is independent from the topological dimension.


\item[(d)] In particular, for infinite superpositions ($\jmax\to\infty$) this plateau takes the value $\alpha$ and extends indefinitely (\fig{SuperposJ}). Formally, one can express this behaviour by
\[
\Delta\tau\big|_{\ds=\alpha} \underset{\jmax\to\infty}{\longrightarrow} \infty\,.
\]
Notice that this only means that the topological dimension $d$ is obtained further away at large $\tau$. Physically, one never takes the infinite limit in practice: for large spin labels, the plateau is long but has finite extension $\Delta\tau$.

\begin{figure}
\includegraphics[width=7cm]{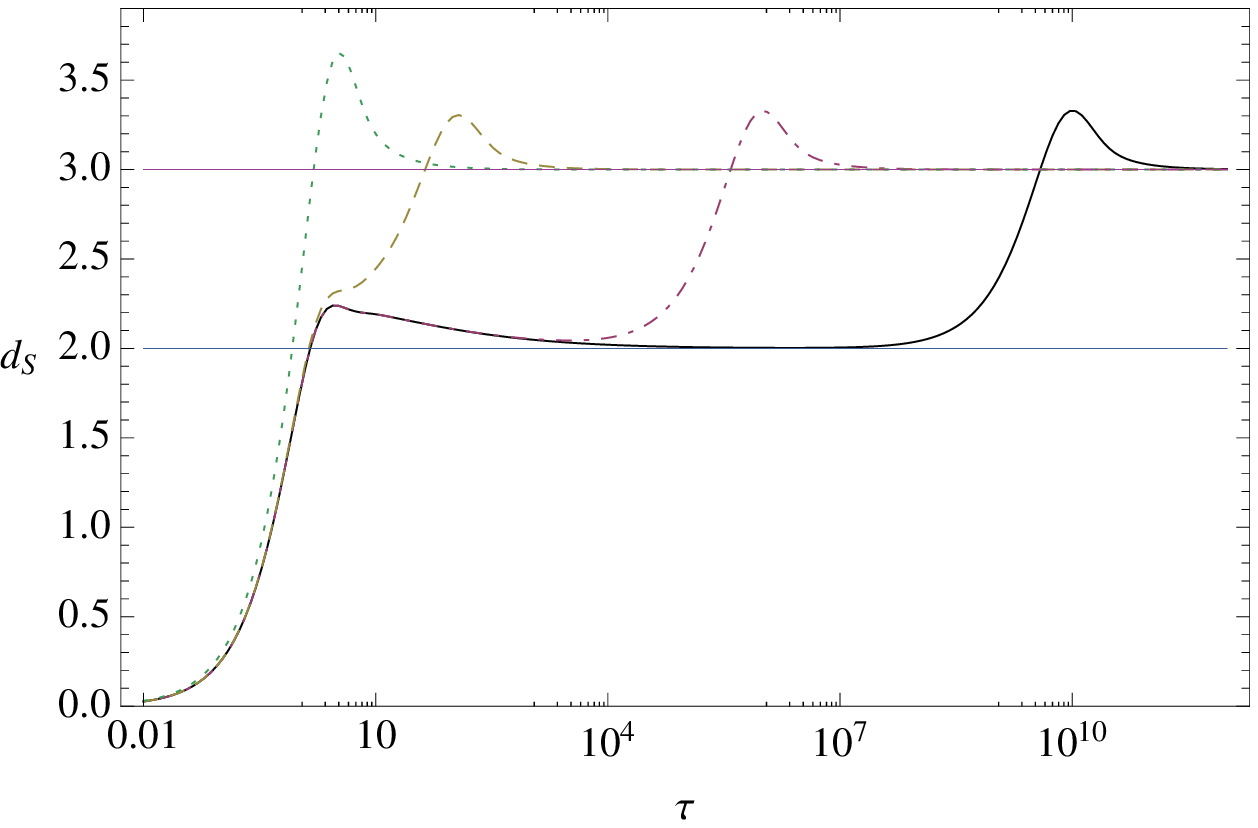}\quad
\includegraphics[width=7cm]{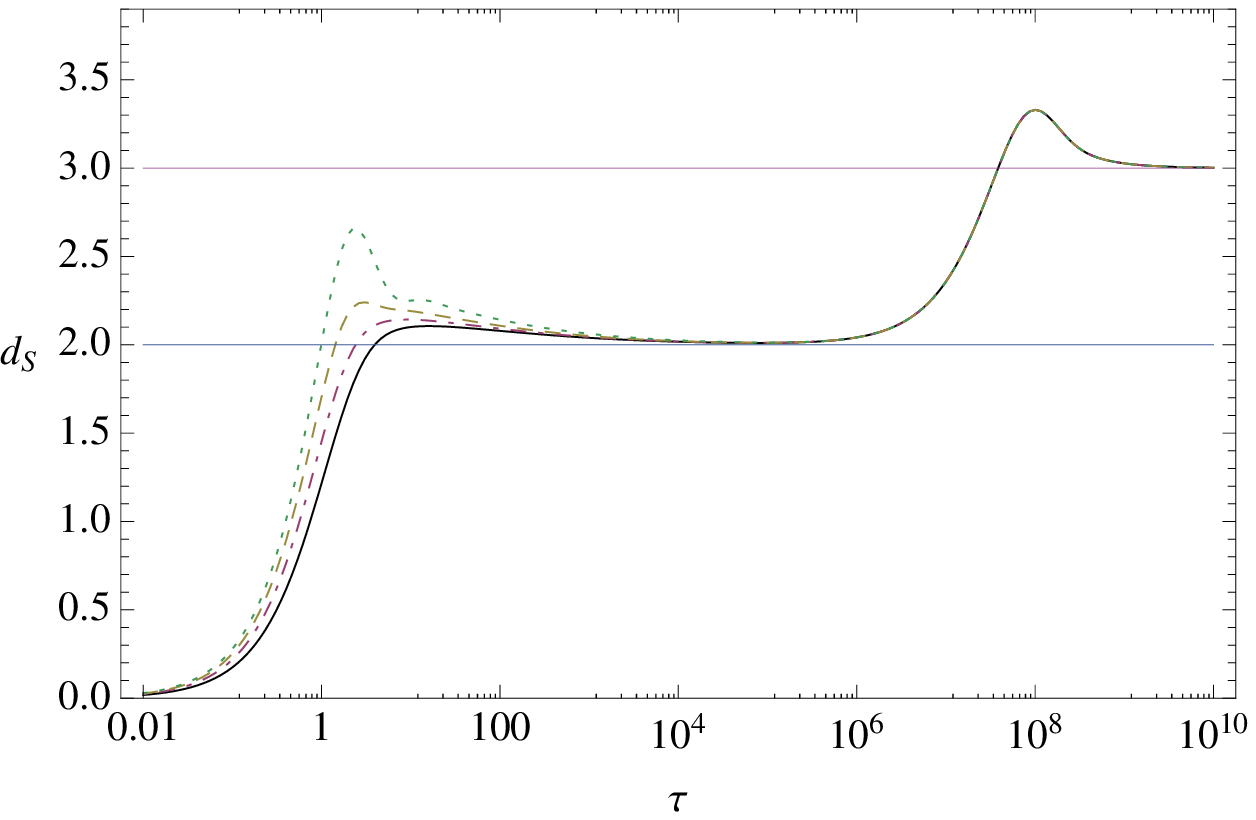}
\caption{Spectral dimension of superpositions with $\alpha = 2$ in $\sd=3$. \\
Left: for cutoffs $\jmax = 1,10,10^3,10^5$  (dotted, dashed, dot-dashed, solid curve).\\
Right: 
summing over positive $j\in \frac 1 q \N$ up to $\jmax =  10^4$ for $q=1/2,1,2,10$  (dotted, dashed, dot-dashed, solid curve).
}
\label{fig:SuperposJ}
\end{figure}

\item[(e)] Moreover, these results are independent of the spacing of the quantum labels $j$. 
Summing over $j\in\frac 1 q \N$ for some $q\in\Q$ slightly changes the results only at the scale $l(\jmin)$ (\fig{SuperposJ}).
Therefore, neither the IR nor the UV regime depends on the spacing of the state label $j$. The numerical calculations show, in particular, that this should also be true in the limit $q\ra\infty$, \ie for positive real $j$.
\end{itemize}


\item For $\alpha<0$, no superposition effect occurs and the profile of the spectral dimension equals approximately the one of the single state $|\jmax,{\hlinf\sd}\rangle$:
\[
\ds^{\ssup}(\tau) \approx \ds^{\jmax,{\hlinf\sd}}(\tau).
\]
This result is numerical, for which, at present, a complete analytical or physical understanding is lacking.
Nevertheless, there is an intuitive explanation for $\ds\approx\alpha$ at sufficiently small scales, in the range $0<\alpha<d$.
On a continuous medium, the case $\alpha<0$ would correspond to an unphysical one with negative dimension. This situation is meaningless both in the conventional diffusion interpretation of the spectral dimension (where the probe would do ``less than not propagating'') and in the resolution interpretation of \cite{Calcagni:2014ig,Calcagni:2014uv}.
In the latter, the return probability $P(\tau)\sim (\sqrt{\tau})^{-\ds}\sim \ell^{-\ds}\sim ({\rm res})^{\ds}$ is the probability to find the probe anywhere when the geometry is probed at scales $\ell$, i.e., with resolution $1/\ell$. 
For positive $\ds$, this probability decreases with the resolution: if $1/\ell$ is too small, there is a chance that the probe is not seen at all. 
On the other hand, a negative $\ds$ implies that the coarser the probe, the greater the chance to find it somewhere. 
In the current case, this pathological behaviour is screened by discreteness effects and $\ds$ is saturated by the lattice with labels $\jmax$. 

For $\alpha>\sd$, no superposition effect occurs and the profile of the spectral dimension equals approximately the one of the single state $|\jmin,{\hlinf\sd}\rangle$,
\[\label{dsmin}
\ds^{\ssup}(\tau) \approx \ds^{\jmin,{\hlinf\sd}}(\tau).
\]
In the continuum limit, $\alpha>d$ would imply a spectral dimension larger than the ambient space. Similarly to the previous case, for an intuitive explanation, one has both the diffusion and the resolution interpretation at hand. In the conventional diffusion interpretation of the spectral dimension, the case $\ds>d$ may be regarded as physical: the probe effectively sees more than $d$ dimensions and tends to superdiffuse. In the resolution interpretation, the probability of finding the probe somewhere grows more steeply than for the normal case (Brownian motion) and probes with large resolution (small scales $\ell$) become even more effective. However, in the present quantum framework there is a limit to which one can probe the microscopic structure of geometry: volume spectra are discrete with minimum eigenvalue determined by $\jmin$. Again, the variation of the spectral dimension is dominated by lattice effects, this time governed by the lower cutoff in the spin labels.

A partial understanding of the results with $0<\alpha<d$, in particular concerning the dependence of the UV value of $\ds$ on the powers $\beta$ and $\gamma$ in \eqref{alpha}, is provided by the following rewriting of the heat trace \eqref{htsuperposition}.
A redefinition of variables $k(j) := l^{-\alpha}(j)$ demands a change of summation-integration measure by
\[
\frac{\d k}{\d j}  = \frac \d {\d j} l^{-\alpha}(j)
=  -\alpha \frac{\d\ln l(j)}{\d j}  l^{-\alpha-1}(j)\,.
\]
In particular, for the power-law spectra \eqref{beta} and the definition of $\alpha$ \eqref{alpha}
\[
\frac{\d k}{\d j} =  -\alpha \beta \, j^{-\alpha \beta -1} \os{\text{\scriptsize \eqref{alpha}}}{=} (2\gamma+1)\, j^{2\gamma}
\]
which is proportional to $|\qsc_j|^2$ for the power-law coefficients \eqref{gamma}.
Thus, the heat trace on these superpositions is a uniformly weighted sum in the $k$-variable, over the range corresponding to \eqref{htsuperposition},
\[\label{ksum}
\langle \widehat P(\tau) \rangle_\ssup  \propto  \sumint_k \left[\e^{-k^{2/\alpha}\tau} I_{0}({ k^{2/\alpha}}\tau)\right]^{d}.
\]
Therefore, genuine dimensional flow comes from a subtle balancing of $\sd$ and $\alpha$ in this expression, while a negative $\alpha$ yields just a dominant $k_{\max} = k(\jmax)$ contribution in the sum. 
Indeed, I have also calculated the spectral dimension directly from \eqref{ksum} for various values of $\sd$, $\alpha$ and summing ranges of integer $k$'s, obtaining qualitatively similar results as discussed above for \eqref{htsuperposition}. As a consequence, dimensional flow has some dependence on the form of the spectrum \eqref{beta} but only in combination with appropriate superposition coefficients.

\ 

Still maintaining the power-law spectrum \eqref{beta} (which is the most reasonable assumption, consistent with known results in LQG and related approaches), I have calculated the spectral dimension for various other classes of coefficient functions. In most cases, there are no surprising results:
(a) For example, exponential coefficients $\qsc_j\propto \e^{bj}$ let either the maximal state $\jmax$ dominate when $a>0$, or the minimal one $\jmin$ when $a<0$. 
(b) Gaussian coefficients, on the other hand, result in a dominance of the $j_0$ at which they are peaked. (c) Trigonometric functions add some oscillations to $\ds^{\jmax,{\hlinf\sd}}$ in the intermediate regime, depending on the relation of the periods to the spacing of $j$ in the sum. In all these cases, therefore, the overall behaviour of the spectral dimension is the same as that found for coefficients given by simple powers. 

More interesting is the case of coefficients which are linear combinations of power functions in $j$. Then one finds, for their asymptotic behaviour $\qsc_j\sim j^\gamma$, the same effect as for power-law coefficients. In particular, if there are several regimes with different approximate scaling $\gamma_1, \gamma_2, \dots$, one obtains plateaux in the spectral dimension plot of different values $\alpha_1, \alpha_2, \dots$ accordingly. An example is shown in \fig{regimes}. This effect coincides, both in its qualitative shape and origin, to the one obtained in the multiscale generalization of the diffusion equation with different powers of the Laplacian \cite{Calcagni:2012rm}. In general, all coefficient functions with an approximate power-law behaviour in some regime give rise to dimensional flow at those scales. Details such as the value of $\jmin$ and the spacing in $j$ are not relevant for the value of the spectral dimension in these intermediate regimes, in agreement with the discussion in \cite{Calcagni:2012rm} on the role of regularization parameters in the profile of $\ds$. The details of regularization schemes are nonphysical and affect only transient regimes in $\ds(\tau)$, not the value of the plateaux.

\begin{figure}
\centering
\includegraphics[scale=0.6]{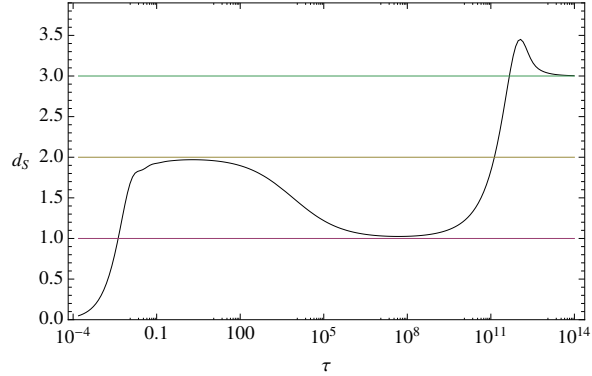}
\caption{Spectral dimension of superpositions with coefficients $|\qsc_j|^2 = j^{-4} + 200 j^{-7}$ summing from $\jmin =1/2$ to $\jmax = 200$ for $\sd=3$ and $\beta = 3$ (to be able to numerically cover enough scales with a feasible number of states in the sum). According to \eqref{alpha}, two different UV regimes with dimension $\ds\approx 2$ and then $\ds\approx 1$ can be observed.
}
\label{fig:regimes}
\end{figure}


\subsection{Walk dimension and Hausdorff dimension of superpositions \label{sec:walk-Hausdorff}}

The spectral dimension is only one of the possible dimensional observables. 
The superposition state model is well-suited to analyse other observables such as the walk dimension \eqref{qdw} and Hausdorff dimension \eqref{qdh} as well, and it is interesting to do so because there exist several relations among them, in classical and continuum spaces. 
Only a detailed analysis of their combined behaviour can give solid indications on the nature of the quantum geometries corresponding to quantum gravity states.

Quantum superpositions $\rsup$ are characterized by the Laplacian \eqref{assumption}.
Therefore, along the same lines as \eqref{htsuperposition} one can evaluate the expectation value of the mean square displacement \eqref{quantumX} in terms of its lattice solution \eqref{Xhlinf},
\begin{eqnarray}
\left\qbra \qX \right\qket_\ssup 
&=& \sumint^\jmax_{j=\jmin} \left|\qsc_j \right|^2 {\langle X^{2}\rangle}^{\hlinf\sd}_{0}[l^{-2}(j)\tau] \\
&=& d \sumint^\jmax_{j=\jmin} \left|\qsc_j \right|^2 l^{-2}(j)\tau\label{dwtem} \\
&\propto& \tau\,.
\end{eqnarray}
According to the definition \eqref{qdw} of the walk dimension $\dw$  as the scaling of the mean square displacement, this result for quantum superpositions yields the standard value
\[\label{dwfinal}
\dw^\ssup=2\,,
\]
independent of the form of the coefficients $\qsc_j$. Notice that the dependence on the topological dimension in \eqref{dwtem} is only through a proportionality coefficient. Therefore, \eqref{dwfinal} is valid both for space and spacetime.

\

The Hausdorff dimension of a quantum state is defined in \eqref{qdh} in terms of the scaling of the expectation value of the volume $\qVr{}$ \eqref{quantumV} of a ball with radius $r$.

With the volume spectra \eqref{spectra}, 
the evaluation of $\qVr{}$  on a single discrete quantum geometry state $\rjb$ in the large size limit $\size\ra\infty$ can be expressed in terms of the classical lattice solutions $V_{\hlinf\sd} (r)$ \eqref{latticeV} as 
\begin{eqnarray}
\qbra V(R) \qket_{j,{\hlinf\sd}}  &\propto& l^\sd(j) V_{\hlinf\sd} (R/l(j))  \\
&\propto& l^\sd(j) \prod_{n=0}^{\sd-1} [R /{l(j)}+n]\,.
\end{eqnarray}

Thus, its Hausdorff dimension corresponds to the lattice solution \eqref{dh-lattice}, \ie
\[\label{dh1}
\dh^{j,{\hlinf\sd}}(r)=r \sum_{k=0}^{\sd-1}\frac{1}{r + k\;l(j)} = \frac r{l(j)} \left[\digamma(r/l(j)+\sd)-\digamma(r/l(j))\right]\,,
\]
 where $\digamma$ is the digamma function. At large scales, $\dh$ approaches the topological dimension $\sd$, while at small scales it tends to 1:
\[\label{dh2}
\dh^{j,{\hlinf\sd}} \simeq\left\{\begin{matrix}\sd\,,\qquad R\gg1\\ 1\,,\qquad R\ll1\end{matrix}\right..
\]




For generic superposition states, this gives a nontrivial expectation value 
\[\label{vrsuperposition}
\langle V(R) \rangle_\ssup \propto \sumint_{j=\jmin}^\jmax |\qsc_j|^2 l^\sd(j)  \prod_{n=0}^{d-1}[R /{l(j)}+n]\,
\]
like for the spectral dimension \eqref{htsuperposition}.

Nevertheless, numerical calculations on the same classes of states as investigated above for the spectral dimension show qualitatively similar results to the Hausdorff dimension $\dh^{j,{\hlinf\sd}}$ on single states (\fig{Hausdorff}). 
That is, in all instances there are plateaux as in the pure lattice case \eqref{dh2}.
Only the scale and steepness of the flow between these plateaux is modified. 
For example, for power-law coefficients \eqref{gamma} the fall-off is much steeper and occurs, as $\alpha$ increases, closer to the scale as in the case of the single state $|\jmin,{\hlinf\sd}\rangle$.

\begin{figure}
\centering
\includegraphics[scale=.6]{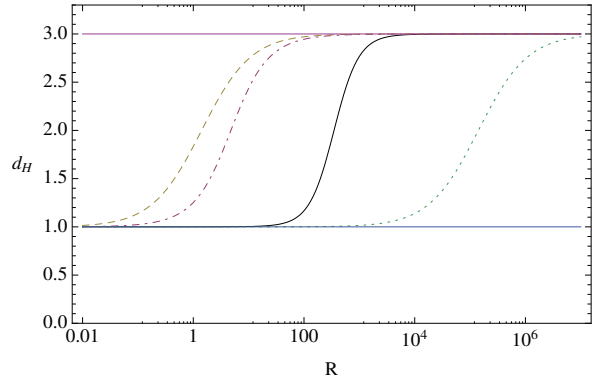}
\caption{Hausdorff dimension $\dh$ of a superposition with $\alpha = 1,2$ (solid and dash-dotted curve) in $\sd=3$ summing up to $\jmax =  10^5$, compared to $\dh$ on single states $|1,{\hlinf\sd}\rangle$ (dashed curve) and $|\jmax,{\hlinf\sd}\rangle$ (dotted curve).
}
\label{fig:Hausdorff}
\end{figure}




\subsection{A fractal structure in (loop) quantum gravity? \label{sec:fractal-structure}}


The calculations presented show that a flow in the spectral dimension occurs in quantum gravity, at least for a specific class of superpositions of regular (both from the combinatorial perspective and for what concerns the assignment of additional quantum labels) quantum states of geometry. 
These quantum states, although restricted by the regularity assumption, are exactly of the type appearing in the discrete quantum gravity formalisms of loop quantum gravity, spin-foam models and group field theory considered in this thesis.
Though they are more general, and can, for example, also simply be seen as quantum states of lattice quantum gravity, in the spirit of quantum Regge calculus. 

On the other hand, no dimensional flow due to quantum effects occurs in the Hausdorff and walk dimension. 
This conclusion is based on the interpretation, maintained throughout this thesis, that the flow of an effective dimension observable of a quantum geometry is a mere effect of discretization whenever it approximately coincides with the form of scale dependence for the related combinatorial complexes. 

The results are of particular interest from the point of view of loop quantum gravity.
Under the assumptions made for the action of the quantum Laplacian on the states (a very simple scaling behaviour), an important example of the states studied are kinematical LQG states where length ($\std=2+1$) or area and volume operators ($\std=3+1$) are diagonalized by spin-network states. 
In this sense, I have identified a class of LQG states with a dimensional flow. 
More precisely, for any $0<\alpha<\sd$ there is a class of states in the kinematical Hilbert space with a dimensional flow from the spatial topological dimension $\sd$ in the IR to a smaller value $\alpha$ in the UV. The UV value depends on the exact superposition considered but not on the topological dimension.

This result is in contrast with earlier arguments in LQG \cite{Modesto:2009bc}. There, the author argues for evidence of dimensional flow for individual spin-network states (thus, for a given graph or complex), and the same result is claimed in \cite{Caravelli:2009td,Magliaro:2009wa} for simple spin-network states with additional weights given by a 1-vertex spin foam (thus, not yet in a truly dynamical context). The starting point in \cite{Modesto:2009bc} is an assumption about the scaling of the expectation value of the Laplacian, very similar to \eqref{assumption}. The essential part of the argument is then the further assumption that the momenta $p$ of the scalar field defining the spectral dimension are directly related to a length scale set by the quantum numbers as $p\propto1/\sqrt j$. The scaling of the Laplacian in $j$ is finally translated in a modified dispersion relation in $p$ and the result depends on the precise form of the spectrum \eqref{area-spectrum} with $\csu=0$.

In the present case, no further assumption beyond \eqref{assumption} is made. Calculations are based on the momenta of the scalar field on the lattice-based geometry, that is, the spectrum of the Laplacian, but the spectral dimension is computed directly as a quantum geometric observable evaluated on quantum states. 
Furthermore, in \sec{coherent-states}, I have found no effects on the spectral dimension for individual quantum-geometry states of LQG based on given graphs or complexes. 
On the other hand, the genuine dimensional flow that occurs here for the states $\rsup$ is crucially related to the superposition of spin-network states also with respect to the underlying combinatorial structures, and it is not solely the result of the discreteness of geometric spectra. 
In this deeper sense, dimensional flow can indeed be seen as an effect of quantum discreteness of geometry.

\

The detailed results of all the three effective dimension observables allow also to address the question, whether such quantum geometries are effectively fractal geometries.
Quite often in the literature of quantum gravity, dimensional flow has been advertised as spacetime being fractal. 
However, strictly speaking not all sets with varying dimensionality are fractals. 
Although no unique operational and rigorous definition of `fractal' exists, one property all fractals generally possess is a special relation among the spectral dimension $\ds$, the Hausdorff dimension $\dh$ and the walk dimension $\dw$:
\begin{equation}
\dh=\frac{\dw}{2}\ds.\label{dimensions}
\end{equation}

On the hypercubic lattice superpositions considered, $\dw=2$ and the above relation simplifies to $\dh=\ds$. 
This is trivially obtained in the IR regime, where both dimensions take the value of the topological dimension. In the UV regime above the lattice scale (recall that below such scale any scaling effect is arguably unphysical), the Hausdorff dimension takes the classical value $\dh^{{\hlinf\sd}}\simeq 1$. Thus, \eqref{dimensions} is only obeyed in the case of a scaling $\alpha = 1$ such that also the spectral dimension takes this value. Only then can one call the quantum superposition $\rsup$ an effective one-dimensional fractal. 
This is indeed a perfectly allowed choice of quantum states and one can conclude that the results show a particular class of quantum geometries that corresponds, by all appearances, to a fractal quantum space.

However, one should mention a caveat here. 
For these geometries to be safely regarded as fractals, the origin of the dimensional flow should be the same in the left- and right-hand side of \eqref{dimensions}, which may not be the case for here: the left-hand side flows due to discreteness effects, while the right-hand side flows due to physical quantum effects. 
This situation might suggest either that one should not place particular significance in the fulfilment of \eqref{dimensions} or that the discrimination between discreteness effects and physical quantum effects is somewhat too strong and should be revised. 
I do not attempt to solve this mild conceptual issue here, which is harmless for the main results. 
Still, it will deserve further attention.

Interestingly, the geometry with $\alpha=1$ is also the only one where the spectral dimension of spacetime reaches the value $\Ds=\ds+1\simeq 2$ which is very often commented upon in the literature of quantum gravity. 
Its appearance across independent approaches such as causal dynamical triangulations, asymptotic safety, Ho\v{r}ava-Lifshitz gravity and others \cite{Hooft:1993vl,Carlip:2009cy,Calcagni:2009fh} triggered the suspicion that this ``magic number'' could be a universal characteristic of frameworks with good ultraviolet properties or, in other words, that a two-dimensional limit of the spectral dimension might be tightly related to the renormalizability or finiteness of quantum gravity. 
By now, it has become clear that this is not the case in general, as there exist counterexamples of nonlocal field theories with good renormalization properties with $\Ds\neq 2$ in the UV \cite{Modesto:2012er}, as well as of local theories whose renormalization properties are not at all improved by dimensional flow \cite{Calcagni:2014hx}. 
The superposition states investigated here provide another instance pointing towards the same conclusion: the value of $\ds$ is governed by a choice of states which, by itself, is not (sufficiently) connected with the dynamical UV properties of the underlying full theory.

\

In this \sec{dimensional-flow}, I have  investigated the effective structure of quantum superpositions of regular (hypercubic and homogeneous in label assignment) states of quantum geometry. 
It is possible to identify states with a flow of the spectral dimension to a dimension $\alpha$ in the UV, provided the superposition includes fine enough combinatorial structures and a large enough number of (kinematical) degrees of freedom of quantum geometry, and a particular set of expansion coefficients \eqref{gamma} related to $\alpha$ \eqref{alpha}.
For the Hausdorff and walk dimension, I do not observe any physical quantum effects, although discreteness effects do alter the value of the Hausdorff dimension across scales.
A fractal structure in the strict sense, \ie where \eqref{dimensions} relating the three dimensions is fulfilled also in the UV, is realized in the case $\alpha=1$.
In particular, these results provide evidence for a dimensional flow in a certain class of kinematical LQG states, also available in the spin-foam and group field theory context.


\addchap[tocentry={}]{Conclusions}
\addcontentsline{toc}{chapter}{\protect\numberline{}Conclusions}

%
%

The main goal of this thesis has been to analyse the fractal structure of discrete quantum geometries in LQG, SF models and GFT and check if there are any more solid hints for a dimensional flow as observed in other quantum gravity approaches \cite{\cdtfractal,Benedetti:2009bi, Coumbe:2015bq,\asfractal,Calcagni:2013jx,\ncfractal,\hlfractal,\snfractal}.
This has been accompanied by several side topics such as question of theory convergence in these approaches, the possibility of a purely combinatorial description for quantum geometries in a fundamentally discrete setting, the development of discrete calculus to provide the conceptual framework for the notion of effective-dimension observables, as well as the extension of group field theories to cover the whole state space of discrete quantum geometries as demanded in particular in LQG.
Addressing all these issues, I have accomplished the following:

First, I have sketched a logical reconstruction of the relevant five discrete quantum gravity proposals, QRC, CDT, LQG, SF models and GFT, according to the structuralist view on scientific theories.
On this basis I have argued that these theories should be understood as distinct theories in the first place. This is the precondition for any meaningful (nontrivial) theory relation between them.
I have then discussed several main relations between the theories, both conceptually and on the level of dynamics, in terms of which all of them turn out to be connected, at least on the level of particular models.
This has led to the conclusion that there is no theory convergence between them in a stricter sense. 
On the other hand, a precise notion of theory crystallization seems to apply well to the web of theories. In this sense, one can therefore understand at least a subset of the approaches as likely  developing into a single theory eventually.

This investigation has not been comprehensive and there are many ways in which it could be expanded.
I have only sketched how the theories are logically reconstructed and the framework of structuralism allows a much more detailed analysis which would strengthen the consequent statements substantially.
From a physicist's perspective, though, it is probably more interesting to stay on a more informal level of theory assessment but retrace in more detail the various intertheoretical relations between the approaches of which I have explained only a number of examples.
Finally, the probably most interesting challenge from a philosophical perspective is to strengthen the hypothesis that a crystallization process indeed improves the epistemological status of the involved theories. 
For this purpose, an interesting framework might be Bayesian epistemology \cite{Bovens:2004tw} where mathematically strict results for such questions (\eg \cite{Dawid:2013tj}) are possible.

\

%
Second, I have given an exhaustive description of the discrete structures playing a role in the discrete quantum gravity approaches under consideration.
I have shown how these can be defined in a purely combinatorial way, extending the notion of abstract simplicial complexes to polyhedral complexes and generalizations of these given by a certain subdivision structure.
A particular focus has been on the molecular 2-complex structure underlying amplitudes in SF models and GFT for which I have given a detailed description and a classification of the various relevant subclasses.
This classification complements, but also clarifies and completes the one in \cite{\KKL}, which formed the basis for the first combinatorial generalization of SF models.

\

%
Third, I have shown how GFT can be extended in a meaningful, manageable way to entail  boundary states based on graphs of arbitrary valence.
With the ground properly set, it is straightforward to define a generalization of the well known simplicial GFT using arbitrary atoms. 
I have presented explicitly how this can be obtained by a multi-field GFT.
Since it is extremely difficult to turn such a formally defined field theory with a potentially infinite number of fields into a controllable one allowing for concrete calculations and physical insights, there is need for an alternative.
To this end I have presented GFT models with a dual weigthing which generate a general class of molecules, but are based on standard simplicial GFT and are similarly manageable.
The dually weighted GFT models are based on the combinatorial possibility to obtain arbitrary boundary graphs from simplicial molecules, implemented on the  dynamical level. 
This implementation provides an example of a useful application of tensor-model techniques to LQG and SF models via a GFT formulation.
In both cases, the inclusion of the dynamics of gravitational SF models, generalized to arbitrary complexes, is possible. 

The extension from simplicial to more general complexes raises questions concerning the implementation of simplicity constraints.
I have shown how their implementation in the known models can be straightforwardly applied to the extended GFT models 
in the same spirit as in \cite{\KKL}.
However, the derivation of all these models rests on the classical geometry of simplices. 
A spin-foam atom with arbitrary combinatorics, in contrast, corresponds rather to a polytope. 
Consequently, taking the arbitrary closed graphs of LQG seriously, an adaption of the simplicity constraints and their quantization to polytopes is needed.

Another issue concerns the possibility of an extension of arbitrary molecules to higher dimensions.
In the simplicial case, the good behaviour of a SF or GFT model depends on the molecules possessing an extension to a full $\std$-dimensional orientable complex.  
As shown, this is necessary for the definition of the differential structure; but it is also important for the  control of divergences \cite{\gftrenorm, \COR},  as well as for diffeomorphism symmetry \cite{Baratin:2011bk,Bonzom:2012tg,Bonzom:2012gwa,Freidel:2003bh,Dittrich:2008pw}. 
A solution to the issue could be to pass over to  coloured GFT models.
This is directly possible in the case of dually weighted GFT
which will then generate effectively all polyhedral $\std$-complexes in terms of their simplicial subdivisions. 
Still, one may want an encoding of the topology of general $\std$-complexes directly at the level of arbitrary SF molecules. This remains an interesting open problem.

Apart from such conceptual issues, the most relevant task is the investigation of the field-theoretic properties of the dually weighted GFT models presented. 
Among them, the large-$N$ \cite{\largeN} and double-scaling \cite{\double} limits 
are of primary interest. 
Finally, the probably most important question concerns the renormalizability of these models.
As explained, the extension from simplicial interactions to a sum over tensor invariant interactions can be directly transferred to the dually weighted GFT.
Investigating the renormalizability of such models is thus possible using the same methods already applied in the GFT literature \cite{\gftrenorm, \COR}.  
 

\


Fourth, I have introduced a discrete calculus on discrete geometries consisting in an assignment of geometric data to purely combinatorial complexes, 
in this way extending formalisms relying on an ambient space \cite{Grady:2010wb, Desbrun:2005ug,Albeverio:1990ii, Adams:1996ul, Teixeira:2013ee}.
This should open up novel ways to investigate the physical and geometric properties of discrete quantum gravity theories.
To this end, I have given explicit expressions of the Laplacian in geometric variables used in any of such approaches: edge lengths, face normals, fluxes, and area-angle variables. 
In particular, the calculus provides a setting to rigorously define a discrete Laplacian operator.
I have shown that this operator satisfies all physically desirable properties of its continuum analogue.
I have used the Laplacian to give a meaningful definition of spectral and walk dimension as observables for discrete quantum geometries, in particular for states in a direct sum $\hs = \bigoplus_\cp \hs_\cp$ of Hilbert spaces $\hs_\cp$ for a class of distinct complexes $\cp$.
This provides the setting for the investigation of these observables in discrete quantum gravity approaches.

A natural field of application of the discrete calculus formalism is the addition of matter fields onto quantum geometry.
In any phenomenologically more relevant application of quantum gravity it is necessary to consider appropriate matter fields coupled to pure gravity.
It is already rather well understood how to canoncially quantize theories of matter coupled to gravity in the spirit of LQG \cite{Thiemann:1998hn,Domagaia:2010iq,Giesel:2012tj},
and there are ideas how matter field degrees of freedom could directly be added to spin-foam configurations \cite{Oriti:2002gy,Fairbairn:2007fk,Bianchi:2013cm}, interpreted as topological defects \cite{Freidel:2004bq,Freidel:2006gc,Baez:2007fm,Fairbairn:2008ko,Fairbairn:2008bx}, or obtained in effective actions of GFT \cite{Fairbairn:2007bu,Oriti:2009ks}.
The discrete calculus allows now to properly define matter coupling directly on the fundamentally discrete geometries as they occur in SF models and GFT setting the stage for a more rigorous analysis of their interaction with the geometric degrees of freedom.

\


Fifth, I have explored the  spectral dimension of states of quantum geometry in the case of (2+1)-dimensional LQG as well as in a general setup of superpositions of discrete quantum geometries in any spatial dimension.
I have shown how one can define ob\-servables of the Hausdorff, spectral and walk dimension for quantum states of geometry in terms of the scaling of the quantum expectation value of the volume, the heat trace and the mean square displacement, respectively.
A systematic classification of topological effects and discreteness artefacts in classical settings has set the ground for the quantum case and I have pushed the analysis to analytic expressions whenever possible.

On this basis I have calculated the spectral dimension of coherent states in LQG in $2+1$ dimensions, 
both as superpositions of geometries on the same complex and as superpositions of complexes, in particular triangulations of the same geometry. 
To this end I have set up an algorithm to compute the spectral properties of the relevant discrete-geometry configurations the states are defined by.
A main result of the calculation of the spectral dimension of coherent states on a single complex is the lack of any strong indication of a dimensional flow. 

On the other hand, I have identified a generic dimensional flow for a particular class of large superposition states in any spatial dimension $\sd$, characterized by superposition coefficients with a power-function dependence on the quantum numbers in the discrete quantum geometries of spin-network type that are summed over. 
For any real number $0<\alpha<\sd$ there is a state in this class which has a spectral dimension flowing from $\sd$ on large lengths scales to $\alpha$ on small scales. 
I have found further that the quantum walk dimension shows no difference from the classical case and that the classical  behaviour of the Hausdorff dimension with a flow to $\sd=1$ on small scales due to discreteness transfers qualitatively to the quantum case.
Taking all these results together, there is a single state in the class of superpositions characterized by the parameter $\alpha =1 $ which can be understood as providing an effectively fractal geometry.

These results can be further generalized in various directions, \eg for lattices of different combinatorics, as well as refined within individual theories of quantum gravity. 
Furthermore, though already based on analytic solutions, the precise reason for the dimensional flow in the superposition is not yet fully understood; a detailed understanding of the relevant aspect in the structure of the superpositions causing the flow could help to generalize and classify further the class of states with potential fractal behaviour.
In parallel, it becomes manageable to explore the phenomenological consequences of the discovered dimensional flow and (when applicable) fractal nature of quantum space as a direct effect of the full quantum theory. 
Such a possibility is especially interesting in the context of quantum cosmology, where a change of dimensionality can bear its imprint in the early stage of cosmic evolution \cite{AmelinoCamelia:2013hk,AmelinoCamelia:2013kj,Calcagni:2013gt}.


\newpage

\appendix



\chapter{Appendix}

\section{Classical expressions of the Laplacian}
\label{sec:classical-expressions}

In this appendix which is based on the appendix of \cite{\COTa}, I give explicit expressions for the Laplacian on discrete geometries in the various variables playing a role in quantum gravity approaches.

The general form of the discrete Laplacian introduced in \sec{calculus} depends both on the combinatorial structure of the underlying combinatorial complex and on its discrete geometry through the various volume factors. 
The Laplacian takes then different concrete expressions, depending on the variables used to encode the geometry of the combinatorial complex. 
These expressions would be needed for explicit calculations in different formulations of classical discrete gravity and, successively, in applications to quantum gravity models. 
In the following, I provide some examples for the discrete Laplacian constructed in the geometric variables used in various approaches to classical and quantum gravity.

\subsection{Regge edge length variables}\label{sec:length-variables}

The most common variables to describe the geometry of an $\m$-dimensional simplicial pseudo-manifold $\simc$ are the squared edge lengths $ l^2_{ij}  $ assigned to all $(ij)\in\simcp1$. 
In the standard version of Regge calculus \cite{Regge:1961ct,Loll:1998ue} which is the starting point for QRC as outlined in \sec{qrc}, these are taken as configuration space for the geometries of piecewise flat triangulations.

The expressions for primal volumes, \ie volumes of the $\p$-simplices $\sigma_\p$ in $\simc$, are well known in the Regge literature.
They can be obtained from the Cayley--Menger determinant 
\[
V_{\sigma_{p}}=\frac{1}{p!}\frac{\left(-1\right)^{\frac{p+1}{2}}}{2^{\frac{p}{2}}}\left|\begin{array}{cccc}
0 & 1 & \cdots & 1\\
1 & 0 & l_{ij}^{2} & \vdots\\
\vdots & l_{ij}^{2} & \ddots\\
1 & \cdots &  & 0
\end{array}\right|^{\frac{1}{2}}.
\]
The particular relevant examples in 4-dimensional spacetime are
\[
V_{\sigma_{2}}=\frac{1}{4}\sqrt{\underset{i}{\sum}\left(2l_{ij}^{2}l_{ik}^{2}-l_{jk}^{4}\right)}
\]
and, after some manipulations, 
\begin{eqnarray}
V_{\sigma_{3}} &=& \frac{1}{12}\sqrt{\underset{(ij)}{\sum}l_{ij}^{2}\left(l_{ik}^{2}l_{jl}^{2}+l_{il}^{2}l_{jk}^{2}-l_{ij}^{2}l_{kl}^{2}\right)-\underset{(ijk)}{\sum}l_{ij}^{2}l_{ik}^{2}l_{jk}^{2}}\,,\\
V_{\sigma_{4}} &=& \frac{1}{96}\left[\underset{(ij)(kl)}{\sum}l_{ij}^{4}l_{kl}^{4}+\underset{(ij)(k)}{\sum}\left(l_{li}^{2}l_{ik}^{2}l_{kj}^{2}l_{jm}^{2}+l_{mi}^{2}l_{ik}^{2}l_{kj}^{2}l_{jl}^{2}-l_{ij}^{4}l_{kl}^{2}l_{km}^{2}\right)\right.\nonumber\\
&& \qquad\left.-2\underset{(ijkl)}{\sum}l_{ij}^{2}l_{jk}^{2}l_{kl}^{2}l_{li}^{2}-4\underset{(ij)}{\sum}l_{ij}^{2}l_{kl}^{2}l_{lm}^{2}l_{mk}^{2}\right]^{\frac{1}{2}}\,,
\end{eqnarray}
where all sums run over all subsimplices of the given kind.

 Thus, the only geometric data needed for defining the dual scalar Laplacian $\Delta $ are the dual edge lengths $\dl_{ab}$ for pairs of dual vertices $v_a, v_b \in \simcsp0$.
I subdivide the dual edges into two parts ${\dl}^{a}$ and ${\dl}^{b}$,
associated respectively with the dual simplex $\sigma=\star v_a$ and $\star v_b$, such that $\dl_{ab}={\dl}^{a}+{\dl}^{b}$. These dual edge lengths depend on the chosen embedding of dual complex into the primal one.
 
In the barycentric case, when ${\dl_{\hat{\iota}}}^{a}$ is the length of the edge dual to the face $(012\dots \hat{\iota}\dots \m)$ contained inside the simplex $\sigma = \star v_a=(012\dots \m)$, the dual (half-)edge is given by 
\[
{\dl_{\hat{\iota}}}^{a}=\frac{1}{\m\left(\m+1\right)}\sqrt{\m\underset{j}{\sum}l_{ij}^{2}-\underset{(jk)}{\sum}l_{jk}^{2}}\,.
\]
This can be seen as follows. In coordinates $x$ the position of the barycentre
$x_{bc}$ of a $p$-simplex $\sigma_\p$ is 
\[
x_{bc}=\frac{1}{p+1}\underset{i=0}{\overset{p}{\sum}}x_{i}\,.
\]
The distance from the barycentre of $\sigma$ to the barycentre
of $(12\dots \m)$ in these coordinates is 
\begin{equation}
{\dl_{\hat 0}}^{a}=\left|\frac{1}{\m+1}\underset{i=0}{\overset{\m}{\sum}}x_{i}-\frac{1}{\m}\underset{i=1}{\overset{\m}{\sum}}x_{i}\right|=\left|\m\,x_{0}-\frac{1}{\m(\m+1)}\underset{i=1}{\overset{\m}{\sum}}x_{i}\right|\,,\label{eq:BaryLength}
\end{equation}
and choosing coordinates where $x_{0}$ is the origin and using $x_{i}\cdot x_{j}=g_{ij}=\frac{1}{2}\left(l_{0i}^{2}+l_{0j}^{2}-l_{ij}^{2}\right)$
\cite{Hamber:2009wl} this reduces to 
\begin{eqnarray}
{\dl_{\hat 0}}^{a} & = &\frac{1}{\m(\m+1)} \sqrt{\left(\underset{i=1}{\overset{\m}{\sum}}x_{i}\right)^{2}}=\frac{1}{\m\left(\m+1\right)}\sqrt{\underset{i}{\sum}x_{i}^{2}-2\underset{(ij)}{\sum}x_{i}\cdot x_{j}}\nonumber\\
 & = & \frac{1}{\m\left(\m+1\right)}\sqrt{\m\underset{i}{\sum}l_{0i}^{2}-\underset{(ij)}{\sum}l_{ij}^{2}}\,.
\end{eqnarray}
Besides the simple two-dimensional case this formula was also proven before for the tetrahedron (theorem 187 in \cite{AltshillerCourt:1964vd}).

Then the matrix elements of the Laplacian \eqref{scalar-laplacian} 
have the form
\[\label{edge-length-laplacian}
\frac{w_{ab}}{V_{\star v_a}} = \m\left(\m+1\right)\frac{1}{V_{012\dots \m}}\frac{V_{12\dots \m}}{{\dl_{\hat{0}}}^{a}+{\dl_{\hat{0}'}}^{b}}\,
\]
where furthermore $\star v_b = (0'12\dots \m)$. 
These are well defined on simplicial geometries satisfying the strong generalized triangular inequalities, that is, $V_{\sigma_{\p}}>0$ for all $\sigma_\p\in\simcp\p\In\simc$. In particular, these conditions ensure that the dual lengths ${\dl}^a$ are non-zero and positive.

This is not the case for the circumcentric dual where each ${\dl}^{a}\in\R$
can be negative or vanishing, and thus it is possible to have ${\dl_{\hat{0}}}^{a}+{\dl_{\hat{0}'}}^{b}=0$. This pole in the expression for the Laplacian, moreover, cannot be absorbed into the volumes as they depend only on the edges of $\sigma=\star v_a$ but not of $\star v_b$. 
On the other hand, except for these singularities, the circumcentric Laplacian might be well defined even on degenerate geometries with $V_{\sigma}=0$. 
This is true, for example, for $\m=2,3$ where explicit expressions of the circumradius are known.
In $\m=2$,
\[
\frac{w_{\left(ijk\right)\left(jkl\right)}}{A_{ijk}} = \frac{8}{\pm\left(l_{ij}^{2}+l_{ik}^{2}-l_{jk}^{2}\right)\pm\frac{A_{ijk}}{A_{jkl}}\left(l_{jl}^{2}+l_{kl}^{2}-l_{jk}^{2}\right)}\,,
\label{CircEdgeLap2d}
\]
and in $\m=3$ 
\begin{eqnarray}
\frac{w_{(ijkl)(ijkm)}}{V_{ijkl}} &=& 12A_{ijk}^{2}
\bigg(\pm\sqrt{\left(2A_{ijk}\mathcal{A}_{ijkl}\right)^{2}-\left(3l_{ij}l_{jk}l_{ki}V_{ijkl}\right)^{2}} \\
&& \pm\frac{V_{ijkl}}{V_{ijkm}}\sqrt{\left(2A_{ijk}\mathcal{A}_{ijkm}\right)^{2}-\left(3l_{ij}l_{jk}l_{ki}V_{ijkm}\right)^{2}}\bigg)^{-1}.\nonumber
\label{CircEdgeLap3d}
\end{eqnarray}
The sign of each dual length part 
is positive if the circumcenter lies inside the $\m$-simplex $\sigma$ and negative if outside. 

This is based on the following calculation of dual circumcentric edge lengths:
In $\m=2$ one gets dual edge lengths from
the circumradius $R_{ijk}=\frac{l_{ij}l_{ik}l_{jk}}{4A_{ijk}}$: 
\begin{eqnarray}
{\dl}_{jk}^{(ijk)} & = & \sqrt{R_{ijk}^{2}-\left(\frac{l_{jk}}{2}\right)^{2}}=\frac{l_{jk}}{2}\sqrt{\frac{l_{ij}^{2}l_{ik}^{2}}{4A_{ijk}^{2}}-1}\nonumber\\
 & = & \frac{l_{jk}}{4A_{ijk}}\sqrt{l_{ij}^{4}+l_{jk}^{2}+l_{ki}^{2}-l_{ij}^{2}l_{ik}^{2}-2(l_{jk}^{2}l_{ij}^{2}+l_{jk}^{2}l_{ik}^{2})}\,.
\end{eqnarray}
Since 
\begin{eqnarray}
4l_{jk}^{2}l_{ik}^{2}-16A_{ijk}^{2} & = & l_{ij}^{4}+l_{jk}^{4}+l_{ki}^{4}-2(l_{jk}^{2}l_{ij}^{2}+l_{jk}^{2}l_{ik}^{2}-l_{ij}^{2}l_{ik}^{2})\nonumber\\
 & = &(l_{ij}^{2}+l_{ik}^{2}-l_{jk}^{2})^{2}\,,
\end{eqnarray}
this simplifies to 
\[
{\dl}_{jk}^{(ijk)}=\frac{l_{ij}^{2}+l_{ik}^{2}-l_{jk}^{2}}{4A_{ijk}}\frac{l_{jk}}{2}\,.
\]
The matrix elements of the Laplacian are 
\begin{eqnarray}
\frac{w_{\left(ijk\right)\left(jkl\right)}}{A_{ijk}} & = &\frac{1}{A_{ijk}}\frac{2}{\pm\sqrt{\frac{l_{ij}^{2}l_{ik}^{2}}{4A_{ijk}^{2}}-1}\pm\sqrt{\frac{l_{jl}^{2}l_{kl}^{2}}{4A_{jkl}^{2}}-1}}\\
& = & \frac{4}{\pm\sqrt{l_{ij}^{2}l_{ik}^{2}-4A_{ijk}^{2}}\pm\frac{A_{ijk}}{A_{jkl}}\sqrt{l_{jl}^{2}l_{kl}^{2}-4A_{jkl}^{2}}}\nonumber\\
 & = & \frac{8}{\pm\left(l_{ij}^{2}+l_{ik}^{2}-l_{jk}^{2}\right)\pm\frac{A_{ijk}}{A_{jkl}}\left(l_{jl}^{2}+l_{kl}^{2}-l_{jk}^{2}\right)}\,.
\end{eqnarray}
For $\m=3$, there is a formula relating the circumradius $R$ of the
tetrahedron $(ijkl)$ to the area $\mathcal{A}_{ijkl}$ of a triangle
with the product of opposite edge lengths in the tetrahedron as its
edge lengths \cite{AltshillerCourt:1964vd}: 
\[
6V_{ijkl}R_{ijkl}=\mathcal{A}_{ijkl}\,.
\]
The circumcentric dual length to a face $(ijk)$ thus is 
\[
{\dl}_{{\hat\iota}}^{(ijkl)}=\sqrt{R_{ijkl}^{2}-R_{ijk}^{2}}=\frac{\sqrt{\left(2A_{ijk}\mathcal{A}_{ijkl}\right)^{2}-\left(3l_{ij}l_{jk}l_{ki}V_{ijkl}\right)^{2}}}{12A_{ijk}V_{ijkl}}\,,
\]
and the Laplace weight 
\begin{eqnarray}
w_{(ijkl)(ijkm)} & = & 12A_{ijk}^{2}\Bigg( \pm\sqrt{\left(\frac{2A_{ijk}\mathcal{A}_{ijkl}}{V_{ijkl}}\right)^{2}-\left(3l_{ij}l_{jk}l_{ki}\right)^{2}}\\
&& \pm\sqrt{\left(\frac{2A_{ijk}\mathcal{A}_{ijkm}}{V_{ijkm}}\right)^{2}-\left(3l_{ij}l_{jk}l_{ki}\right)^{2}}\Bigg)^{-1}\,.\nonumber
\end{eqnarray}
A simplification to avoid the square roots, as in $\m=2$, remains to be found.

With these descriptions of the Laplacian at hand, one can compare with other discrete Laplacians in the literature.

\subsubsection*{Sorkin's discrete Laplacian}

In \cite{Sorkin:1975kv} a formalism with special ``barycentric'' coordinates (not to be confused with the mathematical notion, where unit vectors are attached to corners) is developed. 
As done also in \cite{Dittrich:2012ba}, it can be expressed in terms of the dihedral angles as a ``cotangens Laplacian''
(with inverse volume factor) for primal scalar fields. 
In $\m=2$, with
$\alpha_{ij}^{\sigma_{2}}$ the angle opposite to the edge
$(ij)$ in the triangle $\sigma_{2}$, it is given by
\[
-(\Delta_{0}\phi)_{i}=\frac{1}{V_{\star(i)}}\underset{j}{\sum}\left(\underset{\sigma_{2}>(ij)}{\sum}\cot\alpha_{ij}^{\sigma_{2}}\right)(\phi_{i}-\phi_{j})\,,
\]
and it is easy to show its equivalence to the Laplacian coming from discrete calculus with  circumcentric duals since elementary geometric arguments yield 
\[
{\dl_{\hat{\iota}}}^{\left(ijk\right)}=\sqrt{R^{2}-\left({l_{jk}}/{2}\right)^{2}}
= \frac{l_{jk}}{2} \cot\alpha_{ij}^{\sigma_{2}}\;.
\]
In $\m=3$,
\[
-(\Delta_{0}\phi)_{i}=\frac{1}{V_{\star(i)}}\underset{j}{\sum}\left(\underset{\sigma_{3}>(ij)}{\sum}l_{\hat{\iota}\hat{\j}}^{\sigma_{3}}\cot\alpha_{ij}^{\sigma_{3}}\right)(\phi_{i}-\phi_{j})\,,
\]
where the opposite dihedral angle $\alpha_{ij}^{\sigma_{2}}$
now is between faces sharing the opposite edge $l_{\hat{\iota}\hat{\j}}^{\sigma_{3}}$
in the tetrahedron $\sigma_{3}$. 
From the equivalence in $\m=2$ it is tempting to conjecture equivalence also for $\m\geq 3$, but this remains to be proven.

\subsubsection*{Laplacian in dynamical triangulations}

As discussed with the  introduction of  CDT in \sec{cdt}, a complementary way of encoding the simplicial geometry of a piecewise flat triangulation, still based on the Regge calculus description,
is to fix all edge lengths to some constant value, and allow only changes in the combinatorics of the simplicial complex itself.
For such equilateral configurations, the Laplacian coming from discrete calculus drastically simplifies
(up to an overall factor) to a purely combinatorial graph Laplacian \cite{Chung:1997tk} of the form \eqref{DAeq},
\[
\Delta\propto D-A\,,
\]
where the weights here are $w_{ab}=1$ if $v_a$ and $v_b$ are incident.

While in the Lorentzian version, named causal dynamical triangulations, this is modified by introducing negative length squares for timelike edges parametrized by the factor $\alpha$ \eqref{cdt-wick}, this modification is not performed since the theory is Wick rotated to Euclidean signature and actual calculations are performed in a reduced ensemble of Euclidean triangulations (those that can indeed be obtained by Wick rotating Lorentzian ones) \cite{Ambjorn:2012vc}
Still, in that case, the Laplacian of causal dynamical triangulations has to be calculated explicitly from the full simplicial case, in terms of the functional dependence of volumes on $\alpha$.

\subsection{First-order Regge calculus with face variables}

For a simplicial complex $\simc$ with a geometric realization as a piecewise
linear space of dimension $\std$, the frame field can be considered as a set of discrete
edge vectors 
\begin{equation}
e^{I}=e_{\mu}^{I}\d x^{\mu}\quad\mapsto\quad e_{ij}^{I}(\alpha)=\left[x_{i}(\alpha)-x_{j}(\alpha)\right]^{I}\,,\label{eq:Vielbein}
\end{equation}
where the coordinates $x(\alpha)$ are given by a choice of origin
and frame for every $\std$-simplex
$\sigma^{\alpha}$ the edge $(ij)$ is face of. The index $\alpha=1,2,\dots ,\np{\std}$ labels the $\std$-simplices. Accordingly, the volume form of a $p$-simplex $\sigma_{p}$ in the
coordinates of a $\std$-simplex it is a face of is 
\[
\omega_{I_{p+1}\dots I_{\std}}^{\sigma_{p}}(\alpha)=\epsilon_{I_{1}\dots I_{\std}}\underset{k=1}{\overset{p}{\prod}}e_{k}^{I_{k}}(\alpha)\,,
\]
in terms of $p$ linear independent edge vectors $e_{k}$ belonging
to $\sigma_{p}$. The $p$-volume $\sigma_{p}$ is the norm of the
volume form,
\[
V_{\sigma_{p}}=\left|\omega^{\sigma_{p}}\right|=\frac{1}{p!}\sqrt{\underset{I_{p+1}<\dots <I_{\std}}{\sum}\left|\omega_{I_{p+1}\dots I_{\std}}^{\sigma_{p}}\right|^{2}}\,.
\]

An alternative version to edge-length Regge calculus is in terms of the $(\std-1)$-face normals $\omega^{\sigma_{\std-1}}(\alpha)$ (expressed in the reference frame of the $\std$-simplex $\sigma^{\alpha}$) and Lorentz rotations (parallel transports) $U(\alpha,\alpha')$ from frame to frame across neighbouring simplices. 
In turn, the latter define holonomies (around closed plaquettes)
\[
W_{\alpha}(h)=U_{\alpha,\alpha+1}U_{\alpha+1,\alpha+2}\dots U_{\alpha-1,\alpha}
\]
which are rotations in the plane orthogonal to hinges $h\in \simcp{\std-2}$ \cite{Caselle:1989cd,Barrett:1994ba,Gionti:2005gi} and measure the local curvature. 
The class angles corresponding to the holonomies are therefore the deficit angles $\theta_{h}=2\pi-\sum_\alpha\theta_{h}^{\alpha}$, as could be obtained from the dihedral angles $\theta_{h}^{\alpha}$ at the hinge $h$ in each $\std$-simplex $\sigma^{\alpha}$ sharing it.

I show how all geometric data needed for the Laplacian $\Delta $ have an expression in terms of the face normals $\omega^{\sigma_{\std-1}}(\alpha)$. 

While the $(\std-1)$-volumes are just the modulus of the face normals themselves, 
\[
V_{\sigma_{\std-1}}=\left|\omega^{\sigma_{\std-1}}(\alpha)\right|\,,
\]
the $\std$-volumes of simplices $\sigma^{\alpha}$ can also be expressed
by $\std$ of the face normals $\omega^{i}(\alpha)=\omega^{\sigma_{\std-1}=(012\dots \hat{\iota}\dots \std)}$
as \cite{Caselle:1989cd} 
\[
V_{\alpha}\equiv V_{\sigma^{\alpha}}=\left[\frac{1}{\std!}\epsilon^{I_{1}\dots I_{\std}}\epsilon_{i_{1}\dots i_{\std}j}\omega_{I_{1}}^{i_{1}}(\alpha)\dots \omega_{I_{\std}}^{i_{\std}}(\alpha)\right]^{\frac{1}{\std-1}},
\]
where capital indices $I, J,...$ are in internal space. 
By the closure relations, it does not matter which face $(012\dots \hat{\j}\dots \std)$
is left out if $\sigma^{\alpha}$ is closed. Alternatively, one could also
average over the choices of reference face.

An explicit expression of dual lengths can only be obtained using
position coordinates on $\sigma^{\alpha}$ as functions of the face
normals. Barycentric coordinates $z(\alpha)$, that is coordinates
for which the sum over vertices satisfies $\sum_{i=1}^{\std+1}z_{i}^{I}(\alpha)=0$,
can be derived inverting the expression of the face normals in terms
of discrete vielbeins (see eq.\ (\ref{eq:Vielbein})) in these coordinates
\cite{Caselle:1989cd}, 
\[
\omega_{I}^{i}(\alpha)=\frac{1}{\left(\std-1\right)!^{2}}\underset{k\ne i}{\sum}\epsilon_{J_{1}\dots J_{\std-1}I}\epsilon^{i,i_{1}\dots i_{\std-1},k}z_{i_{1}}^{J_{1}}(\alpha)\dots z_{i_{\std-1}}^{J_{\std-1}}(\alpha)\,,
\]
leading to 
\[
z_{i}^{I}(\alpha)=\frac{1}{\left(\std-1\right)!}\frac{1}{\left(V_{\alpha}\right)^{\std-2}}\underset{k\ne i}{\sum}\epsilon^{J_{1}\dots J_{\std-1}I}\epsilon_{i,i_{1}\dots i_{\std-1},k} \, \omega_{J_{1}}^{i_{1}}(\alpha)\dots \omega_{J_{\std-1}}^{i_{\std-1}}(\alpha)\,.
\]
The barycentric dual length is particularly simple in these coordinates.
It is just the distance from the barycenter of the tetrahedron with
coordinate $z^{I}=0$  to the barycenter of a face
\[
{\dl_{\hat{\iota}}}^{\sigma}=\left|\underset{j\ne i}{\overset{}{\sum}}z_{j}[\omega^{i}(\sigma)]\right|\,.
\]
For the circumcentric case no such simplification can be expected.
Still, primal edge lengths can be expressed in the coordinates $z(\alpha)$, taking then advantage of the above expressions (eqs.\ \eqref{edge-length-laplacian}, \eqref{CircEdgeLap2d} and \eqref{CircEdgeLap3d}).

As an example, I can give the (further simplified) expressions in $\std=3$.
On $\sigma^{\alpha}=(ijkl)$ (suppressing the frame label $\alpha$),
\[
z_{i}^{I}=\frac{1}{2}\frac{1}{V_{\alpha}}\underset{r\ne i}{\sum}\epsilon^{IJK}\epsilon_{imnr}\omega_{J}^{m}\omega_{K}^{n}=\frac{1}{2V_{\alpha}}\left(\omega^{j}\times\omega^{k}+\omega^{k}\times\omega^{l}+\omega^{j} \times\omega^{l}\right)^{I},
\]
and the tetrahedron volume in terms of three of its face triangles is
\[
\left(V_{\alpha}\right)^{2}=\frac{1}{6}\epsilon^{IJK}\epsilon_{ijkl}\omega_{I}^{i}\omega_{J}^{j}\omega_{K}^{k}\,.
\]
Therefore, the dual length is 
\begin{eqnarray}
{\dl_{i}}^{\alpha} & = &\frac{1}{3}\left|z_{j}+z_{k}+z_{l}\right|
=\frac{1}{6V_{\alpha}}\left|\omega^{j}\times\omega^{k}+\omega^{k}\times\omega^{l}+\omega^{l}\times\omega^{j}\right|\\
 & = &\frac{\sqrt{\underset{(mn)\in(jkl)}{\sum}\left[\omega_{m}^{2}\omega_{n}^{2}-(\omega_{m}\cdot\omega_{n})^{2}+(\omega_{m}\cd\omega_{r})(\omega_{r}\cdot\omega_{n})-(\omega_{m}\cdot\omega_{n})\omega_{r}^{2}\right]}}{6V_{\alpha}}\,.\nonumber
\end{eqnarray}
Using the closure condition $\sum\omega^{i}=0$, this further simplifies to 
\[
{\dl_{ij}}^{\alpha}=\frac{1}{2V_{\alpha}}\left|\omega^{j}\times\omega^{k}\right|=\sqrt{\omega_{j}^{2}\omega_{k}^{2}-\left(\omega_{j}\cdot\omega_{k}\right)^{2}}\label{baryfacevar}
\]
for some faces $j,k$. The matrix elements of the Laplacian \eqref{scalar-laplacian}
can then easily be computed combining all the above expressions.

Finally, I note that the volume form $\omega^{h}(\alpha)$ of a hinge
$h=\sigma_{\std-2}$ can be expressed in terms of two normals to two faces $\sigma^{\alpha+1,\alpha}$,
$\sigma^{\alpha,\alpha+1}$ sharing it, in the frame of $\sigma^{\alpha}$ \cite{Caselle:1989cd}: 
\[
\omega_{IJ}^{h}(\alpha)=\frac{1}{V_{\alpha}}\omega_{[I}^{\alpha-1,\alpha}(\alpha)\omega_{J]}^{\alpha,\alpha+1}(\alpha)\,.
\]
This gives a connection to flux variables, discussed in the next
section, which are exactly these $(\std-2)$-face normals. 


\subsection{Flux and area-angle variables}

In $\std=4$, a useful alternative set of variables in simplicial geometry
are the bivectors $b_{ijk}^{IJ}=e_{ij}^{I}\wedge e_{ik}^{J}$ associated
with triangles $(ijk)$ (or their internal Hodge duals $X_{ijk}^{IJ}=\epsilon_{\ KL}^{IJ}b_{ijk}^{KL}$), 
known as fluxes, and playing a prominent role in both canonical loop quantum gravity and spin-foam models \cite{Rovelli:2004wb, Baratin:2010in, Baratin:2011hc}. In a geometric
4-simplex $(ijklm)$, the triangle areas are 
\[
A_{ijk}=|X_{ijk}|\,,
\]
and volumes of tetrahedra can be computed using three of the fluxes associated with the four triangles on their boundary \cite{Barrett:1998fp},
regarding the bivectors as linear maps: 
\[
V_{ijkl}^{2}=\frac{8}{9}\Tr\left(\ast X_{ijk}\left[\ast X_{jkl},\ast X_{kli}\right]\right)\,.
\]
Volumes of 4-simplices can be taken from the wedge product of two
fluxes not lying in the same $3$-hyperplane (thus not belonging to the same tetrahedron), 
\[
V_{ijklm}=|X_{ijk}\wedge X_{ilm}|\,.
\]
Primal edge lengths can be expressed using the generalized sine formula as 
\[
l_{ij}^2=2\frac{|X_{ijk}|^2|X_{ijl}|^2-(X_{ijk}\cdot X_{ijl})^2}{\Tr\left(\ast X_{ijk}\left[\ast X_{jkl},\ast X_{kli}\right]\right)}\,.
\]
This gives all the buildings blocks for explicit expressions  (eqs.\ \eqref{edge-length-laplacian}, \eqref{CircEdgeLap2d} and \eqref{CircEdgeLap3d}) of the
barycentric and circumcentric discrete Laplacian $\Delta $ with elements \eqref{weights}.

In the spin representation in $\std=3+1$ LQG and $\std=4$ spin-foams (adapted to a simplicial context), the easiest variables to use are triangle areas and $3$-volumes of tetrahedra. However, it is known that they form an overcomplete set of data to specify a four-dimensional simplicial geometry, and should be supplemented by additional constraints whose explicit form is not known \cite{Barrett:1999fa,Makela:1994hm}. A more natural choice is to use areas $A_{ijk}$ and dihedral
angles $\phi_{k,l}^{ij}$ between faces $(ijk)$ and $(ijl)$ hinged
at the common edge $(ij)$ \cite{Barrett:1994ba,Dittrich:2008hg}. This set of data encodes the same
information as the fluxes $X_{ijk}$. In these variables, the relevant geometric data to compute the discrete Laplacian have the following expressions.  The $3$-volumes are

\[
V_{ijkl}^{2}=\frac{A_{ijk}}{9}\sqrt{\underset{j}{\sum}A_{ijl}^{2}\sin^{2}\phi_{k,l}^{ij}A_{jkl}^{2}\sin^{2}\phi_{i,l}^{jk}-\underset{(ij)}{\sum}A_{ijl}^{4}\sin^{4}\phi_{k,l}^{ij}}\,.
\]
Using the generalized sine
law according to which the angles $\theta_{l,m}^{ijk}$ between 3-simplices $(ijkl)$ and
$(ijkm)$ are functions of the area dihedral angles of the form \cite{Dittrich:2008hg}
\begin{equation}\nonumber
\cos\theta_{l,m}^{ijk}=\frac{\cos\phi_{k,l}^{ij}-\sin\phi_{l,m}^{ij}\sin\phi_{m,k}^{ij}}{\cos\phi_{l,m}^{ij}\cos\phi_{m,k}^{ij}}\,.
\end{equation}
one obtains the $4$-volumes 
\[
V_{ijklm}=\frac{3}{4}\frac{1}{A_{ijk}}V_{ijkl}V_{ijkm}\sin\theta_{l,m}^{ijk}[\phi]\,,
\]
as well as the primal edge lengths
\[
l_{ij}=\frac{2}{3}\frac{1}{V_{ijkl}}A_{ijk}A_{ijl}\sin\phi_{k,l}^{ij}\,.
\]
Again, this is all the information needed to build the Laplacian $\Delta $.


\backmatter

\bibliographystyle{JHEP}
\bibliography{dissarxiv}

\providecommand{\href}[2]{#2}\begingroup\raggedright\begin{thebibliography}{100}

\bibitem{Oriti:2007uc}
D.~Oriti, ed., {\em {Approaches to Quantum Gravity: Toward a New Understanding
  of Space, Time and Matter}}.
\newblock Cambridge University Press, Cambridge, UK, 2007.

\bibitem{Nicolai:2014hy}
H.~Nicolai, {\it {Quantum Gravity: The View From Particle Physics}},  in {\em
  General Relativity, Cosmology and Astrophysics}, pp.~369--387.
\newblock Springer International Publishing, 2014.
\newblock \href{http://xxx.lanl.gov/abs/1301.5481}{{\tt arXiv:1301.5481}}.

\bibitem{Feynman:1942va}
R.~P. Feynman, {\em {The Principle of Least Action in Quantum Mechanics}}.
\newblock PhD thesis, Princeton University, 1942.

\bibitem{Feynman:1948iw}
R.~Feynman, {\it {Space-Time Approach to Non-Relativistic Quantum Mechanics}},
  {\em Rev. Mod. Phys.} {\bf 20} (1948) 367--387.

\bibitem{Dewitt:1967cs}
B.~Dewitt, {\it {Quantum theory of gravity. II. The manifestly covariant
  theory}},  {\em Phys. Rev.} {\bf 162} (1967) 1195--1239.

\bibitem{Goroff:1985kg}
M.~H. Goroff and A.~Sagnotti, {\it {Quantum gravity at two loops}},  {\em Phys.
  Lett.} {\bf B160} (1985) 81.

\bibitem{Kreimer:2008jm}
D.~Kreimer, {\it {A remark on quantum gravity}},  {\em Annals Phys.} {\bf 323}
  (2008) 49--60, [\href{http://xxx.lanl.gov/abs/0705.3897}{{\tt
  arXiv:0705.3897}}].

\bibitem{Bern:2007cx}
Z.~Bern, L.~J. Dixon, and R.~Roiban, {\it {Is N=8 supergravity ultraviolet
  finite?}},  {\em Phys. Lett. B} {\bf 644} (2007) 265--271,
  [\href{http://xxx.lanl.gov/abs/hep-th/0611086}{{\tt hep-th/0611086}}].

\bibitem{Bern:2011kk}
Z.~Bern, J.~J. Carrasco, L.~J. Dixon, H.~Johansson, and R.~Roiban, {\it
  {Amplitudes and ultraviolet behavior of N = 8 supergravity}},  {\em Fortschr.
  Phys.} {\bf 59} (2011) 561--578,
  [\href{http://xxx.lanl.gov/abs/1103.1848}{{\tt arXiv:1103.1848}}].

\bibitem{Green:1987tq}
M.~B. Green, J.~H. Schwarz, and E.~Witten, {\em {Superstring Theory: Volume 1
  \& 2}}.
\newblock Cambridge University Press, Cambridge, UK, 1987.

\bibitem{Polchinski:1998tm}
J.~Polchinski, {\em {String Theory: Volume 1 \& 2}}.
\newblock Cambridge University Press, Cambridge, UK, 1998.

\bibitem{Blumenhagen:2012wr}
R.~Blumenhagen, D.~L{\"u}st, and S.~Theisen, {\em {Basic Concepts of String
  Theory}}.
\newblock Theoretical and Mathematical Physics. Springer Berlin Heidelberg,
  2012.

\bibitem{Niedermaier:2006up}
M.~Niedermaier and M.~Reuter, {\it {The Asymptotic Safety Scenario in Quantum
  Gravity}},  {\em Living Rev. Relativity} {\bf 9} (2006).

\bibitem{Reuter:2012jx}
M.~Reuter and F.~Saueressig, {\it {Quantum Einstein gravity}},  {\em New J.
  Phys.} {\bf 14} (2012) 055022, [\href{http://xxx.lanl.gov/abs/1202.2274}{{\tt
  arXiv:1202.2274}}].

\bibitem{Thiemann:2007wt}
T.~Thiemann, {\em {Modern canonical quantum general relativity}}.
\newblock Cambridge University Press, Cambridge, UK, 2007.

\bibitem{Rovelli:2004wb}
C.~Rovelli, {\em {Quantum Gravity}}.
\newblock Cambridge University Press, Cambridge, UK, 2004.

\bibitem{Rovelli:2011tk}
C.~Rovelli, {\it {Zakopane lectures on loop gravity}},  in {\em PoS QGQGS2011},
  p.~003, 2011.
\newblock \href{http://xxx.lanl.gov/abs/1102.3660}{{\tt arXiv:1102.3660}}.

\bibitem{Baez:2000kp}
J.~C. Baez, {\it {An Introduction to Spin Foam Models of BF Theory and Quantum
  Gravity}},  in {\em Geometry and Quantum Physics}, pp.~25--93.
\newblock Springer, Berlin, Heidelberg, 2000.
\newblock \href{http://xxx.lanl.gov/abs/gr-qc/9905087}{{\tt gr-qc/9905087}}.

\bibitem{Perez:2003wk}
A.~P{\'e}rez, {\it {Spin foam models for quantum gravity}},  {\em Class. Quant.
  Grav.} {\bf 20} (2003) R43--R104,
  [\href{http://xxx.lanl.gov/abs/gr-qc/0301113}{{\tt gr-qc/0301113}}].

\bibitem{Perez:2013uz}
A.~P{\'e}rez, {\it {The Spin-Foam Approach to Quantum Gravity}},  {\em Living
  Rev. Relativity} {\bf 16} (2013) 3,
  [\href{http://xxx.lanl.gov/abs/1205.2019}{{\tt arXiv:1205.2019}}].

\bibitem{Freidel:2005jy}
L.~Freidel, {\it {Group Field Theory: An Overview}},  {\em Int. J. Theor.
  Phys.} {\bf 44} (2005) 1769,
  [\href{http://xxx.lanl.gov/abs/hep-th/0505016}{{\tt hep-th/0505016}}].

\bibitem{Oriti:2012wt}
D.~Oriti, {\it {The microscopic dynamics of quantum space as a group field
  theory}},  in {\em Foundations of Space and Time}.
\newblock Cambridge University Press, Cambridge, UK, 2012.
\newblock \href{http://xxx.lanl.gov/abs/1110.5606}{{\tt arXiv:1110.5606}}.

\bibitem{Krajewski:2012wm}
T.~Krajewski, {\it {Group Field Theories}},  in {\em PoS QGQGS2011}, p.~005,
  2012.
\newblock \href{http://xxx.lanl.gov/abs/1210.6257}{{\tt arXiv:1210.6257}}.

\bibitem{Williams:1992kw}
R.~M. Williams, {\it {Discrete quantum gravity: The Regge calculus approach}},
  {\em Int. J. Mod. Phys. B} {\bf 6} (1992) 2097--2108.

\bibitem{Williams:1997bn}
R.~M. Williams, {\it {Recent progress in Regge calculus}},  {\em Nucl. Phys. B}
  {\bf 57} (1997) 73--81, [\href{http://xxx.lanl.gov/abs/gr-qc/9702006}{{\tt
  gr-qc/9702006}}].

\bibitem{Williams:2006iu}
R.~M. Williams, {\it {Discrete quantum gravity}},  {\em J. Phys.: Conf. Ser.}
  {\bf 33} (2006) 38--48.

\bibitem{Williams:2007up}
R.~M. Williams, {\it {Quantum Regge calculus}},  in {\em Approaches to Quantum
  Gravity: Toward a New Understanding of Space, Time and Matter} (D.~Oriti,
  ed.).
\newblock Cambridge University Press, Cambridge, UK, 2007.

\bibitem{Hamber:2009wl}
H.~W. Hamber, {\em {Quantum gravitation}}.
\newblock The Feynman Path Integral Approach. Springer, Berlin Heidelberg,
  2009.

\bibitem{Ambjorn:2000hp}
J.~Ambjorn, J.~Jurkiewicz, and R.~Loll, {\it {Lorentzian and Euclidean Quantum
  Gravity --- Analytical and Numerical Results}},  in {\em M-Theory and Quantum
  Geometry}, pp.~381--450.
\newblock Springer Netherlands, Dordrecht, 2000.
\newblock \href{http://xxx.lanl.gov/abs/hep-th/0001124}{{\tt hep-th/0001124}}.

\bibitem{Ambjorn:2007wl}
J.~Ambjorn, J.~Jurkiewicz, and R.~Loll, {\it {Quantum Gravity: the art of
  building spacetime}},  in {\em Approaches to Quantum Gravity: Toward a New
  Understanding of Space, Time and Matter} (D.~Oriti, ed.).
\newblock Cambridge University Press, Cambridge, UK, 2007.

\bibitem{Ambjorn:2012vc}
J.~Ambjorn, A.~G{\"o}rlich, J.~Jurkiewicz, and R.~Loll, {\it {Nonperturbative
  quantum gravity}},  {\em Phys. Rept.} {\bf 519} (2012) 127--210,
  [\href{http://xxx.lanl.gov/abs/1203.3591}{{\tt arXiv:1203.3591}}].

\bibitem{Moulines:2010hi}
C.~U. Moulines, {\it {The crystallization of Clausius's phenomenological
  thermodynamics}},  in {\em Time, Chance, and Reduction} (G.~Ernst and
  A.~H{\"u}ttemann, eds.), pp.~139--158.
\newblock Cambridge University Press, Cambridge, UK, 2010.

\bibitem{Moulines:2013jn}
C.~U. Moulines, {\it {Intertheoretical Relations and the Dynamics of Science}},
   {\em Erkenntnis} {\bf 79} (2013) 1505--1519.

\bibitem{Moulines:2013fy}
C.~U. Moulines, {\it {Crystallizations as a Form of Scientific Semantic Change:
  The Case of Thermodynamics}},  in {\em Evolution of Semantic Systems},
  pp.~209--230.
\newblock Springer, Berlin, Heidelberg, 2013.

\bibitem{Kaminski:2010ba}
W.~Kaminski, M.~Kisielowski, and J.~Lewandowski, {\it {Spin-foams for all loop
  quantum gravity}},  {\em Class. Quant. Grav.} {\bf 27} (2010) 095006,
  [\href{http://xxx.lanl.gov/abs/0909.0939}{{\tt arXiv:0909.0939}}].

\bibitem{Kisielowski:2012bo}
M.~Kisielowski, J.~Lewandowski, and J.~Puchta, {\it {Feynman diagrammatic
  approach to spinfoams}},  {\em Class. Quant. Grav.} {\bf 29} (2012) 015009,
  [\href{http://xxx.lanl.gov/abs/1107.5185}{{\tt arXiv:1107.5185}}].

\bibitem{Oriti:2014wf}
D.~Oriti, {\it {Group Field Theory and Loop Quantum Gravity}},
  \href{http://xxx.lanl.gov/abs/1408.7112}{{\tt arXiv:1408.7112}}.

\bibitem{Hooft:1993vl}
G.~t. Hooft, {\it {Dimensional Reduction in Quantum Gravity}},  in {\em
  Salamfestschrift} (A.~Ali, J.~Ellis, and S.~Randjbar-Daemi, eds.).
\newblock Singapore: World Scientific, 1993.

\bibitem{Carlip:2009cy}
S.~Carlip, {\it {Spontaneous Dimensional Reduction in Short-Distance Quantum
  Gravity?}},  in {\em The Planck Scale: Proceedings of the XXV Max Born
  Symposium}, pp.~72--80, 2009.
\newblock \href{http://xxx.lanl.gov/abs/0909.3329}{{\tt arXiv:0909.3329}}.

\bibitem{Calcagni:2009fh}
G.~Calcagni, {\it {Fractal Universe and Quantum Gravity}},  {\em Phys. Rev.
  Lett.} {\bf 104} (2010) 251301,
  [\href{http://xxx.lanl.gov/abs/0912.3142}{{\tt arXiv:0912.3142}}].

\bibitem{Horava:2009ho}
P.~Horava, {\it {Spectral Dimension of the Universe in Quantum Gravity at a
  Lifshitz Point}},  {\em Phys. Rev. Lett.} {\bf 102} (2009) 161301,
  [\href{http://xxx.lanl.gov/abs/0902.3657}{{\tt arXiv:0902.3657}}].

\bibitem{Lauscher:2005kn}
O.~Lauscher and M.~Reuter, {\it {Fractal spacetime structure in asymptotically
  safe gravity}},  {\em JHEP} {\bf 10} (2005) 050,
  [\href{http://xxx.lanl.gov/abs/hep-th/0508202}{{\tt hep-th/0508202}}].

\bibitem{Calcagni:2013jx}
G.~Calcagni, A.~Eichhorn, and F.~Saueressig, {\it {Probing the quantum nature
  of spacetime by diffusion}},  {\em Phys. Rev. D} {\bf 87} (2013) 124028,
  [\href{http://xxx.lanl.gov/abs/1304.7247}{{\tt arXiv:1304.7247}}].

\bibitem{Ambjorn:2005fj}
J.~Ambjorn, J.~Jurkiewicz, and R.~Loll, {\it {The spectral dimension of the
  universe is scale dependent}},  {\em Phys. Rev. Lett.} {\bf 95} (2005)
  171301, [\href{http://xxx.lanl.gov/abs/hep-th/0505113}{{\tt
  hep-th/0505113}}].

\bibitem{Ambjorn:2005fh}
J.~Ambjorn, J.~Jurkiewicz, and R.~Loll, {\it {Reconstructing the Universe}},
  {\em Phys. Rev. D} {\bf 72} (2005) 064014,
  [\href{http://xxx.lanl.gov/abs/hep-th/0505154}{{\tt hep-th/0505154}}].

\bibitem{Benedetti:2009bi}
D.~Benedetti and J.~Henson, {\it {Spectral geometry as a probe of quantum
  spacetime}},  {\em Phys. Rev. D} {\bf 80} (2009) 124036,
  [\href{http://xxx.lanl.gov/abs/0911.0401}{{\tt arXiv:0911.0401}}].

\bibitem{Coumbe:2015bq}
D.~N. Coumbe and J.~Jurkiewicz, {\it {Evidence for Asymptotic Safety from
  Dimensional Reduction in Causal Dynamical Triangulations}},  {\em JHEP} {\bf
  1503} (2015) 151, [\href{http://xxx.lanl.gov/abs/1411.7712}{{\tt
  arXiv:1411.7712}}].

\bibitem{Benedetti:2009fo}
D.~Benedetti, {\it {Fractal Properties of Quantum Spacetime}},  {\em Phys. Rev.
  Lett.} {\bf 102} (2009) 111303,
  [\href{http://xxx.lanl.gov/abs/0811.1396}{{\tt arXiv:0811.1396}}].

\bibitem{Alesci:2012jl}
E.~Alesci and M.~Arzano, {\it {Anomalous dimension in three-dimensional
  semiclassical gravity}},  {\em Phys. Lett. B} {\bf 707} (2012) 272--277,
  [\href{http://xxx.lanl.gov/abs/1108.1507}{{\tt arXiv:1108.1507}}].

\bibitem{Modesto:2009bc}
L.~Modesto, {\it {Fractal spacetime from the area spectrum}},  {\em Class.
  Quant. Grav.} {\bf 26} (2009) 242002,
  [\href{http://xxx.lanl.gov/abs/0812.2214}{{\tt arXiv:0812.2214}}].

\bibitem{Gielen:2013cr}
S.~Gielen, D.~Oriti, and L.~Sindoni, {\it {Cosmology from Group Field Theory
  Formalism for Quantum Gravity}},  {\em Phys. Rev. Lett.} {\bf 111} (2013)
  031301, [\href{http://xxx.lanl.gov/abs/1303.3576}{{\tt arXiv:1303.3576}}].

\bibitem{Gielen:2014gv}
S.~Gielen, D.~Oriti, and L.~Sindoni, {\it {Homogeneous cosmologies as group
  field theory condensates}},  {\em JHEP} {\bf 06} (2014) 013,
  [\href{http://xxx.lanl.gov/abs/1311.1238}{{\tt arXiv:1311.1238}}].

\bibitem{Gielen:2014ca}
S.~Gielen, {\it {Quantum cosmology of (loop) quantum gravity condensates: an
  example}},  {\em Class. Quant. Grav.} {\bf 31} (2014) 155009,
  [\href{http://xxx.lanl.gov/abs/1404.2944}{{\tt arXiv:1404.2944}}].

\bibitem{Calcagni:2014jt}
G.~Calcagni, {\it {Loop quantum cosmology from group field theory}},  {\em
  Phys. Rev. D} {\bf 90} (2014) 064047,
  [\href{http://xxx.lanl.gov/abs/1407.8166}{{\tt arXiv:1407.8166}}].

\bibitem{Gielen:2014vk}
S.~Gielen and D.~Oriti, {\it {Quantum cosmology from quantum gravity
  condensates: cosmological variables and lattice-refined dynamics}},  {\em New
  J. Phys.} {\bf 16} (2014) 123004,
  [\href{http://xxx.lanl.gov/abs/1407.8167}{{\tt arXiv:1407.8167}}].

\bibitem{Sindoni:2014vs}
L.~Sindoni, {\it {Effective equations for GFT condensates from fidelity}},
  \href{http://xxx.lanl.gov/abs/1408.3095}{{\tt arXiv:1408.3095}}.

\bibitem{Oriti:2015tl}
D.~Oriti, D.~Pranzetti, J.~P. Ryan, and L.~Sindoni, {\it {Generalized quantum
  gravity condensates for homogeneous geometries and cosmology}},
  \href{http://xxx.lanl.gov/abs/1501.0093}{{\tt arXiv:1501.0093}}.

\bibitem{Calcagni:2013ku}
G.~Calcagni, D.~Oriti, and J.~Th{\"u}rigen, {\it {Laplacians on discrete and
  quantum geometries}},  {\em Class. Quant. Grav.} {\bf 30} (2013) 125006,
  [\href{http://xxx.lanl.gov/abs/1208.0354}{{\tt arXiv:1208.0354}}].

\bibitem{Calcagni:2014ep}
G.~Calcagni, D.~Oriti, and J.~Th{\"u}rigen, {\it {Spectral dimension of quantum
  geometries}},  {\em Class. Quant. Grav.} {\bf 31} (2014) 135014,
  [\href{http://xxx.lanl.gov/abs/1311.3340}{{\tt arXiv:1311.3340}}].

\bibitem{Calcagni:2015is}
G.~Calcagni, D.~Oriti, and J.~Th{\"u}rigen, {\it {Dimensional flow in discrete
  quantum geometries}},  {\em Phys. Rev. D} {\bf 91} (2015) 084047,
  [\href{http://xxx.lanl.gov/abs/1412.8390}{{\tt arXiv:1412.8390}}].

\bibitem{Oriti:2015kv}
D.~Oriti, J.~P. Ryan, and J.~Th{\"u}rigen, {\it {Group field theories for all
  loop quantum gravity}},  {\em New J. Phys.} {\bf 17} (2015) 023042,
  [\href{http://xxx.lanl.gov/abs/1409.3150}{{\tt arXiv:1409.3150}}].

\bibitem{Balzer:1987tx}
W.~Balzer, C.~U. Moulines, and J.~Sneed, ``{An Architectonic for Science}.'' D
  Reidel Pub Co, 1987.

\bibitem{Sneed:1971tc}
J.~D. Sneed, {\em {The logical structure of mathematical physics}}.
\newblock Synthese library. Reidel, Dordrecht, 1971.

\bibitem{Moulines:2002ww}
C.~U. Moulines, {\it {Introduction: Structuralism as a Program for Modelling
  Theoretical Science}},  {\em Synthese} {\bf 130} (2002) 1--11.

\bibitem{Balzer:1980}
W.~Balzer and C.~U. Moulines, {\it {On Theoreticity}},  {\em Synthese} {\bf 44}
  (1980).

\bibitem{Bartelborth:1993vi}
T.~Bartelborth, {\it {Hierarchy Versus Holism: A Structuralist View on General
  Relativity}},  {\em Erkenntnis} {\bf 39} (1993) 383--412.

\bibitem{Burgess:2004tw}
C.~P. Burgess, {\it {Quantum Gravity in Everyday Life: General Relativity as an
  Effective Field Theory}},  {\em Living Rev. Relativity} {\bf 7} (2004)
  [\href{http://xxx.lanl.gov/abs/gr-qc/0311082}{{\tt gr-qc/0311082}}].

\bibitem{AmelinoCamelia:2013ct}
G.~Amelino-Camelia, {\it {Quantum-Spacetime Phenomenology}},  {\em Living Rev.
  Relativity} {\bf 16} (2013) 5.

\bibitem{Regge:1961ct}
T.~E. Regge, {\it {General relativity without coordinates}},  {\em Nuovo
  Cimento} {\bf 19} (1961) 558--571.

\bibitem{Barrett:1994ba}
J.~W. Barrett, {\it {First order Regge calculus}},  {\em Class. Quant. Grav.}
  {\bf 11} (1994) 2723--2730,
  [\href{http://xxx.lanl.gov/abs/hep-th/9404124}{{\tt hep-th/9404124}}].

\bibitem{Barrett:1999fa}
J.~W. Barrett, M.~Rocek, and R.~M. Williams, {\it {A note on area variables in
  Regge calculus}},  {\em Class. Quant. Grav.} {\bf 16} (1999) 1373--1376,
  [\href{http://xxx.lanl.gov/abs/gr-qc/9710056}{{\tt gr-qc/9710056}}].

\bibitem{Makela:1994hm}
J.~M{\"a}kel{\"a}, {\it {Phase space coordinates and the Hamiltonian constraint
  of Regge calculus}},  {\em Phys. Rev. D} {\bf 49} (1994) 2882--2896.

\bibitem{Dittrich:2008hg}
B.~Dittrich and S.~Speziale, {\it {Area--angle variables for general
  relativity}},  {\em New J. Phys.} {\bf 10} (2008) 083006,
  [\href{http://xxx.lanl.gov/abs/0802.0864}{{\tt arXiv:0802.0864}}].

\bibitem{Caselle:1989cd}
M.~Caselle, A.~D'Adda, and L.~Magnea, {\it {Regge calculus as a local theory of
  the Poincar{\'e} group}},  {\em Phys. Lett. B} {\bf 232} (1989) 457--461.

\bibitem{Gionti:2005gi}
G.~Gionti, {\it {Discrete gravity as a local theory of the Poincar{\'e} group
  in the first-order formalism}},  {\em Class. Quant. Grav.} {\bf 22} (2005)
  4217--4230, [\href{http://xxx.lanl.gov/abs/gr-qc/0501082}{{\tt
  gr-qc/0501082}}].

\bibitem{Rocek:1984de}
M.~Ro{\v c}ek and R.~M. Williams, {\it {The quantization of Regge calculus}},
  {\em Z. Phys. C - Particles and Fields} {\bf 21} (1984) 371--381.

\bibitem{Hamber:1992jo}
H.~W. Hamber, {\it {Phases of four-dimensional simplicial quantum gravity}},
  {\em Phys. Rev. D} {\bf 45} (1992) 507--512.

\bibitem{Hamber:2000ew}
H.~W. Hamber, {\it {Gravitational scaling dimensions}},  {\em Phys. Rev. D}
  {\bf 61} (2000) 124008, [\href{http://xxx.lanl.gov/abs/hep-th/9912246}{{\tt
  hep-th/9912246}}].

\bibitem{Riedler:1999bq}
J.~Riedler, W.~Beirl, E.~Bittner, A.~Hauke, P.~Homolka, and H.~Markum, {\it
  {Phase structure and graviton propagators in lattice formulations of
  four-dimensional quantum gravity}},  {\em Class. Quant. Grav.} {\bf 16}
  (1999) 1163--1173.

\bibitem{Collins:1973hw}
P.~Collins and R.~M. Williams, {\it {Dynamics of the Friedmann Universe Using
  Regge Calculus}},  {\em Phys. Rev. D} {\bf 7} (1973) 965--971.

\bibitem{Hamber:1999cf}
H.~W. Hamber and R.~M. Williams, {\it {On the measure in simplicial gravity}},
  {\em Phys. Rev. D} {\bf 59} (1999) 064014,
  [\href{http://xxx.lanl.gov/abs/hep-th/9708019}{{\tt hep-th/9708019}}].

\bibitem{Ambjorn:2001kb}
J.~Ambjorn, J.~Jurkiewicz, and R.~Loll, {\it {Dynamically triangulating
  Lorentzian quantum gravity}},  {\em Nucl. Phys. B} {\bf 610} (2001) 347--382,
  [\href{http://xxx.lanl.gov/abs/hep-th/0105267}{{\tt hep-th/0105267}}].

\bibitem{Ambjorn:2012tj}
J.~Ambjorn, J.~Gizbert-Studnicki, A.~G{\"o}rlich, and J.~Jurkiewicz, {\it {The
  transfer matrix in four-dimensional CDT}},  {\em JHEP} {\bf 09} (2012) 017,
  [\href{http://xxx.lanl.gov/abs/1205.3791}{{\tt arXiv:1205.3791}}].

\bibitem{Arnowitt:1962wr}
R.~Arnowitt, S.~Deser, and C.~W. Misner, {\it {The Dynamics of General
  Relativity }},  in {\em Gravitation: An Introduction to Current Research},
  pp.~227--264.
\newblock Wiley, 1962.

\bibitem{Lewandowski:2006ev}
J.~Lewandowski, A.~Oko{\l}{\'o}w, H.~Sahlmann, and T.~Thiemann, {\it
  {Uniqueness of Diffeomorphism Invariant States on Holonomy Flux Algebras}},
  {\em Comm. Math. Phys.} {\bf 267} (2006) 703--733,
  [\href{http://xxx.lanl.gov/abs/gr-qc/0504147}{{\tt gr-qc/0504147}}].

\bibitem{Thiemann:2006cx}
T.~Thiemann, {\it {The Phoenix Project: master constraint programme for loop
  quantum gravity}},  {\em Class. Quant. Grav.} {\bf 23} (2006) 2211--2247,
  [\href{http://xxx.lanl.gov/abs/gr-qc/0305080}{{\tt gr-qc/0305080}}].

\bibitem{Thiemann:2006ir}
T.~Thiemann, {\it {Quantum spin dynamics: VIII. The master constraint}},  {\em
  Class. Quant. Grav.} {\bf 23} (2006) 2249--2265,
  [\href{http://xxx.lanl.gov/abs/gr-qc/0510011}{{\tt gr-qc/0510011}}].

\bibitem{Thiemann:2007un}
T.~Thiemann, {\it {Loop quantum gravity}},  in {\em Approaches to Quantum
  Gravity: Toward a New Understanding of Space, Time and Matter} (D.~Oriti,
  ed.), pp.~235--252.
\newblock Cambridge University Press, Cambridge, UK, 2007.

\bibitem{Rovelli:1995gq}
C.~Rovelli and L.~Smolin, {\it {Discreteness of area and volume in quantum
  gravity}},  {\em Phys. Lett.} {\bf 442} (1995) 593--619,
  [\href{http://xxx.lanl.gov/abs/gr-qc/9411005}{{\tt gr-qc/9411005}}].

\bibitem{Ashtekar:1997bn}
A.~Ashtekar and J.~Lewandowski, {\it {Quantum theory of geometry: I. Area
  operators}},  {\em Class. Quant. Grav.} {\bf 14} (1997) A55--A81,
  [\href{http://xxx.lanl.gov/abs/gr-qc/9602046}{{\tt gr-qc/9602046}}].

\bibitem{Ashtekar:1997we}
A.~Ashtekar and J.~Lewandowski, {\it {Quantum theory of geometry. II. Volume
  operators}},  {\em Adv. Theor. Math. Phys.} {\bf 1} (1997) 388--429,
  [\href{http://xxx.lanl.gov/abs/gr-qc/9711031}{{\tt gr-qc/9711031}}].

\bibitem{Dittrich:2009gi}
B.~Dittrich and T.~Thiemann, {\it {Are the spectra of geometrical operators in
  loop quantum gravity really discrete?}},  {\em J. Math. Phys.} {\bf 50}
  (2009) 012503, [\href{http://xxx.lanl.gov/abs/0708.1721}{{\tt
  arXiv:0708.1721}}].

\bibitem{Bojowald:2008wn}
M.~Bojowald, {\it {Loop Quantum Cosmology}},  {\em Living Rev. Relativity} {\bf
  11} (2008) 4.

\bibitem{Bojowald:2011dn}
M.~Bojowald, {\em {Quantum Cosmology: A Fundamental Description of the
  Universe}}, vol.~835 of {\em Lecture Notes in Physics}.
\newblock Springer New York, New York, NY, 2011.

\bibitem{Ashtekar:2011fh}
A.~Ashtekar and P.~Singh, {\it {Loop quantum cosmology: a status report}},
  {\em Class. Quant. Grav.} {\bf 28} (2011) 213001,
  [\href{http://xxx.lanl.gov/abs/1108.0893}{{\tt arXiv:1108.0893}}].

\bibitem{Banerjee:2012fn}
K.~Banerjee, G.~Calcagni, and M.~Martin-Benito, {\it {Introduction to Loop
  Quantum Cosmology}},  {\em SIGMA} {\bf 8} (2012) 016,
  [\href{http://xxx.lanl.gov/abs/1109.6801}{{\tt arXiv:1109.6801}}].

\bibitem{Ashtekar:2013fk}
A.~Ashtekar, {\it {Introduction to Loop Quantum Gravity and Cosmology}},  in
  {\em Quantum Gravity and Quantum Cosmology}, pp.~31--56.
\newblock Springer, Berlin, Heidelberg, 2013.

\bibitem{BarberoG:2011ur}
F.~Barbero, J.~Lewandowski, and E.~J.~S. Villasenor, {\it {Quantum isolated
  horizons and black hole entropy}},  in {\em PoS QGQGS2011}, p.~023, 2011.

\bibitem{BarberoG:2015ws}
J.~F. Barbero~G and A.~Perez, {\it {Quantum Geometry and Black Holes}},
  \href{http://xxx.lanl.gov/abs/1501.0296}{{\tt arXiv:1501.0296}}.

\bibitem{Sahlmann:2011ch}
H.~Sahlmann, {\it {Black hole horizons from within loop quantum gravity}},
  {\em Phys. Rev. D} {\bf 84} (2011) 044049,
  [\href{http://xxx.lanl.gov/abs/1104.4691}{{\tt arXiv:1104.4691}}].

\bibitem{Baratin:2012gc}
A.~Baratin, C.~Flori, and T.~Thiemann, {\it {The Holst spin foam model via
  cubulations}},  {\em New J. Phys.} {\bf 14} (2012) 103054,
  [\href{http://xxx.lanl.gov/abs/0812.4055}{{\tt arXiv:0812.4055}}].

\bibitem{Atiyah:1988id}
M.~Atiyah, {\it {Topological quantum field theories}},  {\em Publications
  Math{\'e}matiques de l'Institut des Hautes Scientifiques} {\bf 68} (1988)
  175--186.

\bibitem{Plebanski:1977zz}
J.~F. Plebanski, {\it {On the Separation of Einsteinian Substructures}},  {\em
  J. Math. Phys.} {\bf 18} (1977) 2511--2520.

\bibitem{Blau:1991ji}
M.~Blau and G.~Thompson, {\it {Topological gauge theories of antisymmetric
  tensor fields}},  {\em Annals Phys.} {\bf 205} (1991) 130--172.

\bibitem{Rovelli:2006dy}
C.~Rovelli, {\it {Graviton Propagator from Background-Independent Quantum
  Gravity}},  {\em Phys. Rev. Lett.} {\bf 97} (2006) 151301,
  [\href{http://xxx.lanl.gov/abs/gr-qc/0508124}{{\tt gr-qc/0508124}}].

\bibitem{Bianchi:2006ka}
E.~Bianchi, L.~Modesto, C.~Rovelli, and S.~Speziale, {\it {Graviton propagator
  in loop quantum gravity}},  {\em Class. Quant. Grav.} {\bf 23} (2006)
  6989--7028, [\href{http://xxx.lanl.gov/abs/gr-qc/0604044}{{\tt
  gr-qc/0604044}}].

\bibitem{Livine:2006dm}
E.~R. Livine and S.~Speziale, {\it {Group integral techniques for the spinfoam
  graviton propagator}},  {\em JHEP} {\bf 11} (2006) 092,
  [\href{http://xxx.lanl.gov/abs/gr-qc/0608131}{{\tt gr-qc/0608131}}].

\bibitem{Bianchi:2009kc}
E.~Bianchi, E.~Magliaro, and C.~Perini, {\it {LQG propagator from the new spin
  foams}},  {\em Nucl. Phys. B} {\bf 22} (2009) 245--269,
  [\href{http://xxx.lanl.gov/abs/0905.4082}{{\tt arXiv:0905.4082}}].

\bibitem{Bianchi:2010ej}
E.~Bianchi, C.~Rovelli, and F.~Vidotto, {\it {Towards spinfoam cosmology}},
  {\em Phys. Rev. D} {\bf 82} (2010) 084035,
  [\href{http://xxx.lanl.gov/abs/1003.3483}{{\tt arXiv:1003.3483}}].

\bibitem{Ashtekar:2010eh}
A.~Ashtekar, M.~Campiglia, and A.~Henderson, {\it {Path integrals and the WKB
  approximation in loop quantum cosmology}},  {\em Phys. Rev. D} {\bf 82}
  (2010) 124043, [\href{http://xxx.lanl.gov/abs/1011.1024}{{\tt
  arXiv:1011.1024}}].

\bibitem{Rovelli:2010kz}
C.~Rovelli and F.~Vidotto, {\it {On the spinfoam expansion in cosmology}},
  {\em Class. Quant. Grav.} {\bf 27} (2010) 145005,
  [\href{http://xxx.lanl.gov/abs/0911.3097}{{\tt arXiv:0911.3097}}].

\bibitem{Bianchi:2011bd}
E.~Bianchi, T.~Krajewski, C.~Rovelli, and F.~Vidotto, {\it {Cosmological
  constant in spinfoam cosmology}},  {\em Phys. Rev. D} {\bf 83} (2011) 104015,
  [\href{http://xxx.lanl.gov/abs/1101.4049}{{\tt arXiv:1101.4049}}].

\bibitem{Kisielowski:2013du}
M.~Kisielowski, J.~Lewandowski, and J.~Puchta, {\it {One vertex spin-foams with
  the Dipole Cosmology boundary}},  {\em Class. Quant. Grav.} {\bf 30} (2013)
  025007, [\href{http://xxx.lanl.gov/abs/1203.1530}{{\tt arXiv:1203.1530}}].

\bibitem{Baratin:2012ge}
A.~Baratin and D.~Oriti, {\it {Ten questions on Group Field Theory (and their
  tentative answers)}},  {\em J. Phys.: Conf. Ser.} {\bf 360} (2012) 2002,
  [\href{http://xxx.lanl.gov/abs/1112.3270}{{\tt arXiv:1112.3270}}].

\bibitem{Oriti:2009ur}
D.~Oriti, {\it {The group field theory approach to quantum gravity: some recent
  results}},  in {\em The Planck Scale: Proceedings of the XXV Max Born
  Symposium}, 2009.
\newblock \href{http://xxx.lanl.gov/abs/0912.2441}{{\tt arXiv:0912.2441}}.

\bibitem{Francesco:1995ih}
P.~Di~Francesco, P.~Ginsparg, and J.~Zinn-Justin, {\it {2D gravity and random
  matrices}},  {\em Phys. Rept.} {\bf 254} (1995) 1--133,
  [\href{http://xxx.lanl.gov/abs/hep-th/9306153}{{\tt hep-th/9306153}}].

\bibitem{Curiel:2009vo}
E.~Curiel, {\it {Classical mechanics is lagrangian; it is not hamiltonian; the
  semantics of physical theory is not semantical}}, .

\bibitem{Isham:1994db}
C.~J. Isham, {\it {Prima facie questions in quantum gravity}},  in {\em
  Canonical Gravity: From Classical to Quantum: Proceedings of the 117th WE
  Heraeus Seminar} (J.~Ehlers and H.~Friedrich, eds.), pp.~1--21.
\newblock Springer Berlin Heidelberg, Berlin, Heidelberg, 1994.
\newblock \href{http://xxx.lanl.gov/abs/gr-qc/9310031}{{\tt gr-qc/9310031}}.

\bibitem{Huggett:2013uu}
N.~Huggett and C.~W{\"u}thrich, {\it {Emergent spacetime and empirical
  (in)coherence}},  {\em Stud Hist Philos M P} {\bf 44} (2013) 276--285,
  [\href{http://xxx.lanl.gov/abs/1206.6290}{{\tt arXiv:1206.6290}}].

\bibitem{Gurau:2011dw}
R.~Gurau, {\it {Colored Group Field Theory}},  {\em Comm. Math. Phys.} {\bf
  304} (2011) 69--93, [\href{http://xxx.lanl.gov/abs/0907.2582}{{\tt
  arXiv:0907.2582}}].

\bibitem{Gurau:2012hl}
R.~Gurau and J.~P. Ryan, {\it {Colored Tensor Models - a Review}},  {\em SIGMA}
  {\bf 8} (2012) 020, [\href{http://xxx.lanl.gov/abs/1109.4812}{{\tt
  arXiv:1109.4812}}].

\bibitem{Fairbairn:2004db}
W.~J. Fairbairn and C.~Rovelli, {\it {Separable Hilbert space in loop quantum
  gravity}},  {\em J. Math. Phys.} {\bf 45} (2004) 2802,
  [\href{http://xxx.lanl.gov/abs/gr-qc/0403047}{{\tt gr-qc/0403047}}].

\bibitem{Bonzom:2012gw}
V.~Bonzom and M.~Smerlak, {\it {Gauge Symmetries in Spin-Foam Gravity: The Case
  for "Cellular Quantization"}},  {\em Phys. Rev. Lett.} {\bf 108} (2012)
  241303, [\href{http://xxx.lanl.gov/abs/1201.4996}{{\tt arXiv:1201.4996}}].

\bibitem{Bonzom:2012gwa}
V.~Bonzom and M.~Smerlak, {\it {Bubble Divergences from Twisted Cohomology}},
  {\em Comm. Math. Phys.} {\bf 312} (2012) 399--426,
  [\href{http://xxx.lanl.gov/abs/1008.1476}{{\tt arXiv:1008.1476}}].

\bibitem{Bonzom:2012tg}
V.~Bonzom and M.~Smerlak, {\it {Bubble Divergences: Sorting out Topology from
  Cell Structure}},  {\em Ann. Henri Poincar{\'e}} {\bf 13} (2012) 185--208,
  [\href{http://xxx.lanl.gov/abs/1103.3961}{{\tt arXiv:1103.3961}}].

\bibitem{Ponzano:1968wi}
G.~Ponzano and T.~E. Regge, {\it {Semiclassical limit of Racah coefficients}},
  in {\em Spectroscopic and group theoretical methods in physics} (F.~Bloch,
  ed.), pp.~1--58.
\newblock North-Holland, Amsterdam, 1968.

\bibitem{Barrett:1999wu}
J.~W. Barrett and R.~M. Williams, {\it {The asymptotics of an amplitude for the
  4-simplex}},  {\em Adv. Theor. Math. Phys.} {\bf 3} (1999) 209--215,
  [\href{http://xxx.lanl.gov/abs/gr-qc/9809032}{{\tt gr-qc/9809032}}].

\bibitem{Barrett:2011wc}
J.~W. Barrett, R.~J. Dowdall, W.~J. Fairbairn, F.~Hellmann, and R.~Pereira,
  {\it {Asymptotic analysis of Lorentzian spin foam models}},  in {\em PoS
  QGQGS2011}, p.~009, 2011.

\bibitem{Barrett:2009ci}
J.~W. Barrett, R.~J. Dowdall, W.~J. Fairbairn, H.~Gomes, and F.~Hellmann, {\it
  {Asymptotic analysis of the Engle-Pereira-Rovelli-Livine four-simplex
  amplitude}},  {\em J. Math. Phys.} {\bf 50} (2009) 2504,
  [\href{http://xxx.lanl.gov/abs/0902.1170}{{\tt arXiv:0902.1170}}].

\bibitem{Barrett:2011bb}
J.~W. Barrett, R.~J. Dowdall, W.~J. Fairbairn, H.~Gomes, F.~Hellmann, and
  R.~Pereira, {\it {Asymptotics of 4d spin foam models}},  {\em Gen. Rel.
  Grav.} {\bf 43} (2011) 2421--2436,
  [\href{http://xxx.lanl.gov/abs/1003.1886}{{\tt arXiv:1003.1886}}].

\bibitem{Rovelli:2014vg}
C.~Rovelli and F.~Vidotto, {\em {Covariant Loop Quantum Gravity: An Elementary
  Introduction to Quantum Gravity and Spinfoam Theory}}.
\newblock Cambridge University Press, Cambridge, UK, 2014.

\bibitem{Thiemann:2014fn}
T.~Thiemann and A.~Zipfel, {\it {Linking covariant and canonical LQG II: spin
  foam projector}},  {\em Class. Quant. Grav.} {\bf 31} (2014) 125008,
  [\href{http://xxx.lanl.gov/abs/1307.5885}{{\tt arXiv:1307.5885}}].

\bibitem{Noui:2005js}
K.~Noui and A.~Perez, {\it {Three-dimensional loop quantum gravity: physical
  scalar product and spin-foam models}},  {\em Class. Quant. Grav.} {\bf 22}
  (2005) 1739--1761, [\href{http://xxx.lanl.gov/abs/gr-qc/0402110}{{\tt
  gr-qc/0402110}}].

\bibitem{Alesci:2008jg}
E.~Alesci, K.~Noui, and F.~Sardelli, {\it {Spin-foam models and the physical
  scalar product}},  {\em Phys. Rev. D} {\bf 78} (2008) 104009,
  [\href{http://xxx.lanl.gov/abs/0807.3561}{{\tt arXiv:0807.3561}}].

\bibitem{Alesci:2012cu}
E.~Alesci, T.~Thiemann, and A.~Zipfel, {\it {Linking covariant and canonical
  loop quantum gravity: New solutions to the Euclidean scalar constraint}},
  {\em Phys. Rev. D} {\bf 86} (2012) 024017,
  [\href{http://xxx.lanl.gov/abs/1109.1290}{{\tt arXiv:1109.1290}}].

\bibitem{Alesci:2013cd}
E.~Alesci, K.~Liegener, and A.~Zipfel, {\it {Matrix elements of Lorentzian
  Hamiltonian constraint in loop quantum gravity}},  {\em Phys. Rev. D} {\bf
  88} (2013) 084043, [\href{http://xxx.lanl.gov/abs/1306.0861}{{\tt
  arXiv:1306.0861}}].

\bibitem{Gurau:2011aq}
R.~Gurau and V.~Rivasseau, {\it {The 1/N expansion of colored tensor models in
  arbitrary dimension}},  {\em Europhys Lett} {\bf 95} (2011) 50004,
  [\href{http://xxx.lanl.gov/abs/1101.4182}{{\tt arXiv:1101.4182}}].

\bibitem{Bonzom:2011cs}
V.~Bonzom, R.~Gurau, A.~Riello, and V.~Rivasseau, {\it {Critical behavior of
  colored tensor models in the large N limit}},  {\em Nucl. Phys. B} {\bf 853}
  (2011) 174--195, [\href{http://xxx.lanl.gov/abs/1105.3122}{{\tt
  arXiv:1105.3122}}].

\bibitem{Gurau:2012ek}
R.~Gurau, {\it {The Complete 1/N Expansion of Colored Tensor Models in
  Arbitrary Dimension}},  {\em Ann. Henri Poincar{\'e}} {\bf 13} (2012)
  399--423, [\href{http://xxx.lanl.gov/abs/1102.5759}{{\tt arXiv:1102.5759}}].

\bibitem{Gurau:2013th}
R.~Gurau and J.~P. Ryan, {\it {Melons are branched polymers}},  {\em Ann. Henri
  Poincar{\'e}} {\bf 15} (2014) 2085--2131,
  [\href{http://xxx.lanl.gov/abs/1302.4386}{{\tt arXiv:1302.4386}}].

\bibitem{Bartelborth:2002us}
T.~Bartelborth, {\it {Explanatory unification}},  {\em Synthese} {\bf 130}
  (2002) 91--107.

\bibitem{Bartelborth:1996uj}
T.~Bartelborth, {\em {Begr{\"u}ndungsstrategien: ein Weg durch die analytische
  Erkenntnistheorie}}.
\newblock Akademie Verlag, Berlin, 1996.

\bibitem{BonJour:1985tu}
L.~BonJour, {\em {The Structure of Empirical Knowledge}}.
\newblock Harvard University Press, Boston, 1985.

\bibitem{Oriti:2013vv}
D.~Oriti, {\it {Group field theory as the 2nd quantization of Loop Quantum
  Gravity}},  \href{http://xxx.lanl.gov/abs/1310.7786}{{\tt arXiv:1310.7786}}.

\bibitem{Kozlov:2008wc}
D.~Kozlov, {\em {Combinatorial Algebraic Topology}}.
\newblock Algorithms and Computation in Mathematics. Springer, 2008.

\bibitem{Gurau:2010iu}
R.~Gurau, {\it {Lost in translation: topological singularities in group field
  theory}},  {\em Class. Quant. Grav.} {\bf 27} (2010) 235023,
  [\href{http://xxx.lanl.gov/abs/1006.0714}{{\tt arXiv:1006.0714}}].

\bibitem{Reidemeister:1938vf}
K.~Reidemeister, {\em {Topologie der Polyeder}}.
\newblock Mathematik und ihre Anwendungen in Monographien und Lehrb{\"u}chern.
  Akademische Verlagesellschaft M. B. H., 1938.

\bibitem{Desbrun:2005ug}
M.~Desbrun, A.~N. Hirani, M.~Leok, and J.~E. Marsden, {\it {Discrete Exterior
  Calculus}},  \href{http://xxx.lanl.gov/abs/math/0508341}{{\tt math/0508341}}.

\bibitem{Grady:2010wb}
L.~J. Grady and J.~R. Polimeni, {\em {Discrete Calculus: Applied Analysis on
  Graphs for Computational Science}}.
\newblock Springer, Dordrecht, 2010.

\bibitem{McMullen:2009ff}
P.~McMullen and E.~Schulte, {\em {Abstract Regular Polytopes}}.
\newblock Cambridge University Press, Cambridge, UK, 2009.

\bibitem{Hatcher:2002ut}
A.~Hatcher, {\em {Algebraic Topology}}.
\newblock University Press, Cambridge, UK, 2002.

\bibitem{Danzer:1982dp}
L.~Danzer and E.~Schulte, {\it {Regul{\"a}re Inzidenzkomplexe I}},  {\em Geom.
  Dedicata} {\bf 13} (1982) 295--308.

\bibitem{Seifert:1980uo}
H.~Seifert and W.~Threlfall, {\em {Seifert and Threlfall, A textbook of
  topology}}.
\newblock Pure and Applied Mathematics. Elsevier Science, 1980.

\bibitem{Berge:1989wn}
C.~Berge, {\em {Hypergraphs: Combinatorics of Finite Sets}}.
\newblock North-Holland mathematical library. North Holland, 1989.

\bibitem{Smerlak:2011ea}
M.~Smerlak, {\it {Comment on 'Lost in translation: topological singularities in
  group field theory'}},  {\em Class. Quant. Grav.} {\bf 28} (2011) 178001,
  [\href{http://xxx.lanl.gov/abs/1102.1844}{{\tt arXiv:1102.1844}}].

\bibitem{Bahr:2013ek}
B.~Bahr, B.~Dittrich, F.~Hellmann, and W.~Kaminski, {\it {Holonomy spin foam
  models: Definition and coarse graining}},  {\em Phys. Rev. D} {\bf 87} (2013)
  044048, [\href{http://xxx.lanl.gov/abs/1208.3388}{{\tt arXiv:1208.3388}}].

\bibitem{Bahr:2011ey}
B.~Bahr, F.~Hellmann, W.~Kaminski, M.~Kisielowski, and J.~Lewandowski, {\it
  {Operator spin foam models}},  {\em Class. Quant. Grav.} {\bf 28} (2011)
  105003, [\href{http://xxx.lanl.gov/abs/1010.4787}{{\tt arXiv:1010.4787}}].

\bibitem{Bahr:2012iu}
B.~Bahr, {\it {Operator Spin Foams: holonomy formulation and coarse graining}},
   {\em J. Phys.: Conf. Ser.} {\bf 360} (2012) 012042,
  [\href{http://xxx.lanl.gov/abs/1112.3567}{{\tt arXiv:1112.3567}}].

\bibitem{Bonzom:2012bg}
V.~Bonzom, R.~Gurau, and V.~Rivasseau, {\it {Random tensor models in the large
  N limit: Uncoloring the colored tensor models}},  {\em Phys. Rev. D} {\bf 85}
  (2012) 084037, [\href{http://xxx.lanl.gov/abs/1202.3637}{{\tt
  arXiv:1202.3637}}].

\bibitem{Thurigen:2015ba}
J.~Th{\"u}rigen, {\it {Fields and Laplacians on Quantum Geometries}},  in {\em
  The Thirteenth Marcel Grossmann Meeting}, pp.~2168--2170, 2015.

\bibitem{Albeverio:1990ii}
S.~Albeverio and B.~Zegarlinski, {\it {Construction of convergent simplicial
  approximations of quantum fields on Riemannian manifolds}},  {\em Comm. Math.
  Phys.} {\bf 132} (1990) 39--71.

\bibitem{Adams:1996ul}
D.~H. Adams, {\it {R-torsion and linking numbers from simplicial abelian gauge
  theories}},  \href{http://xxx.lanl.gov/abs/hep-th/9612009}{{\tt
  hep-th/9612009}}.

\bibitem{Teixeira:2013ee}
F.~L. Teixeira, {\it {Differential Forms in Lattice Field Theories: An
  Overview}},  {\em ISRN Mathematical Physics} {\bf 2013} (2013) 1--16.

\bibitem{Itzykson:1983vc}
C.~Itzykson, {\it {Fields on a random lattice}},  {\em Progress in Gauge Field
  Theory} (1983).

\bibitem{Christ:1982kr}
N.~H. Christ, R.~Friedberg, and T.~D. Lee, {\it {Random lattice field theory:
  General formulation}},  {\em Nucl. Phys. B} {\bf 202} (1982) 89--125.

\bibitem{Christ:1982hv}
N.~H. Christ, R.~Friedberg, and T.~D. Lee, {\it {Weights of links and
  plaquettes in a random lattice}},  {\em Nucl. Phys. B} {\bf 210} (1982)
  337--346.

\bibitem{Christ:1982bn}
N.~H. Christ, R.~Friedberg, and T.~D. Lee, {\it {Gauge theory on a random
  lattice}},  {\em Nucl. Phys. B} {\bf 210} (1982) 310--336.

\bibitem{Sen:2000cr}
S.~Sen, S.~Sen, J.~Sexton, and D.~H. Adams, {\it {Geometric discretization
  scheme applied to the Abelian Chern-Simons theory}},  {\em Phys. Rev. E} {\bf
  61} (2000) 3174--3185, [\href{http://xxx.lanl.gov/abs/hep-th/0001030}{{\tt
  hep-th/0001030}}].

\bibitem{Mattiussi:1997jp}
C.~Mattiussi, {\it {An Analysis of Finite Volume, Finite Element, and Finite
  Difference Methods Using Some Concepts from Algebraic Topology}},  {\em J.
  Comput. Phys.} {\bf 133} (1997) 289--309.

\bibitem{Teixeira:1999hv}
F.~L. Teixeira and W.~C. Chew, {\it {Lattice electromagnetic theory from a
  topological viewpoint}},  {\em J. Math. Phys.} {\bf 40} (1999) 169--187.

\bibitem{Rosenberg:1997to}
S.~Rosenberg, {\em {The Laplacian on a Riemannian Manifold}}.
\newblock An Introduction to Analysis on Manifolds. Cambridge University Press,
  Cambridge, UK, 1997.

\bibitem{Chung:1997tk}
F.~R.~K. Chung, {\em {Spectral graph theory}}.
\newblock Amer Mathematical Society, 1997.

\bibitem{Wardetzky:2008kk}
M.~Wardetzky, S.~Mathur, F.~K{\"a}lberer, and E.~Grinspun, {\it {Discrete
  laplace operators: no free lunch}},  in {\em Eurographics Symposium on
  Geometry Processing} (A.~Belyaev and M.~Garland, eds.), (New York, New York,
  USA), pp.~1--5, ACM Press, 2007.

\bibitem{Kigami:2001wk}
J.~Kigami, {\em {Analysis on fractals}}.
\newblock University Press, Cambridge, UK, 2001.

\bibitem{Osterwalder:1973hq}
K.~Osterwalder and R.~Schrader, {\it {Axioms for Euclidean Greens Functions}},
  {\em Comm. Math. Phys.} {\bf 31} (1973) 83--112.

\bibitem{Giddings:2001fq}
S.~B. Giddings and M.~Lippert, {\it {Precursors, black holes, and a locality
  bound}},  {\em Phys. Rev. D} {\bf 65} (2001) 024006,
  [\href{http://xxx.lanl.gov/abs/hep-th/0103231}{{\tt hep-th/0103231}}].

\bibitem{Giddings:2004bc}
S.~B. Giddings and M.~Lippert, {\it {The information paradox and the locality
  bound}},  {\em Phys. Rev. D} {\bf 69} (2004) 124019,
  [\href{http://xxx.lanl.gov/abs/hep-th/0402073}{{\tt hep-th/0402073}}].

\bibitem{Giddings:2006cn}
S.~B. Giddings, {\it {Locality in quantum gravity and string theory}},  {\em
  Phys. Rev. D} {\bf 74} (2006) 106006,
  [\href{http://xxx.lanl.gov/abs/hep-th/0604072}{{\tt hep-th/0604072}}].

\bibitem{Calcagni:2012rm}
G.~Calcagni, {\it {Diffusion in multiscale spacetimes}},  {\em Phys. Rev. E}
  {\bf 87} (2013) 012123, [\href{http://xxx.lanl.gov/abs/1205.5046}{{\tt
  arXiv:1205.5046}}].

\bibitem{Calcagni:2012kj}
G.~Calcagni, {\it {Geometry of fractional spaces}},  {\em Adv. Theor. Math.
  Phys.} {\bf 16} (2012) 549--644,
  [\href{http://xxx.lanl.gov/abs/1106.5787}{{\tt arXiv:1106.5787}}].

\bibitem{Calcagni:2011sz}
G.~Calcagni, {\it {Geometry and field theory in multi-fractional spacetime}},
  {\em JHEP} {\bf 2012} (2012) 065,
  [\href{http://xxx.lanl.gov/abs/1107.5041}{{\tt arXiv:1107.5041}}].

\bibitem{Calcagni:2012vd}
G.~Calcagni, {\it {Diffusion in quantum geometry}},  {\em Phys. Rev. D} {\bf
  86} (2012) 044021, [\href{http://xxx.lanl.gov/abs/1204.2550}{{\tt
  arXiv:1204.2550}}].

\bibitem{Calcagni:2012zj}
G.~Calcagni and G.~Nardelli, {\it {Momentum transforms and Laplacians in
  fractional spaces}},  {\em Adv. Theor. Math. Phys.} {\bf 16} (2012)
  1315--1348, [\href{http://xxx.lanl.gov/abs/1202.5383}{{\tt
  arXiv:1202.5383}}].

\bibitem{Reisenberger:2001hd}
M.~P. Reisenberger and C.~Rovelli, {\it {Spacetime as a Feynman diagram: the
  connection formulation}},  {\em Class. Quant. Grav.} {\bf 18} (2001)
  121--140, [\href{http://xxx.lanl.gov/abs/gr-qc/0002095}{{\tt
  gr-qc/0002095}}].

\bibitem{DiFrancesco:1992cn}
P.~Di~Francesco and C.~Itzykson, {\it {A Generating function for fatgraphs}},
  {\em Ann. Henri Poincar{\'e}} {\bf 59} (1993) 117--140,
  [\href{http://xxx.lanl.gov/abs/hep-th/9212108}{{\tt hep-th/9212108}}].

\bibitem{Kazakov:1996et}
V.~A. Kazakov, M.~Staudacher, and T.~Wynter, {\it {Character expansion methods
  for matrix models of dually weighted graphs}},  {\em Comm. Math. Phys.} {\bf
  177} (1996) 451--468, [\href{http://xxx.lanl.gov/abs/hep-th/9502132}{{\tt
  hep-th/9502132}}].

\bibitem{Kazakov:1996fq}
V.~A. Kazakov, M.~Staudacher, and T.~Wynter, {\it {Almost flat planar
  diagrams}},  {\em Comm. Math. Phys.} {\bf 179} (1996) 235--256,
  [\href{http://xxx.lanl.gov/abs/hep-th/9506174}{{\tt hep-th/9506174}}].

\bibitem{Benedetti:2012ed}
D.~Benedetti and R.~Gurau, {\it {Phase transition in dually weighted colored
  tensor models}},  {\em Nucl. Phys. B} {\bf 855} (2012) 420--437,
  [\href{http://xxx.lanl.gov/abs/1108.5389}{{\tt arXiv:1108.5389}}].

\bibitem{Nakahara:2003vx}
M.~Nakahara, {\em {Geometry, Topology, and Physics}}.
\newblock Taylor {\&} Francis, 2003.

\bibitem{Oeckl:2005wg}
R.~Oeckl, {\em {Discrete Gauge Theory: From Lattices to TQFT}}.
\newblock Imperial College Press, 2005.

\bibitem{NIST}
F.~W.~J. Olver, D.~W. Lozier, R.~F. Boisvert, and C.~W. Clark, eds., {\em {NIST
  Handbook of Mathematical Functions}}.
\newblock Cambridge University Press, Cambridge, UK, 2010.

\bibitem{Kittel:2014vo}
T.~Kittel, {\it {The Hilbert Spaces of Loop Quantum Gravity and Group Field
  Theory: A Comparison}},  Master's thesis, Humboldt-Universit{\"a}t zu Berlin,
  2014.

\bibitem{Baratin:2010in}
A.~Baratin and D.~Oriti, {\it {Group Field Theory with Noncommutative Metric
  Variables}},  {\em Phys. Rev. Lett.} {\bf 105} (2010) 221302,
  [\href{http://xxx.lanl.gov/abs/1002.4723}{{\tt arXiv:1002.4723}}].

\bibitem{Baratin:2012br}
A.~Baratin and D.~Oriti, {\it {Group field theory and simplicial gravity path
  integrals: A model for Holst-Plebanski gravity}},  {\em Phys. Rev. D} {\bf
  85} (2012) 044003, [\href{http://xxx.lanl.gov/abs/1111.5842}{{\tt
  arXiv:1111.5842}}].

\bibitem{Dittrich:2012he}
B.~Dittrich, F.~C. Eckert, and M.~Martin-Benito, {\it {Coarse graining methods
  for spin net and spin foam models}},  {\em New J. Phys.} {\bf 14} (2012)
  035008, [\href{http://xxx.lanl.gov/abs/1109.4927}{{\tt arXiv:1109.4927}}].

\bibitem{Dittrich:2012ba}
B.~Dittrich, {\it {From the discrete to the continuous: towards a cylindrically
  consistent dynamics}},  {\em New J. Phys.} {\bf 14} (2012) 123004,
  [\href{http://xxx.lanl.gov/abs/1205.6127}{{\tt arXiv:1205.6127}}].

\bibitem{Engle:2007em}
J.~Engle, R.~Pereira, and C.~Rovelli, {\it {Loop-Quantum-Gravity Vertex
  Amplitude}},  {\em Phys. Rev. Lett.} {\bf 99} (2007) 161301,
  [\href{http://xxx.lanl.gov/abs/0705.2388}{{\tt arXiv:0705.2388}}].

\bibitem{Engle:2008ka}
J.~Engle, R.~Pereira, and C.~Rovelli, {\it {Flipped spinfoam vertex and loop
  gravity}},  {\em Nucl. Phys. B} {\bf 798} (2008) 251--290,
  [\href{http://xxx.lanl.gov/abs/0708.1236}{{\tt arXiv:0708.1236}}].

\bibitem{Engle:2008fj}
J.~Engle, E.~R. Livine, R.~Pereira, and C.~Rovelli, {\it {LQG vertex with
  finite Immirzi parameter}},  {\em Nucl. Phys. B} {\bf 799} (2008) 136--149,
  [\href{http://xxx.lanl.gov/abs/0711.0146}{{\tt arXiv:0711.0146}}].

\bibitem{Freidel:2008fv}
L.~Freidel and K.~Krasnov, {\it {A new spin foam model for 4D gravity}},  {\em
  Class. Quant. Grav.} {\bf 25} (2008) 125018,
  [\href{http://xxx.lanl.gov/abs/0708.1595}{{\tt arXiv:0708.1595}}].

\bibitem{Barrett:1998fp}
J.~W. Barrett and L.~Crane, {\it {Relativistic spin networks and quantum
  gravity}},  {\em J. Math. Phys.} {\bf 39} (1998) 3296,
  [\href{http://xxx.lanl.gov/abs/gr-qc/9709028}{{\tt gr-qc/9709028}}].

\bibitem{Barrett:2000fr}
J.~W. Barrett and L.~Crane, {\it {A Lorentzian signature model for quantum
  general relativity}},  {\em Class. Quant. Grav.} {\bf 17} (2000) 3101--3118,
  [\href{http://xxx.lanl.gov/abs/gr-qc/9903060}{{\tt gr-qc/9903060}}].

\bibitem{Alexandrov:2008im}
S.~Alexandrov, {\it {Simplicity and closure constraints in spin foam models of
  gravity}},  {\em Phys. Rev. D} {\bf 78} (2008) 044033,
  [\href{http://xxx.lanl.gov/abs/0802.3389}{{\tt arXiv:0802.3389}}].

\bibitem{Gielen:2010ek}
S.~Gielen and D.~Oriti, {\it {Classical general relativity as BF-Plebanski
  theory with linear constraints}},  {\em Class. Quant. Grav.} {\bf 27} (2010)
  185017, [\href{http://xxx.lanl.gov/abs/1004.5371}{{\tt arXiv:1004.5371}}].

\bibitem{Livine:2007bq}
E.~R. Livine and S.~Speziale, {\it {New spinfoam vertex for quantum gravity}},
  {\em Phys. Rev. D} {\bf 76} (2007) 084028,
  [\href{http://xxx.lanl.gov/abs/0705.0674}{{\tt arXiv:0705.0674}}].

\bibitem{Dittrich:2014ui}
B.~Dittrich, {\it {The continuum limit of loop quantum gravity - a framework
  for solving the theory}},  \href{http://xxx.lanl.gov/abs/1409.1450}{{\tt
  arXiv:1409.1450}}.

\bibitem{Freidel:2009ek}
L.~Freidel, R.~Gurau, and D.~Oriti, {\it {Group field theory renormalization in
  the 3D case: Power counting of divergences}},  {\em Phys. Rev. D} {\bf 80}
  (2009) 044007, [\href{http://xxx.lanl.gov/abs/0905.3772}{{\tt
  arXiv:0905.3772}}].

\bibitem{BenGeloun:2013fw}
J.~Ben~Geloun and V.~Rivasseau, {\it {A Renormalizable 4-Dimensional Tensor
  Field Theory}},  {\em Comm. Math. Phys.} {\bf 318} (2013) 69--109,
  [\href{http://xxx.lanl.gov/abs/1111.4997}{{\tt arXiv:1111.4997}}].

\bibitem{BenGeloun:2013dl}
J.~Ben~Geloun and D.~O. Samary, {\it {3D Tensor Field Theory: Renormalization
  and One-Loop $\beta$-Functions}},  {\em Ann. Henri Poincar{\'e}} {\bf 14}
  (2013) 1599--1642, [\href{http://xxx.lanl.gov/abs/1201.0176}{{\tt
  arXiv:1201.0176}}].

\bibitem{BenGeloun:2013ek}
J.~Ben~Geloun, {\it {On the finite amplitudes for open graphs in Abelian
  dynamical colored Boulatov-Ooguri models}},  {\em J. Phys. A} {\bf 46} (2013)
  402002, [\href{http://xxx.lanl.gov/abs/1307.8299}{{\tt arXiv:1307.8299}}].

\bibitem{BenGeloun:2013uf}
J.~Ben~Geloun, {\it {Renormalizable Models in Rank $d\geq 2$ Tensorial Group
  Field Theory}},  \href{http://xxx.lanl.gov/abs/1306.1201}{{\tt
  arXiv:1306.1201}}.

\bibitem{Samary:2014bs}
D.~O. Samary and F.~Vignes-Tourneret, {\it {Just Renormalizable TGFT's on
  $U(1)^d$ with Gauge Invariance}},  {\em Comm. Math. Phys.} {\bf 329} (2014)
  545--578, [\href{http://xxx.lanl.gov/abs/1211.2618}{{\tt arXiv:1211.2618}}].

\bibitem{Carrozza:2014ee}
S.~Carrozza, D.~Oriti, and V.~Rivasseau, {\it {Renormalization of Tensorial
  Group Field Theories: Abelian U(1) Models in Four Dimensions}},  {\em Comm.
  Math. Phys.} {\bf 327} (2014) 603--641,
  [\href{http://xxx.lanl.gov/abs/1207.6734}{{\tt arXiv:1207.6734}}].

\bibitem{Carrozza:2014bh}
S.~Carrozza, D.~Oriti, and V.~Rivasseau, {\it {Renormalization of a SU(2)
  Tensorial Group Field Theory in Three Dimensions}},  {\em Comm. Math. Phys.}
  {\bf 330} (2014) 581--637, [\href{http://xxx.lanl.gov/abs/1303.6772}{{\tt
  arXiv:1303.6772}}].

\bibitem{Carrozza:2014tf}
S.~Carrozza, {\it {Discrete Renormalization Group for SU(2) Tensorial Group
  Field Theory}},  \href{http://xxx.lanl.gov/abs/1407.4615}{{\tt
  arXiv:1407.4615}}.

\bibitem{Reisenberger:2000fj}
M.~P. Reisenberger and C.~Rovelli, {\it {Spin foams as Feynman diagrams}},
  \href{http://xxx.lanl.gov/abs/gr-qc/0002083}{{\tt gr-qc/0002083}}.

\bibitem{Ambjorn:2010kv}
J.~Ambjorn, J.~Jurkiewicz, and R.~Loll, {\it {Quantum Gravity as Sum over
  Spacetimes}},  in {\em New Paths Towards Quantum Gravity} (J.~Ambjorn,
  J.~Jurkiewicz, and R.~Loll, eds.), pp.~59--124.
\newblock Springer, Berlin, Heidelberg, 2010.

\bibitem{Reuter:2013ji}
M.~Reuter and F.~Saueressig, {\it {Asymptotic Safety, Fractals, and
  Cosmology}},  in {\em Quantum Gravity and Quantum Cosmology}, pp.~185--226.
\newblock Springer, Berlin, Heidelberg, 2013.
\newblock \href{http://xxx.lanl.gov/abs/1205.5431}{{\tt arXiv:1205.5431}}.

\bibitem{Modesto:2012er}
L.~Modesto, {\it {Super-renormalizable quantum gravity}},  {\em Phys. Rev. D}
  {\bf 86} (2012) 044005, [\href{http://xxx.lanl.gov/abs/1107.2403}{{\tt
  arXiv:1107.2403}}].

\bibitem{Arzano:2013jj}
M.~Arzano and G.~Calcagni, {\it {Black-hole entropy and minimal diffusion}},
  {\em Phys. Rev. D} {\bf D88} (2013) 084017,
  [\href{http://xxx.lanl.gov/abs/1307.6122}{{\tt arXiv:1307.6122}}].

\bibitem{Calcagni:2014ig}
G.~Calcagni and L.~Modesto, {\it {Nonlocality in string theory}},  {\em J.
  Phys. A} {\bf 47} (2014) 355402,
  [\href{http://xxx.lanl.gov/abs/1310.4957}{{\tt arXiv:1310.4957}}].

\bibitem{Calcagni:2013ds}
G.~Calcagni and G.~Nardelli, {\it {Spectral dimension and diffusion in
  multiscale spacetimes}},  {\em Phys. Rev. D} {\bf 88} (2013) 124025,
  [\href{http://xxx.lanl.gov/abs/1304.2709}{{\tt arXiv:1304.2709}}].

\bibitem{Caravelli:2009td}
F.~Caravelli and L.~Modesto, {\it {Fractal Dimension in 3d Spin-Foams}},
  \href{http://xxx.lanl.gov/abs/0905.2170}{{\tt arXiv:0905.2170}}.

\bibitem{Magliaro:2009wa}
E.~Magliaro, C.~Perini, and L.~Modesto, {\it {Fractal Space-Time from
  Spin-Foams}},  \href{http://xxx.lanl.gov/abs/0911.0437}{{\tt
  arXiv:0911.0437}}.

\bibitem{Baratin:2011hc}
A.~Baratin, B.~Dittrich, D.~Oriti, and J.~Tambornino, {\it {Non-commutative
  flux representation for loop quantum gravity}},  {\em Class. Quant. Grav.}
  {\bf 28} (2011) 175011, [\href{http://xxx.lanl.gov/abs/1004.3450}{{\tt
  arXiv:1004.3450}}].

\bibitem{Giasemidis:2012kq}
G.~Giasemidis, J.~F. Wheater, and S.~Zohren, {\it {Dynamical dimensional
  reduction in toy models of 4D causal quantum gravity}},  {\em Phys. Rev. D}
  {\bf 86} (2012) 081503, [\href{http://xxx.lanl.gov/abs/1202.2710}{{\tt
  arXiv:1202.2710}}].

\bibitem{Giasemidis:2012er}
G.~Giasemidis, J.~F. Wheater, and S.~Zohren, {\it {Multigraph models for causal
  quantum gravity and scale dependent spectral dimension}},  {\em J. Phys. A}
  {\bf 45} (2012) 355001, [\href{http://xxx.lanl.gov/abs/1202.6322}{{\tt
  arXiv:1202.6322}}].

\bibitem{Giasemidis:2013wz}
G.~Giasemidis, {\em {Spectral dimension in graph models of causal quantum
  gravity}}.
\newblock PhD thesis, University of Oxford, 2013.
\newblock \href{http://xxx.lanl.gov/abs/1310.8109}{{\tt arXiv:1310.8109}}.

\bibitem{Bell:2011wu}
N.~Bell and A.~N. Hirani, {\it {PyDEC: Software and Algorithms for
  Discretization of Exterior Calculus}},  {\em ACM Transactions on Mathematical
  Software} {\bf 39} (2012) [\href{http://xxx.lanl.gov/abs/1103.3076}{{\tt
  arXiv:1103.3076}}].

\bibitem{Calcagni:2014uv}
G.~Calcagni, L.~Modesto, and G.~Nardelli, {\it {Quantum spectral dimension in
  quantum field theory}},  \href{http://xxx.lanl.gov/abs/1408.0199}{{\tt
  arXiv:1408.0199}}.

\bibitem{Vassilevich:2003fw}
D.~V. Vassilevich, {\it {Heat kernel expansion: user's manual}},  {\em Phys.
  Rept.} {\bf 388} (2003) 279--360,
  [\href{http://xxx.lanl.gov/abs/hep-th/0306138}{{\tt hep-th/0306138}}].

\bibitem{Durhuus:2006bd}
B.~Durhuus, T.~Jonsson, and J.~F. Wheater, {\it {Random walks on combs}},  {\em
  J. Phys. A} {\bf 39} (2006) 1009--1037,
  [\href{http://xxx.lanl.gov/abs/hep-th/0509191}{{\tt hep-th/0509191}}].

\bibitem{Atkin:2011bk}
M.~R. Atkin, G.~Giasemidis, and J.~F. Wheater, {\it {Continuum random combs and
  scale-dependent spectral dimension}},  {\em J. Phys. A} {\bf 44} (2011)
  265001, [\href{http://xxx.lanl.gov/abs/1101.4174}{{\tt arXiv:1101.4174}}].

\bibitem{Freidel:2003kx}
L.~Freidel, E.~R. Livine, and C.~Rovelli, {\it {Spectra of length and area in
  (2+1) Lorentzian loop quantum gravity}},  {\em Class. Quant. Grav.} {\bf 20}
  (2003) 1463--1478, [\href{http://xxx.lanl.gov/abs/gr-qc/0212077}{{\tt
  gr-qc/0212077}}].

\bibitem{Achour:2014gr}
J.~Ben~Achour, M.~Geiller, K.~Noui, and C.~Yu, {\it {Spectra of geometric
  operators in three-dimensional loop quantum gravity: From discrete to
  continuous}},  {\em Phys. Rev. D} {\bf 89} (2014) 64064,
  [\href{http://xxx.lanl.gov/abs/1306.3246}{{\tt arXiv:1306.3246}}].

\bibitem{DePietri:1996en}
R.~De~Pietri and C.~Rovelli, {\it {Geometry eigenvalues and the scalar product
  from recoupling theory in loop quantum gravity}},  {\em Phys. Rev. D} {\bf
  54} (1996) 2664--2690, [\href{http://xxx.lanl.gov/abs/gr-qc/9602023}{{\tt
  gr-qc/9602023}}].

\bibitem{Wiese:2009ub}
U.-J. Wiese, {\it {An Introduction to Lattice Field Theory}},  in {\em 15th
  Saalburg Summer School}, 2009.

\bibitem{Ambjorn:1999in}
J.~Ambjorn, R.~Loll, J.~L. Nielsen, and J.~Rolf, {\it {Euclidean and Lorentzian
  quantum gravity---lessons from two dimensions}},  {\em Chaos, Solitons {\&}
  Fractals} {\bf 10} (1999) 177--195,
  [\href{http://xxx.lanl.gov/abs/hep-th/9806241}{{\tt hep-th/9806241}}].

\bibitem{Wilf:1994to}
H.~S. Wilf, {\em {Generatingfunctionology}}.
\newblock Academic Press, 1994.

\bibitem{Dattoli:1994iu}
G.~Dattoli, S.~Lorenzutta, G.~Maino, A.~Torre, G.~Voykov, and C.~Chiccoli, {\it
  {Theory of two-index Bessel functions and applications to physical
  problems}},  {\em J. Math. Phys.} {\bf 35} (1994) 3636--3649.

\bibitem{Brehm:2011kp}
U.~Brehm and W.~K{\"u}hnel, {\it {Lattice triangulations of $\mathbb{E}^3 $ and
  of the 3-torus}},  {\em Isr. J. Math.} {\bf 189} (2011) 97--133.

\bibitem{Guedes:2013cc}
C.~Guedes, D.~Oriti, and M.~Raasakka, {\it {Quantization maps, algebra
  representation, and non-commutative Fourier transform for Lie groups}},  {\em
  J. Math. Phys.} {\bf 54} (2013) 3508,
  [\href{http://xxx.lanl.gov/abs/1301.7750}{{\tt arXiv:1301.7750}}].

\bibitem{Hall:1994gg}
B.~C. Hall, {\it {The Segal-Bargmann "Coherent State" Transform for Compact Lie
  Groups}},  {\em J. Funct. Anal.} {\bf 122} (1994) 103--151.

\bibitem{Sahlmann:2001bw}
H.~Sahlmann, T.~Thiemann, and O.~Winkler, {\it {Coherent states for canonical
  quantum general relativity and the infinite tensor product extension}},  {\em
  Nucl. Phys. B} {\bf 606} (2001) 401--440,
  [\href{http://xxx.lanl.gov/abs/gr-qc/0102038}{{\tt gr-qc/0102038}}].

\bibitem{Bahr:2009bc}
B.~Bahr and T.~Thiemann, {\it {Gauge-invariant coherent states for loop quantum
  gravity: II. Non-Abelian gauge groups}},  {\em Class. Quant. Grav.} {\bf 26}
  (2009) 045012, [\href{http://xxx.lanl.gov/abs/0709.4636}{{\tt
  arXiv:0709.4636}}].

\bibitem{Oriti:2012kx}
D.~Oriti, R.~Pereira, and L.~Sindoni, {\it {Coherent states in quantum gravity:
  a construction based on the flux representation of loop quantum gravity}},
  {\em J. Phys. A} {\bf 45} (2012) 4004,
  [\href{http://xxx.lanl.gov/abs/1110.5885}{{\tt arXiv:1110.5885}}].

\bibitem{Bianchi:2010fp}
E.~Bianchi, E.~Magliaro, and C.~Perini, {\it {Coherent spin-networks}},  {\em
  Phys. Rev. D} {\bf 82} (2010) 024012,
  [\href{http://xxx.lanl.gov/abs/0912.4054}{{\tt arXiv:0912.4054}}].

\bibitem{Smerlak:2011vt}
M.~Smerlak, {\em {Divergences des mousses de spins (Divergences in spinfoam
  quantum gravity)}}.
\newblock PhD thesis, L'Universit{\'e} de la M{\'e}diterran{\'e}e, Marseille,
  2011.

\bibitem{Calcagni:2014hx}
G.~Calcagni and G.~Nardelli, {\it {Quantum Field Theory with Varying
  Couplings}},  {\em Int. J. Mod. Phys. A} {\bf 29} (2014) 1450012,
  [\href{http://xxx.lanl.gov/abs/1306.0629}{{\tt arXiv:1306.0629}}].

\bibitem{Bovens:2004tw}
L.~Bovens and S.~Hartmann, {\em {Bayesian Epistemology}}.
\newblock Oxford University Press, Oxford, 2004.

\bibitem{Dawid:2013tj}
R.~Dawid, S.~Hartmann, and J.~Sprenger, {\it {The No Alternatives Argument}},
  {\em The British Journal for the Philosophy of Science} {\bf 66} (2015)
  213--234.

\bibitem{Baratin:2011bk}
A.~Baratin, F.~Girelli, and D.~Oriti, {\it {Diffeomorphisms in group field
  theories}},  {\em Phys. Rev. D} {\bf 83} (2011) 104051,
  [\href{http://xxx.lanl.gov/abs/1101.0590}{{\tt arXiv:1101.0590}}].

\bibitem{Freidel:2003bh}
L.~Freidel, {\it {Diffeomorphisms and spin foam models}},  {\em Nucl. Phys. B}
  {\bf 662} (2003) 279--298, [\href{http://xxx.lanl.gov/abs/gr-qc/0212001}{{\tt
  gr-qc/0212001}}].

\bibitem{Dittrich:2008pw}
B.~Dittrich, {\it {Diffeomorphism Symmetry in Quantum Gravity Models}},  {\em
  Adv. Sci. Lett.} {\bf 2} (2009) 151--163,
  [\href{http://xxx.lanl.gov/abs/0810.3594}{{\tt arXiv:0810.3594}}].

\bibitem{Baratin:2014bea}
A.~Baratin, S.~Carrozza, D.~Oriti, J.~P. Ryan, and M.~Smerlak, {\it {Melonic
  Phase Transition in Group Field Theory}},  {\em LMP} {\bf 104} (2014)
  1003--1017, [\href{http://xxx.lanl.gov/abs/1307.5026}{{\tt
  arXiv:1307.5026}}].

\bibitem{Gurau:2011sk}
R.~Gurau, {\it {The Double Scaling Limit in Arbitrary Dimensions: A Toy
  Model}},  {\em Phys. Rev.} {\bf D84} (2011) 124051,
  [\href{http://xxx.lanl.gov/abs/1110.2460}{{\tt arXiv:1110.2460}}].

\bibitem{Thiemann:1998hn}
T.~Thiemann, {\it {Quantum spin dynamics (QSD): V. Quantum gravity as the
  natural regulator of the Hamiltonian constraint of matter quantum field
  theories}},  {\em Class. Quant. Grav.} {\bf 15} (1998) 1281--1314,
  [\href{http://xxx.lanl.gov/abs/gr-qc/9705019}{{\tt gr-qc/9705019}}].

\bibitem{Domagaia:2010iq}
M.~Domaga{\l}a, K.~Giesel, W.~Kaminski, and J.~Lewandowski, {\it {Gravity
  quantized: Loop quantum gravity with a scalar field}},  {\em Phys. Rev. D}
  {\bf 82} (2010) 104038, [\href{http://xxx.lanl.gov/abs/1009.2445}{{\tt
  arXiv:1009.2445}}].

\bibitem{Giesel:2012tj}
K.~Giesel and T.~Thiemann, {\it {Scalar Material Reference Systems and Loop
  Quantum Gravity}},  \href{http://xxx.lanl.gov/abs/1206.3807}{{\tt
  arXiv:1206.3807}}.

\bibitem{Oriti:2002gy}
D.~Oriti and H.~Pfeiffer, {\it {Spin foam model for pure gauge theory coupled
  to quantum gravity}},  {\em Phys. Rev. D} {\bf 66} (2002) 124010,
  [\href{http://xxx.lanl.gov/abs/gr-qc/0207041}{{\tt gr-qc/0207041}}].

\bibitem{Fairbairn:2007fk}
W.~J. Fairbairn, {\it {Fermions in three-dimensional spinfoam quantum
  gravity}},  {\em Gen. Rel. Grav.} {\bf 39} (2007) 427--476,
  [\href{http://xxx.lanl.gov/abs/gr-qc/0609040}{{\tt gr-qc/0609040}}].

\bibitem{Bianchi:2013cm}
E.~Bianchi, M.~Han, C.~Rovelli, W.~M. Wieland, E.~Magliaro, and C.~Perini, {\it
  {Spinfoam fermions}},  {\em Class. Quant. Grav.} {\bf 30} (2013) 235023,
  [\href{http://xxx.lanl.gov/abs/1012.4719}{{\tt arXiv:1012.4719}}].

\bibitem{Freidel:2004bq}
L.~Freidel and D.~Louapre, {\it {Ponzano--Regge model revisited: I. Gauge
  fixing, observables and interacting spinning particles}},  {\em Class. Quant.
  Grav.} {\bf 21} (2004) 5685--5726,
  [\href{http://xxx.lanl.gov/abs/hep-th/0401076}{{\tt hep-th/0401076}}].

\bibitem{Freidel:2006gc}
L.~Freidel and E.~R. Livine, {\it {Ponzano Regge model revisited: III. Feynman
  diagrams and effective field theory}},  {\em Class. Quant. Grav.} {\bf 23}
  (2006) 2021--2061, [\href{http://xxx.lanl.gov/abs/hep-th/0502106}{{\tt
  hep-th/0502106}}].

\bibitem{Baez:2007fm}
J.~C. Baez, A.~S. Crans, and D.~K. Wise, {\it {Exotic statistics for strings in
  4d BF theory}},  {\em Adv. Theor. Math. Phys.} {\bf 11} (2007) 707--749,
  [\href{http://xxx.lanl.gov/abs/gr-qc/0603085}{{\tt gr-qc/0603085}}].

\bibitem{Fairbairn:2008ko}
W.~J. Fairbairn and A.~Perez, {\it {Extended matter coupled to BF theory}},
  {\em Phys. Rev. D} {\bf 78} (2008) 024013,
  [\href{http://xxx.lanl.gov/abs/0709.4235}{{\tt arXiv:0709.4235}}].

\bibitem{Fairbairn:2008bx}
W.~J. Fairbairn, {\it {On gravitational defects, particles and strings}},  {\em
  JHEP} {\bf 09} (2008) 126, [\href{http://xxx.lanl.gov/abs/0807.3188}{{\tt
  arXiv:0807.3188}}].

\bibitem{Fairbairn:2007bu}
W.~J. Fairbairn and E.~R. Livine, {\it {3D spinfoam quantum gravity: matter as
  a phase of the group field theory}},  {\em Class. Quant. Grav.} {\bf 24}
  (2007) 5277--5297, [\href{http://xxx.lanl.gov/abs/gr-qc/0702125}{{\tt
  gr-qc/0702125}}].

\bibitem{Oriti:2009ks}
D.~Oriti, {\it {Emergent non-commutative matter fields from group field theory
  models of quantum spacetime}},  {\em J. Phys.: Conf. Ser.} {\bf 174} (2009)
  2047, [\href{http://xxx.lanl.gov/abs/0903.3970}{{\tt arXiv:0903.3970}}].

\bibitem{AmelinoCamelia:2013hk}
G.~Amelino-Camelia, M.~Arzano, G.~Gubitosi, and J.~Magueijo, {\it {Dimensional
  reduction in the sky}},  {\em Phys. Rev. D} {\bf 87} (2013) 123532,
  [\href{http://xxx.lanl.gov/abs/1305.3153}{{\tt arXiv:1305.3153}}].

\bibitem{AmelinoCamelia:2013kj}
G.~Amelino-Camelia, M.~Arzano, G.~Gubitosi, and J.~Magueijo, {\it {Rainbow
  gravity and scale-invariant fluctuations}},  {\em Phys. Rev. D} {\bf 88}
  (2013) 041303, [\href{http://xxx.lanl.gov/abs/1307.0745}{{\tt
  arXiv:1307.0745}}].

\bibitem{Calcagni:2013gt}
G.~Calcagni, {\it {Multi-scale gravity and cosmology}},  {\em J. Cosmol.
  Astropart. Phys.} {\bf 12} (2013) 041,
  [\href{http://xxx.lanl.gov/abs/1307.6382}{{\tt arXiv:1307.6382}}].

\bibitem{Loll:1998ue}
R.~Loll, {\it {Discrete Approaches to Quantum Gravity in Four Dimensions}},
  {\em Living Rev. Relativity} {\bf 1} (1998) 13,
  [\href{http://xxx.lanl.gov/abs/gr-qc/9805049}{{\tt gr-qc/9805049}}].

\bibitem{AltshillerCourt:1964vd}
N.~Altshiller-Court, {\em {Modern pure solid geometry}}.
\newblock Chelsea Pub. Co., 1964.

\bibitem{Sorkin:1975kv}
R.~D. Sorkin, {\it {The electromagnetic field on a simplicial net}},  {\em J.
  Math. Phys.} {\bf 16} (1975) 2432--2440.

\end{thebibliography}\endgroup
\addcontentsline{toc}{chapter}{\protect\numberline{}Bibliography}


\chapter*{Acknowledgement}

The present work has evolved over the last years at the Albert Einstein Institute in Potsdam.
I am very grateful to my supervisors Daniele Oriti and Gianluca Calcagni for giving me the opportunity to pursue this great endeavour and supporting me with a perfect balance of guidance and freedom.
I am further grateful to Hermann Nicolai for providing this unique environment where all quantum gravity approaches find their place, in particular also sheltering this PhD work.
I am indebted to Daniele Oriti, Gianluca Calcagni, Andreas Pithis and Casey Tomlin for providing very helpful comments on the draft of this thesis.
In addition I also thank James Ryan for an inspiring collaboration on part of this work.

I acknowledge support from the Andrea von Braun foundation motivating the conceptual aspects of this work, from Evangelisches Studienwerk Villigst and from the FAZIT foundation.
Furthermore I would like to thank the quantum gravity groups at University Nijmegen, University Erlangen, and University Copenhagen, as well as the Munich Center for Mathematical Philosophy for hospitality.

Among the many people whose company has made life at the AEI such an enjoyable experience  I would like to devote special thanks to my fellow students
Sylvain Carrozza, Parikshit Dutta,  Frank Eckart, Angelika Fertig, Marco Finocchiaro, Philipp Fleig, Filippo Guarnieri Carlos Guedes, Despoina Katsimpouri, Alexander Kegeles, Tim Kittel, Pan Kessel, Olaf Kr\"uger,  Andreas Pithis, Matti Raasakka Marco Scalisi, Sebastian Steinhaus.

This thesis would not have been possible without the enduring support of my family; but most of all I thank Annika for being there and walking this path together with me.


%
%
%

\end{document}